\documentclass[11pt]{amsart}

\usepackage[margin=1in]{geometry}

\usepackage{amsmath,amssymb,amsthm,mathtools}

\usepackage{graphicx}
\usepackage{booktabs}

\usepackage{tikz}
\usepackage{tikz-cd}
\usetikzlibrary{calc,positioning,arrows.meta,decorations.pathreplacing}
\usepackage{quantikz}
\usepackage[american]{circuitikz}

\usepackage[hidelinks]{hyperref}
\usepackage{bookmark}

\usepackage{xspace}

\newcommand{\Prover}{\mathsf{P}\xspace}
\newcommand{\Verifier}{\mathsf{V}\xspace}

\newcommand{\View}{\mathsf{View}\xspace}

\newcommand{\negl}{\mathsf{negl}}

\makeatletter
\let\origsubsection\subsection
\renewcommand{\subsection}[1]{\origsubsection*{#1}}

\let\origsubsubsection\subsubsection
\renewcommand{\subsubsection}[1]{\origsubsubsection*{#1}}

\let\origparagraph\paragraph
\renewcommand{\paragraph}[1]{\origparagraph*{#1}}
\makeatother

\newcommand{\Mat}{\mathrm{Mat}} 

\DeclareMathOperator{\codim}{codim}
\newcommand{\PP}{\mathbb{P}} 

\makeatletter
\renewcommand{\part}{%
	\PackageError{Manuscript}{Do not use \string\part\space in amsart. Use \string\Part\space instead.}{}%
}

\newcommand{\Part}[1]{%
	\refstepcounter{part}%
	\clearpage
	\thispagestyle{plain}%
	\phantomsection%
	\addcontentsline{toc}{section}{Part \thepart.\ #1}%
	\begin{center}
		\vspace*{2.0em}
		{\Large\bfseries Part \thepart.\ #1\par}
		\vspace*{1.5em}
	\end{center}
	\markboth{Part \thepart.\ #1}{Part \thepart.\ #1}%
}
\makeatother

\newtheoremstyle{artifactstyle}
{0.6\baselineskip}{0.6\baselineskip}
{\itshape}{}{\bfseries}{.}{0.5em}{}
\theoremstyle{artifactstyle}

\newcommand{\R}{\mathbb{R}}
\newcommand{\C}{\mathbb{C}}
\newcommand{\Z}{\mathbb{Z}}
\newcommand{\CP}{\mathbb{CP}}

\DeclareMathOperator{\Tr}{Tr}
\DeclareMathOperator{\Var}{Var}
\DeclareMathOperator{\Cov}{Cov}
\DeclareMathOperator{\rank}{rank}

\DeclareMathOperator{\SU}{SU}
\DeclareMathOperator{\SO}{SO}

\newcommand{\Sim}{\mathsf{Sim}}

\newcommand{\bits}{\{0,1\}}

\theoremstyle{plain}
\newtheorem{thm}{Theorem}[section]
\newtheorem{prop}[thm]{Proposition}

\theoremstyle{definition}
\newtheorem{defn}[thm]{Definition}
\newtheorem{ex}[thm]{Example}
\newtheorem{exercise}[thm]{Exercise}

\theoremstyle{remark}
\newtheorem{rem}[thm]{Remark}

\newenvironment{solution}{\begin{proof}[Solution]}{\end{proof}}

\title[Geometry- and Topology-Informed QC: From States to Real-Time Control]{
	Geometry- and Topology-Informed Quantum Computing: From States to Real-Time Control with FPGA Prototypes
}

\author{Gunhee Cho}
\address{Department of Mathematics, Texas State University}
\email{wvx17@txstate.edu}
\date{}

\begin{document}
	\maketitle
	
	\begin{abstract}
		This book presents a geometry-first and hardware-aware path through modern quantum information workflows, with a concrete goal: connect \emph{states, circuits, and measurement} to the \emph{deterministic classical pipelines} that make hybrid quantum systems actually run. Part~1 builds the mathematical backbone---the linear algebra you truly need, the Bloch-sphere viewpoint, differential-geometric intuition, and quantum Fisher information geometry (QFIM)---so that quantum evolution can be read as motion on curved spaces and measurement as statistics.
		
		Part~2 translates this viewpoint into an FPGA lens: circuits become dataflow graphs, and hybrid loops become streaming pipelines where measurement outcomes are parsed, aggregated, and reduced to small linear-algebra updates that schedule the next pulses. Part~3 develops multi-qubit structure and entanglement as geometry and computation, including teleportation, superdense coding, entanglement detection, and Shor's algorithm through quantum phase estimation.
		
		Part~4 focuses on topological error correction and real-time decoding (Track~A): stabilizer codes, surface-code decoding as \emph{topology $\rightarrow$ graph $\rightarrow$ algorithm}, and the Union--Find decoder down to microarchitectural and RTL-level design constraints. It also treats verification and testbenches (correctness before speed), fault injection, regression methodology, and host/control-stack integration (syndrome in, correction out) with product-level metrics such as bounded latency, p99 behavior, fail-closed policies, and observability.
		
		Optional Track~C studies quantum cryptography and streaming post-processing---BB84/E91, QBER and abort rules, privacy amplification, and zero-knowledge/post-quantum themes---emphasizing that protocol logic is naturally expressed as FSMs, counters, and hash pipelines. The appendices provide hands-on, visualization-driven FPGA examples on the Lattice iCEstick (switch-to-bit conditioning, fixed-point phase arithmetic, experiment sequencing by FSM, and minimal control ISAs), enabling readers to move from geometric principles to implementable systems.
	\end{abstract}
	
	\tableofcontents

\section{Tracks and What This Book Enables}
\label{sec:tracks}

\subsection*{Objective}
This opening section explains \emph{what this book is for} and \emph{how to use it}.
The organizing idea is that modern quantum computing is not ``just circuits'':
it is a \textbf{hybrid stack} where quantum evolution is embedded in a classical
real-time pipeline. The book is therefore written as an \textbf{infrastructure manual}
with three tracks:
\begin{itemize}
	\item \textbf{Track A (Core)}: real-time QEC decoding as FPGA infrastructure,
	\item \textbf{Track B}: geometry/optimization tools (QFIM/QNG) for variational circuits,
	\item \textbf{Track C (Optional)}: cryptography and streaming post-processing.
\end{itemize}
You can follow one track end-to-end, or mix them to match a research/dev workflow.

\medskip
\noindent\textbf{What you should be able to do after this section.}
You should be able to (i) choose a track based on your goal, (ii) understand what artifacts
to produce at each stage (study $\rightarrow$ research $\rightarrow$ development),
and (iii) copy a small set of templates to start producing reproducible results immediately.

\subsection{Why this is an infrastructure book (not a ``quantum computer'' pitch)}

\subsubsection*{1. The central claim}
This book treats quantum computing as a \textbf{systems problem}:
\begin{quote}
	\emph{The useful unit of work is not ``a circuit,'' but a pipeline that converts
		measurement streams into decisions under deadlines.}
\end{quote}
In near-term and even early fault-tolerant regimes, the quantum device is embedded in a
classical control loop. QEC, calibration, adaptive compilation, variational optimization,
and even cryptographic post-processing all become \textbf{streaming computations} with
tight latency and correctness constraints.

\subsubsection*{2. Why ``pitch decks'' fail technically}
A pitch often says: ``Quantum will be exponential.'' An infrastructure view asks:
\begin{itemize}
	\item What is the \textbf{deadline} per control/QEC cycle?
	\item What is the \textbf{tail latency} requirement (p99/p999), not just the mean?
	\item What data is \textbf{streamed}, what is \textbf{buffered}, and what is \textbf{dropped} under overload?
	\item What are the \textbf{interfaces} (wire formats, versioning, host APIs)?
	\item What is the \textbf{verification contract} between a golden model and RTL?
\end{itemize}
This book is organized to answer these questions and to produce artifacts that survive review.

\subsubsection*{3. ``Infrastructure'' in one sentence}
\begin{quote}
	Infrastructure means: \textbf{specifications + benchmarks + test plans + interfaces}
	that make a component reusable across experiments, teams, and hardware platforms.
\end{quote}

\subsubsection*{4. How the math serves infrastructure}
The math is included only when it helps you build a reliable component:
\begin{itemize}
	\item Bloch sphere and circuits $\Rightarrow$ stable reasoning about control and measurement,
	\item QFIM/QNG $\Rightarrow$ geometry-aware updates and sensitivity diagnostics,
	\item topology of codes $\Rightarrow$ clean problem formulations for decoding,
	\item linear algebra $\Rightarrow$ the lingua franca for simulation, verification, and microkernels.
\end{itemize}

\subsection{Track A (Core): Real-Time QEC Decoding as FPGA Infrastructure}

\subsubsection*{1. What Track A builds}
Track A builds a decoder as a deployable infrastructure component:
\begin{quote}
	\emph{Syndrome stream in $\rightarrow$ bounded-time decode $\rightarrow$ correction decisions out.}
\end{quote}
The emphasis is on determinism, tail-latency control, and verification---exactly the reasons FPGAs appear.

\subsubsection*{2. Core objects (what the decoder really ``sees'')}
\begin{itemize}
	\item A stream of measurement bits (or detection events) indexed by space and time;
	\item a local update rule or a small number of bounded passes (to guarantee deadlines);
	\item state stored in structured memory (arrays, banks, FIFOs), not in ``dynamic'' software objects.
\end{itemize}

\subsubsection*{3. The decoding contract (what must be true)}
A decoder component must satisfy:
\begin{itemize}
	\item \textbf{Correctness}: match a golden model under a shared notion of success/failure;
	\item \textbf{Deadline}: worst-case time per window/cycle is bounded by design;
	\item \textbf{Backlog stability}: buffers do not grow without bound under expected load;
	\item \textbf{Observability}: it emits traces/metrics for postmortem and benchmarking;
	\item \textbf{Interfaces}: message formats and versioning are explicit.
\end{itemize}

\subsubsection*{4. Why FPGA is natural here}
\begin{itemize}
	\item Fine-grained parallelism and deterministic latency;
	\item predictable memory access under streaming workloads;
	\item hardware-enforced schedules (FSMs) for bounded passes;
	\item easy integration into control racks as a low-latency endpoint.
\end{itemize}

\subsubsection*{5. Minimum viable Track A outcome}
By the end of Track A, you should have:
\begin{itemize}
	\item a benchmarkable decoder pipeline (even if small-distance / toy-scale),
	\item a golden Python reference and a conformance test plan,
	\item a message schema for syndrome input and correction output,
	\item a latency/throughput budget worksheet with p99/p999 targets.
\end{itemize}

\subsection{Track B: Geometry/Optimization (QFIM, QNG, and Variational Circuits)}

\subsubsection*{1. What Track B builds}
Track B builds geometry-aware optimization tools that plug into hybrid workflows:
\begin{quote}
	\emph{parameters $\theta$ $\rightarrow$ circuit $U(\theta)$ $\rightarrow$ measurement statistics
		$\rightarrow$ QFIM/QNG update $\rightarrow$ new $\theta$.}
\end{quote}
The goal is to compute and use sensitivity information to improve stability and convergence.

\subsubsection*{2. Why geometry matters}
Plain gradient descent treats parameters as Euclidean coordinates.
But circuits live on curved manifolds (state space / projective space), and noise deforms that geometry.
QFIM/QNG provides:
\begin{itemize}
	\item a principled preconditioner (natural gradient),
	\item a diagnostic for barren plateaus and identifiability,
	\item a bridge between information geometry and hardware-aware calibration.
\end{itemize}

\subsubsection*{3. Minimum viable Track B outcome}
By the end of Track B, you should have:
\begin{itemize}
	\item reproducible scripts to compute QFIM for small ans\"atze (pure and basic mixed cases),
	\item worked examples where QNG outperforms Euclidean updates,
	\item a measurement plan explaining how QFIM entries are estimated from shots,
	\item a compact report template (metrics, plots, and failure modes).
\end{itemize}

\subsection{Track C (Optional): Cryptography and Streaming Post-Processing}

\subsubsection*{1. What Track C builds}
Track C treats quantum cryptography and post-processing as \textbf{streaming verification} tasks:
\begin{quote}
	\emph{raw key / outcomes $\rightarrow$ sifting $\rightarrow$ QBER estimation
		$\rightarrow$ error correction $\rightarrow$ privacy amplification.}
\end{quote}
This is a classical pipeline with strict correctness guarantees and clear security thresholds.

\subsubsection*{2. Why it belongs in the same book}
Because it uses the same infrastructure concepts:
\begin{itemize}
	\item streaming data, fixed message formats, versioned schemas,
	\item throughput and tail latency constraints,
	\item end-to-end test vectors and reproducible benchmarks.
\end{itemize}

\subsubsection*{3. Minimum viable Track C outcome}
By the end of Track C, you should have:
\begin{itemize}
	\item an end-to-end BB84 (or E91) simulation with logged metrics,
	\item a streaming post-processing design (sifting + QBER + hashing) suitable for hardware,
	\item a test suite with known-answer tests (KATs) for each stage.
\end{itemize}

\subsection{How to use tracks in a research/dev workflow}

\subsubsection*{1. Three usage patterns}
\begin{enumerate}
	\item \textbf{Track-first (deep)}: finish one track end-to-end (best for a semester project).
	\item \textbf{Artifact-first (fast)}: start from templates, generate artifacts early, then learn theory as needed.
	\item \textbf{Workflow-first (team)}: split tracks across collaborators and align on shared interfaces and benchmarks.
\end{enumerate}

\subsubsection*{2. Suggested cadence (practical)}
A reasonable cadence for a small team:
\begin{itemize}
	\item Weeks 1--2: choose a target pipeline and fill the latency budget worksheet.
	\item Weeks 3--5: produce a golden model + benchmark report skeleton (even for toy scale).
	\item Weeks 6--8: implement a bounded-pass hardware-friendly version (prototype level).
	\item Weeks 9--10: verification + regressions + documentation polish (artifact readiness).
\end{itemize}

\subsubsection*{3. The one rule that keeps projects from failing}
\begin{quote}
	\textbf{Decide interfaces and metrics before you optimize.}
\end{quote}
If you postpone message formats, versioning, and benchmark definitions, you will rewrite everything.

\subsection{Deliverables and artifacts (study $\rightarrow$ research $\rightarrow$ development)}
\label{subsec:deliverables}

\subsubsection*{1. The artifact ladder}
This book is designed so that every chapter can produce at least one artifact that is useful later.
The ladder is:
\[
\text{study artifacts} \;\rightarrow\; \text{research artifacts} \;\rightarrow\; \text{development artifacts}.
\]
The point is continuity: you should not throw away your early work.

\subsubsection{Study artifacts: notes, problem sets, and reproducible calculations}
\begin{itemize}
	\item Clean notes that capture definitions, identities, and standard computations.
	\item Problem sets with fully worked solutions and sanity-check tests.
	\item Reproducible calculations (scripts/notebooks) that regenerate figures and tables.
	\item Small ``unit tests'' for identities used later (unitarity checks, trace identities, etc.).
\end{itemize}

\subsubsection{Research artifacts: experiment templates, metrics, and benchmark reports}
\begin{itemize}
	\item Experiment templates: dataset generation, parameter sweeps, and logging format.
	\item Metrics definitions: latency, throughput, correctness, stability, and tail behavior.
	\item Benchmark report: a fixed structure so results are comparable over time.
	\item Failure-mode catalog: what breaks, how you detect it, and what you log.
\end{itemize}

\subsubsection{Development artifacts: RTL-facing specs, test plans, and host APIs}
\begin{itemize}
	\item RTL-facing specification: state machines, bounded passes, memory maps, and timing assumptions.
	\item Test plan: conformance tests against golden model, randomized tests, and fault injection.
	\item Host API: command protocol, message schema, versioning rules, and error handling.
	\item Regression harness: automated runs that catch performance/correctness drift.
\end{itemize}

\subsection{Quickstart: what to copy and reuse}
\label{subsec:quickstart}

\subsubsection*{0. Quickstart principle}
Copy templates first, then fill them with the smallest nontrivial example.
Do not wait for ``full understanding'' before producing artifacts.

\subsubsection{Latency budget worksheet}
\begin{itemize}
	\item Define the cycle deadline (per QEC round / control loop iteration).
	\item Separate \textbf{mean} and \textbf{tail} targets: p50, p99, p999.
	\item List stages: readout $\rightarrow$ preprocess $\rightarrow$ decode $\rightarrow$ emit correction.
	\item Assign budgets per stage and identify the dominant contributor (usually I/O or memory).
	\item Declare buffering/backpressure policy under overload.
\end{itemize}

\subsubsection{Benchmark report skeleton}
\begin{itemize}
	\item System-under-test: version, parameters, code distance, window size, and platform.
	\item Workload: syndrome model, noise model (if simulated), and trace format.
	\item Metrics: correctness, latency distribution, throughput, memory usage, and failure modes.
	\item Plots: latency CDF/CCDF, throughput vs. load, correctness vs. parameters.
	\item Reproducibility: seed handling, config files, and exact command lines.
\end{itemize}

\subsubsection{Message schema and versioning checklist}
\begin{itemize}
	\item Schema includes: timestamps, indices, payload, and optional metadata fields.
	\item Define endianness, field sizes, alignment, and framing rules.
	\item Add a version field and strict backward/forward compatibility rules.
	\item Define error responses (invalid packet, unsupported version, overflow).
	\item Provide known-answer test vectors for encoding/decoding.
\end{itemize}

	\Part{Foundations: States, Circuits, Geometry, and Measurement}
	
\section{How to Read These Notes (for Engineers and Mathematicians)}
\label{sec:how-to-read}

\subsection*{Objective}

These notes are designed to be used in two modes:
\begin{enumerate}
	\item \textbf{Engineer mode (build-first).}
	You want a concrete pipeline: \emph{state preparation $\to$ circuit $\to$ measurement stream $\to$ classical processing},
	with explicit latency/throughput constraints and a clear path to FPGA-friendly components.
	\item \textbf{Mathematician mode (structure-first).}
	You want the invariants: \emph{projective state spaces, metrics, curvature, tangent vectors, and topological protection},
	with proofs and reusable lemmas that explain why algorithms behave the way they do.
\end{enumerate}

The common language between these modes is \textbf{geometry}:
\[
\text{states live on curved spaces} \quad+\quad
\text{circuits are paths}
\quad\Longrightarrow\quad
\text{noise deforms geometry}
\]
\[
\quad\Longrightarrow\quad
\text{QEC enforces topological constraints}.
\]
By the end of the first part, you should be able to read a circuit not merely as a diagram,
but as a \emph{geometric procedure} whose sensitivity and robustness can be predicted.

\medskip
\noindent\textbf{How to use this section.}
Read it once, then return to it whenever you feel lost.
It is the ``map'' that tells you \emph{which object lives where} and \emph{what you are allowed to compute} at each layer.

\subsection*{Key Concepts: What You Should Expect to Learn}

\begin{itemize}
	\item \textbf{State space is not a vector space.}
	Pure states are rays (projective points), mixed states are density matrices.
	The correct geometry is \emph{projective} and \emph{metric}.
	\item \textbf{Gates are motions.}
	Unitary gates move states along isometries; parameterized gates trace paths.
	\item \textbf{Measurements create classical streams.}
	The output of a quantum device is not ``a state'' but a sequence of bits; the classical side must estimate, filter, and decide.
	\item \textbf{Distances and sensitivity matter.}
	The Quantum Fisher Information Metric (QFIM) and related notions quantify distinguishability and provide geometry-aware updates (QNG).
	\item \textbf{Errors are structured.}
	Noise is not an abstract annoyance; it has locality and symmetries.
	Topological QEC turns local errors into global constraints.
	\item \textbf{Real-time is a physics constraint.}
	Decoding and feedback must meet deadlines; this forces deterministic pipelines (FPGA/ASIC).
\end{itemize}

\medskip
\noindent\textbf{What you will repeatedly do:}
\begin{enumerate}
	\item write a state model (vector or density matrix),
	\item apply a circuit model (matrix/unitary or parameterized family),
	\item compute a statistic (Born probabilities / expectations),
	\item interpret the result geometrically (point, path, distance, curvature),
	\item connect it to a hardware pipeline (shots $\to$ streaming estimates $\to$ decisions).
\end{enumerate}

\subsection*{How engineers should think about the core objects}

This subsection is a translation table. If you read it like a specification, it will save you months.

\subsubsection*{(A) ``State''}
\begin{itemize}
	\item \textbf{Math view:} $\ket{\psi}\in\C^d$ with $\|\psi\|=1$ up to global phase; or $\rho\succeq 0$, $\Tr\rho=1$.
	\item \textbf{Engineering view:} a state is \emph{not directly observable}.
	What you can access are repeated shots producing classical bits that allow you to estimate expectation values.
	\item \textbf{Operational interface:} choose a measurement setting (basis / POVM), run $N$ shots, get a bitstring stream.
\end{itemize}

\subsubsection*{(B) ``Circuit''}
\begin{itemize}
	\item \textbf{Math view:} a unitary $U\in U(d)$; parameterized circuit $U(\theta)$ is a smooth map from parameters to unitaries.
	\item \textbf{Engineering view:} a circuit is a \emph{control program}:
	pulses, timing, calibrations, and compilation constraints.
	\item \textbf{Operational interface:} upload a pulse schedule or gate list; the device executes it, then measures.
\end{itemize}

\subsubsection*{(C) ``Measurement''}
\begin{itemize}
	\item \textbf{Math view:} POVM $\{M_x\}$ with $p(x)=\Tr(\rho M_x)$.
	\item \textbf{Engineering view:} measurement is an analog process producing a voltage trace,
	which is digitized into bits with thresholds and classifiers (introducing readout error).
	\item \textbf{Operational interface:} you get bit outcomes (and maybe confidence scores), not wavefunctions.
\end{itemize}

\subsubsection*{(D) ``Noise''}
\begin{itemize}
	\item \textbf{Math view:} CPTP map $\mathcal{E}$; Markovian channels; Kraus operators; Lindbladians.
	\item \textbf{Engineering view:} drift, dephasing, leakage, cross-talk, calibration error, and readout error.
	\item \textbf{Pipeline implication:} noise changes the geometry (distances/sensitivity) and forces error correction.
\end{itemize}

\subsubsection*{(E) ``Real-time classical processing''}
\begin{itemize}
	\item \textbf{Math view:} estimation, filtering, optimization, decoding; often online / iterative.
	\item \textbf{Engineering view:} deterministic bounded-latency computation, often near the experiment.
	\item \textbf{Pipeline implication:} FPGA becomes a first-class component, not an optional accelerator.
\end{itemize}

\subsection*{Geometry as a bridge: why these notes are ``geometry-first''}

The same phenomenon appears in many disguises:
\begin{quote}
	\emph{If two quantum states are hard to distinguish, then learning/optimization is slow and error-prone.}
\end{quote}
Geometry packages this into distances and metrics.

\subsubsection*{A minimal geometric dictionary}
\begin{itemize}
	\item \textbf{Manifold:} the set of states (pure or mixed) forms a curved space.
	\item \textbf{Tangent vector:} an infinitesimal change in state induced by a small parameter change $\delta\theta$.
	\item \textbf{Metric:} a rule that turns tangent vectors into lengths; the QFIM is the metric that matters for distinguishability.
	\item \textbf{Geodesic:} the ``straightest'' path; many algorithms are best understood as controlled motion along geodesics.
	\item \textbf{Curvature:} how geometry deviates from flatness; curvature influences optimization landscapes and sensitivity.
\end{itemize}

\subsubsection*{Why geometry helps engineering decisions}
Geometry gives you:
\begin{itemize}
	\item \textbf{Sensitivity analysis:} which parameters matter most (large QFIM directions) vs.\ which are sloppy (small QFIM directions).
	\item \textbf{Stable update rules:} QNG rescales updates by the local geometry, reducing step-size tuning pain.
	\item \textbf{A unifying language:} Bloch sphere $\to$ $\CP^{d-1}$ $\to$ QFIM $\to$ decoding graphs/topology.
\end{itemize}

\subsection*{Geometry of states: Bloch sphere and beyond}

\subsubsection*{Single qubit: the Bloch sphere/ball}
For a qubit, pure states correspond to points on a sphere and mixed states fill the interior.
This is the first place where ``projective, not linear'' becomes visually undeniable.

\begin{figure}[t]
	\centering
	\begin{tikzpicture}[scale=2.3, line cap=round, line join=round]
		\draw (0,0) circle (1);
		\draw[dashed] (-1,0) arc (180:360:1 and 0.35);
		\draw (-1,0) arc (180:0:1 and 0.35);
		
		\draw[->] (0,0) -- (1.25,0) node[right] {$x$};
		\draw[->] (0,0) -- (0,1.25) node[above] {$z$};
		\draw[->] (0,0) -- (-0.75,-0.55) node[left] {$y$};
		
		\fill (0,1) circle (0.03) node[above] {\small $\ket{0}$};
		\fill (0,-1) circle (0.03) node[below] {\small $\ket{1}$};
		\fill (1,0) circle (0.03) node[right] {\small $\ket{+}$};
		\fill (-1,0) circle (0.03) node[left] {\small $\ket{-}$};
		
		\coordinate (M) at (0.45,0.2);
		\fill (M) circle (0.03) node[right] {\small mixed};
		\draw[thick] (0,0) -- (M);
		\node at (0.22,0.12) {\small $\vec r$};
		
		\node at (0,-1.25) {\small pure: $\|\vec r\|=1$ (surface), mixed: $\|\vec r\|<1$ (interior)};
	\end{tikzpicture}
	\caption{State geometry for one qubit. Pure states live on the Bloch sphere (surface),
		while mixed states fill the Bloch ball (interior).}
	\label{fig:howto-bloch}
\end{figure}
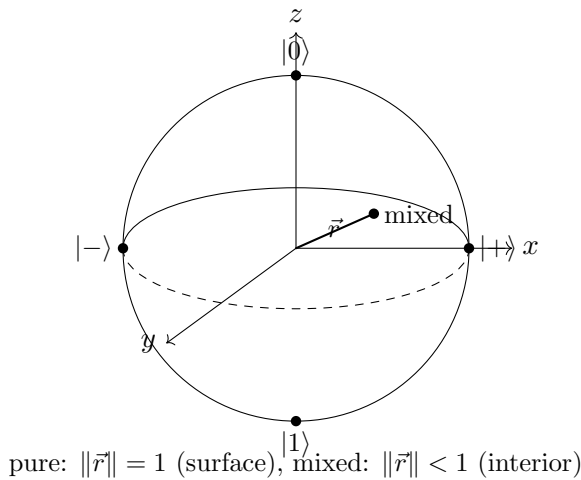

\subsubsection*{Many qubits: projective space and density matrices}
For $n$ qubits, the Hilbert space is $\C^{2^n}$, but the pure-state space is
\[
\CP^{2^n-1},
\]
which is curved and high-dimensional.
When you trace out subsystems (to model partial access), you land in mixed states:
density matrices on $2^k\times 2^k$ for subsystems of size $k$.

\subsubsection*{What you should keep in mind}
\begin{itemize}
	\item The \emph{dimension} grows exponentially, but many circuits explore a structured low-dimensional submanifold.
	\item Geometry tells you which directions are accessible and distinguishable under your measurement model.
\end{itemize}

\subsection*{Geometry of circuits: paths, distances, and robustness}

A parameterized circuit $\theta\mapsto U(\theta)$ induces a path in state space:
\[
\theta \mapsto \rho(\theta)=U(\theta)\rho_0 U(\theta)^\dagger.
\]
Small changes $\delta\theta$ correspond to tangent vectors $\partial_i \rho$.
A metric (QFIM) turns these into meaningful lengths and angles, i.e.\ a local notion of distance in \emph{state space}.

\subsubsection*{Why ``path length'' is not a metaphor}
If the path is long, you can distinguish many intermediate states with finite shots.
If the path is almost flat in some direction (small QFIM eigenvalue),
then changing that parameter barely changes observable statistics, making learning slow and fragile.

\begin{figure}[t]
	\centering
	\begin{tikzpicture}[scale=1.15, line cap=round, line join=round]
		\draw[thick] (-4,0) circle (1.2);
		\node at (-4,1.5) {\small parameter space};
		\node at (-4,-1.5) {\small $\theta$-updates};
		
		\coordinate (t0) at (-4,0.2);
		\coordinate (t1) at (-3.2,0.65);
		\fill (t0) circle (0.04) node[left] {\small $\theta$};
		\fill (t1) circle (0.04) node[right] {\small $\theta+\delta\theta$};
		\draw[->, thick] (t0) -- (t1);
		
		\draw[->, thick] (-2.6,0.4) -- (-1.2,0.4);
		\node at (-1.9,0.65) {\small $U(\theta)$};
		
		\draw[thick] (1,0) .. controls (2.4,1.1) and (3.6,0.9) .. (4.1,0)
		.. controls (3.6,-1.0) and (2.4,-1.2) .. (1,0);
		\node at (2.6,1.55) {\small state space};
		
		\coordinate (s0) at (2.0,0.2);
		\coordinate (s1) at (3.2,0.55);
		\fill (s0) circle (0.04) node[left] {\small $\rho(\theta)$};
		\fill (s1) circle (0.04) node[right] {\small $\rho(\theta+\delta\theta)$};
		\draw[->, thick] (s0) -- (s1);
		\node at (2.6,-1.55) {\small distances measured by QFIM / FS metric};
	\end{tikzpicture}
	\caption{A circuit $U(\theta)$ maps parameter updates into motion in state space.
		Geometry defines which motions are large (observable) or small (hard to learn).}
	\label{fig:howto-param-to-state}
\end{figure}
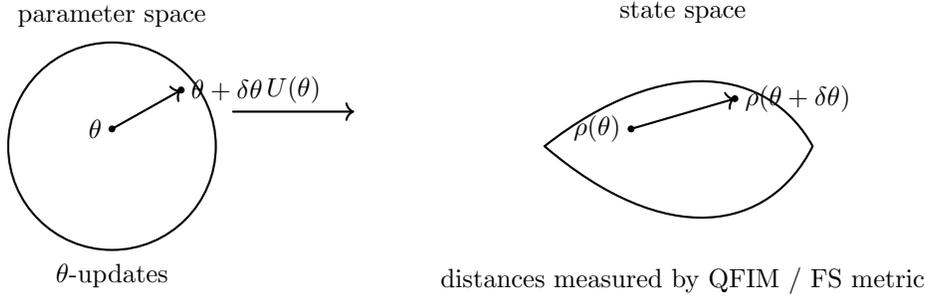

\subsubsection*{Robustness as ``small geometric change under perturbation''}
A circuit is robust if small hardware perturbations (noise, calibration drift, readout bias)
do not significantly change the induced state or the final decision statistic.
Geometry formalizes this: robustness is about \emph{Lipschitz behavior} of the map
\[
\theta \mapsto \text{observable statistics}.
\]

\subsection*{Geometry of errors: topology explains protection}

Topological quantum error correction (QEC) is the idea that local noise should not easily alter global information.
The geometric/topological story is:

\begin{itemize}
	\item Physical noise creates local error events.
	\item Syndrome measurements reveal \emph{boundaries} of error chains (not the chains themselves).
	\item Decoding chooses a consistent error chain with the same boundary.
	\item Logical failure corresponds to choosing a chain in the wrong homology class (a global/topological mistake).
\end{itemize}

\begin{figure}[t]
	\centering
	\begin{tikzpicture}[scale=0.9, line cap=round, line join=round]
		\foreach \x in {0,1,2,3,4,5,6} {
			\foreach \y in {0,1,2,3,4} {
				\fill (\x,\y) circle (0.04);
			}
		}
		\foreach \x in {0,1,2,3,4,5} {
			\foreach \y in {0,1,2,3,4} {
				\draw (\x,\y) -- (\x+1,\y);
			}
		}
		\foreach \x in {0,1,2,3,4,5,6} {
			\foreach \y in {0,1,2,3} {
				\draw (\x,\y) -- (\x,\y+1);
			}
		}
		
		\draw[very thick] (1,1) -- (2,1) -- (3,1) -- (3,2) -- (3,3) -- (4,3);
		\fill (1,1) circle (0.10);
		\fill (4,3) circle (0.10);
		\node[below left] at (1,1) {\small syndrome};
		\node[above right] at (4,3) {\small syndrome};
		
		\node at (3.2,-0.8) {\small decoder sees endpoints, must infer a plausible chain};
	\end{tikzpicture}
	\caption{Topology-flavored picture: an error chain produces syndrome at its boundary endpoints.
		The decoder reconstructs a chain consistent with the observed boundary.}
	\label{fig:howto-topology-errors}
\end{figure}
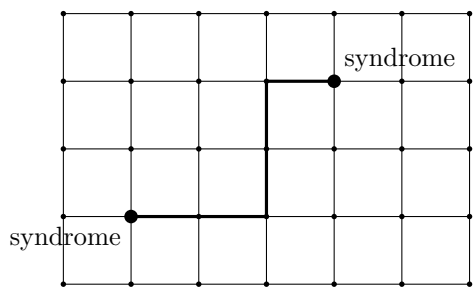

\subsubsection*{Why this matters for infrastructure}
The syndrome stream arrives continuously.
Decoding is not an offline graph problem: it is an online decision problem with deadlines.
This is where FPGA shows up naturally: bounded latency is part of the correctness criterion.

\subsection*{Where FPGA fits into quantum computing}

There are two recurring reasons:
\begin{enumerate}
	\item \textbf{Deadline constraints (feedback and QEC).}
	If you miss the cycle budget, your correction arrives too late and the physical qubits decohere.
	\item \textbf{Determinism and jitter control.}
	Software stacks can have tail latency spikes (p99/p999) that break stability.
	FPGAs offer bounded, clocked pipelines.
\end{enumerate}

\subsubsection*{A minimal ``hybrid loop'' template}
\[
\boxed{\text{QPU: prepare $\to$ evolve $\to$ measure}}
\;\Rightarrow\;
\boxed{\text{FPGA: parse stream $\to$ decode/estimate $\to$ decide}}
\]
\[
\boxed{\text{QPU: apply correction / next pulse}}
\;\Rightarrow\;
\boxed{\text{(back to next round)}}
\]
These notes will treat FPGA not as an afterthought but as an \emph{API-stable infrastructure component}.

\begin{figure}[t]
	\centering
	\begin{tikzpicture}[scale=1.05, line cap=round, line join=round]
		\draw[thick, rounded corners=3pt] (0,0) rectangle (4,1.4);
		\node at (2,1.05) {\textbf{QPU}};
		\node at (2,0.45) {\small prepare $\to$ evolve $\to$ measure};
		
		\draw[thick, rounded corners=3pt] (6,0) rectangle (10,1.4);
		\node at (8,1.05) {\textbf{FPGA}};
		\node at (8,0.75) {\small streaming parse / estimate};
		\node at (8,0.35) {\small decode / decide};
		
		\draw[->, thick] (4,0.7) -- (6,0.7) node[midway, above] {\small bits / timestamps};
		\draw[->, thick] (10,0.7) .. controls (11.2,0.7) and (11.2,2.0) .. (2,2.0)
		.. controls (-0.5,2.0) and (-0.5,0.7) .. (0,0.7)
		node[midway, above] {\small correction / next controls};
		
		\node at (5,2.35) {\small closed-loop control / QEC cycle};
	\end{tikzpicture}
	\caption{Hybrid pipeline viewpoint: the QPU produces a measurement stream; FPGA computes bounded-latency decisions and drives the next control step.}
	\label{fig:howto-hybrid-loop}
\end{figure}
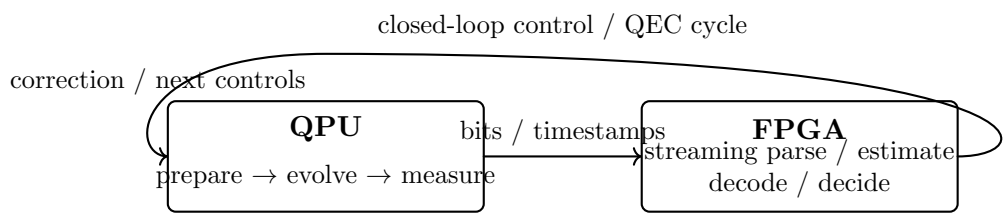

\subsection*{Exercises (warm-up and orientation)}

\begin{exercise}[State vs.\ statistics]
	In one sentence each, explain the difference between:
	(i) a quantum state $\rho$,
	(ii) a measurement setting (basis/POVM),
	(iii) the classical data you receive after $N$ shots.
\end{exercise}
\noindent\textbf{Solution sketch.}
$\rho$ is the (unobserved) mathematical object that generates probabilities;
a measurement setting defines which probabilities you are sampling;
the data are i.i.d.\ (approximately) samples from that distribution, possibly corrupted by readout noise.

\begin{exercise}[Projective sanity check]
	Show by counting real degrees of freedom that pure qubit states have 2 real parameters.
\end{exercise}
\noindent\textbf{Solution sketch.}
$\C^2$ has 4 real parameters.
Normalization removes 1.
Global phase removes 1.
Total $4-1-1=2$.

\begin{exercise}[Circuit as a path]
	Let $\rho(\theta)=U(\theta)\rho_0 U(\theta)^\dagger$ for a smooth one-parameter circuit.
	Write down the first-order variation $\frac{d}{d\theta}\rho(\theta)$ in terms of $U'(\theta)$ and $U(\theta)$.
\end{exercise}
\noindent\textbf{Solution sketch.}
Differentiate:
\[
\frac{d}{d\theta}\rho(\theta)=U'(\theta)\rho_0 U(\theta)^\dagger
+U(\theta)\rho_0 (U'(\theta))^\dagger.
\]
(If you introduce the generator $G(\theta)=i\,U(\theta)^\dagger U'(\theta)$, this becomes a commutator form later.)

\begin{exercise}[Topology vocabulary]
	In the decoder picture, you see syndrome bits corresponding to ``endpoints.''
	Explain in one sentence why decoding is not simply ``finding the unique error.''
\end{exercise}
\noindent\textbf{Solution sketch.}
Many different error chains have the same boundary (same syndrome); decoding chooses one consistent explanation,
and the dangerous ambiguity is the homology class (global/topological choice).

\begin{exercise}[Where p99/p999 matters]
	Give one concrete reason why an algorithm with good \emph{average} runtime might still fail in a real-time QEC loop.
\end{exercise}
\noindent\textbf{Solution sketch.}
Tail latency spikes can exceed the cycle deadline, making corrections arrive too late;
a rare spike can still cause logical failure, so p99/p999 (not mean) is a correctness-relevant metric.

\begin{exercise}[Reading plan]
	Pick one track (A/B/C). Write a 5-line plan:
	(1) which chapters you will read first,
	(2) what artifact you will produce (note, code, benchmark, RTL sketch),
	(3) how you will test correctness.
\end{exercise}
\noindent\textbf{Solution sketch.}
This is personalized; the minimal good answer must mention a deliverable and a test method.


	
\section{Crash Course: Linear Algebra You Actually Need}
\label{sec:la-crash}

\subsection*{Objective}
This chapter is a \emph{working} linear-algebra kit for quantum computing and
hardware-aware implementations.
We focus on the few objects you will use repeatedly---inner products, adjoints,
unitary/Hermitian matrices, eigenvalues/spectra, tensor products, and density matrices---and we practice
them with \emph{explicit calculations} and \emph{algorithm-facing interpretations}.

\medskip
\noindent\textbf{Implementation-aware note.}
Every time you compile a circuit, simulate a state, compute a QFIM entry,
or stream measurement statistics, you are executing linear algebra under constraints:
finite precision, memory layout, bandwidth, and latency. The goal here is not to be abstract,
but to train the ability to:
\begin{quote}
	\emph{turn quantum procedures into matrix identities and then into stable, testable computation.}
\end{quote}

\medskip
\noindent By the end of this section, you should be able to:
\begin{itemize}
	\item move fluently between $\ket{\psi}$, $\bra{\psi}$, $\braket{\phi}{\psi}$, and matrix forms;
	\item compute adjoints, check unitarity/Hermiticity, and interpret them operationally;
	\item diagonalize small matrices, compute spectra, and use spectral decompositions as algorithms;
	\item compute tensor products for multi-qubit bookkeeping and interpret them as wiring rules;
	\item form density matrices, take partial traces (basic examples), and read out mixedness;
	\item solve exercises in a ``how-to-do-it'' style suitable for code and verification.
\end{itemize}

\subsection*{Key Concepts: the ``physics dictionary'' (used every day)}

\subsubsection*{0. Minimal dictionary (keep this on your desk)}
\begin{center}
	\renewcommand{\arraystretch}{1.2}
	\begin{tabular}{ll}
		\textbf{Object} & \textbf{Meaning / how it is used} \\
		\hline
		$\ket{\psi}\in\C^n$ & state vector (column) \\
		$\bra{\psi}=\ket{\psi}^\dagger$ & conjugate transpose (row) \\
		$\braket{\phi}{\psi}=\phi^\dagger\psi$ & inner product; amplitude \\
		$\|\psi\|^2=\braket{\psi}{\psi}$ & norm / probability normalization \\
		$A^\dagger$ & adjoint; used for unitaries and observables \\
		$U^\dagger U=I$ & unitarity; norm-preserving evolution \\
		$H^\dagger=H$ & Hermitian; measurable observable \\
		$\rho\succeq 0,\ \Tr\rho=1$ & density matrix; mixed states \\
		$\Tr(\rho O)$ & expectation value of observable $O$ \\
		$A\otimes B$ & tensor product; multi-qubit wiring \\
	\end{tabular}
\end{center}

\subsection*{Implementation viewpoint (why this chapter also matters for FPGA)}

\subsubsection*{1. Two compute realities: bandwidth and stability}
\begin{itemize}
	\item \textbf{Bandwidth dominates arithmetic} in many hybrid pipelines:
	you are moving vectors/matrices (state, gradients, QFIM blocks) through memory and interfaces.
	\item \textbf{Stability matters}: probabilities are quadratic in amplitudes.
	Small numerical errors in amplitudes can become noticeable errors after squaring and averaging.
\end{itemize}

\subsubsection*{2. Verification viewpoint: golden identities}
In decoder/control stacks, correctness is usually ensured by a small set of linear-algebra invariants:
\begin{itemize}
	\item $\|\psi\|=1$ is preserved by unitaries;
	\item $\rho\succeq 0$ and $\Tr\rho=1$ are preserved by $U\rho U^\dagger$;
	\item projectors satisfy $\Pi^2=\Pi$ and $\Pi^\dagger=\Pi$;
	\item spectral decompositions reconstruct the original matrix: $A=\sum_k \lambda_k \ket{v_k}\bra{v_k}$.
\end{itemize}
These become \emph{unit tests} for simulation code and hardware microkernels.

\subsection*{Notation (Dirac notation) and why it is useful}

\subsubsection*{1. Column/row and outer products}
\begin{defn}[Bra-ket basics]
	For $\ket{\psi}\in\C^n$ (a column vector), define $\bra{\psi}:=\ket{\psi}^\dagger$.
	For two vectors $\ket{\phi},\ket{\psi}$:
	\[
	\braket{\phi}{\psi} := \bra{\phi}\ket{\psi} = \phi^\dagger \psi\in\C.
	\]
	The \emph{outer product} $\ket{\psi}\bra{\phi}$ is the $n\times n$ matrix
	\[
	\ket{\psi}\bra{\phi} : v \mapsto \ket{\psi}\braket{\phi}{v}.
	\]
\end{defn}

\begin{ex}[Explicit outer product computation]
	Let $\ket{\psi}=\binom{1}{i}/\sqrt{2}$ and $\ket{\phi}=\binom{1}{1}/\sqrt{2}$.
	Then
	\[
	\ket{\psi}\bra{\phi}
	=
	\frac{1}{2}
	\begin{pmatrix}
		1\\ i
	\end{pmatrix}
	\begin{pmatrix}
		1 & 1
	\end{pmatrix}
	=
	\frac{1}{2}
	\begin{pmatrix}
		1 & 1\\
		i & i
	\end{pmatrix}.
	\]
\end{ex}

\subsubsection*{2. Projectors as ``measurement'' matrices}
\begin{defn}[Rank-one projector]
	If $\|\psi\|=1$, define
	\[
	\Pi_\psi := \ket{\psi}\bra{\psi}.
	\]
	Then $\Pi_\psi^\dagger=\Pi_\psi$ and $\Pi_\psi^2=\Pi_\psi$.
\end{defn}

\begin{proof}[Idempotence in one line]
	\[
	\Pi_\psi^2=(\ket{\psi}\bra{\psi})(\ket{\psi}\bra{\psi})
	=\ket{\psi}\underbrace{\braket{\psi}{\psi}}_{=1}\bra{\psi}
	=\Pi_\psi.
	\]
\end{proof}

\subsection*{Complex vector spaces and inner products}

\subsubsection*{1. Inner product conventions}
\begin{defn}[Standard inner product on $\C^n$]
	For $\ket{\phi},\ket{\psi}\in\C^n$,
	\[
	\braket{\phi}{\psi}:=\sum_{j=1}^n \overline{\phi_j}\psi_j.
	\]
	It is conjugate-linear in the first slot and linear in the second.
\end{defn}

\begin{prop}[Cauchy--Schwarz]
	For all $\phi,\psi\in\C^n$,
	\[
	|\braket{\phi}{\psi}|\le \|\phi\|\cdot \|\psi\|.
	\]
\end{prop}

\begin{rem}[Operational meaning]
	Amplitudes are inner products; probabilities are squared magnitudes.
	Cauchy--Schwarz is the basic inequality behind ``probabilities are bounded by 1''.
\end{rem}

\subsubsection*{2. Worked probability computation (Born rule in coordinates)}
Let $\ket{\psi}=\alpha\ket{0}+\beta\ket{1}$ with $|\alpha|^2+|\beta|^2=1$.
Measurement in the computational basis gives:
\[
p(0)=|\braket{0}{\psi}|^2=|\alpha|^2,\qquad
p(1)=|\braket{1}{\psi}|^2=|\beta|^2.
\]
Nothing here is ``mystical'': it is inner products and squaring.

\subsection*{Linear maps, matrices, and the two special classes}

\subsubsection*{1. Adjoint and unitarity}
\begin{defn}[Adjoint]
	For a matrix $A\in\C^{n\times n}$, $A^\dagger$ is the conjugate transpose.
\end{defn}

\begin{defn}[Unitary and Hermitian]
	$U$ is \emph{unitary} if $U^\dagger U=I$.
	$H$ is \emph{Hermitian} if $H^\dagger=H$.
\end{defn}

\begin{prop}[Unitaries preserve norms and inner products]
	If $U$ is unitary, then for all $\phi,\psi$,
	\[
	\|U\psi\|=\|\psi\|,\qquad
	\braket{U\phi}{U\psi}=\braket{\phi}{\psi}.
	\]
\end{prop}

\begin{proof}
	\[
	\|U\psi\|^2=\braket{U\psi}{U\psi}=\psi^\dagger U^\dagger U\psi=\psi^\dagger\psi=\|\psi\|^2,
	\]
	and similarly
	\[
	\braket{U\phi}{U\psi}=\phi^\dagger U^\dagger U\psi=\phi^\dagger\psi=\braket{\phi}{\psi}.
	\]
\end{proof}

\subsubsection*{2. Concrete gate checks (do this once, reuse forever)}
\begin{ex}[Pauli $X$ is unitary and Hermitian]
	\[
	X=\begin{pmatrix}0&1\\1&0\end{pmatrix},\quad X^\dagger=X,\quad X^2=I.
	\]
	Hence $X^\dagger X=I$ (unitary) and $X^\dagger=X$ (Hermitian).
\end{ex}

\begin{ex}[Hadamard is unitary]
	\[
	H=\frac{1}{\sqrt2}\begin{pmatrix}1&1\\1&-1\end{pmatrix}.
	\]
	Compute
	\[
	H^\dagger H = H^T H
	=
	\frac12
	\begin{pmatrix}1&1\\1&-1\end{pmatrix}
	\begin{pmatrix}1&1\\1&-1\end{pmatrix}
	=
	\frac12
	\begin{pmatrix}2&0\\0&2\end{pmatrix}
	=I.
	\]
\end{ex}

\subsection*{Eigenvalues, diagonalization, and spectra}

\subsubsection*{1. Why eigenvalues show up in QC}
Eigenvalues appear whenever you have:
\begin{itemize}
	\item \textbf{measurement} (observables are Hermitian; outcomes are eigenvalues);
	\item \textbf{time evolution} $e^{-itH}$ (spectral decomposition turns exponentials into scalars);
	\item \textbf{phase estimation / QPE} (eigenphases are what you measure);
	\item \textbf{stability / conditioning} (spectra control sensitivity).
\end{itemize}

\subsubsection*{2. Full worked diagonalization: a $2\times2$ Hermitian matrix}
\begin{ex}[Diagonalize explicitly]
	Let
	\[
	A=\begin{pmatrix}2&1+i\\1-i&0\end{pmatrix}.
	\]
	\textbf{Step 1: eigenvalues.}
	Solve $\det(A-\lambda I)=0$:
	\[
	\det\begin{pmatrix}2-\lambda&1+i\\1-i&-\lambda\end{pmatrix}
	=(2-\lambda)(-\lambda)-(1+i)(1-i).
	\]
	Compute $(1+i)(1-i)=1- i^2=2$. Hence
	\[
	-\lambda(2-\lambda)-2 = -2\lambda+\lambda^2-2.
	\]
	Set to zero:
	\[
	\lambda^2-2\lambda-2=0
	\quad\Rightarrow\quad
	\lambda=1\pm\sqrt3.
	\]
	
	\smallskip
	\noindent\textbf{Step 2: eigenvectors.}
	For $\lambda_+=1+\sqrt3$, solve $(A-\lambda_+ I)v=0$:
	\[
	\begin{pmatrix}1-\sqrt3&1+i\\1-i&-(1+\sqrt3)\end{pmatrix}
	\binom{v_1}{v_2}=0.
	\]
	From the first row: $(1-\sqrt3)v_1+(1+i)v_2=0$, so
	\[
	v_1=\frac{1+i}{\sqrt3-1}v_2.
	\]
	Choose $v_2=1$. Then $v^{(+)}=\binom{\frac{1+i}{\sqrt3-1}}{1}$.
	Normalize: set $\ket{v_+}=v^{(+)}/\|v^{(+)}\|$.
	
	Similarly for $\lambda_-=1-\sqrt3$, solve $(A-\lambda_-I)v=0$ to obtain an orthonormal eigenvector $\ket{v_-}$.
	
	\smallskip
	\noindent\textbf{Step 3: spectral decomposition.}
	Once $\{\ket{v_+},\ket{v_-}\}$ is orthonormal, you have
	\[
	A=(1+\sqrt3)\ket{v_+}\bra{v_+}+(1-\sqrt3)\ket{v_-}\bra{v_-}.
	\]
	This is the form used for exponentials and measurements.
\end{ex}

\subsubsection*{3. Visualization: eigenvectors as axes (2D picture)}
\begin{figure}[t]
	\centering
	\begin{tikzpicture}[scale=2.0, line cap=round, line join=round]
		\draw[->] (-1.2,0) -- (1.2,0) node[right] {$x$};
		\draw[->] (0,-1.2) -- (0,1.2) node[above] {$y$};
		\draw[thick] (0,0) -- (0.9,0.35) node[right] {$\ket{v_+}$};
		\draw[thick] (0,0) -- (-0.35,0.9) node[above] {$\ket{v_-}$};
		\draw[dashed] (0,0) ellipse (0.95 and 0.55);
		\node at (0,-1.35) {\small Diagonalization picks orthogonal axes (eigenvectors) for the quadratic form.};
	\end{tikzpicture}
	\caption{Eigenvectors define preferred axes: in the eigenbasis, a Hermitian matrix acts by stretching along orthogonal directions.}
	\label{fig:eigen-axes}
\end{figure}

\subsection*{Tensor products (multi-qubit systems)}

\subsubsection*{1. Definition and the ``wiring rule''}
\begin{defn}[Tensor product on basis vectors]
	For basis vectors $\ket{i}\in\C^m$ and $\ket{j}\in\C^n$,
	the tensor product space $\C^m\otimes\C^n$ has basis $\{\ket{i}\otimes\ket{j}\}$.
	For vectors $\ket{a}=\sum_i a_i\ket{i}$ and $\ket{b}=\sum_j b_j\ket{j}$:
	\[
	\ket{a}\otimes\ket{b}=\sum_{i,j} a_i b_j\,(\ket{i}\otimes\ket{j}).
	\]
\end{defn}

\begin{defn}[Kronecker product of matrices]
	If $A\in\C^{m\times m}$ and $B\in\C^{n\times n}$, define $A\otimes B\in\C^{mn\times mn}$ by blocks:
	\[
	A\otimes B=
	\begin{pmatrix}
		a_{11}B & \cdots & a_{1m}B\\
		\vdots & \ddots & \vdots\\
		a_{m1}B & \cdots & a_{mm}B
	\end{pmatrix}.
	\]
\end{defn}

\begin{prop}[Action on product states (the wiring rule)]
	\[
	(A\otimes B)(\ket{a}\otimes\ket{b})=(A\ket{a})\otimes(B\ket{b}).
	\]
\end{prop}

\begin{proof}[Basis expansion]
	Write $\ket{a}=\sum_i a_i\ket{i}$ and $\ket{b}=\sum_j b_j\ket{j}$.
	Then
	\[
	(A\otimes B)(\ket{a}\otimes\ket{b})
	=\sum_{i,j}a_i b_j\,(A\ket{i})\otimes(B\ket{j})
	=(A\sum_i a_i\ket{i})\otimes(B\sum_j b_j\ket{j})
	=(A\ket{a})\otimes(B\ket{b}).
	\]
\end{proof}

\subsubsection*{2. Concrete computation: $H\otimes H$}
\begin{ex}[Compute $H\otimes H$ explicitly]
	\[
	H=\frac{1}{\sqrt2}\begin{pmatrix}1&1\\1&-1\end{pmatrix}
	\quad\Rightarrow\quad
	H\otimes H=\frac{1}{2}
	\begin{pmatrix}
		1&1&1&1\\
		1&-1&1&-1\\
		1&1&-1&-1\\
		1&-1&-1&1
	\end{pmatrix}.
	\]
	This matrix performs a basis change on two qubits simultaneously.
\end{ex}

\subsubsection*{3. Visualization: tensor products as parallel wires}
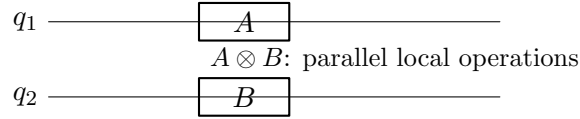
\begin{figure}[t]
	\centering
	\begin{tikzpicture}[scale=1.0, line cap=round, line join=round]
		\draw (0,0) -- (6,0);
		\draw (0,-1) -- (6,-1);
		\node[left] at (0,0) {$q_1$};
		\node[left] at (0,-1) {$q_2$};
		\draw[thick] (2,0.25) rectangle (3.2,-0.25);
		\node at (2.6,0) {$A$};
		\draw[thick] (2,-0.75) rectangle (3.2,-1.25);
		\node at (2.6,-1) {$B$};
		\node at (4.6,-0.5) {\small $A\otimes B$: parallel local operations};
	\end{tikzpicture}
	\caption{Tensor product gates correspond to parallel operations on separate wires. This ``wiring rule'' is the algebra behind parallelism.}
	\label{fig:tensor-wires}
\end{figure}

\subsection*{Density matrices and mixed states}

\subsubsection*{1. Density matrices: definition and two core invariants}
\begin{defn}[Density matrix]
	A density matrix is a matrix $\rho\in\C^{n\times n}$ such that
	\[
	\rho^\dagger=\rho,\qquad
	\rho\succeq 0,\qquad
	\Tr(\rho)=1.
	\]
\end{defn}

\begin{prop}[Pure states are rank-one density matrices]
	If $\|\psi\|=1$, then $\rho=\ket{\psi}\bra{\psi}$ is a density matrix and satisfies
	\[
	\rho^2=\rho,\qquad \Tr(\rho^2)=1.
	\]
	Conversely, if $\Tr(\rho^2)=1$ for a density matrix, then $\rho$ is pure (rank one).
\end{prop}

\begin{proof}
	We already showed $\rho^2=\rho$ for rank-one projectors, so $\Tr(\rho^2)=\Tr(\rho)=1$.
	Conversely, for density matrices one always has $\Tr(\rho^2)\le 1$ with equality iff $\rho$ has rank one,
	so $\Tr(\rho^2)=1$ characterizes purity.
\end{proof}

\subsubsection*{2. Worked example: classical mixture vs. coherent superposition}
\begin{ex}[Same populations, different physics]
	Consider two states on one qubit:
	\[
	\rho_{\text{mix}}=\frac12\ket0\bra0+\frac12\ket1\bra1
	=
	\frac12\begin{pmatrix}1&0\\0&1\end{pmatrix}
	=\frac12 I,
	\]
	and
	\[
	\ket{+}=\frac{1}{\sqrt2}(\ket0+\ket1),\qquad
	\rho_{+}=\ket+\bra+=\frac12\begin{pmatrix}1&1\\1&1\end{pmatrix}.
	\]
	Both have $p(0)=p(1)=1/2$ in the computational basis.
	But they differ by the off-diagonal terms (coherence).
	Indeed,
	\[
	\Tr(\rho_{\text{mix}}^2)=\Tr\!\left(\frac14 I\right)=\frac12,
	\qquad
	\Tr(\rho_{+}^2)=\Tr(\rho_{+})=1.
	\]
	So $\rho_{\text{mix}}$ is mixed, while $\rho_{+}$ is pure.
\end{ex}

\subsubsection*{3. (Optional but useful) Partial trace: the simplest nontrivial computation}
\begin{ex}[Tracing out one qubit of a Bell state]
	Let
	\[
	\ket{\Phi^+}=\frac{1}{\sqrt2}(\ket{00}+\ket{11}),\qquad
	\rho=\ket{\Phi^+}\bra{\Phi^+}.
	\]
	Write $\rho$ in the basis $\{\ket{00},\ket{01},\ket{10},\ket{11}\}$:
	\[
	\rho=\frac12
	\begin{pmatrix}
		1&0&0&1\\
		0&0&0&0\\
		0&0&0&0\\
		1&0&0&1
	\end{pmatrix}.
	\]
	The reduced state on the first qubit is $\rho_A=\Tr_B(\rho)$, computed by
	\[
	(\rho_A)_{i,i'}=\sum_{j\in\{0,1\}} \rho_{(i,j),(i',j)}.
	\]
	Compute entries:
	\[
	(\rho_A)_{00}=\rho_{00,00}+\rho_{01,01}=\frac12+0=\frac12,\qquad
	(\rho_A)_{11}=\rho_{10,10}+\rho_{11,11}=0+\frac12=\frac12,
	\]
	\[
	(\rho_A)_{01}=\rho_{00,10}+\rho_{01,11}=0+0=0,\qquad
	(\rho_A)_{10}=\rho_{10,00}+\rho_{11,01}=0+0=0.
	\]
	So
	\[
	\rho_A=\frac12\begin{pmatrix}1&0\\0&1\end{pmatrix}=\frac12 I.
	\]
	This is the ``entanglement $\Rightarrow$ local mixedness'' phenomenon in one calculation.
\end{ex}

\subsection*{Applications: how this chapter feeds the rest of the notes}

\subsubsection*{1. Where these tools appear next (concrete pointers)}
\begin{itemize}
	\item \textbf{Bloch sphere}: inner products, Pauli basis, rotations.
	\item \textbf{QFIM/QNG (later)}: adjoints, expectation values $\Tr(\rho O)$, tensor-product bookkeeping.
	\item \textbf{Circuits (later)}: unitarity, gate composition, Kronecker products.
	\item \textbf{QEC decoding (Track A)}: parity constraints become linear algebra over $\mathbb{F}_2$ (conceptually analogous).
	\item \textbf{Hardware pipelines}: diagonalization/spectral tools motivate stable computations; tensor products become wiring constraints.
\end{itemize}

\subsection*{Exercises}

\begin{exercise}
	Let $\ket{\psi}=\frac{1}{\sqrt{5}}\binom{1+2i}{0}$.
	Compute $\|\psi\|$, normalize it, and compute the projector $\ket{\psi}\bra{\psi}$.
\end{exercise}
\noindent\textbf{How to do it.}
Compute $\|\psi\|^2=\braket{\psi}{\psi}=\frac{1}{5}|1+2i|^2=\frac{1}{5}(1^2+2^2)=1$,
so it is already normalized. Then
\[
\ket{\psi}\bra{\psi}
=\frac{1}{5}
\begin{pmatrix}
	1+2i\\0
\end{pmatrix}
\begin{pmatrix}
	1-2i & 0
\end{pmatrix}
=
\frac{1}{5}
\begin{pmatrix}
	5 & 0\\
	0 & 0
\end{pmatrix}
=
\begin{pmatrix}
	1&0\\0&0
\end{pmatrix}.
\]

\begin{exercise}
	Show that if $U$ is unitary, then $\ket{\psi}$ is normalized iff $U\ket{\psi}$ is normalized.
\end{exercise}
\noindent\textbf{How to do it.}
Use $\|U\psi\|^2=\psi^\dagger U^\dagger U\psi=\psi^\dagger\psi=\|\psi\|^2$.

\begin{exercise}
	Check that the matrix
	\[
	U=\frac{1}{\sqrt2}\begin{pmatrix}1&i\\ i&1\end{pmatrix}
	\]
	is unitary. Then compute $U\ket0$ and $U\ket1$.
\end{exercise}
\noindent\textbf{How to do it.}
Compute $U^\dagger=\frac{1}{\sqrt2}\begin{pmatrix}1&-i\\ -i&1\end{pmatrix}$ and verify $U^\dagger U=I$.
Then multiply:
\[
U\ket0=\frac{1}{\sqrt2}\binom{1}{i},\qquad
U\ket1=\frac{1}{\sqrt2}\binom{i}{1}.
\]

\begin{exercise}
	Diagonalize $Z=\begin{pmatrix}1&0\\0&-1\end{pmatrix}$ and write its spectral decomposition.
\end{exercise}
\noindent\textbf{How to do it.}
Eigenpairs: $(+1,\ket0)$ and $(-1,\ket1)$. Hence
\[
Z = 1\cdot \ket0\bra0 + (-1)\cdot \ket1\bra1.
\]

\begin{exercise}
	Compute $(A\otimes B)(\ket{0}\otimes\ket{1})$ where
	$A=X=\begin{pmatrix}0&1\\1&0\end{pmatrix}$ and $B=Z=\begin{pmatrix}1&0\\0&-1\end{pmatrix}$.
\end{exercise}
\noindent\textbf{How to do it.}
Use the wiring rule:
\[
(A\otimes B)(\ket0\otimes\ket1)=(A\ket0)\otimes(B\ket1)=\ket1\otimes(-\ket1)=-\ket{11}.
\]

\begin{exercise}
	Let $\rho=\frac12 I$ and $\sigma=\ket+\bra+$ with $\ket+=\frac{1}{\sqrt2}(\ket0+\ket1)$.
	Compute $\Tr(\rho^2)$ and $\Tr(\sigma^2)$ and interpret.
\end{exercise}
\noindent\textbf{How to do it.}
$\rho^2=\frac14 I$ so $\Tr(\rho^2)=\frac12$ (mixed).
$\sigma^2=\sigma$ so $\Tr(\sigma^2)=1$ (pure).

\begin{exercise}[Optional: partial trace]
	Compute $\Tr_B(\ket{00}\bra{00})$ and $\Tr_B(\ket{01}\bra{01})$ explicitly.
\end{exercise}
\noindent\textbf{How to do it.}
Use $(\rho_A)_{i,i'}=\sum_{j} \rho_{(i,j),(i',j)}$.
For $\ket{00}\bra{00}$ you get $\ket0\bra0$.
For $\ket{01}\bra{01}$ you also get $\ket0\bra0$ (the first qubit is still $\ket0$).

\begin{exercise}[Optional: implementation-flavored]
	Explain why matrix multiplication order matters in circuits: give a $2\times2$ example where $AB\neq BA$.
\end{exercise}
\noindent\textbf{How to do it.}
Pick $A=X$ and $B=Z$. Compute
\[
XZ=\begin{pmatrix}0&-1\\1&0\end{pmatrix},\qquad
ZX=\begin{pmatrix}0&1\\-1&0\end{pmatrix},
\]
so $XZ=-ZX\neq ZX$.


\section{Physical Platforms Overview (Why Control Hardware Matters)}
\label{sec:platforms}

\subsection*{Objective}

This section explains \emph{why} quantum computing inevitably becomes a \textbf{control-hardware problem}.
Across essentially all platforms, the story is the same:
\[
\text{Quantum device (Hamiltonian dynamics)}
\quad+\quad
\text{control pulses (classical waveforms)}
\quad+\quad
\text{measurement (analog) $\to$ bits}
\quad+\quad
\text{real-time classical processing (feedback / decoding)}.
\]
We will treat the QPU as a physical system whose evolution is governed by a Hamiltonian,
and treat ``circuits'' as \emph{idealized descriptions} of what control hardware tries to implement.

\medskip
\noindent By the end of this section, you should be able to:
\begin{itemize}
	\item explain the universal control model ``Hamiltonian + pulses'' and how it implements gates,
	\item describe what makes superconducting and trapped-ion platforms different in timing and control,
	\item sketch a realistic classical control stack (AWGs, DAC/ADC, FPGA, host),
	\item reason with a minimal latency model using concrete symbols (cycle time, decode time, jitter, p99/p999),
	\item identify which computations belong near the device (FPGA/ASIC) versus offline (CPU/GPU).
\end{itemize}

\subsection*{Key takeaways}

\begin{enumerate}
	\item \textbf{Every platform is a control system.}
	Quantum ``gates'' are implemented by precisely shaped classical pulses; calibration is continuous.
	\item \textbf{Measurement produces a stream.}
	You do not read out a wavefunction; you read analog signals that are thresholded/classified into bits.
	\item \textbf{Real-time matters.}
	For feedback, active reset, and QEC, the computation must complete within a strict cycle budget.
	Average runtime is not enough; p99/p999 and jitter are correctness-relevant.
	\item \textbf{Control hardware is part of the algorithm.}
	If the control loop misses deadlines, the effective algorithm changes (because the physical system keeps evolving and decohering).
	\item \textbf{FPGA/ASIC is the natural ``near-QPU'' compute layer.}
	It provides bounded latency, deterministic scheduling, and streaming dataflow.
\end{enumerate}

\subsection*{The universal control picture: Hamiltonian + pulses}

\subsubsection*{From ``circuit gate'' to physical evolution}

An ideal circuit applies a unitary $U$.
A physical platform instead evolves continuously:
\[
\frac{d}{dt}\ket{\psi(t)} = -i H(t)\ket{\psi(t)},
\qquad
\ket{\psi(T)} = \mathcal{T}\exp\!\left(-i\int_0^T H(t)\,dt\right)\ket{\psi(0)},
\]
where $\mathcal{T}$ denotes time-ordering.

The key modeling move is:
\[
H(t) = H_0 + \sum_{k=1}^m u_k(t)\,H_k,
\]
where
\begin{itemize}
	\item $H_0$ is the \textbf{drift} Hamiltonian (always on),
	\item $H_k$ are \textbf{control} Hamiltonians (how fields couple to the qubits),
	\item $u_k(t)$ are classical control amplitudes (pulses) produced by AWGs/DACs.
\end{itemize}

\subsubsection*{Why calibration is never ``done''}

If $H_0$ drifts or $u_k(t)$ is distorted (bandwidth limits, phase noise, crosstalk),
the implemented unitary is not the intended gate.
Thus a core engineering loop is:
\[
\text{choose pulse parameters} \to \text{run experiment} \to \text{estimate error} \to \text{update}.
\]
This is already an online optimization problem, and it becomes \emph{real-time} when feedback is inside the experiment.

\subsubsection*{Visualization: gate as a pulse-driven trajectory}

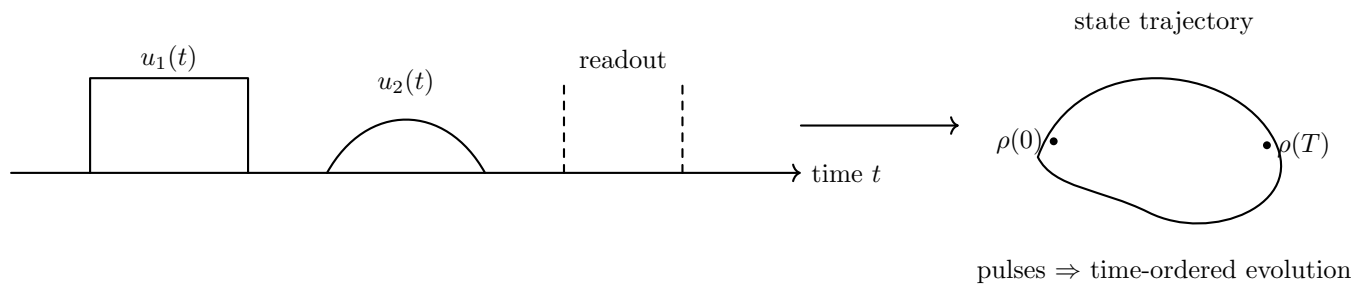
\begin{figure}[t]
	\centering
	\begin{tikzpicture}[scale=1.05, line cap=round, line join=round]
		\draw[->, thick] (0,0) -- (10,0) node[right] {\small time $t$};
		
		\draw[thick] (1,0) -- (1,1.2) -- (3,1.2) -- (3,0);
		\node at (2,1.45) {\small $u_1(t)$};
		
		\draw[thick] (4,0) .. controls (4.5,0.9) and (5.5,0.9) .. (6,0);
		\node at (5,1.15) {\small $u_2(t)$};
		
		\draw[thick, dashed] (7,0) -- (7,1.2);
		\draw[thick, dashed] (8.5,0) -- (8.5,1.2);
		\node at (7.75,1.45) {\small readout};
		
		\draw[->, thick] (10,0.6) -- (12,0.6);
		
		\draw[thick] (13,0.2) .. controls (13.5,1.6) and (15.5,1.4) .. (16,0.4)
		.. controls (16.4,-0.4) and (15.2,-0.9) .. (14.4,-0.5)
		.. controls (13.8,-0.2) and (13.2,-0.2) .. (13,0.2);
		\node at (14.6,1.9) {\small state trajectory};
		
		\fill (13.2,0.4) circle (0.05) node[left] {\small $\rho(0)$};
		\fill (15.9,0.35) circle (0.05) node[right] {\small $\rho(T)$};
		
		\node at (14.6,-1.25) {\small pulses $\Rightarrow$ time-ordered evolution};
	\end{tikzpicture}
	\caption{Universal control picture: the control hardware outputs analog waveforms $u_k(t)$,
		which drive the Hamiltonian $H(t)$ and produce a state trajectory; measurement occurs in a readout window.}
	\label{fig:platforms-pulse-to-path}
\end{figure}

\subsection*{Platform I: Superconducting qubits (transmons)}

\subsubsection*{What a transmon is (one paragraph model)}
A transmon is a superconducting circuit behaving like an anharmonic oscillator.
You encode $\ket{0},\ket{1}$ in its lowest two energy levels.
Because the oscillator is only weakly anharmonic, \textbf{leakage} to $\ket{2}$ is a real engineering concern.

\subsubsection*{Control and readout in practice}
\begin{itemize}
	\item \textbf{Control:} microwave pulses implement single-qubit rotations; two-qubit gates come from tunable couplers or cross-resonance-type interactions.
	\item \textbf{Readout:} dispersive measurement via a resonator; you measure an analog IQ signal and classify it into 0/1.
	\item \textbf{Timing intuition:} gates are fast (tens of ns to a few hundred ns), but coherence times are finite;
	feedback/QEC cycles are time-critical.
\end{itemize}

\subsubsection*{Why control hardware matters here}
\begin{itemize}
	\item \textbf{Fast repetition:} calibration and error characterization require many shots quickly.
	\item \textbf{Low-latency feedback:} active reset and conditional operations benefit from near-device processing.
	\item \textbf{Streaming syndrome:} surface-code experiments naturally produce a continuous stream of stabilizer outcomes.
\end{itemize}

\subsubsection*{Visualization: readout as analog-to-bit conversion}
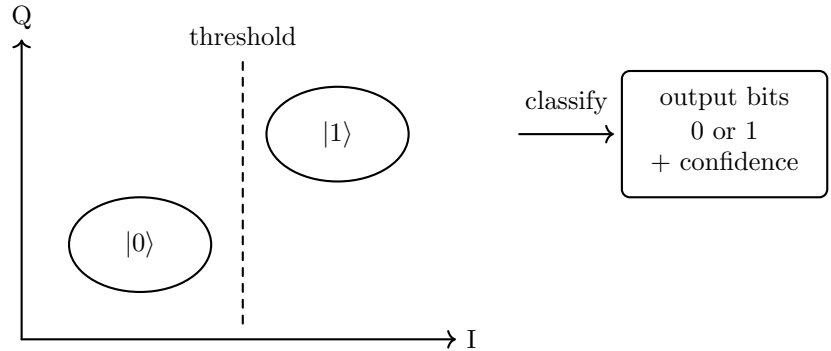
\begin{figure}[t]
	\centering
	\begin{tikzpicture}[scale=1.05, line cap=round, line join=round]
		\draw[->, thick] (0,0) -- (5.5,0) node[right] {\small I};
		\draw[->, thick] (0,0) -- (0,3.8) node[above] {\small Q};
		
		\draw[thick] (1.5,1.2) ellipse (0.9 and 0.6);
		\draw[thick] (4.0,2.6) ellipse (0.9 and 0.6);
		\node at (1.5,1.2) {\small $\ket0$};
		\node at (4.0,2.6) {\small $\ket1$};
		
		\draw[dashed, thick] (2.8,0.2) -- (2.8,3.6);
		\node[above] at (2.8,3.6) {\small threshold};
		
		\draw[->, thick] (6.3,2.6) -- (7.5,2.6);
		\node at (6.9,3.0) {\small classify};
		
		\draw[thick, rounded corners=3pt] (7.6,1.8) rectangle (10.2,3.4);
		\node at (8.9,3.05) {\small output bits};
		\node at (8.9,2.65) {\small 0 or 1};
		\node at (8.9,2.25) {\small + confidence};
		
		\node at (4.8,-1.0) {\small readout is analog; the ``bit'' is a classified estimate};
	\end{tikzpicture}
	\caption{Superconducting readout schematic: analog IQ samples cluster by state; a classifier/threshold outputs bits.
		This classification step is part of the classical pipeline and contributes latency and error.}
	\label{fig:platforms-readout}
\end{figure}

\subsection*{Platform II: Trapped-ion qubits}

\subsubsection*{Physical picture}
Trapped-ion qubits encode states in internal electronic levels of ions held in electromagnetic traps.
The ions share collective motional modes, which can mediate entangling gates.

\subsubsection*{Control and readout in practice}
\begin{itemize}
	\item \textbf{Control:} laser pulses drive single-qubit rotations and multi-qubit entangling gates via motional modes.
	\item \textbf{Readout:} state-dependent fluorescence; photon counts are thresholded into bits.
	\item \textbf{Timing intuition:} gates can be slower than superconducting gates but often have very high fidelity;
	measurement windows can be relatively long (photon counting).
\end{itemize}

\subsubsection*{Why control hardware still matters}
Even when gates are slower, the pipeline constraints remain:
\begin{itemize}
	\item repeated-shot estimation still creates data streams,
	\item feedback still benefits from deterministic scheduling,
	\item decoding/processing still becomes streaming at scale.
\end{itemize}

\subsubsection*{Visualization: photon-count thresholding}
\begin{figure}[t]
	\centering
	\begin{tikzpicture}[scale=1.0, line cap=round, line join=round]
		\draw[->, thick] (0,0) -- (10,0) node[right] {\small photon count};
		\draw[->, thick] (0,0) -- (0,4.0) node[above] {\small frequency};
		
		\draw[thick] (1.0,0.2) .. controls (2.2,2.6) and (3.0,2.6) .. (4.2,0.2);
		\draw[thick] (5.2,0.2) .. controls (6.4,3.4) and (7.2,3.4) .. (8.4,0.2);
		
		\node at (2.6,2.8) {\small dark};
		\node at (6.8,3.6) {\small bright};
		
		\draw[dashed, thick] (4.8,0) -- (4.8,3.8);
		\node[above] at (4.8,3.8) {\small threshold};
		
		\node at (2.6,-0.5) {\small output 0};
		\node at (6.8,-0.5) {\small output 1};
		
		\node at (5,-1.3) {\small fluorescence $\Rightarrow$ count distribution $\Rightarrow$ threshold $\Rightarrow$ bit};
	\end{tikzpicture}
	\caption{Trapped-ion readout schematic: photon counts form distributions; thresholding converts counts into a bit.
		Readout is again an analog-to-bit classification step with error and latency.}
	\label{fig:platforms-photon}
\end{figure}
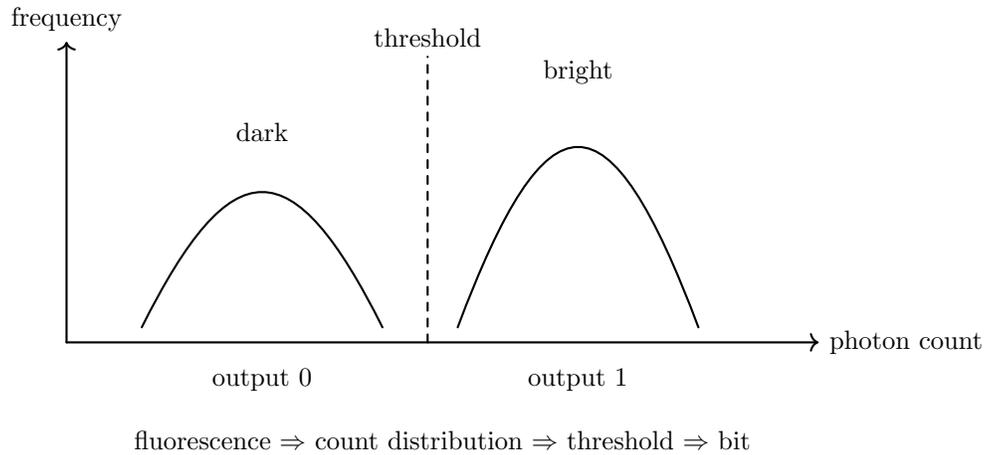

\subsection*{Other platforms (brief)}

The same architecture pattern appears across other technologies:
\begin{itemize}
	\item \textbf{Neutral atoms / Rydberg:} laser control, imaging-based readout; large arrays; timing and calibration remain central.
	\item \textbf{Photonic:} sources, interferometers, and detectors; feed-forward and switching impose strict timing constraints.
	\item \textbf{Spin qubits (semiconductor):} microwave/voltage pulses; cryogenic constraints; readout via charge sensors; low-latency feedback is valuable.
\end{itemize}
Platform specifics change, but the universal control loop does not.

\subsection*{The classical control stack}

\subsubsection*{Stack layers (conceptual)}
A practical lab stack can be summarized as:

\begin{enumerate}
	\item \textbf{Experiment logic / host (CPU):}
	compiles circuits, schedules runs, logs data, runs offline analysis.
	\item \textbf{Real-time controller (FPGA/ASIC):}
	handles deterministic timing, streaming parsing, feedback, decoding, and near-device decisions.
	\item \textbf{Waveform generation (DAC/AWG):}
	outputs analog control pulses $u_k(t)$ to the device.
	\item \textbf{Analog front-end:}
	mixers, filters, amplifiers, cryogenic chain (platform dependent).
	\item \textbf{Digitization (ADC):}
	converts readout waveforms into samples.
	\item \textbf{Readout processing:}
	demodulation/integration/classification into bits, often on FPGA.
\end{enumerate}

\subsubsection*{Visualization: control stack and data direction}
\begin{figure}[t]
	\centering
	\begin{tikzpicture}[
		font=\small,
		box/.style={
			draw, thick, rounded corners=3pt,
			align=center,
			minimum height=1.05cm,
			inner xsep=6pt, inner ysep=5pt,
			text width=3.0cm 
		},
		qpubox/.style={
			draw, thick, rounded corners=3pt,
			align=center,
			minimum height=1.05cm,
			inner xsep=6pt, inner ysep=5pt,
			text width=2.8cm
		},
		arr/.style={-Latex, thick},
		lab/.style={font=\footnotesize, inner sep=1pt}
		]
		
		\matrix (top) [matrix of nodes,
		nodes={box},
		column sep=10mm,
		row sep=10mm
		]{
			{\textbf{Host CPU}\\[-1pt]\footnotesize compile / log / analyze} &
			{\textbf{FPGA/ASIC}\\[-1pt]\footnotesize timing / decode / feedback} &
			{\textbf{DAC/AWG}\\[-1pt]\footnotesize pulses $u_k(t)$} &
			|[qpubox]| {\textbf{QPU}\\[-1pt]\footnotesize Hamiltonian evolution} \\
		};
		
		\node[box, text width=3.0cm] (adc) at ($(top-1-3.south)+(0,-1.55cm)$)
		{\textbf{ADC}\\[-1pt]\footnotesize samples};
		
		\draw[arr] (top-1-1.east) -- node[lab, above] {circuits / params} (top-1-2.west);
		\draw[arr] (top-1-2.east) -- node[lab, above] {waveforms}        (top-1-3.west);
		
		\draw[arr] (top-1-3.east) -- node[lab, above] {control} (top-1-4.west);
		
		\draw[arr] (top-1-4.west) ++(0,-2mm) |- node[lab, near end, right] {readout} (adc.east);
		
		\draw[arr] (adc.west) -- node[lab, below] {bits / samples} ($(top-1-2.south)+(0,-1.55cm)$);
		
		\draw[arr] ($(top-1-2.south)+(-6mm,0)$) -- node[lab, below] {logs / results} ($(top-1-1.south)+(6mm,0)$);
		
		\node[align=center, font=\footnotesize, text width=0.92\textwidth]
		at ($(top.south)+(0,-2.55cm)$)
		{near-QPU real-time path is \textbf{FPGA} $\leftrightarrow$ (DAC/ADC) $\leftrightarrow$ QPU};
		
	\end{tikzpicture}
	\caption{Classical control stack schematic. Real-time tasks (readout processing, feedback, decoding)
		sit naturally on FPGA/ASIC between host software and analog hardware.}
	\label{fig:platforms-stack}
\end{figure}

\subsection*{Minimal latency model}

\subsubsection*{Symbols (use these throughout the book)}
We will model a single ``cycle'' (one measurement-and-decision unit) with the following times:
\begin{itemize}
	\item $T_{\mathrm{meas}}$: physical measurement window (integration / photon counting / etc.)
	\item $T_{\mathrm{adc}}$: digitization and transfer time
	\item $T_{\mathrm{ro}}$: readout processing time (demodulate / integrate / classify)
	\item $T_{\mathrm{dec}}$: decoding/estimation/decision time (Union--Find, QNG update, etc.)
	\item $T_{\mathrm{cmd}}$: time to send the next command / correction to the control layer
	\item $T_{\mathrm{margin}}$: safety margin (clock domain crossings, buffering, worst-case slack)
\end{itemize}
Define the \textbf{cycle budget} (deadline) $T_{\mathrm{cycle}}$ and the \textbf{total latency}
\[
T_{\mathrm{tot}} := T_{\mathrm{meas}} + T_{\mathrm{adc}} + T_{\mathrm{ro}} + T_{\mathrm{dec}} + T_{\mathrm{cmd}} + T_{\mathrm{margin}}.
\]
A real-time loop requires
\[
T_{\mathrm{tot}} \le T_{\mathrm{cycle}}.
\]

\subsubsection*{Tail latency (p99/p999) belongs in the inequality}
In practice $T_{\mathrm{ro}}$ and $T_{\mathrm{dec}}$ are random variables under software scheduling,
I/O contention, and caching effects.
Therefore the relevant inequality is a tail constraint:
\[
\mathrm{p99}(T_{\mathrm{tot}}) \le T_{\mathrm{cycle}}
\quad\text{or even}\quad
\mathrm{p999}(T_{\mathrm{tot}}) \le T_{\mathrm{cycle}},
\]
depending on the required reliability.

\subsubsection*{Backlog stability (streaming view)}
If syndrome/measurement events arrive every $T_{\mathrm{cycle}}$ but processing occasionally takes longer,
a backlog accumulates in a FIFO.
A minimal stability condition is that the long-run average processing rate exceeds the arrival rate,
and that tail spikes are bounded so the FIFO does not overflow.
This is why ``bounded passes'' and deterministic pipelines are emphasized in Track A.

\subsubsection*{Visualization: cycle budget and deadlines}
\begin{figure}[t]
	\centering
	\begin{tikzpicture}[scale=1.05, line cap=round, line join=round]
		\draw[->, thick] (0,0) -- (12,0) node[right] {\small time};
		
		\draw[thick] (0,0) -- (2,0);
		\draw[thick] (0,0.7) rectangle (2,1.3);
		\node at (1,1.0) {\small $T_{\mathrm{meas}}$};
		
		\draw[thick] (2,0.7) rectangle (3.2,1.3);
		\node at (2.6,1.0) {\small $T_{\mathrm{adc}}$};
		
		\draw[thick] (3.2,0.7) rectangle (5.4,1.3);
		\node at (4.3,1.0) {\small $T_{\mathrm{ro}}$};
		
		\draw[thick] (5.4,0.7) rectangle (7.8,1.3);
		\node at (6.6,1.0) {\small $T_{\mathrm{dec}}$};
		
		\draw[thick] (7.8,0.7) rectangle (9.0,1.3);
		\node at (8.4,1.0) {\small $T_{\mathrm{cmd}}$};
		
		\draw[thick] (9.0,0.7) rectangle (10.2,1.3);
		\node at (9.6,1.0) {\small margin};
		
		\draw[dashed, thick] (10.2,0.2) -- (10.2,1.8);
		\node[above] at (10.2,1.8) {\small deadline $T_{\mathrm{cycle}}$};
		
		\draw[->, thick] (10.2,1.5) -- (11.7,1.5);
		\node at (11.0,1.75) {\small next cycle};
		
		\node at (6.0,-0.8) {\small real-time condition: total latency must fit before the deadline};
	\end{tikzpicture}
	\caption{Minimal latency model: measurement and classical processing must complete within a fixed cycle budget.
		For QEC/feedback, missing the deadline can cause physical failure.}
	\label{fig:platforms-latency}
\end{figure}

\subsection*{Exercises with explanations}

\begin{exercise}[Hamiltonian + pulse model]
	Suppose $H(t)=H_0+u(t)H_1$ with a piecewise-constant pulse
	$u(t)=u_0$ for $0\le t\le T$.
	Write the unitary $U(T)$ explicitly.
\end{exercise}
\noindent\textbf{Explanation / solution.}
If $u(t)$ is constant, then $H(t)$ is constant:
\[
H = H_0 + u_0 H_1.
\]
Hence the evolution is
\[
U(T)=\exp(-iHT)=\exp\!\bigl(-i(H_0+u_0H_1)T\bigr).
\]
(If $H_0$ and $H_1$ do not commute, this formula still holds because $H$ is constant; noncommutativity becomes an issue when $u(t)$ varies with time and the Hamiltonian changes.)

\begin{exercise}[Time-ordering intuition]
	Let $u(t)$ switch between two values $u_a$ on $[0,T/2]$ and $u_b$ on $[T/2,T]$.
	Express $U(T)$ as a product of exponentials.
\end{exercise}
\noindent\textbf{Explanation / solution.}
On each interval the Hamiltonian is constant:
\[
H_a=H_0+u_aH_1,\qquad H_b=H_0+u_bH_1.
\]
Thus
\[
U(T)=\exp\!\left(-iH_b\frac{T}{2}\right)\exp\!\left(-iH_a\frac{T}{2}\right),
\]
with later time on the left (time-ordering).

\begin{exercise}[Readout is classification]
	Explain (in 2--3 sentences) why measurement error can occur even if the underlying quantum projection is perfect.
\end{exercise}
\noindent\textbf{Explanation / solution.}
The physical measurement produces an analog signal (IQ samples or photon counts).
To output a bit, the classical system must classify that signal using thresholds or a learned discriminator.
Noise in the analog chain and overlap of the signal distributions cause misclassification, producing readout error.

\begin{exercise}[Latency inequality]
	Given a cycle budget $T_{\mathrm{cycle}}=1$ (in normalized units),
	suppose $T_{\mathrm{meas}}=0.35$, $T_{\mathrm{adc}}=0.10$, $T_{\mathrm{ro}}=0.25$,
	$T_{\mathrm{dec}}=0.20$, $T_{\mathrm{cmd}}=0.05$.
	How much margin remains? What happens if $T_{\mathrm{dec}}$ occasionally spikes to $0.35$?
\end{exercise}
\noindent\textbf{Explanation / solution.}
Compute
\[
T_{\mathrm{tot}}=0.35+0.10+0.25+0.20+0.05=0.95,
\]
so margin is $0.05$.
If $T_{\mathrm{dec}}$ spikes to $0.35$, then
\[
T_{\mathrm{tot}}=0.35+0.10+0.25+0.35+0.05=1.10>1,
\]
so the deadline is missed; in a real-time loop this can cause backlog growth or incorrect/tardy correction.

\begin{exercise}[p99 vs.\ mean]
	Give one concrete reason why a controller with mean latency $0.2$ could still be unacceptable,
	even if $T_{\mathrm{cycle}}=1$.
\end{exercise}
\noindent\textbf{Explanation / solution.}
If rare events produce tail spikes (e.g.\ p99 or p999 near or above $1$),
then even though the mean is small, occasional deadline misses can destabilize the loop (FIFO overflow, stale corrections).
Real-time correctness depends on tail behavior, not just the mean.

\begin{exercise}[Where does FPGA belong?]
	For each of the following tasks, label it as ``near-QPU'' (FPGA/ASIC) or ``host/offline'' (CPU/GPU),
	and give one sentence of justification:
	(i) readout demodulation/integration,
	(ii) QEC decoding,
	(iii) large-scale parameter sweep to fit a noise model,
	(iv) logging and plotting.
\end{exercise}
\noindent\textbf{Explanation / solution.}
(i) near-QPU: must run every shot with bounded latency.
(ii) near-QPU (for real-time QEC): deadline critical and streaming.
(iii) host/offline: computationally heavy but not per-cycle deadline critical.
(iv) host/offline: not latency-critical; better handled in software.

	
\section{Single-Qubit States and the Bloch Sphere}
\label{sec:bloch}

\subsection*{Objective (geometry-first, algorithm-first)}

A single qubit is the smallest nontrivial quantum system, yet it already contains
the full set of ``quantum'' phenomena that drive algorithms and hardware constraints:
\emph{superposition}, \emph{interference}, and \emph{phase sensitivity}.
This section builds a geometric model of single-qubit states and uses it as a
working language for computation:

\begin{quote}
	\emph{prepare a state on the Bloch sphere $\to$ rotate it $\to$ measure along an axis $\to$ convert analog readout to bits.}
\end{quote}

\noindent The main message is that many single-qubit procedures are best understood as
\textbf{explicit rotations of a point on a sphere}.

\medskip
\noindent\textbf{Implementation-aware note.}
In hardware, gates are implemented by control pulses and readout returns classical bits (after analog processing).
Drift, dephasing, miscalibration (axis tilt, over/under-rotation), and readout error deform the ideal picture.
A Bloch-sphere viewpoint makes sensitivity and robustness visually obvious and gives a direct bridge to
calibration, feedback, and near-device control (often on FPGA/ASIC).

\medskip
\noindent By the end of this section, you should be able to:
\begin{itemize}
	\item explain why pure qubit states form $\CP^1\simeq S^2$ (projective, not linear),
	\item compute Bloch coordinates from a state vector or density matrix,
	\item interpret common gates ($X,Y,Z,H,R_Z,R_X,R_Y$) as rotations (with explicit formulas),
	\item predict measurement statistics from geometry (and vice versa),
	\item connect ideal rotations to practical control/readout steps.
\end{itemize}

\subsection*{Key Concepts}

\subsubsection*{1. Pure states, global phase, and why the state space is a sphere}

\begin{defn}[Pure single-qubit state]
	A pure single-qubit state is a unit vector
	\[
	\ket{\psi}=\alpha\ket{0}+\beta\ket{1}\in\C^2,
	\qquad |\alpha|^2+|\beta|^2=1.
	\]
\end{defn}

\begin{prop}[Global phase invariance]
	The vectors $\ket{\psi}$ and $e^{i\theta}\ket{\psi}$ represent the same physical state
	for any $\theta\in\R$.
\end{prop}

\begin{proof}
	Measurement probabilities in any orthonormal basis $\{\ket{x}\}$ are
	\[
	p(x)=|\braket{x}{\psi}|^2.
	\]
	For $\ket{\psi'}=e^{i\theta}\ket{\psi}$,
	\[
	|\braket{x}{\psi'}|^2
	=
	|\braket{x}{e^{i\theta}\psi}|^2
	=
	|e^{i\theta}\braket{x}{\psi}|^2
	=
	|\braket{x}{\psi}|^2.
	\]
	So all measurement statistics agree.
\end{proof}

\begin{rem}[Algorithmic meaning]
	Only \emph{relative} phase affects interference; global phase does not.
	Hence the true state space of pure states is projective.
\end{rem}

\begin{defn}[Complex projective line]
	\[
	\CP^1 := (\C^2\setminus\{0\})/\!\sim,
	\qquad v\sim \lambda v \text{ for } \lambda\in\C^\times.
	\]
\end{defn}

\begin{rem}[Degrees of freedom]
	$(\alpha,\beta)\in\C^2$ has $4$ real parameters.
	Normalization removes $1$, global phase removes $1$.
	Thus pure qubit states have $2$ real degrees of freedom, matching $S^2$.
\end{rem}

\subsubsection*{2. Pauli matrices, density matrices, and the Bloch map}

\begin{defn}[Pauli matrices]
	\[
	X=\begin{pmatrix}0&1\\1&0\end{pmatrix},\quad
	Y=\begin{pmatrix}0&-i\\ i&0\end{pmatrix},\quad
	Z=\begin{pmatrix}1&0\\0&-1\end{pmatrix}.
	\]
\end{defn}

\begin{defn}[Density matrix]
	A (single-qubit) density matrix is a $2\times2$ matrix $\rho$ such that:
	\[
	\rho^\dagger=\rho,\qquad \rho\succeq 0,\qquad \Tr(\rho)=1.
	\]
	A pure state $\ket{\psi}$ corresponds to $\rho=\ket{\psi}\bra{\psi}$.
\end{defn}

\begin{defn}[Bloch vector]
	For any single-qubit state $\rho$, define
	\[
	x=\Tr(\rho X),\qquad y=\Tr(\rho Y),\qquad z=\Tr(\rho Z).
	\]
	The vector $\vec r=(x,y,z)\in\R^3$ is the \emph{Bloch vector}.
\end{defn}

\begin{rem}[Operational meaning]
	Each coordinate is an expectation of a $\pm1$-valued measurement:
	\[
	x=\mathbb E[X\text{-outcome}],\quad
	y=\mathbb E[Y\text{-outcome}],\quad
	z=\mathbb E[Z\text{-outcome}].
	\]
	In experiments each is estimated by repeated shots and averaging classical readout bits.
\end{rem}

\begin{prop}[Bloch expansion (any qubit state)]
	For any $2\times2$ density matrix $\rho$ with Bloch vector $\vec r=(x,y,z)$,
	\[
	\rho=\frac12\bigl(I+xX+yY+zZ\bigr).
	\]
\end{prop}

\begin{proof}
	The Hermitian $2\times2$ matrices form a $4$-dimensional real vector space with basis $\{I,X,Y,Z\}$.
	Write $\rho=aI+bX+cY+dZ$.
	From $\Tr(\rho)=1$ we get $2a=1$ so $a=\frac12$.
	Using $\frac12\Tr(\sigma_i\sigma_j)=\delta_{ij}$ and $\Tr(\sigma_i)=0$,
	\[
	\Tr(\rho X)=2b,\quad \Tr(\rho Y)=2c,\quad \Tr(\rho Z)=2d,
	\]
	so $b=\frac{x}{2}$, $c=\frac{y}{2}$, $d=\frac{z}{2}$.
\end{proof}

\begin{prop}[Bloch ball and purity]
	For any qubit state $\rho$ with Bloch vector $\vec r$,
	\[
	\Tr(\rho^2)=\frac12\bigl(1+\|\vec r\|^2\bigr).
	\]
	In particular, $\rho$ is pure $\iff \|\vec r\|=1$, and mixed states satisfy $\|\vec r\|<1$.
\end{prop}

\begin{proof}
	Write $\rho=\frac12(I+\vec r\cdot\vec\sigma)$.
	Using $(\vec r\cdot\vec\sigma)^2=\|\vec r\|^2 I$ and $\Tr(\sigma_i)=0$,
	\[
	\rho^2=\frac14\bigl(I+2\vec r\cdot\vec\sigma+\|\vec r\|^2I\bigr)
	\quad\Rightarrow\quad
	\Tr(\rho^2)=\frac14\bigl(2+0+2\|\vec r\|^2\bigr)=\frac12(1+\|\vec r\|^2).
	\]
\end{proof}

\begin{prop}[Bloch sphere identification $\CP^1\simeq S^2$]
	For pure states $\rho=\ket{\psi}\bra{\psi}$, the Bloch vector satisfies $\|\vec r\|=1$,
	and the map $\CP^1\to S^2$, $[\psi]\mapsto \vec r$, is a diffeomorphism.
\end{prop}

\begin{proof}[Concrete construction]
	Let $\ket{\psi}=\begin{psmallmatrix}\alpha\\ \beta\end{psmallmatrix}$ with $|\alpha|^2+|\beta|^2=1$.
	Then
	\[
	x=2\Re(\alpha\overline{\beta}),\quad
	y=2\Im(\alpha\overline{\beta}),\quad
	z=|\alpha|^2-|\beta|^2,
	\]
	and a direct computation gives $x^2+y^2+z^2=1$.
	Conversely, given $(x,y,z)\in S^2$, choose angles $\theta,\phi$ with
	$z=\cos\theta$, $x=\sin\theta\cos\phi$, $y=\sin\theta\sin\phi$ and define
	\[
	\ket{\psi(\theta,\phi)}=\cos(\theta/2)\ket0+e^{i\phi}\sin(\theta/2)\ket1,
	\]
	which has Bloch vector $(x,y,z)$.
\end{proof}

\begin{figure}[t]
	\centering
	\begin{tikzpicture}[scale=2.2, line cap=round, line join=round]
		\draw (0,0) circle (1);
		\draw[dashed] (-1,0) arc (180:360:1 and 0.35);
		\draw (-1,0) arc (180:0:1 and 0.35);
		
		\draw[->] (0,0) -- (1.25,0) node[right] {$x$};
		\draw[->] (0,0) -- (0,1.25) node[above] {$z$};
		\draw[->] (0,0) -- (-0.75,-0.55) node[left] {$y$};
		
		\fill (0,1) circle (0.03) node[above] {$\ket{0}$};
		\fill (0,-1) circle (0.03) node[below] {$\ket{1}$};
		
		\fill (1,0) circle (0.03) node[right] {$\ket{+}$};
		\fill (-1,0) circle (0.03) node[left] {$\ket{-}$};
		
		\coordinate (P) at (0.55,0.35);
		\fill (P) circle (0.03) node[right] {$\ket{\psi(\theta,\phi)}$};
		\draw[thick] (0,0) -- (P);
		\node at (0.25,0.2) {$\vec r$};
	\end{tikzpicture}
	\caption{Bloch sphere: pure qubit states form $S^2$.
		The Bloch vector $\vec r=(x,y,z)$ encodes expectation values of $(X,Y,Z)$.}
	\label{fig:bloch-sphere}
\end{figure}
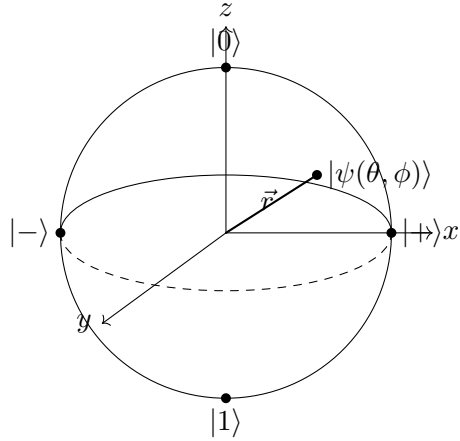

\begin{figure}[t]
	\centering
	\begin{tikzpicture}[scale=2.2, line cap=round, line join=round]
		\draw (0,0) circle (1);
		\draw[dashed] (-1,0) arc (180:360:1 and 0.35);
		\draw (-1,0) arc (180:0:1 and 0.35);
		
		\draw[thick] (0,0) circle (0.55);
		\node at (0.78,0.52) {\small mixed};
		
		\draw[->] (0,0) -- (1.25,0) node[right] {$x$};
		\draw[->] (0,0) -- (0,1.25) node[above] {$z$};
		\draw[->] (0,0) -- (-0.75,-0.55) node[left] {$y$};
		
		\node at (0,-1.2) {\small $\|\vec r\|=1$ pure (surface), $\ \|\vec r\|<1$ mixed (interior)};
	\end{tikzpicture}
	\caption{Bloch ball: mixed states fill the interior. Purity is detected by $\Tr(\rho^2)=\frac12(1+\|\vec r\|^2)$.}
	\label{fig:bloch-ball}
\end{figure}

\subsubsection*{3. Unitary gates are rotations (SU(2) $\to$ SO(3))}

\begin{defn}[Single-qubit gate action]
	A single-qubit gate is a unitary $U\in\SU(2)$ acting on density matrices by
	\[
	\rho \longmapsto U\rho U^\dagger.
	\]
\end{defn}

\begin{prop}[Conjugation preserves states]
	If $\rho$ is a density matrix, then $U\rho U^\dagger$ is also a density matrix.
\end{prop}

\begin{proof}
	Hermiticity is preserved since $(U\rho U^\dagger)^\dagger=U\rho U^\dagger$.
	Trace is preserved:
	$\Tr(U\rho U^\dagger)=\Tr(\rho U^\dagger U)=\Tr(\rho)=1$.
	Positivity is preserved because $\bra{v}U\rho U^\dagger\ket{v}=\bra{U^\dagger v}\rho\ket{U^\dagger v}\ge0$.
\end{proof}

\begin{prop}[Rotation action on Bloch vectors]
	Let $\rho=\frac12(I+\vec r\cdot\vec\sigma)$ and $\rho'=U\rho U^\dagger$.
	Then there exists a rotation $R(U)\in\SO(3)$ such that
	\[
	\vec r' = R(U)\,\vec r.
	\]
	Moreover $R(U)=R(-U)$ (double cover).
\end{prop}

\begin{proof}
	Define $R_{ij}:=\frac12\Tr(\sigma_i U\sigma_j U^\dagger)$.
	Then $U\sigma_jU^\dagger=\sum_i R_{ij}\sigma_i$ and
	\[
	\rho'=\frac12\Bigl(I+\sum_j r_j U\sigma_jU^\dagger\Bigr)
	=\frac12\Bigl(I+\sum_i (\sum_j R_{ij}r_j)\sigma_i\Bigr),
	\]
	so $\vec r'=R(U)\vec r$.
	Conjugation preserves the Hilbert--Schmidt inner product, hence $R(U)^TR(U)=I$.
	Continuity and connectedness imply $\det R(U)=+1$, so $R(U)\in\SO(3)$.
	Finally $(-U)\sigma_j(-U)^\dagger=U\sigma_jU^\dagger$, so $R(-U)=R(U)$.
\end{proof}

\begin{rem}[Practical computation]
	To find $R(U)$ explicitly, compute $UXU^\dagger, UYU^\dagger, UZU^\dagger$ and express each in the Pauli basis.
\end{rem}

\subsubsection*{4. Standard rotation gates and their geometric meaning}

\begin{defn}[Axis rotations]
	For $\theta\in\R$ define
	\[
	R_X(\theta)=e^{-i\frac{\theta}{2}X},\qquad
	R_Y(\theta)=e^{-i\frac{\theta}{2}Y},\qquad
	R_Z(\theta)=e^{-i\frac{\theta}{2}Z}.
	\]
\end{defn}

\begin{prop}[Closed forms]
	Using $X^2=Y^2=Z^2=I$,
	\[
	R_\sigma(\theta)=\cos(\tfrac{\theta}{2})I-i\sin(\tfrac{\theta}{2})\,\sigma
	\quad (\sigma\in\{X,Y,Z\}).
	\]
	In particular,
	\[
	R_Z(\theta)=
	\begin{pmatrix}
		e^{-i\theta/2} & 0\\
		0 & e^{i\theta/2}
	\end{pmatrix}
	\sim
	\begin{pmatrix}
		1 & 0\\
		0 & e^{i\theta}
	\end{pmatrix}
	\quad(\text{equal up to global phase}).
	\]
\end{prop}

\begin{proof}
	Expand the exponential using $\sigma^2=I$:
	odd powers contribute $\sigma$, even powers contribute $I$.
\end{proof}

\begin{prop}[$Z$-rotation rotates the equator]
	Under the Bloch map, $R_Z(\theta)$ rotates $(x,y)$ by angle $\theta$ and leaves $z$ unchanged:
	\[
	(x,y,z)\mapsto (x\cos\theta-y\sin\theta,\ x\sin\theta+y\cos\theta,\ z).
	\]
\end{prop}

\begin{proof}
	Compute $R_Z(\theta)XR_Z(\theta)^\dagger = (\cos\theta)X+(\sin\theta)Y$ and
	$R_Z(\theta)YR_Z(\theta)^\dagger = (-\sin\theta)X+(\cos\theta)Y$, while $Z$ commutes with $R_Z(\theta)$.
	Substitute into $\rho=\frac12(I+xX+yY+zZ)$ and read off coefficients.
\end{proof}

\begin{figure}[t]
	\centering
	\begin{tikzpicture}[scale=2.3, line cap=round, line join=round]
		\draw (0,0) circle (1);
		\draw[dashed] (-1,0) arc (180:360:1 and 0.35);
		\draw (-1,0) arc (180:0:1 and 0.35);
		
		\draw[->] (0,0) -- (1.25,0) node[right] {$x$};
		\draw[->] (0,0) -- (0,1.25) node[above] {$z$};
		
		\draw[->, thick] (0.55,0.0) arc (0:60:0.55 and 0.2);
		\node at (0.58,0.28) {\small $R_Z(\theta)$};
		
		\fill (0.75,0) circle (0.03) node[right] {\small before};
		\fill (0.40,0.20) circle (0.03) node[above right] {\small after};
		
		\node at (0,-1.25) {\small $R_Z(\theta)$ rotates about $z$: phase motion becomes visible after basis change};
	\end{tikzpicture}
	\caption{$Z$-axis rotations move states along the equator (changing relative phase).}
	\label{fig:bloch-rz}
\end{figure}

\subsection*{Concrete algorithms on the Bloch sphere}

\subsubsection*{Algorithm template (used repeatedly)}
A typical single-qubit procedure is:
\begin{enumerate}
	\item \textbf{Prepare} a reference state (often $\ket0$).
	\item \textbf{Rotate} by a gate sequence (a path on the sphere).
	\item \textbf{Measure} along an axis (often $Z$ in hardware).
	\item \textbf{Post-process} classically (average bits, estimate expectations, update parameters).
\end{enumerate}
Only the unitary evolution is quantum; everything else is classical control and estimation.

\subsubsection*{Algorithm 1: Basis change (turn $Z$ readout into $X$ readout)}

\begin{defn}[Hadamard gate]
	\[
	H=\frac{1}{\sqrt{2}}
	\begin{pmatrix}
		1 & 1\\
		1 & -1
	\end{pmatrix}.
	\]
\end{defn}

\begin{prop}[Hadamard swaps axes]
	\[
	HZH = X,\qquad HXH = Z,\qquad HYH = -Y.
	\]
\end{prop}

\begin{proof}
	Direct multiplication gives $HZH=X$. The others follow similarly.
\end{proof}

\begin{prop}[Measurement equivalence]
	Measuring $Z$ after applying $H$ is equivalent to measuring $X$ before applying $H$:
	\[
	\Pr(Z=\pm1 \text{ on } H\ket{\psi}) = \Pr(X=\pm1 \text{ on } \ket{\psi}).
	\]
\end{prop}

\begin{proof}
	In general, measuring observable $A$ on $U\ket{\psi}$ is equivalent to measuring $U^\dagger A U$ on $\ket{\psi}$.
	Here $U=H$ and $A=Z$, so $H^\dagger Z H = X$.
\end{proof}

\begin{figure}[t]
	\centering
	\begin{tikzpicture}[scale=2.2, line cap=round, line join=round]
		\draw (0,0) circle (1);
		\draw[dashed] (-1,0) arc (180:360:1 and 0.35);
		\draw (-1,0) arc (180:0:1 and 0.35);
		
		\draw[->] (0,0) -- (1.25,0) node[right] {$x$};
		\draw[->] (0,0) -- (0,1.25) node[above] {$z$};
		
		\fill (0,1) circle (0.03) node[above] {\small $+Z$};
		\fill (1,0) circle (0.03) node[right] {\small $+X$};
		
		\draw[->, thick] (0,1.0) to[bend left=25] (1.0,0);
		\node at (0.55,0.75) {\small $H$};
		
		\node at (0,-1.15) {\small pre-rotate by $H$ to reuse $Z$-hardware readout for $X$-measurement};
	\end{tikzpicture}
	\caption{Hadamard as an axis swap: it maps $Z$ to $X$ under conjugation.}
	\label{fig:bloch-hadamard}
\end{figure}

\subsubsection*{Algorithm 2: Phase becomes observable after basis change}

\begin{prop}[Phase-to-amplitude conversion]
	Let $\ket{\psi}=\ket{+}=\frac{1}{\sqrt2}(\ket0+\ket1)$.
	Apply $R_Z(\theta)$ and then apply $H$, and finally measure $Z$.
	The probability of outcome $0$ is
	\[
	\Pr(0)=\cos^2(\tfrac{\theta}{2}),
	\qquad
	\Pr(1)=\sin^2(\tfrac{\theta}{2}).
	\]
\end{prop}

\begin{proof}
	Start with $\ket{+}$, then
	\[
	R_Z(\theta)\ket{+}
	=
	\frac{1}{\sqrt2}\bigl(e^{-i\theta/2}\ket0+e^{i\theta/2}\ket1\bigr)
	\sim
	\frac{1}{\sqrt2}\bigl(\ket0+e^{i\theta}\ket1\bigr).
	\]
	Apply $H$:
	\[
	H\frac{1}{\sqrt2}\bigl(e^{-i\theta/2}\ket0+e^{i\theta/2}\ket1\bigr)
	=
	\frac12
	\Bigl((e^{-i\theta/2}+e^{i\theta/2})\ket0 + (e^{-i\theta/2}-e^{i\theta/2})\ket1\Bigr)
	=
	\cos(\tfrac{\theta}{2})\ket0 - i\sin(\tfrac{\theta}{2})\ket1.
	\]
	Measuring $Z$ gives $\Pr(0)=|\cos(\theta/2)|^2$ and $\Pr(1)=|\sin(\theta/2)|^2$.
\end{proof}

\begin{rem}
	The geometric story is: $R_Z(\theta)$ moves along the equator (phase motion),
	and $H$ rotates axes so that this equatorial displacement becomes a change in $Z$-projection.
\end{rem}

\subsubsection*{Algorithm 3: Gate synthesis by Euler angles}

\begin{prop}[Euler-angle decomposition in $\SU(2)$]
	Any $U\in\SU(2)$ can be written as
	\[
	U = R_Z(\alpha)\,R_Y(\beta)\,R_Z(\gamma)
	\quad\text{for some }\alpha,\beta,\gamma\in\R.
	\]
\end{prop}

\begin{proof}[Constructive sketch]
	Every $U\in\SU(2)$ has the form
	\[
	U=
	\begin{pmatrix}
		a & b\\
		-\overline{b} & \overline{a}
	\end{pmatrix},
	\qquad |a|^2+|b|^2=1.
	\]
	Choose $\beta\in[0,\pi]$ with $\cos(\beta/2)=|a|$, $\sin(\beta/2)=|b|$.
	Choose phases $\alpha,\gamma$ so that
	\[
	a=e^{-i(\alpha+\gamma)/2}\cos(\beta/2),\qquad
	b=-e^{-i(\alpha-\gamma)/2}\sin(\beta/2).
	\]
	Then multiplying out $R_Z(\alpha)R_Y(\beta)R_Z(\gamma)$ matches $U$.
\end{proof}

\begin{rem}[Engineering relevance]
	This underlies compiler synthesis and pulse decomposition: two fixed axes suffice to generate any single-qubit gate (up to global phase).
\end{rem}

\subsection*{Exercises}

\begin{exercise}[Global phase $\Rightarrow$ same Bloch vector]
	Show explicitly that $\ket{\psi}$ and $e^{i\theta}\ket{\psi}$ have identical Bloch vectors.
\end{exercise}

\begin{solution}
	Let
	\[
	\ket{\psi}=\alpha\ket0+\beta\ket1,\qquad |\alpha|^2+|\beta|^2=1,
	\]
	and define the global-phase shifted vector $\ket{\psi'}:=e^{i\theta}\ket{\psi}$.
	
	\smallskip
	\noindent\textbf{Step 1: compute density matrices.}
	\[
	\rho=\ket{\psi}\bra{\psi},\qquad \rho'=\ket{\psi'}\bra{\psi'}.
	\]
	Since $\ket{\psi'}=e^{i\theta}\ket{\psi}$ and $\bra{\psi'}=e^{-i\theta}\bra{\psi}$,
	\[
	\rho'=\ket{\psi'}\bra{\psi'}=(e^{i\theta}\ket{\psi})(e^{-i\theta}\bra{\psi})
	=\ket{\psi}\bra{\psi}=\rho.
	\]
	
	\smallskip
	\noindent\textbf{Step 2: take Bloch coordinates.}
	By definition,
	\[
	x=\Tr(\rho X),\quad y=\Tr(\rho Y),\quad z=\Tr(\rho Z),
	\]
	and similarly
	\[
	x'=\Tr(\rho' X),\quad y'=\Tr(\rho' Y),\quad z'=\Tr(\rho' Z).
	\]
	Because $\rho'=\rho$, we get $x'=x$, $y'=y$, $z'=z$.
	Hence the Bloch vectors are identical.
\end{solution}

\begin{exercise}[Bloch vector of $\ket{\psi(\theta,\phi)}$]
	Compute the Bloch vector of
	\[
	\ket{\psi(\theta,\phi)}=\cos(\theta/2)\ket0+e^{i\phi}\sin(\theta/2)\ket1.
	\]
\end{exercise}

\begin{solution}
	Write the column vector
	\[
	\ket{\psi}=
	\begin{pmatrix}
		\cos(\theta/2)\\ e^{i\phi}\sin(\theta/2)
	\end{pmatrix},
	\qquad
	\bra{\psi}=
	\begin{pmatrix}
		\cos(\theta/2) & e^{-i\phi}\sin(\theta/2)
	\end{pmatrix}.
	\]
	
	\smallskip
	\noindent\textbf{Step 1: compute the density matrix.}
	\[
	\rho=\ket{\psi}\bra{\psi}
	=
	\begin{pmatrix}
		\cos^2(\theta/2) & \cos(\theta/2)\,e^{-i\phi}\sin(\theta/2)\\[0.3em]
		e^{i\phi}\sin(\theta/2)\cos(\theta/2) & \sin^2(\theta/2)
	\end{pmatrix}.
	\]
	So the entries are
	\[
	\rho_{00}=\cos^2(\theta/2),\quad
	\rho_{11}=\sin^2(\theta/2),\quad
	\rho_{01}=\cos(\theta/2)\sin(\theta/2)e^{-i\phi},\quad
	\rho_{10}=\cos(\theta/2)\sin(\theta/2)e^{i\phi}.
	\]
	
	\smallskip
	\noindent\textbf{Step 2: compute $x=\Tr(\rho X)$.}
	Since
	\(
	X=\begin{psmallmatrix}0&1\\1&0\end{psmallmatrix},
	\)
	we have
	\[
	\Tr(\rho X)=\rho_{01}+\rho_{10}.
	\]
	Thus
	\[
	x=\rho_{01}+\rho_{10}
	=\cos(\theta/2)\sin(\theta/2)\bigl(e^{-i\phi}+e^{i\phi}\bigr)
	=2\cos(\theta/2)\sin(\theta/2)\cos\phi
	=\sin\theta\cos\phi.
	\]
	
	\smallskip
	\noindent\textbf{Step 3: compute $y=\Tr(\rho Y)$.}
	Since
	\(
	Y=\begin{psmallmatrix}0&-i\\ i&0\end{psmallmatrix},
	\)
	\[
	\Tr(\rho Y)=-i\rho_{01}+i\rho_{10}.
	\]
	Hence
	\[
	y=-i\rho_{01}+i\rho_{10}
	=i\cos(\theta/2)\sin(\theta/2)\bigl(e^{i\phi}-e^{-i\phi}\bigr)
	=i\cos(\theta/2)\sin(\theta/2)\cdot 2i\sin\phi
	= -2\cos(\theta/2)\sin(\theta/2)\sin\phi
	=\sin\theta\sin\phi.
	\]
	(Here we used $e^{i\phi}-e^{-i\phi}=2i\sin\phi$.)
	
	\smallskip
	\noindent\textbf{Step 4: compute $z=\Tr(\rho Z)$.}
	Since
	\(
	Z=\begin{psmallmatrix}1&0\\0&-1\end{psmallmatrix},
	\)
	\[
	z=\Tr(\rho Z)=\rho_{00}-\rho_{11}
	=\cos^2(\theta/2)-\sin^2(\theta/2)=\cos\theta.
	\]
	
	\smallskip
	\noindent Therefore the Bloch vector is
	\[
	\vec r=(x,y,z)=(\sin\theta\cos\phi,\ \sin\theta\sin\phi,\ \cos\theta).
	\]
\end{solution}

\begin{exercise}[$X,Y,Z$ as $\pi$-rotations on the Bloch sphere]
	Describe geometrically the action of $X$, $Y$, and $Z$ on the Bloch sphere.
\end{exercise}

\begin{solution}
	Let
	\[
	\rho=\frac12(I+xX+yY+zZ)
	\]
	be any qubit state. Applying a Pauli gate $P\in\{X,Y,Z\}$ sends $\rho$ to
	\[
	\rho' = P\rho P^\dagger = P\rho P,
	\]
	since Pauli matrices are Hermitian and unitary.
	
	\smallskip
	\noindent\textbf{Step 1: reduce to conjugating Pauli matrices.}
	\[
	\rho'=\frac12\bigl(I + x(PXP)+y(PYP)+z(PZP)\bigr).
	\]
	So we only need $PXP$, $PYP$, $PZP$.
	
	\smallskip
	\noindent\textbf{Step 2: use Pauli multiplication identities.}
	Using $X^2=Y^2=Z^2=I$ and the anticommutation relations
	\[
	XY=-YX=iZ,\quad YZ=-ZY=iX,\quad ZX=-XZ=iY,
	\]
	we get the key conjugations.
	
	\medskip
	\noindent\underline{Case $P=X$.}
	\[
	X X X = X,\qquad
	X Y X = -Y,\qquad
	X Z X = -Z.
	\]
	(For example, $XYX=(XY)X=(iZ)X=i(ZX)=i(-XZ)=-iXZ=-(iY)=-Y$.)
	Thus
	\[
	\rho'=\frac12(I + xX - yY - zZ),
	\quad\Rightarrow\quad
	(x',y',z')=(x,-y,-z).
	\]
	
	\medskip
	\noindent\underline{Case $P=Y$.}
	\[
	Y X Y = -X,\qquad
	Y Y Y = Y,\qquad
	Y Z Y = -Z.
	\]
	Hence
	\[
	(x',y',z')=(-x,y,-z).
	\]
	
	\medskip
	\noindent\underline{Case $P=Z$.}
	\[
	Z X Z = -X,\qquad
	Z Y Z = -Y,\qquad
	Z Z Z = Z.
	\]
	Hence
	\[
	(x',y',z')=(-x,-y,z).
	\]
	
	\smallskip
	\noindent\textbf{Step 3: geometric interpretation.}
	Each map flips the two coordinates orthogonal to its axis and keeps the axis coordinate fixed:
	\[
	X:\ (x,y,z)\mapsto (x,-y,-z)\quad\text{($\pi$-rotation about the $x$-axis),}
	\]
	\[
	Y:\ (x,y,z)\mapsto (-x,y,-z)\quad\text{($\pi$-rotation about the $y$-axis),}
	\]
	\[
	Z:\ (x,y,z)\mapsto (-x,-y,z)\quad\text{($\pi$-rotation about the $z$-axis).}
	\]
\end{solution}

\begin{exercise}[Hadamard measurement trick]
	Show that applying $H$ followed by a $Z$ measurement is equivalent to measuring $X$.
\end{exercise}

\begin{solution}
	There are two clean routes; both are worth knowing.
	
	\medskip
	\noindent\textbf{Route A (observable conjugation).}
	
	\smallskip
	\noindent\textbf{Step 1: compute the conjugation identity.}
	A direct multiplication gives
	\[
	H Z H = X.
	\]
	Since $H^\dagger=H$, equivalently $H^\dagger Z H = X$.
	
	\smallskip
	\noindent\textbf{Step 2: convert this to measurement statistics.}
	For any state $\rho$, the expectation value of observable $A$ is $\langle A\rangle_\rho=\Tr(\rho A)$.
	After applying $H$, the state becomes $\rho' = H\rho H^\dagger$ and
	\[
	\langle Z\rangle_{\rho'}=\Tr(\rho' Z)=\Tr(H\rho H^\dagger Z)
	=\Tr(\rho\,H^\dagger Z H)
	=\Tr(\rho X)=\langle X\rangle_\rho.
	\]
	So measuring $Z$ after $H$ produces the same $\pm1$ statistics as measuring $X$ before $H$.
	
	\medskip
	\noindent\textbf{Route B (projectors).}
	
	The $Z$-measurement projectors are
	\[
	\Pi_0=\ket0\bra0,\qquad \Pi_1=\ket1\bra1.
	\]
	After applying $H$, the probability of outcome $0$ is
	\[
	p(0)=\Tr(\rho' \Pi_0)=\Tr(H\rho H^\dagger \Pi_0)=\Tr(\rho\,H^\dagger \Pi_0 H).
	\]
	But $H\ket0=\ket{+}$ and $H\ket1=\ket{-}$, hence
	\[
	H^\dagger \Pi_0 H = H\ket0\bra0 H = \ket{+}\bra{+},
	\]
	which is exactly the projector for outcome $+1$ of an $X$-measurement.
	Similarly $H^\dagger \Pi_1 H=\ket{-}\bra{-}$.
	Thus $Z$ after $H$ reproduces $X$-measurement outcomes.
\end{solution}

\begin{exercise}[Axis-angle unitary]
	Given an axis $\vec n\in S^2$ and angle $\theta$, write an explicit unitary that rotates about $\vec n$.
\end{exercise}

\begin{solution}
	Let $\vec n=(n_x,n_y,n_z)$ with $\|\vec n\|=1$ and define
	\[
	\vec n\cdot\vec\sigma := n_xX+n_yY+n_zZ.
	\]
	
	\smallskip
	\noindent\textbf{Step 1: check the square.}
	Using the Pauli multiplication rule,
	\[
	(\vec n\cdot\vec\sigma)^2
	=
	n_x^2X^2+n_y^2Y^2+n_z^2Z^2
	+\text{(cross terms)}.
	\]
	Cross terms cancel because $\sigma_i\sigma_j+\sigma_j\sigma_i=0$ for $i\neq j$.
	Since $X^2=Y^2=Z^2=I$ and $n_x^2+n_y^2+n_z^2=1$,
	\[
	(\vec n\cdot\vec\sigma)^2 = I.
	\]
	
	\smallskip
	\noindent\textbf{Step 2: expand the exponential.}
	Define
	\[
	U(\vec n,\theta):=\exp\!\Bigl(-i\frac{\theta}{2}\,\vec n\cdot\vec\sigma\Bigr).
	\]
	Because $(\vec n\cdot\vec\sigma)^2=I$, the power series separates into even/odd parts:
	\[
	U(\vec n,\theta)
	=
	\sum_{k\ge0}\frac{1}{(2k)!}\Bigl(-i\frac{\theta}{2}\Bigr)^{2k}I
	+
	\sum_{k\ge0}\frac{1}{(2k+1)!}\Bigl(-i\frac{\theta}{2}\Bigr)^{2k+1}(\vec n\cdot\vec\sigma).
	\]
	Hence
	\[
	U(\vec n,\theta)=\cos(\tfrac{\theta}{2})I - i\sin(\tfrac{\theta}{2})(\vec n\cdot\vec\sigma).
	\]
	
	\smallskip
	\noindent\textbf{Step 3: write the explicit $2\times2$ matrix.}
	Compute
	\[
	\vec n\cdot\vec\sigma
	=
	n_x\begin{pmatrix}0&1\\1&0\end{pmatrix}
	+n_y\begin{pmatrix}0&-i\\ i&0\end{pmatrix}
	+n_z\begin{pmatrix}1&0\\0&-1\end{pmatrix}
	=
	\begin{pmatrix}
		n_z & n_x - i n_y\\
		n_x + i n_y & -n_z
	\end{pmatrix}.
	\]
	Therefore
	\[
	U(\vec n,\theta)
	=
	\begin{pmatrix}
		\cos(\tfrac{\theta}{2})-i n_z\sin(\tfrac{\theta}{2})
		&
		-i\sin(\tfrac{\theta}{2})(n_x-i n_y)
		\\[0.4em]
		-i\sin(\tfrac{\theta}{2})(n_x+i n_y)
		&
		\cos(\tfrac{\theta}{2})+i n_z\sin(\tfrac{\theta}{2})
	\end{pmatrix}.
	\]
	This is a standard axis-angle parametrization of $\SU(2)$.
\end{solution}

\begin{exercise}[Phase visibility]
	Explain why a $Z$-rotation may not change $Z$-measurement statistics, and how a basis change makes the phase visible.
\end{exercise}

\begin{solution}
	\textbf{Step 1: $Z$-measurement depends only on $z$.}
	For a state with Bloch vector $\vec r=(x,y,z)$, the $Z$ expectation is
	\[
	\langle Z\rangle=\Tr(\rho Z)=z.
	\]
	Equivalently, the $Z$-measurement probabilities are
	\[
	p(0)=\frac{1+z}{2},\qquad p(1)=\frac{1-z}{2}.
	\]
	So $Z$-statistics depend only on the $z$-coordinate.
	
	\smallskip
	\noindent\textbf{Step 2: $R_Z(\theta)$ does not change $z$.}
	Under $R_Z(\theta)$,
	\[
	(x,y,z)\mapsto(x\cos\theta-y\sin\theta,\ x\sin\theta+y\cos\theta,\ z).
	\]
	So $z$ is unchanged, hence $Z$-measurement probabilities can be unchanged.
	
	\smallskip
	\noindent\textbf{Step 3: a basis change converts equatorial motion into $Z$-readout.}
	Applying $H$ swaps axes: $H^\dagger Z H = X$.
	Thus measuring $Z$ after $H$ on a state $\rho$ is equivalent to measuring $X$ on $\rho$:
	\[
	\langle Z\rangle_{H\rho H}=\Tr(H\rho H Z)=\Tr(\rho\,H Z H)=\Tr(\rho X)=\langle X\rangle_\rho.
	\]
	Since $R_Z(\theta)$ rotates $(x,y)$, it changes $x$ and $y$ in general, so it becomes visible if you read out $X$ (or $Y$) instead of $Z$.
\end{solution}

\begin{exercise}[Hardware-flavored: sensitivity to $\varepsilon$]
	Suppose an intended $R_Z(\theta)$ is implemented as $R_Z(\theta+\varepsilon)$.
	Which measurement setting is most sensitive to $\varepsilon$?
\end{exercise}

\begin{solution}
	Let the ideal rotation be $R_Z(\theta)$ and the implemented rotation be $R_Z(\theta+\varepsilon)$.
	
	\smallskip
	\noindent\textbf{Step 1: interpret the error geometrically.}
	On the Bloch sphere, a $Z$-rotation rotates the azimuthal angle in the equatorial plane.
	Replacing $\theta$ by $\theta+\varepsilon$ is an \emph{over/under-rotation} about the $z$-axis by $\varepsilon$.
	
	\smallskip
	\noindent\textbf{Step 2: see which observables change to first order.}
	If a state has Bloch vector $(x,y,z)$, then after $R_Z(\theta)$,
	\[
	x(\theta)=x\cos\theta-y\sin\theta,\qquad
	y(\theta)=x\sin\theta+y\cos\theta,\qquad
	z(\theta)=z.
	\]
	Differentiate with respect to $\theta$:
	\[
	\frac{d}{d\theta}x(\theta)=-(x\sin\theta+y\cos\theta)=-y(\theta),
	\qquad
	\frac{d}{d\theta}y(\theta)=(x\cos\theta-y\sin\theta)=x(\theta),
	\qquad
	\frac{d}{d\theta}z(\theta)=0.
	\]
	Thus a small additional angle $\varepsilon$ changes $(x,y)$ by an amount of order $\varepsilon$ (first order),
	but changes $z$ by $0$ (no first-order effect, in fact no effect at all).
	
	\smallskip
	\noindent\textbf{Step 3: translate to measurement sensitivity.}
	Measuring $Z$ reads out $z$, which is unchanged under any $Z$-rotation, so it is \emph{insensitive} to $\varepsilon$.
	Measuring $X$ reads out $x$, and measuring $Y$ reads out $y$; both typically change linearly with $\varepsilon$.
	
	\smallskip
	\noindent\textbf{Conclusion.}
	The most sensitive settings are equatorial measurements ($X$ or $Y$), or equivalently
	$Z$-hardware readout preceded by a basis-change gate to implement an $X$- or $Y$-measurement (e.g.\ $H$ or $S^\dagger H$).
\end{solution}


	
\section{Differential Geometry Viewpoint}
\label{sec:dg}

\subsection*{Objective}
This chapter gives the minimum differential-geometry toolkit needed to read the rest of these notes
\emph{as geometry rather than as symbol manipulation}.
The goal is not to turn you into a geometer; it is to make the following pipeline feel natural:
\[
\text{states as points on a manifold}
\;\Longrightarrow\;
\text{gates as motion (paths)}
\;\Longrightarrow\;
\text{robustness as geometry (distance/curvature)}
\]
\[
\text{hardware as a constraint on motion (latency + noise).}
\;\Longrightarrow\;
\text{(what we can actually run).}
\]

By the end of this chapter, you should be able to:
\begin{itemize}
	\item compute lengths of curves and distances from a metric,
	\item recognize geodesics as ``straightest'' motions and why they matter for optimal control,
	\item compute curvature on $S^2$ and interpret Gauss--Bonnet as a global constraint,
	\item compare two models of the torus: the flat quotient (for distances) and the embedded torus (for curvature),
	\item understand why pure quantum states live on $\CP^{d-1}$ and how the Fubini--Study distance measures distinguishability.
\end{itemize}

\subsection*{Minimum geometric vocabulary}

\subsubsection*{Manifolds (what the ``space of states'' really is)}
A \emph{smooth manifold} $M$ is a space that looks locally like $\R^n$ and supports calculus.
We will use manifolds as \emph{configuration spaces}:
\begin{itemize}
	\item $S^2$ for pure single-qubit states (Bloch sphere),
	\item $\CP^{d-1}$ for pure $d$-level quantum states (projective Hilbert space),
	\item tori $T^2$ and higher tori as convenient ``flat'' models for phases and periodic parameters.
\end{itemize}

\subsubsection*{Tangent vectors and curves}
A smooth curve $\gamma:[0,1]\to M$ has a velocity vector $\dot\gamma(t)\in T_{\gamma(t)}M$.
If $M\subset \R^N$ is embedded, $\dot\gamma(t)$ is literally the derivative in $\R^N$; otherwise it is defined intrinsically.

\subsubsection*{Riemannian metric (the object that turns calculus into geometry)}
A \emph{Riemannian metric} $g$ assigns to each point $p\in M$ an inner product
\[
g_p(\cdot,\cdot):T_pM\times T_pM\to \R
\]
varying smoothly in $p$.

Once you have $g$, you can define:
\begin{itemize}
	\item \textbf{Length of a curve:}
	\[
	L(\gamma)=\int_0^1 \sqrt{g_{\gamma(t)}(\dot\gamma(t),\dot\gamma(t))}\,dt.
	\]
	\item \textbf{Distance between points:}
	\[
	d(p,q)=\inf_{\gamma:\,p\to q} L(\gamma).
	\]
	\item \textbf{Geodesics:} curves that locally minimize length (``straight lines'' in the metric).
\end{itemize}

\subsubsection*{Geodesics (why ``shortest path'' is not a metaphor)}
Geodesics can be defined in several equivalent ways.
We will use two:
\begin{enumerate}
	\item \textbf{Variational definition:} $\gamma$ is a geodesic if it is a critical point of the energy functional
	\[
	E(\gamma)=\frac12\int_0^1 g_{\gamma(t)}(\dot\gamma(t),\dot\gamma(t))\,dt.
	\]
	\item \textbf{Embedded picture:} on a surface embedded in $\R^3$,
	geodesics have acceleration normal to the surface (no tangential component),
	so they look ``as straight as possible'' while remaining on the surface.
\end{enumerate}

\subsubsection*{Curvature (local bending that affects global behavior)}
Curvature tells you how geometry deviates from flat Euclidean space.
On a 2D surface, the relevant scalar is \emph{Gaussian curvature} $K$.
You do not need general tensor formulas yet; you need:
\begin{itemize}
	\item $K>0$ (sphere-like): geodesics reconverge, triangles have angle sum $>\pi$,
	\item $K=0$ (flat): Euclidean behavior,
	\item $K<0$ (hyperbolic-like): geodesics diverge, triangles have angle sum $<\pi$.
\end{itemize}

\subsubsection*{Gauss--Bonnet (one global identity you should remember)}
For a compact oriented surface $M$ without boundary,
\[
\int_M K\,dA = 2\pi\,\chi(M),
\]
where $\chi(M)$ is the Euler characteristic.
This is a bridge between \emph{geometry} (curvature) and \emph{topology} (Euler characteristic).

\subsection*{Warm-up on $S^2$: metric, distance, curvature, Gauss--Bonnet}

\subsubsection*{The round metric}
Let
\[
S^2=\{(x,y,z)\in\R^3:\ x^2+y^2+z^2=1\}.
\]
The \emph{round metric} on $S^2$ is the one induced by the Euclidean inner product in $\R^3$:
for tangent vectors $u,v\in T_pS^2\subset \R^3$,
\[
g_p(u,v)=u\cdot v.
\]

\subsubsection*{Great circles are geodesics (the key fact for Bloch-sphere intuition)}
A \emph{great circle} is the intersection of $S^2$ with a plane through the origin.
Great circles are geodesics for the round metric.

\begin{prop}[Distance on $S^2$ via central angle]
	Let $p,q\in S^2$. Then the geodesic distance is
	\[
	d_{S^2}(p,q)=\arccos(p\cdot q).
	\]
\end{prop}

\begin{proof}
	Let $\alpha=\arccos(p\cdot q)\in[0,\pi]$ be the central angle between $p$ and $q$.
	The great circle through $p,q$ has radius $1$, so the arc length subtending angle $\alpha$ is exactly $\alpha$.
	Any other curve connecting $p$ and $q$ has length $\ge \alpha$ by the definition of distance as an infimum.
	Hence $d(p,q)=\alpha=\arccos(p\cdot q)$.
\end{proof}

\subsubsection*{Curvature of the round sphere}
For the unit sphere, the Gaussian curvature is constant:
\[
K\equiv 1.
\]
(For a sphere of radius $R$, it is $K\equiv 1/R^2$.)

\subsubsection*{Gauss--Bonnet on $S^2$ (a concrete computation)}
The area form $dA$ on the unit sphere integrates to $\mathrm{Area}(S^2)=4\pi$.
Since $K\equiv 1$,
\[
\int_{S^2}K\,dA=\int_{S^2}1\,dA=4\pi.
\]
Gauss--Bonnet gives
\[
4\pi = 2\pi\,\chi(S^2)\quad\Rightarrow\quad \chi(S^2)=2.
\]
This single computation is the canonical example of ``curvature integrates to topology.''

\subsubsection*{Spherical triangles (why curvature matters for algorithms)}
On $S^2$, a geodesic triangle with interior angles $A,B,C$ satisfies
\[
A+B+C-\pi = \mathrm{Area}(\triangle),
\]
for the unit sphere.
So the angle excess measures area. On curved state spaces, similar ``excess'' phenomena appear in how phases
and distinguishability accumulate along paths.

\begin{figure}[t]
	\centering
	\begin{tikzpicture}[scale=2.0, line cap=round, line join=round]
		\draw (0,0) circle (1);
		\draw[dashed] (-1,0) arc (180:360:1 and 0.35);
		\draw (-1,0) arc (180:0:1 and 0.35);
		
		\coordinate (A) at (0.15,0.85);
		\coordinate (B) at (-0.65,0.10);
		\coordinate (C) at (0.75,0.20);
		
		\fill (A) circle (0.03) node[above] {$A$};
		\fill (B) circle (0.03) node[left] {$B$};
		\fill (C) circle (0.03) node[right] {$C$};
		
		\draw[thick] (A) to[bend left=15] (B);
		\draw[thick] (B) to[bend left=10] (C);
		\draw[thick] (C) to[bend left=20] (A);
		
		\node at (0,-1.2) {\small Curved geometry: angle sum $>\pi$ on $S^2$ (area $\leftrightarrow$ curvature)};
	\end{tikzpicture}
	\caption{Spherical geometry in one picture: geodesics are great-circle arcs, and curvature shows up as an angle excess.}
	\label{fig:spherical-triangle}
\end{figure}

\subsection*{The torus: flat model for distances, embedded model for curvature}

\subsubsection*{Model 1: flat torus as a quotient (best for distances)}
Define the 2-torus as the quotient
\[
T^2=\R^2/\Z^2.
\]
A point is represented by $(x,y)\in\R^2$ with identification $(x,y)\sim(x+m,y+n)$ for $(m,n)\in\Z^2$.
The flat metric on $\R^2$ descends to a flat metric on $T^2$.

\begin{prop}[Geodesics on the flat torus]
	Geodesics on $T^2=\R^2/\Z^2$ are precisely the projections of straight lines in $\R^2$:
	\[
	\gamma(t)=(x_0,y_0)+t(v_x,v_y)\quad \text{mod }\Z^2.
	\]
\end{prop}

\begin{proof}
	Straight lines minimize length in $\R^2$ under the Euclidean metric.
	The quotient map $\R^2\to T^2$ is a local isometry, so local length-minimizers project to local length-minimizers.
	Hence projected lines are geodesics.
\end{proof}

\subsubsection*{Distance on the flat torus (explicit formula you can implement)}
Let $[p],[q]\in T^2$ be represented by $p,q\in\R^2$.
Then the distance is
\[
d_{T^2}([p],[q])=\min_{k\in\Z^2}\| (p-q)+k\|.
\]
Interpretation: among all lattice-shifted copies of $q$, choose the nearest one to $p$.

\subsubsection*{Model 2: embedded torus in $\R^3$ (best for curvature intuition)}
An embedded ``donut'' torus can be parameterized by
\[
\Phi(\theta,\varphi)
=
\bigl((R+r\cos\theta)\cos\varphi,\ (R+r\cos\theta)\sin\varphi,\ r\sin\theta\bigr),
\qquad
\theta,\varphi\in[0,2\pi).
\]
Here $R>r>0$.
This torus has \emph{non-constant} Gaussian curvature:
positive on the outer side and negative on the inner side.
So:
\[
\text{same topology ($T^2$), different geometry (curvature distribution).}
\]

\begin{rem}[Why two models?]
	The quotient torus is the right model when you care about periodic parameters and distances.
	The embedded torus is the right model when you care about curvature effects and visualization.
\end{rem}

\begin{figure}[t]
	\centering
	\begin{tikzpicture}[scale=2.2, line cap=round, line join=round]
		\draw[thick] (0,0) rectangle (1.6,1.2);
		\draw[->] (0.2,1.25) -- (1.4,1.25) node[midway,above] {\small identify};
		\draw[->] (0.2,-0.05) -- (1.4,-0.05) node[midway,below] {\small identify};
		\draw[->] (-0.05,0.2) -- (-0.05,1.0) node[midway,left] {\small identify};
		\draw[->] (1.65,0.2) -- (1.65,1.0) node[midway,right] {\small identify};
		\node at (0.8,0.6) {\small flat $T^2=\R^2/\Z^2$};

	\end{tikzpicture}
	\caption{Two complementary torus models: the flat quotient for distances and periodic parameters, and the embedded torus for curvature intuition.}
	\label{fig:torus-two-models}
\end{figure}
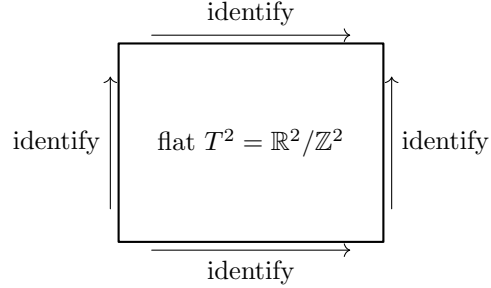

\subsection*{Back to quantum: $\CP^{d-1}$ and Fubini--Study distance}

\subsubsection*{Why projective space appears}
A pure quantum state in $\C^d$ is a unit vector $\ket{\psi}\in\C^d$.
Physical predictions are invariant under global phase:
\[
\ket{\psi}\sim e^{i\theta}\ket{\psi}.
\]
Therefore the true space of pure states is the projective space
\[
\CP^{d-1} = \bigl(\C^d\setminus\{0\}\bigr)/\sim,\qquad v\sim\lambda v \ (\lambda\in\C^\times).
\]

\subsubsection*{The Fubini--Study metric (the canonical metric on $\CP^{d-1}$)}
There are many equivalent ways to introduce it.
For computation, you can treat the induced distance formula as the primary definition:

\begin{defn}[Fubini--Study distance (operational form)]
	For two normalized pure states $\ket{\psi},\ket{\phi}\in\C^d$,
	the Fubini--Study distance on $\CP^{d-1}$ is
	\[
	d_{\mathrm{FS}}([\psi],[\phi]) := \arccos\bigl(|\braket{\psi}{\phi}|\bigr)\in[0,\pi/2].
	\]
\end{defn}

\begin{rem}[What it measures]
	$|\braket{\psi}{\phi}|$ is the overlap (fidelity for pure states).
	So $d_{\mathrm{FS}}$ is an \emph{angle} measuring distinguishability:
	\[
	|\braket{\psi}{\phi}|=1 \ \Rightarrow\ d_{\mathrm{FS}}=0 \quad(\text{same ray}),
	\qquad
	|\braket{\psi}{\phi}|=0 \ \Rightarrow\ d_{\mathrm{FS}}=\frac{\pi}{2}\quad(\text{orthogonal}).
	\]
\end{rem}

\subsubsection*{Special case: qubit ($d=2$) recovers the Bloch-sphere angle}
For $d=2$, $\CP^1\simeq S^2$ (Bloch sphere), and the distances satisfy:
\[
d_{\mathrm{FS}}([\psi],[\phi])=\frac12\, d_{S^2}(\vec r_\psi,\vec r_\phi),
\]
i.e.\ the FS distance is half the central angle between Bloch vectors.
This is one reason ``circuits are rotations'' works so cleanly for a qubit.

\begin{figure}[t]
	\centering
	\begin{tikzpicture}[scale=2.3, line cap=round, line join=round]
		\draw (0,0) circle (1);
		\draw[dashed] (-1,0) arc (180:360:1 and 0.35);
		\draw (-1,0) arc (180:0:1 and 0.35);
		
		\coordinate (P) at (0.25,0.80);
		\coordinate (Q) at (0.85,0.05);
		\fill (P) circle (0.03) node[above] {\small $\vec r_\psi$};
		\fill (Q) circle (0.03) node[right] {\small $\vec r_\phi$};
		
		\draw[thick] (P) to[bend left=15] (Q);
		\node at (0.55,0.55) {\small $d_{S^2}$};
		
		\node at (0,-1.2) {\small For qubits: $d_{\mathrm{FS}}=\tfrac12 d_{S^2}$ (FS angle is half Bloch angle)};
	\end{tikzpicture}
	\caption{For a qubit, the Fubini--Study distance on $\CP^1$ corresponds to half the geodesic angle on the Bloch sphere.}
	\label{fig:fs-half-bloch}
\end{figure}

\subsubsection*{Why this matters for circuits}
A parameterized circuit $U(\theta)$ generates a path of states
\[
[\psi(\theta)] = [U(\theta)\ket{\psi_0}]\in \CP^{d-1}.
\]
Geometry gives:
\begin{itemize}
	\item a notion of \textbf{speed} along the path (from the metric),
	\item a notion of \textbf{shortest route} (geodesics),
	\item a notion of \textbf{conditioning} for optimization (local metric tensor $\Rightarrow$ QFIM),
	\item a way to quantify \textbf{robustness} against small perturbations (distances/curvature).
\end{itemize}

\subsection*{Exercises}

\begin{exercise}[Distance on $S^2$]
	Let $p,q\in S^2\subset\R^3$ be unit vectors. Prove that the geodesic distance satisfies
	\[
	d_{S^2}(p,q)=\arccos(p\cdot q).
	\]
\end{exercise}

\begin{solution}
	Let $\alpha=\arccos(p\cdot q)\in[0,\pi]$ be the angle between $p$ and $q$.
	The plane spanned by $p$ and $q$ passes through the origin, hence its intersection with $S^2$ is a great circle.
	The shorter arc connecting $p$ to $q$ on that great circle subtends angle $\alpha$ at the origin.
	Since the radius is $1$, arc length equals the angle: $L=\alpha$.
	Because distance is the infimum over all connecting curves, $d_{S^2}(p,q)\le \alpha$.
	
	Conversely, any curve $\gamma$ on $S^2$ connecting $p$ to $q$ projects (via radial projection) to a curve on the great circle whose length is at most the original length; the great-circle arc is the unique length minimizer for $\alpha<\pi$.
	Hence no curve can have length $<\alpha$, so $d_{S^2}(p,q)=\alpha$.
\end{solution}

\begin{exercise}[Gauss--Bonnet on the unit sphere]
	Assume $K\equiv 1$ on the unit sphere and $\mathrm{Area}(S^2)=4\pi$.
	Use Gauss--Bonnet to compute $\chi(S^2)$.
\end{exercise}

\begin{solution}
	Gauss--Bonnet gives
	\[
	\int_{S^2}K\,dA = 2\pi\,\chi(S^2).
	\]
	Since $K\equiv 1$ and $\int_{S^2}1\,dA=\mathrm{Area}(S^2)=4\pi$,
	\[
	4\pi=2\pi\,\chi(S^2)\quad\Rightarrow\quad \chi(S^2)=2.
	\]
\end{solution}

\begin{exercise}[Distance on the flat torus]
	Let $T^2=\R^2/\Z^2$ with the induced flat metric.
	Show that for $[p],[q]\in T^2$ represented by $p,q\in\R^2$,
	\[
	d_{T^2}([p],[q])=\min_{k\in\Z^2}\|(p-q)+k\|.
	\]
\end{exercise}

\begin{solution}
	A curve $\gamma$ on $T^2$ lifts locally to curves $\widetilde\gamma$ in $\R^2$.
	Any path from $[p]$ to $[q]$ lifts to a path from $p$ to $q+k$ for some lattice vector $k\in\Z^2$
	(depending on how many times the path wraps around).
	The length of the path in $T^2$ equals the Euclidean length of the lift in $\R^2$ because the quotient is a local isometry.
	
	In $\R^2$, the shortest path from $p$ to $q+k$ is the straight line segment, with length $\|(p-q)+k\|$.
	Therefore the shortest path in $T^2$ corresponds to choosing the lattice shift $k$ that minimizes this length,
	giving the formula.
\end{solution}

\begin{exercise}[Fubini--Study distance for qubits]
	Let $\ket{\psi},\ket{\phi}\in\C^2$ be normalized pure qubit states with Bloch vectors $\vec r_\psi,\vec r_\phi\in S^2$.
	Assuming the identity
	\[
	|\braket{\psi}{\phi}|^2=\frac{1+\vec r_\psi\cdot \vec r_\phi}{2},
	\]
	prove that
	\[
	d_{\mathrm{FS}}([\psi],[\phi])=\frac12\,\arccos(\vec r_\psi\cdot\vec r_\phi).
	\]
\end{exercise}

\begin{solution}
	By definition,
	\[
	d_{\mathrm{FS}}([\psi],[\phi])=\arccos\bigl(|\braket{\psi}{\phi}|\bigr).
	\]
	Using the given identity,
	\[
	|\braket{\psi}{\phi}|=\sqrt{\frac{1+\vec r_\psi\cdot \vec r_\phi}{2}}.
	\]
	Write $\vec r_\psi\cdot \vec r_\phi=\cos\alpha$ with $\alpha\in[0,\pi]$.
	Then
	\[
	|\braket{\psi}{\phi}|=\sqrt{\frac{1+\cos\alpha}{2}}=\cos(\alpha/2),
	\]
	so
	\[
	d_{\mathrm{FS}}=\arccos(\cos(\alpha/2))=\alpha/2=\frac12\,\arccos(\vec r_\psi\cdot\vec r_\phi).
	\]
\end{solution}

\subsection*{What to visualize}

\begin{itemize}
	\item \textbf{$S^2$ as a distance machine:} pick two points $p,q$ and draw the great-circle arc; its length is the distance.
	\item \textbf{Curvature as angle excess:} draw a spherical triangle and remember ``sum of angles $>\pi$'' on $S^2$.
	\item \textbf{Flat vs.\ embedded torus:} the square-with-identified-edges for shortest paths and wrap-around distance;
	the donut picture for where curvature changes sign.
	\item \textbf{Projective state space:} rays in $\C^d$ (global phase ignored) $\Rightarrow \CP^{d-1}$.
	\item \textbf{Fubini--Study distance as overlap angle:} $|\langle\psi|\phi\rangle|$ close to $1$ means small distance, orthogonal means distance $\pi/2$.
\end{itemize}

\subsection*{Summary: why this chapter matters for circuits and hardware}

\begin{itemize}
	\item \textbf{Circuits generate paths.} A parameterized unitary $U(\theta)$ moves a state along a curve in $\CP^{d-1}$.
	\item \textbf{Optimization is geometry.} The local metric that measures ``how far the state moves'' under parameter changes
	is exactly what later becomes the QFIM and the Quantum Natural Gradient.
	\item \textbf{Robustness is distance/curvature.} Small control errors and noise perturb the path; geometry quantifies sensitivity.
	\item \textbf{Hardware imposes time.} Real-time controllers discretize the path into clocked steps; latency budgets constrain how quickly you can correct or adapt.
	\item \textbf{Two toy spaces cover most intuition.} $S^2$ (sphere) and $T^2$ (torus) already teach:
	\emph{curved vs.\ flat}, \emph{global constraints}, and \emph{periodic parameters}.
\end{itemize}

	
\section{Quantum Fisher Information Geometry (QFIM)}
\label{sec:qfim}

\subsection*{Objective}
This chapter builds the \emph{geometry of quantum state families} from the ground up and connects it
directly to \emph{computable formulas} used in variational circuits and hardware-aware optimization.
You will learn:
\begin{itemize}
	\item how measurement turns a quantum state into a classical probability distribution,
	\item how classical Fisher information defines a metric on distributions (and why it measures sensitivity),
	\item why quantum mechanics forces a choice of measurement and how the \emph{QFIM} is the ``best possible'' Fisher metric over all measurements,
	\item how to compute QFIM using the SLD (symmetric logarithmic derivative),
	\item the clean pure-state identity: QFIM is the pullback of the Fubini--Study metric,
	\item fast circuit formulas (generator/variance and covariance forms) that avoid density-matrix differentiation,
	\item how QFIM becomes a preconditioner: Quantum Natural Gradient (QNG),
	\item worked examples for one- and multi-parameter circuits, with practical notes (regularization, shot noise, singular metrics).
\end{itemize}

\subsection*{Minimal measurement vocabulary: POVMs and Born probabilities}

\subsubsection*{Born rule as a map: quantum state $\to$ classical distribution}
A measurement converts quantum information into classical data.
Mathematically, this is a map
\[
\rho \quad\longmapsto\quad p_\rho(x)
\]
where $p_\rho$ is a probability distribution on outcomes $x$.

\begin{defn}[POVM]
	A POVM (positive operator-valued measure) on a finite outcome set $\mathcal X$ is a collection of operators
	\[
	\{E_x\}_{x\in\mathcal X},\qquad E_x\succeq 0,\qquad \sum_{x\in\mathcal X} E_x = I.
	\]
\end{defn}

\begin{defn}[Born probabilities for a POVM]
	Given a density matrix $\rho$ and POVM $\{E_x\}$, the outcome probabilities are
	\[
	p_\rho(x) := \Tr(\rho E_x).
	\]
\end{defn}

\begin{rem}[Projective measurements are a special case]
	If $\{E_x\}$ are orthogonal projectors $\Pi_x$ summing to $I$, then this is the usual projective measurement.
	POVMs are strictly more general and include noisy or coarse-grained measurements.
\end{rem}

\subsubsection*{Parameter estimation viewpoint}
Let $\rho_\theta$ be a family of states depending on parameters $\theta=(\theta^1,\dots,\theta^m)$.
A measurement produces a family of distributions
\[
p_\theta(x)=\Tr(\rho_\theta E_x).
\]
The central question:
\begin{quote}
	How sensitive are the observable probabilities to small changes in $\theta$?
\end{quote}
Fisher information answers this in a coordinate-free way.

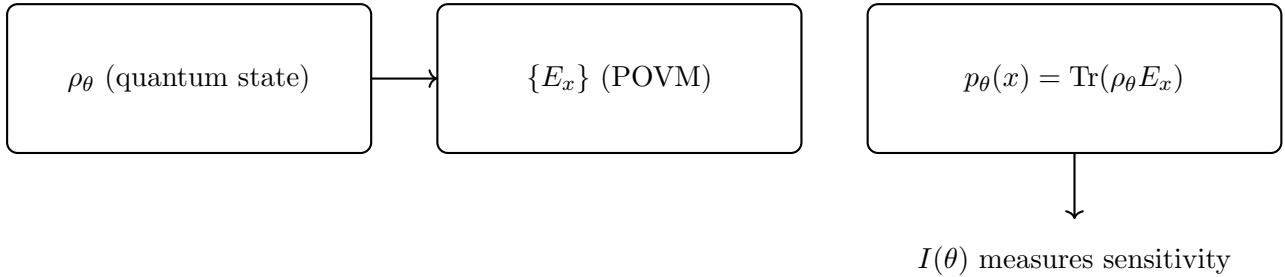
\begin{figure}[t]
	\centering
	\begin{tikzpicture}[scale=2.2, line cap=round, line join=round]
		\draw[thick, rounded corners] (0,0.5) rectangle (2.2,1.4);
		\node at (1.1,0.95) {$\rho_\theta$ (quantum state)};
		
		\draw[thick, rounded corners] (2.6,0.5) rectangle (4.8,1.4);
		\node at (3.7,0.95) {$\{E_x\}$ (POVM)};
		
		\draw[->, thick] (2.2,0.95) -- (2.6,0.95);
		
		\draw[thick, rounded corners] (5.2,0.5) rectangle (7.7,1.4);
		\node at (6.45,0.95) {$p_\theta(x)=\Tr(\rho_\theta E_x)$};
		
		\draw[->, thick] (6.45,0.5) -- (6.45,0.1);
		\node at (6.45,-0.15) {$I(\theta)$ measures sensitivity};
		
		\node at (3.85,1.75) {\small Quantum-to-classical: a measurement induces a statistical model};
	\end{tikzpicture}
	\caption{A POVM turns a parameterized quantum state family $\rho_\theta$ into a classical probability model $p_\theta(x)$. Fisher information measures how sharply $p_\theta$ changes with $\theta$.}
	\label{fig:qfim-povm-pipeline}
\end{figure}

\subsection*{Classical Fisher information: geometry of distributions}

\subsubsection*{Single-parameter Fisher information}
Let $p_\theta(x)$ be a parametric probability distribution on a finite set $\mathcal X$.
Define the score (log-derivative)
\[
\partial_\theta \log p_\theta(x) = \frac{\partial_\theta p_\theta(x)}{p_\theta(x)}\quad (\text{where }p_\theta(x)>0).
\]

\begin{defn}[Classical Fisher information (1D)]
	The Fisher information is
	\[
	I(\theta) := \sum_{x\in\mathcal X} p_\theta(x)\,\bigl(\partial_\theta \log p_\theta(x)\bigr)^2.
	\]
\end{defn}

\begin{rem}[Sensitivity interpretation]
	If a small parameter change $\theta\mapsto \theta+\delta$ causes a large change in the likelihood,
	then $I(\theta)$ is large. In estimation, large $I$ means \emph{small variance lower bounds} (Cram\'er--Rao).
\end{rem}

\subsubsection*{Multi-parameter Fisher information matrix (FIM)}
For $\theta=(\theta^1,\dots,\theta^m)$ define the score components
\[
\partial_i \log p_\theta(x) := \frac{\partial_i p_\theta(x)}{p_\theta(x)}.
\]

\begin{defn}[Classical Fisher information matrix]
	The Fisher information matrix is
	\[
	I_{ij}(\theta) := \sum_{x\in\mathcal X} p_\theta(x)\,(\partial_i \log p_\theta(x))(\partial_j \log p_\theta(x)).
	\]
	Equivalently, it is the covariance of the score:
	\[
	I_{ij}(\theta) = \mathbb E_\theta\!\left[ \partial_i \log p_\theta(X)\,\partial_j \log p_\theta(X)\right].
	\]
\end{defn}

\begin{prop}[FIM is positive semidefinite]
	For any vector $v\in\R^m$,
	\[
	v^T I(\theta) v \ge 0.
	\]
\end{prop}

\begin{proof}
	Compute
	\[
	v^T I v = \sum_{i,j} v_i I_{ij} v_j
	= \sum_x p_\theta(x)\left(\sum_i v_i \partial_i \log p_\theta(x)\right)^2 \ge 0.
	\]
\end{proof}

\subsubsection*{Fisher metric as a Riemannian metric on a statistical manifold}
The matrix $I(\theta)$ defines an inner product on tangent vectors $u,v\in\R^m$ by
\[
\langle u,v\rangle_\theta := u^T I(\theta) v.
\]
This is the first appearance of ``geometry of learning'': distance and conditioning depend on $I(\theta)$.

\subsection*{From classical to quantum: measurement choice and the ``best'' metric}

\subsubsection*{Measurement dependence}
Given $\rho_\theta$, different POVMs $\{E_x\}$ produce different classical models $p_\theta(x)$,
hence different classical Fisher matrices $I(\theta;\{E_x\})$.
So: \emph{there is no unique Fisher metric unless we decide how we measure.}

\subsubsection*{Best possible Fisher information over all measurements}
Quantum theory has a canonical answer: take the supremum over all POVMs.
The resulting matrix is the \emph{Quantum Fisher Information Matrix}.

\begin{defn}[QFIM as the best classical FIM]
	For a state family $\rho_\theta$, define the QFIM by
	\[
	F(\theta) := \sup_{\{E_x\}} I(\theta;\{E_x\}),
	\]
	where the supremum is taken over all POVMs.
\end{defn}

\begin{rem}[Why ``best'' is not vague]
	There are precise theorems (Helstrom, Holevo) that identify the QFIM with an operator formula (SLD)
	and show it upper-bounds any classical Fisher matrix from any measurement.
	So $F(\theta)$ is the intrinsic information geometry of $\rho_\theta$.
\end{rem}

\subsection*{Computable definition via SLD}

\subsubsection*{The symmetric logarithmic derivative (SLD)}
\begin{defn}[SLD operators]
	For each parameter $\theta^i$, the SLD $L_i$ is the Hermitian operator solving
	\[
	\partial_i \rho_\theta = \frac12\bigl(L_i \rho_\theta + \rho_\theta L_i\bigr).
	\]
\end{defn}

\begin{defn}[QFIM via SLD]
	The quantum Fisher information matrix is
	\[
	F_{ij}(\theta) := \Tr\!\bigl(\rho_\theta\,\frac12(L_iL_j+L_jL_i)\bigr)
	= \Re\,\Tr(\rho_\theta L_i L_j).
	\]
\end{defn}

\begin{prop}[QFIM upper-bounds any measurement Fisher matrix]
	For any POVM $\{E_x\}$ producing $p_\theta(x)=\Tr(\rho_\theta E_x)$, the classical FIM satisfies
	\[
	I(\theta;\{E_x\}) \preceq F(\theta).
	\]
\end{prop}

\begin{rem}[Practical consequence]
	You can design measurement strategies to approach the QFIM, but for optimization you typically
	use $F(\theta)$ as the intrinsic preconditioner, even if you estimate it approximately.
\end{rem}

\subsubsection*{Spectral formula (useful for mixed states)}
If $\rho=\sum_a \lambda_a \ket{a}\bra{a}$ with $\lambda_a\ge 0$, then a standard identity gives
\[
F_{ij}
=
\sum_{a,b:\,\lambda_a+\lambda_b>0}
\frac{2}{\lambda_a+\lambda_b}\,
\Re\!\Bigl(\bra{a}\,\partial_i\rho\,\ket{b}\ \bra{b}\,\partial_j\rho\,\ket{a}\Bigr).
\]
This avoids explicitly solving for $L_i$ and is often the most stable analytic expression for mixed states.

\subsection*{Pure states: QFIM equals the Fubini--Study pullback (derivation)}

\subsubsection*{Setup}
Let $\ket{\psi_\theta}$ be a smooth family of \emph{normalized} pure states.
Then $\rho_\theta=\ket{\psi_\theta}\bra{\psi_\theta}$.

\begin{prop}[Pure-state QFIM formula]
	For pure states,
	\[
	F_{ij}(\theta)
	=
	4\,\Re\!\Bigl(\braket{\partial_i\psi}{\partial_j\psi}
	-
	\braket{\partial_i\psi}{\psi}\braket{\psi}{\partial_j\psi}
	\Bigr).
	\]
\end{prop}

\begin{proof}
	Start with $\rho=\ket{\psi}\bra{\psi}$ and differentiate:
	\[
	\partial_i\rho = \ket{\partial_i\psi}\bra{\psi} + \ket{\psi}\bra{\partial_i\psi}.
	\]
	A convenient SLD choice for pure states is
	\[
	L_i = 2\bigl(\ket{\partial_i\psi}\bra{\psi}+\ket{\psi}\bra{\partial_i\psi}\bigr),
	\]
	which is Hermitian. Then
	\[
	\frac12(L_i\rho+\rho L_i)
	=
	\bigl(\ket{\partial_i\psi}\bra{\psi}+\ket{\psi}\bra{\partial_i\psi}\bigr)\ket{\psi}\bra{\psi}
	+
	\ket{\psi}\bra{\psi}\bigl(\ket{\partial_i\psi}\bra{\psi}+\ket{\psi}\bra{\partial_i\psi}\bigr).
	\]
	Use $\braket{\psi}{\psi}=1$ and $\rho\ket{\psi}=\ket{\psi}$:
	after simplification, the result equals $\partial_i\rho$ (details expand mechanically),
	so this $L_i$ satisfies the SLD equation.
	
	Now compute
	\[
	F_{ij}=\Re\,\Tr(\rho L_i L_j) = \Re\,\bra{\psi}L_iL_j\ket{\psi}.
	\]
	Insert $L_i$ and expand the products.
	The only non-vanishing contributions are those that return to $\ket{\psi}$ after acting.
	A direct expansion yields exactly
	\[
	F_{ij}
	=
	4\,\Re\!\Bigl(\braket{\partial_i\psi}{\partial_j\psi}
	-
	\braket{\partial_i\psi}{\psi}\braket{\psi}{\partial_j\psi}\Bigr).
	\]
\end{proof}

\begin{rem}[Geometric meaning: a metric on projective space]
	The expression
	\[
	g^{\mathrm{FS}}_{ij}(\theta)
	:=
	\Re\!\Bigl(\braket{\partial_i\psi}{\partial_j\psi}
	-
	\braket{\partial_i\psi}{\psi}\braket{\psi}{\partial_j\psi}\Bigr)
	\]
	is precisely the pullback of the Fubini--Study metric on $\CP^{d-1}$ along the map $\theta\mapsto [\psi_\theta]$.
	So
	\[
	F(\theta)=4\,g^{\mathrm{FS}}(\theta).
	\]
	This is the cleanest conceptual bridge between ``quantum states'' and ``Riemannian geometry.''
\end{rem}

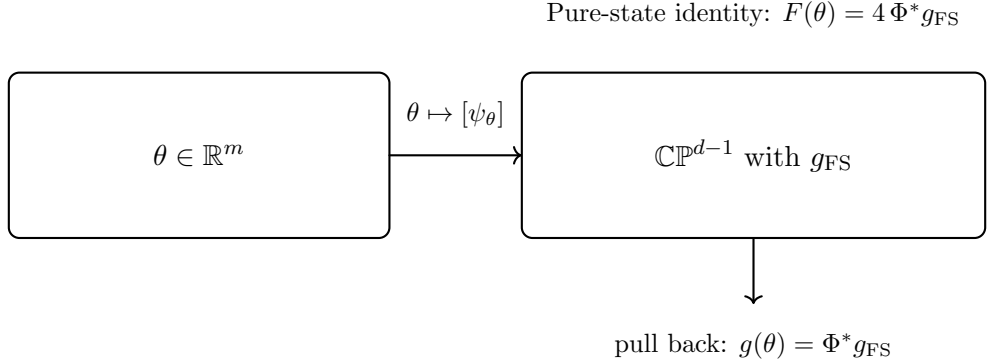
\begin{figure}[t]
	\centering
	\begin{tikzpicture}[scale=2.2, line cap=round, line join=round]
		\draw[thick, rounded corners] (0,0.5) rectangle (2.3,1.5);
		\node at (1.15,1.0) {$\theta\in\R^m$};
		
		\draw[->, thick] (2.3,1.0) -- (3.1,1.0);
		\node at (2.7,1.25) {\small $\theta\mapsto[\psi_\theta]$};
		
		\draw[thick, rounded corners] (3.1,0.5) rectangle (5.9,1.5);
		\node at (4.5,1.0) {$\CP^{d-1}$ with $g_{\mathrm{FS}}$};
		
		\draw[->, thick] (4.5,0.5) -- (4.5,0.1);
		\node at (4.5,-0.15) {\small pull back: $g(\theta)=\Phi^*g_{\mathrm{FS}}$};
		
		\node at (4.5,1.85) {\small Pure-state identity: $F(\theta)=4\,\Phi^*g_{\mathrm{FS}}$};
	\end{tikzpicture}
	\caption{For pure states, QFIM is the pullback of the Fubini--Study metric on projective Hilbert space.}
	\label{fig:qfim-pure-pullback}
\end{figure}

\subsection*{Fast computation tool: generator/variance form}

\subsubsection*{Single-parameter unitary families}
Let
\[
\ket{\psi_\theta}=U(\theta)\ket{\psi_0},\qquad U(\theta)=e^{-i\theta G},
\]
where $G$ is Hermitian (the generator).

\begin{prop}[Single-parameter QFI as variance]
	For a pure state family generated by $G$,
	\[
	F(\theta)=4\,\Var_{\psi_\theta}(G)
	=4\Bigl(\bra{\psi_\theta}G^2\ket{\psi_\theta}-\bra{\psi_\theta}G\ket{\psi_\theta}^2\Bigr).
	\]
\end{prop}

\begin{proof}
	Differentiate $\ket{\psi_\theta}$:
	\[
	\partial_\theta \ket{\psi_\theta} = -iG\ket{\psi_\theta}.
	\]
	Insert into the pure-state formula:
	\[
	F = 4\,\Re\Bigl(\braket{\partial\psi}{\partial\psi}-|\braket{\psi}{\partial\psi}|^2\Bigr)
	=4\Bigl(\bra{\psi}G^2\ket{\psi}-|\bra{\psi}G\ket{\psi}|^2\Bigr),
	\]
	using $\braket{\psi}{\partial\psi}=-i\bra{\psi}G\ket{\psi}$ and $\Re(i\cdot \text{real})=0$.
\end{proof}

\begin{rem}[Engineering translation]
	QFI is large when the generator fluctuates strongly in the state.
	If the state is (nearly) an eigenstate of $G$, variance is small and parameter becomes hard to estimate/optimize.
	This is the simplest geometric explanation of barren-plateau-style flatness in certain directions.
\end{rem}

\subsection*{Multi-parameter circuits: covariance form (practical for ans\"atze)}

\subsubsection*{Parameterized circuits as products of exponentials}
A common ansatz form is
\[
U(\theta)=U_L(\theta^L)\cdots U_2(\theta^2)U_1(\theta^1),
\qquad
U_k(\theta^k)=e^{-i\theta^k G_k},
\]
with fixed Hermitian generators $G_k$ (Pauli strings, hardware Hamiltonians, etc.).

Define the \emph{effective generators} (also called dressed generators)
\[
\widetilde G_i(\theta)
:=
U_{1:i-1}(\theta)^\dagger\,G_i\,U_{1:i-1}(\theta),
\qquad
U_{1:i-1}:=U_{i-1}\cdots U_1,
\]
so that the parameter derivative acts like
\[
\partial_i \ket{\psi_\theta} = -i\,\widetilde G_i(\theta)\ket{\psi_\theta}.
\]

\begin{prop}[Covariance form of pure-state QFIM for circuits]
	For pure $\ket{\psi_\theta}=U(\theta)\ket{\psi_0}$,
	\[
	F_{ij}(\theta)=4\,\Cov_{\psi_\theta}\bigl(\widetilde G_i(\theta),\widetilde G_j(\theta)\bigr),
	\]
	where the (symmetrized) covariance is
	\[
	\Cov_\psi(A,B)
	:=
	\frac12\bra{\psi}(AB+BA)\ket{\psi}
	-\bra{\psi}A\ket{\psi}\,\bra{\psi}B\ket{\psi}.
	\]
\end{prop}

\begin{proof}
	Use $\partial_i\ket{\psi}=-i\widetilde G_i\ket{\psi}$ and insert into the pure-state formula
	\[
	F_{ij}=4\,\Re\Bigl(\braket{\partial_i\psi}{\partial_j\psi}-\braket{\partial_i\psi}{\psi}\braket{\psi}{\partial_j\psi}\Bigr).
	\]
	Compute:
	\[
	\braket{\partial_i\psi}{\partial_j\psi}
	=
	(-i)(i)\bra{\psi}\widetilde G_i\widetilde G_j\ket{\psi}
	=
	\bra{\psi}\widetilde G_i\widetilde G_j\ket{\psi},
	\]
	and
	\[
	\braket{\partial_i\psi}{\psi}=-i\bra{\psi}\widetilde G_i\ket{\psi},\qquad
	\braket{\psi}{\partial_j\psi}=+i\bra{\psi}\widetilde G_j\ket{\psi}.
	\]
	Thus
	\[
	\braket{\partial_i\psi}{\psi}\braket{\psi}{\partial_j\psi}
	=
	\bra{\psi}\widetilde G_i\ket{\psi}\,\bra{\psi}\widetilde G_j\ket{\psi}.
	\]
	Taking the real part replaces $\bra{\psi}\widetilde G_i\widetilde G_j\ket{\psi}$
	by $\frac12\bra{\psi}(\widetilde G_i\widetilde G_j+\widetilde G_j\widetilde G_i)\ket{\psi}$,
	yielding the covariance form.
\end{proof}

\begin{rem}[Why this is practical]
	This formula computes QFIM by expectation values of (dressed) generators and their pairwise products.
	On simulators this is exact; on hardware it can be estimated via measurement grouping strategies.
	It avoids differentiating density matrices and is the default workhorse in variational algorithms.
\end{rem}

\subsection*{QFIM and Quantum Natural Gradient (QNG): geometry-aware preconditioning}

\subsubsection*{Problem: Euclidean gradient ignores geometry}
Suppose you minimize a cost $C(\theta)$ (energy, infidelity, negative log-likelihood, etc.).
A standard update uses the Euclidean gradient:
\[
\theta \leftarrow \theta - \eta \nabla C(\theta).
\]
But $\theta$-coordinates are arbitrary; reparameterizations change the meaning of ``step size''.

\subsubsection*{Geometric fix: take steepest descent in the QFIM metric}
In a Riemannian metric $F(\theta)$, the steepest descent direction is
\[
\Delta\theta_{\mathrm{QNG}} = -\eta\,F(\theta)^{-1}\,\nabla C(\theta).
\]
This is the \emph{Quantum Natural Gradient} update.

\begin{rem}[Singular and noisy QFIM: regularization is not optional]
	In practice $F$ can be ill-conditioned or singular (redundant parameters, symmetries, barren directions).
	Use a damped inverse:
	\[
	(F+\lambda I)^{-1}
	\]
	or truncated eigenspace inversion. This is analogous to Levenberg--Marquardt / Tikhonov regularization.
\end{rem}

\begin{figure}[t]
	\centering
	\begin{tikzpicture}[scale=2.1, line cap=round, line join=round]
		\draw[->] (-0.2,0) -- (3.2,0) node[right] {$\theta^1$};
		\draw[->] (0,-0.2) -- (0,2.4) node[above] {$\theta^2$};
		
		\draw (1.7,1.1) ellipse (1.2 and 0.55);
		\draw (1.7,1.1) ellipse (0.85 and 0.38);
		\draw (1.7,1.1) ellipse (0.50 and 0.22);
		
		\fill (0.9,1.8) circle (0.03) node[above] {$\theta$};
		
		\draw[->, thick] (0.9,1.8) -- (1.35,1.3) node[midway,right] {\small $-\nabla C$};
		
		\draw[->, thick] (0.9,1.8) -- (1.05,1.55) node[midway,left] {\small $-F^{-1}\nabla C$};
		
		\node at (1.7,-0.55) {\small QNG rescales directions according to the local state-space geometry};
	\end{tikzpicture}
	\caption{Contour picture intuition: with an anisotropic metric, the steepest descent direction differs from the Euclidean gradient. QNG uses the QFIM to normalize parameter directions by state-space sensitivity.}
	\label{fig:qng-contours}
\end{figure}

\subsection*{Worked examples}

\subsubsection*{Example 1: single-qubit $R_Y(\theta)$ from $\ket{0}$}
Let
\[
\ket{\psi_\theta}=R_Y(\theta)\ket{0},\qquad R_Y(\theta)=e^{-i\frac{\theta}{2}Y}.
\]
Then $G=\frac12 Y$, so for pure states
\[
F(\theta)=4\,\Var_{\psi_\theta}(G).
\]
Compute $\Var(G)$ by evaluating $\langle Y\rangle$ and $\langle Y^2\rangle$.
Since $Y^2=I$,
\[
\langle G^2\rangle = \frac14\langle I\rangle=\frac14.
\]
Also $\langle Y\rangle_{\psi_\theta}=0$ because $R_Y(\theta)\ket0$ stays on the $xz$-plane of the Bloch sphere,
so the $y$-coordinate is zero. Hence $\langle G\rangle=0$ and
\[
\Var(G)=\frac14 \quad\Rightarrow\quad F(\theta)=4\cdot \frac14 = 1.
\]
So this parameter is uniformly well-conditioned: the state moves at constant speed in Fubini--Study geometry.

\subsubsection*{Example 2: single-qubit $R_Z(\theta)$ acting on $\ket{0}$ (a ``dead'' direction)}
Let $\ket{\psi_\theta}=R_Z(\theta)\ket0$ with $R_Z(\theta)=e^{-i\frac{\theta}{2}Z}$.
But $\ket0$ is an eigenstate of $Z$, so the physical state is unchanged up to global phase.
Thus the QFI must vanish. Using variance:
\[
G=\frac12 Z,\qquad \Var_{\ket0}(G)=\frac14-\left(\frac12\right)^2=0,
\]
so
\[
F(\theta)=0.
\]
This is the cleanest illustration of a singular QFIM direction: the parameter does not move you in projective state space.

\subsubsection*{Example 3: two-parameter circuit with noncommuting dressed generators}
Let
\[
\ket{\psi(\theta_1,\theta_2)} = R_Z(\theta_2)\,R_Y(\theta_1)\ket0,
\]
with generators $G_1=\frac12 Y$, $G_2=\frac12 Z$.
The dressed generator for $\theta_1$ is $\widetilde G_1=G_1$ (nothing before it).
For $\theta_2$, the dressed generator is
\[
\widetilde G_2(\theta_1) = R_Y(\theta_1)^\dagger\,\frac12 Z\,R_Y(\theta_1)
= \frac12\bigl(Z\cos\theta_1 + X\sin\theta_1\bigr),
\]
since $Y$-rotation rotates the Bloch axes.

Now compute QFIM using covariance in state $\ket{\psi}$.
Because $R_Z(\theta_2)$ is last, it does not change expectation values of operators conjugated consistently;
you can evaluate in $\ket{\phi}=R_Y(\theta_1)\ket0$.

First,
\[
F_{11}=4\,\Var_\phi\!\left(\frac12 Y\right)=1
\]
as in Example 1.

Next,
\[
F_{22}=4\,\Var_\phi\!\left(\frac12(Z\cos\theta_1+X\sin\theta_1)\right).
\]
In state $\ket{\phi}$, the Bloch vector is $(\sin\theta_1,0,\cos\theta_1)$.
So
\[
\langle X\rangle_\phi = \sin\theta_1,\qquad \langle Z\rangle_\phi=\cos\theta_1,
\qquad \langle Z\cos\theta_1+X\sin\theta_1\rangle_\phi = 1.
\]
Also $(Z\cos\theta_1+X\sin\theta_1)^2=I$ (Pauli along a unit axis), hence
\[
\left\langle \left(\frac12(Z\cos\theta_1+X\sin\theta_1)\right)^2\right\rangle
= \frac14.
\]
Thus
\[
\Var\left(\frac12(\cdot)\right)=\frac14-\left(\frac12\cdot 1\right)^2=0,
\quad\Rightarrow\quad F_{22}=0.
\]
So the second parameter is again a dead direction for this particular ansatz: once you rotate to $\ket{\phi}$,
the axis of the subsequent $Z$-rotation aligns with the state, producing only phase.

Finally,
\[
F_{12}=4\,\Cov_\phi\!\left(\frac12Y,\frac12(Z\cos\theta_1+X\sin\theta_1)\right).
\]
But in the $xz$-plane state $\ket{\phi}$, $\langle Y\rangle_\phi=0$ and also
$\langle YX+XY\rangle=\langle YZ+ZY\rangle=0$ in any qubit state because these anticommutators vanish:
$YX+XY=0$ and $YZ+ZY=0$.
Hence $F_{12}=0$.

Conclusion:
\[
F(\theta_1,\theta_2)=
\begin{pmatrix}
	1 & 0\\
	0 & 0
\end{pmatrix}.
\]
This example is intentionally simple but extremely important: \emph{QFIM reveals redundant parameters immediately.}

\subsection*{Exercises}

\begin{exercise}[Classical Fisher information for Bernoulli]
	Let $X\sim \mathrm{Bernoulli}(p)$ with parameter $\theta=p\in(0,1)$:
	\[
	\Pr(X=1)=p,\quad \Pr(X=0)=1-p.
	\]
	Compute the Fisher information $I(p)$.
\end{exercise}

\begin{solution}
	The log-likelihood derivatives are
	\[
	\partial_p \log p^X(1-p)^{1-X}
	=
	\frac{X}{p}-\frac{1-X}{1-p}.
	\]
	Square and take expectation:
	\[
	I(p)=\mathbb E\left[\left(\frac{X}{p}-\frac{1-X}{1-p}\right)^2\right]
	=
	p\left(\frac{1}{p}\right)^2+(1-p)\left(\frac{1}{1-p}\right)^2
	=
	\frac{1}{p}+\frac{1}{1-p}
	=
	\frac{1}{p(1-p)}.
	\]
\end{solution}

\begin{exercise}[POVM-induced Fisher information]
	Let $\rho_\theta$ be a qubit state family and $\{E_x\}_{x\in\{0,1\}}$ a two-outcome POVM.
	Write the classical Fisher information (1D) explicitly in terms of $p_\theta=\Tr(\rho_\theta E_1)$.
\end{exercise}

\begin{solution}
	With two outcomes,
	\[
	p_\theta(1)=p_\theta,\qquad p_\theta(0)=1-p_\theta.
	\]
	Thus the Fisher information is the Bernoulli formula:
	\[
	I(\theta)=\frac{(\partial_\theta p_\theta)^2}{p_\theta(1-p_\theta)}.
	\]
\end{solution}

\begin{exercise}[Pure-state QFI equals generator variance]
	Let $\ket{\psi_\theta}=e^{-i\theta G}\ket{\psi_0}$ with $G$ Hermitian.
	Using the pure-state QFIM formula
	\[
	F=4\,\Re\bigl(\braket{\partial\psi}{\partial\psi}-|\braket{\psi}{\partial\psi}|^2\bigr),
	\]
	derive $F=4\,\Var(G)$.
\end{exercise}

\begin{solution}
	Differentiate:
	\[
	\partial_\theta\ket{\psi_\theta}=-iG\ket{\psi_\theta}.
	\]
	Then
	\[
	\braket{\partial\psi}{\partial\psi}=\bra{\psi}G^2\ket{\psi},\qquad
	\braket{\psi}{\partial\psi}=-i\bra{\psi}G\ket{\psi}.
	\]
	So $|\braket{\psi}{\partial\psi}|^2=\bra{\psi}G\ket{\psi}^2$.
	Insert into the formula to obtain
	\[
	F=4\bigl(\bra{\psi}G^2\ket{\psi}-\bra{\psi}G\ket{\psi}^2\bigr)=4\,\Var_\psi(G).
	\]
\end{solution}

\begin{exercise}[A dead parameter direction]
	Show that if $\ket{\psi_0}$ is an eigenstate of $G$, then the QFI for $\ket{\psi_\theta}=e^{-i\theta G}\ket{\psi_0}$ is zero.
\end{exercise}

\begin{solution}
	If $G\ket{\psi_0}=\lambda\ket{\psi_0}$, then for all $\theta$,
	\[
	\ket{\psi_\theta}=e^{-i\theta\lambda}\ket{\psi_0},
	\]
	which differs only by global phase, so the projective state is unchanged. Hence QFI must vanish.
	Using variance:
	\[
	\Var_{\psi_\theta}(G)=\bra{\psi_0}G^2\ket{\psi_0}-\bra{\psi_0}G\ket{\psi_0}^2
	=\lambda^2-\lambda^2=0,
	\]
	so $F=4\Var(G)=0$.
\end{solution}

\begin{exercise}[Two-parameter covariance form]
	Let $\ket{\psi(\theta)}=U(\theta)\ket{\psi_0}$ be pure and define dressed generators $\widetilde G_i(\theta)$ by
	$\partial_i\ket{\psi}=-i\widetilde G_i\ket{\psi}$.
	Show that
	\[
	F_{ij}=4\left(\frac12\bra{\psi}(\widetilde G_i\widetilde G_j+\widetilde G_j\widetilde G_i)\ket{\psi}
	-\bra{\psi}\widetilde G_i\ket{\psi}\bra{\psi}\widetilde G_j\ket{\psi}\right).
	\]
\end{exercise}

\begin{solution}
	Insert $\partial_i\ket{\psi}=-i\widetilde G_i\ket{\psi}$ into the pure-state QFIM identity:
	\[
	F_{ij}=4\,\Re\Bigl(\bra{\psi}\widetilde G_i\widetilde G_j\ket{\psi}
	-\bra{\psi}\widetilde G_i\ket{\psi}\bra{\psi}\widetilde G_j\ket{\psi}\Bigr).
	\]
	Since $\widetilde G_i$ are Hermitian, the real part of $\bra{\psi}\widetilde G_i\widetilde G_j\ket{\psi}$
	is $\frac12\bra{\psi}(\widetilde G_i\widetilde G_j+\widetilde G_j\widetilde G_i)\ket{\psi}$.
	This yields the covariance form.
\end{solution}

\begin{exercise}[QNG step as a constrained minimizer]
	Let $C(\theta)$ be a smooth cost. Consider the quadratic model
	\[
	C(\theta+\Delta)\approx C(\theta)+\nabla C(\theta)^T\Delta
	\]
	subject to a constraint on the squared step length in the QFIM metric:
	\[
	\Delta^T F(\theta)\Delta \le \varepsilon^2.
	\]
	Show that the minimizing direction is proportional to $-F^{-1}\nabla C$.
\end{exercise}

\begin{solution}
	Use Lagrange multipliers for minimizing the linear approximation under a quadratic constraint.
	Minimize
	\[
	\mathcal L(\Delta,\lambda)=\nabla C^T\Delta+\lambda(\Delta^T F\Delta-\varepsilon^2),\qquad \lambda\ge 0.
	\]
	Stationarity in $\Delta$ gives
	\[
	\nabla C + 2\lambda F\Delta = 0
	\quad\Rightarrow\quad
	\Delta = -\frac{1}{2\lambda}\,F^{-1}\nabla C.
	\]
	Thus the minimizing direction is $-F^{-1}\nabla C$, with scale chosen to satisfy the constraint.
\end{solution}


	
\section{Quantum Circuits I: From Diagrams to Linear Maps (Geometry-Aware Edition)}
\label{sec:circuits1}

\subsection{Objective}
The goal of this chapter is to translate \emph{circuit diagrams} into precise
\emph{linear maps} (unitary matrices) without any ambiguity about ordering,
tensor products, or basis conventions. At the same time, we keep a
\emph{geometric interpretation} in view: single-qubit subcircuits act as
\emph{rotations on the Bloch sphere}, and two-qubit primitives (like CNOT)
create and manipulate correlation/entanglement in a way that becomes visible
both algebraically (as matrices) and geometrically (via reduced Bloch vectors).

\subsection{Conventions (read once, reuse everywhere)}
\subsubsection{Wires, ordering, and basis}
We fix the computational basis
\[
\ket{0}=\begin{pmatrix}1\\0\end{pmatrix},\qquad
\ket{1}=\begin{pmatrix}0\\1\end{pmatrix}.
\]
For two qubits we use the ordered basis
\[
\ket{00},\ket{01},\ket{10},\ket{11}
\quad\text{where}\quad
\ket{ab}:=\ket{a}\otimes\ket{b}.
\]
(First wire = left tensor factor, second wire = right tensor factor.)

\subsubsection{Time order = matrix multiplication}
If a circuit applies $U$ and then applies $V$, the overall map is
\[
\ket{\psi}\ \mapsto\ V(U\ket{\psi})=(VU)\ket{\psi}.
\]
So the matrix product is \emph{right-to-left} when reading gates in time order.

\subsubsection{Parallel gates = tensor product}
If, at the same time step, the first wire gets gate $A$ and the second wire gets gate $B$,
then the joint operation is
\[
A\otimes B.
\]

\subsubsection{ASCII circuit drawings (TikZ)}
Throughout, we use the following simple TikZ style for circuit sketches.
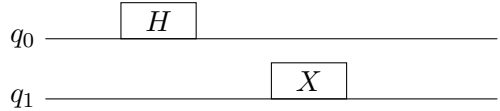
\begin{center}
	\begin{tikzpicture}[x=1.0cm,y=0.8cm, line cap=round]
		\draw (0,0) -- (6,0);
		\draw (0,-1) -- (6,-1);
		\node[left] at (0,0) {$q_0$};
		\node[left] at (0,-1) {$q_1$};
		\draw (1,0) rectangle (2,0.6); \node at (1.5,0.3) {$H$};
		\draw (3,-1) rectangle (4,-0.4); \node at (3.5,-0.7) {$X$};
	\end{tikzpicture}
\end{center}

\subsection{Basic single-qubit gates as matrices}
\subsubsection{The standard gates}
\[
I=\begin{pmatrix}1&0\\0&1\end{pmatrix},\quad
X=\begin{pmatrix}0&1\\1&0\end{pmatrix},\quad
Y=\begin{pmatrix}0&-i\\ i&0\end{pmatrix},\quad
Z=\begin{pmatrix}1&0\\0&-1\end{pmatrix}.
\]
\[
H=\frac1{\sqrt2}\begin{pmatrix}1&1\\1&-1\end{pmatrix},\qquad
S=\begin{pmatrix}1&0\\0&i\end{pmatrix},\qquad
T=\begin{pmatrix}1&0\\0&e^{i\pi/4}\end{pmatrix}.
\]

\subsubsection{Rotation gates and Bloch-sphere meaning}
Define (physics convention)
\[
R_X(\theta):=e^{-i\frac{\theta}{2}X},\qquad
R_Y(\theta):=e^{-i\frac{\theta}{2}Y},\qquad
R_Z(\theta):=e^{-i\frac{\theta}{2}Z}.
\]
Using $X^2=Y^2=Z^2=I$ one obtains the closed forms
\[
R_\alpha(\theta)=\cos\!\Bigl(\frac{\theta}{2}\Bigr)I
-i\sin\!\Bigl(\frac{\theta}{2}\Bigr)\alpha,
\quad \alpha\in\{X,Y,Z\}.
\]
\emph{Bloch-sphere interpretation:} for a pure state with Bloch vector $\vec r\in S^2$,
\[
\rho=\ket{\psi}\bra{\psi}=\frac12(I+\vec r\cdot\vec\sigma),
\qquad \vec\sigma=(X,Y,Z),
\]
conjugation by $R_\alpha(\theta)$ rotates $\vec r$ by angle $\theta$ about the corresponding axis.

\begin{figure}[t]
	\centering
	\begin{tikzpicture}[scale=2.2, line cap=round, line join=round]
		\draw (0,0) circle (1);
		\draw[dashed] (-1,0) arc (180:360:1 and 0.35);
		\draw (-1,0) arc (180:0:1 and 0.35);
		\draw[->] (0,0) -- (1.25,0) node[right] {$x$};
		\draw[->] (0,0) -- (0,1.25) node[above] {$z$};
		\draw[->] (0,0) -- (-0.75,-0.55) node[left] {$y$};
		\coordinate (P) at (0.7,0.15);
		\fill (P) circle (0.03) node[right] {\small $\vec r$};
		\draw[thick] (0,0)--(P);
		\draw[->, thick] (0.55,0) arc (0:55:0.55 and 0.2);
		\node at (0.55,0.25) {\small $R_Z(\theta)$};
	\end{tikzpicture}
	\caption{Single-qubit unitaries act as rotations on the Bloch sphere: $R_Z(\theta)$ rotates the Bloch vector in the $(x,y)$-plane while keeping $z$ fixed.}
	\label{fig:bloch-rotations-circuits1}
\end{figure}
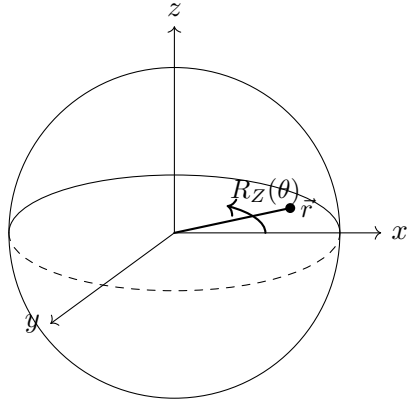

\subsection{Gate composition is matrix multiplication (worked)}
\subsubsection{Worked example 1: $HZH = X$ as a circuit identity}
Circuit:
\begin{center}
	\begin{tikzpicture}[x=1.1cm,y=0.8cm, line cap=round]
		\draw (0,0) -- (6,0);
		\node[left] at (0,0) {$q$};
		\draw (1,-0.35) rectangle (2,0.35); \node at (1.5,0) {$H$};
		\draw (3,-0.35) rectangle (4,0.35); \node at (3.5,0) {$Z$};
		\draw (5,-0.35) rectangle (6,0.35); \node at (5.5,0) {$H$};
	\end{tikzpicture}
\end{center}
Algebra:
\[
HZH=\frac1{\sqrt2}\begin{pmatrix}1&1\\1&-1\end{pmatrix}
\begin{pmatrix}1&0\\0&-1\end{pmatrix}
\frac1{\sqrt2}\begin{pmatrix}1&1\\1&-1\end{pmatrix}
=
\begin{pmatrix}0&1\\1&0\end{pmatrix}=X.
\]
Geometric meaning: Hadamard rotates axes so that a $Z$-flip becomes an $X$-flip on the Bloch sphere.

\subsubsection{Worked example 2: phase becomes observable after a basis change}
Consider
\[
\ket{\psi}=\ket{0}
\ \xrightarrow{H}\ 
\ket{+}=\frac{\ket0+\ket1}{\sqrt2}
\ \xrightarrow{R_Z(\theta)}\ 
\frac{\ket0+e^{i\theta}\ket1}{\sqrt2}
\ \xrightarrow{H}\ 
\text{measure }Z.
\]
Circuit:
\begin{center}
	\begin{tikzpicture}[x=1.05cm,y=0.8cm, line cap=round]
		\draw (0,0) -- (7,0);
		\node[left] at (0,0) {$q$};
		\draw (1,-0.35) rectangle (2,0.35); \node at (1.5,0) {$H$};
		\draw (3,-0.35) rectangle (4,0.35); \node at (3.5,0) {$R_Z(\theta)$};
		\draw (5,-0.35) rectangle (6,0.35); \node at (5.5,0) {$H$};
		\node at (6.6,0) {\small measure $Z$};
	\end{tikzpicture}
\end{center}
Compute the state right before measurement:
\[
H\,R_Z(\theta)\,H\,\ket0
=
H\Bigl(\frac{\ket0+e^{i\theta}\ket1}{\sqrt2}\Bigr)
=
\frac1{2}\Bigl((1+e^{i\theta})\ket0+(1-e^{i\theta})\ket1\Bigr).
\]
Thus
\[
p(0)=\Bigl|\frac{1+e^{i\theta}}{2}\Bigr|^2
=
\frac14(2+2\cos\theta)
=
\frac{1+\cos\theta}{2},
\qquad
p(1)=\frac{1-\cos\theta}{2}.
\]
Geometric meaning:
$R_Z(\theta)$ is a rotation around the $z$-axis (equator motion);
the final $H$ maps that equatorial displacement into the $z$-projection, making phase visible in $Z$-readout.

\subsection{Two-qubit systems: tensor products and bookkeeping}
\subsubsection{Basis and tensor products}
For $\ket{\psi}=\alpha\ket0+\beta\ket1$ and $\ket{\phi}=\gamma\ket0+\delta\ket1$,
\[
\ket{\psi}\otimes\ket{\phi}
=
\alpha\gamma\ket{00}+\alpha\delta\ket{01}+\beta\gamma\ket{10}+\beta\delta\ket{11}.
\]
Matrix rule:
\[
(A\otimes B)(\ket{v}\otimes\ket{w})=(A\ket v)\otimes(B\ket w).
\]

\subsubsection{Worked example: parallel gates are tensor products}
Circuit:
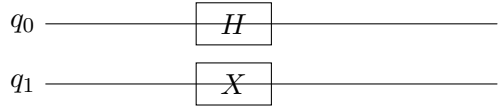
\begin{center}
	\begin{tikzpicture}[x=1.0cm,y=0.8cm, line cap=round]
		\draw (0,0) -- (6,0);
		\draw (0,-1) -- (6,-1);
		\node[left] at (0,0) {$q_0$};
		\node[left] at (0,-1) {$q_1$};
		\draw (2,-0.35) rectangle (3,0.35); \node at (2.5,0) {$H$};
		\draw (2,-1.35) rectangle (3,-0.65); \node at (2.5,-1) {$X$};
	\end{tikzpicture}
\end{center}
This is $H\otimes X$.
Apply it to $\ket{00}$:
\[
(H\otimes X)\ket{00}=(H\ket0)\otimes(X\ket0)=\ket{+}\otimes\ket1=\frac{\ket{01}+\ket{11}}{\sqrt2}.
\]

\subsection{The CNOT gate as a linear map}
\subsubsection{Definition in the computational basis}
CNOT (control = first qubit, target = second qubit) is defined by
\[
\ket{a}\ket{b}\ \mapsto\ \ket{a}\ket{b\oplus a},
\quad a,b\in\{0,1\}.
\]
So
\[
\ket{00}\mapsto\ket{00},\ 
\ket{01}\mapsto\ket{01},\ 
\ket{10}\mapsto\ket{11},\ 
\ket{11}\mapsto\ket{10}.
\]
Matrix (in basis $\ket{00},\ket{01},\ket{10},\ket{11}$):
\[
\mathrm{CNOT}=
\begin{pmatrix}
	1&0&0&0\\
	0&1&0&0\\
	0&0&0&1\\
	0&0&1&0
\end{pmatrix}.
\]

\subsubsection{Circuit symbol and intuition}
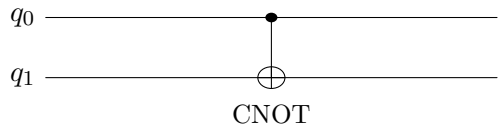
\begin{center}
	\begin{tikzpicture}[x=1.0cm,y=0.8cm, line cap=round]
		\draw (0,0) -- (6,0);
		\draw (0,-1) -- (6,-1);
		\node[left] at (0,0) {$q_0$};
		\node[left] at (0,-1) {$q_1$};
		\fill (3,0) circle (0.08);
		\draw (3,0) -- (3,-1);
		\draw (3,-1) circle (0.18);
		\draw (3-0.18,-1) -- (3+0.18,-1);
		\draw (3,-1-0.18) -- (3,-1+0.18);
		\node at (3,-1.6) {\small CNOT};
	\end{tikzpicture}
\end{center}

\subsubsection{Worked example: CNOT creates entanglement}
Start with $\ket{00}$ and apply $H$ on the control, then CNOT:
\[
\ket{00}\xrightarrow{H\otimes I}\frac{\ket{00}+\ket{10}}{\sqrt2}
\xrightarrow{\mathrm{CNOT}}\frac{\ket{00}+\ket{11}}{\sqrt2}=\ket{\Phi^+}.
\]
Circuit:
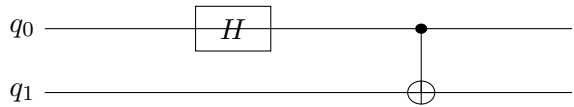
\begin{center}
	\begin{tikzpicture}[x=1.0cm,y=0.85cm, line cap=round]
		\draw (0,0) -- (7,0);
		\draw (0,-1) -- (7,-1);
		\node[left] at (0,0) {$q_0$};
		\node[left] at (0,-1) {$q_1$};
		\draw (2,-0.35) rectangle (3,0.35); \node at (2.5,0) {$H$};
		\fill (5,0) circle (0.08);
		\draw (5,0) -- (5,-1);
		\draw (5,-1) circle (0.18);
		\draw (4.82,-1) -- (5.18,-1);
		\draw (5,-1.18) -- (5,-0.82);
	\end{tikzpicture}
\end{center}

\paragraph{Geometric reading via reduced Bloch vectors.}
For the Bell state $\ket{\Phi^+}$, each single-qubit reduced density matrix is maximally mixed:
\[
\rho_{0}=\Tr_{1}\bigl(\ket{\Phi^+}\bra{\Phi^+}\bigr)=\frac{I}{2},
\qquad
\rho_{1}=\Tr_{0}\bigl(\ket{\Phi^+}\bra{\Phi^+}\bigr)=\frac{I}{2}.
\]
Thus each qubit has Bloch vector $\vec r=\vec 0$ (center of the Bloch ball).
This is the cleanest geometric signal that the \emph{joint pure state is entangled}:
globally pure, locally maximally mixed.

\subsection{Hardware-efficient ansatz (HEA) in matrix language}
\subsubsection{A canonical 2-qubit HEA layer}
A common HEA layer is:
(1) single-qubit rotations on each qubit, then (2) entangling block (e.g.\ CNOT chain).
For two qubits:
\[
U(\boldsymbol{\theta})
=
\mathrm{CNOT}\cdot
\bigl(R_Z(\theta_1)R_Y(\theta_0)\bigr)\otimes
\bigl(R_Z(\theta_3)R_Y(\theta_2)\bigr).
\]
Circuit sketch:
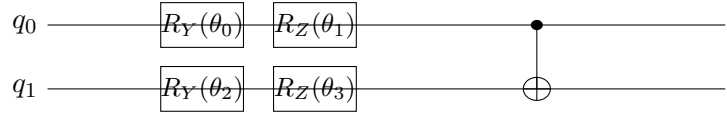
\begin{center}
	\begin{tikzpicture}[x=1.0cm,y=0.85cm, line cap=round]
		\draw (0,0) -- (9,0);
		\draw (0,-1) -- (9,-1);
		\node[left] at (0,0) {$q_0$};
		\node[left] at (0,-1) {$q_1$};
		\draw (1.5,-0.35) rectangle (2.6,0.35); \node at (2.05,0) {\small $R_Y(\theta_0)$};
		\draw (3.0,-0.35) rectangle (4.1,0.35); \node at (3.55,0) {\small $R_Z(\theta_1)$};
		\draw (1.5,-1.35) rectangle (2.6,-0.65); \node at (2.05,-1) {\small $R_Y(\theta_2)$};
		\draw (3.0,-1.35) rectangle (4.1,-0.65); \node at (3.55,-1) {\small $R_Z(\theta_3)$};
		\fill (6.5,0) circle (0.08);
		\draw (6.5,0) -- (6.5,-1);
		\draw (6.5,-1) circle (0.18);
		\draw (6.32,-1) -- (6.68,-1);
		\draw (6.5,-1.18) -- (6.5,-0.82);
	\end{tikzpicture}
\end{center}

\subsubsection{Geometry comment (why HEA is ``paths on a manifold'')}
Each $R_Y,R_Z$ is a Bloch-sphere rotation for the corresponding qubit \emph{before entangling}.
The entangler (CNOT) couples the two subsystems so that the joint state moves on $\CP^{3}$
(for pure 2-qubit states), while single-qubit reduced Bloch vectors typically shrink (mixedness increases)
when entanglement is generated.

\subsection{Exercises}

\begin{exercise}
	Write the $4\times4$ matrix for $(H\otimes I)\cdot \mathrm{CNOT}$ and apply it to $\ket{00}$.
\end{exercise}

\noindent\textbf{Solution.}
We work in the ordered computational basis
\[
\ket{00},\ \ket{01},\ \ket{10},\ \ket{11}.
\]

\smallskip
\noindent\textbf{Step 1: write the matrices.}
The CNOT (control = first qubit, target = second qubit) is
\[
\mathrm{CNOT}=
\begin{pmatrix}
	1&0&0&0\\
	0&1&0&0\\
	0&0&0&1\\
	0&0&1&0
\end{pmatrix}.
\]
Also
\[
H=\frac1{\sqrt2}\begin{pmatrix}1&1\\1&-1\end{pmatrix},
\qquad
I=\begin{pmatrix}1&0\\0&1\end{pmatrix}.
\]
Compute the tensor product:
\[
H\otimes I
=
\frac1{\sqrt2}
\begin{pmatrix}
	1\cdot I & 1\cdot I\\
	1\cdot I & -1\cdot I
\end{pmatrix}
=
\frac1{\sqrt2}
\begin{pmatrix}
	1&0&1&0\\
	0&1&0&1\\
	1&0&-1&0\\
	0&1&0&-1
\end{pmatrix}.
\]

\smallskip
\noindent\textbf{Step 2: multiply $(H\otimes I)\mathrm{CNOT}$.}
A fast way is to multiply by columns, because $\mathrm{CNOT}$ permutes basis vectors:
\[
\mathrm{CNOT}\,e_1=e_1,\quad \mathrm{CNOT}\,e_2=e_2,\quad \mathrm{CNOT}\,e_3=e_4,\quad \mathrm{CNOT}\,e_4=e_3,
\]
where $e_1,e_2,e_3,e_4$ correspond to $\ket{00},\ket{01},\ket{10},\ket{11}$.
Hence $(H\otimes I)\mathrm{CNOT}$ is obtained from $H\otimes I$ by swapping the 3rd and 4th columns:
\[
(H\otimes I)\mathrm{CNOT}
=
\frac1{\sqrt2}
\begin{pmatrix}
	1&0&0&1\\
	0&1&1&0\\
	1&0&0&-1\\
	0&1&-1&0
\end{pmatrix}.
\]

\smallskip
\noindent\textbf{Step 3: apply it to $\ket{00}$.}
In this basis, $\ket{00}=e_1=(1,0,0,0)^T$, so
\[
(H\otimes I)\mathrm{CNOT}\,\ket{00}
=
\frac1{\sqrt2}
\begin{pmatrix}
	1\\0\\1\\0
\end{pmatrix}
=
\frac{\ket{00}+\ket{10}}{\sqrt2}
=
\ket{+}\otimes\ket{0}.
\]

\smallskip
\noindent\textbf{Sanity check (interpretation).}
This makes sense: $\mathrm{CNOT}$ does nothing to $\ket{00}$, then $H$ on the first qubit produces $\ket{+}$.

\begin{exercise}
	Compute $p(0)$ and $p(1)$ for the circuit
	\[
	\ket0 \xrightarrow{H} \xrightarrow{R_Z(\theta)} \xrightarrow{H} \text{measure } Z.
	\]
\end{exercise}

\noindent\textbf{Solution.}
We compute the final state explicitly and read off the measurement probabilities.

\smallskip
\noindent\textbf{Step 1: apply $H$ to $\ket0$.}
\[
H\ket0=\ket{+}=\frac{\ket0+\ket1}{\sqrt2}.
\]

\smallskip
\noindent\textbf{Step 2: apply $R_Z(\theta)$.}
Using
\[
R_Z(\theta)=e^{-i\frac{\theta}{2}Z}
=
\begin{pmatrix}
	e^{-i\theta/2} & 0\\
	0 & e^{i\theta/2}
\end{pmatrix},
\]
we get
\[
R_Z(\theta)\ket{+}
=
\frac{e^{-i\theta/2}\ket0+e^{i\theta/2}\ket1}{\sqrt2}.
\]

\smallskip
\noindent\textbf{Step 3: apply $H$ again.}
Recall
\[
H\ket0=\frac{\ket0+\ket1}{\sqrt2},\qquad
H\ket1=\frac{\ket0-\ket1}{\sqrt2}.
\]
So
\begin{align*}
	H\,R_Z(\theta)\ket{+}
	&=
	\frac{1}{\sqrt2}\left(e^{-i\theta/2}H\ket0+e^{i\theta/2}H\ket1\right)\\
	&=
	\frac{1}{\sqrt2}\left(
	e^{-i\theta/2}\frac{\ket0+\ket1}{\sqrt2}
	+
	e^{i\theta/2}\frac{\ket0-\ket1}{\sqrt2}
	\right)\\
	&=
	\frac{1}{2}\left((e^{-i\theta/2}+e^{i\theta/2})\ket0 + (e^{-i\theta/2}-e^{i\theta/2})\ket1\right)\\
	&=
	\frac{1}{2}\left(2\cos(\theta/2)\ket0 -2i\sin(\theta/2)\ket1\right)\\
	&=
	\cos(\theta/2)\ket0 - i\sin(\theta/2)\ket1.
\end{align*}

\smallskip
\noindent\textbf{Step 4: measure $Z$.}
The probability of outcome $0$ is the squared magnitude of the $\ket0$ amplitude:
\[
p(0)=|\cos(\theta/2)|^2=\cos^2(\theta/2)=\frac{1+\cos\theta}{2}.
\]
Similarly,
\[
p(1)=|{-i\sin(\theta/2)}|^2=\sin^2(\theta/2)=\frac{1-\cos\theta}{2}.
\]

\smallskip
\noindent\textbf{Geometric meaning.}
This circuit converts a $Z$-axis rotation (equatorial motion on the Bloch sphere) into a change in the final $z$-projection, which is exactly what $Z$-measurement reads out.

\begin{exercise}
	Show that $\mathrm{CNOT}(\ket{+}\ket0)=\ket{\Phi^+}$ and compute the reduced density matrix of the first qubit.
\end{exercise}

\noindent\textbf{Solution.}

\smallskip
\noindent\textbf{Part A: compute the output state.}
First write
\[
\ket{+}\ket0=\left(\frac{\ket0+\ket1}{\sqrt2}\right)\otimes\ket0
=
\frac{\ket{00}+\ket{10}}{\sqrt2}.
\]
CNOT acts as $\ket{a}\ket{b}\mapsto\ket{a}\ket{b\oplus a}$, hence
\[
\mathrm{CNOT}\ket{00}=\ket{00},\qquad
\mathrm{CNOT}\ket{10}=\ket{11}.
\]
Therefore
\[
\mathrm{CNOT}(\ket{+}\ket0)
=
\frac{\ket{00}+\ket{11}}{\sqrt2}
=:\ket{\Phi^+}.
\]

\smallskip
\noindent\textbf{Part B: compute the reduced density matrix $\rho_0$ of the first qubit.}
Let
\[
\rho=\ket{\Phi^+}\bra{\Phi^+}
=
\frac12\Bigl(\ket{00}\bra{00}+\ket{00}\bra{11}+\ket{11}\bra{00}+\ket{11}\bra{11}\Bigr).
\]
We trace out the second qubit. Use the rule (basis $\{\ket0,\ket1\}$ for the second qubit):
\[
\Tr_1\!\bigl(\ket{a b}\bra{c d}\bigr)=\langle d|b\rangle\,\ket{a}\bra{c}
=
\delta_{b d}\,\ket{a}\bra{c}.
\]
Apply this to each term:

\begin{itemize}
	\item $\Tr_1(\ket{00}\bra{00})=\delta_{0,0}\ket0\bra0=\ket0\bra0$.
	\item $\Tr_1(\ket{00}\bra{11})=\delta_{0,1}\ket0\bra1=0$.
	\item $\Tr_1(\ket{11}\bra{00})=\delta_{1,0}\ket1\bra0=0$.
	\item $\Tr_1(\ket{11}\bra{11})=\delta_{1,1}\ket1\bra1=\ket1\bra1$.
\end{itemize}
So
\[
\rho_0=\Tr_1(\rho)=\frac12\left(\ket0\bra0+\ket1\bra1\right)=\frac{I}{2}.
\]

\smallskip
\noindent\textbf{Interpretation (Bloch-sphere).}
$\rho_0=I/2$ has Bloch vector $\vec r=\vec 0$ (center of the Bloch ball), i.e.\ the first qubit is maximally mixed.
This is the clean hallmark of entanglement: the \emph{joint} state is pure, but each \emph{part} is mixed.

\begin{exercise}
	Verify explicitly that
	\[
	(A\otimes B)(\ket v\otimes\ket w)=(A\ket v)\otimes(B\ket w)
	\]
	and use it to compute $(H\otimes X)\ket{01}$.
\end{exercise}

\noindent\textbf{Solution.}

\smallskip
\noindent\textbf{Part A: verify the identity from coordinates.}
Write $\ket v=\sum_i v_i\ket i$ and $\ket w=\sum_j w_j\ket j$ in the computational basis ($i,j\in\{0,1\}$).
Then
\[
\ket v\otimes\ket w
=
\sum_{i,j} v_i w_j\,\ket i\otimes\ket j.
\]
Apply $A\otimes B$ linearly and use the defining action on basis tensors:
\[
(A\otimes B)(\ket i\otimes\ket j)=(A\ket i)\otimes(B\ket j).
\]
Thus
\[
(A\otimes B)(\ket v\otimes\ket w)
=
\sum_{i,j} v_i w_j\,(A\ket i)\otimes(B\ket j)
=
\left(\sum_i v_i A\ket i\right)\otimes\left(\sum_j w_j B\ket j\right)
=
(A\ket v)\otimes(B\ket w),
\]
which proves the identity.

\smallskip
\noindent\textbf{Part B: compute $(H\otimes X)\ket{01}$.}
We have $\ket{01}=\ket0\otimes\ket1$. Therefore
\[
(H\otimes X)\ket{01}=(H\ket0)\otimes(X\ket1).
\]
Compute each piece:
\[
H\ket0=\ket{+}=\frac{\ket0+\ket1}{\sqrt2},\qquad
X\ket1=\ket0.
\]
So
\[
(H\otimes X)\ket{01}
=
\ket{+}\otimes\ket0
=
\frac{\ket{00}+\ket{10}}{\sqrt2}.
\]

\smallskip
\noindent\textbf{Optional check (as a $4$-vector).}
In the basis $(\ket{00},\ket{01},\ket{10},\ket{11})$,
\[
\ket{01}=(0,1,0,0)^T,\qquad
(H\otimes X)=\frac1{\sqrt2}
\begin{pmatrix}
	0&1&0&1\\
	1&0&1&0\\
	0&1&0&-1\\
	1&0&-1&0
\end{pmatrix},
\]
and multiplying indeed gives $\frac1{\sqrt2}(1,0,1,0)^T$.

\section{Quantum Circuits II: Interference, Cancellation, and Measurement}
\label{sec:circuits2}

\subsection{Objective}

This chapter explains the single most important phenomenon behind quantum speedups:
\emph{interference of amplitudes}. The core idea is simple but subtle:

\begin{quote}
	\emph{Probabilities never interfere. Amplitudes interfere, and probabilities are squared magnitudes of amplitudes.}
\end{quote}

\subsection{Setup: amplitudes, probabilities, and one-line rules}

We work in the computational basis $\{\ket{x}\}$, where $x\in\{0,1\}^n$.
A general pure $n$-qubit state is
\[
\ket{\psi}=\sum_{x\in\{0,1\}^n} \alpha_x\,\ket{x},
\qquad
\sum_x |\alpha_x|^2=1.
\]

\medskip
\noindent\textbf{Born rule (one line).}
If you measure in the computational basis, the probability of outcome $x$ is
\[
p(x)=|\alpha_x|^2.
\]

\medskip
\noindent\textbf{Unitary evolution (one line).}
A circuit is a unitary $U$; it maps amplitudes linearly:
\[
\ket{\psi}\longmapsto U\ket{\psi}.
\]

\medskip
\noindent\textbf{Interference (one line).}
If an amplitude for outcome $x$ is a \emph{sum} of contributions,
\[
\alpha_x = a_x + b_x + \cdots,
\]
then the probability is
\[
p(x)=|\alpha_x|^2 = |a_x+b_x+\cdots|^2,
\]
and \emph{cross terms} appear. This is the origin of interference.

\medskip
\noindent\textbf{Global phase (one line).}
$\ket{\psi}$ and $e^{i\theta}\ket{\psi}$ produce identical measurement probabilities.

\medskip
\noindent\textbf{Relative phase (one line).}
Relative phase between components \emph{does} affect future interference, hence affects future probabilities after additional gates.

\subsection{Superposition is about amplitudes, not probabilities}

\subsubsection*{Superposition vs.\ classical mixture}

A coherent superposition of two basis states is
\[
\ket{\psi}=\frac{\ket{0}+\ket{1}}{\sqrt2}=\ket{+}.
\]
A classical mixture that outputs $\ket{0}$ with prob.\ $1/2$ and $\ket{1}$ with prob.\ $1/2$
is represented by a density matrix
\[
\rho_{\text{mix}}=\frac12\ket0\bra0+\frac12\ket1\bra1=\frac{I}{2}.
\]
They have identical $Z$-measurement probabilities:
\[
p_{\ket{+}}(0)=p_{\rho_{\text{mix}}}(0)=\frac12,
\qquad
p_{\ket{+}}(1)=p_{\rho_{\text{mix}}}(1)=\frac12.
\]
But they are \emph{physically different} because they respond differently to basis changes.

\begin{prop}[Basis change distinguishes coherence]
	Let $H$ be the Hadamard gate. Then measuring $Z$ after applying $H$ gives:
	\[
	H\ket{+}=\ket0 \quad\Rightarrow\quad p(0)=1,
	\]
	but
	\[
	H\rho_{\text{mix}}H^\dagger=\frac{I}{2}\quad\Rightarrow\quad p(0)=\frac12.
	\]
\end{prop}

\begin{proof}
	First,
	\[
	H\ket{+}=H\left(\frac{\ket0+\ket1}{\sqrt2}\right)
	=
	\frac{H\ket0+H\ket1}{\sqrt2}
	=
	\frac{\ket{+}+\ket{-}}{\sqrt2}
	=
	\ket0.
	\]
	So $Z$-measurement outputs $0$ with certainty.
	
	For the mixture, use $H I H^\dagger=I$ so $\rho$ stays maximally mixed under any unitary:
	\[
	H\left(\frac{I}{2}\right)H^\dagger=\frac{I}{2}.
	\]
	Hence $Z$-measurement remains uniform.
\end{proof}

\begin{rem}[Geometric meaning (Bloch sphere)]
	$\ket{+}$ is a point on the \emph{surface} (pure state), while $I/2$ is the \emph{center} (maximally mixed).
	Hadamard is a rotation of the Bloch sphere; pure points move, the center does not.
\end{rem}

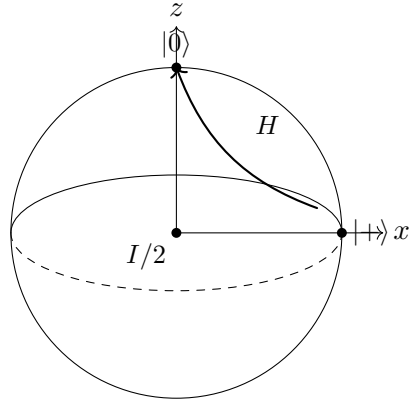
\begin{figure}[t]
	\centering
	\begin{tikzpicture}[scale=2.2, line cap=round, line join=round]
		\draw (0,0) circle (1);
		\draw[dashed] (-1,0) arc (180:360:1 and 0.35);
		\draw (-1,0) arc (180:0:1 and 0.35);
		\draw[->] (0,0) -- (1.25,0) node[right] {$x$};
		\draw[->] (0,0) -- (0,1.25) node[above] {$z$};
		\fill (1,0) circle (0.03) node[right] {\small $\ket{+}$};
		\fill (0,0) circle (0.03) node[below left] {\small $I/2$};
		\draw[->, thick] (0.85,0.15) to[bend left=25] (0,1.0);
		\node at (0.55,0.65) {\small $H$};
		\fill (0,1) circle (0.03) node[above] {\small $\ket0$};
	\end{tikzpicture}
	\caption{Superposition vs.\ mixture: $H$ rotates the pure state $\ket{+}$ to $\ket0$ (certain outcome),
		while the maximally mixed state $I/2$ remains at the center (always uniform).}
	\label{fig:superposition-vs-mixture}
\end{figure}

\subsection{Constructive vs.\ destructive interference}

\subsubsection*{The two-path calculation (the interference identity)}

Suppose an amplitude for outcome $\ket{x}$ is the sum of two contributions:
\[
\alpha_x = a_x + b_x.
\]
Then
\[
p(x)=|\alpha_x|^2 = |a_x+b_x|^2 = |a_x|^2 + |b_x|^2 + 2\Re(a_x\overline{b_x}).
\]
The last term is the interference term:
\begin{itemize}
	\item \textbf{constructive} if $\Re(a_x\overline{b_x})>0$,
	\item \textbf{destructive} if $\Re(a_x\overline{b_x})<0$,
	\item \textbf{perfect cancellation} if $a_x=-b_x$.
\end{itemize}

\subsubsection*{Canonical interference demo: $H^2=I$ as cancellation}

A single qubit:
\[
\ket0 \xrightarrow{H} \ket{+} \xrightarrow{H} \ket0.
\]
This is not magic; it is a cancellation of amplitude contributions.

\begin{prop}[Hadamard squared equals identity]
	\[
	H^2=I.
	\]
\end{prop}

\begin{proof}
	Compute
	\[
	H^2=\frac12
	\begin{pmatrix}
		1&1\\
		1&-1
	\end{pmatrix}
	\begin{pmatrix}
		1&1\\
		1&-1
	\end{pmatrix}
	=
	\frac12
	\begin{pmatrix}
		2&0\\
		0&2
	\end{pmatrix}
	=I.
	\]
\end{proof}

\begin{rem}[Where the cancellation is hiding]
	Track the amplitude of $\ket1$ after two Hadamards:
	\[
	\ket0 \xrightarrow{H} \frac{\ket0+\ket1}{\sqrt2}
	\xrightarrow{H}
	\frac{1}{\sqrt2}\left(\frac{\ket0+\ket1}{\sqrt2}+\frac{\ket0-\ket1}{\sqrt2}\right)
	=
	\ket0.
	\]
	The $\ket1$ contributions appear with opposite signs and cancel.
\end{rem}

\begin{figure}[t]
	\centering
	\begin{tikzpicture}[scale=1.0]
		\node at (0,0) {$\ket0$};
		\node[draw, rounded corners, minimum width=1.0cm, minimum height=0.7cm] (H1) at (2,0) {$H$};
		\node[draw, rounded corners, minimum width=1.0cm, minimum height=0.7cm] (H2) at (4,0) {$H$};
		\node at (6,0) {$\ket0$};
		
		\draw[->] (0.5,0) -- (H1.west);
		\draw[->] (H1.east) -- (H2.west) node[midway, above] {\small $\ket{+}$};
		\draw[->] (H2.east) -- (5.5,0);
	\end{tikzpicture}
	\caption{$H^2=I$: a minimal interference/cancellation circuit.}
	\label{fig:h2}
\end{figure}
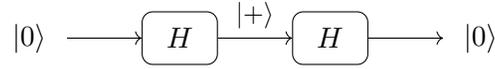

\subsubsection*{Interference as a design pattern}

To build an algorithm, you often:
\begin{enumerate}
	\item create a superposition (many computational paths),
	\item apply phase factors so ``bad'' paths pick up minus signs (or complex phases),
	\item recombine so that bad paths cancel and good paths add,
	\item measure.
\end{enumerate}
Grover's algorithm and phase estimation are refined versions of this pattern.

\subsection{Measurement as projection (and why phase is lost)}

\subsubsection*{Projective measurement in the computational basis}

Let $\ket{\psi}=\sum_x \alpha_x\ket{x}$.
Measuring in the computational basis returns outcome $x$ with probability $|\alpha_x|^2$,
and the post-measurement state becomes $\ket{x}$ (ideal projective measurement).

\begin{defn}[Projectors]
	For a single qubit (computational basis),
	\[
	\Pi_0=\ket0\bra0,\qquad \Pi_1=\ket1\bra1,
	\qquad \Pi_0+\Pi_1=I.
	\]
\end{defn}

\begin{prop}[Measurement probabilities and post-measurement update]
	For a pure state $\ket{\psi}$,
	\[
	p(0)=\bra{\psi}\Pi_0\ket{\psi},\qquad p(1)=\bra{\psi}\Pi_1\ket{\psi}.
	\]
	Conditioned on outcome $k\in\{0,1\}$, the post-measurement state is
	\[
	\ket{\psi}\ \mapsto\ \frac{\Pi_k\ket{\psi}}{\sqrt{\bra{\psi}\Pi_k\ket{\psi}}}.
	\]
\end{prop}

\subsubsection*{Why phase is lost}

Take
\[
\ket{\psi_\phi}=\frac{\ket0+e^{i\phi}\ket1}{\sqrt2}.
\]
A $Z$-measurement gives
\[
p(0)=\frac12,\qquad p(1)=\frac12
\]
independent of $\phi$. This is the precise sense in which $Z$-measurement discards relative phase information.

However, $\phi$ \emph{is real and usable}: apply a basis change and it becomes visible.

\begin{prop}[Phase becomes visible after a basis change]
	Let $\ket{\psi_\phi}=\frac{\ket0+e^{i\phi}\ket1}{\sqrt2}$.
	After applying $H$ and measuring $Z$,
	\[
	p(0)=\frac{1+\cos\phi}{2},\qquad p(1)=\frac{1-\cos\phi}{2}.
	\]
\end{prop}

\begin{proof}
	Compute
	\[
	H\ket{\psi_\phi}
	=
	\frac{1}{\sqrt2}\bigl(H\ket0+e^{i\phi}H\ket1\bigr)
	=
	\frac{1}{\sqrt2}\left(\frac{\ket0+\ket1}{\sqrt2}+e^{i\phi}\frac{\ket0-\ket1}{\sqrt2}\right)
	=
	\frac{(1+e^{i\phi})\ket0+(1-e^{i\phi})\ket1}{2}.
	\]
	Thus
	\[
	p(0)=\left|\frac{1+e^{i\phi}}{2}\right|^2=\frac{1+\cos\phi}{2},
	\qquad
	p(1)=\left|\frac{1-e^{i\phi}}{2}\right|^2=\frac{1-\cos\phi}{2}.
	\]
\end{proof}

\begin{rem}[Bloch-sphere interpretation]
	$\ket{\psi_\phi}$ lies on the equator; $\phi$ is the azimuth angle.
	A $Z$ measurement only sees the $z$-coordinate; it cannot see azimuth.
	Hadamard rotates axes so that azimuthal motion becomes a change in the measured projection.
\end{rem}

\subsection{Algorithmic examples}

\subsubsection*{Example 1: ``phase-to-amplitude'' conversion (the $H$--phase--$H$ pattern)}

The circuit
\[
\ket0 \xrightarrow{H} \xrightarrow{R_Z(\theta)} \xrightarrow{H} \text{measure }Z
\]
produces
\[
p(0)=\frac{1+\cos\theta}{2},\qquad p(1)=\frac{1-\cos\theta}{2}.
\]
This is the smallest example of the central algorithmic pattern:
\emph{convert an invisible phase into a visible probability bias.}

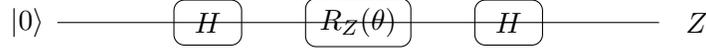
\begin{figure}[t]
	\centering
	\begin{tikzpicture}[scale=1.0]
		\draw (0,0) -- (8,0);
		\node at (-0.4,0) {$\ket0$};
		\node[draw, rounded corners, minimum width=0.9cm, minimum height=0.6cm] at (2,0) {$H$};
		\node[draw, rounded corners, minimum width=1.4cm, minimum height=0.6cm] at (4,0) {$R_Z(\theta)$};
		\node[draw, rounded corners, minimum width=0.9cm, minimum height=0.6cm] at (6,0) {$H$};
		\node at (8.5,0) {$Z$};
	\end{tikzpicture}
	\caption{Phase-to-amplitude conversion: $H$--$R_Z(\theta)$--$H$ turns a phase parameter into a measurable $Z$-bias.}
	\label{fig:phase-to-amplitude}
\end{figure}

\subsubsection*{Example 2: two-qubit interference and entanglement creation (Bell state)}

Prepare $\ket{00}$, apply $H$ on the first qubit, then CNOT:
\[
\ket{00}\xrightarrow{H\otimes I}\frac{\ket{00}+\ket{10}}{\sqrt2}
\xrightarrow{\mathrm{CNOT}}\frac{\ket{00}+\ket{11}}{\sqrt2}
=\ket{\Phi^+}.
\]
Here ``two paths'' $\ket{00}$ and $\ket{10}$ are recombined into correlated outcomes.

\begin{figure}[t]
	\centering
	\begin{tikzpicture}[scale=1.0]
		\draw (0,1) -- (8,1);
		\draw (0,0) -- (8,0);
		\node at (-0.4,1) {$\ket0$};
		\node at (-0.4,0) {$\ket0$};
		\node[draw, rounded corners, minimum width=0.9cm, minimum height=0.6cm] at (2,1) {$H$};
		\fill (5,1) circle (0.07);
		\draw (5,1) -- (5,0);
		\draw (5,0) circle (0.18);
		\draw (5-0.18,0) -- (5+0.18,0);
		\draw (5,0-0.18) -- (5,0+0.18);
		\node at (8.6,1) {$Z$};
		\node at (8.6,0) {$Z$};
	\end{tikzpicture}
	\caption{Bell-state circuit: $(H\otimes I)$ then CNOT creates $\ket{\Phi^+}=(\ket{00}+\ket{11})/\sqrt2$.}
	\label{fig:bell-circuit}
\end{figure}
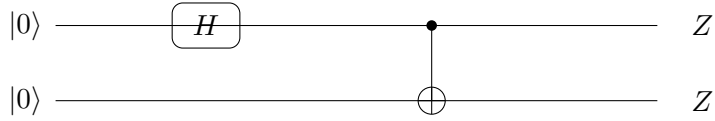

\subsubsection*{Example 3: cancellation by design (the ``minus sign trick'')}

Consider
\[
\ket0 \xrightarrow{H} \frac{\ket0+\ket1}{\sqrt2}
\xrightarrow{Z} \frac{\ket0-\ket1}{\sqrt2}=\ket{-}
\xrightarrow{H} \ket1.
\]
So $HZH\ket0=\ket1$ deterministically.
The gate $Z$ inserts a relative phase ($-1$ on $\ket1$), which flips constructive into destructive interference
on the final recombination.

\begin{figure}[t]
	\centering
	\begin{tikzpicture}[scale=1.0]
		\draw (0,0) -- (8,0);
		\node at (-0.4,0) {$\ket0$};
		\node[draw, rounded corners, minimum width=0.9cm, minimum height=0.6cm] at (2,0) {$H$};
		\node[draw, rounded corners, minimum width=0.9cm, minimum height=0.6cm] at (4,0) {$Z$};
		\node[draw, rounded corners, minimum width=0.9cm, minimum height=0.6cm] at (6,0) {$H$};
		\node at (8.5,0) {$\ket1$};
	\end{tikzpicture}
	\caption{A single inserted minus sign ($Z$) flips interference: $HZH\ket0=\ket1$.}
	\label{fig:hzh}
\end{figure}
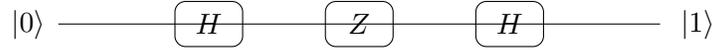

\subsection{Exercises}

\begin{exercise}[Interference identity]
	Let $\alpha=a+b$ with $a,b\in\C$. Show
	\[
	|\alpha|^2 = |a|^2+|b|^2+2\Re(a\overline b).
	\]
	Give an example where the interference term is (i) positive, (ii) negative, (iii) makes the probability zero.
\end{exercise}

\noindent\textbf{Solution.}
Expand:
\[
|a+b|^2=(a+b)\overline{(a+b)}=a\overline a+b\overline b+a\overline b+\overline a b
=|a|^2+|b|^2+2\Re(a\overline b).
\]
Examples: take $a=b=1$ gives positive interference ($|2|^2=4$).
Take $a=1$, $b=-1/2$ gives negative interference ($|1-1/2|^2=1/4<|1|^2+|1/2|^2$).
Take $a=1$, $b=-1$ gives perfect cancellation ($|0|^2=0$).

\begin{exercise}[Superposition vs.\ mixture under Hadamard]
	Compute the measurement distribution of $Z$ after applying $H$ to:
	(i) $\ket{+}$, and (ii) the mixed state $\rho=\frac12\ket0\bra0+\frac12\ket1\bra1$.
\end{exercise}

\noindent\textbf{Solution.}
(i) $H\ket{+}=\ket0$, so $p(0)=1$.
(ii) $H\rho H^\dagger=\rho=I/2$, so $p(0)=p(1)=1/2$.

\begin{exercise}[Phase invisibility and visibility]
	Let $\ket{\psi_\phi}=(\ket0+e^{i\phi}\ket1)/\sqrt2$.
	Show that a $Z$-measurement gives uniform outcomes, but an $H$ then $Z$ measurement gives
	\[
	p(0)=\frac{1+\cos\phi}{2},\qquad p(1)=\frac{1-\cos\phi}{2}.
	\]
\end{exercise}

\noindent\textbf{Solution.}
Uniform for $Z$ because amplitudes are $(1/\sqrt2,e^{i\phi}/\sqrt2)$ so probabilities are both $1/2$.
After $H$, compute as in the proof above:
\[
H\ket{\psi_\phi}=\frac{(1+e^{i\phi})\ket0+(1-e^{i\phi})\ket1}{2},
\]
then square magnitudes.

\begin{exercise}[Bell state marginals]
	Let $\ket{\Phi^+}=(\ket{00}+\ket{11})/\sqrt2$ and $\rho=\ket{\Phi^+}\bra{\Phi^+}$.
	Compute $\rho_0=\Tr_1(\rho)$ and $\rho_1=\Tr_0(\rho)$ and interpret geometrically on the Bloch ball.
\end{exercise}

\noindent\textbf{Solution.}
As computed earlier,
\[
\Tr_1(\rho)=\frac12(\ket0\bra0+\ket1\bra1)=\frac{I}{2}.
\]
By symmetry $\Tr_0(\rho)=I/2$ too.
Geometrically each qubit is at the center of the Bloch ball (maximally mixed) even though the joint state is pure.

\begin{exercise}[Minus sign flips the outcome]
	Show that $HZH\ket0=\ket1$ and $HZH\ket1=\ket0$.
	Conclude that $HZH=X$.
\end{exercise}

\noindent\textbf{Solution.}
Compute
\[
H\ket0=\frac{\ket0+\ket1}{\sqrt2},\qquad
Z(H\ket0)=\frac{\ket0-\ket1}{\sqrt2},\qquad
H\left(\frac{\ket0-\ket1}{\sqrt2}\right)=\ket1.
\]
Similarly,
\[
H\ket1=\frac{\ket0-\ket1}{\sqrt2},\quad
Z(H\ket1)=\frac{\ket0+\ket1}{\sqrt2},\quad
H\left(\frac{\ket0+\ket1}{\sqrt2}\right)=\ket0.
\]
Thus $HZH$ swaps $\ket0$ and $\ket1$, so it equals $X$.

\subsection{Optional drill problems}

\begin{exercise}[Drill: amplitude bookkeeping]
	Compute the final state (as amplitudes in the computational basis) for each circuit on input $\ket0$:
	\begin{enumerate}
		\item $H\to H$,
		\item $H\to Z\to H$,
		\item $H\to R_Z(\theta)\to H$,
		\item $H\to X\to H$.
	\end{enumerate}
\end{exercise}

\noindent\textbf{Solutions.}
\begin{enumerate}
	\item $H^2\ket0=\ket0$.
	\item $HZH\ket0=\ket1$.
	\item As derived: $\cos(\theta/2)\ket0-i\sin(\theta/2)\ket1$.
	\item Use conjugation $HXH=Z$; hence $HXH\ket0=Z\ket0=\ket0$.
\end{enumerate}

\begin{exercise}[Drill: which measurements see phase?]
	Let $\ket{\psi_\phi}=(\ket0+e^{i\phi}\ket1)/\sqrt2$.
	Compute $\langle X\rangle$, $\langle Y\rangle$, $\langle Z\rangle$ and identify which ones depend on $\phi$.
\end{exercise}

\noindent\textbf{Solution.}
This state lies on the equator, so $z=0$.
Compute using Bloch coordinates:
\[
\langle X\rangle=\cos\phi,\qquad
\langle Y\rangle=\sin\phi,\qquad
\langle Z\rangle=0.
\]
Thus $X,Y$ measurements see phase; $Z$ does not.

\section{Quantum Circuits III: Circuit Identities, Simplification, and Compilation}
\label{sec:circuits3}

\subsection{Objective}

This chapter builds the ``algebra of circuits'' that you will use constantly:
\emph{rewrite}, \emph{simplify}, and \emph{compile} circuits without changing their action.
The guiding principle is:

\begin{quote}
	\emph{Circuit identities are matrix identities (up to global phase).}
\end{quote}

\subsection{Circuit identities are matrix identities}

\subsubsection*{Equality conventions}

Two circuits $C_1$ and $C_2$ acting on the same number of qubits are \emph{exactly equal} if
their matrices are equal:
\[
U(C_1)=U(C_2).
\]
In many algorithmic contexts, we treat circuits as equivalent \emph{up to global phase}:
\[
U(C_1)=e^{i\theta}U(C_2),
\]
because a global phase does not change measurement statistics.

\begin{rem}[When global phase matters]
	Global phase never affects any measurement probabilities.
	However, when circuits are used as subroutines inside larger interference patterns,
	a \emph{relative} phase between branches matters.
	A global phase on the entire unitary still does not matter, but be careful not to confuse
	``global phase'' with ``a phase on a subspace''.
\end{rem}

\subsubsection*{Foundational one-qubit identities}

We will use these repeatedly:
\[
H^2=I,\qquad X^2=Y^2=Z^2=I,
\]
\[
HXH=Z,\qquad HZH=X,\qquad HYH=-Y,
\]
and (up to global phase)
\[
S^2=Z,\qquad T^2=S,\qquad S^\dagger=S^{-1},\qquad T^\dagger=T^{-1},
\]
where
\[
S=\begin{pmatrix}1&0\\0&i\end{pmatrix},
\qquad
T=\begin{pmatrix}1&0\\0&e^{i\pi/4}\end{pmatrix}.
\]

\begin{prop}[Conjugation table for $H$]
	\[
	H X H = Z,\qquad H Z H = X,\qquad H Y H = -Y.
	\]
\end{prop}

\begin{proof}
	We already computed $HZH=X$ in \S\ref{sec:circuits2}. Similarly:
	\[
	HXH=\frac12\begin{pmatrix}1&1\\1&-1\end{pmatrix}
	\begin{pmatrix}0&1\\1&0\end{pmatrix}
	\begin{pmatrix}1&1\\1&-1\end{pmatrix}
	=
	\begin{pmatrix}1&0\\0&-1\end{pmatrix}
	=Z.
	\]
	Then $HYH = H(iXZ)H = i(HXH)(HZH)=i(Z)(X)=-iXZ=-Y$.
\end{proof}

\subsubsection*{Two-qubit identities: CNOT and friends}

Recall the controlled-not gate
\[
\mathrm{CNOT}\ket{a,b}=\ket{a,\,b\oplus a}.
\]
It is its own inverse:
\[
\mathrm{CNOT}^2=I.
\]

\begin{prop}[CNOT involution]
	\[
	\mathrm{CNOT}^2=I.
	\]
\end{prop}

\begin{proof}
	Apply CNOT twice:
	\[
	(a,b)\mapsto (a,b\oplus a)\mapsto (a,(b\oplus a)\oplus a)=(a,b).
	\]
\end{proof}

\begin{prop}[Hadamard swaps control/target (CNOT $\leftrightarrow$ CZ)]
	Define the controlled-$Z$ gate $\mathrm{CZ}=\mathrm{diag}(1,1,1,-1)$.
	Then
	\[
	(I\otimes H)\,\mathrm{CNOT}\,(I\otimes H)=\mathrm{CZ}.
	\]
	Equivalently,
	\[
	(I\otimes H)\,\mathrm{CZ}\,(I\otimes H)=\mathrm{CNOT}.
	\]
\end{prop}

\begin{proof}
	Use the rule that $H$ maps $X\leftrightarrow Z$ by conjugation.
	CNOT can be characterized as the unique 2-qubit gate that conjugates Pauli operators as:
	\[
	\mathrm{CNOT}:\ 
	\begin{cases}
		X\otimes I \mapsto X\otimes X,\\
		I\otimes Z \mapsto Z\otimes Z,\\
		I\otimes X \mapsto I\otimes X,\\
		Z\otimes I \mapsto Z\otimes I,
	\end{cases}
	\]
	while CZ conjugates as:
	\[
	\mathrm{CZ}:\ 
	\begin{cases}
		X\otimes I \mapsto X\otimes Z,\\
		I\otimes X \mapsto Z\otimes X,\\
		Z\otimes I \mapsto Z\otimes I,\\
		I\otimes Z \mapsto I\otimes Z.
	\end{cases}
	\]
	Conjugating CNOT by $(I\otimes H)$ swaps $X$ and $Z$ on the second qubit, producing CZ.
	(You can also verify by multiplying $4\times4$ matrices directly; see exercises.)
\end{proof}

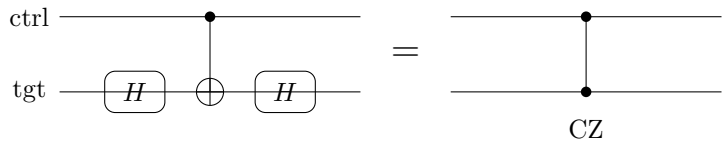
\begin{figure}[t]
	\centering
	\begin{tikzpicture}[scale=1.0]
		\draw (0,1) -- (4,1);
		\draw (0,0) -- (4,0);
		\node at (-0.4,1) {\small ctrl};
		\node at (-0.4,0) {\small tgt};
		
		\node[draw, rounded corners, minimum width=0.8cm, minimum height=0.55cm] at (1,0) {\small $H$};
		\fill (2,1) circle (0.07);
		\draw (2,1) -- (2,0);
		\draw (2,0) circle (0.18);
		\draw (2-0.18,0) -- (2+0.18,0);
		\draw (2,0-0.18) -- (2,0+0.18);
		\node[draw, rounded corners, minimum width=0.8cm, minimum height=0.55cm] at (3,0) {\small $H$};
		
		\node at (4.6,0.5) {\Large $=$};
		
		\draw (5.2,1) -- (8.8,1);
		\draw (5.2,0) -- (8.8,0);
		\fill (7,1) circle (0.07);
		\fill (7,0) circle (0.07);
		\draw (7,1) -- (7,0);
		\node at (7.0,-0.5) {\small CZ};
		
	\end{tikzpicture}
	\caption{A standard compilation identity: $(I\otimes H)\,\mathrm{CNOT}\,(I\otimes H)=\mathrm{CZ}$.}
	\label{fig:cnot-cz}
\end{figure}

\subsection{Commutation, reordering, and parallelization}

\subsubsection*{When can you swap gates?}

If two gates act on disjoint sets of wires, they commute:
\[
(A\otimes I)(I\otimes B)=(I\otimes B)(A\otimes I)=A\otimes B.
\]
If gates overlap on a wire, commutation is delicate and usually false.

\begin{prop}[Disjoint-support commutation]
	If $A$ acts on qubit set $S$ and $B$ acts on disjoint qubit set $T$ ($S\cap T=\varnothing$), then
	\[
	AB=BA.
	\]
\end{prop}

\begin{proof}
	In a tensor-factor ordering where $S$ wires come first and $T$ wires come next,
	the operators have the form $A\otimes I$ and $I\otimes B$, which commute by matrix multiplication.
\end{proof}

\subsubsection*{Moving single-qubit gates through CNOT}

These are compilation workhorses. The clean way to remember them:
CNOT propagates $X$ forward from control to target and propagates $Z$ backward from target to control.

\begin{prop}[Pauli propagation through CNOT]
	Let $C=\mathrm{CNOT}$ (control=first qubit, target=second). Then:
	\[
	C\,(X\otimes I)\,C = X\otimes X,
	\qquad
	C\,(I\otimes X)\,C = I\otimes X,
	\]
	\[
	C\,(Z\otimes I)\,C = Z\otimes I,
	\qquad
	C\,(I\otimes Z)\,C = Z\otimes Z.
	\]
	Equivalently (rearranged as commuting rules):
	\[
	(X\otimes I)\,C = C\,(X\otimes X),
	\qquad
	(I\otimes Z)\,C = C\,(Z\otimes Z),
	\]
	and $I\otimes X$ and $Z\otimes I$ commute through unchanged.
\end{prop}

\begin{proof}
	Verify on computational basis states $\ket{a,b}$:
	\[
	C\ket{a,b}=\ket{a,b\oplus a}.
	\]
	For example, compute $C(X\otimes I)C\ket{a,b}$:
	\[
	C(X\otimes I)\ket{a,b\oplus a}=C\ket{a\oplus1,\,b\oplus a}
	=\ket{a\oplus1,\,(b\oplus a)\oplus(a\oplus1)}
	=\ket{a\oplus1,\,b\oplus1}.
	\]
	But $(X\otimes X)\ket{a,b}=\ket{a\oplus1,b\oplus1}$, so
	$C(X\otimes I)C = X\otimes X$.
	The other identities are similar.
\end{proof}

\begin{figure}[t]
	\centering
	\begin{tikzpicture}[scale=1.0]
		\draw (0,1) -- (4,1);
		\draw (0,0) -- (4,0);
		\node[draw, rounded corners, minimum width=0.8cm, minimum height=0.55cm] at (1,1) {\small $X$};
		\fill (2.5,1) circle (0.07);
		\draw (2.5,1) -- (2.5,0);
		\draw (2.5,0) circle (0.18);
		\draw (2.5-0.18,0) -- (2.5+0.18,0);
		\draw (2.5,0-0.18) -- (2.5,0+0.18);
		
		\node at (4.6,0.5) {\Large $=$};
		
		\draw (5.2,1) -- (9.2,1);
		\draw (5.2,0) -- (9.2,0);
		\fill (6.7,1) circle (0.07);
		\draw (6.7,1) -- (6.7,0);
		\draw (6.7,0) circle (0.18);
		\draw (6.7-0.18,0) -- (6.7+0.18,0);
		\draw (6.7,0-0.18) -- (6.7,0+0.18);
		\node[draw, rounded corners, minimum width=0.8cm, minimum height=0.55cm] at (8.2,1) {\small $X$};
		\node[draw, rounded corners, minimum width=0.8cm, minimum height=0.55cm] at (8.2,0) {\small $X$};
	\end{tikzpicture}
	\caption{Pauli propagation: $X$ on the control ``copies'' to the target when pushed through CNOT.}
	\label{fig:pauli-propagation-x}
\end{figure}
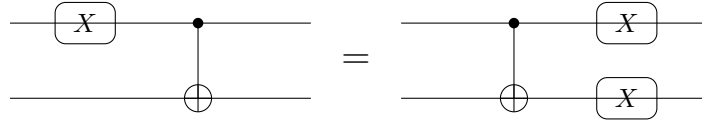

\subsubsection*{Parallelization: depth vs.\ gate count}

Two circuits can have the same total number of gates but different \emph{depth}.
Depth matters because noise accumulates with time and because hardware has coherence limits.

\begin{defn}[Circuit depth]
	The depth of a circuit is the minimum number of time steps needed if gates on disjoint wires can be executed simultaneously.
\end{defn}

\begin{rem}[Hardware reality check]
	On real devices, ``disjoint wires'' is not always enough:
	parallel gates may be disallowed due to crosstalk, frequency crowding, shared control lines, or readout conflicts.
	Compilation targets a \emph{hardware schedule}, not just an abstract circuit.
\end{rem}

\subsection{Normal forms, universality, and compilation}

\subsubsection*{Universal gate sets}

A gate set $\mathcal{G}$ is universal if any target unitary can be approximated arbitrarily well by circuits from $\mathcal{G}$.

\begin{prop}[A standard universal set]
	The gate set $\{H,\,T,\,\mathrm{CNOT}\}$ is universal for quantum computation.
\end{prop}

\begin{rem}[What universality means in practice]
	Universality is an existence statement. Compilation is the constructive task:
	given a target algorithmic unitary $U$, produce a circuit over your hardware gate set with acceptable:
	\begin{itemize}
		\item approximation error,
		\item depth,
		\item two-qubit gate count,
		\item scheduling constraints.
	\end{itemize}
\end{rem}

\subsubsection*{Normal forms (single-qubit)}

A very common compilation step is to rewrite a single-qubit unitary into Euler angles:
\[
U \in \SU(2)
\quad\Rightarrow\quad
U = R_Z(\alpha)\,R_Y(\beta)\,R_Z(\gamma),
\]
as already discussed in \S\ref{sec:bloch}. This is a continuous, geometry-based normal form.

If your hardware supports $R_Z$ virtually (``software phase'') and has calibrated $R_X$ or $R_Y$ pulses,
Euler decomposition is physically meaningful.

\subsubsection*{Clifford+T compilation viewpoint}

Clifford gates (generated by $H,S,\mathrm{CNOT}$) have rich algebraic structure and compile easily.
The $T$ gate supplies non-Clifford power needed for universality.

A canonical workflow:
\begin{enumerate}
	\item reduce/normalize Clifford parts by identities (often using stabilizer rules),
	\item approximate non-Clifford rotations using $T$ gates (Solovay--Kitaev or modern synthesis),
	\item then resimplify and schedule respecting hardware constraints.
\end{enumerate}

\begin{rem}[Compilation target depends on platform]
	Superconducting systems often natively support continuous rotations $R_X(\theta),R_Y(\theta),R_Z(\theta)$ and a calibrated entangler (e.g.\ cross-resonance).
	Trapped ions often implement M{\o}lmer--S{\o}rensen-like entanglers and collective rotations.
	Compilation means different normal forms on different platforms.
\end{rem}

\subsection{Why simplification matters physically}

Simplification is not aesthetic; it is an error-mitigation step.

\subsubsection*{Depth reduction reduces decoherence exposure}
If each time step has decoherence rate $\lambda$, a crude model says survival fidelity decays roughly like $e^{-\lambda \cdot \text{depth}}$.
So commuting/parallelizing gates can change success probabilities dramatically.

\subsubsection*{Two-qubit gates are expensive}
In many platforms, two-qubit gates have significantly higher error rates than single-qubit gates.
Compilation strategies often optimize:
\[
\text{(two-qubit gate count)} \quad \text{first, then depth}.
\]

\subsubsection*{Crosstalk and scheduling constraints}
Even if two gates commute algebraically, hardware may forbid parallel execution.
Real compilation is:
\[
\text{algebraic rewrite} \ +\ \text{device constraints} \ +\ \text{calibration/parameter choices}.
\]

\subsubsection*{FPGA/control relevance}
On control hardware, simplification can:
\begin{itemize}
	\item shorten real-time feedback loops (lower latency),
	\item reduce classical post-processing load per shot,
	\item simplify message schemas (fewer pulses/events to stream),
	\item reduce memory bandwidth in control pipelines.
\end{itemize}

\subsection{Conceptual summary}

\begin{itemize}
	\item Circuit equalities are matrix equalities (often modulo global phase).
	\item You simplify by applying identities and by commuting gates when permitted.
	\item CNOT has predictable propagation rules for Paulis (the core of many rewrites).
	\item Universality explains \emph{why} compilation is possible; normal forms explain \emph{how}.
	\item Physical motivation: fewer gates, less depth, fewer two-qubit operations $\Rightarrow$ higher fidelity and lower latency.
\end{itemize}

\subsection{Exercises}

\begin{exercise}[Warm-up: prove an identity by matrices]
	Compute $HXH$ and $HYH$ directly as $2\times2$ matrix products and verify
	\[
	HXH=Z,\qquad HYH=-Y.
	\]
\end{exercise}

\noindent\textbf{Solution.}
We already computed $HXH=Z$ in the text.
For $HYH$:
\[
HYH
=
\frac{1}{\sqrt2}
\begin{pmatrix}1&1\\1&-1\end{pmatrix}
\begin{pmatrix}0&-i\\ i&0\end{pmatrix}
\frac{1}{\sqrt2}
\begin{pmatrix}1&1\\1&-1\end{pmatrix}.
\]
First multiply $YH$:
\[
YH
=
\begin{pmatrix}0&-i\\ i&0\end{pmatrix}
\frac{1}{\sqrt2}\begin{pmatrix}1&1\\1&-1\end{pmatrix}
=
\frac{1}{\sqrt2}\begin{pmatrix}-i&i\\ i&i\end{pmatrix}.
\]
Then multiply $H(YH)$:
\[
H(YH)
=
\frac{1}{\sqrt2}\begin{pmatrix}1&1\\1&-1\end{pmatrix}
\frac{1}{\sqrt2}\begin{pmatrix}-i&i\\ i&i\end{pmatrix}
=
\frac12\begin{pmatrix}0&2i\\ -2i&0\end{pmatrix}
=
\begin{pmatrix}0&i\\ -i&0\end{pmatrix}
=-Y.
\]

\begin{exercise}[CNOT to CZ compilation]
	Verify directly that
	\[
	(I\otimes H)\,\mathrm{CNOT}\,(I\otimes H)=\mathrm{CZ}
	\]
	by multiplying $4\times 4$ matrices in the computational basis order
	$\{\ket{00},\ket{01},\ket{10},\ket{11}\}$.
\end{exercise}

\noindent\textbf{Solution.}
Write the matrices:
\[
\mathrm{CNOT}=
\begin{pmatrix}
	1&0&0&0\\
	0&1&0&0\\
	0&0&0&1\\
	0&0&1&0
\end{pmatrix},
\qquad
I\otimes H=
\frac{1}{\sqrt2}
\begin{pmatrix}
	1&1&0&0\\
	1&-1&0&0\\
	0&0&1&1\\
	0&0&1&-1
\end{pmatrix}.
\]
Compute $M=(I\otimes H)\,\mathrm{CNOT}\,(I\otimes H)$.
A quick way is to apply $M$ to basis vectors:
\begin{itemize}
	\item $M\ket{00}=\ket{00}$ (no phase),
	\item $M\ket{01}=\ket{01}$,
	\item $M\ket{10}=\ket{10}$,
	\item $M\ket{11}=-\ket{11}$.
\end{itemize}
Thus $M=\mathrm{diag}(1,1,1,-1)=\mathrm{CZ}$.

\begin{exercise}[Pauli propagation through CNOT]
	Let $C=\mathrm{CNOT}$ (control first, target second). Prove
	\[
	C(X\otimes I)C = X\otimes X,
	\qquad
	C(I\otimes Z)C = Z\otimes Z
	\]
	by checking their action on all basis states $\ket{a,b}$ with $a,b\in\{0,1\}$.
\end{exercise}

\noindent\textbf{Solution.}
We show the first identity; the second is similar.
For any $\ket{a,b}$:
\[
C(X\otimes I)C\ket{a,b}
=
C(X\otimes I)\ket{a,\,b\oplus a}
=
C\ket{a\oplus1,\,b\oplus a}
=
\ket{a\oplus1,\,(b\oplus a)\oplus(a\oplus1)}
=
\ket{a\oplus1,\,b\oplus1}.
\]
But $(X\otimes X)\ket{a,b}=\ket{a\oplus1,\,b\oplus1}$, so the operators agree.

\begin{exercise}[Commutation and illegal swaps]
	Decide whether each pair of gates commute (answer ``yes'' or ``no'') and justify:
	\begin{enumerate}
		\item $X\otimes I$ and $I\otimes Z$,
		\item $X\otimes I$ and $\mathrm{CNOT}$,
		\item $I\otimes X$ and $\mathrm{CNOT}$,
		\item $Z\otimes I$ and $\mathrm{CNOT}$.
	\end{enumerate}
\end{exercise}

\noindent\textbf{Solution.}
\begin{enumerate}
	\item Yes: disjoint support $\Rightarrow$ commute.
	\item No in general: $(X\otimes I)\mathrm{CNOT}=\mathrm{CNOT}(X\otimes X)$, not $\mathrm{CNOT}(X\otimes I)$.
	\item Yes: $C(I\otimes X)C=I\otimes X$ implies $(I\otimes X)$ commutes with CNOT.
	\item Yes: $C(Z\otimes I)C=Z\otimes I$ implies $(Z\otimes I)$ commutes with CNOT.
\end{enumerate}

\begin{exercise}[Depth reduction by commuting]
	Consider a 3-qubit circuit:
	apply $H$ on qubit 1, then $H$ on qubit 3, then CNOT with control=1 target=2.
	\begin{enumerate}
		\item Draw the circuit.
		\item Rewrite it to minimize depth using commutation/parallelization.
		\item State the original depth and the optimized depth.
	\end{enumerate}
\end{exercise}

\noindent\textbf{Solution.}
(1) Circuit:
\[
(H\ \text{on }1)\ \rightarrow\ (H\ \text{on }3)\ \rightarrow\ \mathrm{CNOT}_{1\to2}.
\]
(2) The $H$ on qubit 3 is disjoint from both $H$ on qubit 1 and the CNOT on (1,2),
so it can be moved and executed in parallel.
Schedule:
\begin{itemize}
	\item time step 1: apply $H$ on qubit 1 and $H$ on qubit 3 in parallel,
	\item time step 2: apply $\mathrm{CNOT}_{1\to2}$.
\end{itemize}
(3) Original depth: 3 (if executed strictly serially).
Optimized depth: 2.

	\Part{FPGA Viewpoint: Deterministic Pipelines for Hybrid Workflows}
	
\section{FPGA I: Why Hardware Acceleration Appears in Quantum Workflows}
\label{sec:fpga1}

\subsection{Objective}

This chapter explains, in concrete engineering terms, why ``quantum computing'' quickly becomes
a \emph{quantum--classical real-time system}. The quantum device (QPU) produces measurement outcomes,
but the workflow is dominated by \emph{classical} tasks:
\begin{itemize}
	\item decoding (QEC), filtering, and feedback control,
	\item low-latency decision making (adaptive circuits / mid-circuit measurement),
	\item streaming data reduction and aggregation (shots $\to$ estimates),
	\item deterministic, jitter-bounded I/O orchestration (timing is a feature, not an afterthought).
\end{itemize}
We focus on why and where FPGA acceleration naturally appears, and what you should learn first to use it.

\subsection{A fundamental reality of quantum computing}

\subsubsection*{The quantum computer is not a standalone box}

A practical quantum system is a closed loop:
\[
\text{(classical control)} \rightarrow \text{(pulses)} \rightarrow \text{(QPU evolution)} \rightarrow
\text{(measurement)} \rightarrow \text{(classical processing)} \rightarrow \text{(updated control)}.
\]
Even when the algorithm looks ``purely quantum'' on paper, the experiment is typically run as a
high-throughput, repeated-shot pipeline:
\[
\text{circuit parameters} \ \to\ \text{many shots}\ \to\ \text{bitstrings}\ \to\ \text{statistics/gradients}\ \to\ \text{parameter update}.
\]

\begin{rem}[Where the latency pressure comes from]
	Two sources create hard latency requirements:
	\begin{enumerate}
		\item \textbf{Error correction cycles:} syndrome extraction is periodic, and correction decisions
		must be made before errors accumulate.
		\item \textbf{Mid-circuit feedback:} adaptive measurement (feed-forward) requires classical decisions
		\emph{during} a circuit, not after the entire experiment.
	\end{enumerate}
	Both requirements point to deterministic, low-jitter classical processing close to the hardware.
\end{rem}

\subsubsection*{A minimal timing model (engineering lens)}

Let
\[
T_{\mathrm{cycle}}
=
T_{\mathrm{meas}}
+
T_{\mathrm{readout}}
+
T_{\mathrm{classical}}
+
T_{\mathrm{apply}},
\]
where:
\begin{itemize}
	\item $T_{\mathrm{meas}}$: physical measurement window (integrating an analog signal),
	\item $T_{\mathrm{readout}}$: digitization + thresholding + I/O transfer,
	\item $T_{\mathrm{classical}}$: decoding / filtering / decision logic,
	\item $T_{\mathrm{apply}}$: time to emit the next control action (pulse trigger, conditional gate).
\end{itemize}
The engineering problem is not only small average time, but also small \emph{tail latency} and jitter:
\[
T_{\mathrm{classical}} \text{ must be bounded and predictable.}
\]
This is exactly where FPGA-style hardware is strong.

\subsection{What is an FPGA (conceptually)?}

\subsubsection*{Reconfigurable hardware: ``hardware you can program''}

An FPGA (Field-Programmable Gate Array) is a chip containing:
\begin{itemize}
	\item many small logic blocks (LUTs: look-up tables),
	\item flip-flops (registers),
	\item block RAMs (BRAM),
	\item DSP blocks (fast multiply-accumulate),
	\item programmable routing fabric (wires you configure),
	\item I/O transceivers and clocking resources.
\end{itemize}
You do not ``run a program'' on an FPGA the way you run a program on a CPU.
Instead, you \emph{synthesize} a digital circuit that exists physically on the chip.

\subsubsection*{A mental model: from Boolean expressions to gates}

A combinational logic function $f:\{0,1\}^k\to\{0,1\}$ can be built from gates.
For example, the XOR function $a\oplus b$ can be built from AND/OR/NOT gates:
\[
a\oplus b = (a\wedge \neg b)\ \vee\ (\neg a \wedge b).
\]

\begin{figure}[t]
	\centering
	\begin{tikzpicture}[x=1.1cm,y=1.0cm, line cap=round, line join=round]
		\node at (0,2) {$a$};
		\node at (0,0) {$b$};
		\draw (0.2,2) -- (1.0,2);
		\draw (0.2,0) -- (1.0,0);
		
		\node[draw, minimum width=0.6cm, minimum height=0.45cm] (nota) at (1.6,2) {\small NOT};
		\node[draw, minimum width=0.6cm, minimum height=0.45cm] (notb) at (1.6,0) {\small NOT};
		\draw (1.0,2) -- (nota.west);
		\draw (1.0,0) -- (notb.west);
		
		\draw (1.0,2) -- (1.0,1.2);
		\draw (1.0,0) -- (1.0,0.8);
		
		\node[draw, minimum width=0.8cm, minimum height=0.55cm] (and1) at (3.2,1.2) {\small AND};
		\node[draw, minimum width=0.8cm, minimum height=0.55cm] (and2) at (3.2,0.8) {\small AND};
		
		\draw (1.0,1.2) -- (and1.west);
		\draw (nota.east) -- (2.3,2) -- (2.3,1.2) -- (and1.west);
		
		\draw (notb.east) -- (2.3,0) -- (2.3,0.8) -- (and2.west);
		\draw (1.0,0.8) -- (and2.west);
		
		\node[draw, minimum width=0.8cm, minimum height=0.55cm] (or1) at (4.8,1.0) {\small OR};
		\draw (and1.east) -- (or1.west);
		\draw (and2.east) -- (or1.west);
		
		\draw (or1.east) -- (6.0,1.0);
		\node at (6.3,1.0) {$a\oplus b$};
		
	\end{tikzpicture}
	\caption{A gate-level realization of XOR using AND/OR/NOT (pure combinational logic).}
	\label{fig:xor-gates}
\end{figure}
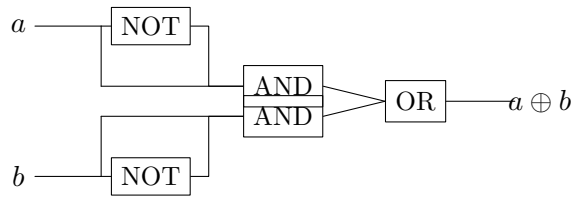

\subsubsection*{Sequential logic: registers and clocks}

Many FPGA tasks are streaming pipelines.
For that, we need state (memory), implemented using flip-flops (registers) updated by a clock.

A single-bit register can be pictured as a D flip-flop:

\begin{figure}[t]
	\centering
	\begin{tikzpicture}[x=1.1cm,y=1.0cm, line cap=round, line join=round]
		\node at (0,1) {$D$};
		\draw (0.2,1) -- (1.0,1);
		
		\node[draw, minimum width=1.1cm, minimum height=1.0cm] (dff) at (2.0,1) {\small DFF};
		\node at (1.45,0.35) {\small clk};
		
		\draw (1.0,1) -- (dff.west);
		
		\draw (1.2,0.35) -- (dff.south);
		\draw (1.2,0.35) -- (0.2,0.35);
		\node at (0,0.35) {$\mathrm{clk}$};
		
		\draw (dff.east) -- (3.4,1);
		\node at (3.7,1) {$Q$};
		
		\node at (2.0,-0.2) {\small $Q(t^+)=D(t)$ on clock edge};
	\end{tikzpicture}
	\caption{A register (D flip-flop) stores state and updates on a clock edge.}
	\label{fig:dff}
\end{figure}
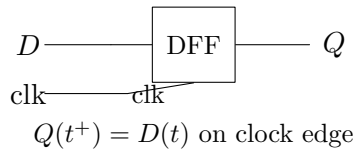

The key FPGA advantage: you can build \emph{many} such pipelines in parallel and keep them
synchronized with deterministic timing.

\subsection{Why classical processing becomes the bottleneck}

\subsubsection*{The QPU produces data; the classical side must keep up}

A measurement round produces bitstrings (or analog samples that become bits).
If your workflow needs to decide what happens next (feedback or QEC),
then the classical computation is on the critical path.

Two bottleneck patterns appear repeatedly:

\begin{enumerate}
	\item \textbf{Streaming bottleneck:} data arrive continuously; you must process at line rate.
	\item \textbf{Latency bottleneck:} you must decide within a time budget, every cycle, with bounded jitter.
\end{enumerate}

\subsubsection*{Example: parity checks as a toy model for syndrome processing}

A very simple syndrome-like computation is a parity (XOR) over several bits:
\[
s = b_1 \oplus b_2 \oplus \cdots \oplus b_m.
\]
On a CPU, this is a loop (fast on average, but with scheduling jitter and OS noise).
On an FPGA, it is a fixed XOR tree with known depth.

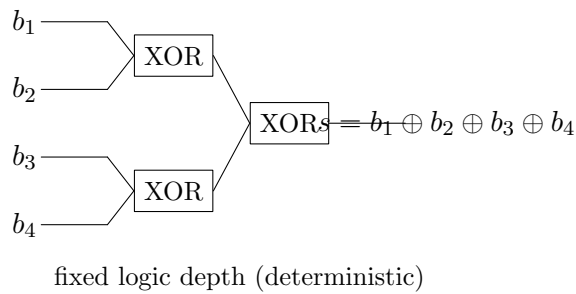
\begin{figure}[t]
	\centering
	\begin{tikzpicture}[x=1.1cm,y=0.9cm, line cap=round, line join=round]
		\node at (0,4) {$b_1$}; \draw (0.2,4) -- (1.0,4);
		\node at (0,3) {$b_2$}; \draw (0.2,3) -- (1.0,3);
		\node at (0,2) {$b_3$}; \draw (0.2,2) -- (1.0,2);
		\node at (0,1) {$b_4$}; \draw (0.2,1) -- (1.0,1);
		
		\node[draw, minimum width=0.8cm, minimum height=0.55cm] (x12) at (1.8,3.5) {\small XOR};
		\node[draw, minimum width=0.8cm, minimum height=0.55cm] (x34) at (1.8,1.5) {\small XOR};
		\draw (1.0,4) -- (x12.west);
		\draw (1.0,3) -- (x12.west);
		\draw (1.0,2) -- (x34.west);
		\draw (1.0,1) -- (x34.west);
		
		\node[draw, minimum width=0.8cm, minimum height=0.55cm] (xall) at (3.2,2.5) {\small XOR};
		\draw (x12.east) -- (xall.west);
		\draw (x34.east) -- (xall.west);
		
		\draw (xall.east) -- (4.6,2.5);
		\node at (5.1,2.5) {$s=b_1\oplus b_2\oplus b_3\oplus b_4$};
		
		\node at (2.6,0.2) {\small fixed logic depth (deterministic)};
	\end{tikzpicture}
	\caption{Parity as an XOR tree. This is a toy model for ``syndrome extraction'' style computations.}
	\label{fig:parity-tree}
\end{figure}

This toy is deliberately simple, but it captures the control-hardware theme:
\emph{the computation is a circuit} with fixed latency and predictable throughput.

\subsection{Why CPUs and GPUs are often insufficient}

\subsubsection*{CPUs: great flexibility, weak determinism}

CPUs excel at irregular control flow and complex software stacks.
But low-latency control loops suffer from:
\begin{itemize}
	\item OS scheduling jitter,
	\item cache/memory variability,
	\item interrupts and context switches,
	\item unpredictable tail latency under load.
\end{itemize}
Even if average latency is good, the worst-case (or tail) can break real-time constraints.

\subsubsection*{GPUs: massive throughput, high overhead and poor tight-loop latency}

GPUs excel at:
\begin{itemize}
	\item large batch linear algebra,
	\item embarrassingly parallel workloads,
	\item high throughput with large data blocks.
\end{itemize}
But they often struggle with:
\begin{itemize}
	\item tight per-cycle feedback decisions (kernel launch overhead, queueing),
	\item low-latency I/O and deterministic response time,
	\item bit-level streaming logic and protocol handling.
\end{itemize}

\begin{rem}[A useful rule of thumb]
	If the task is ``do a huge matrix multiply'' $\Rightarrow$ GPU.
	If the task is ``for every cycle, decide in a few hundred nanoseconds with bounded jitter'' $\Rightarrow$ FPGA.
	Many quantum pipelines require \emph{both}, but in different places.
\end{rem}

\subsection{Key roles of FPGA in quantum pipelines}

We list the FPGA roles in a way that matches quantum workflows in later chapters.

\subsubsection*{Role 1: Real-time QEC decoding and correction triggers}

Syndromes are produced repeatedly. A decoder maps syndrome history $\to$ correction actions.
Even if you do not implement the full decoder, you often need:
\begin{itemize}
	\item fast parity check accumulation,
	\item windowed filtering and reliability scoring,
	\item threshold triggers and event flags,
	\item routing of correction instructions to the control sequencer.
\end{itemize}

\subsubsection*{Role 2: Fast classical feed-forward (mid-circuit measurement)}

Adaptive circuits require:
\[
\text{measure} \Rightarrow \text{classical predicate} \Rightarrow \text{conditional gate}.
\]
This is naturally represented as a small state machine with deterministic timing.

\begin{figure}[t]
	\centering
	\begin{tikzpicture}[x=1.25cm,y=0.95cm, line cap=round, line join=round]
		\node[draw, rounded corners, minimum width=2.5cm, minimum height=0.9cm] (meas) at (0,0) {\small readout bit $m$};
		\node[draw, rounded corners, minimum width=3.2cm, minimum height=0.9cm] (logic) at (3,0) {\small predicate / decode};
		\node[draw, rounded corners, minimum width=3.0cm, minimum height=0.9cm] (act) at (6.3,0) {\small conditional action};
		
		\draw[->] (meas.east) -- (logic.west) node[midway,above] {\small stream};
		\draw[->] (logic.east) -- (act.west) node[midway,above] {\small decision};
		
		\node at (3,-1.2) {\small deterministic low-jitter loop is the design goal};
	\end{tikzpicture}
	\caption{The basic feedback-control motif: readout $\to$ decision logic $\to$ conditional action.}
	\label{fig:feedback-loop}
\end{figure}
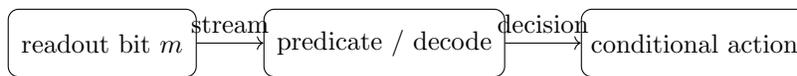

\subsubsection*{Role 3: Readout DSP and data reduction}

Before you even get bits, you often start with analog readout samples.
FPGA DSP blocks can do:
\begin{itemize}
	\item digital downconversion / filtering,
	\item integration windows,
	\item thresholding and classification,
	\item shot-by-shot aggregation (counts, moments).
\end{itemize}
This reduces the data-rate pressure on the CPU/GPU and keeps latency low.

\subsubsection*{Role 4: Protocol glue and time orchestration}

Quantum control stacks involve multiple devices and timing domains.
FPGA is commonly used for:
\begin{itemize}
	\item deterministic triggering and sequencing,
	\item I/O protocol bridging (custom links, timestamped packets),
	\item buffering, formatting, and versioned message schemas.
\end{itemize}

\subsubsection*{Role 5: ``Hybrid'' acceleration for geometry/optimization pipelines}

Even when training/optimization happens on GPUs/CPUs, FPGA can contribute:
\begin{itemize}
	\item fast statistics reduction (shots $\to$ expectations $\to$ gradients),
	\item low-latency preconditioning approximations (limited precision, structured matrices),
	\item streaming evaluation loops for variational circuits when parameters change rapidly.
\end{itemize}

\subsection{Connection to earlier chapters}

This chapter connects directly to the earlier themes:

\begin{itemize}
	\item \textbf{Bloch sphere and circuits:} gates are rotations / linear maps, but experiments are repeated-shot pipelines
	with classical post-processing. FPGA helps keep the pipeline real-time and deterministic.
	\item \textbf{Interference and measurement:} measurement collapses quantum information into bits;
	the rest is classical computation---often on the critical path.
	\item \textbf{Circuit identities and compilation:} simplification reduces depth and two-qubit gates;
	that reduces the burden on both the QPU and the classical control system (fewer events, tighter schedules).
\end{itemize}

\subsection{Exercises}

\begin{exercise}[Gate-level parity]
	Build a gate-level circuit (AND/OR/NOT or XOR gates) that computes
	\[
	s=b_1\oplus b_2\oplus b_3\oplus b_4
	\]
	and estimate its combinational depth if implemented as a balanced XOR tree.
\end{exercise}

\noindent\textbf{Solution.}
A balanced XOR tree is shown in Fig.~\ref{fig:parity-tree}.
Depth is 2 XOR levels (first XOR pairs, then XOR the results). If each XOR is built from basic gates,
the depth becomes ``(depth of XOR implementation)$\times 2$''.

\begin{exercise}[Registers and streaming]
	Suppose a stream of readout bits $m_t\in\{0,1\}$ arrives each cycle.
	Design a small circuit that outputs the running parity
	\[
	p_t = m_1\oplus m_2\oplus \cdots \oplus m_t.
	\]
	Draw the sequential circuit (use one DFF and one XOR) and write the update equation for the register.
\end{exercise}

\noindent\textbf{Solution.}
Let the register hold the current parity $p$.
Each cycle, update $p \leftarrow p \oplus m_t$.
Circuit: XOR the incoming $m_t$ with the register output $p$ and feed back into the DFF input.
Formally:
\[
p_{t+1}=p_t\oplus m_{t+1},\qquad p_0=0.
\]

\begin{figure}[t]
	\centering
	\begin{tikzpicture}[x=1.2cm,y=1.0cm, line cap=round, line join=round]
		\node at (0,1) {$m_t$};
		\draw (0.2,1) -- (1.1,1);
		
		\node[draw, minimum width=0.9cm, minimum height=0.55cm] (xor) at (2.0,1) {\small XOR};
		\draw (1.1,1) -- (xor.west);
		
		\node[draw, minimum width=1.1cm, minimum height=1.0cm] (dff) at (4.0,1) {\small DFF};
		\draw (xor.east) -- (dff.west);
		
		\draw (dff.east) -- (5.4,1);
		\node at (5.9,1) {$p_t$};
		
		\draw (5.0,1) -- (5.0,0.2) -- (1.6,0.2) -- (1.6,0.8) -- (xor.west);
		
		\draw (3.3,0.1) -- (4.0,0.1) -- (4.0,0.5) -- (dff.south);
		\node at (3.0,0.1) {\small clk};
		
		\node at (3.2,-0.6) {\small $p_{t+1} = p_t \oplus m_{t+1}$};
	\end{tikzpicture}
	\caption{A 1-register streaming parity accumulator (toy model of running syndrome aggregation).}
	\label{fig:streaming-parity}
\end{figure}
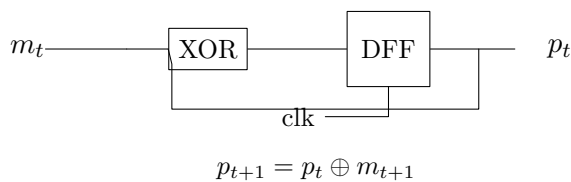

\begin{exercise}[Latency budgeting thought experiment]
	In the cycle model
	\[
	T_{\mathrm{cycle}}=T_{\mathrm{meas}}+T_{\mathrm{readout}}+T_{\mathrm{classical}}+T_{\mathrm{apply}},
	\]
	explain which term(s) an FPGA most directly improves, and why (two to five sentences).
\end{exercise}

\noindent\textbf{Solution.}
An FPGA most directly reduces $T_{\mathrm{classical}}$ and part of $T_{\mathrm{readout}}$ by performing
thresholding / filtering / decision logic in hardware at line rate with fixed latency.
It can also reduce $T_{\mathrm{apply}}$ by integrating tightly with pulse sequencers and I/O triggers,
avoiding OS and software scheduling overhead.
The physical measurement window $T_{\mathrm{meas}}$ is set by device physics and usually not reduced by FPGA.

\begin{exercise}[CPU vs.\ GPU vs.\ FPGA classification]
	For each task, choose the best default accelerator (CPU/GPU/FPGA) and justify briefly:
	\begin{enumerate}
		\item large-batch simulation of $10^6$ parameter settings of a variational circuit,
		\item deterministic per-cycle processing of a $1$-bit readout stream with a hard deadline,
		\item parsing and logging experiment metadata to disk,
		\item real-time predicate: ``if $m_t=1$ then trigger pulse A else pulse B'' within a tight budget.
	\end{enumerate}
\end{exercise}

\noindent\textbf{Solution.}
\begin{enumerate}
	\item GPU: throughput-dominated batch computation.
	\item FPGA: tiny per-cycle logic with hard latency/jitter constraints.
	\item CPU: software + file system interaction and irregular control flow.
	\item FPGA: fixed-latency branching logic close to the hardware.
\end{enumerate}


\section{FPGA II: Circuits as Dataflow Graphs (A Hardware View)}
\label{sec:fpga2}

\subsection{Objective}

This chapter reinterprets quantum circuits using a hardware lens:
\begin{center}
	\emph{a circuit diagram is not only an algebraic expression, but also a \textbf{dataflow graph}
		that suggests a staged execution plan.}
\end{center}
We connect:
\[
\text{operator composition} \ (U=U_k\cdots U_2U_1)
\quad \longleftrightarrow \quad
\text{pipeline stages and signal flow},
\]
and show how tensor products, parallel gates, and locality constraints become wiring rules.
The goal is to make the step from ``circuit on paper'' to ``implementation plan'' feel natural,
\emph{before} you write RTL.

\subsection{From operator composition to staged execution}

\subsubsection*{The algebraic view (what you already know)}

A quantum circuit on $n$ qubits is a product of operators:
\[
U = U_L U_{L-1}\cdots U_2 U_1,
\qquad
\ket{\psi_{\mathrm{out}}}=U\ket{\psi_{\mathrm{in}}}.
\]
In simulation (or linear-algebra thinking), the state vector is updated by repeated matrix-vector multiplication.
That is logically correct, but it hides the implementation question:
\begin{quote}
	\emph{Which pieces can run in parallel? What is the critical path? Where do we store intermediate data?
		Which steps must wait for measurement outcomes?}
\end{quote}

\subsubsection*{The staged execution view (what hardware sees)}

Hardware tends to execute in \emph{stages}:
\begin{itemize}
	\item each stage consumes inputs available at that stage,
	\item performs a fixed set of operations,
	\item produces outputs registered at a clock boundary.
\end{itemize}
So we group a circuit into \emph{layers} (a.k.a.\ moments):
\[
U = \underbrace{(\text{Layer }r)}_{\text{parallel}} \cdots \underbrace{(\text{Layer }2)}_{\text{parallel}}\underbrace{(\text{Layer }1)}_{\text{parallel}}.
\]
Inside a layer, gates that act on disjoint qubits commute and can be parallelized physically (subject to routing and control constraints).

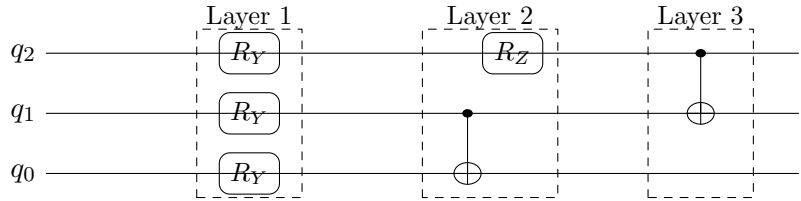
\begin{figure}[t]
	\centering
	\begin{tikzpicture}[x=1.0cm,y=0.8cm, line cap=round, line join=round]
		\draw (0,2) -- (10,2);
		\draw (0,1) -- (10,1);
		\draw (0,0) -- (10,0);
		\node[left] at (0,2) {$q_2$};
		\node[left] at (0,1) {$q_1$};
		\node[left] at (0,0) {$q_0$};
		
		\draw[dashed] (2.0,-0.4) rectangle (3.4,2.4);
		\node at (2.7,2.6) {\small Layer 1};
		
		\draw[dashed] (5.0,-0.4) rectangle (6.8,2.4);
		\node at (5.9,2.6) {\small Layer 2};
		
		\draw[dashed] (8.0,-0.4) rectangle (9.4,2.4);
		\node at (8.7,2.6) {\small Layer 3};
		
		\node[draw, rounded corners, minimum width=0.8cm, minimum height=0.5cm] at (2.7,2) {\small $R_Y$};
		\node[draw, rounded corners, minimum width=0.8cm, minimum height=0.5cm] at (2.7,1) {\small $R_Y$};
		\node[draw, rounded corners, minimum width=0.8cm, minimum height=0.5cm] at (2.7,0) {\small $R_Y$};
		
		\fill (5.6,1) circle (0.07);
		\draw (5.6,1) -- (5.6,0);
		\draw (5.6,0) circle (0.18);
		\draw (5.42,0) -- (5.78,0);
		\draw (5.6,-0.18) -- (5.6,0.18);
		
		\node[draw, rounded corners, minimum width=0.8cm, minimum height=0.5cm] at (6.2,2) {\small $R_Z$};
		
		\fill (8.7,2) circle (0.07);
		\draw (8.7,2) -- (8.7,1);
		\draw (8.7,1) circle (0.18);
		\draw (8.52,1) -- (8.88,1);
		\draw (8.7,0.82) -- (8.7,1.18);
		
	\end{tikzpicture}
	\caption{Layering a circuit: each dashed box is a parallelizable stage. Hardware cares about stage boundaries and critical path.}
	\label{fig:circuit-layering}
\end{figure}

\begin{rem}[Layering is not unique]
	Different compilers choose different layerings:
	\begin{itemize}
		\item commuting gates may move across each other,
		\item a hardware scheduler may insert idle cycles to satisfy constraints (crosstalk, routing, calibration),
		\item measurement/feedback introduces hard barriers.
	\end{itemize}
	The important point is that the circuit naturally suggests a directed acyclic graph (DAG) of dependencies.
\end{rem}

\subsection{Dataflow graphs: computation as signal flow}

\subsubsection*{A dataflow graph (DAG) viewpoint}

A dataflow graph is a directed graph where:
\begin{itemize}
	\item nodes are operations (gates, measurement, classical post-processing),
	\item edges carry values (qubit state dependencies, measurement bits, parameters),
	\item an edge $u\to v$ means ``$v$ cannot run until $u$ produces its output.''
\end{itemize}

For a circuit with no mid-circuit measurement, the dependency graph is essentially:
\[
\text{gate on qubit }i \text{ depends on the previous gate(s) that touched qubit } i.
\]
Two gates on disjoint qubits have no dependency and can be parallelized.

\subsubsection*{Hardware meaning: schedule and resources}

Given a dataflow graph, hardware planning asks:
\begin{itemize}
	\item What is the \textbf{critical path length}? (minimum time)
	\item What is the \textbf{parallelism} at each stage? (resource usage)
	\item Where are the \textbf{barriers}? (measurement, synchronization, I/O)
\end{itemize}

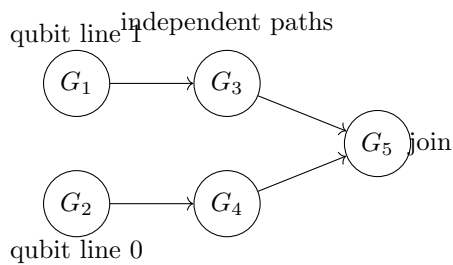
\begin{figure}[t]
	\centering
	\begin{tikzpicture}[x=1.0cm,y=0.8cm, line cap=round, line join=round]
		\node[draw, circle, minimum size=0.7cm] (g1) at (0,2) {\small $G_1$};
		\node[draw, circle, minimum size=0.7cm] (g2) at (0,0) {\small $G_2$};
		\node[draw, circle, minimum size=0.7cm] (g3) at (2,2) {\small $G_3$};
		\node[draw, circle, minimum size=0.7cm] (g4) at (2,0) {\small $G_4$};
		\node[draw, circle, minimum size=0.7cm] (g5) at (4,1) {\small $G_5$};
		
		\draw[->] (g1) -- (g3);
		\draw[->] (g2) -- (g4);
		\draw[->] (g3) -- (g5);
		\draw[->] (g4) -- (g5);
		
		\node at (0,2.8) {\small qubit line 1};
		\node at (0,-0.8) {\small qubit line 0};
		
		\node at (2,3.0) {\small independent paths};
		\node at (4.7,1.0) {\small join};
		
	\end{tikzpicture}
	\caption{A simple dataflow DAG: $G_1\!\to\!G_3$ and $G_2\!\to\!G_4$ can run in parallel, then must join at $G_5$. Critical path is the longest chain.}
	\label{fig:dataflow-dag}
\end{figure}

\subsubsection*{Adding measurement creates a mixed quantum--classical dataflow}

Mid-circuit measurement introduces classical edges:
\[
\text{(measure)} \rightarrow \text{(classical function)} \rightarrow \text{(conditional gate)}.
\]
This turns the circuit into a hybrid dataflow graph.

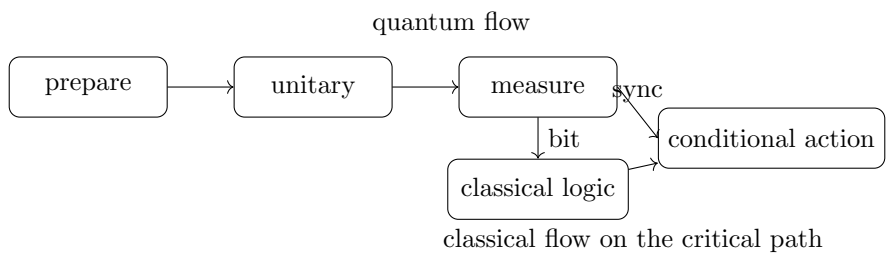
\begin{figure}[t]
	\centering
	\begin{tikzpicture}[x=1.15cm,y=0.85cm, line cap=round, line join=round]
		\node[draw, rounded corners, minimum width=2.1cm, minimum height=0.8cm] (prep) at (0,1.2) {\small prepare};
		\node[draw, rounded corners, minimum width=2.1cm, minimum height=0.8cm] (unit) at (2.6,1.2) {\small unitary};
		\node[draw, rounded corners, minimum width=2.1cm, minimum height=0.8cm] (meas) at (5.2,1.2) {\small measure};
		
		\node[draw, rounded corners, minimum width=2.4cm, minimum height=0.8cm] (cl) at (5.2,-0.4) {\small classical logic};
		\node[draw, rounded corners, minimum width=2.4cm, minimum height=0.8cm] (cond) at (7.9,0.4) {\small conditional action};
		
		\draw[->] (prep) -- (unit);
		\draw[->] (unit) -- (meas);
		\draw[->] (meas) -- (cl) node[midway,right] {\small bit};
		\draw[->] (cl) -- (cond);
		\draw[->] (meas.east) -- (cond.west) node[midway,above] {\small sync};
		
		\node at (4.2,2.2) {\small quantum flow};
		\node at (6.3,-1.2) {\small classical flow on the critical path};
		
	\end{tikzpicture}
	\caption{Hybrid dataflow: measurement outputs classical bits that feed a decision block, which triggers a conditional action.}
	\label{fig:hybrid-dataflow}
\end{figure}

\subsection{Tensor products as wiring rules}

\subsubsection*{The algebraic rule}

For disjoint subsystems, gates combine by tensor product:
\[
U_{\mathrm{total}} = U_A \otimes U_B.
\]
On basis states:
\[
(A\otimes B)(\ket v \otimes \ket w) = (A\ket v)\otimes (B\ket w).
\]
In circuit diagrams, this is ``parallel gates on different wires.''

\subsubsection*{The wiring rule (hardware view)}

Tensor product corresponds to:
\begin{itemize}
	\item \textbf{separate wires}, separate data paths,
	\item parallel operations, no cross-coupling,
	\item direct product state spaces in the simulator,
	\item independent control signals in the hardware schedule (until a two-qubit gate occurs).
\end{itemize}

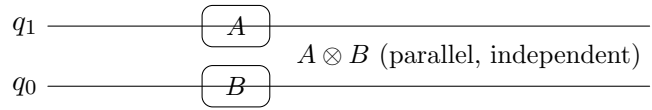
\begin{figure}[t]
	\centering
	\begin{tikzpicture}[x=1.0cm,y=0.8cm, line cap=round, line join=round]
		\draw (0,1) -- (8,1);
		\draw (0,0) -- (8,0);
		\node[left] at (0,1) {$q_1$};
		\node[left] at (0,0) {$q_0$};
		
		\node[draw, rounded corners, minimum width=0.9cm, minimum height=0.55cm] at (2.5,1) {\small $A$};
		\node[draw, rounded corners, minimum width=0.9cm, minimum height=0.55cm] at (2.5,0) {\small $B$};
		
		\node at (5.6,0.5) {\small $A\otimes B$ (parallel, independent)};
	\end{tikzpicture}
	\caption{Parallel gates on distinct wires represent tensor products. Hardware sees two independent datapaths.}
	\label{fig:tensor-as-wiring}
\end{figure}

\subsubsection*{Two-qubit gates are the ``wiring join''}

A two-qubit gate (e.g.\ CNOT, CZ) creates coupling and hence a dependency edge across wires.
That is where:
\begin{itemize}
	\item scheduling constraints appear,
	\item physical connectivity matters (nearest-neighbor couplers),
	\item control and calibration complexity increases,
	\item noise sensitivity typically increases.
\end{itemize}

\subsection{Parallelism and locality: what hardware cares about}

\subsubsection*{Parallelism is constrained by locality and shared resources}

Even if two gates are algebraically independent, hardware may limit concurrency due to:
\begin{itemize}
	\item limited number of control channels,
	\item cross-talk constraints (cannot drive neighboring qubits simultaneously),
	\item shared readout resonators / multiplexed measurement chains,
	\item routing constraints (two-qubit gates only on connected pairs).
\end{itemize}

\subsubsection*{Locality as a graph constraint}

Model the device by a connectivity graph $G=(V,E)$ where vertices are qubits and edges allow two-qubit gates.
Then any two-qubit gate must be on an edge $(i,j)\in E$.
If a circuit wants a gate between distant qubits, compilation inserts SWAP networks, increasing depth.

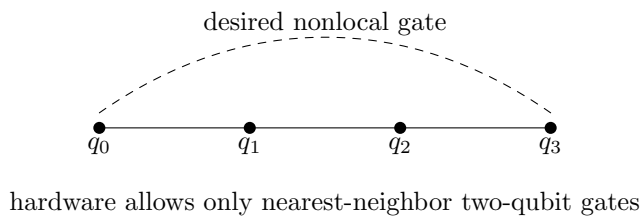
\begin{figure}[t]
	\centering
	\begin{tikzpicture}[x=1.0cm,y=1.0cm, line cap=round, line join=round]
		\foreach \x/\lab in {0/$q_0$,2/$q_1$,4/$q_2$,6/$q_3$}{
			\fill (\x,0) circle (0.08);
			\node[below] at (\x,0) {\small \lab};
		}
		\draw (0,0) -- (2,0) -- (4,0) -- (6,0);
		\draw[dashed] (0,0.2) to[bend left=35] (6,0.2);
		\node[above] at (3,1.1) {\small desired nonlocal gate};
		\node at (3,-1.0) {\small hardware allows only nearest-neighbor two-qubit gates};
	\end{tikzpicture}
	\caption{Connectivity constraint: two-qubit gates must follow the device graph. Nonlocal interactions require SWAP compilation, increasing depth.}
	\label{fig:connectivity}
\end{figure}

\subsubsection*{Hardware implication: depth is not just ``number of gates''}

Depth increases due to:
\begin{itemize}
	\item locality routing overhead (SWAPs),
	\item measurement barriers (must wait),
	\item classical feedback edges (decision time),
	\item scheduling constraints (parallelism limits).
\end{itemize}
Therefore, in hardware-oriented design, we track:
\[
\text{critical path depth} \quad + \quad \text{barrier count} \quad + \quad \text{I/O and decision latency}.
\]

\subsection{Worked example: 2-qubit HEA as a pipeline}

We build a minimal 2-qubit hardware-efficient ansatz (HEA) layer:
\[
U(\vec\theta) =
\bigl(R_Y^{(0)}(\theta_0)\,R_Z^{(0)}(\theta_1)\bigr)
\otimes
\bigl(R_Y^{(1)}(\theta_2)\,R_Z^{(1)}(\theta_3)\bigr)
\ \cdot\
\mathrm{CNOT}_{0\to 1}.
\]
(You may choose a different ordering; the point here is to show the pipeline and dependencies.)

\subsubsection*{Circuit diagram (algebra $\to$ graph)}

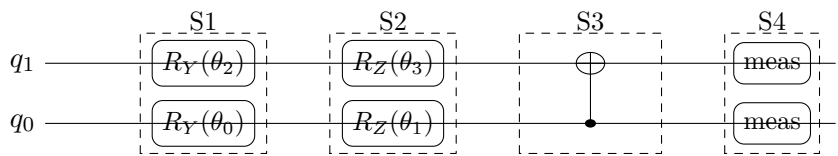
\begin{figure}[t]
	\centering
	\begin{tikzpicture}[x=1.05cm,y=0.8cm, line cap=round, line join=round]
		\draw (0,1) -- (10,1);
		\draw (0,0) -- (10,0);
		\node[left] at (0,1) {$q_1$};
		\node[left] at (0,0) {$q_0$};
		
		\draw[dashed] (1.2,-0.5) rectangle (2.8,1.5);
		\node at (2.0,1.7) {\small S1};
		\node[draw, rounded corners, minimum width=0.9cm, minimum height=0.55cm] at (2.0,1) {\small $R_Y(\theta_2)$};
		\node[draw, rounded corners, minimum width=0.9cm, minimum height=0.55cm] at (2.0,0) {\small $R_Y(\theta_0)$};
		
		\draw[dashed] (3.6,-0.5) rectangle (5.2,1.5);
		\node at (4.4,1.7) {\small S2};
		\node[draw, rounded corners, minimum width=0.9cm, minimum height=0.55cm] at (4.4,1) {\small $R_Z(\theta_3)$};
		\node[draw, rounded corners, minimum width=0.9cm, minimum height=0.55cm] at (4.4,0) {\small $R_Z(\theta_1)$};
		
		\draw[dashed] (6.0,-0.5) rectangle (7.8,1.5);
		\node at (6.9,1.7) {\small S3};
		
		\fill (6.9,0) circle (0.07);
		\draw (6.9,0) -- (6.9,1);
		\draw (6.9,1) circle (0.18);
		\draw (6.72,1) -- (7.08,1);
		\draw (6.9,0.82) -- (6.9,1.18);
		
		\draw[dashed] (8.6,-0.5) rectangle (9.8,1.5);
		\node at (9.2,1.7) {\small S4};
		\node[draw, rounded corners, minimum width=0.9cm, minimum height=0.55cm] at (9.2,1) {\small meas};
		\node[draw, rounded corners, minimum width=0.9cm, minimum height=0.55cm] at (9.2,0) {\small meas};
		
	\end{tikzpicture}
	\caption{A 2-qubit HEA layer as staged execution: (S1,S2) are parallel single-qubit rotations; (S3) is the entangling barrier; (S4) readout.}
	\label{fig:hea2q-pipeline}
\end{figure}

\subsubsection*{Dataflow interpretation}

In S1 and S2, the two wires are independent (tensor-product structure):
\[
\text{S1/S2: two parallel datapaths (control pulses can be scheduled concurrently if hardware permits).}
\]
At S3 (CNOT), the datapaths must join:
\[
\text{S3: cross-wire coupling (requires connectivity, more calibration, often dominates error).}
\]
At S4, measurement turns quantum state into classical bits:
\[
\text{S4: classical pipeline begins (aggregation, decoding, feedback).}
\]

\subsubsection*{Pipeline skeleton (hardware planner view)}

If you were implementing a repeated-shot experiment, the pipeline resembles:

\begin{figure}[t]
	\centering
	\begin{tikzpicture}[x=1.1cm,y=0.9cm, line cap=round, line join=round]
		\node[draw, rounded corners, minimum width=2.4cm, minimum height=0.8cm] (load) at (0,0) {\small load $\vec\theta$};
		\node[draw, rounded corners, minimum width=2.4cm, minimum height=0.8cm] (s12) at (2.8,0) {\small S1+S2 pulses};
		\node[draw, rounded corners, minimum width=2.4cm, minimum height=0.8cm] (ent) at (5.6,0) {\small S3 entangle};
		\node[draw, rounded corners, minimum width=2.4cm, minimum height=0.8cm] (meas) at (8.4,0) {\small S4 readout};
		
		\node[draw, rounded corners, minimum width=2.8cm, minimum height=0.8cm] (proc) at (11.7,0) {\small classical reduce};
		
		\draw[->] (load) -- (s12);
		\draw[->] (s12) -- (ent);
		\draw[->] (ent) -- (meas);
		\draw[->] (meas) -- (proc);
		
		\node at (5.8,-1.2) {\small repeat for many shots; FPGA aims to keep both timing and throughput deterministic};
		
	\end{tikzpicture}
	\caption{Repeated-shot pipeline skeleton for a 2-qubit HEA experiment: parameter load, pulse schedule, entangling barrier, readout, classical reduction.}
	\label{fig:hea2q-shot-pipeline}
\end{figure}

Hardware design questions now become explicit:
\begin{itemize}
	\item Can S1 and S2 be fully parallelized or must they be serialized by control constraints?
	\item What is the maximum rate at which $\vec\theta$ can change (parameter streaming bandwidth)?
	\item What is the deterministic budget for ``classical reduce'' if you want feedback between shots?
\end{itemize}

\subsection{Why this matters before RTL}

RTL is where you commit to:
\begin{itemize}
	\item clock domains, registers, and buffering,
	\item message schemas and I/O protocols,
	\item exact fixed-point formats for arithmetic,
	\item finite-state machines (FSMs) for sequencing and feedback,
	\item worst-case latency analysis and resource accounting.
\end{itemize}
If you do not first express the quantum workflow as a dataflow graph,
RTL becomes guesswork: you risk building a fast circuit that accelerates the wrong part of the pipeline.

A good pre-RTL discipline is:
\begin{enumerate}
	\item draw the dataflow DAG (quantum + classical edges),
	\item mark barriers (measurement, synchronization),
	\item define the message interfaces (inputs/outputs at each barrier),
	\item allocate deterministic latency budgets per stage,
	\item only then implement the stages as RTL modules.
\end{enumerate}

\subsection{Exercises}

\begin{exercise}[Layering and critical path]
	Given a circuit with gates
	\[
	U = (A^{(0)}\otimes B^{(1)})\ \cdot\ \mathrm{CNOT}_{0\to1}\ \cdot\ (C^{(0)}\otimes D^{(1)}),
	\]
	draw a layered circuit diagram and give the dependency DAG.
	Which stage(s) can be parallelized?
\end{exercise}

\noindent\textbf{Solution.}
$(A^{(0)}\otimes B^{(1)})$ is parallel single-qubit work (one layer).
$\mathrm{CNOT}_{0\to1}$ is a coupling barrier (next layer).
$(C^{(0)}\otimes D^{(1)})$ is again parallel single-qubit work (final layer).
DAG: two independent nodes in layer 1 feeding the CNOT node, then feeding two independent nodes in layer 3.
Parallelism exists in layers 1 and 3, but not across the CNOT barrier.

\begin{exercise}[Tensor product wiring rule]
	Verify explicitly that
	\[
	(H\otimes X)\ket{01} = \ket{+}\ket0,
	\qquad \ket{+}=\frac{\ket0+\ket1}{\sqrt2}.
	\]
	Then interpret the computation as a wiring rule rather than a $4\times4$ matrix multiplication.
\end{exercise}

\noindent\textbf{Solution.}
Use $(A\otimes B)(\ket v\otimes\ket w)=(A\ket v)\otimes(B\ket w)$:
\[
(H\otimes X)\ket{01}=(H\ket0)\otimes(X\ket1)=\ket+\otimes\ket0.
\]
Interpretation: the top wire applies $H$ to $\ket0$ while the bottom wire applies $X$ to $\ket1$, independently and in parallel.

\begin{exercise}[Connectivity overhead]
	Assume a linear hardware graph $q_0-q_1-q_2-q_3$.
	You want a CNOT between $q_0$ and $q_3$.
	Explain why SWAP insertion is needed, and estimate the additional two-qubit gate depth required in a naive approach.
\end{exercise}

\noindent\textbf{Solution.}
Because $(0,3)$ is not an edge, the hardware cannot directly couple $q_0$ and $q_3$.
A naive method moves the state of $q_0$ next to $q_3$ using SWAPs:
swap along the chain $q_0\leftrightarrow q_1$, then $q_1\leftrightarrow q_2$, then $q_2\leftrightarrow q_3$
(3 SWAPs) to bring the logical qubit to position 3, apply CNOT, then optionally swap back (another 3 SWAPs).
Each SWAP is typically 3 CNOTs (or similar cost), so the two-qubit depth can grow substantially.
Even without swapping back, you already add multiple two-qubit layers before the desired interaction.

\begin{exercise}[Hybrid edge and latency]
	Draw a hybrid dataflow graph for:
	\[
	\ket0 \xrightarrow{H} \text{measure } Z \xrightarrow{\text{if }m=1} X \xrightarrow{} \text{measure } Z.
	\]
	Mark which edge is classical, and explain why the classical node is on the critical path for the conditional $X$.
\end{exercise}

\noindent\textbf{Solution.}
The first measurement produces a classical bit $m$.
The conditional $X$ depends on $m$, so there is an edge
\[
\text{measure }Z \ \to\ m \ \to\ \text{(apply or skip }X\text{)}.
\]
That classical computation (even if it is just ``if $m=1$'') must complete before the second measurement stage.
Hence its latency contributes directly to the end-to-end circuit time.

\begin{exercise}[Design question: what to accelerate]
	You have a variational experiment where each shot returns a bitstring of length $n=50$.
	You only need the Hamming weight (number of 1s) and a few parities.
	Explain why computing these features on FPGA can reduce bandwidth and improve determinism.
\end{exercise}

\noindent\textbf{Solution.}
If you compute Hamming weight and parities on FPGA, you transmit only a small set of features
(e.g.\ an integer weight plus a few syndrome bits) rather than the full 50-bit string per shot.
This reduces I/O bandwidth and host-side parsing overhead.
Moreover, the feature extraction can be implemented as a fixed-latency combinational/sequential circuit,
removing OS jitter and making per-shot processing time predictable, which helps tight feedback loops.


\section{Algorithm I: Grover's Algorithm as a Geodesic Rotation}
\label{sec:grover}

\subsection{Objective}
Grover's algorithm is often taught as a clever trick involving an oracle and a mysterious
``diffusion operator.''  In this chapter we treat it as a \textbf{pure geometry story}:

\begin{quote}
	\emph{Grover iteration is a rotation in a 2D plane. Each step follows a constant-angle turn,
		so the algorithm is a discrete geodesic-like motion toward the marked direction.}
\end{quote}

This viewpoint clarifies:
\begin{itemize}
	\item why the algorithm reduces to a 2-dimensional subspace,
	\item why the success probability has a closed form,
	\item why the number of iterations is $\Theta(\sqrt{N})$,
	\item how the ``rotation speed'' can be interpreted via distinguishability (QFI/QFIM).
\end{itemize}

\subsection{Visual roadmap}

\subsubsection*{1) The three pictures (same algorithm)}
Keep three mental pictures in sync:

\begin{enumerate}
	\item \textbf{Circuit picture:} repeated block $G = D\,O_f$ applied $t$ times.
	\item \textbf{2D plane picture:} a vector rotating toward the marked axis by a fixed angle $2\theta$.
	\item \textbf{Probability picture:} $\Pr(\text{success after }t)=\sin^2((2t+1)\theta)$.
\end{enumerate}

\begin{figure}[t]
	\centering
	\begin{tikzpicture}[scale=1.05, line cap=round, line join=round]
		\node[draw, rounded corners, minimum width=3.0cm, minimum height=0.9cm] (O) at (-5.0,0.8) {$O_f$};
		\node[draw, rounded corners, minimum width=3.0cm, minimum height=0.9cm] (D) at (-1.5,0.8) {$D$};
		\node[draw, rounded corners, minimum width=3.0cm, minimum height=0.9cm] (G) at (2.0,0.8) {$G=D\,O_f$};
		\draw[->, thick] (-7.0,0.8) -- (O);
		\draw[->, thick] (O) -- (D);
		\draw[->, thick] (D) -- (G);
		\draw[->, thick] (G) -- (4.0,0.8);
		\node at (5.0,0.8) {\small repeat $t$ times};
		
		\draw[->] (-3.2,-1.0) -- (-3.2,-3.2) node[below] {\small $\ket{w_\perp}$};
		\draw[->] (-3.2,-1.0) -- (-1.0,-1.0) node[right] {\small $\ket{w}$};
		\draw[thick] (-3.2,-1.0) -- (-2.2,-2.4);
		\node at (-2.0,-2.6) {\small $\ket{\psi_0}$};
		
		\draw[->, thick] (-2.55,-1.75) arc (-120:-70:1.0);
		\node at (-2.1,-1.5) {\small $2\theta$};
		
		\draw[thick] (-3.2,-1.0) -- (-1.5,-1.6);
		\node at (-1.3,-1.6) {\small $\ket{\psi_t}$};
		
		\draw[->] (1.2,-3.0) -- (6.2,-3.0) node[right] {\small $t$};
		\draw[->] (1.2,-3.0) -- (1.2,-1.0) node[above] {\small $p_{\text{succ}}$};
		\draw[thick, domain=0:4.6, samples=80]
		plot (\x+1.2, {-3.0 + 1.7*(sin(35*\x))^2});
		\node at (3.8,-1.2) {\small $\sin^2((2t+1)\theta)$};
		
		\node[align=center] at (0.0,-4.0)
		{\small Circuit block $\leftrightarrow$ 2D rotation $\leftrightarrow$ closed-form success probability.};
	\end{tikzpicture}
	\caption{Grover in one page: repeated oracle+diffusion equals a fixed-angle rotation in a 2D plane, yielding a sinusoidal success probability.}
	\label{fig:grover-roadmap}
\end{figure}
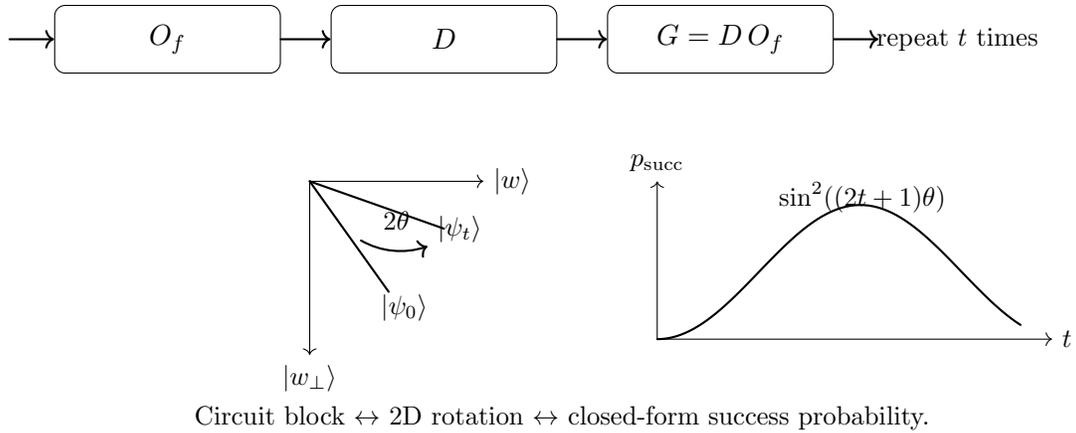

\subsection{Unstructured search and the phase oracle}

\subsubsection*{1) Problem statement}
Let $N=2^n$ and let $f:\{0,1\}^n\to\{0,1\}$ mark solutions:
\[
f(x)=1 \iff x \in \mathcal{M},
\]
where $\mathcal{M}$ is the set of marked items.
The simplest case is $|\mathcal{M}|=1$ (one solution), but we will allow $|\mathcal{M}|=M$.

\subsubsection*{2) Phase oracle}
Grover uses a \emph{phase oracle}:
\[
O_f\ket{x} = (-1)^{f(x)}\ket{x}.
\]
So marked basis states pick up a minus sign; unmarked states do nothing.

\begin{rem}[Why phase matters]
	A phase flip is invisible if you measure immediately in the computational basis.
	Grover works by converting phase information into amplitude bias via interference.
\end{rem}

\subsubsection*{3) The uniform starting state}
Define
\[
\ket{s} := \frac{1}{\sqrt{N}} \sum_{x\in\{0,1\}^n}\ket{x}.
\]
This can be prepared by applying $H^{\otimes n}$ to $\ket{0^n}$.

\subsection{The diffusion operator (inversion about the average)}

\subsubsection*{1) Definition}
The diffusion operator is
\[
D := 2\ket{s}\bra{s} - I.
\]
It is a reflection about the vector $\ket{s}$.

\begin{prop}[Reflection property]
	$D$ is unitary, Hermitian, and satisfies $D^2=I$.
	Moreover, for any vector $\ket{v}$, $D$ reflects $\ket{v}$ across the axis spanned by $\ket{s}$.
\end{prop}

\begin{proof}
	Hermitian: $(2\ket{s}\bra{s}-I)^\dagger=2\ket{s}\bra{s}-I$.
	Then
	\[
	D^2 = (2\ket{s}\bra{s}-I)^2
	=4\ket{s}\bra{s}\ket{s}\bra{s} -4\ket{s}\bra{s}+I
	=4\ket{s}\bra{s}-4\ket{s}\bra{s}+I=I
	\]
	since $\braket{s}{s}=1$ and $\ket{s}\bra{s}\ket{s}\bra{s}=\ket{s}\bra{s}$.
	Hence $D^{-1}=D$ and $D$ is unitary.
	A Hermitian involution is a reflection across its $+1$ eigenspace (here $\mathrm{span}\{\ket{s}\}$).
\end{proof}

\subsubsection*{2) ``Inversion about the average'' in coordinates}
Let $\ket{v}=\sum_x a_x\ket{x}$ and define the average amplitude
\[
\bar a := \frac{1}{N}\sum_x a_x.
\]
Then a direct computation gives
\[
(D\ket{v})_x = 2\bar a - a_x.
\]
So each amplitude is mapped to its reflection about the mean.

\begin{proof}[Coordinate computation]
	Since $\ket{s}=\frac1{\sqrt N}\sum_x \ket{x}$, we have
	\[
	\braket{s}{v}=\frac1{\sqrt N}\sum_x a_x = \sqrt N\,\bar a.
	\]
	Thus
	\[
	D\ket{v} = 2\ket{s}\braket{s}{v}-\ket{v}
	=2\left(\frac1{\sqrt N}\sum_x \ket{x}\right)(\sqrt N\,\bar a) - \sum_x a_x\ket{x}
	=\sum_x (2\bar a-a_x)\ket{x}.
	\]
\end{proof}

\subsection{Grover iterate = two reflections = rotation}

\subsubsection*{1) Define the Grover iterate}
The Grover iterate is
\[
G := D\,O_f.
\]

\subsubsection*{2) Geometric principle: two reflections make a rotation}
In a real 2D plane, the product of two reflections is a rotation by twice the angle between the reflecting axes.
Grover is exactly this, because:
\begin{itemize}
	\item $O_f$ is a reflection about the subspace of \emph{unmarked} vectors (it flips the marked subspace),
	\item $D$ is a reflection about $\ket{s}$.
\end{itemize}
The miracle is that the whole action collapses to a 2D subspace.

\subsection{Exact two-dimensional reduction (the plane $K$)}

\subsubsection*{1) Marked/unmarked superpositions}
Let $M=|\mathcal{M}|$.
Define normalized vectors
\[
\ket{w} := \frac{1}{\sqrt M}\sum_{x\in\mathcal{M}}\ket{x},
\qquad
\ket{w_\perp} := \frac{1}{\sqrt{N-M}}\sum_{x\notin\mathcal{M}}\ket{x}.
\]
Then $\{\ket{w},\ket{w_\perp}\}$ is an orthonormal pair spanning a 2D plane:
\[
K := \mathrm{span}\{\ket{w},\ket{w_\perp}\}.
\]

\subsubsection*{2) The starting state lies in $K$}
Write $\ket{s}$ in this basis:
\[
\ket{s}
=\frac{1}{\sqrt N}\left(\sum_{x\in\mathcal{M}}\ket{x}+\sum_{x\notin\mathcal{M}}\ket{x}\right)
=\sqrt{\frac{M}{N}}\,\ket{w}+\sqrt{\frac{N-M}{N}}\,\ket{w_\perp}.
\]
Define an angle $\theta\in(0,\pi/2)$ by
\[
\sin\theta = \sqrt{\frac{M}{N}},\qquad \cos\theta=\sqrt{\frac{N-M}{N}}.
\]
Then
\[
\ket{s}=\sin\theta\,\ket{w}+\cos\theta\,\ket{w_\perp}.
\]

\subsubsection*{3) $K$ is invariant under both $O_f$ and $D$}
First,
\[
O_f\ket{w}=-\ket{w},\qquad O_f\ket{w_\perp}=+\ket{w_\perp},
\]
so $O_f$ preserves $K$.

Next, $D=2\ket{s}\bra{s}-I$ preserves any subspace that contains $\ket{s}$ and is stable under orthogonal complement within that subspace.
Since $\ket{s}\in K$ and $K$ is 2D, one can check explicitly that $D$ maps both basis vectors into $K$:
\[
D\ket{w} = 2\ket{s}\braket{s}{w}-\ket{w}
=2(\sin\theta\,\ket{s})-\ket{w}\in K,
\]
and similarly $D\ket{w_\perp}\in K$.
Thus $G$ acts as a $2\times2$ rotation on $K$.

\subsection{Rotation angle and closed-form probabilities}

\subsubsection*{1) Matrix of $G$ in the basis $\{\ket{w},\ket{w_\perp}\}$}
We compute $G=D\,O_f$ restricted to $K$.

First note:
\[
O_f =
\begin{pmatrix}
	-1 & 0\\
	0 & 1
\end{pmatrix}
\quad\text{on }(\ket{w},\ket{w_\perp}).
\]

Next compute $D=2\ket{s}\bra{s}-I$ in this basis.
Since $\ket{s}=\sin\theta\,\ket{w}+\cos\theta\,\ket{w_\perp}$,
\[
\ket{s}\bra{s}=
\begin{pmatrix}
	\sin^2\theta & \sin\theta\cos\theta\\
	\sin\theta\cos\theta & \cos^2\theta
\end{pmatrix},
\]
so
\[
D=
2\begin{pmatrix}
	\sin^2\theta & \sin\theta\cos\theta\\
	\sin\theta\cos\theta & \cos^2\theta
\end{pmatrix}
-
\begin{pmatrix}
	1&0\\0&1
\end{pmatrix}
=
\begin{pmatrix}
	2\sin^2\theta-1 & 2\sin\theta\cos\theta\\
	2\sin\theta\cos\theta & 2\cos^2\theta-1
\end{pmatrix}.
\]
Use identities:
\[
2\sin^2\theta-1 = -\cos 2\theta,\quad
2\cos^2\theta-1 = \cos 2\theta,\quad
2\sin\theta\cos\theta=\sin 2\theta.
\]
Hence
\[
D=
\begin{pmatrix}
	-\cos 2\theta & \sin 2\theta\\
	\sin 2\theta & \cos 2\theta
\end{pmatrix}.
\]

Now multiply:
\[
G = D\,O_f
=
\begin{pmatrix}
	-\cos 2\theta & \sin 2\theta\\
	\sin 2\theta & \cos 2\theta
\end{pmatrix}
\begin{pmatrix}
	-1&0\\0&1
\end{pmatrix}
=
\begin{pmatrix}
	\cos 2\theta & \sin 2\theta\\
	-\sin 2\theta & \cos 2\theta
\end{pmatrix}.
\]
This is exactly a rotation by angle $2\theta$ in the $(\ket{w},\ket{w_\perp})$ plane.

\subsubsection*{2) Closed-form state after $t$ iterations}
Start from $\ket{\psi_0}=\ket{s}=\sin\theta\,\ket{w}+\cos\theta\,\ket{w_\perp}$.
Applying a rotation $t$ times:
\[
\ket{\psi_t} = G^t\ket{\psi_0}
= \sin((2t+1)\theta)\,\ket{w} + \cos((2t+1)\theta)\,\ket{w_\perp}.
\]

\subsubsection*{3) Success probability}
Measuring in the computational basis, success means landing in the marked set.
Since $\ket{w}$ is the equal superposition of marked basis states,
the total success probability is
\[
p_{\mathrm{succ}}(t) = \|\text{projection of }\ket{\psi_t}\text{ onto marked subspace}\|^2
= |\sin((2t+1)\theta)|^2
= \sin^2((2t+1)\theta).
\]

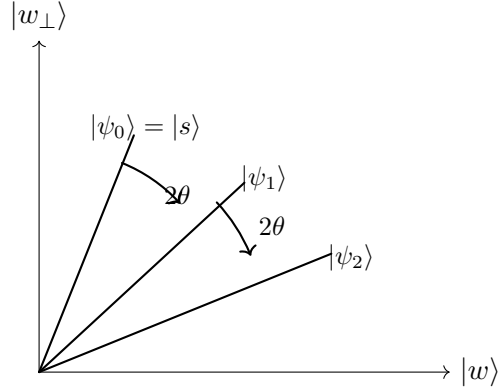
\begin{figure}[t]
	\centering
	\begin{tikzpicture}[scale=1.05, line cap=round, line join=round]
		\draw[->] (0,0) -- (5.2,0) node[right] {$\ket{w}$};
		\draw[->] (0,0) -- (0,4.2) node[above] {$\ket{w_\perp}$};
		
		\draw[thick] (0,0) -- (1.2,3.0);
		\node at (1.35,3.1) {\small $\ket{\psi_0}=\ket{s}$};
		
		\draw[thick] (0,0) -- (2.6,2.4);
		\node at (2.85,2.45) {\small $\ket{\psi_1}$};
		
		\draw[thick] (0,0) -- (3.7,1.5);
		\node at (3.95,1.45) {\small $\ket{\psi_2}$};
		
		\draw[->, thick] (1.05,2.65) arc (68:42:2.0);
		\node at (1.75,2.25) {\small $2\theta$};
		
		\draw[->, thick] (2.25,2.15) arc (43:22:2.2);
		\node at (2.95,1.85) {\small $2\theta$};
		
		\node[align=center] at (2.6,-0.9)
		{\small Each Grover step rotates by $2\theta$ toward $\ket{w}$.};
	\end{tikzpicture}
	\caption{The invariant plane $K=\mathrm{span}\{\ket{w},\ket{w_\perp}\}$. The Grover iterate is a fixed-angle rotation by $2\theta$, so the marked amplitude grows sinusoidally.}
	\label{fig:grover-rotation}
\end{figure}

\subsection{Concrete numerical example}

Take $N=16$ ($n=4$) and $M=1$ marked item.
Then
\[
\sin\theta=\sqrt{\frac{1}{16}}=\frac14,\qquad \theta=\arcsin\!\left(\frac14\right)\approx 0.25268.
\]
Success after $t$ iterations:
\[
p_{\mathrm{succ}}(t)=\sin^2((2t+1)\theta).
\]
Compute a few values:
\[
t=0:\ (2t+1)\theta=\theta \Rightarrow p\approx \sin^2(0.25268)\approx 0.0625,
\]
\[
t=1:\ 3\theta\approx 0.7580 \Rightarrow p\approx \sin^2(0.7580)\approx 0.472,
\]
\[
t=2:\ 5\theta\approx 1.2634 \Rightarrow p\approx \sin^2(1.2634)\approx 0.902,
\]
\[
t=3:\ 7\theta\approx 1.7688 \Rightarrow p\approx \sin^2(1.7688)\approx 0.961,
\]
\[
t=4:\ 9\theta\approx 2.2741 \Rightarrow p\approx \sin^2(2.2741)\approx 0.580.
\]
So the peak occurs near $t=3$ for this small instance.

\begin{rem}[Rule of thumb]
	The first peak is near $(2t+1)\theta\approx \frac{\pi}{2}$, i.e.
	\[
	t \approx \frac{\pi}{4\theta}-\frac12 \sim \frac{\pi}{4}\sqrt{\frac{N}{M}}.
	\]
\end{rem}

\subsection{Optimality: why you cannot beat $\sqrt{N}$}

\subsubsection*{1) Intuition from rotation speed}
Each Grover step rotates by $2\theta$ where $\sin\theta=\sqrt{M/N}$.
For small $M/N$,
\[
\theta \approx \sqrt{\frac{M}{N}},
\quad\Rightarrow\quad
2\theta \approx 2\sqrt{\frac{M}{N}}.
\]
To rotate from the initial angle $\theta$ to near $\pi/2$ requires $\Theta(1/\theta)$ steps:
\[
t = \Theta\!\left(\sqrt{\frac{N}{M}}\right).
\]

\subsubsection*{2) Why faster is impossible (geometric distinguishability sketch)}
The oracle $O_f$ is the only operation that depends on which item is marked.
Between oracle calls you can apply arbitrary unitaries independent of $f$.
A standard adversary/hybrid argument shows that after $T$ oracle queries,
the state cannot depend enough on $f$ (in inner product distance) to identify the marked item
unless $T=\Omega(\sqrt{N/M})$.
Geometrically: each oracle query can only change the state by a limited angle
when the marked subspace weight is small; accumulating a constant total angle needs $\sqrt{N}$ queries.

\begin{rem}[What to remember]
	Grover is optimal up to constants in the black-box (oracle) model.
	The geometry picture explains the scaling without heavy machinery.
\end{rem}

\subsection{Geodesic and QFI/QFIM viewpoint}

\subsubsection*{1) State-space geometry: projective distance}
Pure states live in projective space with the Fubini--Study distance
\[
d_{\mathrm{FS}}(\ket{\psi},\ket{\phi})
= \arccos(|\braket{\psi}{\phi}|).
\]
In the plane $K$, Grover steps move along a great-circle-like arc in the induced metric:
inner products evolve exactly as a rotation.

\subsubsection*{2) Query = bounded speed in distinguishability}
A useful mental model:
\begin{quote}
	\emph{Oracle calls are the only source of motion that depends on the hidden solution.
		They bound how fast the state can become distinguishable under the Fubini--Study metric.}
\end{quote}
This connects to QFI: for a one-parameter family $\ket{\psi(\varphi)}$,
QFI controls the rate at which nearby parameters become distinguishable.
Grover can be viewed as repeatedly applying a fixed ``generator'' that increases
the marked amplitude at a bounded rate; accumulating a constant distinguishability
costs $\Theta(\sqrt{N})$ oracle uses.

\subsubsection*{3) Multi-parameter extension (QFIM intuition)}
In amplitude amplification, you effectively steer probability mass in a low-dimensional subspace.
When the algorithm depends on multiple continuous parameters (as in variational settings),
the QFIM plays the role of a local metric controlling how parameter changes move the state.
Grover is the discrete, exactly-solvable prototype: the metric reduces to a constant 2D rotation.

\subsection{Grover beyond search: amplitude amplification}

\subsubsection*{1) General statement}
Suppose you have a subroutine $A$ such that
\[
A\ket{0^n}=\sqrt{p}\,\ket{\mathrm{good}}+\sqrt{1-p}\,\ket{\mathrm{bad}}.
\]
Amplitude amplification defines reflections:
\[
S_0 := 2\ket{0^n}\bra{0^n}-I,\qquad
S_{\mathrm{good}} := I-2\Pi_{\mathrm{good}},
\]
and constructs
\[
Q := A\,S_0\,A^\dagger\,S_{\mathrm{good}}.
\]
Then $Q$ acts as a rotation in $\mathrm{span}\{\ket{\mathrm{good}},\ket{\mathrm{bad}}\}$,
amplifying success probability from $p$ to near $1$ in $\Theta(1/\sqrt{p})$ iterations.

\subsubsection*{2) Why this matters}
Grover is just the special case where $A=H^{\otimes n}$ and $\Pi_{\mathrm{good}}$ selects marked basis states.
The rotation mechanism is universal: whenever you have a reflection about ``good''
and a reflection about the prepared direction, you get amplification.

\subsection{Exercises}

\begin{exercise}[Compute the diffusion operator explicitly for $n=2$]
	Let $n=2$ so $N=4$ and
	\[
	\ket{s}=\frac12(\ket{00}+\ket{01}+\ket{10}+\ket{11}).
	\]
	Write the $4\times4$ matrix of $D=2\ket{s}\bra{s}-I$ in the computational basis.
	Verify on a sample vector that it performs inversion about the average amplitude.
\end{exercise}
\noindent\textbf{Solution.}
We have $\ket{s}\bra{s}=\frac14\mathbf{1}\mathbf{1}^T$ where $\mathbf{1}=(1,1,1,1)^T$.
Thus $2\ket{s}\bra{s}=\frac12\mathbf{1}\mathbf{1}^T$, i.e.\ every entry is $1/2$.
Hence
\[
D=
\begin{pmatrix}
	-1/2& 1/2& 1/2& 1/2\\
	1/2&-1/2& 1/2& 1/2\\
	1/2& 1/2&-1/2& 1/2\\
	1/2& 1/2& 1/2&-1/2
\end{pmatrix}.
\]
For $\ket{v}=(a_{00},a_{01},a_{10},a_{11})^T$ the output is
$a'_{x}=2\bar a-a_x$ with $\bar a=\frac14\sum_x a_x$, matching inversion about the average.

\begin{exercise}[Two-dimensional reduction]
	Assume there is exactly one marked item $x^\star$.
	Define $\ket{w}=\ket{x^\star}$ and
	\[
	\ket{w_\perp}=\frac{1}{\sqrt{N-1}}\sum_{x\neq x^\star}\ket{x}.
	\]
	Show that both $O_f$ and $D$ preserve $K=\mathrm{span}\{\ket{w},\ket{w_\perp}\}$.
\end{exercise}
\noindent\textbf{Solution.}
$O_f$ flips $\ket{w}$ and fixes every unmarked basis state, so it fixes $\ket{w_\perp}$.
Thus $O_f(K)\subset K$.
For $D=2\ket{s}\bra{s}-I$, note $\ket{s}\in K$ and for any $\ket{u}\in K$ we have
$D\ket{u}=2\ket{s}\braket{s}{u}-\ket{u}\in K$ because it is a linear combination of $\ket{s}$ and $\ket{u}$, both in $K$.

\begin{exercise}[Derive the closed-form success probability]
	With $M$ marked items, define $\sin\theta=\sqrt{M/N}$.
	Starting from $\ket{s}$, show that after $t$ Grover iterations,
	\[
	\ket{\psi_t}=\sin((2t+1)\theta)\ket{w}+\cos((2t+1)\theta)\ket{w_\perp},
	\]
	hence $p_{\mathrm{succ}}(t)=\sin^2((2t+1)\theta)$.
\end{exercise}
\noindent\textbf{Solution.}
Compute the $2\times2$ matrix of $G$ on the basis $(\ket{w},\ket{w_\perp})$ as in the text:
\[
G=
\begin{pmatrix}
	\cos 2\theta & \sin 2\theta\\
	-\sin 2\theta & \cos 2\theta
\end{pmatrix}.
\]
This is a rotation by $2\theta$, so $G^t$ is rotation by $2t\theta$.
Apply it to $\ket{s}=(\sin\theta,\cos\theta)$ to get
\[
(\sin\theta,\cos\theta)\mapsto (\sin(\theta+2t\theta),\cos(\theta+2t\theta))
=(\sin((2t+1)\theta),\cos((2t+1)\theta)).
\]
Projecting onto the marked direction $\ket{w}$ yields success probability $\sin^2((2t+1)\theta)$.

\begin{exercise}[Find the optimal iteration count]
	Assume $M\ll N$. Using $\theta\approx \sqrt{M/N}$, show that the first near-maximum occurs at
	\[
	t^\star \approx \left\lfloor \frac{\pi}{4}\sqrt{\frac{N}{M}} \right\rfloor.
	\]
	Explain why iterating beyond $t^\star$ decreases success probability.
\end{exercise}
\noindent\textbf{Solution.}
The success is $\sin^2((2t+1)\theta)$, maximized near $(2t+1)\theta\approx \pi/2$.
Solve:
\[
(2t+1)\theta\approx \frac{\pi}{2}
\quad\Rightarrow\quad
t\approx \frac{\pi}{4\theta}-\frac12.
\]
For small $\theta$, $\theta\approx \sqrt{M/N}$ gives
$t\approx \frac{\pi}{4}\sqrt{N/M}$.
Because the function is sinusoidal in $t$, after the peak the angle passes $\pi/2$ and $\sin^2$ decreases (the state overshoots the marked axis).

\begin{exercise}[Amplitude amplification template]
	Let $A\ket{0^n}=\sqrt{p}\ket{\mathrm{good}}+\sqrt{1-p}\ket{\mathrm{bad}}$.
	Define $S_{\mathrm{good}}=I-2\Pi_{\mathrm{good}}$ and $S_0=2\ket{0^n}\bra{0^n}-I$.
	Show that
	\[
	Q:=A S_0 A^\dagger S_{\mathrm{good}}
	\]
	acts as a rotation in $\mathrm{span}\{\ket{\mathrm{good}},\ket{\mathrm{bad}}\}$.
	Find the rotation angle in terms of $p$.
\end{exercise}
\noindent\textbf{Solution.}
In the 2D basis $(\ket{\mathrm{good}},\ket{\mathrm{bad}})$, the state after $A$ is
\[
\ket{\psi}=\sin\theta\,\ket{\mathrm{good}}+\cos\theta\,\ket{\mathrm{bad}}
\quad\text{where}\quad \sin\theta=\sqrt{p}.
\]
$S_{\mathrm{good}}$ flips the sign of the good component (a reflection),
and $A S_0 A^\dagger = 2\ket{\psi}\bra{\psi}-I$ reflects about $\ket{\psi}$.
Thus $Q$ is a product of two reflections, hence a rotation by $2\theta$ in that plane.
Therefore success can be amplified to near $1$ in $\Theta(1/\sin\theta)=\Theta(1/\sqrt{p})$ iterations.
	
\section{Algorithm II: Quantum Natural Gradient and Variational Circuits}
\label{sec:qng}

\subsection{Objective}
Variational quantum algorithms optimize a parameterized circuit
\[
\ket{\psi(\boldsymbol{\theta})}=U(\boldsymbol{\theta})\ket{\psi_{\mathrm{in}}}
\]
to minimize (or maximize) a scalar objective (loss) such as an energy expectation
\[
\mathcal{L}(\boldsymbol{\theta})=\bra{\psi(\boldsymbol{\theta})}H\ket{\psi(\boldsymbol{\theta})},
\]
a negative log-likelihood, or a task-specific cost (fidelity, classification loss, etc.).
The key issue: \emph{parameters are not the true geometry}.
Two different parameterizations can represent the \emph{same} physical states with different distortions.
Quantum Natural Gradient (QNG) fixes this by updating parameters using the \emph{intrinsic geometry}
induced by the quantum state manifold (QFIM / Fubini--Study metric).

\medskip
\noindent\textbf{Pipeline viewpoint (hybrid, implementation-aware).}
A typical loop is:
\[
\text{prepare shots} \rightarrow \text{estimate }\mathcal{L},\nabla \mathcal{L}, G(\theta)
\rightarrow \text{solve }G\,\Delta\theta=\nabla \mathcal{L}
\rightarrow \theta\leftarrow \theta-\eta\Delta\theta,
\]
where $G(\theta)$ is (an approximation to) the QFIM. This is a \emph{streaming linear-algebra loop}
that often dominates classical latency, motivating FPGA/GPU acceleration in later chapters.

\subsection{From Euclidean updates to geometric motion}

\subsubsection*{Why vanilla gradient depends on coordinates}
A Euclidean update
\[
\boldsymbol{\theta}_{t+1}=\boldsymbol{\theta}_t-\eta\,\nabla_{\boldsymbol{\theta}}\mathcal{L}(\boldsymbol{\theta}_t)
\]
implicitly assumes a flat metric on parameter space (the identity matrix).
But the physically meaningful distance is in \emph{state space}:
for pure states, a canonical local distance is the Fubini--Study line element
\[
ds^2_{\mathrm{FS}} = 4\Bigl(\langle d\psi|d\psi\rangle - |\langle\psi|d\psi\rangle|^2\Bigr).
\]
Pulling back to parameters $\boldsymbol{\theta}$ yields a Riemannian metric
\[
ds^2_{\mathrm{FS}} = \sum_{i,j} G_{ij}(\boldsymbol{\theta})\,d\theta^i\,d\theta^j,
\]
where $G(\theta)$ is the (pure-state) quantum Fisher information metric (QFIM).
This is the object that should shape ``steepest descent.''

\subsubsection*{State-space first-order model}
A small parameter step $\delta\boldsymbol{\theta}$ induces a state change
\[
\ket{\psi(\theta+\delta\theta)} \approx \ket{\psi(\theta)} + \sum_i \partial_i\ket{\psi(\theta)}\,\delta\theta^i.
\]
The intrinsic squared distance (to second order) is
\[
\|\delta\psi\|_{\mathrm{FS}}^2
\;=\;
\sum_{i,j} G_{ij}(\theta)\,\delta\theta^i\delta\theta^j.
\]
Thus, $G$ plays the role of ``how expensive a parameter move is'' in the physical state manifold.

\subsection{Natural gradient as steepest descent in state space}

\subsubsection*{Derivation by constrained minimization}
Define the local linearization
\[
\mathcal{L}(\theta+\delta\theta)\approx \mathcal{L}(\theta)+\nabla\mathcal{L}(\theta)^{\!\top}\delta\theta.
\]
To pick the step that decreases $\mathcal{L}$ most while keeping the \emph{state-space move} bounded,
solve
\[
\min_{\delta\theta}\ \nabla\mathcal{L}(\theta)^{\!\top}\delta\theta
\quad
\text{subject to }
\delta\theta^{\!\top}G(\theta)\delta\theta\le \varepsilon^2.
\]
Lagrange multipliers give
\[
\nabla\mathcal{L}(\theta) + 2\lambda\,G(\theta)\delta\theta = 0
\quad\Rightarrow\quad
\delta\theta = -\frac{1}{2\lambda}\,G(\theta)^{-1}\nabla\mathcal{L}(\theta).
\]
Absorbing the scalar $\frac{1}{2\lambda}$ into the learning rate $\eta$ yields the QNG update
\[
\boxed{
	\boldsymbol{\theta}_{t+1}=\boldsymbol{\theta}_t-\eta\,G(\boldsymbol{\theta}_t)^{-1}\nabla \mathcal{L}(\boldsymbol{\theta}_t).
}
\]
This is coordinate-invariant: it is the steepest descent direction with respect to the Riemannian metric $G$.

\subsubsection*{Regularization and practical solve}
In practice $G$ may be singular or noisy. Common fixes:
\[
(G+\lambda I)\,\Delta\theta = \nabla\mathcal{L}
\quad\text{(Tikhonov / damping)},
\qquad
\Delta\theta = G^{\dagger}\nabla\mathcal{L}
\quad\text{(pseudoinverse)}.
\]
For $p$ parameters, solving $p\times p$ systems is $O(p^3)$ naively, but with structure
(blocks, low-rank updates, diagonal/triangular approximations) it becomes much cheaper.

\subsection{Explicit calculation: a two-parameter single-qubit circuit}

We now compute \emph{everything explicitly} for a simple but representative 1-qubit ansatz:
\[
U(\theta,\phi)=R_Z(\phi)\,R_Y(\theta),
\qquad
\ket{\psi(\theta,\phi)} = U(\theta,\phi)\ket{0}.
\]
Recall
\[
R_Y(\theta)=e^{-i\theta Y/2}=
\begin{pmatrix}
	\cos\frac{\theta}{2} & -\sin\frac{\theta}{2}\\
	\sin\frac{\theta}{2} & \cos\frac{\theta}{2}
\end{pmatrix},
\quad
R_Z(\phi)=e^{-i\phi Z/2}=
\begin{pmatrix}
	e^{-i\phi/2} & 0\\
	0 & e^{i\phi/2}
\end{pmatrix}.
\]

\subsubsection*{Step 1: write the state and Bloch vector}
Apply $R_Y(\theta)$ to $\ket0$:
\[
R_Y(\theta)\ket0=\cos\frac{\theta}{2}\ket0+\sin\frac{\theta}{2}\ket1.
\]
Then apply $R_Z(\phi)$:
\[
\ket{\psi(\theta,\phi)}
=
e^{-i\phi/2}\cos\frac{\theta}{2}\ket0
+
e^{i\phi/2}\sin\frac{\theta}{2}\ket1.
\]
Up to a global phase, this matches the Bloch-sphere parametrization, so the Bloch vector is
\[
\vec r(\theta,\phi)=(\sin\theta\cos\phi,\ \sin\theta\sin\phi,\ \cos\theta).
\]

\subsubsection*{Step 2: compute the QFIM $G(\theta,\phi)$ for pure states}
For pure states, one convenient formula is
\[
G_{ij} = 4\,\Re\Bigl(\langle \partial_i\psi|\partial_j\psi\rangle
-\langle \partial_i\psi|\psi\rangle\langle \psi|\partial_j\psi\rangle\Bigr),
\qquad i,j\in\{\theta,\phi\}.
\]
Compute derivatives:
\[
\partial_\theta\ket{\psi}
=
e^{-i\phi/2}\Bigl(-\frac12\sin\frac{\theta}{2}\Bigr)\ket0
+
e^{i\phi/2}\Bigl(\frac12\cos\frac{\theta}{2}\Bigr)\ket1,
\]
\[
\partial_\phi\ket{\psi}
=
\Bigl(-\frac{i}{2}\Bigr)e^{-i\phi/2}\cos\frac{\theta}{2}\ket0
+
\Bigl(\frac{i}{2}\Bigr)e^{i\phi/2}\sin\frac{\theta}{2}\ket1.
\]
Inner products:
\[
\langle\psi|\partial_\theta\psi\rangle
=
\cos\frac{\theta}{2}\Bigl(-\frac12\sin\frac{\theta}{2}\Bigr)
+
\sin\frac{\theta}{2}\Bigl(\frac12\cos\frac{\theta}{2}\Bigr)
=0,
\]
\[
\langle\psi|\partial_\phi\psi\rangle
=
-\frac{i}{2}\cos^2\frac{\theta}{2}+\frac{i}{2}\sin^2\frac{\theta}{2}
=
-\frac{i}{2}\cos\theta.
\]
Next,
\[
\langle\partial_\theta\psi|\partial_\theta\psi\rangle
=
\frac14\sin^2\frac{\theta}{2}+\frac14\cos^2\frac{\theta}{2}
=\frac14,
\]
\[
\langle\partial_\phi\psi|\partial_\phi\psi\rangle
=
\frac14\cos^2\frac{\theta}{2}+\frac14\sin^2\frac{\theta}{2}
=\frac14,
\]
\[
\langle\partial_\theta\psi|\partial_\phi\psi\rangle
=
\Bigl(-\frac12\sin\frac{\theta}{2}\Bigr)\Bigl(-\frac{i}{2}\cos\frac{\theta}{2}\Bigr)
+
\Bigl(\frac12\cos\frac{\theta}{2}\Bigr)\Bigl(\frac{i}{2}\sin\frac{\theta}{2}\Bigr)
=
\frac{i}{4}\sin\frac{\theta}{2}\cos\frac{\theta}{2}
-
\frac{i}{4}\sin\frac{\theta}{2}\cos\frac{\theta}{2}
=0.
\]
Now assemble:
\[
G_{\theta\theta}=4\Re\bigl(\tfrac14-0\bigr)=1,
\]
\[
G_{\phi\phi}=4\Re\Bigl(\tfrac14-|\langle\psi|\partial_\phi\psi\rangle|^2\Bigr)
=
4\Bigl(\tfrac14-\tfrac14\cos^2\theta\Bigr)
=
\sin^2\theta,
\]
\[
G_{\theta\phi}=4\Re\Bigl(0-\langle\partial_\theta\psi|\psi\rangle\langle\psi|\partial_\phi\psi\rangle\Bigr)=0.
\]
Therefore
\[
\boxed{
	G(\theta,\phi)=
	\begin{pmatrix}
		1 & 0\\
		0 & \sin^2\theta
	\end{pmatrix}.
}
\]
This is exactly the round metric on the Bloch sphere (up to convention): $ds^2=d\theta^2+\sin^2\theta\,d\phi^2$.

\subsubsection*{Step 3: pick a concrete objective and compute $\nabla\mathcal{L}$}
Choose a simple, physically meaningful cost:
\[
\mathcal{L}(\theta,\phi)=\bra{\psi(\theta,\phi)}Z\ket{\psi(\theta,\phi)}.
\]
Using $\vec r=(x,y,z)$, we have $\langle Z\rangle=z=\cos\theta$, hence
\[
\mathcal{L}(\theta,\phi)=\cos\theta,
\qquad
\nabla\mathcal{L}=
\begin{pmatrix}
	\partial_\theta\mathcal{L}\\
	\partial_\phi\mathcal{L}
\end{pmatrix}
=
\begin{pmatrix}
	-\sin\theta\\
	0
\end{pmatrix}.
\]

\subsubsection*{Step 4: compare Euclidean gradient vs.\ QNG}
Euclidean step:
\[
\Delta\theta_{\mathrm{Euc}}=
-\eta
\begin{pmatrix}
	-\sin\theta\\
	0
\end{pmatrix}
=
\eta
\begin{pmatrix}
	\sin\theta\\
	0
\end{pmatrix}.
\]
QNG step:
\[
\Delta\theta_{\mathrm{QNG}}=-\eta\,G^{-1}\nabla\mathcal{L}
=
-\eta
\begin{pmatrix}
	1 & 0\\
	0 & \sin^{-2}\theta
\end{pmatrix}
\begin{pmatrix}
	-\sin\theta\\
	0
\end{pmatrix}
=
\eta
\begin{pmatrix}
	\sin\theta\\
	0
\end{pmatrix}.
\]
Here they coincide because the loss depends only on $\theta$ and $G$ is diagonal with $G_{\theta\theta}=1$.
But the key phenomenon appears when the loss has $\phi$-dependence:
for instance, take
\[
\mathcal{L}(\theta,\phi)=\langle X\rangle = \sin\theta\cos\phi.
\]
Then
\[
\nabla\mathcal{L}=
\begin{pmatrix}
	\cos\theta\cos\phi\\
	-\sin\theta\sin\phi
\end{pmatrix},
\qquad
\Delta\theta_{\mathrm{QNG}}=-\eta
\begin{pmatrix}
	\cos\theta\cos\phi\\
	-\dfrac{\sin\theta\sin\phi}{\sin^2\theta}
\end{pmatrix}
=
-\eta
\begin{pmatrix}
	\cos\theta\cos\phi\\
	-\dfrac{\sin\phi}{\sin\theta}
\end{pmatrix}.
\]
So the $\phi$-update is rescaled by $1/\sin\theta$, i.e.\ \emph{near the poles} ($\sin\theta\approx0$),
changing $\phi$ costs almost nothing in Euclidean coordinates but is \emph{physically meaningless}
(because all longitudes meet at the pole). QNG corrects this by using the proper geometry.

\subsection{Circuit-level geometric interpretation}

\subsubsection*{QNG as ``shortest physical move that improves the loss''}
The update $\delta\theta=-\eta G^{-1}\nabla\mathcal{L}$ is the direction that,
for a given infinitesimal state-space step length,
achieves maximal first-order decrease in $\mathcal{L}$.

\subsubsection*{Geodesics and coordinate singularities}
For the 1-qubit example, $G$ is (up to convention) the sphere metric.
Thus, QNG turns parameter updates into motion consistent with spherical geometry:
\begin{itemize}
	\item Near the equator ($\sin\theta\approx1$), $\phi$ is a good coordinate and updates behave normally.
	\item Near the poles ($\sin\theta\approx0$), $\phi$ is ill-conditioned: many $\phi$ represent nearly identical states.
	QNG naturally damps or rescales $\phi$-motion by the metric.
\end{itemize}

\subsubsection*{Generator/variance view (link to QFIM computation)}
If your circuit has the form
\[
U(\theta)=e^{-i\theta A}\cdot(\text{rest}),
\]
then in many cases the pure-state QFIM element reduces to a variance-like object:
\[
G_{\theta\theta}=4\bigl(\langle A^2\rangle-\langle A\rangle^2\bigr),
\]
where the expectation is taken in the current state (possibly with conjugations depending on parameter ordering).
This is why QNG often ``slows down'' directions with small variance (flat directions) and re-scales sensitive ones.

\subsection{Algorithmic applications}

\subsubsection*{Application 1: Variational energy minimization (VQE)}
In VQE, $\mathcal{L}(\theta)=\langle H\rangle$ for a Hamiltonian $H$.
You estimate:
\begin{itemize}
	\item $\mathcal{L}$ by measuring Pauli terms,
	\item $\nabla\mathcal{L}$ by parameter-shift or finite differences,
	\item $G$ by overlap/SLD/variance-based estimators (depending on ansatz and measurement access).
\end{itemize}
Then solve $(G+\lambda I)\Delta\theta=\nabla\mathcal{L}$ and update $\theta\leftarrow\theta-\eta\Delta\theta$.
Empirically, QNG can reduce sensitivity to parameterization and improve conditioning.

\subsubsection*{Application 2: Quantum circuit learning / QML}
If $\mathcal{L}$ is a supervised loss depending on measured probabilities, QNG
acts like ``natural gradient'' in the model manifold: it corrects distortion introduced by
measurement nonlinearity and redundant parameterization.

\subsubsection*{Application 3: Noise-aware and constrained optimization}
Under noise, the effective geometry can change (mixed-state QFIM).
A common practical strategy is to:
\begin{itemize}
	\item compute/estimate a mixed-state QFIM (or an approximation),
	\item add damping ($\lambda I$) to avoid overreacting to noisy metric estimates,
	\item restrict to blocks (layerwise QNG) for scalability.
\end{itemize}

\subsubsection*{Implementation note: why this is hardware-relevant}
Each iteration requires repeated, structured operations:
\[
\text{stream measurements} \rightarrow \text{accumulate statistics} \rightarrow \text{form }G,\nabla\mathcal{L}
\rightarrow \text{solve small linear systems}.
\]
This is exactly the type of deterministic, low-latency pipeline where an FPGA can help:
online accumulation, fixed-shape matrix kernels, and stable tail-latency.

\subsection{Geometric parallel with Grover's algorithm}

Grover's iterate is a product of two reflections, producing a \emph{rotation} in a 2D subspace.
QNG, in contrast, is a \emph{local} geometric rule that chooses the steepest descent direction
with respect to the Fubini--Study/QFIM metric.
The parallel is conceptual:

\begin{itemize}
	\item \textbf{Grover:} global, exact rotation by a fixed angle in a known 2D plane (closed form).
	\item \textbf{QNG:} local, adaptive ``best direction'' computed from the metric and gradient,
	approximating geodesic-like motion on the state manifold.
\end{itemize}

In both cases, \emph{state space geometry} is the right language:
Grover exploits it analytically; QNG exploits it algorithmically.

\subsection{Exercises}

\begin{exercise}[Compute QFIM for a different 1-qubit parametrization]
	Let
	\[
	\ket{\psi(\alpha,\beta)} = R_X(\alpha)\,R_Z(\beta)\ket{0}.
	\]
	Compute the pure-state QFIM $G(\alpha,\beta)$ explicitly.
\end{exercise}

\noindent\textbf{Solution.}
First compute the state. Since $R_Z(\beta)\ket0=e^{-i\beta/2}\ket0$ (global phase),
\[
\ket{\psi(\alpha,\beta)} \sim R_X(\alpha)\ket0.
\]
Now
\[
R_X(\alpha)=e^{-i\alpha X/2}=\cos\frac{\alpha}{2}I-i\sin\frac{\alpha}{2}X
=
\begin{pmatrix}
	\cos\frac{\alpha}{2} & -i\sin\frac{\alpha}{2}\\
	-i\sin\frac{\alpha}{2} & \cos\frac{\alpha}{2}
\end{pmatrix},
\]
so
\[
\ket{\psi(\alpha,\beta)} \sim
\cos\frac{\alpha}{2}\ket0 - i\sin\frac{\alpha}{2}\ket1.
\]
This has \emph{no dependence} on $\beta$ in the projective state, hence $\partial_\beta\ket{\psi}$ is purely
a gauge (global phase) contribution and the physical metric must satisfy $G_{\beta\beta}=G_{\alpha\beta}=0$.
Compute $G_{\alpha\alpha}$ using the formula (with $\partial_\alpha$ only):
\[
\partial_\alpha\ket{\psi}
=
-\frac12\sin\frac{\alpha}{2}\ket0 - i\frac12\cos\frac{\alpha}{2}\ket1,
\qquad
\langle\psi|\partial_\alpha\psi\rangle = 0,
\]
\[
\langle\partial_\alpha\psi|\partial_\alpha\psi\rangle
=\frac14\sin^2\frac{\alpha}{2}+\frac14\cos^2\frac{\alpha}{2}=\frac14.
\]
Thus
\[
G_{\alpha\alpha}=4\cdot\frac14=1,
\qquad
G(\alpha,\beta)=\begin{pmatrix}1&0\\0&0\end{pmatrix}.
\]
Interpretation: $\beta$ is redundant here (acts as a global phase on this input state).

\medskip

\begin{exercise}[QNG step on the Bloch sphere metric]
	For the two-parameter circuit in the chapter,
	\[
	G(\theta,\phi)=\mathrm{diag}(1,\sin^2\theta).
	\]
	Given the loss $\mathcal{L}(\theta,\phi)=\sin\theta\cos\phi$ (i.e.\ $\langle X\rangle$),
	compute the Euclidean gradient step and the QNG step explicitly.
\end{exercise}

\noindent\textbf{Solution.}
Compute gradients:
\[
\partial_\theta \mathcal{L}=\cos\theta\cos\phi,
\qquad
\partial_\phi \mathcal{L}=-\sin\theta\sin\phi.
\]
Euclidean step ($\delta=-\eta\nabla\mathcal{L}$):
\[
\delta\theta_{\mathrm{Euc}}=-\eta\cos\theta\cos\phi,
\qquad
\delta\phi_{\mathrm{Euc}}=+\eta\sin\theta\sin\phi.
\]
QNG step uses $G^{-1}=\mathrm{diag}(1,\sin^{-2}\theta)$:
\[
\begin{pmatrix}\delta\theta_{\mathrm{QNG}}\\ \delta\phi_{\mathrm{QNG}}\end{pmatrix}
=
-\eta
\begin{pmatrix}
	1 & 0\\
	0 & \sin^{-2}\theta
\end{pmatrix}
\begin{pmatrix}
	\cos\theta\cos\phi\\
	-\sin\theta\sin\phi
\end{pmatrix}
=
-\eta
\begin{pmatrix}
	\cos\theta\cos\phi\\
	-\dfrac{\sin\phi}{\sin\theta}
\end{pmatrix}.
\]
So
\[
\delta\theta_{\mathrm{QNG}}=-\eta\cos\theta\cos\phi,
\qquad
\delta\phi_{\mathrm{QNG}}=+\eta\,\frac{\sin\phi}{\sin\theta}.
\]
Key point: near $\theta\approx 0,\pi$, $\phi$ is ill-conditioned, and the metric compensates.

\medskip

\begin{exercise}[Redundancy and singular metrics]
	Show that if two different parameter values $\theta$ and $\theta'$ produce the same state
	(up to global phase), then the QFIM must be singular along some direction.
	Give an example using a 1-qubit circuit.
\end{exercise}

\noindent\textbf{Solution.}
If there is a continuous redundancy, there exists a curve $\theta(t)$ such that
$\ket{\psi(\theta(t))}$ is constant in $\CP^{d-1}$.
Then the physical tangent vector $|\partial_t\psi\rangle$ is purely a gauge direction,
so the Fubini--Study/QFIM norm of that direction is zero:
\[
\|\partial_t\psi\|_{\mathrm{FS}}^2=\partial_t\theta^i\,G_{ij}\,\partial_t\theta^j=0.
\]
Hence $G$ has a nontrivial null vector $\partial_t\theta$, i.e.\ it is singular.

Example: from the previous exercise,
\[
\ket{\psi(\alpha,\beta)}=R_X(\alpha)R_Z(\beta)\ket0 \sim R_X(\alpha)\ket0,
\]
so changing $\beta$ does not change the physical state. Therefore the $\beta$-direction has zero metric length and
$G_{\beta\beta}=0$, making $G$ singular.

\medskip

\begin{exercise}[Solve a damped QNG system]
	Let
	\[
	G=\begin{pmatrix}1&0\\0&\sin^2\theta\end{pmatrix},\quad
	g=\nabla\mathcal{L}=
	\begin{pmatrix}g_\theta\\ g_\phi\end{pmatrix}.
	\]
	Solve $(G+\lambda I)\Delta = g$ explicitly and write the update $\theta\leftarrow\theta-\eta\Delta$.
\end{exercise}

\noindent\textbf{Solution.}
Since $G+\lambda I$ is diagonal,
\[
G+\lambda I=
\begin{pmatrix}
	1+\lambda & 0\\
	0 & \sin^2\theta+\lambda
\end{pmatrix}.
\]
Thus
\[
\Delta=
\begin{pmatrix}
	\dfrac{g_\theta}{1+\lambda}\\[0.6em]
	\dfrac{g_\phi}{\sin^2\theta+\lambda}
\end{pmatrix}.
\]
Update:
\[
\theta\leftarrow\theta-\eta\,\frac{g_\theta}{1+\lambda},
\qquad
\phi\leftarrow\phi-\eta\,\frac{g_\phi}{\sin^2\theta+\lambda}.
\]
Damping prevents blow-up when $\sin\theta\approx 0$.

\medskip

\begin{exercise}[Geometry meets Grover (conceptual)]
	In Grover search, state evolution stays in a 2D plane and performs a rotation by a fixed angle.
	In QNG, updates are local and metric-shaped.
	Write a short paragraph explaining in what sense both are ``geometry-driven''.
\end{exercise}

\noindent\textbf{Solution.}
Grover's algorithm uses exact geometric structure: two reflections produce a rotation in a known 2D invariant subspace,
so the amplitude on the marked state increases by deterministic angular motion with closed-form probabilities.
QNG uses differential geometry: it defines ``steepest descent'' with respect to the intrinsic state-space metric (QFIM),
so parameter updates correspond to the most efficient local movement in the physical manifold, correcting coordinate distortions.
Both replace naive coordinate intuition by the geometry of the state space; Grover exploits it analytically (global rotation),
while QNG exploits it algorithmically (local metric-preconditioned descent).


\section{FPGA III: Streaming Linear Algebra for Hybrid Quantum Algorithms}
\label{sec:fpga3}

\subsection{Objective}
Hybrid quantum algorithms (VQE, QAOA, QNG, adaptive tomography, error-mitigation loops)
are not ``just quantum circuits.'' They are \emph{streaming pipelines}:
\[
\text{shots} \;\to\; \text{classical statistics} \;\to\; \text{linear algebra} \;\to\; \text{parameter update}
\;\to\; \text{next shots}.
\]
The quantum device produces a high-rate, low-bitwidth stream (measurement outcomes),
but the \emph{classical} side must:
\begin{itemize}
	\item aggregate (online mean/variance/covariance),
	\item build small matrices (Gram/QFIM/covariance),
	\item solve small linear systems (damped inverse / preconditioned step),
	\item keep tail latency (p99/p999) tight enough for feedback/control.
\end{itemize}
This section explains why these workloads fit FPGA well: deterministic microkernels,
fixed shapes, high reuse, low latency, and clean dataflow.

\subsection{From quantum measurements to classical data streams}

\subsubsection*{The data that actually arrives}
On most platforms, the classical readout per shot is a bitstring
\[
\mathbf{b}^{(s)}\in\{0,1\}^n
\quad \text{or} \quad
\mathbf{y}^{(s)}\in\{\pm1\}^n
\]
(after mapping $0\mapsto +1$, $1\mapsto -1$ per qubit).
From this, we form estimators of expectations like
\[
\langle Z_i\rangle,\quad \langle Z_i Z_j\rangle,\quad
\langle P_k\rangle \;\text{for Pauli strings }P_k,
\]
depending on what basis rotations were used before measurement.

\subsubsection*{Measurement basis choice = which random variable you sample}
A circuit typically ends with a basis-change layer, then computational-basis measurement.
So each shot effectively samples a random variable $X$ (observable) with outcomes $\pm 1$,
and the hardware produces \emph{bits}; your pipeline reconstructs \emph{statistics}:
\[
\widehat{\mu} \approx \mathbb{E}[X],
\qquad
\widehat{\sigma}^2 \approx \mathbb{E}[X^2]-\mathbb{E}[X]^2.
\]
FPGA becomes relevant because this is exactly the structure of online estimation.

\begin{figure}[t]
	\centering
	\resizebox{\textwidth}{!}{%
		\begin{tikzpicture}[
			>=Latex,
			rounded corners,
			node distance=8mm,
			font=\small
			]
			\tikzset{
				block/.style={draw, thick, align=center, minimum height=10mm},
				arr/.style={->, thick}
			}
			
			\node[block, minimum width=26mm] (qpu) {QPU\\(shots)};
			\node[block, right=10mm of qpu, minimum width=34mm] (readout) {Readout\\bits/words};
			\node[block, right=10mm of readout, minimum width=38mm] (agg) {Online\\aggregation};
			\node[block, right=10mm of agg, minimum width=40mm] (linalg) {Small\\linear algebra};
			
			\node[block, below=10mm of linalg, minimum width=40mm] (update) {Parameter\\update};
			\node[block, below=10mm of update, minimum width=40mm] (sched) {Next\\pulse schedule};
			
			\draw[arr] (qpu) -- node[above]{shot stream} (readout);
			\draw[arr] (readout) -- node[above]{unpacked $\pm1$} (agg);
			\draw[arr] (agg) -- node[above]{means/covs} (linalg);
			\draw[arr] (linalg) -- node[right]{$\Delta\theta$} (update);
			\draw[arr] (update) -- (sched);
			\draw[arr] (sched.west) |- (qpu.south);
			
			\node[align=center] at ($(agg)+(0,-9mm)$)
			{\small FPGA sweet spot:\\ \small deterministic streaming kernels};
		\end{tikzpicture}%
	}
	\caption{Hybrid loop as a feedback pipeline: measurements become streams; streams become statistics; statistics drive small linear algebra; updates schedule the next quantum shots.}
	\label{fig:hybrid-streaming-loop}
\end{figure}
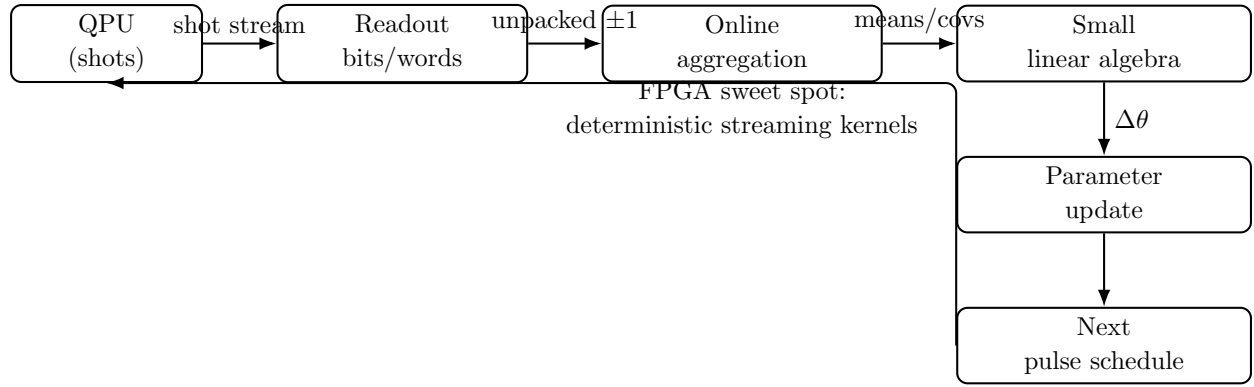

\subsection{Pattern I: shot aggregation as online estimation}

\subsubsection*{1) Online mean (one observable)}
Suppose each shot produces $x_s\in\{\pm1\}$ for some observable (e.g.\ $Z$, or $X$ after a basis change).
The empirical mean after $N$ shots is
\[
\widehat{\mu}_N = \frac{1}{N}\sum_{s=1}^N x_s.
\]
Streaming update:
\[
S_N := S_{N-1}+x_N,\qquad \widehat{\mu}_N=\frac{S_N}{N}.
\]
This is a single adder + counter; extremely FPGA-friendly.

\subsubsection*{2) Online variance (Welford, numerically stable)}
For calibration and error bars you want variance. A stable streaming update is Welford's algorithm:
\[
\mu_N = \mu_{N-1} + \frac{x_N-\mu_{N-1}}{N},
\qquad
M2_N = M2_{N-1} + (x_N-\mu_{N-1})(x_N-\mu_N),
\]
then $\widehat{\sigma}^2 = \frac{M2_N}{N-1}$.
This avoids catastrophic cancellation when $N$ is large.

\subsubsection*{3) Online covariance (many observables)}
For a vector of measured features $\mathbf{x}_s\in\R^d$ (e.g.\ multiple Pauli terms),
the empirical mean and covariance are
\[
\widehat{\mu}=\frac1N\sum_s \mathbf{x}_s,
\qquad
\widehat{\Sigma}=\frac1N\sum_s (\mathbf{x}_s-\widehat{\mu})(\mathbf{x}_s-\widehat{\mu})^\top.
\]
In streaming form, maintain:
\[
\mathbf{S}=\sum_s \mathbf{x}_s,
\qquad
\mathbf{C}=\sum_s \mathbf{x}_s\mathbf{x}_s^\top,
\]
then
\[
\widehat{\mu}=\frac{\mathbf{S}}{N},
\qquad
\widehat{\Sigma}=\frac{\mathbf{C}}{N}-\widehat{\mu}\,\widehat{\mu}^\top.
\]
This is a rank-1 outer-product accumulation per shot, i.e.\ a fixed microkernel.

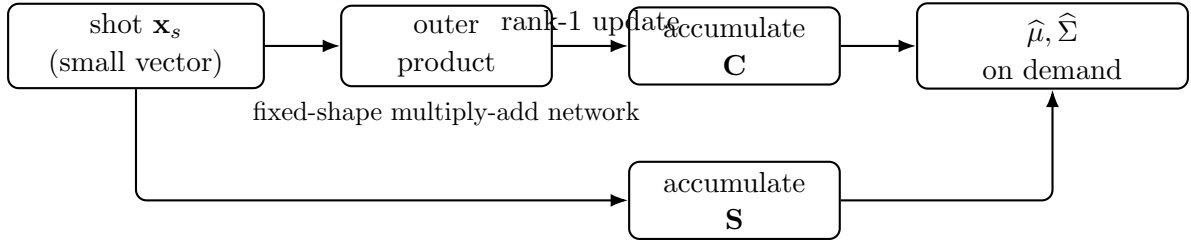
\begin{figure}[t]
	\centering
	\begin{tikzpicture}[>=Latex, rounded corners]
		\node[draw, thick, align=center, minimum width=34mm, minimum height=10mm] (in) {shot $\mathbf{x}_s$\\(small vector)};
		\node[draw, thick, right=10mm of in, align=center, minimum width=28mm, minimum height=10mm] (outer) {outer\\product};
		\node[draw, thick, right=10mm of outer, align=center, minimum width=28mm, minimum height=10mm] (acc) {accumulate\\$\mathbf{C}$};
		\node[draw, thick, below=10mm of acc, align=center, minimum width=28mm, minimum height=10mm] (sum) {accumulate\\$\mathbf{S}$};
		\node[draw, thick, right=10mm of acc, align=center, minimum width=36mm, minimum height=10mm] (out) {$\widehat{\mu},\widehat{\Sigma}$\\on demand};
		
		\draw[->, thick] (in) -- (outer);
		\draw[->, thick] (outer) -- node[above]{rank-1 update} (acc);
		\draw[->, thick] (in) |- (sum);
		\draw[->, thick] (acc) -- (out);
		\draw[->, thick] (sum) -| (out);
		
		\node[align=center] at ($(outer)+(0,-9mm)$) {\small fixed-shape multiply-add network};
	\end{tikzpicture}
	\caption{Streaming covariance accumulation: each shot contributes a rank-1 outer product $\mathbf{x}_s\mathbf{x}_s^\top$ plus a vector sum. This pattern repeats across many hybrid algorithms.}
	\label{fig:streaming-cov}
\end{figure}

\subsection{Pattern II: QFIM updates as a matrix--vector microkernel}

\subsubsection*{Where QFIM shows up in practice}
In QNG and related methods, you need a metric matrix $G(\theta)\in\R^{p\times p}$ and a gradient $g(\theta)\in\R^p$.
Then you solve a small system
\[
(G+\lambda I)\Delta\theta = g
\quad\Rightarrow\quad
\theta \leftarrow \theta - \eta\,\Delta\theta.
\]
Even if $p$ is only 10--200 (typical for modest ans\"atze), this happens many times and must be \emph{stable and fast}.

\subsubsection*{Common estimator structure (Gram / covariance form)}
Many QFIM estimators reduce to a Gram/covariance form
\[
G \approx \mathbb{E}\bigl[\mathbf{v}\mathbf{v}^\top\bigr]-\mathbb{E}[\mathbf{v}]\,\mathbb{E}[\mathbf{v}]^\top,
\]
where $\mathbf{v}$ is a per-shot (or per-batch) feature vector derived from measurements
(parameter-shift features, generator features, score-function features, etc.).
Thus QFIM estimation is again a streaming outer-product accumulation, exactly Pattern I.

\subsubsection*{Microkernel view}
At the hardware level, a dominant primitive is:
\[
\mathbf{C} \leftarrow \mathbf{C} + \mathbf{v}\mathbf{v}^\top
\qquad\text{and/or}\qquad
\mathbf{y} \leftarrow \mathbf{y} + \mathbf{A}\mathbf{x}
\]
for small, fixed sizes. FPGAs excel at these because:
\begin{itemize}
	\item the shape is known at compile time,
	\item memory access is predictable (streaming),
	\item multiply-add pipelines can be fully unrolled or partially unrolled,
	\item the latency is deterministic.
\end{itemize}

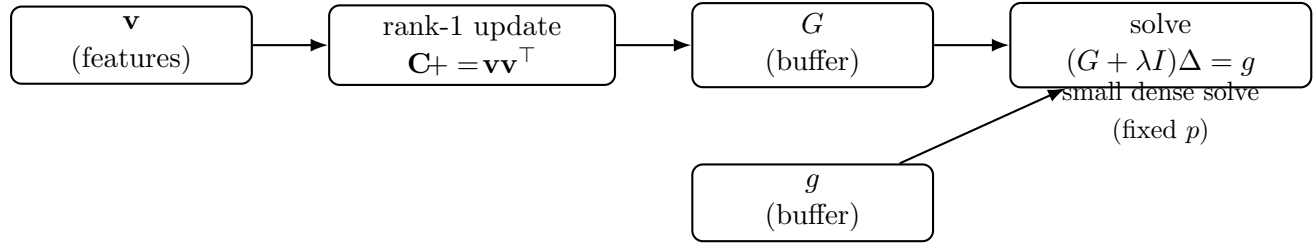
\begin{figure}[t]
	\centering
	\begin{tikzpicture}[>=Latex, rounded corners]
		\node[draw, thick, align=center, minimum width=32mm, minimum height=10mm] (v) {$\mathbf{v}$\\(features)};
		\node[draw, thick, right=10mm of v, align=center, minimum width=38mm, minimum height=10mm] (gemm) {rank-1 update\\$\mathbf{C}\!+=\!\mathbf{v}\mathbf{v}^\top$};
		\node[draw, thick, right=10mm of gemm, align=center, minimum width=32mm, minimum height=10mm] (G) {$G$\\(buffer)};
		\node[draw, thick, below=10mm of G, align=center, minimum width=32mm, minimum height=10mm] (g) {$g$\\(buffer)};
		\node[draw, thick, right=10mm of G, align=center, minimum width=40mm, minimum height=10mm] (solve) {solve\\$(G+\lambda I)\Delta=g$};
		
		\draw[->, thick] (v) -- (gemm);
		\draw[->, thick] (gemm) -- (G);
		\draw[->, thick] (g) -- (solve);
		\draw[->, thick] (G) -- (solve);
		
		\node[align=center] at ($(solve)+(0,-9mm)$) {\small small dense solve\\\small (fixed $p$)};
	\end{tikzpicture}
	\caption{QFIM/QNG pipeline: (i) streaming outer-product updates to form $G$ (and $g$), then (ii) a small dense linear solve. Both are dominated by fixed-shape linear algebra.}
	\label{fig:qfim-microkernel}
\end{figure}

\subsubsection*{Solving small systems (practical options)}
For modest $p$, typical FPGA-friendly solve strategies include:
\begin{itemize}
	\item \textbf{Cholesky} for SPD matrices ($G+\lambda I \succ 0$), $\;O(p^3/3)$, stable.
	\item \textbf{LDL$^\top$} variant (no square roots), often good for fixed-point pipelines.
	\item \textbf{Iterative (CG)} if you can accept a few iterations; very regular compute.
	\item \textbf{Diagonal / block-diagonal} approximations (layerwise QNG), turning solve into divides.
\end{itemize}
In many workflows, the best first step is to build a reliable streaming $G$ and use a \emph{damped} solve.

\subsection{Pattern III: tight online optimization loops}

\subsubsection*{The controlling constraint: tail latency}
If the hybrid loop is used for real-time adaptation or drift tracking, you care less about throughput and more about:
\[
\text{bounded worst-case response time (p99/p999)}.
\]
GPU/CPU scheduling jitter, cache effects, and OS interference can break real-time guarantees.
FPGA pipelines, once placed and timed, provide much tighter jitter bounds.

\subsubsection*{A minimal loop model}
Let
\begin{itemize}
	\item $T_{\mathrm{shot}}$ = time to execute one shot (including reset, pulses, measurement),
	\item $N$ = number of shots per iteration,
	\item $T_{\mathrm{agg}}$ = aggregation time (often overlapped with readout),
	\item $T_{G,g}$ = time to assemble $G$ and $g$ (or finalize streaming buffers),
	\item $T_{\mathrm{solve}}$ = time to solve $(G+\lambda I)\Delta=g$,
	\item $T_{\mathrm{update}}$ = time to update parameters and schedule next pulses.
\end{itemize}
Then the iteration latency is roughly
\[
T_{\mathrm{iter}}
\approx
N\,T_{\mathrm{shot}}
+\max\{T_{\mathrm{agg}},T_{G,g}\}
+T_{\mathrm{solve}}
+T_{\mathrm{update}}.
\]
Streaming design aims to hide $T_{\mathrm{agg}}$ inside the shot loop so that only a short tail remains.

\subsubsection*{Hardware sketch: overlapped aggregation}
A canonical pipeline is:
\[
\text{while shots stream: }
\mathbf{S}\!+=\!\mathbf{x},\ \mathbf{C}\!+=\!\mathbf{x}\mathbf{x}^\top
\quad\text{(no stalls)}
\quad\Rightarrow\quad
\text{after last shot: finalize + solve}.
\]
This structure is exactly what hardware dataflow is designed for.

\begin{figure}[t]
	\centering
	\begin{tikzpicture}[>=Latex, rounded corners]
		\node[draw, thick, align=center, minimum width=26mm, minimum height=8mm] (s1) {shot 1};
		\node[draw, thick, right=2mm of s1, align=center, minimum width=26mm, minimum height=8mm] (s2) {shot 2};
		\node[draw, thick, right=2mm of s2, align=center, minimum width=26mm, minimum height=8mm] (s3) {shot 3};
		\node[right=2mm of s3] (dots) {$\cdots$};
		\node[draw, thick, right=4mm of dots, align=center, minimum width=26mm, minimum height=8mm] (sN) {shot $N$};
		
		\node[draw, thick, below=7mm of s2, align=center, minimum width=84mm, minimum height=8mm] (agg) {streaming aggregation runs continuously};
		\node[draw, thick, below=7mm of sN, align=center, minimum width=44mm, minimum height=8mm] (solve) {finalize + solve};
		
		\draw[->, thick] (s1.east) -- (s2.west);
		\draw[->, thick] (s2.east) -- (s3.west);
		\draw[->, thick] (s3.east) -- (dots.west);
		\draw[->, thick] (dots.east) -- (sN.west);
		
		\draw[->, thick] (s1.south) |- (agg.west);
		\draw[->, thick] (sN.south) -- (solve.north);
		
		\node[align=center] at ($(solve)+(0,-7mm)$) {\small short deterministic tail};
	\end{tikzpicture}
	\caption{Latency hiding by design: aggregation overlaps with shot execution; only a short ``finalize + solve'' tail remains.}
	\label{fig:overlap-aggregation}
\end{figure}
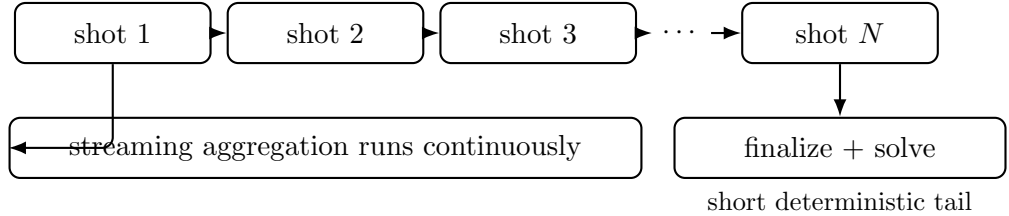

\subsection{Why this matters for quantum algorithms}

\subsubsection*{1) QNG and adaptive methods are only as good as their loop time}
If you cannot iterate fast enough, you cannot:
\begin{itemize}
	\item track drift / recalibrate frequently,
	\item run inner loops (line search, trust region),
	\item do adaptive measurement allocation (importance sampling),
	\item maintain stability in feedback-controlled experiments.
\end{itemize}

\subsubsection*{2) Measurement dominates, but classical jitter breaks ``physics timing''}
Even when $N\,T_{\mathrm{shot}}$ dominates, classical jitter can destroy:
\begin{itemize}
	\item deterministic control scheduling,
	\item consistent batching (important for estimator variance),
	\item clean latency budgets in QEC/QND-style protocols.
\end{itemize}

\subsubsection*{3) Algorithms with small matrices appear everywhere}
Small dense linear algebra is a recurring motif:
\[
\text{QFIM/QNG},\quad \text{Gauss--Newton / LM},\quad
\text{Kalman-style tracking},\quad
\text{adaptive tomography},\quad
\text{error mitigation fits}.
\]
If your system can stream and solve small systems reliably, many hybrid algorithms become viable.

\subsection{Conceptual takeaway}
The ``classical part'' of hybrid quantum computing is not a generic CPU workload.
It is a repeated, fixed-shape, low-latency linear-algebra microkernel pipeline fed by a measurement stream.
That is why hardware acceleration appears naturally---not as a luxury, but as an architectural consequence.

\subsection{Exercises}

\begin{exercise}[Online mean/variance for $\pm1$ shots]
	Let $x_s\in\{\pm1\}$ be shot outcomes of a fixed observable.
	Show that the streaming mean update
	\[
	\mu_N=\mu_{N-1}+\frac{x_N-\mu_{N-1}}{N}
	\]
	produces the same $\mu_N$ as $\mu_N=\frac1N\sum_{s=1}^N x_s$.
	Then implement Welford's variance update formulas and verify on a toy sequence.
\end{exercise}

\noindent\textbf{Solution.}
For the mean, write
\[
\mu_N=\frac1N\sum_{s=1}^N x_s=\frac{N-1}{N}\cdot\frac1{N-1}\sum_{s=1}^{N-1}x_s+\frac{x_N}{N}
=\frac{N-1}{N}\mu_{N-1}+\frac{x_N}{N}
=\mu_{N-1}+\frac{x_N-\mu_{N-1}}{N}.
\]
For Welford variance, use:
\[
\mu_N=\mu_{N-1}+\frac{x_N-\mu_{N-1}}{N},\qquad
M2_N=M2_{N-1}+(x_N-\mu_{N-1})(x_N-\mu_N).
\]
Finally $\widehat{\sigma}^2=M2_N/(N-1)$.
Verify by comparing against the batch formula $\frac1{N-1}\sum_{s=1}^N(x_s-\mu_N)^2$ on a short sequence.

\medskip

\begin{exercise}[Streaming covariance via outer products]
	Let $\mathbf{x}_s\in\R^d$ be a per-shot feature vector.
	Define $\mathbf{S}=\sum_s\mathbf{x}_s$ and $\mathbf{C}=\sum_s\mathbf{x}_s\mathbf{x}_s^\top$.
	Show that
	\[
	\widehat{\Sigma}=\frac{\mathbf{C}}{N}-\widehat{\mu}\,\widehat{\mu}^\top,
	\qquad \widehat{\mu}=\frac{\mathbf{S}}{N},
	\]
	equals the usual empirical covariance $\frac1N\sum_s(\mathbf{x}_s-\widehat{\mu})(\mathbf{x}_s-\widehat{\mu})^\top$.
\end{exercise}

\noindent\textbf{Solution.}
Expand the empirical covariance:
\[
\frac1N\sum_s(\mathbf{x}_s-\widehat{\mu})(\mathbf{x}_s-\widehat{\mu})^\top
=
\frac1N\sum_s\Bigl(\mathbf{x}_s\mathbf{x}_s^\top-\mathbf{x}_s\widehat{\mu}^\top-\widehat{\mu}\mathbf{x}_s^\top+\widehat{\mu}\widehat{\mu}^\top\Bigr).
\]
Now use $\sum_s\mathbf{x}_s=\mathbf{S}=N\widehat{\mu}$ to get
\[
\frac1N\sum_s\mathbf{x}_s\mathbf{x}_s^\top - \widehat{\mu}\widehat{\mu}^\top-\widehat{\mu}\widehat{\mu}^\top+\widehat{\mu}\widehat{\mu}^\top
=
\frac{\mathbf{C}}{N}-\widehat{\mu}\widehat{\mu}^\top,
\]
as claimed.

\medskip

\begin{exercise}[Damped QNG solve as a streaming-friendly primitive]
	Suppose your QNG update solves
	\[
	(G+\lambda I)\Delta=g
	\]
	for $p=2$ with
	\[
	G=\begin{pmatrix}a&c\\c&b\end{pmatrix},\quad \lambda>0.
	\]
	Write an explicit closed-form expression for $\Delta$ in terms of $(a,b,c,\lambda)$ and $(g_1,g_2)$.
	Explain why this is attractive for hardware when $p$ is small and fixed.
\end{exercise}

\noindent\textbf{Solution.}
Let
\[
M=G+\lambda I=
\begin{pmatrix}a+\lambda&c\\c&b+\lambda\end{pmatrix}.
\]
Its inverse (since $\lambda>0$ typically makes it SPD) is
\[
M^{-1}=\frac{1}{\det M}
\begin{pmatrix}b+\lambda&-c\\-c&a+\lambda\end{pmatrix},
\qquad
\det M=(a+\lambda)(b+\lambda)-c^2.
\]
Thus
\[
\Delta = M^{-1}g
=
\frac{1}{(a+\lambda)(b+\lambda)-c^2}
\begin{pmatrix}
	(b+\lambda)g_1-cg_2\\
	-cg_1+(a+\lambda)g_2
\end{pmatrix}.
\]
Hardware attractiveness: for fixed small $p$, this becomes a fixed datapath of multiply-adds plus one reciprocal (or iterative reciprocal),
with deterministic latency and no dynamic memory.

\medskip

\begin{exercise}[Latency model sanity check]
	Assume:
	\[
	T_{\mathrm{shot}}=2\ \mu s,\quad N=2000,\quad
	T_{\mathrm{solve}}=20\ \mu s,\quad T_{\mathrm{update}}=5\ \mu s,
	\]
	and aggregation is perfectly overlapped.
	Compute $T_{\mathrm{iter}}$.
	Then explain qualitatively what happens to $T_{\mathrm{iter}}$ if aggregation is \emph{not} overlapped and costs $150\ \mu s$.
\end{exercise}

\noindent\textbf{Solution.}
If aggregation overlaps completely, the iteration time is
\[
T_{\mathrm{iter}}\approx N T_{\mathrm{shot}} + T_{\mathrm{solve}} + T_{\mathrm{update}}
=2000\cdot 2\mu s + 20\mu s + 5\mu s
=4000\mu s + 25\mu s
=4025\ \mu s.
\]
If aggregation does not overlap, add its cost to the tail:
\[
T_{\mathrm{iter}}\approx 4000\mu s + 150\mu s + 20\mu s + 5\mu s=4175\ \mu s.
\]
The key point: even modest non-overlapped work increases the tail, and tail latency matters for tight feedback loops.

\medskip

\begin{exercise}[Design a feature vector for QFIM-style accumulation]
	Propose a per-shot (or per-batch) feature vector $\mathbf{v}\in\R^p$ such that accumulating
	$\sum \mathbf{v}\mathbf{v}^\top$ could produce a useful approximation to a QFIM or a preconditioner.
	Your answer can be conceptual (no need to match a specific estimator), but it must be explicit about:
	(i) what $\mathbf{v}$ is computed from, and (ii) why $\mathbf{v}\mathbf{v}^\top$ is informative.
\end{exercise}

\noindent\textbf{Solution.}
One simple explicit choice is to use parameter-shift gradient samples. For each parameter $\theta_i$,
estimate $\partial_i \mathcal{L}$ from two shifted circuits and treat the resulting gradient vector
$g^{(s)}\in\R^p$ (from a batch $s$) as a feature:
\[
\mathbf{v}^{(s)} := g^{(s)}.
\]
Then
\[
\sum_s \mathbf{v}^{(s)}\mathbf{v}^{(s)\top}
=
\sum_s g^{(s)}g^{(s)\top}
\]
is an empirical second-moment matrix of gradients, which serves as a curvature/preconditioning proxy
(similar in spirit to Gauss--Newton / Fisher-style approximations).
It is informative because directions with consistently large gradient magnitude/variance receive larger diagonal (and possibly correlated) weight,
helping scale updates and stabilize optimization. This is not the exact QFIM, but it matches the same streaming outer-product pattern and is hardware-friendly.

	\Part{Multi-Qubit Structure: Entanglement and Its Geometry}
	
\section{Multi-Qubit Systems and Entanglement}
\label{sec:entanglement}

\subsection{Objective}
A single qubit already forces you to think geometrically (Bloch sphere / projective space).
Multiple qubits force a deeper shift: the state space is no longer a simple sphere-like object,
and the most important new phenomenon is \emph{entanglement}.

This section builds the multi-qubit formalism in a way that is:
\begin{itemize}
	\item \textbf{structural}: tensor products explain what a composite system means,
	\item \textbf{computational}: we compute reduced states via partial trace explicitly,
	\item \textbf{algorithmic}: we connect entanglement to interference and speedups,
	\item \textbf{hardware-aware}: entangling gates are expensive and define compilation constraints.
\end{itemize}

\subsection{Tensor products and the meaning of ``composite system''}

\subsubsection*{1) The postulate (what composite means in quantum)}
If system $A$ has state space $\mathcal{H}_A$ and system $B$ has state space $\mathcal{H}_B$,
then the composite system has state space
\[
\mathcal{H}_{AB}=\mathcal{H}_A\otimes \mathcal{H}_B.
\]
For two qubits, $\mathcal{H}_A\simeq \C^2$, $\mathcal{H}_B\simeq \C^2$, hence
\[
\mathcal{H}_{AB}\simeq \C^2\otimes \C^2 \simeq \C^4.
\]

\subsubsection*{2) Meaning of simple tensors (separable preparation)}
A \emph{product state} is a simple tensor
\[
\ket{\psi}_{AB}=\ket{\psi}_A\otimes \ket{\phi}_B.
\]
Operationally: ``prepare $A$ in $\ket{\psi}$ and $B$ in $\ket{\phi}$ independently.''

\begin{rem}[Basis and bookkeeping]
	Fix computational bases $\{\ket0,\ket1\}$ on each qubit.
	Then the tensor product basis is
	\[
	\{\ket{00},\ket{01},\ket{10},\ket{11}\},
	\qquad
	\ket{ij}:=\ket{i}\otimes\ket{j}.
	\]
	A general two-qubit pure state is
	\[
	\ket{\Psi}=a\ket{00}+b\ket{01}+c\ket{10}+d\ket{11},
	\qquad a,b,c,d\in\C,\ \ |a|^2+|b|^2+|c|^2+|d|^2=1.
	\]
\end{rem}

\subsubsection*{3) Matrix view (very useful for Schmidt)}
For two qubits, group coefficients into a $2\times 2$ matrix
\[
M_\Psi=
\begin{pmatrix}
	a & b\\
	c & d
\end{pmatrix}.
\]
Then
\[
\ket{\Psi}=\sum_{i,j\in\{0,1\}}(M_\Psi)_{ij}\ket{i}\ket{j}.
\]
This matrix viewpoint turns Schmidt decomposition into an SVD.

\begin{figure}[t]
	\centering
	\begin{tikzpicture}[>=Latex, rounded corners]
		\node[draw, thick, align=center, minimum width=30mm, minimum height=10mm] (A) {Qubit $A$\\$\mathcal{H}_A\simeq\C^2$};
		\node[draw, thick, right=12mm of A, align=center, minimum width=30mm, minimum height=10mm] (B) {Qubit $B$\\$\mathcal{H}_B\simeq\C^2$};
		\node[draw, thick, below=10mm of $(A)!0.5!(B)$, align=center, minimum width=74mm, minimum height=10mm] (AB) {Composite system\\$\mathcal{H}_{AB}=\mathcal{H}_A\otimes\mathcal{H}_B\simeq\C^4$};
		
		\draw[->, thick] (A.south) -- (AB.north west);
		\draw[->, thick] (B.south) -- (AB.north east);
		
		\node[align=center] at ($(AB)+(0,-9mm)$) {\small product state: $\ket{\psi}\otimes\ket{\phi}$ \qquad vs.\qquad entangled: not factorizable};
	\end{tikzpicture}
	\caption{Composite system rule: the joint Hilbert space is a tensor product. Product states form a tiny subset; generic states are entangled.}
	\label{fig:composite-tensor}
\end{figure}
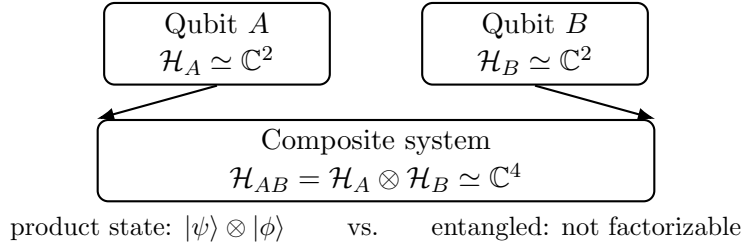

\subsection{Product vs.\ entangled pure states}

\begin{defn}[Product (separable) pure state]
	A two-qubit pure state $\ket{\Psi}\in\mathcal{H}_A\otimes\mathcal{H}_B$ is a \emph{product state}
	if there exist $\ket{\psi}\in\mathcal{H}_A$ and $\ket{\phi}\in\mathcal{H}_B$ such that
	\[
	\ket{\Psi}=\ket{\psi}\otimes\ket{\phi}.
	\]
	Otherwise it is \emph{entangled}.
\end{defn}

\subsubsection*{A fast test for two qubits (determinant test)}
Write $\ket{\Psi}=a\ket{00}+b\ket{01}+c\ket{10}+d\ket{11}$ and $M_\Psi=\begin{psmallmatrix}a&b\\ c&d\end{psmallmatrix}$.
Then:
\[
\ket{\Psi}\ \text{is product}
\quad\Longleftrightarrow\quad
\mathrm{rank}(M_\Psi)=1
\quad\Longleftrightarrow\quad
\det(M_\Psi)=ad-bc=0.
\]
So for two qubits, $ad-bc$ detects entanglement (up to normalization/phase).

\begin{proof}[Why the determinant test works]
	If $\ket{\Psi}=\ket{\psi}\otimes\ket{\phi}$ with $\ket{\psi}=(\alpha,\beta)^\top$, $\ket{\phi}=(\gamma,\delta)^\top$, then
	\[
	\ket{\Psi}=
	(\alpha\ket0+\beta\ket1)\otimes(\gamma\ket0+\delta\ket1)
	=\alpha\gamma\ket{00}+\alpha\delta\ket{01}+\beta\gamma\ket{10}+\beta\delta\ket{11}.
	\]
	Thus $a=\alpha\gamma$, $b=\alpha\delta$, $c=\beta\gamma$, $d=\beta\delta$, hence $ad-bc=0$.
	Conversely, if $ad-bc=0$, then the $2\times2$ matrix $M_\Psi$ has rank $1$,
	so it can be written as an outer product $M_\Psi=\mathbf{u}\mathbf{v}^\top$,
	which corresponds exactly to $\ket{\psi}\otimes\ket{\phi}$.
\end{proof}

\begin{ex}[Bell state (maximally entangled)]
	\[
	\ket{\Phi^+}=\frac{1}{\sqrt2}(\ket{00}+\ket{11})
	\quad\Rightarrow\quad
	M_{\Phi^+}=\frac{1}{\sqrt2}\begin{pmatrix}1&0\\0&1\end{pmatrix},
	\]
	so $\det(M_{\Phi^+})=\frac12\neq0$, hence entangled.
\end{ex}

\subsection{Schmidt decomposition (main structural theorem)}

\begin{thm}[Schmidt decomposition]
	Let $\ket{\Psi}\in \mathcal{H}_A\otimes\mathcal{H}_B$ be a pure state with $\dim\mathcal{H}_A=m$, $\dim\mathcal{H}_B=n$.
	Then there exist orthonormal sets $\{\ket{u_k}\}\subset\mathcal{H}_A$, $\{\ket{v_k}\}\subset\mathcal{H}_B$
	and nonnegative numbers $\lambda_k\ge0$ such that
	\[
	\ket{\Psi}=\sum_{k=1}^{r}\lambda_k\,\ket{u_k}\otimes\ket{v_k},
	\qquad r\le \min(m,n),
	\qquad \sum_{k=1}^r \lambda_k^2=1.
	\]
	The integer $r$ (the number of nonzero $\lambda_k$) is the \emph{Schmidt rank}.
\end{thm}

\subsubsection*{Proof idea you can compute with (SVD)}
Choose orthonormal bases $\{\ket{i}_A\}_{i=1}^m$, $\{\ket{j}_B\}_{j=1}^n$ and write
\[
\ket{\Psi}=\sum_{i,j}\Psi_{ij}\ket{i}_A\ket{j}_B.
\]
View $\Psi=(\Psi_{ij})$ as an $m\times n$ matrix.
Take its singular value decomposition:
\[
\Psi = U\,\mathrm{diag}(\lambda_1,\dots,\lambda_r)\,V^\dagger,
\]
where $U\in U(m)$, $V\in U(n)$, and $\lambda_k>0$ are singular values.
Define
\[
\ket{u_k}=\sum_i U_{ik}\ket{i}_A,
\qquad
\ket{v_k}=\sum_j \overline{V_{jk}}\ket{j}_B.
\]
Then plugging into the expansion yields the Schmidt form
$\ket{\Psi}=\sum_k\lambda_k\ket{u_k}\ket{v_k}$.

\begin{rem}[Two-qubit specialization]
	For two qubits, Schmidt decomposition always has at most two terms:
	\[
	\ket{\Psi}=\lambda_1\ket{u_1}\ket{v_1}+\lambda_2\ket{u_2}\ket{v_2},
	\qquad \lambda_1^2+\lambda_2^2=1.
	\]
\end{rem}

\subsection{Separability via Schmidt rank}

\begin{prop}[Product $\iff$ Schmidt rank $1$]
	A bipartite pure state $\ket{\Psi}$ is a product state iff its Schmidt rank equals $1$.
\end{prop}

\begin{proof}
	If Schmidt rank is $1$, then $\ket{\Psi}=\lambda_1\ket{u_1}\ket{v_1}$,
	which is a product state (global phase absorbed into vectors; $\lambda_1=1$ by normalization).
	Conversely, if $\ket{\Psi}=\ket{\psi}\otimes\ket{\phi}$, then in any bases the coefficient matrix has rank $1$,
	so the SVD has exactly one nonzero singular value, hence Schmidt rank $1$.
\end{proof}

\begin{rem}[What Schmidt rank measures]
	Schmidt rank is the most basic invariant: it detects whether entanglement exists.
	To measure \emph{how much} entanglement, you use the Schmidt coefficients $\lambda_k$
	(entanglement entropy, concurrence for two qubits, etc.).
\end{rem}

\subsection{Worked Schmidt decompositions}

\subsubsection*{Example 1: Bell state $\ket{\Phi^+}$}
Start from
\[
\ket{\Phi^+}=\frac{1}{\sqrt2}(\ket{00}+\ket{11}).
\]
This is already in Schmidt form with
\[
\lambda_1=\lambda_2=\frac{1}{\sqrt2},\quad
\ket{u_1}=\ket0,\ \ket{v_1}=\ket0,\quad
\ket{u_2}=\ket1,\ \ket{v_2}=\ket1.
\]
So Schmidt rank $2$ (maximal for two qubits).

\subsubsection*{Example 2: a partially entangled state}
Let
\[
\ket{\Psi}=\sqrt{\tfrac34}\ket{00}+\sqrt{\tfrac14}\ket{11}.
\]
This is also Schmidt form with $\lambda_1=\sqrt{3}/2$, $\lambda_2=1/2$.
It is entangled (rank $2$) but not maximally entangled (unequal coefficients).

\subsubsection*{Example 3: nontrivial basis entanglement}
Consider
\[
\ket{\Psi}=\frac12(\ket{00}+\ket{01}+\ket{10}-\ket{11}).
\]
Coefficient matrix in computational basis:
\[
M_\Psi=\frac12
\begin{pmatrix}
	1 & 1\\
	1 & -1
\end{pmatrix}.
\]
Compute $M_\Psi M_\Psi^\dagger$:
\[
M_\Psi M_\Psi^\dagger
=\frac14
\begin{pmatrix}
	1+1 & 1-1\\
	1-1 & 1+1
\end{pmatrix}
=
\frac12
\begin{pmatrix}
	1 & 0\\
	0 & 1
\end{pmatrix}
=\frac12 I.
\]
Thus singular values are both $1/\sqrt2$.
So Schmidt coefficients are $\lambda_1=\lambda_2=1/\sqrt2$ (maximally entangled).
Moreover, the left singular vectors can be chosen as $\ket{u_1}=\ket0$, $\ket{u_2}=\ket1$ (any orthonormal basis works since $M M^\dagger$ is proportional to identity),
and similarly on $B$.
One convenient Schmidt form is
\[
\ket{\Psi}
=
\frac{1}{\sqrt2}\ket{0}\otimes\ket{+}
+\frac{1}{\sqrt2}\ket{1}\otimes\ket{-},
\qquad
\ket{\pm}=\frac{\ket0\pm\ket1}{\sqrt2},
\]
which you can verify by expanding the RHS.

\begin{figure}[t]
	\centering
	\begin{tikzpicture}[>=Latex, rounded corners]
		\node[draw, thick, align=center, minimum width=35mm, minimum height=10mm] (state) {two-qubit\\$\ket{\Psi}$};
		\node[draw, thick, right=12mm of state, align=center, minimum width=42mm, minimum height=10mm] (matrix) {reshape\\into $2\times2$ matrix $M_\Psi$};
		\node[draw, thick, right=12mm of matrix, align=center, minimum width=38mm, minimum height=10mm] (svd) {SVD:\\$M=U\Sigma V^\dagger$};
		\node[draw, thick, below=10mm of svd, align=center, minimum width=48mm, minimum height=10mm] (schmidt) {Schmidt:\\$\ket{\Psi}=\sum_k\lambda_k\ket{u_k}\ket{v_k}$};
		
		\draw[->, thick] (state) -- (matrix);
		\draw[->, thick] (matrix) -- (svd);
		\draw[->, thick] (svd) -- (schmidt);
		
		\node[align=center] at ($(schmidt)+(0,-9mm)$) {\small Schmidt coefficients $\lambda_k$ quantify entanglement};
	\end{tikzpicture}
	\caption{Practical workflow: coefficient table $\to$ matrix reshape $\to$ SVD $\to$ Schmidt decomposition. For two qubits, this is always a $2\times2$ SVD.}
	\label{fig:schmidt-workflow}
\end{figure}
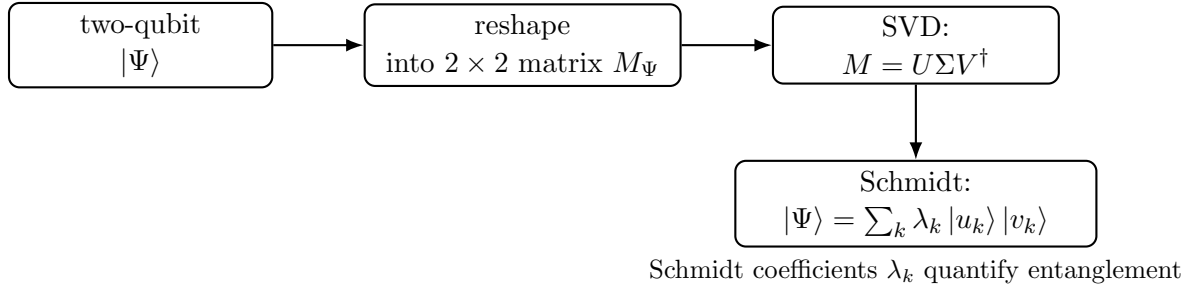

\subsection{Reduced density matrices and ``mixedness = entanglement''}

\subsubsection*{1) Reduced state via partial trace}
Given a bipartite density matrix $\rho_{AB}$, the reduced state of $A$ is
\[
\rho_A=\Tr_B(\rho_{AB}),
\qquad
\rho_B=\Tr_A(\rho_{AB}).
\]
For a pure state $\ket{\Psi}$, we take $\rho_{AB}=\ket{\Psi}\bra{\Psi}$.

\subsubsection*{2) Compute $\rho_A$ from Schmidt data}
If $\ket{\Psi}=\sum_{k=1}^r \lambda_k\ket{u_k}\ket{v_k}$, then
\[
\rho_{AB}=\sum_{k,\ell}\lambda_k\lambda_\ell
\bigl(\ket{u_k}\bra{u_\ell}\bigr)\otimes \bigl(\ket{v_k}\bra{v_\ell}\bigr).
\]
Taking $\Tr_B$ kills cross terms because $\{\ket{v_k}\}$ is orthonormal:
\[
\Tr_B\!\Bigl(\ket{v_k}\bra{v_\ell}\Bigr)=\langle v_\ell|v_k\rangle=\delta_{k\ell}.
\]
Hence
\[
\rho_A=\Tr_B(\rho_{AB})=\sum_{k=1}^r \lambda_k^2\,\ket{u_k}\bra{u_k}.
\]
Similarly
\[
\rho_B=\sum_{k=1}^r \lambda_k^2\,\ket{v_k}\bra{v_k}.
\]

\begin{prop}[Pure marginal $\iff$ product state]
	For a bipartite pure state $\ket{\Psi}$:
	\[
	\rho_A\ \text{is pure}
	\quad\Longleftrightarrow\quad
	\ket{\Psi}\ \text{is a product state}
	\quad\Longleftrightarrow\quad
	\rho_B\ \text{is pure}.
	\]
\end{prop}

\begin{proof}
	$\rho_A$ is pure iff it has rank $1$.
	From $\rho_A=\sum_k\lambda_k^2\ket{u_k}\bra{u_k}$, its rank equals the number of nonzero $\lambda_k$, i.e.\ Schmidt rank.
	Thus $\rho_A$ pure $\iff$ Schmidt rank $1$ $\iff$ product state.
\end{proof}

\subsubsection*{3) Explicit computation for $\ket{\Phi^+}$}
Let
\[
\ket{\Phi^+}=\frac{1}{\sqrt2}(\ket{00}+\ket{11}),
\qquad
\rho_{AB}=\ket{\Phi^+}\bra{\Phi^+}.
\]
Expand:
\[
\rho_{AB}=\frac12\Bigl(
\ket{00}\bra{00}+\ket{00}\bra{11}+\ket{11}\bra{00}+\ket{11}\bra{11}
\Bigr).
\]
Now partial trace over $B$ using $\Tr_B(\ket{i}\bra{j}\otimes\ket{k}\bra{\ell})=\ket{i}\bra{j}\,\Tr(\ket{k}\bra{\ell})=\ket{i}\bra{j}\,\delta_{k\ell}$:
\[
\Tr_B(\ket{00}\bra{00})=\ket0\bra0,
\quad
\Tr_B(\ket{00}\bra{11})=\ket0\bra1\,\Tr(\ket0\bra1)=0,
\]
\[
\Tr_B(\ket{11}\bra{00})=\ket1\bra0\,\Tr(\ket1\bra0)=0,
\quad
\Tr_B(\ket{11}\bra{11})=\ket1\bra1.
\]
Thus
\[
\rho_A=\Tr_B(\rho_{AB})
=\frac12(\ket0\bra0+\ket1\bra1)
=\frac12 I.
\]
So the marginal is maximally mixed. In Bloch-vector language, $\vec r=\vec 0$.

\begin{figure}[t]
	\centering
	\begin{tikzpicture}[scale=2.0, line cap=round, line join=round]
		\draw (0,0) circle (1);
		\draw[dashed] (-1,0) arc (180:360:1 and 0.35);
		\draw (-1,0) arc (180:0:1 and 0.35);
		\draw[->] (0,0) -- (1.2,0) node[right] {$x$};
		\draw[->] (0,0) -- (0,1.2) node[above] {$z$};
		\draw[->] (0,0) -- (-0.7,-0.5) node[left] {$y$};
		\fill (0,0) circle (0.04);
		\node[align=center] at (0,-1.25) {\small entangled pure $\ket{\Phi^+}$\\\small $\Rightarrow$ marginal $\rho_A=I/2$ at center};
	\end{tikzpicture}
	\caption{``Mixedness = entanglement'' for pure bipartite states: if $\ket{\Psi}_{AB}$ is entangled, then each subsystem is mixed. For $\ket{\Phi^+}$, $\rho_A=I/2$ (center of Bloch ball).}
	\label{fig:mixedness-entanglement}
\end{figure}
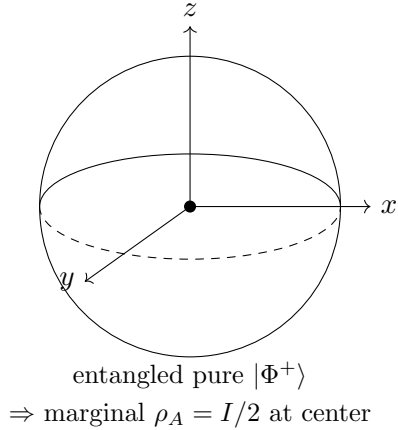

\subsection{Why entanglement is a computational resource}

\subsubsection*{1) Entanglement enables non-classical correlations}
In $\ket{\Phi^+}$, measurements are perfectly correlated in the same basis:
\[
Z\text{-measurements: outcomes match with prob. }1.
\]
But the state is not a classical mixture of $\ket{00}$ and $\ket{11}$ because it also has coherent cross terms
$\ket{00}\bra{11}$ and $\ket{11}\bra{00}$ that drive interference in other bases.
Those off-diagonal terms are exactly what classical probability models lack.

\subsubsection*{2) Entangling gates create the ``hard part'' of circuits}
Single-qubit gates are local rotations; they are cheap and parallelizable.
Entangling gates (CNOT/CZ/iSWAP) are:
\begin{itemize}
	\item the main source of depth and error,
	\item constrained by hardware connectivity,
	\item the defining resource that compilers optimize.
\end{itemize}

\subsubsection*{3) Speedups and structure (high-level view)}
Entanglement is not sufficient for speedup by itself, but in many algorithmic primitives it is essential:
\begin{itemize}
	\item amplitude amplification and multi-qubit interference patterns require coherent multi-qubit structure,
	\item phase estimation and Fourier-type routines rely on controlled operations (entangling),
	\item error correction is \emph{built out of} structured entanglement across many qubits.
\end{itemize}


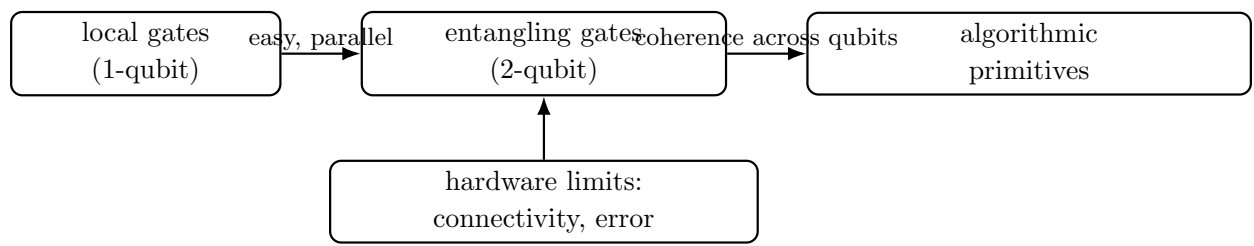
\begin{figure}[t]
	\centering
	\resizebox{\textwidth}{!}{%
		\begin{tikzpicture}[
			>=Latex,
			rounded corners,
			font=\small,
			box/.style={draw, thick, align=center, minimum height=10mm},
			arr/.style={-Latex, thick},
			node distance=10mm and 14mm
			]
			\node[box, minimum width=34mm] (local) {local gates\\(1-qubit)};
			\node[box, right=10mm of local, minimum width=46mm] (ent) {entangling gates\\(2-qubit)};
			\node[box, right=10mm of ent, minimum width=56mm] (algo) {algorithmic\\primitives};
			
			\node[box, below=8mm of ent, minimum width=54mm] (hardware)
			{hardware limits:\\connectivity, error};
			
			\draw[arr] (local) -- node[midway, above, font=\footnotesize, inner sep=1pt]
			{easy, parallel} (ent);
			
			\draw[arr] (ent) -- node[midway, above, font=\footnotesize, inner sep=1pt, text width=40mm, align=center]
			{coherence across qubits} (algo);
			
			\draw[arr] (hardware) -- (ent);
			
			\node[below=6mm of hardware, align=center, font=\footnotesize, text width=90mm]
			{compilers mostly optimize entangling structure};
			
		\end{tikzpicture}%
	}
	\caption{Entanglement as a resource in practice: the algorithm needs entangling gates; hardware constrains them; compilation tries to preserve algorithmic structure under hardware limits.}
	\label{fig:entanglement-resource}
\end{figure}

\subsection{Exercises}

\begin{exercise}[Tensor product expansion]
	Let $\ket{\psi}=\alpha\ket0+\beta\ket1$ and $\ket{\phi}=\gamma\ket0+\delta\ket1$.
	Compute $\ket{\psi}\otimes\ket{\phi}$ in the computational basis and identify the coefficient matrix $M$.
	Show $\det(M)=0$.
\end{exercise}

\noindent\textbf{Solution.}
Expand:
\[
\ket{\psi}\otimes\ket{\phi}
=(\alpha\ket0+\beta\ket1)\otimes(\gamma\ket0+\delta\ket1)
=\alpha\gamma\ket{00}+\alpha\delta\ket{01}+\beta\gamma\ket{10}+\beta\delta\ket{11}.
\]
Thus $M=\begin{psmallmatrix}\alpha\gamma&\alpha\delta\\ \beta\gamma&\beta\delta\end{psmallmatrix}
=\begin{psmallmatrix}\alpha\\ \beta\end{psmallmatrix}\begin{psmallmatrix}\gamma&\delta\end{psmallmatrix}$ is rank $1$,
and
\[
\det(M)=\alpha\gamma\cdot \beta\delta-\alpha\delta\cdot \beta\gamma = 0.
\]

\medskip

\begin{exercise}[Determinant entanglement test]
	For the state $\ket{\Psi}=a\ket{00}+b\ket{01}+c\ket{10}+d\ket{11}$:
	prove that $\ket{\Psi}$ is entangled iff $ad-bc\neq0$.
\end{exercise}

\noindent\textbf{Solution.}
Let $M=\begin{psmallmatrix}a&b\\ c&d\end{psmallmatrix}$.
If $\ket{\Psi}$ is product, then $M$ is an outer product (rank $1$) hence $\det(M)=0$.
Conversely, if $\det(M)=0$, then $\mathrm{rank}(M)<2$, so $\mathrm{rank}(M)=1$ (unless $M=0$, impossible after normalization).
Thus $M=\mathbf{u}\mathbf{v}^\top$ for some $\mathbf{u},\mathbf{v}\in\C^2$,
and $\ket{\Psi}=(u_0\ket0+u_1\ket1)\otimes(v_0\ket0+v_1\ket1)$ is product.
So entangled iff $\det(M)\neq0$, i.e.\ $ad-bc\neq0$.

\medskip

\begin{exercise}[Schmidt decomposition via SVD: a worked case]
	Find the Schmidt decomposition of
	\[
	\ket{\Psi}=\frac{1}{\sqrt3}\ket{00}+\sqrt{\frac{2}{3}}\ket{11}.
	\]
	Compute $\rho_A$ and its Bloch vector.
\end{exercise}

\noindent\textbf{Solution.}
The state is already in Schmidt form with
\[
\lambda_1=\frac{1}{\sqrt3},\quad \lambda_2=\sqrt{\frac{2}{3}},
\quad \ket{u_1}=\ket0,\ \ket{v_1}=\ket0,\quad \ket{u_2}=\ket1,\ \ket{v_2}=\ket1.
\]
Thus
\[
\rho_A=\lambda_1^2\ket0\bra0+\lambda_2^2\ket1\bra1
=\frac13\ket0\bra0+\frac23\ket1\bra1
=
\begin{pmatrix}
	1/3 & 0\\
	0 & 2/3
\end{pmatrix}.
\]
Bloch vector: $x=\Tr(\rho_A X)=0$, $y=\Tr(\rho_A Y)=0$, and
\[
z=\Tr(\rho_A Z)=\rho_{00}-\rho_{11}=\frac13-\frac23=-\frac13.
\]
So $\vec r=(0,0,-1/3)$ (mixed, inside the Bloch ball).

\medskip

\begin{exercise}[Bell state marginal is maximally mixed]
	Let $\ket{\Phi^+}=\frac{1}{\sqrt2}(\ket{00}+\ket{11})$.
	Compute $\rho_A=\Tr_B(\ket{\Phi^+}\bra{\Phi^+})$ explicitly and verify $\rho_A=I/2$.
\end{exercise}

\noindent\textbf{Solution.}
As computed in the text:
\[
\rho_{AB}=\frac12(\ket{00}\bra{00}+\ket{00}\bra{11}+\ket{11}\bra{00}+\ket{11}\bra{11}).
\]
Take $\Tr_B$ term-by-term:
\[
\Tr_B(\ket{00}\bra{00})=\ket0\bra0,\quad
\Tr_B(\ket{00}\bra{11})=0,\quad
\Tr_B(\ket{11}\bra{00})=0,\quad
\Tr_B(\ket{11}\bra{11})=\ket1\bra1.
\]
Thus
\[
\rho_A=\frac12(\ket0\bra0+\ket1\bra1)=\frac12 I.
\]

\medskip

\begin{exercise}[Circuit that creates entanglement]
	Consider the circuit: apply $H$ on qubit $A$ then apply $\mathrm{CNOT}$ with control $A$ and target $B$.
	Starting from $\ket{00}$, compute the output state and show it is entangled.
	Then compute $ad-bc$ for the output coefficients.
\end{exercise}

\noindent\textbf{Solution.}
First apply $H$ to the first qubit:
\[
(H\otimes I)\ket{00}=(H\ket0)\otimes\ket0=\ket{+}\ket0
=\frac{1}{\sqrt2}(\ket0+\ket1)\ket0
=\frac{1}{\sqrt2}(\ket{00}+\ket{10}).
\]
Now apply $\mathrm{CNOT}$ (control first, target second): $\ket{00}\mapsto\ket{00}$, $\ket{10}\mapsto\ket{11}$, so
\[
\mathrm{CNOT}\,\frac{1}{\sqrt2}(\ket{00}+\ket{10})
=\frac{1}{\sqrt2}(\ket{00}+\ket{11})
=\ket{\Phi^+}.
\]
This is entangled.
Coefficients: $a=\frac1{\sqrt2}$, $b=0$, $c=0$, $d=\frac1{\sqrt2}$, so
\[
ad-bc=\frac12-0=\frac12\neq0,
\]
confirming entanglement by the determinant test.

\section{Algorithm III: Teleportation and Superdense Coding}
\label{sec:teleport}

\subsection{Objective}
Teleportation and superdense coding are the two ``infrastructure-grade'' demonstrations that
\emph{entanglement is a resource that changes what communication is possible}.

\begin{itemize}
	\item \textbf{Teleportation:} move an \emph{unknown} qubit state using
	\[
	\text{1 Bell pair} \ +\ \text{2 classical bits} \ +\ \text{local Clifford + Pauli correction}.
	\]
	\item \textbf{Superdense coding:} send \emph{two} classical bits using
	\[
	\text{1 Bell pair} \ +\ \text{1 transmitted qubit} \ +\ \text{Bell measurement}.
	\]
\end{itemize}

Both protocols share the same control pattern:
\[
\text{resource entanglement}
\ \to\
\text{short Clifford circuit}
\ \to\
\text{measurement (bits)}
\ \to\
\text{feed-forward correction/decoding}.
\]
This ``measure $\Rightarrow$ branch $\Rightarrow$ correct'' pattern is exactly what reappears in
fault-tolerant stacks (Pauli frame updates) and hardware control loops (low-latency conditional logic).

\subsection{Prerequisites and notation}

\subsubsection{Registers and ownership}
Teleportation uses three qubits:
\[
A \ (\text{Alice: unknown input}),\qquad
A' \ (\text{Alice: half of Bell pair}),\qquad
B \ (\text{Bob: half of Bell pair}).
\]
The unknown input is
\[
\ket{\psi}_A=\alpha\ket0+\beta\ket1,\qquad |\alpha|^2+|\beta|^2=1.
\]
The shared Bell pair is
\[
\ket{\Phi^+}_{A'B}=\frac{1}{\sqrt2}(\ket{00}+\ket{11}).
\]
The initial joint state is
\[
\ket{\psi}_A\otimes\ket{\Phi^+}_{A'B}.
\]

\subsubsection{Gates}
\[
H=\frac{1}{\sqrt2}\begin{pmatrix}1&1\\1&-1\end{pmatrix},\qquad
X=\begin{pmatrix}0&1\\1&0\end{pmatrix},\qquad
Z=\begin{pmatrix}1&0\\0&-1\end{pmatrix}.
\]
CNOT is control $\to$ target:
\[
\mathrm{CNOT}\ket{a}\ket{b}=\ket{a}\ket{b\oplus a}.
\]

\subsubsection{Bell basis}
\[
\ket{\Phi^\pm}=\frac{1}{\sqrt2}(\ket{00}\pm\ket{11}),
\qquad
\ket{\Psi^\pm}=\frac{1}{\sqrt2}(\ket{01}\pm\ket{10}).
\]

\subsubsection{Measurement outcomes and corrections}
Alice measures two qubits in the computational basis, producing two classical bits
\[
(m_1,m_2)\in\{0,1\}^2.
\]
In the convention used below, Bob applies the correction
\[
\boxed{\text{Bob applies } Z^{m_1}X^{m_2}.}
\]
If you prefer $X^{m_2}Z^{m_1}$, that is equivalent up to a global phase when both exponents are $1$.

\subsection{Quantum teleportation}

\subsubsection{Circuit (resource + measurement + feed-forward)}
\begin{figure}[t]
	\centering
	\begin{tikzpicture}[x=1.0cm,y=0.92cm,>=Latex]
		\draw (0,2) -- (13.4,2) node[right] {$A$ (Alice: unknown)};
		\draw (0,1) -- (13.4,1) node[right] {$A'$ (Alice)};
		\draw (0,0) -- (13.4,0) node[right] {$B$ (Bob)};
		
		\node[left] at (0,2) {$\ket{\psi}$};
		\node[left] at (0,1) {$\ket0$};
		\node[left] at (0,0) {$\ket0$};
		
		\node[draw, thick, minimum width=6mm, minimum height=5mm] at (1.4,1) {$H$};
		\fill[black] (2.7,1) circle (0.06);
		\draw[thick] (2.7,1) -- (2.7,0);
		\draw[thick] (2.7,0) circle (0.16);
		\draw[thick] (2.7,0.16) -- (2.7,-0.16);
		\node at (2.05,0.55) {\scriptsize create $\ket{\Phi^+}$};
		
		\fill[black] (4.6,2) circle (0.06);
		\draw[thick] (4.6,2) -- (4.6,1);
		\draw[thick] (4.6,1) circle (0.16);
		\draw[thick] (4.6,1.16) -- (4.6,0.84);
		\node at (4.6,2.35) {\scriptsize CNOT};
		
		\node[draw, thick, minimum width=6mm, minimum height=5mm] at (6.0,2) {$H$};
		
		\node[draw, thick, minimum width=8mm, minimum height=6mm] (M1) at (7.7,2) {$M$};
		\node[draw, thick, minimum width=8mm, minimum height=6mm] (M2) at (7.7,1) {$M$};
		\node at (7.7,2.45) {\scriptsize $m_1$};
		\node at (7.7,1.45) {\scriptsize $m_2$};
		
		\draw[dashed, thick] (8.15,2) -- (10.6,0.55);
		\draw[dashed, thick] (8.15,1) -- (10.6,0.25);
		\node at (9.5,0.95) {\scriptsize 2 classical bits};
		
		\node[draw, thick, minimum width=10mm, minimum height=6mm] at (11.3,0.55) {$Z^{m_1}$};
		\node[draw, thick, minimum width=10mm, minimum height=6mm] at (11.3,0.25) {$X^{m_2}$};
		
		\node at (9.5,-0.65) {\scriptsize Bob recovers $\ket{\psi}$ on $B$};
	\end{tikzpicture}
	\caption{Teleportation circuit. A Bell pair plus two classical bits suffice to move an unknown qubit state.}
	\label{fig:teleport}
\end{figure}
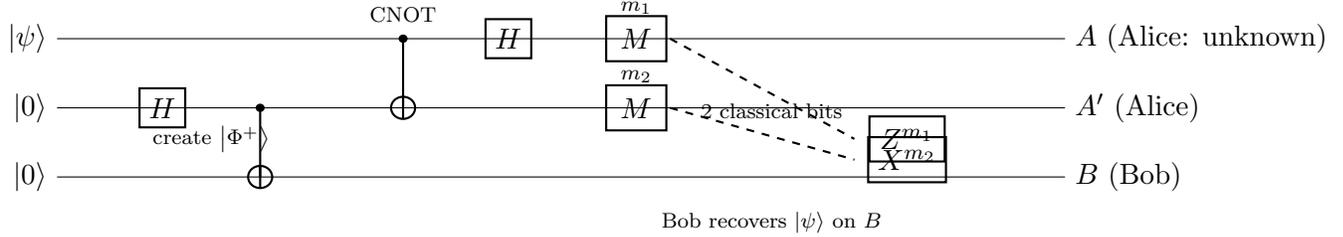

\subsubsection{The teleportation identity (expanded algebra)}
Start with
\[
\ket{\psi}_A\otimes\ket{\Phi^+}_{A'B}
=
(\alpha\ket0+\beta\ket1)_A\otimes\frac1{\sqrt2}(\ket{00}+\ket{11})_{A'B}.
\]
Expand:
\[
=\frac1{\sqrt2}\Bigl(
\alpha\ket0\ket0\ket0+\alpha\ket0\ket1\ket1+\beta\ket1\ket0\ket0+\beta\ket1\ket1\ket1
\Bigr)_{AA'B}.
\]

\paragraph{Step 1: apply $\mathrm{CNOT}_{A\to A'}$.}
Only terms with $A=\ket1$ flip $A'$:
\[
\ket1\ket0\mapsto\ket1\ket1,\qquad \ket1\ket1\mapsto\ket1\ket0.
\]
So the state becomes
\[
\frac1{\sqrt2}\Bigl(
\alpha\ket0\ket0\ket0+\alpha\ket0\ket1\ket1+\beta\ket1\ket1\ket0+\beta\ket1\ket0\ket1
\Bigr).
\]

\paragraph{Step 2: apply $H$ on $A$.}
Use $H\ket0=\frac1{\sqrt2}(\ket0+\ket1)$ and $H\ket1=\frac1{\sqrt2}(\ket0-\ket1)$.
After distributing and collecting, we get an overall factor $1/2$:
\[
\ket{\Psi_{\mathrm{pre}}}
=
\frac12\Bigl[
\ket{00}(\alpha\ket0+\beta\ket1)
+\ket{01}(\alpha\ket1+\beta\ket0)
+\ket{10}(\alpha\ket0-\beta\ket1)
+\ket{11}(\alpha\ket1-\beta\ket0)
\Bigr]_{(AA')B}.
\]
Recognize each Bob branch as a Pauli-applied version of $\ket{\psi}$:
\[
(\alpha\ket0+\beta\ket1)=I\ket{\psi},\qquad
(\alpha\ket1+\beta\ket0)=X\ket{\psi},
\]
\[
(\alpha\ket0-\beta\ket1)=Z\ket{\psi},\qquad
(\alpha\ket1-\beta\ket0)=XZ\ket{\psi}.
\]
Therefore
\[
\boxed{
	\ket{\Psi_{\mathrm{pre}}}
	=
	\frac12\sum_{m_1,m_2\in\{0,1\}}
	\ket{m_1m_2}_{AA'}\ \Bigl(X^{m_2}Z^{m_1}\ket{\psi}\Bigr)_B.
}
\]

\subsubsection{Measurement and recovery}
Alice measures $(A,A')$.
If she observes outcome $(m_1,m_2)$, the state collapses to
\[
\ket{m_1m_2}_{AA'}\otimes X^{m_2}Z^{m_1}\ket{\psi}_B.
\]
After receiving $(m_1,m_2)$, Bob applies $Z^{m_1}X^{m_2}$:
\[
(Z^{m_1}X^{m_2})(X^{m_2}Z^{m_1})\ket{\psi}=\ket{\psi}
\]
(the only possible discrepancy is a global phase when $m_1=m_2=1$, which has no physical effect).

\subsubsection{Pauli frame note}
Instead of physically applying $X$ and $Z$ immediately, many architectures track the correction in software
as a \emph{Pauli frame} and reinterpret subsequent measurements/gates accordingly.
Teleportation is therefore naturally compatible with low-latency classical control.

\subsection{Superdense coding}

\subsubsection{Goal and resource trade}
Superdense coding achieves:
\[
\boxed{\text{send 2 classical bits by transmitting 1 qubit, if a Bell pair is shared.}}
\]
It works because a shared Bell pair gives access to \emph{four orthogonal global states} using only local operations.

\subsubsection{Encoding}
Let Alice and Bob share $\ket{\Phi^+}_{AB}$.
To send $(b_1,b_2)\in\{0,1\}^2$, Alice applies
\[
U_{b_1b_2}=Z^{b_1}X^{b_2}
\]
to her qubit $A$, then sends qubit $A$ to Bob.

\subsubsection{Paulis permute the Bell basis}
Compute:
\[
(I\otimes I)\ket{\Phi^+}=\ket{\Phi^+},
\]
\[
(Z\otimes I)\ket{\Phi^+}
=\frac1{\sqrt2}(\ket{00}-\ket{11})=\ket{\Phi^-},
\]
\[
(X\otimes I)\ket{\Phi^+}
=\frac1{\sqrt2}(\ket{10}+\ket{01})=\ket{\Psi^+},
\]
\[
(XZ\otimes I)\ket{\Phi^+}
=\frac1{\sqrt2}(\ket{10}-\ket{01})=\ket{\Psi^-}
\quad(\text{global phase ignored}).
\]
Thus the two classical bits are mapped to one of four orthogonal Bell states.

\subsubsection{Decoding (Bell measurement)}
Once Bob has both qubits, he performs a Bell measurement via a fixed circuit:
apply $\mathrm{CNOT}_{A\to B}$, then $H$ on $A$, then measure $(A,B)$ in the computational basis.

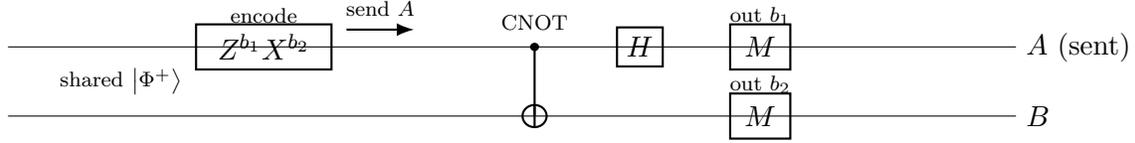
\begin{figure}[t]
	\centering
	\begin{tikzpicture}[x=1.0cm,y=0.92cm,>=Latex]
		\draw (0,1) -- (13.4,1) node[right] {$A$ (sent)};
		\draw (0,0) -- (13.4,0) node[right] {$B$};
		
		\node at (1.5,0.5) {\scriptsize shared $\ket{\Phi^+}$};
		
		\node[draw, thick, minimum width=18mm, minimum height=6mm] at (3.4,1) {$Z^{b_1}X^{b_2}$};
		\node at (3.4,1.45) {\scriptsize encode};
		
		\draw[->, thick] (4.5,1.25) -- (5.4,1.25);
		\node at (4.95,1.55) {\scriptsize send $A$};
		
		\fill[black] (7.0,1) circle (0.06);
		\draw[thick] (7.0,1) -- (7.0,0);
		\draw[thick] (7.0,0) circle (0.16);
		\draw[thick] (7.0,0.16) -- (7.0,-0.16);
		\node at (7.0,1.35) {\scriptsize CNOT};
		
		\node[draw, thick, minimum width=6mm, minimum height=5mm] at (8.4,1) {$H$};
		
		\node[draw, thick, minimum width=8mm, minimum height=6mm] at (10.0,1) {$M$};
		\node[draw, thick, minimum width=8mm, minimum height=6mm] at (10.0,0) {$M$};
		\node at (10.0,1.45) {\scriptsize out $b_1$};
		\node at (10.0,0.45) {\scriptsize out $b_2$};
	\end{tikzpicture}
	\caption{Superdense coding: local Pauli encoding + Bell measurement decoding.}
	\label{fig:dense}
\end{figure}

\subsubsection{Why the decoder works (explicit mapping)}
Apply the decoder $D=(H\otimes I)\cdot\mathrm{CNOT}$:
\[
\ket{\Phi^+}\mapsto\ket{00},\quad
\ket{\Phi^-}\mapsto\ket{10},\quad
\ket{\Psi^+}\mapsto\ket{01},\quad
\ket{\Psi^-}\mapsto\ket{11}.
\]
Hence computational-basis measurement recovers $(b_1,b_2)$ deterministically.

\subsection{Conceptual interpretation: what entanglement is doing}

\subsubsection{Teleportation: entanglement + two bits selects the correct Pauli frame}
Teleportation produces a classical label $(m_1,m_2)$ that specifies \emph{which Pauli-twisted copy} of $\ket{\psi}$
Bob currently holds. The classical bits do not carry $\alpha,\beta$; they carry the \emph{frame index}.

\subsubsection{Dense coding: entanglement turns local moves into globally distinguishable states}
With a shared Bell pair, Alice's four local Pauli choices produce four orthogonal Bell states.
The channel transmits one qubit, but the \emph{space of perfectly distinguishable global states} is size $4$.

\subsubsection{Hardware/control takeaway}
Both protocols are small instances of the fault-tolerant control pattern:
\[
\text{Clifford circuit} \to \text{syndrome-like bits} \to \text{conditional Pauli update}.
\]
This is why teleportation-based primitives and Pauli-frame tracking are a practical ``control-hardware story.''

\subsection{Exercises}

\begin{exercise}[Teleportation identity (full derivation)]
	Let $\ket{\psi}=\alpha\ket0+\beta\ket1$.
	Derive explicitly that after $\mathrm{CNOT}_{A\to A'}$ and $H$ on $A$,
	\[
	\ket{\Psi_{\mathrm{pre}}}
	=
	\frac12\sum_{m_1,m_2\in\{0,1\}}
	\ket{m_1m_2}_{AA'}\ \Bigl(X^{m_2}Z^{m_1}\ket{\psi}\Bigr)_B.
	\]
\end{exercise}

\noindent\textbf{Solution.}
Use the expansion-and-collect calculation in the teleportation subsection:
expand $\ket{\psi}\otimes\ket{\Phi^+}$ into four computational basis terms,
apply CNOT (flip $A'$ iff $A=1$), apply $H$ on $A$,
then group by Alice outcomes $\ket{00},\ket{01},\ket{10},\ket{11}$ and identify the corresponding Paulis.

\medskip

\begin{exercise}[Teleport a concrete state and list all branches]
	Take $\ket{\psi}=\ket{+}=\frac1{\sqrt2}(\ket0+\ket1)$.
	List Bob's state for each measurement outcome $(m_1,m_2)$ before correction,
	and verify correction returns $\ket{+}$.
\end{exercise}

\noindent\textbf{Solution.}
Bob's pre-correction state is $X^{m_2}Z^{m_1}\ket{+}$.
Compute $X\ket{+}=\ket{+}$ and $Z\ket{+}=\ket{-}$.
Thus:
\[
(0,0)\Rightarrow \ket{+},\quad
(0,1)\Rightarrow \ket{+},\quad
(1,0)\Rightarrow \ket{-},\quad
(1,1)\Rightarrow \ket{-}.
\]
Applying $Z^{m_1}X^{m_2}$ cancels the same Pauli (up to global phase), giving $\ket{+}$.

\medskip

\begin{exercise}[Dense coding table]
	Fill the table mapping bits $(b_1,b_2)$ to the Bell state $U_{b_1b_2}\ket{\Phi^+}$ for
	$U_{b_1b_2}=Z^{b_1}X^{b_2}$.
\end{exercise}

\noindent\textbf{Solution.}
Direct computation gives:
\[
(0,0)\mapsto\ket{\Phi^+},\quad
(1,0)\mapsto\ket{\Phi^-},\quad
(0,1)\mapsto\ket{\Psi^+},\quad
(1,1)\mapsto\ket{\Psi^-}
\]
(up to global phase for the last case depending on your $XZ$ convention).

\medskip

\begin{exercise}[Bell measurement decoder check]
	Show that the decode circuit $D=(H\otimes I)\cdot \mathrm{CNOT}$ satisfies
	\[
	D\ket{\Phi^+}=\ket{00},\ 
	D\ket{\Phi^-}=\ket{10},\ 
	D\ket{\Psi^+}=\ket{01},\ 
	D\ket{\Psi^-}=\ket{11}.
	\]
\end{exercise}

\noindent\textbf{Solution.}
Compute each case by two short steps:
apply CNOT to the Bell state (rewriting it as $\ket{\pm}\ket{0/1}$),
then apply $H$ on the first qubit ($H\ket{\pm}=\ket{0/1}$).
For example,
\[
\ket{\Phi^+}=\tfrac1{\sqrt2}(\ket{00}+\ket{11})
\xrightarrow{\mathrm{CNOT}}
\tfrac1{\sqrt2}(\ket{00}+\ket{10})=\ket{+}\ket0
\xrightarrow{H\text{ on 1st}}
\ket0\ket0.
\]

\medskip

\begin{exercise}[No-signaling: Bob sees maximally mixed without the bits]
	Assume Alice does not send $(m_1,m_2)$.
	Show Bob's reduced density matrix is $\rho_B=I/2$ for any $\ket{\psi}$.
\end{exercise}

\noindent\textbf{Solution.}
From the teleportation identity, Bob holds one of the four states
$X^{m_2}Z^{m_1}\ket{\psi}$ with equal probability $1/4$.
So
\[
\rho_B=\frac14\sum_{m_1,m_2} X^{m_2}Z^{m_1}\,\rho\,Z^{m_1}X^{m_2},
\qquad \rho=\ket{\psi}\bra{\psi}.
\]
Write $\rho=\frac12(I+\vec r\cdot\vec\sigma)$ (Bloch form).
Conjugation by the four Paulis flips the components of $\vec r$ in all sign patterns,
so the average cancels $\vec r$ and leaves $\frac12 I$:
\[
\rho_B=\frac12 I.
\]
This confirms Alice cannot signal to Bob without the classical message.
	
\section{Entanglement via Algebraic Geometry: Segre Varieties}
\label{sec:segre}

\subsection*{Objective}
This section gives a structural, algebraic--geometric characterization of
\emph{bipartite pure-state separability}. The key dictionary is one line:
\begin{quote}
	\emph{Product states are exactly rank-one tensors; projectively, they form the Segre variety.}
\end{quote}
Equivalently, entanglement is the geometric statement ``the state does not lie on the Segre variety.''
Schmidt rank becomes \emph{matrix rank}, and rank constraints become \emph{minor vanishing conditions}.

\subsection*{Key idea: product states are rank-one tensors}

\subsubsection*{(1) Projective state space}
Let the local Hilbert spaces be
\[
\mathcal H_A \cong \C^{m},\qquad \mathcal H_B \cong \C^{n}.
\]
Then the composite system is
\[
\mathcal H_{AB}=\mathcal H_A\otimes \mathcal H_B \cong \C^{mn}.
\]
A pure state is a nonzero vector $\psi\in\mathcal H_{AB}\setminus\{0\}$, but physically
global scaling (in particular global phase) is irrelevant, so the true space of pure states is
\[
[\psi]\in \PP(\mathcal H_{AB})\cong \CP^{mn-1}.
\]

\subsubsection*{(2) Coefficient tensor and reshaping}
Choose orthonormal bases $\{\ket{i}_A\}_{i=1}^m$ and $\{\ket{j}_B\}_{j=1}^n$. Write
\[
\ket{\psi}=\sum_{i=1}^{m}\sum_{j=1}^{n}\psi_{ij}\,\ket{i}_A\otimes\ket{j}_B.
\]
The coefficients $\psi_{ij}$ form an $m\times n$ array. Define the reshaped matrix
\[
M_\psi := (\psi_{ij})\in \Mat_{m\times n}(\C).
\]

\subsubsection*{(3) Product state $\Leftrightarrow$ rank-one}
A product state has the form $\ket{\psi}=\ket a\otimes \ket b$ with
\[
\ket a=\sum_{i=1}^m a_i\ket{i}_A,\qquad
\ket b=\sum_{j=1}^n b_j\ket{j}_B.
\]
Then
\[
\psi_{ij}=a_i b_j,\qquad\text{so}\qquad M_\psi = a\, b^{T},
\]
an outer product, hence $\rank(M_\psi)=1$. Conversely, if $\rank(M_\psi)=1$ then
$M_\psi=a b^T$ for some vectors $a,b$, which implies $\psi=a\otimes b$.

\begin{prop}[Product state $\iff$ rank-one tensor]
	For $[\psi]\in \PP(\C^m\otimes \C^n)$ the following are equivalent:
	\[
	\psi \text{ is a product state}
	\ \Longleftrightarrow\
	\rank(M_\psi)=1.
	\]
\end{prop}

\begin{proof}
	If $\psi=a\otimes b$, then $M_\psi=a b^T$ so $\rank(M_\psi)=1$.
	Conversely, if $\rank(M_\psi)=1$, then $M_\psi=a b^T$ for some $a,b$, and
	\[
	\psi=\sum_{i,j} a_i b_j\,\ket{i}\otimes\ket{j}=a\otimes b,
	\]
	so $\psi$ is a product state.
\end{proof}

\subsection*{Segre embedding and the Segre variety}

\subsubsection*{(1) Segre embedding}
Define the Segre map
\[
\sigma:\ \CP^{m-1}\times\CP^{n-1}\ \longrightarrow\ \CP^{mn-1},
\qquad
([a],[b])\ \longmapsto\ [a\otimes b].
\]
In homogeneous coordinates:
\[
([a_1:\cdots:a_m],[b_1:\cdots:b_n])
\ \longmapsto\
[\;a_i b_j\;]_{1\le i\le m,\ 1\le j\le n}.
\]
Its image
\[
\Sigma_{m,n}:=\sigma(\CP^{m-1}\times\CP^{n-1})\subset \CP^{mn-1}
\]
is the \emph{Segre variety}.

\begin{prop}[Segre variety = projectivized product states]
	$\Sigma_{m,n}$ is exactly the set of projective classes of bipartite pure product states in
	$\PP(\C^m\otimes\C^n)$.
\end{prop}

\begin{proof}
	By definition, every point of $\Sigma_{m,n}$ is of the form $[a\otimes b]$, a product state.
	Conversely, any product state $[a\otimes b]$ is $\sigma([a],[b])$, hence lies in the image.
\end{proof}

\subsubsection*{(2) Defining equations: $2\times2$ minors}
A matrix has rank one iff all its $2\times2$ minors vanish. Concretely,
for all $1\le i<k\le m$ and $1\le j<\ell\le n$,
\[
\psi_{ij}\psi_{k\ell}-\psi_{i\ell}\psi_{kj}=0.
\]
These are homogeneous quadrics in the projective coordinates, so $\Sigma_{m,n}$ is a projective
algebraic variety.

\begin{prop}[Ideal of $\Sigma_{m,n}$ (bipartite case)]
	$\Sigma_{m,n}\subset \CP^{mn-1}$ is cut out by the homogeneous ideal generated by all
	$2\times2$ minors of the reshaped matrix $M_\psi$.
\end{prop}

\begin{proof}
	If $\psi=a\otimes b$, then $M_\psi=a b^T$ and every $2\times2$ minor vanishes.
	Conversely, if all $2\times2$ minors vanish then $\rank(M_\psi)\le 1$.
	In projective space $\psi\neq0$, hence $\rank(M_\psi)=1$, so $\psi$ is a product state.
\end{proof}

\subsubsection*{(3) Dimension check}
\[
\dim(\Sigma_{m,n})=\dim(\CP^{m-1}\times \CP^{n-1})=(m-1)+(n-1)=m+n-2,
\]
while $\dim(\CP^{mn-1})=mn-1$, hence
\[
\codim(\Sigma_{m,n}\subset\CP^{mn-1})
=(mn-1)-(m+n-2)=(m-1)(n-1).
\]

\subsubsection*{(4) Two qubits: a single quadric}
For $m=n=2$,
\[
\ket{\psi}=a\ket{00}+b\ket{01}+c\ket{10}+d\ket{11},
\qquad [a:b:c:d]\in\CP^3,
\qquad
M_\psi=\begin{pmatrix}a&b\\c&d\end{pmatrix}.
\]
The product-state condition is
\[
\det(M_\psi)=ad-bc=0.
\]
Thus $\Sigma_{2,2}\subset\CP^3$ is the quadric hypersurface $\{ad-bc=0\}$.

\begin{figure}[t]
	\centering
	\begin{tikzpicture}[x=1cm,y=1cm,>=Latex]
		\node[draw, rounded corners, align=left, inner sep=6pt] (All) at (0,0)
		{All bipartite pure states\\ $\mathbb{P}(\C^{mn})=\CP^{mn-1}$};
		\node[draw, rounded corners, align=left, inner sep=6pt] (Segre) at (7.4,0)
		{Product states\\ Segre variety $\Sigma_{m,n}$\\ (rank-one tensors)};
		\draw[->, thick] (All) -- node[above, align=center]
		{\small reshape $\psi\mapsto M_\psi$\\[-0.2em]\small check minors} (Segre);
		
		\node[align=left] at (0,-2.0)
		{\small Coefficients $\psi_{ij}$ form a matrix $M_\psi$};
		\node[align=left] at (7.4,-2.0)
		{\small Product $\iff \rank(M_\psi)=1$\\ \small $\iff$ all $2\times2$ minors vanish};
	\end{tikzpicture}
	\caption{Segre viewpoint: product states form an algebraic subvariety defined by rank-one constraints.}
	\label{fig:segre-dictionary}
\end{figure}
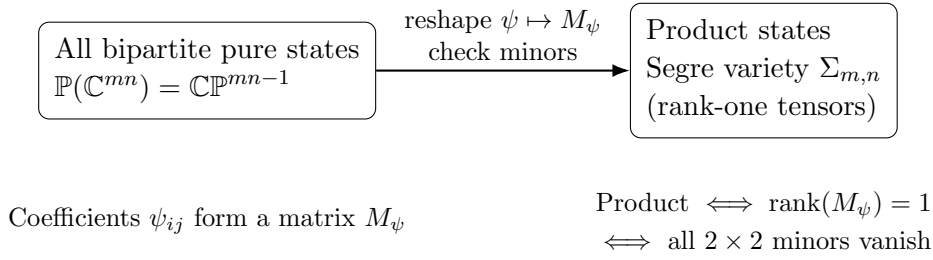

\subsection*{Schmidt rank, matrix rank, and Segre strata}

\subsubsection*{(1) Schmidt decomposition as SVD}
Every bipartite pure state admits a Schmidt decomposition
\[
\ket{\psi}=\sum_{r=1}^{R} s_r\,\ket{u_r}\otimes\ket{v_r},
\qquad s_r>0,
\]
where $R$ is the Schmidt rank.

\begin{prop}[Schmidt rank = rank of reshaped matrix]
	For $\ket{\psi}\in\C^m\otimes\C^n$,
	\[
	\mathrm{SchmidtRank}(\psi)=\rank(M_\psi).
	\]
\end{prop}

\begin{proof}
	View $\psi$ as a linear map
	\[
	\widetilde{\psi}:\C^n\to\C^m,\qquad
	\widetilde{\psi}(\ket{j})=\sum_{i=1}^m \psi_{ij}\ket{i}.
	\]
	The matrix of $\widetilde{\psi}$ in the chosen bases is exactly $M_\psi$.
	An SVD of $\widetilde{\psi}$,
	\[
	\widetilde{\psi}=U\,\mathrm{diag}(s_1,\dots,s_R)\,V^\dagger,
	\]
	produces the Schmidt decomposition with the same nonzero singular values $s_r$.
	Hence $R$ equals the number of nonzero singular values, i.e.\ $\rank(M_\psi)$.
\end{proof}

\subsubsection*{(2) Determinantal rank strata}
Define the rank strata
\[
\mathcal{R}_{\le r}:=\{[\psi]\in\CP^{mn-1}:\ \rank(M_\psi)\le r\}.
\]
Then
\[
\mathcal{R}_{\le 1}=\Sigma_{m,n}.
\]
More generally, the Zariski closure of $\mathcal{R}_{\le r}$ is the determinantal variety
\[
\overline{\mathcal{R}_{\le r}}
=\{[\psi]: \text{all $(r+1)\times(r+1)$ minors of $M_\psi$ vanish}\}.
\]
Thus the condition ``Schmidt rank $\le r$'' becomes an explicit system of polynomial equations.

\subsection*{General bipartite systems: higher Segre varieties and minors}

\subsubsection*{(1) Higher local dimension is the same story}
For arbitrary $m,n$,
\[
\Sigma_{m,n}\subset\CP^{mn-1}
\]
is always the set of rank-one tensors, and separability is still equivalent to
vanishing of all $2\times2$ minors of $M_\psi$.

\subsubsection*{(2) Multipartite remark: flattenings}
For three or more subsystems there is no single Schmidt rank, but one still obtains
useful algebraic tests by \emph{flattening} the coefficient tensor into a matrix along a chosen bipartition
and checking minors (rank constraints) for that bipartition.

\subsection*{Worked examples}

\subsubsection*{Example 1 (two qubits): Segre equation and concurrence}
For
\[
\ket{\psi}=a\ket{00}+b\ket{01}+c\ket{10}+d\ket{11},
\]
\[
\ket{\psi}\text{ is product}\iff ad-bc=0.
\]
Moreover, for two-qubit pure states the concurrence is
\[
C(\psi)=2|ad-bc|,
\]
so the Segre equation is exactly the condition $C(\psi)=0$.

\subsubsection*{Example 2 (Bell state): off the Segre variety}
\[
\ket{\Phi^+}=\frac1{\sqrt2}(\ket{00}+\ket{11})
\Rightarrow (a,b,c,d)=\Bigl(\frac1{\sqrt2},0,0,\frac1{\sqrt2}\Bigr),
\]
hence
\[
ad-bc=\frac12\neq 0,
\]
so $\ket{\Phi^+}$ is entangled.

\subsubsection*{Example 3 ($2\times3$): minors as explicit equations}
Let
\[
\ket{\psi}=
a\ket{00}+b\ket{01}+c\ket{02}
+d\ket{10}+e\ket{11}+f\ket{12},
\qquad
M_\psi=\begin{pmatrix}a&b&c\\ d&e&f\end{pmatrix}.
\]
Then $\ket{\psi}$ is product iff all $2\times2$ minors vanish:
\[
ae-bd=0,\qquad af-cd=0,\qquad bf-ce=0.
\]

\subsubsection*{Example 4 (Schmidt rank via rank computation)}
In the $2\times3$ case the only possibilities are
\[
\rank(M_\psi)=1 \iff \text{product},\qquad
\rank(M_\psi)=2 \iff \text{entangled}.
\]

\subsection*{Exercises}

\begin{exercise}[Two-qubit Segre equation]
	Let $\ket{\psi}=a\ket{00}+b\ket{01}+c\ket{10}+d\ket{11}$.
	Show that $\ket{\psi}$ is a product state iff $ad-bc=0$.
\end{exercise}

\noindent\textbf{Solution.}
($\Rightarrow$) If $\ket{\psi}=(\alpha\ket0+\beta\ket1)\otimes(\gamma\ket0+\delta\ket1)$, then
\[
a=\alpha\gamma,\quad b=\alpha\delta,\quad c=\beta\gamma,\quad d=\beta\delta,
\]
so
\[
ad-bc=(\alpha\gamma)(\beta\delta)-(\alpha\delta)(\beta\gamma)=0.
\]
($\Leftarrow$) If $ad-bc=0$, then $\det(M_\psi)=0$ for
$M_\psi=\begin{psmallmatrix}a&b\\c&d\end{psmallmatrix}$, hence $\rank(M_\psi)\le 1$.
Since $\psi\neq0$ in projective space, $\rank(M_\psi)=1$, so $M_\psi=u v^T$ for some $u,v$.
Therefore $\psi=u\otimes v$ is a product state.

\medskip

\begin{exercise}[Bell states and minors]
	Compute $ad-bc$ for $\ket{\Phi^\pm}$ and $\ket{\Psi^\pm}$ and conclude they are entangled.
\end{exercise}

\noindent\textbf{Solution.}
\[
\ket{\Phi^\pm}=\tfrac1{\sqrt2}(\ket{00}\pm\ket{11})
\Rightarrow
(a,b,c,d)=\bigl(\tfrac1{\sqrt2},0,0,\pm\tfrac1{\sqrt2}\bigr)
\Rightarrow ad-bc=\pm\tfrac12\neq0.
\]
\[
\ket{\Psi^\pm}=\tfrac1{\sqrt2}(\ket{01}\pm\ket{10})
\Rightarrow
(a,b,c,d)=\bigl(0,\tfrac1{\sqrt2},\pm\tfrac1{\sqrt2},0\bigr)
\Rightarrow ad-bc=-\,(\pm\tfrac12)\neq0.
\]
Hence all four Bell states are entangled.

\medskip

\begin{exercise}[$2\times3$ minors]
	For $\ket{\psi}=a\ket{00}+b\ket{01}+c\ket{02}+d\ket{10}+e\ket{11}+f\ket{12}$,
	prove that $\ket{\psi}$ is product iff
	\[
	ae-bd=0,\quad af-cd=0,\quad bf-ce=0.
	\]
\end{exercise}

\noindent\textbf{Solution.}
Let $M_\psi=\begin{psmallmatrix}a&b&c\\ d&e&f\end{psmallmatrix}$.
A state is product iff $\rank(M_\psi)=1$, and $\rank(M_\psi)=1$ iff all $2\times2$ minors vanish.
Those minors are exactly
\[
\det\!\begin{pmatrix}a&b\\ d&e\end{pmatrix}=ae-bd,\quad
\det\!\begin{pmatrix}a&c\\ d&f\end{pmatrix}=af-cd,\quad
\det\!\begin{pmatrix}b&c\\ e&f\end{pmatrix}=bf-ce.
\]
Thus the condition is necessary and sufficient.

\medskip

\begin{exercise}[Schmidt rank equals matrix rank]
	Let $M_\psi$ be the coefficient matrix of $\ket{\psi}\in\C^m\otimes\C^n$.
	Prove $\mathrm{SchmidtRank}(\psi)=\rank(M_\psi)$.
\end{exercise}

\noindent\textbf{Solution.}
Interpret $\psi$ as a linear map
\[
\widetilde{\psi}:\C^n\to\C^m,\qquad
\widetilde{\psi}(\ket{j})=\sum_{i=1}^{m}\psi_{ij}\ket{i}.
\]
In the chosen bases, the matrix of $\widetilde{\psi}$ is $M_\psi$.
An SVD
\[
\widetilde{\psi}=U\,\mathrm{diag}(s_1,\dots,s_R)\,V^\dagger
\]
yields
\[
\psi=\sum_{r=1}^{R} s_r\,\ket{u_r}\otimes\ket{v_r},
\]
a Schmidt decomposition with $R$ nonzero Schmidt coefficients.
Hence $R$ equals the number of nonzero singular values, i.e.\ $\rank(M_\psi)$.

\medskip

\begin{exercise}[Flattening drill (three qubits)]
	Let $\ket{\psi}=\sum_{i,j,k\in\{0,1\}}\psi_{ijk}\ket{i}\ket{j}\ket{k}$ be a three-qubit state.
	Form the $2\times4$ flattening matrix for the bipartition $A|BC$ and write down the
	condition for $\ket{\psi}$ to be product across $A|BC$ in terms of $2\times2$ minors.
\end{exercise}

\noindent\textbf{Solution.}
Order the $BC$ basis as $\{\ket{00},\ket{01},\ket{10},\ket{11}\}$ and set
\[
M^{A|BC}_\psi=
\begin{pmatrix}
	\psi_{0,0,0}&\psi_{0,0,1}&\psi_{0,1,0}&\psi_{0,1,1}\\
	\psi_{1,0,0}&\psi_{1,0,1}&\psi_{1,1,0}&\psi_{1,1,1}
\end{pmatrix}.
\]
The state is product across $A|BC$ iff $\rank(M^{A|BC}_\psi)=1$,
equivalently iff every $2\times2$ minor of this $2\times4$ matrix vanishes.
For example, using the first two columns gives the condition
\[
\psi_{000}\psi_{101}-\psi_{001}\psi_{100}=0,
\]
and similarly for each pair of columns.

\section{Algorithm IV: Entanglement Detection and Measurement}
\label{sec:ent-detect}

\subsection{Objective}
Entanglement is not a philosophical label; it is a \emph{measurable resource} that changes what circuits can do,
how noise propagates, and what classical post-processing must estimate.
This section gives two practical detection toolchains:
\begin{enumerate}
	\item \textbf{Exact (two qubits):} compute \emph{concurrence} from a reconstructed (or known) density matrix.
	\item \textbf{Fast (many qubits, few shots):} use \emph{entanglement witnesses} from a small set of observables.
\end{enumerate}
We emphasize operational meaning (what you actually measure), algorithmic cost (shots and classical compute),
and circuit templates to prepare test states.

\subsection{Why these objects matter (engineering + algorithm viewpoint)}
\paragraph{Engineering reality.}
In real devices, you do not ``see'' amplitudes. You see \emph{classical measurement outcomes} and their statistics.
Entanglement detection matters because:
\begin{itemize}
	\item \textbf{Calibration/verification:} your compiler may claim ``prepare a Bell state''; you must verify it.
	\item \textbf{Debugging noise:} many errors (dephasing, crosstalk) collapse entanglement before they visibly change
	single-qubit marginals.
	\item \textbf{Downstream algorithms:} teleportation, QEC primitives, and variational circuits rely on entanglement.
	\item \textbf{Latency/throughput:} full tomography is expensive; witnesses can be low-latency.
\end{itemize}

\paragraph{Two detection regimes.}
\begin{center}
	\begin{tabular}{p{0.33\textwidth} p{0.62\textwidth}}
		\textbf{Regime} & \textbf{Tool and tradeoff}\\ \hline
		Two qubits, need an exact number & Concurrence $C(\rho)\in[0,1]$; requires $\rho$ (tomography or a trusted model).\\
		Many qubits / quick screening & Witness $W$; requires only a few expectation values, but is only sufficient (not necessary).\\
	\end{tabular}
\end{center}


\begin{figure}[t]
	\centering
	\resizebox{\textwidth}{!}{%
		\begin{tikzpicture}[
			>=Latex,
			line cap=round,
			line join=round,
			font=\small,
			box/.style={draw, rounded corners, align=center, minimum height=9mm},
			arr/.style={->, thick},
			node distance=10mm
			]
			\node[box, minimum width=32mm] (prep) {Prepare circuit\\$\rho$ (unknown)};
			
			\node[box, right=10mm of prep, minimum width=40mm] (meas)
			{Measure chosen observables\\shots $\to$ bitstrings};
			
			\node[box, right=10mm of meas, minimum width=44mm] (post)
			{Classical post-processing\\estimate $\langle O_k\rangle$};
			
			\node[box, right=10mm of post, minimum width=42mm] (dec)
			{Decision\\$C(\rho)$ or witness};
			
			\draw[arr] (prep) -- (meas);
			\draw[arr] (meas) -- (post);
			\draw[arr] (post) -- (dec);
			
			\node[align=center, font=\small] at ($(post.south)+(0,-8mm)$)
			{Two choices: reconstruct $\rho$ (tomography) or compute a witness directly.};
		\end{tikzpicture}%
	}
	\caption{Detection pipeline: circuits $\to$ measurements $\to$ classical estimation $\to$ entanglement decision.}
	\label{fig:ent-detect-pipeline}
\end{figure}

\subsection{Warm-up: what measurement means here}

\subsubsection{Observables, expectation values, and Pauli strings}
For $n$ qubits, an important measurement family is the \emph{Pauli strings}:
\[
P = \sigma_{a_1}\otimes \sigma_{a_2}\otimes \cdots \otimes \sigma_{a_n},
\quad
\sigma_{a_k}\in\{I,X,Y,Z\}.
\]
Given a state $\rho$, the expectation value is
\[
\langle P\rangle_\rho := \Tr(\rho P).
\]
In experiments, you estimate $\langle P\rangle$ from repeated shots.
If $P$ has eigenvalues $\pm 1$ and you can measure in its eigenbasis, then
\[
\widehat{\langle P\rangle}=\frac{1}{N}\sum_{s=1}^N o_s,
\qquad o_s\in\{+1,-1\}.
\]

\subsubsection{Basis changes for Pauli measurements}
Hardware often measures in the computational ($Z$) basis.
To measure $X$ on a qubit, apply $H$ then measure $Z$.
To measure $Y$, apply $S^\dagger H$ then measure $Z$ (with $S=\begin{psmallmatrix}1&0\\0&i\end{psmallmatrix}$).
For Pauli strings, apply these basis changes qubit-wise.

\begin{figure}[t]
	\centering
	\begin{tikzpicture}[scale=1.0, line cap=round, line join=round]
		\draw (0,0) -- (7.4,0);
		\draw (0,-0.9) -- (7.4,-0.9);
		\node[left] at (0,0) {$q_1$};
		\node[left] at (0,-0.9) {$q_2$};
		
		\node[draw, rounded corners, minimum width=0.9cm, minimum height=0.55cm] (U1) at (2.0,0) {$H$};
		\node[draw, rounded corners, minimum width=0.9cm, minimum height=0.55cm] (U2) at (2.0,-0.9) {$S^\dagger H$};
		
		\draw (6.2,0) node[draw, circle, inner sep=1.2pt] (m1) {};
		\draw (6.2,-0.9) node[draw, circle, inner sep=1.2pt] (m2) {};
		\node[right] at (6.35,0) {$Z$};
		\node[right] at (6.35,-0.9) {$Z$};
		
		\node[align=center] at (3.9,0.55) {\small Measures $X\otimes Y$ via $Z$-readout};
	\end{tikzpicture}
	\caption{Measuring $X\otimes Y$: apply $H$ on qubit 1 and $S^\dagger H$ on qubit 2, then measure both in $Z$.}
	\label{fig:measure-xy}
\end{figure}
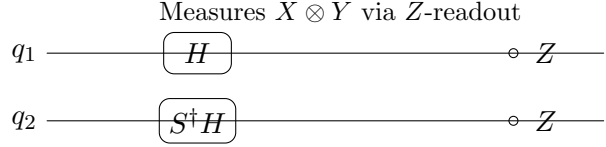

\subsubsection{Separable states and a simple necessary condition}
A bipartite state $\rho_{AB}$ is \emph{separable} if it can be written as
\[
\rho_{AB}=\sum_k p_k\,\rho_A^{(k)}\otimes \rho_B^{(k)},\qquad p_k\ge 0,\ \sum_k p_k=1.
\]
Entangled states are those that are not of this form.
Many tests (witnesses) certify \emph{non-separability} using inequalities satisfied by all separable states.

\subsection{Concurrence (two qubits): what it measures and why it matters}

\subsubsection{Definition (Wootters)}
For a two-qubit density matrix $\rho$ (a $4\times 4$ PSD matrix with trace $1$), define the \emph{spin-flipped} state
\[
\tilde{\rho}:=(\sigma_y\otimes\sigma_y)\,\rho^*\,(\sigma_y\otimes\sigma_y),
\]
where $\rho^*$ is complex conjugation in the computational basis and
\[
\sigma_y=
\begin{pmatrix}
	0 & -i\\
	i & 0
\end{pmatrix}.
\]
Let the eigenvalues of $\rho\tilde{\rho}$ be $\lambda_1\ge \lambda_2\ge \lambda_3\ge \lambda_4\ge 0$.
Define
\[
C(\rho):=\max\Bigl\{0,\ \sqrt{\lambda_1}-\sqrt{\lambda_2}-\sqrt{\lambda_3}-\sqrt{\lambda_4}\Bigr\}.
\]
Then $C(\rho)\in[0,1]$, with $C(\rho)=0$ for separable states and $C(\rho)=1$ for Bell states.

\subsubsection{Geometric/operational meaning}
Concurrence is an \emph{entanglement monotone}: it does not increase under LOCC operations.
For two-qubit pure states $\ket{\psi}$, concurrence reduces to a simple determinant formula (see below),
making it an ideal ``ground truth'' metric for benchmarking state-preparation and noise models.

\subsubsection{Closed form for pure states}
Let
\[
\ket{\psi}=a\ket{00}+b\ket{01}+c\ket{10}+d\ket{11}.
\]
Then the concurrence is
\[
C(\ket{\psi})=2|ad-bc|.
\]
This can be derived from the Schmidt decomposition, or by directly evaluating $\rho\tilde{\rho}$ for $\rho=\ket{\psi}\bra{\psi}$.

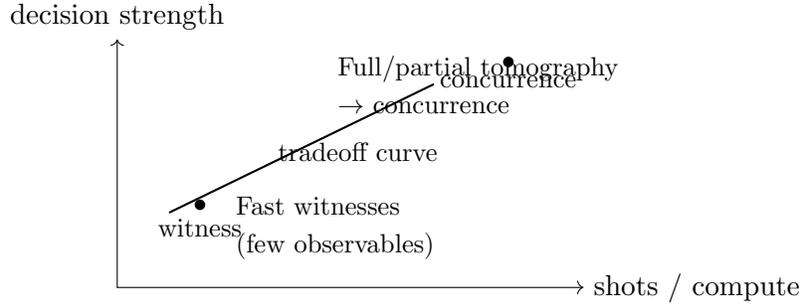
\begin{figure}[t]
	\centering
	\begin{tikzpicture}[scale=1.0, line cap=round, line join=round]
		\draw[->] (0,0) -- (6.2,0) node[right] {shots / compute};
		\draw[->] (0,0) -- (0,3.3) node[above] {decision strength};
		
		\draw[thick] (0.7,1.0) -- (4.2,2.7);
		\node[align=left] at (2.9,0.8) {\small Fast witnesses\\\small (few observables)};
		\node[align=left] at (4.8,2.7) {\small Full/partial tomography\\\small $\to$ concurrence};
		
		\fill (1.1,1.1) circle (2pt);
		\node[below] at (1.1,1.05) {\small witness};
		
		\fill (5.2,3.0) circle (2pt);
		\node[below] at (5.2,2.95) {\small concurrence};
		
		\node at (3.2,1.8) {\small tradeoff curve};
	\end{tikzpicture}
	\caption{Concurrence vs.\ witnesses: concurrence is stronger but typically needs more estimation effort.}
	\label{fig:tradeoff-conc-wit}
\end{figure}

\subsection{Mixed states: explicit concurrence algorithm (two qubits)}

\subsubsection{Algorithm (step-by-step)}
Given a two-qubit density matrix $\rho$ in the computational basis ordering
$\{\ket{00},\ket{01},\ket{10},\ket{11}\}$:

\begin{enumerate}
	\item Form $\rho^*$ (complex conjugate entrywise in this basis).
	\item Compute $\tilde{\rho}=(\sigma_y\otimes\sigma_y)\rho^*(\sigma_y\otimes\sigma_y)$.
	\item Compute $R=\rho\tilde{\rho}$.
	\item Compute the eigenvalues of $R$, sort them decreasing: $\lambda_1\ge\lambda_2\ge\lambda_3\ge\lambda_4$.
	\item Output $C(\rho)=\max\{0,\sqrt{\lambda_1}-\sqrt{\lambda_2}-\sqrt{\lambda_3}-\sqrt{\lambda_4}\}$.
\end{enumerate}

\subsubsection{Worked example: Werner state}
The Werner state is
\[
\rho_W(p)=p\ket{\Phi^+}\bra{\Phi^+}+(1-p)\frac{I_4}{4},
\qquad 0\le p\le 1,
\]
where $\ket{\Phi^+}=\frac{1}{\sqrt2}(\ket{00}+\ket{11})$.
One finds (exercise) that
\[
C(\rho_W(p))=\max\Bigl\{0,\ \frac{3p-1}{2}\Bigr\}.
\]
This is a concrete ``noise knob'': entanglement disappears at $p\le 1/3$.

\subsection{Entanglement witnesses: fast detection from few measurements}

\subsubsection{Definition and decision rule}
A Hermitian operator $W$ is an \emph{entanglement witness} if
\[
\Tr(W\sigma)\ge 0\quad\text{for all separable }\sigma,
\]
but there exists at least one entangled $\rho$ such that $\Tr(W\rho)<0$.
Thus, a negative expectation value \emph{certifies entanglement}.

\subsubsection{Bell-state witness (two qubits)}
For $\ket{\Phi^+}$, a standard witness is
\[
W_{\Phi^+}=\frac12 I_4-\ket{\Phi^+}\bra{\Phi^+}.
\]
If $\Tr(W_{\Phi^+}\rho)<0$, then $\rho$ is entangled and has fidelity with $\ket{\Phi^+}$ exceeding $1/2$.

\subsubsection{Pauli-decomposition form (measurement-friendly)}
Using
\[
\ket{\Phi^+}\bra{\Phi^+}
=\frac14\Bigl(I\otimes I + X\otimes X - Y\otimes Y + Z\otimes Z\Bigr),
\]
we get
\[
W_{\Phi^+}
=\frac14\Bigl(I\otimes I - X\otimes X + Y\otimes Y - Z\otimes Z\Bigr).
\]
So you can evaluate $\Tr(W_{\Phi^+}\rho)$ by estimating only
$\langle X\otimes X\rangle$, $\langle Y\otimes Y\rangle$, $\langle Z\otimes Z\rangle$.

\paragraph{Decision inequality.}
\[
\Tr(W_{\Phi^+}\rho)<0
\quad\Longleftrightarrow\quad
\langle X\otimes X\rangle - \langle Y\otimes Y\rangle + \langle Z\otimes Z\rangle > 1.
\]
This is extremely practical: three correlators, one inequality.

\subsubsection{Shot complexity (rule of thumb)}
If each correlator is estimated from $N$ shots, then standard deviation scales like $O(1/\sqrt{N})$.
Witness tests are attractive because you can make a decision with far fewer settings than tomography.

\subsection{Circuit examples: preparing states to test}

\subsubsection{Bell state preparation}
\[
\ket{00}\ \xrightarrow{H\otimes I}\ \frac{1}{\sqrt2}(\ket{00}+\ket{10})
\ \xrightarrow{\mathrm{CNOT}}\ \ket{\Phi^+}.
\]

\begin{figure}[t]
	\centering
	\begin{tikzpicture}[scale=1.0, line cap=round, line join=round]
		\draw (0,0) -- (8.5,0);
		\draw (0,-1.0) -- (8.5,-1.0);
		\node[left] at (0,0) {$q_1$};
		\node[left] at (0,-1.0) {$q_2$};
		
		\node[draw, rounded corners, minimum width=0.8cm, minimum height=0.55cm] (H) at (1.6,0) {$H$};
		
		\draw (4.4,0) node[draw, circle, inner sep=1.2pt] (ctrl) {};
		\draw (4.4,-1.0) node[draw, circle, inner sep=1.2pt] (tgt) {};
		\draw (tgt) ++(-0.22,0) -- ++(0.44,0);
		\draw (tgt) ++(0,-0.22) -- ++(0,0.44);
		\draw[thick] (ctrl) -- (tgt);
		
		\draw (7.5,0) node[draw, circle, inner sep=1.2pt] (m1) {};
		\draw (7.5,-1.0) node[draw, circle, inner sep=1.2pt] (m2) {};
		\node[right] at (7.65,0) {\small $Z$};
		\node[right] at (7.65,-1.0) {\small $Z$};
		
		\node[align=center] at (5.9,0.65) {\small Bell prep + readout};
	\end{tikzpicture}
	\caption{Circuit to prepare $\ket{\Phi^+}$ from $\ket{00}$.}
	\label{fig:bell-prep}
\end{figure}
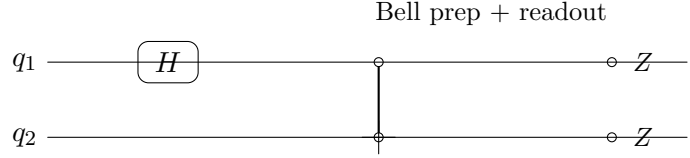

\subsubsection{Producing a partially entangled pure state}
Replace $H$ by $R_Y(\theta)$ on the first qubit:
\[
(R_Y(\theta)\otimes I)\ket{00}=\cos(\tfrac\theta2)\ket{00}+\sin(\tfrac\theta2)\ket{10},
\]
then apply CNOT to obtain
\[
\ket{\psi(\theta)}=\cos(\tfrac\theta2)\ket{00}+\sin(\tfrac\theta2)\ket{11}.
\]
Its concurrence is
\[
C(\ket{\psi(\theta)})=|\sin\theta|.
\]
This is a clean knob to validate a detector across entanglement strengths.

\subsubsection{Noisy Bell state (dephasing toy model)}
A simple dephasing model yields
\[
\rho(\gamma)=(1-\gamma)\ket{\Phi^+}\bra{\Phi^+}+\gamma\frac{\ket{00}\bra{00}+\ket{11}\bra{11}}{2},
\quad 0\le \gamma\le 1.
\]
Witness values degrade smoothly with $\gamma$; concurrence drops and eventually hits $0$.

\subsection{Concurrence vs.\ witness: when to use which}

\paragraph{Use concurrence when:}
\begin{itemize}
	\item you are in the two-qubit regime and want a \emph{quantitative} entanglement score,
	\item you have (or can cheaply reconstruct) $\rho$,
	\item you are producing benchmark plots across noise knobs.
\end{itemize}

\paragraph{Use witnesses when:}
\begin{itemize}
	\item you need \emph{fast certification} in a loop (calibration, debugging, regression tests),
	\item you have more qubits or limited measurement settings,
	\item you only need a \emph{yes/no} entanglement certificate for a target family (e.g.\ Bell-like states).
\end{itemize}

\paragraph{A practical hybrid strategy.}
In workflows: use a witness for daily/CI tests; run concurrence/tomography periodically or when a witness fails.

\subsection{Exercises}

\begin{exercise}[Pure-state concurrence formula]
	Let $\ket{\psi}=a\ket{00}+b\ket{01}+c\ket{10}+d\ket{11}$.
	Show that $C(\ket{\psi})=2|ad-bc|$.
\end{exercise}
\noindent\textbf{Solution.}
Write $\rho=\ket{\psi}\bra{\psi}$ and compute the spin-flip:
\[
\ket{\tilde{\psi}}=(\sigma_y\otimes\sigma_y)\ket{\psi^*}.
\]
Using $\sigma_y\ket{0}=i\ket{1}$ and $\sigma_y\ket{1}=-i\ket{0}$, one checks
\[
\ket{\tilde{\psi}} = d^*\ket{00}-c^*\ket{01}-b^*\ket{10}+a^*\ket{11}.
\]
For pure states, Wootters shows $C(\ket{\psi})=|\langle \psi|\tilde{\psi}\rangle|$.
Compute
\[
\langle\psi|\tilde{\psi}\rangle
= a d^* - b c^* - c b^* + d a^*
= 2(ad-bc)^*,
\]
hence $C(\ket{\psi})=|2(ad-bc)^*|=2|ad-bc|$.

\begin{exercise}[Bell witness via correlators]
	Prove the Pauli expansion
	\[
	\ket{\Phi^+}\bra{\Phi^+}
	=\frac14(I\otimes I+X\otimes X-Y\otimes Y+Z\otimes Z),
	\]
	and deduce the witness inequality
	\[
	\Tr(W_{\Phi^+}\rho)<0
	\iff
	\langle X\otimes X\rangle - \langle Y\otimes Y\rangle + \langle Z\otimes Z\rangle > 1.
	\]
\end{exercise}
\noindent\textbf{Solution.}
Compute the action of $X\otimes X$, $Y\otimes Y$, $Z\otimes Z$ on the Bell basis.
For $\ket{\Phi^+}=\frac{1}{\sqrt2}(\ket{00}+\ket{11})$:
\[
(X\otimes X)\ket{\Phi^+}=\ket{\Phi^+},\quad
(Y\otimes Y)\ket{\Phi^+}=-\ket{\Phi^+},\quad
(Z\otimes Z)\ket{\Phi^+}=\ket{\Phi^+}.
\]
Also $\Tr(\ket{\Phi^+}\bra{\Phi^+})=1$ and the projector must be a Hermitian operator
supported on the Bell subspace. The stated Pauli combination is Hermitian, has trace $1$,
and has $\ket{\Phi^+}$ as eigenvector with eigenvalue $1$ while annihilating the orthogonal Bell states
(check similarly), hence equals $\ket{\Phi^+}\bra{\Phi^+}$.
Then
\[
W_{\Phi^+}=\frac12 I-\ket{\Phi^+}\bra{\Phi^+}
=\frac14(I\otimes I - X\otimes X + Y\otimes Y - Z\otimes Z).
\]
Taking expectation in $\rho$ gives
\[
\Tr(W_{\Phi^+}\rho)
=\frac14\Bigl(1-\langle X\otimes X\rangle+\langle Y\otimes Y\rangle-\langle Z\otimes Z\rangle\Bigr),
\]
so negativity is equivalent to
\[
1-\langle X\otimes X\rangle+\langle Y\otimes Y\rangle-\langle Z\otimes Z\rangle<0
\iff
\langle X\otimes X\rangle-\langle Y\otimes Y\rangle+\langle Z\otimes Z\rangle>1.
\]

\begin{exercise}[Werner concurrence threshold]
	For $\rho_W(p)=p\ket{\Phi^+}\bra{\Phi^+}+(1-p)\frac{I_4}{4}$, show that
	\[
	C(\rho_W(p))=\max\Bigl\{0,\frac{3p-1}{2}\Bigr\}.
	\]
\end{exercise}
\noindent\textbf{Solution.}
Use the fact that $\rho_W(p)$ is diagonal in the Bell basis with eigenvalues
\[
\lambda_{\Phi^+}=\frac{1+3p}{4},\qquad
\lambda_{\Phi^-}=\lambda_{\Psi^+}=\lambda_{\Psi^-}=\frac{1-p}{4}.
\]
For Werner states, $\tilde{\rho}=\rho$ (in the Bell basis), hence $\rho\tilde{\rho}=\rho^2$
and the eigenvalues of $\rho\tilde{\rho}$ are $\lambda_i^2$.
Thus $\sqrt{\lambda_i(\rho\tilde{\rho})}=\lambda_i(\rho)$.
Plug into Wootters:
\[
C(\rho)=\max\{0,\lambda_1-\lambda_2-\lambda_3-\lambda_4\}
=\max\Bigl\{0,\frac{1+3p}{4}-3\cdot\frac{1-p}{4}\Bigr\}
=\max\Bigl\{0,\frac{3p-1}{2}\Bigr\}.
\]

\begin{exercise}[Tomography-lite for Bell states]
	Assume $\rho$ is close to $\ket{\Phi^+}$ and you can estimate
	$\langle X\otimes X\rangle,\langle Y\otimes Y\rangle,\langle Z\otimes Z\rangle$.
	Give an estimator for the fidelity $F=\langle \Phi^+|\rho|\Phi^+\rangle$.
\end{exercise}
\noindent\textbf{Solution.}
Using the Pauli expansion,
\[
F=\Tr(\rho\,\ket{\Phi^+}\bra{\Phi^+})
=\frac14\Bigl(1+\langle X\otimes X\rangle-\langle Y\otimes Y\rangle+\langle Z\otimes Z\rangle\Bigr).
\]
So the plug-in estimator is the same expression with empirical correlators.

\begin{exercise}[Circuit knob: concurrence vs.\ $\theta$]
	Prepare $\ket{\psi(\theta)}=\cos(\tfrac\theta2)\ket{00}+\sin(\tfrac\theta2)\ket{11}$
	using $R_Y(\theta)$ + CNOT. Compute $C(\ket{\psi(\theta)})$ and the witness value
	$\Tr(W_{\Phi^+}\ket{\psi(\theta)}\bra{\psi(\theta)})$.
\end{exercise}
\noindent\textbf{Solution.}
Here $a=\cos(\tfrac\theta2)$, $d=\sin(\tfrac\theta2)$, $b=c=0$, so
\[
C(\ket{\psi(\theta)})=2|ad|=2\Bigl|\cos(\tfrac\theta2)\sin(\tfrac\theta2)\Bigr|=|\sin\theta|.
\]
For the witness,
\[
\Tr(W_{\Phi^+}\rho)=\frac12 - |\langle \Phi^+|\psi(\theta)\rangle|^2.
\]
Compute
\[
\langle \Phi^+|\psi(\theta)\rangle
=\frac{1}{\sqrt2}\Bigl(\cos(\tfrac\theta2)+\sin(\tfrac\theta2)\Bigr),
\]
so
\[
|\langle \Phi^+|\psi(\theta)\rangle|^2
=\frac12\Bigl(\cos(\tfrac\theta2)+\sin(\tfrac\theta2)\Bigr)^2
=\frac12\bigl(1+\sin\theta\bigr).
\]
Therefore
\[
\Tr(W_{\Phi^+}\rho)=\frac12-\frac12(1+\sin\theta)=-\frac12\sin\theta.
\]
So the witness detects entanglement for $\sin\theta>0$ (and fails for $\sin\theta\le 0$),
illustrating that witnesses can be directionally biased toward a target family.


\section{Algorithm V: Shor's Algorithm via Quantum Phase Estimation}
\label{sec:shor}

\subsection{Objective}
Shor's algorithm is best understood as a \emph{hybrid pipeline}:
\[
\text{factoring }N
\ \longrightarrow\
\text{period finding }r
\ \longrightarrow\
\text{eigenphase estimation (QPE)}
\ \longrightarrow\
\text{continued fractions + gcd (classical)}.
\]
This section explains why the quantum part is essentially \emph{phase estimation of a modular-multiplication unitary},
how the circuit is organized, and where classical/FPGA acceleration naturally appears (modular arithmetic, streaming statistics,
continued fractions, gcd, and scheduling/latency constraints).

\subsection{Why Shor's algorithm matters}
\paragraph{Cryptography impact (conceptual).}
The widely taught reason is that Shor's algorithm factors integers and computes discrete logs in polynomial time,
threatening RSA and ECC \emph{if} sufficiently large fault-tolerant quantum computers exist.

\paragraph{Systems impact (engineering).}
Shor is also a canonical example of:
\begin{itemize}
	\item a deep quantum circuit with large reversible arithmetic blocks,
	\item repeated controlled powers of a unitary (controlled-$U^{2^k}$),
	\item heavy classical post-processing tightly coupled to the measurement outcomes,
	\item a demanding error-correction/control stack (latency and scheduling).
\end{itemize}
So even before full-scale factoring is feasible, Shor is a design template for how quantum and classical infrastructure must co-exist.

\subsection{From factoring to period finding}
\subsubsection{Reduction (the number-theory skeleton)}
Given an odd composite integer $N$ and an integer $a$ with $1<a<N$:
\begin{enumerate}
	\item If $\gcd(a,N)\neq 1$, you already found a nontrivial factor.
	\item Otherwise, consider the function $f(x)=a^x \bmod N$.
	This function is periodic with some period $r$ (the order of $a$ modulo $N$):
	\[
	a^r \equiv 1 \pmod N,\qquad r=\mathrm{ord}_N(a).
	\]
	\item If $r$ is even and $a^{r/2}\not\equiv -1\pmod N$, then
	\[
	\gcd\!\bigl(a^{r/2}-1,N\bigr),\quad \gcd\!\bigl(a^{r/2}+1,N\bigr)
	\]
	yield nontrivial factors of $N$.
\end{enumerate}

\subsubsection{Why the gcd trick works}
If $a^r\equiv 1\pmod N$, then
\[
a^r-1 \equiv 0 \pmod N
\quad\Rightarrow\quad
(a^{r/2}-1)(a^{r/2}+1)\equiv 0 \pmod N.
\]
If neither factor is $0\pmod N$, then $N$ shares a nontrivial gcd with each factor, producing nontrivial factors.

\begin{exercise}
	Show that if $r$ is even and $a^{r/2}\not\equiv -1 \pmod N$, then at least one of
	$\gcd(a^{r/2}-1,N)$ or $\gcd(a^{r/2}+1,N)$ is a nontrivial factor of $N$.
\end{exercise}

\subsection{Period finding as an eigenphase problem}

\subsubsection{The key unitary: modular multiplication}
Work on the computational basis states $\ket{y}$ with $y\in\{0,1,\dots,N-1\}$.
Define the unitary (a permutation on basis states)
\[
U_a\ket{y}=\ket{a y \bmod N}.
\]
This is unitary because multiplication by $a$ mod $N$ is a bijection on $\{0,\dots,N-1\}$ whenever $\gcd(a,N)=1$.

\subsubsection{Eigenstates from periodic orbits}
Fix $y=1$ and define the orbit
\[
\ket{1},\ \ket{a},\ \ket{a^2},\ \dots,\ \ket{a^{r-1}},
\qquad a^r\equiv 1\pmod N.
\]
Consider the following superpositions (Fourier modes on the cycle):
\[
\ket{\psi_s}
=\frac{1}{\sqrt r}\sum_{k=0}^{r-1} e^{-2\pi i s k/r}\,\ket{a^k \bmod N},
\qquad s=0,1,\dots,r-1.
\]
Then
\[
U_a\ket{\psi_s}=e^{2\pi i s/r}\ket{\psi_s}.
\]
So the period $r$ is encoded in the eigenphases $2\pi s/r$ of $U_a$.

\begin{proof}[Verification]
	Compute
	\[
	U_a\ket{\psi_s}
	=\frac{1}{\sqrt r}\sum_{k=0}^{r-1} e^{-2\pi i s k/r}\,U_a\ket{a^k}
	=\frac{1}{\sqrt r}\sum_{k=0}^{r-1} e^{-2\pi i s k/r}\,\ket{a^{k+1}}.
	\]
	Reindex $k'=k+1$ (mod $r$):
	\[
	U_a\ket{\psi_s}
	=\frac{1}{\sqrt r}\sum_{k'=0}^{r-1} e^{-2\pi i s (k'-1)/r}\,\ket{a^{k'}}
	=e^{2\pi i s/r}\frac{1}{\sqrt r}\sum_{k'=0}^{r-1} e^{-2\pi i s k'/r}\,\ket{a^{k'}}
	=e^{2\pi i s/r}\ket{\psi_s}.
	\]
\end{proof}

\subsection{Quantum phase estimation (QPE) as a measurement}

\subsubsection{What QPE estimates}
QPE estimates an eigenphase $\varphi\in[0,1)$ in
\[
U\ket{u}=e^{2\pi i \varphi}\ket{u}.
\]
In Shor, $U=U_a$ and $\varphi=s/r$ for an unknown $r$ (and a random $s$ depending on the eigenstate component).

\subsubsection{QPE circuit skeleton}
Let the top register have $t$ qubits (precision register) and the bottom register store $\ket{1}$ (work register).
QPE performs:
\begin{enumerate}
	\item Apply $H^{\otimes t}$ to the top register to create a uniform superposition over $x\in\{0,\dots,2^t-1\}$.
	\item Apply controlled-$U^{2^k}$ gates (powers of $U$) conditioned on each top qubit.
	\item Apply inverse QFT on the top register.
	\item Measure the top register to obtain an integer $m\in\{0,\dots,2^t-1\}$ approximating $2^t\varphi$.
\end{enumerate}

\begin{figure}[t]
	\centering
	\begin{tikzpicture}[scale=1.0, line cap=round, line join=round]
		\foreach \y/\lab in {0/{q_{t-1}},-0.7/{\vdots},-1.4/{q_1},-2.1/{q_0}}{
			\draw (0,\y) -- (12.2,\y);
			\node[left] at (0,\y) {$\lab$};
		}
		\draw (0,-3.2) -- (12.2,-3.2);
		\node[left] at (0,-3.2) {$\ket{1}$};
		
		\foreach \x/\y in {1.2/0,1.2/-1.4,1.2/-2.1}{
			\node[draw, rounded corners, minimum width=0.7cm, minimum height=0.45cm] at (\x,\y) {$H$};
		}
		\node at (1.2,-0.7) {\small $H$};
		
		\foreach \y/\pow/\x in {0/{2^{t-1}}/5.4,-1.4/{2^1}/5.4,-2.1/{2^0}/5.4}{
			\fill (\x,\y) circle (1.3pt);
			\draw[thick] (\x,\y) -- (\x,-3.2);
			\node[draw, rounded corners, minimum width=2.4cm, minimum height=0.7cm, align=center] at (\x+2.2,-3.2)
			{\small $U^{\pow}$\\\small (mod mult)};
		}
		\node at (5.4,-0.7) {\small $\cdots$};
		
		\node[draw, rounded corners, minimum width=2.8cm, minimum height=2.1cm, align=center] (iqft) at (10.2,-1.05)
		{\small inverse\\\small QFT};
		
		\foreach \y in {0,-0.7,-1.4,-2.1}{
			\draw (11.7,\y) node[draw, circle, inner sep=1.2pt] {};
			\node[right] at (11.85,\y) {\small meas};
		}
		
		\node[align=left] at (6.2,-4.0)
		{\small QPE measures an eigenphase $\varphi\approx m/2^t$ of $U_a$. In Shor, $\varphi=s/r$.};
	\end{tikzpicture}
	\caption{QPE skeleton (schematic): controlled powers of $U_a$ (modular multiplication) followed by inverse QFT.}
	\label{fig:qpe-skeleton}
\end{figure}
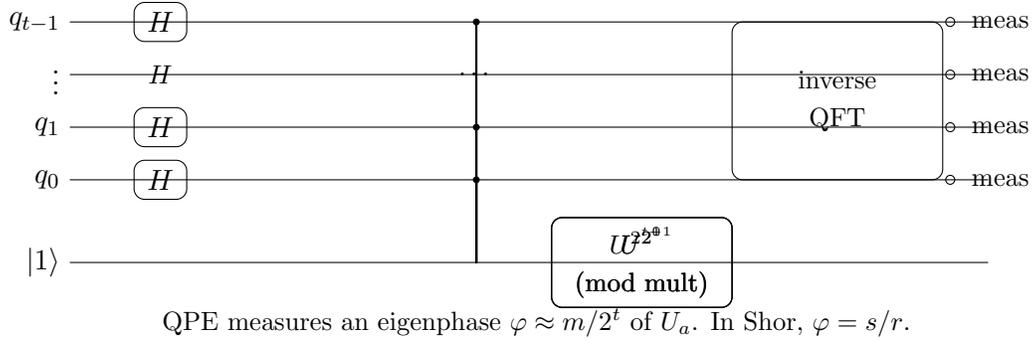

\subsubsection{One-line math of QPE output}
If the work register is an eigenstate $\ket{u}$ of $U$, the measurement $m$ concentrates near $2^t\varphi$.
In Shor, you obtain an estimate
\[
\frac{m}{2^t}\approx \frac{s}{r}.
\]
Then classical continued fractions recovers a candidate denominator $r$.

\subsection{Circuit anatomy of Shor's algorithm (hybrid pipeline)}

\subsubsection{Registers}
\begin{itemize}
	\item \textbf{Phase register:} $t$ qubits (precision).
	\item \textbf{Work register:} $n=\lceil \log_2 N\rceil$ qubits to store numbers mod $N$.
	\item \textbf{Ancillas:} for reversible modular multiplication/addition circuits.
\end{itemize}

\subsubsection{Quantum block: controlled modular multiplication}
Controlled-$U_a^{2^k}$ is implemented as controlled modular multiplication by $a^{2^k}\bmod N$:
\[
\ket{y}\mapsto \ket{a^{2^k}y \bmod N}.
\]
This is the heavy part: large reversible arithmetic (adders, modular reduction, uncomputation).

\subsubsection{Classical block: from $m$ to $r$ to factors}
Given measured $m$, do:
\begin{enumerate}
	\item compute the continued fraction expansion of $m/2^t$ to find $s/r$,
	\item test candidate $r$ by checking whether $a^r\equiv 1\pmod N$,
	\item if $r$ even and $a^{r/2}\not\equiv -1\pmod N$, compute gcds to get factors.
\end{enumerate}


\begin{figure}[t]
	\centering
	\resizebox{\textwidth}{!}{%
		\begin{tikzpicture}[
			font=\small,
			box/.style={draw, rounded corners, align=center, minimum height=9mm},
			arr/.style={-Latex, thick},
			lab/.style={font=\footnotesize},
			node distance=12mm and 14mm
			]
			\node[box, minimum width=40mm] (choose)
			{Choose $a$\\compute $\gcd(a,N)$};
			
			\node[box, right=14mm of choose, minimum width=44mm] (qpe)
			{Quantum: QPE on $U_a$\\measure $m$};
			
			\node[box, right=14mm of qpe, minimum width=52mm] (cf)
			{Classical: CF on $m/2^t$\\candidate $r$};
			
			\node[box, right=14mm of cf, minimum width=46mm] (gcd)
			{$\gcd\!\bigl(a^{r/2}\pm 1,\,N\bigr)$\\output factors};
			
			\draw[arr] (choose) -- (qpe);
			\draw[arr] (qpe) -- (cf);
			\draw[arr] (cf) -- (gcd);
			
			\node[
			below=10mm of cf,
			align=center,
			text width=150mm
			] (note)
			{Hybrid loop: if $r$ fails tests, repeat with new $a$ or new shot outcomes.};
			
		\end{tikzpicture}%
	}
	\caption{Shor pipeline: classical pre-check $\to$ quantum period finding (QPE) $\to$ classical continued fractions + gcd.}
	\label{fig:shor-pipeline}
\end{figure}

\subsection{Worked toy example: $N=15$, $a=2$}

\subsubsection{Step 1: gcd pre-check}
\[
\gcd(2,15)=1,
\]
so continue.

\subsubsection{Step 2: compute the order $r$ (ground truth)}
Compute powers of $2$ mod $15$:
\[
2^1\equiv 2,\quad
2^2\equiv 4,\quad
2^3\equiv 8,\quad
2^4\equiv 16\equiv 1\pmod{15}.
\]
So the period is
\[
r=4.
\]

\subsubsection{Step 3: what eigenphases QPE is estimating}
The eigenphases are
\[
\varphi=\frac{s}{r}\in\Bigl\{0,\frac14,\frac12,\frac34\Bigr\}.
\]
With enough precision $t$, QPE returns $m/2^t$ close to one of these values.

\subsubsection{Step 4: classical recovery of $r$ from a QPE output}
Pick $t=4$ (so $2^t=16$) for an easy toy picture.
If QPE outputs $m=4$, then
\[
\frac{m}{2^t}=\frac{4}{16}=\frac14 \Rightarrow r=4.
\]
If it outputs $m=12$, then $12/16=3/4$ also yields denominator $4$.

\subsubsection{Step 5: factor extraction}
Since $r=4$ is even,
\[
2^{r/2}=2^2=4\not\equiv -1\equiv 14\pmod{15}.
\]
Compute gcds:
\[
\gcd(4-1,15)=\gcd(3,15)=3,\qquad
\gcd(4+1,15)=\gcd(5,15)=5.
\]
Thus $15=3\cdot 5$.

\begin{figure}[t]
	\centering
	\begin{tikzpicture}[scale=1.05, line cap=round, line join=round]
		\draw (0,0) circle (2.0);
		\draw[->] (-2.4,0) -- (2.6,0) node[right] {};
		\draw[->] (0,-2.4) -- (0,2.6) node[above] {};
		
		\foreach \ang/\lab in {0/{0},90/{1/4},180/{1/2},270/{3/4}}{
			\fill ({2*cos(\ang)},{2*sin(\ang)}) circle (2.2pt);
			\node at ({2.25*cos(\ang)},{2.25*sin(\ang)}) {\small $\lab$};
		}
		
		\node[align=center] at (0,-3.0) {\small For $r=4$, eigenphases $\varphi=s/r$ lie at quarter-turns on the unit circle.};
	\end{tikzpicture}
	\caption{Toy geometry: QPE estimates phases $e^{2\pi i s/r}$ on the unit circle; denominators reveal the period $r$.}
	\label{fig:qpe-phases-r4}
\end{figure}

\subsection{Why exponential speedup is possible (and what the bottleneck is)}
\subsubsection{Where the speedup comes from}
Classically, finding the order $r$ is believed to require super-polynomial time in general (for large $N$).
Quantumly, QPE extracts $\varphi=s/r$ with $t=O(\log N)$ precision in poly$(\log N)$ time,
assuming we can implement controlled-$U_a^{2^k}$ efficiently (reversible modular arithmetic).

\subsubsection{What becomes the practical bottleneck}
In real hardware stacks, the bottleneck is not the continued fractions or gcd.
It is:
\begin{itemize}
	\item circuit depth / error correction overhead for large modular arithmetic,
	\item controlled-power scheduling (many sequential stages),
	\item classical control latency for feed-forward (in some variants),
	\item throughput of measurement and classical post-processing at scale.
\end{itemize}

\subsection{Classical post-processing and FPGA acceleration}

\subsubsection{What classical work is in the loop}
Even in the textbook version (no adaptive feed-forward inside the quantum circuit),
you still run a high-throughput classical pipeline:
\begin{itemize}
	\item parse many $m$ samples (shots),
	\item run continued fractions repeatedly,
	\item validate candidate periods ($a^r\bmod N$),
	\item compute gcds and stop when factors found,
	\item manage retries over different $a$ values if needed.
\end{itemize}

\subsubsection{Why FPGA fits naturally}
FPGA acceleration can appear in two places:
\begin{enumerate}
	\item \textbf{Control/IO side:} fast ingestion of readout bits, online histogramming, and low-latency decision logic
	(e.g.\ accept/reject candidate $r$, schedule next runs).
	\item \textbf{Arithmetic microkernels:} modular exponentiation checks ($a^r\bmod N$), gcd computations,
	and continued-fraction kernels implemented as streaming integer arithmetic.
\end{enumerate}

\paragraph{A minimal streaming model.}
Assume a stream of measured integers $m_j\in[0,2^t)$ arrives.
A lightweight FPGA pipeline can compute:
\[
x_j=\frac{m_j}{2^t}\ \rightsquigarrow\ \text{CF candidates } \frac{s}{r}
\ \rightsquigarrow\ \text{validate } r
\ \rightsquigarrow\ \text{emit best $r$ to host}.
\]
Even if the heavy reversible modular arithmetic stays quantum-side, the \emph{classical} verification loop
is a natural place to enforce low latency and high throughput.

\subsection{Exercises}

\begin{exercise}[Order finding and gcd extraction]
	Let $N$ be odd composite and $\gcd(a,N)=1$.
	Assume $r=\mathrm{ord}_N(a)$ is even and $a^{r/2}\not\equiv -1\pmod N$.
	Show that $\gcd(a^{r/2}-1,N)$ is a nontrivial factor of $N$ or $\gcd(a^{r/2}+1,N)$ is.
\end{exercise}
\noindent\textbf{Solution.}
We have $a^r\equiv 1\pmod N$, so
\[
(a^{r/2}-1)(a^{r/2}+1)=a^r-1\equiv 0\pmod N.
\]
Thus $N$ divides the product. If $\gcd(a^{r/2}-1,N)=1$ and $\gcd(a^{r/2}+1,N)=1$ simultaneously,
then $N$ would be coprime to the product, contradicting $N\mid (a^{r/2}-1)(a^{r/2}+1)$.
So at least one gcd is $>1$. It is also $<N$ because $a^{r/2}\not\equiv \pm 1\pmod N$
(the $+1$ case would contradict minimality of $r$, and the $-1$ case is excluded by assumption).
Hence at least one gcd is a nontrivial factor.

\begin{exercise}[Toy Shor instance]
	Work through Shor's classical post-processing for $N=15$ and $a=2$:
	compute $r$, then compute the gcds and recover the factors.
\end{exercise}
\noindent\textbf{Solution.}
As in the section:
\[
2^1\equiv 2,\ 2^2\equiv 4,\ 2^3\equiv 8,\ 2^4\equiv 1\pmod{15},
\]
so $r=4$. Then $2^{r/2}=4\not\equiv -1\pmod{15}$.
Compute
\[
\gcd(4-1,15)=\gcd(3,15)=3,\qquad \gcd(4+1,15)=\gcd(5,15)=5.
\]
Thus $15=3\cdot 5$.

\begin{exercise}[QPE output interpretation]
	Assume QPE returns $m$ with $m/2^t$ close to $\varphi=s/r$ where $0\le s<r$ and $\gcd(s,r)=1$.
	Explain why a continued-fraction expansion of $m/2^t$ can recover $r$ when $t$ is large enough.
\end{exercise}
\noindent\textbf{Solution.}
The continued fraction algorithm returns best rational approximations to a real number.
If $|m/2^t - s/r| < 1/(2r^2)$, then $s/r$ is the unique convergent of the continued fraction with denominator $\le r$.
For sufficiently large precision $t$, QPE ensures $m/2^t$ is within $O(1/2^t)$ of $s/r$ with high probability.
Choosing $t$ so that $2^t \gg r^2$ makes the inequality hold, allowing recovery of $r$ from the convergent denominator.

\begin{exercise}[Eigenstates from modular orbits]
	Let $U_a\ket{y}=\ket{ay\bmod N}$ with $\gcd(a,N)=1$ and let $r$ be the order of $a$ mod $N$.
	Define $\ket{\psi_s}$ as
	\[
	\ket{\psi_s}=\frac{1}{\sqrt r}\sum_{k=0}^{r-1} e^{-2\pi i s k/r}\,\ket{a^k \bmod N}.
	\]
	Show $U_a\ket{\psi_s}=e^{2\pi i s/r}\ket{\psi_s}$.
\end{exercise}
\noindent\textbf{Solution.}
Exactly as in the derivation earlier:
\[
U_a\ket{\psi_s}
=\frac{1}{\sqrt r}\sum_{k=0}^{r-1} e^{-2\pi i s k/r}\,\ket{a^{k+1}}
=e^{2\pi i s/r}\frac{1}{\sqrt r}\sum_{k'=0}^{r-1} e^{-2\pi i s k'/r}\,\ket{a^{k'}}
=e^{2\pi i s/r}\ket{\psi_s}.
\]

\begin{exercise}[Measurement settings for witness-style checks]
	Suppose you want to certify that the work register after some arithmetic block remains close to a target eigenstate
	by measuring a small set of Pauli strings (a witness-like approach).
	Explain why basis-change layers (single-qubit gates) are essential in practice.
\end{exercise}
\noindent\textbf{Solution.}
Hardware typically measures in the $Z$ basis. A generic Pauli string includes $X$ and $Y$ factors, whose eigenbases
are not the computational basis. By inserting local basis changes (e.g.\ $H$ to measure $X$, $S^\dagger H$ to measure $Y$),
one can map the desired Pauli eigenbasis measurement into $Z$-readout.
Without these basis-change layers, you would only access $Z$-diagonal information and miss phase/coherence signals
that are crucial for eigenphase structure and entanglement created by arithmetic blocks.


	
	\Part{Topological Error Correction and Real-Time Decoding (Track A Core)}
	
\section{Noise, Errors, and Stabilizer Codes}
\label{sec:stabilizers}

\subsection{Overview and motivation}
Quantum information is fragile: the same superposition and entanglement that enable speedups
also make states sensitive to unwanted interactions with the environment and control imperfections.
The modern approach to reliability is:

\begin{quote}
	\emph{Model noise \(\to\) extract error information (syndrome) without collapsing data \(\to\)
		classically decode \(\to\) apply corrections (or track them in software).}
\end{quote}

This chapter builds the minimal, reusable toolkit:
\begin{itemize}
	\item how physical noise is abstracted into Pauli error models,
	\item what stabilizers are (as symmetry constraints),
	\item how syndrome extraction works using ancillas,
	\item how simple examples (Bell pair) already behave like a code,
	\item why surface-code geometry turns errors into chains,
	\item how homology packages “what is correctable” into a topological statement,
	\item why the classical side (decoding + real-time control) is a natural FPGA workload.
\end{itemize}

\subsection{From physical noise to Pauli error models}

\subsubsection{Noise as a quantum channel}
A general (Markovian) noise process acting on a density matrix \(\rho\) is modeled as a
\emph{quantum channel} (completely positive trace-preserving map, CPTP):
\[
\mathcal{E}(\rho)=\sum_k E_k \rho E_k^\dagger,
\qquad
\sum_k E_k^\dagger E_k = I.
\]
The \(E_k\) are Kraus operators. This is the most general “memoryless” noise model.

\subsubsection{Pauli operators and Pauli channels}
For one qubit, the Pauli operators are
\[
I=\begin{pmatrix}1&0\\0&1\end{pmatrix},\quad
X=\begin{pmatrix}0&1\\1&0\end{pmatrix},\quad
Y=\begin{pmatrix}0&-i\\ i&0\end{pmatrix},\quad
Z=\begin{pmatrix}1&0\\0&-1\end{pmatrix}.
\]
A \emph{Pauli channel} is
\[
\mathcal{E}(\rho) = p_I \rho + p_X X\rho X + p_Y Y\rho Y + p_Z Z\rho Z,
\qquad
p_I+p_X+p_Y+p_Z=1.
\]
Pauli channels are special because they “stay in the Pauli basis” under many operations (Cliffords),
which makes fault-tolerant analysis tractable.

\subsubsection{Why Pauli errors are enough (the working intuition)}
Two key reasons Pauli error models appear everywhere:
\begin{itemize}
	\item \textbf{Twirling / randomization:} inserting random Cliffords can convert many coherent errors
	into an effective stochastic Pauli channel (at the level of averaged behavior).
	\item \textbf{Stabilizer measurements only care about Pauli components:} for stabilizer codes,
	syndrome extraction effectively diagnoses Pauli-type error patterns; coherent errors often
	manifest as mixtures after repeated syndrome measurements and classical processing.
\end{itemize}

\subsubsection{Single-qubit physical noise examples mapped to Pauli language}

\paragraph{Bit-flip noise.}
With probability \(p\), apply \(X\):
\[
\mathcal{E}_{\mathrm{bit}}(\rho)=(1-p)\rho + p X\rho X.
\]

\paragraph{Phase-flip noise.}
With probability \(p\), apply \(Z\):
\[
\mathcal{E}_{\mathrm{phase}}(\rho)=(1-p)\rho + p Z\rho Z.
\]

\paragraph{Depolarizing noise.}
With probability \(p\), replace the state by maximally mixed:
\[
\mathcal{E}_{\mathrm{dep}}(\rho)=(1-p)\rho + \frac{p}{3}\,(X\rho X + Y\rho Y + Z\rho Z).
\]

\paragraph{Amplitude damping (not Pauli, but often approximated).}
Amplitude damping has Kraus operators
\[
E_0=\begin{pmatrix}1&0\\0&\sqrt{1-\gamma}\end{pmatrix},\quad
E_1=\begin{pmatrix}0&\sqrt{\gamma}\\0&0\end{pmatrix}.
\]
It is not a Pauli channel, but in many regimes one approximates its effect by a mixture that includes
\(X\)-like and \(Z\)-like components when analyzing stabilizer measurements.

\subsection{Stabilizers as symmetry constraints}

\subsubsection{The Pauli group}
For \(n\) qubits, the \(n\)-qubit Pauli group \(\mathcal{P}_n\) is generated by tensor products
of single-qubit Paulis with phases \(\{\pm 1, \pm i\}\):
\[
\mathcal{P}_n=\left\{\omega\, P_1\otimes \cdots \otimes P_n \;:\;
\omega\in\{\pm 1,\pm i\},\ P_j\in\{I,X,Y,Z\}\right\}.
\]
Key fact: any two Paulis either commute or anticommute.

\subsubsection{Definition of a stabilizer code}
A \emph{stabilizer group} \(S\subset \mathcal{P}_n\) is an abelian subgroup not containing \(-I\).
The \emph{codespace} is the joint \(+1\) eigenspace:
\[
\mathcal{C}=\left\{\ket{\psi}\in(\mathbb{C}^2)^{\otimes n} \;:\; g\ket{\psi}=\ket{\psi}\ \text{for all } g\in S\right\}.
\]

If \(S\) has \(m\) independent generators, then \(\dim \mathcal{C}=2^{n-m}\).
We write this as an \([[n,k]]\) stabilizer code with \(k=n-m\) logical qubits.

\subsubsection{Syndrome: how stabilizers detect errors}
Given an error \(E\in\mathcal{P}_n\) acting on a codeword \(\ket{\psi}\in\mathcal{C}\),
measure each generator \(g_i\in S\).
If \(E\) commutes with \(g_i\), the outcome is \(+1\); if \(E\) anticommutes with \(g_i\), the outcome is \(-1\).
Thus the syndrome is the bit-string
\[
s(E) = \big(s_1,\dots,s_m\big)\in\{0,1\}^m,\qquad
s_i=
\begin{cases}
	0,& Eg_i=g_iE\\
	1,& Eg_i=-g_iE.
\end{cases}
\]
This is the algebraic heart of stabilizer decoding: \emph{errors are inferred from commutation patterns}.

\subsection{Ancilla-mediated syndrome extraction}

\subsubsection{Measuring a Pauli stabilizer without measuring the data}
To measure a multi-qubit Pauli operator \(g\) (e.g.\ \(Z_1Z_2Z_3Z_4\)) without collapsing the data in the computational basis,
we use an ancilla qubit and controlled gates so that the ancilla accumulates the eigenvalue information.

\paragraph{Measuring \(Z\)-type stabilizers.}
To measure \(g=Z_{i_1}\cdots Z_{i_w}\):
\begin{enumerate}
	\item Prepare ancilla in \(\ket{0}\).
	\item Apply CNOTs with each data qubit \(i_j\) as control and the ancilla as target.
	\item Measure ancilla in \(Z\). The outcome bit encodes the parity (eigenvalue) of \(g\).
\end{enumerate}

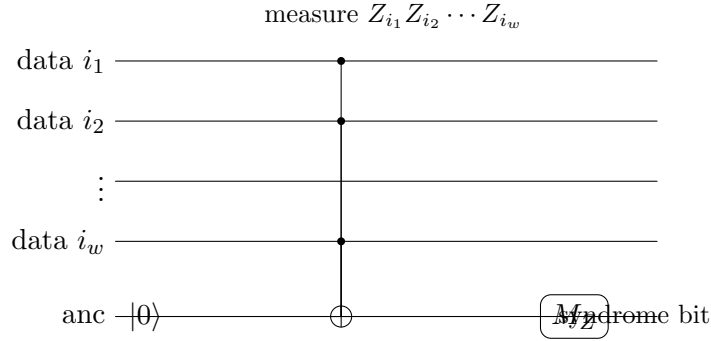
\begin{figure}[t]
	\centering
	\begin{tikzpicture}[scale=1.0, line cap=round, line join=round]
		\draw (0,0) -- (7.2,0); \node[left] at (0,0) {data $i_1$};
		\draw (0,-0.8) -- (7.2,-0.8); \node[left] at (0,-0.8) {data $i_2$};
		\draw (0,-1.6) -- (7.2,-1.6); \node[left] at (0,-1.6) {$\vdots$};
		\draw (0,-2.4) -- (7.2,-2.4); \node[left] at (0,-2.4) {data $i_w$};
		\draw (0,-3.4) -- (7.2,-3.4); \node[left] at (0,-3.4) {anc};
		
		\node at (0.4,-3.4) {$\ket{0}$};
		
		\foreach \y in {0,-0.8,-2.4}{
			\fill (3.0,\y) circle (1.5pt);
			\draw (3.0,\y) -- (3.0,-3.4);
		}
		\draw (3.0,-3.4) circle (0.14);
		\draw (2.86,-3.4) -- (3.14,-3.4);
		\draw (3.0,-3.54) -- (3.0,-3.26);
		
		\node[draw, rounded corners, minimum width=0.9cm, minimum height=0.5cm] at (6.1,-3.4) {$M_Z$};
		\node at (6.9,-3.4) {\small syndrome bit};
		
		\node[align=center] at (3.7,0.6) {\small measure $Z_{i_1}Z_{i_2}\cdots Z_{i_w}$};
	\end{tikzpicture}
	\caption{Ancilla-mediated measurement of a \(Z\)-type stabilizer via parity into the ancilla.}
	\label{fig:measure-z-stabilizer}
\end{figure}

\paragraph{Measuring \(X\)-type stabilizers.}
To measure \(g=X_{i_1}\cdots X_{i_w}\), conjugate by Hadamards because \(H X H = Z\):
\begin{enumerate}
	\item Apply \(H\) to each involved data qubit (or equivalently to the ancilla pattern, depending on implementation).
	\item Measure as a \(Z\)-type stabilizer parity (CNOT fan-in).
	\item Undo the basis change with \(H\).
\end{enumerate}

\begin{figure}[t]
	\centering
	\begin{tikzpicture}[scale=1.0, line cap=round, line join=round]
		\draw (0,0) -- (8.2,0); \node[left] at (0,0) {data $i_1$};
		\draw (0,-0.8) -- (8.2,-0.8); \node[left] at (0,-0.8) {data $i_2$};
		\draw (0,-1.6) -- (8.2,-1.6); \node[left] at (0,-1.6) {$\vdots$};
		\draw (0,-2.4) -- (8.2,-2.4); \node[left] at (0,-2.4) {data $i_w$};
		\draw (0,-3.4) -- (8.2,-3.4); \node[left] at (0,-3.4) {anc};
		
		\foreach \y in {0,-0.8,-2.4}{
			\node[draw, rounded corners, minimum width=0.6cm, minimum height=0.45cm] at (1.2,\y) {$H$};
		}
		
		\foreach \y in {0,-0.8,-2.4}{
			\fill (4.0,\y) circle (1.5pt);
			\draw (4.0,\y) -- (4.0,-3.4);
		}
		\draw (4.0,-3.4) circle (0.14);
		\draw (3.86,-3.4) -- (4.14,-3.4);
		\draw (4.0,-3.54) -- (4.0,-3.26);
		
		\foreach \y in {0,-0.8,-2.4}{
			\node[draw, rounded corners, minimum width=0.6cm, minimum height=0.45cm] at (6.3,\y) {$H$};
		}
		
		\node[draw, rounded corners, minimum width=0.9cm, minimum height=0.5cm] at (7.5,-3.4) {$M_Z$};
		
		\node[align=center] at (4.3,0.7) {\small measure $X_{i_1}X_{i_2}\cdots X_{i_w}$ via basis change};
	\end{tikzpicture}
	\caption{Measuring an \(X\)-type stabilizer by conjugating it into a \(Z\)-type parity measurement.}
	\label{fig:measure-x-stabilizer}
\end{figure}
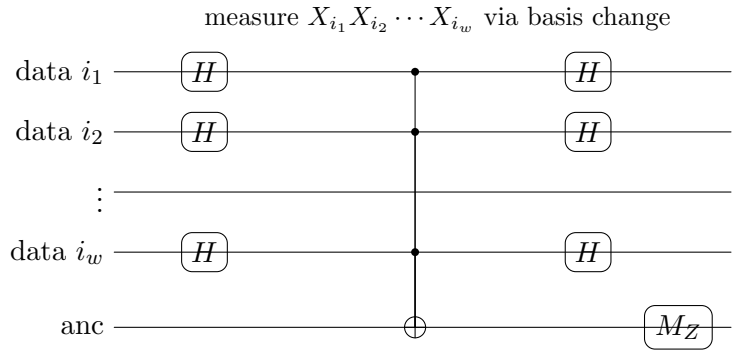

\subsection{Worked example: Bell state as a stabilizer code}

\subsubsection{The Bell state and its stabilizers}
Consider
\[
\ket{\Phi^+}=\frac{1}{\sqrt2}\bigl(\ket{00}+\ket{11}\bigr).
\]
It is stabilized by the commuting operators
\[
g_1 = Z\otimes Z,\qquad g_2 = X\otimes X,
\]
meaning
\[
(Z\otimes Z)\ket{\Phi^+}=\ket{\Phi^+},\qquad (X\otimes X)\ket{\Phi^+}=\ket{\Phi^+}.
\]
Thus \(\ket{\Phi^+}\) is the unique joint \(+1\) eigenstate of \(\langle ZZ, XX\rangle\).
This is an \([[2,0]]\) stabilizer code: it encodes \(k=0\) logical qubits (a single protected state),
but it is still extremely instructive as an \emph{error-detecting} gadget.

\subsubsection{Error detection table (single-qubit Pauli errors)}
Let \(E\) be a single-qubit Pauli on qubit 1. The syndrome is determined by commutation with \(ZZ\) and \(XX\):

\[
\begin{array}{c|cc}
	E\ \text{on qubit 1} & \text{with }ZZ & \text{with }XX\\
	\hline
	X\otimes I & \text{anticommutes} & \text{commutes}\\
	Z\otimes I & \text{commutes} & \text{anticommutes}\\
	Y\otimes I & \text{anticommutes} & \text{anticommutes}\\
\end{array}
\]
So the syndrome bits \((s_{ZZ}, s_{XX})\) are:
\[
X\otimes I \mapsto (1,0),\quad
Z\otimes I \mapsto (0,1),\quad
Y\otimes I \mapsto (1,1).
\]
Similarly for errors on qubit 2 (same pattern).

\subsubsection{Preparing \(\ket{\Phi^+}\) and verifying stabilizers in-circuit}
A standard preparation circuit is \(H\) then CNOT:
\[
\ket{00}\xrightarrow{H\otimes I}\frac{1}{\sqrt2}(\ket{00}+\ket{10})
\xrightarrow{\mathrm{CNOT}}\frac{1}{\sqrt2}(\ket{00}+\ket{11})=\ket{\Phi^+}.
\]

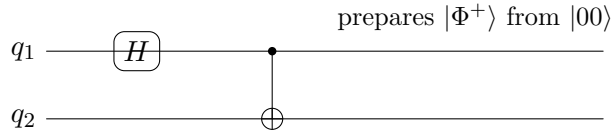
\begin{figure}[t]
	\centering
	\begin{tikzpicture}[scale=1.0, line cap=round, line join=round]
		\draw (0,0) -- (7.4,0); \node[left] at (0,0) {$q_1$};
		\draw (0,-0.9) -- (7.4,-0.9); \node[left] at (0,-0.9) {$q_2$};
		
		\node[draw, rounded corners, minimum width=0.6cm, minimum height=0.45cm] at (1.2,0) {$H$};
		\fill (3.0,0) circle (1.6pt);
		\draw (3.0,0) -- (3.0,-0.9);
		\draw (3.0,-0.9) circle (0.14);
		\draw (2.86,-0.9) -- (3.14,-0.9);
		\draw (3.0,-1.04) -- (3.0,-0.76);
		
		\node[align=center] at (5.7,0.45) {\small prepares $\ket{\Phi^+}$ from $\ket{00}$};
	\end{tikzpicture}
	\caption{Bell-pair preparation circuit.}
	\label{fig:bell-prep}
\end{figure}

To \emph{measure} the stabilizers \(ZZ\) and \(XX\), use the ancilla circuits from Figures
\ref{fig:measure-z-stabilizer} and \ref{fig:measure-x-stabilizer} with \(w=2\).

\subsection{Surface code geometry: errors as chains}

\subsubsection{The geometric picture}
Surface codes live on a 2D lattice of qubits with local parity checks.
The essential idea is:

\begin{quote}
	\emph{local checks detect the endpoints of an error chain; decoding is finding a likely chain consistent with those endpoints.}
\end{quote}

A common abstraction uses:
\begin{itemize}
	\item \textbf{Data qubits} on edges of a square lattice,
	\item \textbf{\(Z\)-type (plaquette) stabilizers} on faces: product of \(Z\) on the four surrounding edges,
	\item \textbf{\(X\)-type (star) stabilizers} on vertices: product of \(X\) on the four incident edges.
\end{itemize}

\subsubsection{Chains and endpoints (a minimal example)}
A \(Z\)-error on a data qubit anticommutes with neighboring \(X\)-type checks,
so it flips the outcomes of the adjacent star checks. If multiple \(Z\) errors occur along a path,
interior flips cancel and only the endpoints remain as syndrome defects.

\begin{figure}[t]
	\centering
	\begin{tikzpicture}[scale=1.0, line cap=round, line join=round]
		\foreach \x in {0,1,2,3}{
			\foreach \y in {0,1,2,3}{
				\fill (\x,\y) circle (1.2pt);
			}
		}
		\foreach \x in {0,1,2}{
			\foreach \y in {0,1,2,3}{
				\draw (\x,\y) -- (\x+1,\y);
			}
		}
		\foreach \x in {0,1,2,3}{
			\foreach \y in {0,1,2}{
				\draw (\x,\y) -- (\x,\y+1);
			}
		}
		
		\draw[very thick] (0,1) -- (1,1) -- (2,1) -- (2,2) -- (2,3);
		
		\node at (0,1.35) {\Large $\ast$};
		\node at (2.35,3.0) {\Large $\ast$};
		
		\node[align=center] at (1.5,-0.8)
		{\small Thick path = a chain of \(Z\) errors on data edges.\\
			\small Syndromes appear at endpoints (defects).};
	\end{tikzpicture}
	\caption{Surface-code intuition: a chain of \(Z\) errors produces syndrome defects at its endpoints.}
	\label{fig:surface-chain}
\end{figure}

\subsubsection{Decoding as matching (high-level)}
Given many defects (endpoints), a decoder tries to pair them with likely paths.
For i.i.d.\ noise, “likely” means short (or low-weight) paths, but real devices use
weighted edges based on calibration data, crosstalk, leakage, and measurement error rates.

\subsection{Homological interpretation}

\subsubsection{Boundary operator viewpoint}
The phrase “errors are chains” can be made precise using a chain complex.
Informally (for a square lattice):
\begin{itemize}
	\item \(C_1\): formal sums of edges (error chains),
	\item \(C_0\): formal sums of vertices (defect locations),
	\item \(\partial: C_1\to C_0\): boundary map sending an edge to its two endpoints.
\end{itemize}
An error chain \(e\in C_1\) has syndrome \(\partial e\in C_0\), its set of endpoints.
Thus syndrome extraction computes a boundary.

\subsubsection{Why “topology explains protection”}
If two different chains \(e\) and \(e'\) have the same boundary, then \(\partial(e-e')=0\),
so \(e-e'\) is a cycle. Cycles that are themselves boundaries correspond to “trivial” differences
that do not change logical information; nontrivial cycles (representing nonzero homology classes)
correspond to \emph{logical operators} that implement a nontrivial transformation on encoded qubits.
That is the core protection statement:

\begin{quote}
	\emph{local measurements cannot distinguish chains differing by a homologically nontrivial cycle.}
\end{quote}

In practice, code distance is the minimal weight of a nontrivial cycle:
it is the minimum number of physical errors needed to implement a logical error.

\subsection{Hardware perspective and FPGA relevance}

\subsubsection{Where the classical workload sits}
Each QEC cycle produces classical bits:
\begin{itemize}
	\item syndrome bits from stabilizer measurements,
	\item possibly flags for leakage, ancilla faults, or measurement confidence,
	\item timing/alignment metadata (cycle index, shot id).
\end{itemize}
The classical side must do (often in real time):
\begin{enumerate}
	\item \textbf{Syndrome processing:} compute differences between rounds (to handle measurement noise),
	\item \textbf{Decoding:} map syndromes to corrections (or Pauli frame updates),
	\item \textbf{Feedback/control:} schedule next operations, apply conditional resets, update parameters.
\end{enumerate}

\subsubsection{Why FPGA is a natural fit}
FPGA advantages match the QEC loop structure:
\begin{itemize}
	\item \textbf{Deterministic low latency:} stable cycle-to-cycle timing (jitter control),
	\item \textbf{Streaming bit-level compute:} parity, XOR networks, small-state machines,
	\item \textbf{Massive parallelism:} many checks updated concurrently,
	\item \textbf{Custom datapaths:} decode kernels (e.g.\ union-find, MWPM variants, lookup-table microdecoders).
\end{itemize}
In surface-code-scale systems, the decoder and the data-movement fabric (syndrome routing, buffering, time alignment)
become first-class architecture problems. This is precisely where FPGA/ASIC control hardware appears.

\subsection{Exercises}

\begin{exercise}[Pauli commutation and syndrome bits]
	Let \(g=Z_1Z_2Z_3Z_4\). For each single-qubit Pauli error \(E=X_j\), \(Z_j\), \(Y_j\) acting on qubit \(j\in\{1,2,3,4\}\),
	determine whether \(E\) commutes or anticommutes with \(g\).
\end{exercise}
\noindent\textbf{Solution.}
Use: \(X\) anticommutes with \(Z\), \(Y\) anticommutes with \(Z\), and \(Z\) commutes with \(Z\).
Also tensor factors on different qubits commute.
Thus for \(E=X_j\) or \(E=Y_j\), the \(j\)-th factor anticommutes with the corresponding \(Z_j\) in \(g\),
so \(Eg=-gE\) (anticommute). For \(E=Z_j\), all factors commute, so \(Eg=gE\).

\begin{exercise}[Bell stabilizers]
	Verify directly that \((Z\otimes Z)\ket{\Phi^+}=\ket{\Phi^+}\) and \((X\otimes X)\ket{\Phi^+}=\ket{\Phi^+}\).
\end{exercise}
\noindent\textbf{Solution.}
Compute
\[
(Z\otimes Z)\ket{00}=\ket{00},\qquad (Z\otimes Z)\ket{11}=( -\ket{1})\otimes(-\ket{1})=\ket{11},
\]
so \((Z\otimes Z)\ket{\Phi^+}=\ket{\Phi^+}\).
Also
\[
(X\otimes X)\ket{00}=\ket{11},\qquad (X\otimes X)\ket{11}=\ket{00},
\]
so
\[
(X\otimes X)\ket{\Phi^+}=\frac{1}{\sqrt2}(\ket{11}+\ket{00})=\ket{\Phi^+}.
\]

\begin{exercise}[Syndrome of single-qubit errors on a Bell pair]
	Let \(S=\langle ZZ,XX\rangle\) and \(\ket{\Phi^+}\) be the codespace state.
	Compute the syndrome bits \((s_{ZZ},s_{XX})\) for \(E=X\otimes I\), \(Z\otimes I\), and \(Y\otimes I\).
\end{exercise}
\noindent\textbf{Solution.}
Compute commutation:
\[
(X\otimes I)(Z\otimes Z)=-(Z\otimes Z)(X\otimes I)\quad\Rightarrow\quad s_{ZZ}=1,
\]
\[
(X\otimes I)(X\otimes X)=(X\otimes X)(X\otimes I)\quad\Rightarrow\quad s_{XX}=0.
\]
So \(X\otimes I\mapsto (1,0)\).
Similarly,
\[
(Z\otimes I)\ \text{commutes with }ZZ \Rightarrow s_{ZZ}=0,\qquad
(Z\otimes I)\ \text{anticommutes with }XX \Rightarrow s_{XX}=1,
\]
so \(Z\otimes I\mapsto (0,1)\).
Finally \(Y=iXZ\) anticommutes with both \(ZZ\) and \(XX\), hence \(Y\otimes I\mapsto(1,1)\).

\begin{exercise}[Why ancilla parity measurement works (two-qubit \(ZZ\))]
	Consider data qubits in a joint eigenstate of \(Z\otimes Z\) with eigenvalue \(\lambda\in\{+1,-1\}\).
	Show that the ancilla circuit “CNOT from each data qubit into ancilla, then measure ancilla in \(Z\)”
	returns a classical bit encoding \(\lambda\).
\end{exercise}
\noindent\textbf{Solution.}
Write the data in the computational basis:
\(\ket{00},\ket{01},\ket{10},\ket{11}\).
The operator \(Z\otimes Z\) has eigenvalue \(+1\) on even parity states \(\ket{00},\ket{11}\)
and eigenvalue \(-1\) on odd parity states \(\ket{01},\ket{10}\).
The two CNOTs compute parity into the ancilla:
starting ancilla \(\ket{0}\), it flips once for each data qubit in state \(\ket{1}\),
so it ends in \(\ket{0}\) for even parity and \(\ket{1}\) for odd parity.
Thus ancilla measurement reveals the \(ZZ\) eigenvalue without directly measuring data.

\begin{exercise}[Chains and endpoints]
	In Figure~\ref{fig:surface-chain}, interpret the thick path as an element of \(C_1\).
	Explain why only its endpoints are detected by local checks (syndrome bits), not the interior.
\end{exercise}
\noindent\textbf{Solution.}
Each individual \(Z\) error flips the two neighboring \(X\)-type checks adjacent to that edge.
Along a chain, interior vertices touch two error edges, so the corresponding check is flipped twice:
\((-1)\cdot(-1)=+1\), canceling out. Only vertices incident to an odd number (here, one) of error edges flip,
i.e.\ the boundary \(\partial e\), the endpoints.

\begin{exercise}[FPGA-friendly microtasks]
	List three decoder-adjacent tasks that are naturally streaming/bitwise and thus FPGA-friendly.
	For each, state the input and output bitwidth at a high level.
\end{exercise}
\noindent\textbf{Solution.}
Examples:
\begin{itemize}
	\item Syndrome differencing across rounds: input \(m\) syndrome bits for round \(t\) and \(m\) bits for round \(t-1\),
	output \(m\) bits via XOR.
	\item Defect extraction: input \(m\) syndrome bits, output a sparse list or bitmask of defect locations (same \(m\) bits or compressed indices).
	\item Local rule / micro-decoder stage: input a neighborhood window of syndrome bits (e.g.\ \(O(10)\)–\(O(100)\) bits),
	output correction bits for a subset of edges (few to tens of bits) or updates to a Pauli frame record.
\end{itemize}
All are low-precision, high-throughput, fixed-latency logic patterns well-suited to FPGA pipelines.
	
\section{Decoder Requirements and Metrics (Real-Time, Correctness, and Product Knobs)}
\label{sec:req-metrics}

\subsection{Objective and central claim}
This section translates ``decoder quality'' into \emph{engineering requirements} that you can
measure, test, and negotiate as product knobs.

\medskip
\noindent\textbf{Central claim.}
A QEC decoder is not just an algorithm; it is a \emph{real-time service}.
Therefore, the only meaningful notion of ``good'' is a \emph{joint} statement about:
\[
\text{(deadline stability)}\ \wedge\ \text{(tail latency)}\ \wedge\ \text{(throughput under bursts)}\ \wedge\ \text{(correctness under a noise model)}.
\]
You do not ship “lowest average latency” if p999 misses deadlines.
You do not ship “highest accuracy” if it creates backlog growth.
You do not claim “works on distance \(d\)” without stating the cadence and the time-window policy.

\medskip
\noindent\textbf{Running mental model.}
Each QEC round (cycle) produces a syndrome slice \(s_t\).
The decoder ingests \(s_t\) (often together with a window \(s_{t-W+1:t}\)),
and must output a correction decision \(c_t\) (or a Pauli-frame update) before a deadline.
If the decoder ever falls behind (backlog), the system becomes unstable.

\subsection{What ``real-time'' means: deadlines, cadence, and backlog stability}

\subsubsection{Cadence and deadlines}
Let:
\begin{itemize}
	\item \(T_{\mathrm{cycle}}\): QEC cycle period (cadence), e.g.\ microseconds to milliseconds depending on platform,
	\item \(D\): decoding deadline (often \(D \le T_{\mathrm{cycle}}\) after accounting for IO and actuation),
	\item \(L_t\): decoder latency for round \(t\) (end-to-end: ingest \(\to\) compute \(\to\) output visible to control),
	\item \(A_t\): arrival time of syndrome slice \(s_t\).
\end{itemize}
The real-time requirement is:
\[
\forall t,\qquad A_t + L_t \le A_t + D.
\]
Equivalently, \(L_t \le D\) for all \(t\).

\subsubsection{Backlog and stability: the queueing statement}
If a decoder is a server and syndrome slices are jobs, backlog grows when the server falls behind.
A minimal discrete-time backlog model:
\[
B_{t+1} = \max\{0,\ B_t + L_t - T_{\mathrm{cycle}}\}.
\]
Here \(B_t\) is “how far behind schedule” you are just before round \(t\) arrives.

\medskip
\noindent\textbf{Stability condition (practical).}
Even if \(\mathbb{E}[L_t] < T_{\mathrm{cycle}}\), occasional spikes can make \(B_t\) drift upward.
So you must care about \emph{tails} and \emph{burst behavior}, not just means.

\subsubsection{The product definition of real-time}
In a product spec, ``real-time'' should be written as a concrete SLA:
\begin{itemize}
	\item \textbf{Cadence:} \(T_{\mathrm{cycle}}=\) \underline{\hspace{2.5cm}}.
	\item \textbf{Deadline:} \(D=\) \underline{\hspace{2.5cm}}.
	\item \textbf{SLA:} \( \Pr[L \le D] \ge 1-\epsilon \) with \(\epsilon=\) \underline{\hspace{1.5cm}} (e.g.\ \(10^{-6}\)).
	\item \textbf{Backlog:} over a run of \(N\) rounds, maximum backlog \(B_{\max}\le\) \underline{\hspace{2cm}}.
\end{itemize}

\subsection{Latency metrics: mean vs.\ p99/p999, jitter, and spikes}

\subsubsection{Core latency metrics}
Given measured latencies \(L_1,\dots,L_N\):
\begin{itemize}
	\item \(\mu = \frac{1}{N}\sum_{t=1}^N L_t\) (mean),
	\item \(p99 = \text{99th percentile of } L\),
	\item \(p999 = \text{99.9th percentile of } L\),
	\item \(L_{\max} = \max_t L_t\).
\end{itemize}

\subsubsection{Jitter (cycle-to-cycle variability)}
Two practical jitter measures:
\[
J_{\mathrm{abs}} := \mathrm{std}(L_t),
\qquad
J_{\mathrm{diff}} := \mathrm{std}(L_{t+1}-L_t).
\]
The first measures overall spread; the second measures rapid variability that tends to break pipelining.

\subsubsection{Spike taxonomy (what you log)}
When you see p999 failures, they are almost always attributable to a small set of causes:
\begin{itemize}
	\item cache misses / memory stalls,
	\item buffer contention / burst arrivals,
	\item branch-heavy code paths (data-dependent),
	\item OS scheduling / interrupts (CPU),
	\item PCIe / network jitter (host-device),
	\item garbage collection / allocator spikes (managed runtimes),
	\item thermal throttling (rare, but real).
\end{itemize}
Your benchmark harness should label latency samples with a \emph{reason code}
(or at least enough telemetry to infer one).

\subsection{Throughput metrics: sustained rate, burst rate, and buffering}

\subsubsection{Definitions}
Let:
\begin{itemize}
	\item \(R_{\mathrm{in}} = 1/T_{\mathrm{cycle}}\): required sustained input rate (rounds/sec),
	\item \(R_{\mathrm{svc}} = 1/\mu\): service rate implied by mean latency (rounds/sec),
	\item \(B_{\mathrm{buf}}\): buffer capacity (in rounds or in bytes),
	\item \(R_{\mathrm{burst}}\): burst arrival rate when the system delivers multiple slices quickly (common with batching / transport).
\end{itemize}

\subsubsection{Sustained throughput requirement}
For stability, you need at least:
\[
R_{\mathrm{svc}} > R_{\mathrm{in}}
\quad\Longleftrightarrow\quad
\mu < T_{\mathrm{cycle}}.
\]
But to survive bursts and tails, you typically need a safety margin:
\[
\mu \le \alpha\, T_{\mathrm{cycle}}
\quad\text{with}\quad \alpha \in [0.3,0.7]\ \text{(platform dependent)}.
\]
The remaining budget goes to IO, control, and contingency for spikes.

\subsubsection{Burst handling and buffer sizing}
If bursts of size \(K\) rounds can arrive within a short transport interval,
your buffer must absorb them:
\[
B_{\mathrm{buf}} \ge K + \text{(worst-case temporary backlog in compute)}.
\]
A practical buffer spec is written using a \emph{burst profile}:
\begin{itemize}
	\item maximum burst size \(K_{\max}\),
	\item maximum burst frequency (how often bursts occur),
	\item maximum gap between bursts.
\end{itemize}

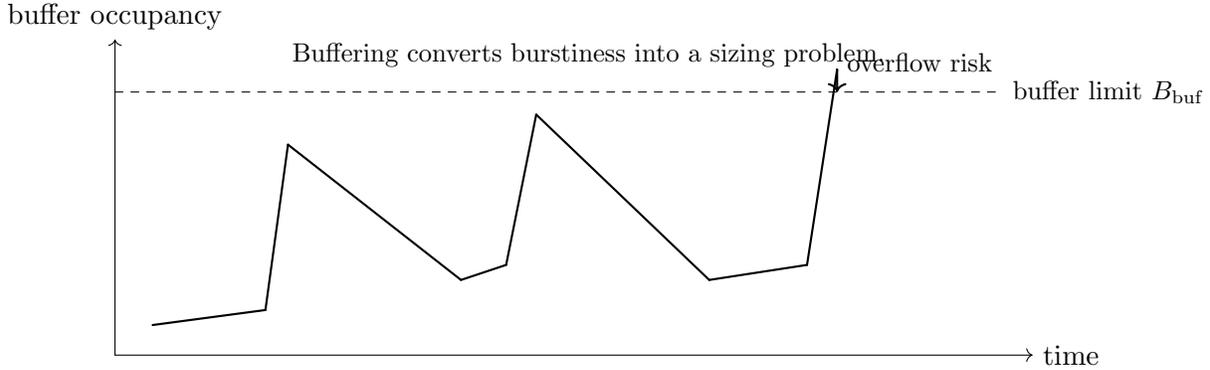
\begin{figure}[t]
	\centering
	\begin{tikzpicture}[scale=1.0, line cap=round, line join=round]
		\draw[->] (0,0) -- (12.2,0) node[right] {time};
		\draw[->] (0,0) -- (0,4.2) node[above] {buffer occupancy};
		
		\draw[dashed] (0,3.5) -- (11.8,3.5);
		\node[right] at (11.8,3.5) {\small buffer limit $B_{\mathrm{buf}}$};
		
		\draw[thick] (0.5,0.4) -- (2.0,0.6);
		\draw[thick] (2.0,0.6) -- (2.3,2.8); 
		\draw[thick] (2.3,2.8) -- (4.6,1.0); 
		\draw[thick] (4.6,1.0) -- (5.2,1.2);
		\draw[thick] (5.2,1.2) -- (5.6,3.2); 
		\draw[thick] (5.6,3.2) -- (7.9,1.0); 
		\draw[thick] (7.9,1.0) -- (9.2,1.2);
		\draw[thick] (9.2,1.2) -- (9.6,3.8); 
		
		\draw[->, thick] (9.6,3.8) -- (9.6,3.5);
		\node[right] at (9.6,3.9) {\small overflow risk};
		
		\node[align=center] at (6.3,4.0)
		{\small Buffering converts burstiness into a sizing problem.};
	\end{tikzpicture}
	\caption{Buffer occupancy under bursts: even with adequate mean throughput, insufficient buffering can overflow during clustered arrivals.}
	\label{fig:buffer-bursts}
\end{figure}

\subsection{Correctness metrics: success criteria and failure modes}

\subsubsection{What correctness means for a decoder}
Decoders do not need to output the exact physical error \(E\); they need to output a correction
\(\widehat{E}\) that produces the same syndrome and differs from the true error only by a stabilizer:
\[
\widehat{E}E \in S
\quad\text{(success condition)}.
\]
Equivalently: the net effect on the logical state is trivial.

\subsubsection{Decoder correctness metrics}
Typical correctness metrics depend on your test model and whether you simulate logical qubits:
\begin{itemize}
	\item \textbf{Logical failure rate:} \(p_L(d)\) per round (or per unit time) at code distance \(d\).
	\item \textbf{Failure per shot:} probability that a full experiment run produces a logical error.
	\item \textbf{Decoder mismatch rate:} fraction of rounds where \(\widehat{E}\) is inconsistent with syndrome constraints
	(often indicates implementation bug).
	\item \textbf{Time-correlated failure:} failure probability conditioned on recent events (bursts, measurement errors, leakage flags).
\end{itemize}

\subsubsection{Failure modes you must name in the spec}
A product-quality decoder spec should explicitly list failure modes:
\begin{itemize}
	\item \textbf{Deadline miss:} correction arrives too late (even if logically correct).
	\item \textbf{Backlog collapse:} sustained overload causes cascading misses.
	\item \textbf{Model mismatch:} decoder assumes i.i.d.\ Pauli noise but device has strong bias, drift, or correlations.
	\item \textbf{Measurement errors:} syndrome bits wrong; requires time-differencing or 3D decoding.
	\item \textbf{Leakage / non-Pauli events:} error model invalidates assumptions unless flagged and handled.
	\item \textbf{Implementation bugs:} wrong boundary conditions, indexing, windowing, or version mismatch.
\end{itemize}

\subsection{Scaling knobs: code distance, window size, and policy variants}

\subsubsection{Knob 1: code distance \(d\)}
Distance \(d\) controls robustness but increases workload:
\begin{itemize}
	\item number of data qubits \(\sim O(d^2)\),
	\item number of stabilizers \(\sim O(d^2)\),
	\item number of syndrome bits per round \(\sim O(d^2)\),
	\item decoding complexity depends on algorithm (often near-linear in number of defects for union-find, etc.).
\end{itemize}

\subsubsection{Knob 2: decoding window size \(W\)}
If measurement errors exist, decoders often use multiple rounds (a spacetime graph).
Window size \(W\) trades:
\[
\text{better handling of measurement noise}
\quad \leftrightarrow \quad
\text{higher latency/memory/bandwidth}.
\]
In real-time systems, a key product knob is whether you do:
\begin{itemize}
	\item \textbf{strict streaming:} fixed small \(W\), constant-time updates,
	\item \textbf{micro-batching:} accumulate \(W\) rounds then decode (better accuracy, worse latency),
	\item \textbf{hybrid:} streaming baseline with periodic “repair passes.”
\end{itemize}

\subsubsection{Knob 3: policy variants}
Decoding is full of policy choices that affect speed/quality:
\begin{itemize}
	\item boundary conditions (planar vs.\ toric),
	\item weighting model (uniform vs.\ calibrated weights),
	\item tie-breaking (deterministic vs.\ randomized),
	\item post-processing passes (local cleanups),
	\item leakage-aware rules (extra flags and constraints).
\end{itemize}
A good benchmark suite treats these as named profiles, not hidden code paths.

\subsection{Measurement plan: traces, logs, and reproducible benchmarks}

\subsubsection{The minimal trace you should collect}
To make results reproducible and debuggable, record:
\begin{itemize}
	\item \textbf{Configuration:} code family, \(d\), window \(W\), noise model parameters, seed,
	\item \textbf{IO sizes:} number of syndrome bits per round, metadata format/version,
	\item \textbf{Timing:} timestamps for ingest, compute start/end, output commit,
	\item \textbf{Latency sample:} \(L_t\) per round with percentiles computed after run,
	\item \textbf{Backlog/buffer:} occupancy over time and overflow events,
	\item \textbf{Correctness:} logical failure events with context (recent syndromes, flags).
\end{itemize}

\subsubsection{Determinism rules for benchmarks}
For a benchmark to be “scientific enough,” enforce:
\begin{itemize}
	\item fixed RNG seeds (with explicit reporting),
	\item pinned CPU frequency or dedicated FPGA clock domain,
	\item no hidden warm-up effects (separate warm-up from measurement),
	\item logs include build hash / git commit and compiler flags,
	\item fixed input trace files archived alongside results.
\end{itemize}

\subsection{Exercises (design-a-metric)}

\begin{exercise}[SLA writing]
	Write an SLA for a decoder that must run at \(T_{\mathrm{cycle}}=1\ \mu s\) with deadline \(D=0.7\ \mu s\).
	Specify a tail-latency requirement using p999 and a backlog requirement over \(N=10^9\) rounds.
\end{exercise}
\noindent\textbf{Solution sketch.}
A concrete SLA could be:
\[
p999(L) \le 0.7\ \mu s,\qquad \Pr[L > 0.7\ \mu s] \le 10^{-6}.
\]
Backlog requirement: with discrete update \(B_{t+1}=\max(0,B_t+L_t-T_{\mathrm{cycle}})\),
require \(B_{\max}\le 5\ \mu s\) (i.e.\ at most 5 cycles behind) over \(10^9\) rounds, with zero buffer overflows.

\begin{exercise}[Latency budget decomposition]
	Assume the end-to-end path is:
	\[
	\text{readout decode} (t_{\mathrm{io}})\ +\ \text{syndrome preprocess} (t_{\mathrm{prep}})\ +\ \text{decode core} (t_{\mathrm{core}})\ +\ \text{commit/output} (t_{\mathrm{out}}).
	\]
	For deadline \(D\), propose a budget split (percentages) that accounts for spikes in \(t_{\mathrm{core}}\).
\end{exercise}
\noindent\textbf{Solution sketch.}
A common pattern is to allocate conservative headroom to the core:
e.g.\ \(t_{\mathrm{io}}:20\%\), \(t_{\mathrm{prep}}:10\%\), \(t_{\mathrm{core}}:50\%\), \(t_{\mathrm{out}}:10\%\),
leaving \(10\%\) explicit slack. If \(t_{\mathrm{core}}\) is spiky, reduce its nominal budget
and reserve slack for p99/p999 events (or pipeline the core to flatten variability).

\begin{exercise}[Correctness metric choice]
	You have two decoder variants:
	A has lower logical failure rate but occasional deadline misses.
	B has slightly worse logical failure rate but no deadline misses.
	Define a single scalar score that reflects a product requirement: ``never miss deadlines'' is strict.
\end{exercise}
\noindent\textbf{Solution sketch.}
Use a hard constraint: if any deadline miss occurs, score is zero (or marked FAIL).
Otherwise score by logical failure rate, e.g.\ \(\text{Score}=-\log_{10}(p_L)\).
This reflects that real-time misses are catastrophic in closed-loop control.

\begin{exercise}[Burst profile]
	Create a burst profile for syndrome delivery that could arise from batching:
	specify \(K_{\max}\), burst interval, and required buffer \(B_{\mathrm{buf}}\).
	Then explain how you would measure it in a trace.
\end{exercise}
\noindent\textbf{Solution sketch.}
Example: bursts of up to \(K_{\max}=64\) rounds arriving within \(5\ \mu s\), at most once per \(1\ ms\).
Buffer should exceed \(64\) plus worst-case compute backlog, e.g.\ \(B_{\mathrm{buf}}=128\) rounds.
Measure by logging arrival timestamps and computing inter-arrival gaps; detect clusters and count their sizes.

\subsection{Templates (copy-paste)}

\subsubsection{Latency-budget worksheet}
\begin{quote}
	\textbf{Latency budget worksheet (decoder service)}\\
	\begin{tabular}{@{}p{0.42\linewidth}p{0.25\linewidth}p{0.25\linewidth}@{}}
		\textbf{Item} & \textbf{Target (mean)} & \textbf{Target (p999)}\\
		\hline
		Cycle period \(T_{\mathrm{cycle}}\) & \underline{\hspace{3cm}} & ---\\
		Deadline \(D\) & \underline{\hspace{3cm}} & ---\\
		IO ingest \(t_{\mathrm{io}}\) & \underline{\hspace{3cm}} & \underline{\hspace{3cm}}\\
		Preprocess \(t_{\mathrm{prep}}\) & \underline{\hspace{3cm}} & \underline{\hspace{3cm}}\\
		Decode core \(t_{\mathrm{core}}\) & \underline{\hspace{3cm}} & \underline{\hspace{3cm}}\\
		Commit/output \(t_{\mathrm{out}}\) & \underline{\hspace{3cm}} & \underline{\hspace{3cm}}\\
		Slack / margin & \underline{\hspace{3cm}} & \underline{\hspace{3cm}}\\
		\hline
		\textbf{Total} & \underline{\hspace{3cm}} & \underline{\hspace{3cm}}\\
	\end{tabular}
	
	\medskip
	\textbf{Pass/Fail rules}\\
	(1) \(p999(L)\le D\). \quad
	(2) Over run length \(N=\underline{\hspace{1.5cm}}\), backlog \(B_{\max}\le \underline{\hspace{1.5cm}}\). \quad
	(3) Buffer overflow count \(=0\).
\end{quote}

\subsubsection{Benchmark report skeleton}
\begin{quote}
	\textbf{Benchmark report (decoder)}\\
	\textbf{Build:} git=\underline{\hspace{3cm}}, flags=\underline{\hspace{3cm}}, platform=\underline{\hspace{3cm}}\\
	\textbf{Config:} code=\underline{\hspace{2cm}}, distance \(d=\underline{\hspace{1cm}}\), window \(W=\underline{\hspace{1cm}}\),
	noise=\underline{\hspace{2.5cm}}, seed=\underline{\hspace{1.5cm}}\\
	\textbf{Cadence:} \(T_{\mathrm{cycle}}=\underline{\hspace{1.5cm}}\), deadline \(D=\underline{\hspace{1.5cm}}\),
	run length \(N=\underline{\hspace{2cm}}\) rounds\\
	\textbf{Latency:} mean=\underline{\hspace{1.5cm}}, p99=\underline{\hspace{1.5cm}}, p999=\underline{\hspace{1.5cm}}, max=\underline{\hspace{1.5cm}}\\
	\textbf{Throughput/Buffer:} sustained rate=\underline{\hspace{1.5cm}}, burst profile=\underline{\hspace{3cm}},
	max occupancy=\underline{\hspace{1.5cm}}, overflows=\underline{\hspace{1cm}}\\
	\textbf{Correctness:} logical failure rate \(p_L=\underline{\hspace{1.5cm}}\),
	deadline misses=\underline{\hspace{1cm}}, notes=\underline{\hspace{4cm}}\\
	\textbf{Artifacts:} trace file=\underline{\hspace{3cm}}, plots=\underline{\hspace{3cm}}, raw logs=\underline{\hspace{3cm}}
\end{quote}
	

\section{Algorithm VI: Surface-Code Decoding (\texorpdfstring{Topology $\to$ Graph $\to$ Algorithm}{Topology -> Graph -> Algorithm})}
\label{sec:surface-decode}

\newcommand{\EMDot}{1.3pt}
\newcommand{\EMNodeDot}{1.6pt}
\newcommand{\EMSyndromeR}{4.2pt}

\subsection{Objective (what we decode and why it matters)}

Surface codes turn \emph{physical Pauli errors} into \emph{geometric objects} on a lattice.
You do not directly observe the errors themselves; you observe \emph{syndrome bits} that identify where the errors
\emph{begin and end} (their boundary).
Decoding is the online task of turning a stream of syndrome bits into a correction decision (often a Pauli-frame update)
fast enough to keep pace with the quantum error-correction (QEC) cycle.

\medskip
\noindent\textbf{Decoding problem in one sentence.}
Given syndrome outcomes over space (and usually time), output a correction chain whose boundary matches the syndrome,
minimizing a cost derived from the noise model, while avoiding logical errors.

\medskip
\noindent\textbf{Why it matters.}
Surface codes are attractive because checks are local (nearest-neighbor circuits) and scalable.
But the code protects logical information only if:
(i) the decoder keeps up with the cadence, and
(ii) the correction does not accidentally implement a logical operator.

\subsection{Key objects I: lattice, checks, and the measured data}

\subsubsection{A minimal planar patch picture}

We describe a planar surface-code patch (a square-ish region).
There are two dual sides:
\begin{itemize}
	\item \(Z\)-type data errors are detected by \(X\)-type checks, and
	\item \(X\)-type data errors are detected by \(Z\)-type checks.
\end{itemize}
We present the \(Z\)-error/\(X\)-check side explicitly; the other side follows by exchanging \(X \leftrightarrow Z\)
and primal/dual conventions.

\begin{itemize}
	\item \textbf{Data qubits} live on \emph{edges} of a square lattice.
	\item \textbf{\(X\)-type checks} (often called ``stars'') live on \emph{vertices}.
	Each check touches the data-qubit edges incident to that vertex.
	\item \textbf{Measured data per round} is one bit per check, typically recorded as \(\pm 1\) or \(0/1\).
\end{itemize}

\subsubsection{What is a syndrome bit?}

Let \(S_v\) be the \(X\)-type stabilizer at vertex \(v\).
In the ideal, noiseless-measurement abstraction:
\[
S_v \;=\; \prod_{e\ni v} X_e,
\]
where the product ranges over data-qubit edges \(e\) incident to \(v\).
A \(Z\)-error on an edge \(e=(u,v)\) anticommutes with \(X_e\) and therefore flips the outcomes of the two
adjacent \(X\)-checks \(S_u\) and \(S_v\).
Thus, the set of flipped checks marks the \emph{endpoints} of the \(Z\)-error chain.

\medskip
\noindent\textbf{Hardware reality.}
You typically do not measure \(S_v\) directly.
You entangle neighboring data qubits with an ancilla (via CNOT patterns) and measure the ancilla.
This produces a streaming syndrome bit per check per cycle.

\subsection{Key objects II: errors as chains, syndrome as boundary}

\subsubsection{Errors as 1-chains over \(\mathbb{Z}_2\)}

Let \(E\) be the set of data-qubit edges that suffered a \(Z\)-error in a given cycle.
Treat \(E\) as a 1-chain with \(\mathbb{Z}_2\) coefficients:
\[
E \;=\; \sum_{e} x_e\, e,
\qquad x_e\in\{0,1\}.
\]
Here \(x_e=1\) means ``edge \(e\) is in the error chain.''

\subsubsection{Boundary operator produces the syndrome}

Let \(\partial: C_1 \to C_0\) be the boundary operator modulo 2.
For an edge \(e=(u,v)\),
\[
\partial e \;=\; u+v \quad (\text{mod }2).
\]
Therefore,
\[
\partial E \;=\; \sum_e x_e\,\partial e
\]
is exactly the set of vertices incident to an odd number of error edges.
Those are precisely the vertices whose \(X\)-checks flip.
So in the ideal picture:
\[
s \;=\; \partial E.
\]

\subsubsection{Correction, cycles, and why ``matching the boundary'' is not enough}

A decoder outputs a correction chain \(\widehat{E}\) such that
\[
\partial \widehat{E} \;=\; s.
\]
Then
\[
\partial(\widehat{E}+E)=\partial\widehat{E}+\partial E=s+s=0,
\]
so \(\widehat{E}+E\) is a cycle.

\medskip
\noindent\textbf{Core success criterion.}
\[
\widehat{E}\text{ matches the syndrome, and the cycle }\widehat{E}+E\text{ is homologically trivial.}
\]
If \(\widehat{E}+E\) represents a nontrivial homology class (connects the wrong boundaries in the planar patch, or wraps
around in a toric code), it implements a logical operator and decoding fails.

\subsection{Useful visualization 1: chains vs.\ endpoints (self-contained TikZ)}

The figure below shows a \(Z\)-error chain \(E\) on edges and the syndrome \(\partial E\) on vertices.

\begin{figure}[t]
	\centering
	\begin{tikzpicture}[scale=0.95, line cap=round, line join=round]
		\def\n{5}
		\pgfmathtruncatemacro{\NmOne}{\n-1}
		
		\foreach \i in {0,...,\n}{
			\foreach \j in {0,...,\n}{
				\fill (\i,\j) circle (\EMDot);
			}
		}
		
		\foreach \i in {0,...,\n}{
			\foreach \j in {0,...,\NmOne}{
				\draw[gray!55, line width=0.4pt] (\i,\j) -- (\i,\j+1);
				\draw[gray!55, line width=0.4pt] (\j,\i) -- (\j+1,\i);
			}
		}
		
		\draw[line width=1.2pt] (1,1) -- (2,1);
		\draw[line width=1.2pt] (2,1) -- (3,1);
		\draw[line width=1.2pt] (3,1) -- (3,2);
		\draw[line width=1.2pt] (3,2) -- (3,3);
		\draw[line width=1.2pt] (3,3) -- (4,3);
		
		\fill[white] (1,1) circle (\EMSyndromeR);
		\draw[line width=1.0pt] (1,1) circle (\EMSyndromeR);
		\fill[white] (4,3) circle (\EMSyndromeR);
		\draw[line width=1.0pt] (4,3) circle (\EMSyndromeR);
		
		\node[align=center] at (2.6,5.6) {\small \(Z\)-errors form a chain \(E\) on edges};
		\node[align=center] at (2.6,5.1) {\small syndrome \(s=\partial E\) is the set of endpoints};
		\node[anchor=west] at (5.3,3.0) {\small circle = flipped check};
		\node[anchor=west] at (5.3,2.5) {\small thick edge = data-qubit \(Z\)-error};
		
		\draw[->, line width=0.8pt] (4.7,3.2) -- (4.15,3.05);
		\draw[->, line width=0.8pt] (4.7,2.7) -- (1.25,1.15);
	\end{tikzpicture}
	\caption{A \(Z\)-error chain \(E\) (thick edges) produces syndrome defects at \(\partial E\) (circled vertices). A decoder chooses \(\widehat{E}\) with \(\partial\widehat{E}=\partial E\); failure occurs if \(\widehat{E}+E\) is a nontrivial cycle (logical operator).}
	\label{fig:surface:chain-boundary}
\end{figure}
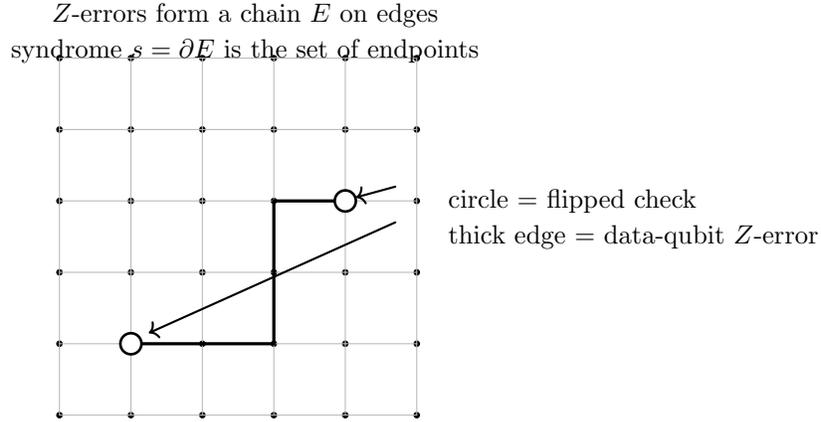

\subsection{Probability weights and the decoding objective}

\subsubsection{Why weights appear}

If each data qubit suffers a \(Z\)-error independently with probability \(p\), then a chain \(E\) with \(|E|\) edges has
likelihood (up to a constant factor)
\[
\Pr(E)\propto p^{|E|}(1-p)^{M-|E|},
\]
where \(M\) is the number of data qubits (edges).
Conditioned on a fixed syndrome, maximizing likelihood is equivalent to minimizing \(|E|\),
i.e.\ ``shortest correction'' in uniform noise.

\subsubsection{From likelihood to additive edge costs}

If edges have different error rates \(p_e\), then
\[
\Pr(E)=\prod_{e\in E} p_e \prod_{e\notin E} (1-p_e).
\]
Taking \(-\log\) and discarding constants independent of \(E\) yields an additive objective:
\[
\text{cost}(E)=\sum_{e\in E} w_e,
\qquad
w_e:=\log\!\left(\frac{1-p_e}{p_e}\right).
\]
Thus the decoding objective is:
\[
\text{find }\widehat{E}\text{ such that }\partial\widehat{E}=s
\quad\text{and}\quad
\widehat{E}\in\arg\min_{E:\,\partial E=s}\ \sum_{e\in E} w_e.
\]
This is the precise bridge from topology to graph algorithms: \(\sum w_e\) is a shortest-path/matching weight.

\subsection{Graph formulation (bridge from topology to algorithms)}

\subsubsection{Defects and pairing logic}

In the single-round, noiseless-measurement picture, the syndrome is a set of defect vertices
\[
s=\{v_1,\dots,v_{2k}\}.
\]
(Endpoints come in pairs unless boundaries are involved.)
A correction \(\widehat{E}\) can be thought of as pairing defects by lattice paths.
The total cost is the sum of edge weights along the chosen paths.

\subsubsection{Weighted distance between defects}

Define the weighted shortest-path distance
\[
d(u,v):=\min_{\gamma:u\to v}\ \sum_{e\in\gamma} w_e,
\]
where \(\gamma\) ranges over lattice paths connecting defects \(u\) and \(v\).
This is computed by Dijkstra (or BFS in uniform weights) on the underlying lattice graph.

\subsubsection{MWPM in one line}

Build a complete graph on the defect set with edge weights \(d(v_i,v_j)\), then solve
\[
\text{minimum-weight perfect matching (MWPM).}
\]
Each matched pair is expanded back into a lattice shortest path;
the union of those paths is the correction chain \(\widehat{E}\).

\subsubsection{Planar boundaries as special match partners}

In a planar code, a defect can also be matched to an appropriate boundary (rough/smooth).
Algorithmically, boundaries are handled by adding virtual boundary nodes and distances to them,
so MWPM can pair a defect to a boundary when that is cheaper and physically correct.

\begin{figure}[t]
	\centering
	\begin{tikzpicture}[scale=1.0, line cap=round, line join=round]
		\fill (0,0) circle (\EMNodeDot) node[below] {\small $v_1$};
		\fill (2.0,1.2) circle (\EMNodeDot) node[above] {\small $v_2$};
		\fill (4.2,0.2) circle (\EMNodeDot) node[below] {\small $v_3$};
		\fill (1.2,3.0) circle (\EMNodeDot) node[above] {\small $v_4$};
		
		\draw[line width=0.6pt] (0,0) -- (2.0,1.2) node[midway, above left] {\small $d_{12}$};
		\draw[line width=0.6pt] (0,0) -- (4.2,0.2) node[midway, below] {\small $d_{13}$};
		\draw[line width=0.6pt] (0,0) -- (1.2,3.0) node[midway, left] {\small $d_{14}$};
		
		\draw[line width=0.6pt] (2.0,1.2) -- (4.2,0.2) node[midway, right] {\small $d_{23}$};
		\draw[line width=0.6pt] (2.0,1.2) -- (1.2,3.0) node[midway, above] {\small $d_{24}$};
		\draw[line width=0.6pt] (4.2,0.2) -- (1.2,3.0) node[midway, right] {\small $d_{34}$};
		
		\draw[line width=1.2pt] (0,0) -- (2.0,1.2);
		\draw[line width=1.2pt] (4.2,0.2) -- (1.2,3.0);
		
		\node[align=center] at (2.2,-1.0) {\small Compute distances $d(v_i,v_j)$ from lattice weights};
		\node[align=center] at (2.2,-1.45) {\small MWPM chooses a minimum-weight pairing (thick edges)};
	\end{tikzpicture}
	\caption{Graph bridge: shortest-path distances between defects become weights in a complete graph; MWPM chooses a pairing, then each pair expands back to a lattice path forming \(\widehat{E}\).}
	\label{fig:surface:mwpm-bridge}
\end{figure}

\subsection{Decoder families: MWPM, Union--Find, and local rules (overview)}

\subsubsection{MWPM (Minimum-Weight Perfect Matching)}
\begin{itemize}
	\item \textbf{Interpretation:} maximum-likelihood decoding for common independent Pauli noise models.
	\item \textbf{Strength:} excellent logical error rates; well-understood theoretically.
	\item \textbf{Engineering cost:} matching and repeated shortest-path computations can be heavy at scale;
	real-time use requires careful batching, windowing, and optimized implementations.
\end{itemize}

\subsubsection{Union--Find (cluster growth / peeling)}
\begin{itemize}
	\item \textbf{Idea:} grow clusters around defects until they merge; resolve within each cluster using local structure.
	\item \textbf{Strength:} near-linear time; hardware-friendly data structures; good candidate for low-latency systems.
	\item \textbf{Tradeoff:} can be slightly suboptimal vs.\ MWPM under some noise/correlation models (variant-dependent).
\end{itemize}

\subsubsection{Local-rule / cellular decoders}
\begin{itemize}
	\item \textbf{Idea:} update rules depend only on a local neighborhood, iterated over time.
	\item \textbf{Strength:} extreme parallelism; maps naturally to FPGA/ASIC fabrics.
	\item \textbf{Risk:} may lose global optimality; performance depends strongly on rule design and noise correlations.
\end{itemize}

\subsection{Spacetime decoding graphs (space-like vs.\ time-like edges)}

\subsubsection{Why time enters: measurement errors}

If stabilizer measurements are noisy, a single flipped measurement bit can create a fake defect.
A standard move is to decode on \emph{syndrome differences}:
\[
\Delta s_t \;:=\; s_t \oplus s_{t-1}.
\]
Now:
\begin{itemize}
	\item a \textbf{data-qubit error} changes the stabilizer eigenvalue and persists, producing a change between rounds (space-like edge),
	\item a \textbf{measurement error} flips only one round, producing changes across time (time-like edge).
\end{itemize}

\subsubsection{3D decoding picture}

In spacetime \((x,y,t)\), defects live on vertices and paths may use:
\begin{itemize}
	\item \textbf{space-like edges:} model data-qubit errors,
	\item \textbf{time-like edges:} model measurement flips.
\end{itemize}
Weights come from two physical rates \(p_{\text{data}}\) and \(p_{\text{meas}}\) (and can be refined per edge/check).

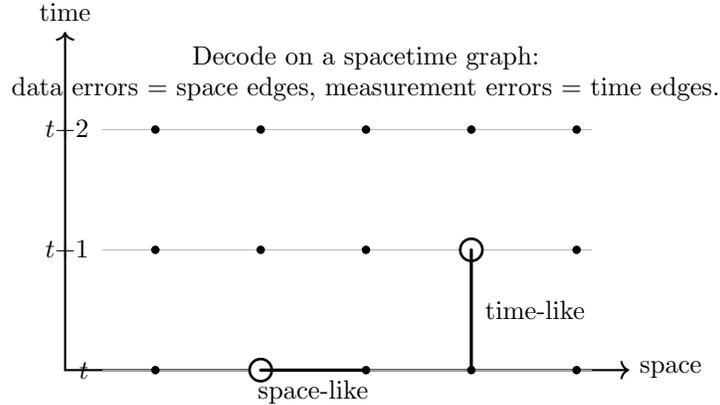
\begin{figure}[t]
	\centering
	\begin{tikzpicture}[scale=1.0, line cap=round, line join=round]
		\draw[->, line width=0.8pt] (0,0) -- (7.5,0) node[right] {\small space};
		\draw[->, line width=0.8pt] (0,0) -- (0,4.5) node[above] {\small time};
		
		\foreach \ty/\lbl in {0/{\small $t$},1.6/{\small $t{+}1$},3.2/{\small $t{+}2$}}{
			\draw[gray!55, line width=0.4pt] (0.5,\ty) -- (7.0,\ty);
			\node[left] at (0.45,\ty) {\lbl};
			\foreach \x in {1.2,2.6,4.0,5.4,6.8}{
				\fill (\x,\ty) circle (\EMNodeDot);
			}
		}
		
		\fill[white] (2.6,0.0) circle (\EMSyndromeR);
		\draw[line width=1.0pt] (2.6,0.0) circle (\EMSyndromeR);
		\fill[white] (5.4,1.6) circle (\EMSyndromeR);
		\draw[line width=1.0pt] (5.4,1.6) circle (\EMSyndromeR);
		
		\draw[line width=1.2pt] (2.6,0.0) -- (4.0,0.0);
		\node[below] at (3.3,0.0) {\small space-like};
		
		\draw[line width=1.2pt] (5.4,0.0) -- (5.4,1.6);
		\node[right] at (5.45,0.8) {\small time-like};
		
		\node[align=center] at (4.0,4.15) {\small Decode on a spacetime graph:};
		\node[align=center] at (4.0,3.75) {\small data errors = space edges,\ measurement errors = time edges.};
	\end{tikzpicture}
	\caption{Spacetime decoding sketch: with noisy measurements, decoding occurs on a 3D graph. Time-like edges model measurement flips; space-like edges model data-qubit errors.}
	\label{fig:surface:spacetime-graph}
\end{figure}

\subsection{Circuit origin: why syndromes exist as a stream}

\subsubsection{Syndrome extraction repeats every cycle}

A check is implemented by an ancilla-mediated circuit:
\[
\text{prepare ancilla}\ \to\ \text{entangle with neighbors}\ \to\ \text{measure ancilla}\ \to\ \text{reset}.
\]
Noise is continuous in time, so this cycle must repeat indefinitely.
Therefore the syndrome is not a static dataset; it is a time-indexed stream, and decoding must be online.

\subsubsection{Why this pushes toward FPGA acceleration}

The syndrome stream is high-rate, low-latency, and structurally local:
\begin{itemize}
	\item \textbf{High-rate:} \(\Theta(d^2)\) checks per cycle for distance \(d\).
	\item \textbf{Low-latency:} decisions must meet deadlines tied to the control loop.
	\item \textbf{Local structure:} lattice locality \(\Rightarrow\) parallel, hardware-friendly computation patterns.
\end{itemize}
This is exactly the regime where FPGA implementations (pipelines, locality, predictable latency) can dominate CPUs/GPUs.

\subsection{Exercises}

\begin{exercise}[Boundary computation on a small patch]
	Consider the chain \(E\) consisting of the edges
	\((1,1)\!\!-\!\!(2,1)\), \((2,1)\!\!-\!\!(3,1)\), \((3,1)\!\!-\!\!(3,2)\).
	Compute \(\partial E\) explicitly (mod \(2\)) and identify the syndrome vertices.
\end{exercise}

\noindent\textbf{Solution.}
Compute each boundary:
\[
\partial\bigl((1,1)\!-\!(2,1)\bigr)=(1,1)+(2,1),\qquad
\partial\bigl((2,1)\!-\!(3,1)\bigr)=(2,1)+(3,1),
\]
\[
\partial\bigl((3,1)\!-\!(3,2)\bigr)=(3,1)+(3,2).
\]
Sum mod \(2\) (repeated vertices cancel):
\[
\partial E
=\bigl((1,1)+(2,1)\bigr)+\bigl((2,1)+(3,1)\bigr)+\bigl((3,1)+(3,2)\bigr)
=(1,1)+(3,2).
\]
So the syndrome defects are at \((1,1)\) and \((3,2)\).

\begin{exercise}[Why the decoding cost is additive]
	Assume independent edge error probabilities \(p_e\).
	Show that maximizing \(\Pr(E)\) is equivalent to minimizing \(\sum_{e\in E} w_e\) with
	\(w_e=\log((1-p_e)/p_e)\), up to an additive constant independent of \(E\).
\end{exercise}

\noindent\textbf{Solution.}
Independence gives
\[
\Pr(E)=\prod_{e\in E} p_e \prod_{e\notin E} (1-p_e).
\]
Take \(-\log\):
\[
-\log\Pr(E)=\sum_{e\in E}(-\log p_e)+\sum_{e\notin E}(-\log(1-p_e)).
\]
Rewrite the second term:
\[
\sum_{e\notin E}(-\log(1-p_e))
=\sum_e(-\log(1-p_e))-\sum_{e\in E}(-\log(1-p_e)).
\]
Therefore
\[
-\log\Pr(E)
=\underbrace{\sum_e(-\log(1-p_e))}_{\text{constant in }E}
+\sum_{e\in E}\bigl(-\log p_e+\log(1-p_e)\bigr)
=\text{const}+\sum_{e\in E}\log\!\Bigl(\frac{1-p_e}{p_e}\Bigr).
\]
Thus maximizing \(\Pr(E)\) is equivalent to minimizing \(\sum_{e\in E} w_e\).

\begin{exercise}[Why matching the boundary can still fail]
	Explain how \(\partial\widehat{E}=\partial E\) can still cause a logical failure.
	Use the cycle \(\widehat{E}+E\) and the notion of a nontrivial logical operator.
\end{exercise}

\noindent\textbf{Solution.}
If \(\partial\widehat{E}=\partial E\), then \(\partial(\widehat{E}+E)=0\), so \(\widehat{E}+E\) is a cycle.
Some cycles are boundaries of faces (stabilizers) and act trivially on the code space.
But others represent nontrivial homology classes: in a planar code they connect the wrong boundaries;
in a toric code they wrap around the torus.
Such a cycle implements a logical Pauli operator. If \(\widehat{E}+E\) equals a nontrivial cycle, the correction differs
from the true error by a logical operator, yielding a logical error despite perfectly matching the syndrome.

\begin{exercise}[Spacetime decoding: why use syndrome differences]
	Assume check measurements can flip with probability \(q\).
	Explain why decoding on \(\Delta s_t=s_t\oplus s_{t-1}\) separates data errors (space-like)
	from measurement errors (time-like). Describe a minimal spacetime picture in words.
\end{exercise}

\noindent\textbf{Solution.}
A data error changes the stabilizer eigenvalue from the time it occurs onward, so it causes a change between consecutive rounds;
this appears as a defect in \(\Delta s_t\) that is connected by a space-like edge.
A measurement flip affects only one round, so it creates a ``blip'' in time: the reported value differs at time \(t\) but not at \(t\pm1\),
producing two changes, one between \(t-1\) and \(t\) and another between \(t\) and \(t+1\).
In a spacetime graph, that corresponds to a vertical (time-like) edge. Hence \(\Delta s_t\) naturally maps the problem to paths
in a 3D graph with two edge types and two noise rates.

\begin{exercise}[Engineering: minimum per-round logging fields]
	List a minimal set of log fields per round that lets you diagnose deadline misses and correctness failures in simulation.
	Include configuration IDs, timestamps, and correctness indicators.
\end{exercise}

\noindent\textbf{Solution.}
A minimal per-round log includes:
\begin{itemize}
	\item \textbf{Configuration identifiers:} code family, distance \(d\), window \(W\), policy variant, build/hash ID.
	\item \textbf{Round index:} \(t\).
	\item \textbf{Timing:} arrival timestamp \(A_t\), decode start timestamp, decode end timestamp (latency \(L_t\)).
	\item \textbf{Backlog/buffer:} queue depth at arrival and after service (or bytes/frames buffered).
	\item \textbf{Syndrome summary:} defect count and/or a compressed syndrome representation (for replay).
	\item \textbf{Output decision:} Pauli-frame update ID or compressed correction descriptor.
	\item \textbf{Deadline flag:} \(1[L_t>D]\) and (optional) measured slack \(D-L_t\).
	\item \textbf{Correctness signal (simulation):} logical-failure flag, or mismatch indicator against injected errors.
	\item \textbf{Counters for spikes:} optional instrumentation counters (memory stalls, branch path ID, transport jitter flag).
\end{itemize}
	

\section{FPGA V: Decoder as a Finite-State Machine on a Lattice}
\label{sec:fpga5}

\subsection{Objective and guiding viewpoint}

A real-time surface-code decoder is not ``just an algorithm.'' On hardware, it is best understood as:

\medskip
\noindent\textbf{A finite-state machine (FSM) that repeatedly updates a lattice-state array,}
driven by a streaming input of detection events, under strict latency and throughput constraints.

\medskip
This viewpoint forces three design commitments that software often hides:
\begin{itemize}
	\item \textbf{Bounded phases:} the decoder must complete in a bounded number of passes per round/window.
	\item \textbf{Locality as a contract:} updates use local neighborhoods so wiring and memory access remain structured.
	\item \textbf{Time is physics:} missing deadlines is logically equivalent to injecting extra noise into the QEC loop.
\end{itemize}

\subsection{Syndromes and detection events (streaming input)}

\subsubsection{Syndrome bits vs.\ detection events}

Let \(s_t\) denote the syndrome bits measured at cycle \(t\) (one bit per check).
With noisy measurements, it is standard to form \emph{detection events} (also called defects)
by taking differences:
\[
\Delta s_t := s_t \oplus s_{t-1}.
\]
Intuition:
\begin{itemize}
	\item a \textbf{data error} changes the underlying stabilizer eigenvalue and persists, so it creates a consistent change
	visible in \(\Delta s_t\) and is modeled as a \emph{space-like} edge;
	\item a \textbf{measurement error} is a one-round flip, producing two consecutive differences,
	modeled as a \emph{time-like} edge.
\end{itemize}

\subsubsection{Streaming input contract}

On hardware, the input is not ``a full graph instance.'' It arrives as a stream:
\[
\texttt{(t, check\_id, bit)}\quad \text{or}\quad \texttt{(t, check\_id, delta\_bit)}.
\]
A practical interface is an AXI-stream-like packet:
\begin{itemize}
	\item \texttt{t} (round index mod \(W\) for a sliding window),
	\item \texttt{addr} (check coordinate or linear index),
	\item \texttt{val} (syndrome or detection bit),
	\item \texttt{valid/last} markers for framing.
\end{itemize}
The decoder must begin processing before an entire round is received if latency is tight.

\subsection{The decoding graph and locality}

\subsubsection{Spacetime decoding graph}

With detection events, decoding occurs on a 3D spacetime graph \(G=(V,E)\) where vertices correspond to
check locations at times \(t\), and edges correspond to plausible single-fault explanations:
\begin{itemize}
	\item \textbf{space-like edges} connect neighboring checks within the same time slice,
	\item \textbf{time-like edges} connect a check across consecutive time slices.
\end{itemize}
Weights encode negative log-likelihoods:
\[
w_e = \log\!\Bigl(\frac{1-p_e}{p_e}\Bigr),
\]
with different \(p_e\) for data vs.\ measurement faults.

\subsubsection{Locality: what hardware actually needs}

Locality is not a philosophical preference; it is a hardware constraint:
\begin{itemize}
	\item \textbf{bounded-degree neighborhoods} \(\Rightarrow\) fixed number of memory reads per update,
	\item \textbf{structured adjacency} \(\Rightarrow\) address arithmetic, not pointer chasing,
	\item \textbf{tiling} \(\Rightarrow\) you can stream tiles through on-chip buffers.
\end{itemize}

\subsection{Decoder as a global FSM (bounded phases, bounded passes)}

\subsubsection{Why a global FSM exists even for ``local'' decoders}

Even if the update rule is local, the hardware must orchestrate:
\begin{itemize}
	\item receiving and buffering detection events,
	\item initializing per-round state,
	\item running a bounded number of update sweeps,
	\item emitting decisions (pairings, corrections, or Pauli-frame deltas),
	\item advancing the sliding window.
\end{itemize}
That orchestration is a \emph{global} FSM controlling local compute.

\subsubsection{A canonical phase decomposition}

A typical bounded-phase decoder looks like:

\begin{enumerate}
	\item \textbf{INGEST:} accept detection-event stream for round \(t\), update defect markers.
	\item \textbf{SEED:} initialize cluster/labels/frontiers from defects.
	\item \textbf{GROW / PROPAGATE:} run \(P\) passes of local propagation (cluster growth, belief messages, wavefronts).
	\item \textbf{RESOLVE:} finalize pairings/paths inside clusters (peeling / parity fix / local traceback).
	\item \textbf{EMIT:} output Pauli-frame update (or correction summary) and round statistics.
	\item \textbf{ADVANCE:} slide window, recycle memory pages, clear old time slice.
\end{enumerate}

\subsubsection{Finite-state means bounded work}

The key engineering requirement is: for each round/window, the number of cycles is bounded:
\[
T_{\text{decode}} \le T_{\max}(d,W,\text{policy}),
\]
where the bound is derived from:
\begin{itemize}
	\item maximum number of passes \(P\),
	\item maximum neighborhood size per update,
	\item maximum tile count to cover the lattice.
\end{itemize}
This is where the ``decoder as FSM'' framing becomes a correctness requirement, not just a style.

\subsection{Latency is part of the physics}

\subsubsection{Deadline misses behave like extra noise}

Let the QEC cycle time be \(T_{\text{cycle}}\) and the decoder deadline be \(D\le T_{\text{cycle}}\).
If the decoder output arrives late, the control system must either:
\begin{itemize}
	\item postpone corrective action (increasing effective error accumulation), or
	\item apply stale/partial decisions (increasing logical failure probability).
\end{itemize}
In both cases, lateness increases the effective physical error rate seen by the logical qubit.

\subsubsection{Two practical latency numbers}

Hardware teams usually track:
\begin{itemize}
	\item \textbf{critical-path latency} (worst-case bound),
	\item \textbf{tail latency} (p99/p999) because jitter breaks stability even when the mean is fine.
\end{itemize}

\subsection{Throughput: bandwidth dominates arithmetic}

\subsubsection{Why memory traffic is the bottleneck}

Many decoding updates are simple (XOR, min-plus, small integer ops), but they touch many lattice sites.
Thus, performance is frequently limited by:
\[
\text{bytes moved per round} \;\gg\; \text{flops per round}.
\]
A good first model is:
\[
\text{Throughput} \approx \frac{\text{memory bandwidth}}{\text{bytes accessed per lattice update}}.
\]

\subsubsection{The ``bytes per detection event'' accounting}

A practical budgeting method:
\begin{itemize}
	\item count how many lattice words must be read/updated per detection event (including neighbor fetches),
	\item multiply by event rate (checks per round),
	\item compare to available DDR/BRAM bandwidth.
\end{itemize}
This quickly predicts when your design must:
\begin{itemize}
	\item compress state,
	\item cache tiles on-chip,
	\item or increase parallel banks (multiport via replication/partitioning).
\end{itemize}

\subsection{Circuit origin of the classical stream (why FIFO appears)}

\subsubsection{Why the decoder sees a stream, not a batch}

Syndrome extraction circuits produce one ancilla measurement per check per cycle.
Electronics and firmware deliver those bits as packets with:
\begin{itemize}
	\item transport latency variability,
	\item burstiness (many checks arrive together),
	\item occasional missing/invalid samples.
\end{itemize}
Therefore the FPGA front end almost always includes a \textbf{FIFO} to decouple input timing from the decoder pipeline.

\subsubsection{A minimal streaming block diagram (TikZ)}

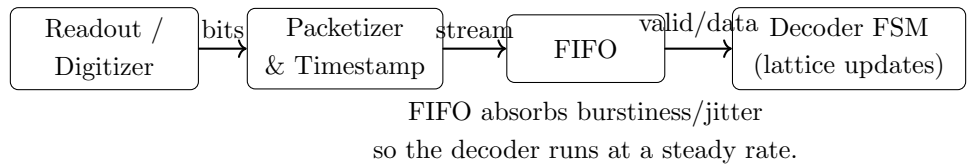
\begin{figure}[t]
	\centering
	\begin{tikzpicture}[scale=1.0, line cap=round, line join=round]
		\node[draw, rounded corners=3pt, minimum width=2.5cm, minimum height=0.9cm, align=center] (adc) at (0,0) {\small Readout /\\ \small Digitizer};
		\node[draw, rounded corners=3pt, minimum width=2.6cm, minimum height=0.9cm, align=center] (pkt) at (3.2,0) {\small Packetizer\\ \small \& Timestamp};
		\node[draw, rounded corners=3pt, minimum width=2.1cm, minimum height=0.9cm, align=center] (fifo) at (6.4,0) {\small FIFO};
		\node[draw, rounded corners=3pt, minimum width=3.1cm, minimum height=0.9cm, align=center] (dec) at (9.9,0) {\small Decoder FSM\\ \small (lattice updates)};
		
		\draw[->, line width=0.8pt] (adc) -- (pkt) node[midway, above] {\small bits};
		\draw[->, line width=0.8pt] (pkt) -- (fifo) node[midway, above] {\small stream};
		\draw[->, line width=0.8pt] (fifo) -- (dec) node[midway, above] {\small valid/data};
		
		\node[align=center] at (6.4,-1.1) {\small FIFO absorbs burstiness/jitter\\ \small so the decoder runs at a steady rate.};
	\end{tikzpicture}
	\caption{Why FIFO appears: syndrome bits arrive as a timestamped stream with jitter/bursts; the decoder FSM needs steady service to guarantee latency bounds.}
	\label{fig:fpga5:fifo-origin}
\end{figure}

\subsection{Lattice state and memory layout (RAM words)}

\subsubsection{The lattice state is a data structure, not a picture}

A hardware decoder maintains per-site state:
\begin{itemize}
	\item defect flags (present/absent),
	\item cluster labels / parents (Union--Find) or messages (belief-like),
	\item growth radii / parity info / timestamps,
	\item ``resolved'' flags and traceback pointers (optional).
\end{itemize}

\subsubsection{Linear address mapping}

Let the lattice have \(N_x\times N_y\) checks per time slice, and a window of \(W\) time slices.
A common linearization:
\[
\texttt{idx}(x,y,t)= (t \bmod W)\cdot (N_x N_y) + yN_x + x.
\]
This gives:
\begin{itemize}
	\item simple address arithmetic,
	\item contiguous scans over \(x,y\),
	\item easy page-rotation when advancing the window.
\end{itemize}

\subsubsection{Packing into RAM words}

Suppose each site needs fields:
\[
\texttt{defect:1},\quad \texttt{label:k},\quad \texttt{parity:1},\quad \texttt{age:r},\quad \texttt{misc:m}.
\]
Pack into a fixed-width word:
\[
\texttt{word}[x,y,t] \in \{0,1\}^{B},\qquad B = 1+k+1+r+m.
\]
Design choices:
\begin{itemize}
	\item \textbf{narrow words} reduce bandwidth but complicate bit slicing,
	\item \textbf{wide words} simplify updates but increase bandwidth.
\end{itemize}
This is where throughput math forces tradeoffs.

\subsection{Local rules and update schedules (hazard-avoidance)}

\subsubsection{Why hazards occur}

Local update rules often read neighbor states and then write an updated state.
If two neighboring sites update simultaneously and both write, you can get:
\begin{itemize}
	\item write-write conflicts,
	\item read-after-write ambiguity (depends on pipeline timing),
	\item inconsistent symmetry breaking (e.g.\ two clusters both claim a site).
\end{itemize}

\subsubsection{Two standard hazard-avoidance patterns}

\paragraph{(1) Red-black (checkerboard) updates.}
Partition sites into two colors; update all red sites using frozen black neighbors, then swap.
This guarantees neighbors do not write at the same time.

\paragraph{(2) Two-phase commit (read phase / write phase).}
Compute proposed updates into a shadow buffer, then commit writes in a separate phase.
This costs extra memory but yields deterministic behavior.

\subsubsection{A simple checkerboard visualization (TikZ)}

\begin{figure}[t]
	\centering
	\begin{tikzpicture}[scale=0.9, line cap=round, line join=round]
		\def\n{6}
		\pgfmathtruncatemacro{\NmOne}{\n-1}
		
		\foreach \i in {0,...,\NmOne}{
			\foreach \j in {0,...,\NmOne}{
				\pgfmathtruncatemacro{\c}{mod(\i+\j,2)}
				\ifnum\c=0
				\fill[gray!20] (\i,\j) rectangle (\i+1,\j+1);
				\else
				\fill[gray!5] (\i,\j) rectangle (\i+1,\j+1);
				\fi
				\draw[gray!60, line width=0.4pt] (\i,\j) rectangle (\i+1,\j+1);
			}
		}
		
		\node[align=center] at (3,-0.7) {\small Red-black schedule: update one color while the other is read-only.};
	\end{tikzpicture}
	\caption{Checkerboard scheduling is a hardware-friendly way to avoid neighbor write conflicts in local lattice updates.}
	\label{fig:fpga5:checkerboard}
\end{figure}

\subsection{From algorithms to hardware budgets (worst-case bounds)}

\subsubsection{A budgeting template}

Let:
\begin{itemize}
	\item \(N = N_xN_yW\) be the number of lattice sites in the active window,
	\item \(P\) be the maximum number of passes,
	\item \(r\) be the sites updated per cycle (parallel lanes),
	\item \(c_{\text{upd}}\) be cycles per site update (pipeline depth / multi-read).
\end{itemize}
A conservative bound:
\[
T_{\text{decode}} \;\le\; P\cdot \left\lceil \frac{N}{r}\right\rceil \cdot c_{\text{upd}} \;+\; T_{\text{ingest}} \;+\; T_{\text{emit}}.
\]
This bound is what you compare against the deadline \(D\).

\subsubsection{Bandwidth bound (often tighter than cycle count)}

If each site update requires \(R\) RAM reads and \(W_r\) RAM writes of \(B\)-bit words, then bytes moved per update:
\[
\text{bytes/update} = \frac{B}{8}(R+W_r).
\]
Total bytes per decode:
\[
\text{bytes/decode} \approx P\cdot N \cdot \frac{B}{8}(R+W_r).
\]
The design is feasible only if:
\[
\frac{\text{bytes/decode}}{D} \;\le\; \text{available bandwidth}.
\]
This inequality is a fast, brutally honest feasibility test.

\subsection{Exercises}

\begin{exercise}[Detection-event stream formation]
	Assume a single check produces outcomes \(s_{t-1}=1\), \(s_t=0\), \(s_{t+1}=1\).
	Compute \(\Delta s_t\) and \(\Delta s_{t+1}\), and interpret whether this looks like a measurement error or a persistent data error.
\end{exercise}

\noindent\textbf{Solution.}
\[
\Delta s_t = s_t\oplus s_{t-1} = 0\oplus 1 = 1,\qquad
\Delta s_{t+1} = s_{t+1}\oplus s_t = 1\oplus 0 = 1.
\]
Two consecutive detection events at the same check are consistent with a one-round measurement flip at time \(t\):
the reported value toggled at \(t\) relative to both neighbors.
A persistent data error would typically produce a \emph{single} change at the time the error occurs (then the syndrome stays flipped),
not a blip that immediately reverts.

\begin{exercise}[Linear memory layout]
	Let \(N_x=8\), \(N_y=6\), \(W=4\). Write \(\texttt{idx}(x,y,t)\) and compute \(\texttt{idx}(x{=}3,y{=}2,t{=}5)\).
	Assume \(t\) is reduced modulo \(W\).
\end{exercise}

\noindent\textbf{Solution.}
\[
\texttt{idx}(x,y,t)=(t\bmod W)\cdot (N_xN_y)+yN_x+x.
\]
Here \(N_xN_y=48\) and \(t\bmod W = 5\bmod 4 = 1\). So
\[
\texttt{idx}(3,2,5)=1\cdot 48 + 2\cdot 8 + 3 = 48 + 16 + 3 = 67.
\]

\begin{exercise}[Worst-case pass budget]
	A decoder uses \(P=6\) passes on a window with \(N=12{,}288\) sites.
	Your FPGA updates \(r=32\) sites per cycle and needs \(c_{\text{upd}}=2\) cycles per site update.
	Ignore ingest/emit and compute the worst-case decode cycles.
\end{exercise}

\noindent\textbf{Solution.}
\[
T_{\text{decode}} \le P\cdot \left\lceil \frac{N}{r}\right\rceil\cdot c_{\text{upd}}
=6\cdot \left\lceil \frac{12288}{32}\right\rceil \cdot 2.
\]
Since \(12288/32 = 384\) exactly,
\[
T_{\text{decode}} \le 6\cdot 384\cdot 2 = 4608 \text{ cycles}.
\]

\begin{exercise}[Bandwidth sanity check]
	Each lattice update reads \(R=3\) words and writes \(W_r=1\) word.
	The packed word width is \(B=64\) bits. The decoder performs \(P=4\) passes over \(N=10^5\) sites per round.
	Estimate bytes moved per round and the required bandwidth if the deadline is \(D=100\,\mu\text{s}\).
\end{exercise}

\noindent\textbf{Solution.}
Bytes per update:
\[
\frac{B}{8}(R+W_r)=\frac{64}{8}(3+1)=8\cdot 4=32\ \text{bytes/update}.
\]
Bytes per round:
\[
\text{bytes/round} \approx P\cdot N \cdot 32 = 4\cdot 10^5 \cdot 32 = 1.28\times 10^7\ \text{bytes}
\approx 12.8\ \text{MB}.
\]
Bandwidth required to finish within \(D=100\,\mu\text{s}\):
\[
\frac{12.8\ \text{MB}}{100\,\mu\text{s}} = 128\ \text{GB/s}.
\]
This number is extremely high, signaling that you must reduce bytes moved (state compression, fewer passes, more on-chip tiling/caching, fewer reads/writes, or algorithmic changes) or relax the deadline.

\begin{exercise}[Hazard-avoidance reasoning]
	Explain (in 3--5 sentences) why a red-black update schedule avoids neighbor write conflicts.
	What is the main downside relative to a fully parallel update?
\end{exercise}

\noindent\textbf{Solution.}
In a checkerboard (red-black) partition, every red site has only black neighbors and vice versa.
If you update all red sites simultaneously while keeping black sites read-only, no two concurrently-updating sites are neighbors,
so they never write to each other's state in the same cycle.
Then you swap roles and update black sites using the newly written red states.
The downside is that it can reduce peak parallelism (you effectively update only half the sites per phase) and may increase the number
of global phases needed for convergence or completion.
	
\section{Union--Find Decoder Microarchitecture (From Data Structure to RTL Pipeline)}
\label{sec:uf-microarch}

\subsection{Objective: deterministic Union--Find under streaming constraints}
Union--Find (UF) decoders are attractive for surface codes because they can be made
\emph{near-linear}, \emph{locality-friendly}, and \emph{hardware-realistic}.
But textbook UF is an \emph{unbounded} data structure:
path compression and union-by-rank are great on average, yet they can generate
variable-latency pointer-chasing (bad for deadlines and p99/p999 control).

\medskip
\noindent\textbf{Goal of this section.}
We refactor UF decoding into an RTL-facing microarchitecture that:
\begin{itemize}
	\item ingests a \emph{stream} of detection events (syndrome differences) at fixed cadence,
	\item runs in a \emph{fixed number of bounded passes} (deterministic latency),
	\item uses a \emph{memory layout} that avoids bank conflicts and long pointer chains,
	\item produces a correction / Pauli-frame update with explicit worst-case bounds.
\end{itemize}

\medskip
\noindent\textbf{Guiding viewpoint.}
Think of UF decoding as a lattice-wide FSM with three repeating phases:
\[
\textbf{GROW} \;\to\; \textbf{MERGE} \;\to\; \textbf{PEEL},
\]
implemented by streaming, banked RAM accesses and hazard-aware scheduling.

\subsection{Union--Find recap: \texttt{find}, \texttt{union}, ranks, compression}

\subsubsection{The textbook structure}
We maintain a forest of rooted trees over a set of elements (nodes).
Each element $i$ stores:
\[
\texttt{parent}[i]\in\{0,\dots,N-1\},\qquad
\texttt{rank}[i]\in\mathbb{N}.
\]
A root satisfies $\texttt{parent}[r]=r$.

\begin{defn}[\texttt{find} (root query)]
	\label{def:uf-find}
	Given a node $i$, \texttt{find(i)} returns the representative (root) of the tree containing $i$:
	\[
	\mathrm{root}(i)=
	\begin{cases}
		i, & \text{if }\texttt{parent}[i]=i,\\
		\mathrm{root}(\texttt{parent}[i]), & \text{otherwise}.
	\end{cases}
	\]
\end{defn}

\begin{defn}[\texttt{union} (merge two sets)]
	\label{def:uf-union}
	Given $a,b$, compute $r_a=\texttt{find}(a)$, $r_b=\texttt{find}(b)$.
	If $r_a\neq r_b$, attach one root under the other root (union-by-rank):
	\[
	\text{if }\texttt{rank}[r_a]<\texttt{rank}[r_b]\text{ then set }\texttt{parent}[r_a]=r_b,
	\]
	else set $\texttt{parent}[r_b]=r_a$ and increase rank if tied.
\end{defn}

\subsubsection{Why naive UF is hardware-hostile}
\begin{itemize}
	\item \textbf{Pointer chasing:} \texttt{find} follows parent pointers with data-dependent depth.
	\item \textbf{Path compression:} good for amortized complexity, but creates bursty writes.
	\item \textbf{Concurrent unions:} two unions touching the same root create write hazards.
	\item \textbf{Variable runtime:} even if average is small, p99 tail can violate deadlines.
\end{itemize}

\subsubsection{Hardware-aware alternatives to full path compression}
We will use \emph{bounded} variants:
\begin{itemize}
	\item \textbf{Path halving:} on each \texttt{find} step, rewrite $\texttt{parent}[i]\leftarrow \texttt{parent}[\texttt{parent}[i]]$.
	\item \textbf{Pointer jumping in passes:} run a fixed number $P$ of pointer-jump passes over all nodes in parallel.
	\item \textbf{No compression in the critical loop:} postpone compression to a non-critical stage or do it partially.
\end{itemize}

\begin{figure}[t]
	\centering
	\begin{tikzpicture}[scale=1.0, line cap=round, line join=round]
		\node (r) at (0,0) [circle, draw, inner sep=1.2pt] {$r$};
		\node (a) at (-1.2,-1.0) [circle, draw, inner sep=1.2pt] {$a$};
		\node (b) at (0.0,-1.0) [circle, draw, inner sep=1.2pt] {$b$};
		\node (c) at (1.2,-1.0) [circle, draw, inner sep=1.2pt] {$c$};
		\node (d) at (0.0,-2.0) [circle, draw, inner sep=1.2pt] {$d$};
		\draw (a) -- (r);
		\draw (b) -- (r);
		\draw (c) -- (b);
		\draw (d) -- (c);
		\node[align=center] at (0,-2.9) {\small variable-depth parent pointers};
		
		\begin{scope}[xshift=7.0cm]
			\node (r2) at (0,0) [circle, draw, inner sep=1.2pt] {$r$};
			\node (a2) at (-1.2,-1.0) [circle, draw, inner sep=1.2pt] {$a$};
			\node (b2) at (0.0,-1.0) [circle, draw, inner sep=1.2pt] {$b$};
			\node (c2) at (1.2,-1.0) [circle, draw, inner sep=1.2pt] {$c$};
			\node (d2) at (0.0,-2.0) [circle, draw, inner sep=1.2pt] {$d$};
			\draw (a2) -- (r2);
			\draw (b2) -- (r2);
			\draw (c2) -- (r2);
			\draw (d2) -- (r2);
			\node[align=center] at (0,-2.9) {\small after bounded pointer jumping};
		\end{scope}
		
		\node[align=center] at (3.5,0.55) {\small \textbf{Idea:} replace data-dependent depth\\[-0.2em]\small with fixed-pass flattening.};
		\draw[->] (2.1,-1.3) -- (4.9,-1.3);
	\end{tikzpicture}
	\caption{UF pointer chasing is variable-latency. A hardware-friendly approach uses fixed-pass pointer jumping/halving to reduce depth deterministically.}
	\label{fig:uf-pointer-jump}
\end{figure}
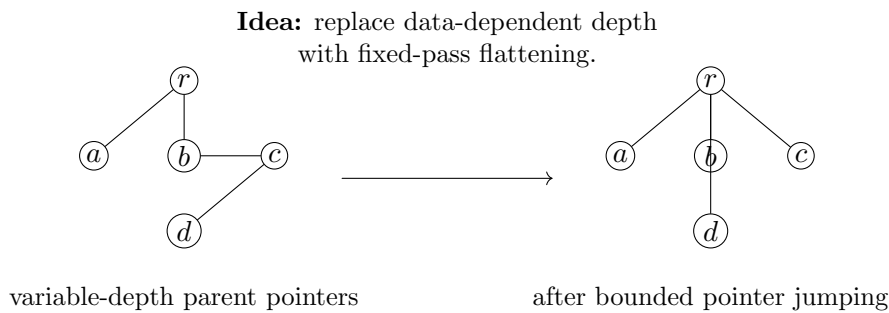

\subsection{Turning UF into bounded passes (real-time guarantee strategy)}

\subsubsection{What the decoder must compute (surface-code UF view)}
Input (per time step or per window):
\begin{itemize}
	\item detection events (defects) on a spacetime graph,
	\item edge weights / probabilities (possibly uniform),
	\item boundary conditions (virtual nodes for rough/smooth boundaries).
\end{itemize}
Output:
\begin{itemize}
	\item a set of correction edges or a Pauli-frame update that pairs defects appropriately.
\end{itemize}

\subsubsection{Bounded-pass philosophy}
We enforce determinism by organizing work into a fixed schedule:
\[
\underbrace{\textbf{PASS 1}}_{\text{GROW}} \to
\underbrace{\textbf{PASS 2}}_{\text{MERGE}} \to
\underbrace{\textbf{PASS 3}}_{\text{GROW}} \to \cdots \to
\underbrace{\textbf{FINAL}}_{\text{PEEL}}.
\]
Each pass:
\begin{itemize}
	\item is a streaming sweep over lattice sites / edges,
	\item performs only local reads/writes (bounded fan-in),
	\item uses arbitration to avoid write conflicts.
\end{itemize}

\subsubsection{A concrete deterministic bound}
Let $W$ be the decoding window size in time, and let $d$ be code distance.
A practical strategy is to limit cluster growth radius to a maximum $R_{\max}$
derived from hardware budgets and failure tolerance:
\[
R_{\max} \le c_1 d \quad\text{(planar)}\qquad\text{or}\qquad R_{\max} \le c_2 W \quad\text{(spacetime)}.
\]
Then run exactly $R_{\max}$ growth rounds (or fewer if convergence is detected early, but do not rely on early exit for deadlines).

\begin{rem}[Determinism vs.\ adaptivity]
	You may still \emph{record} early convergence to improve average power/throughput,
	but \emph{deadline planning} must assume the full $R_{\max}$ passes.
\end{rem}

\subsubsection{Two implementation patterns}
\begin{enumerate}
	\item \textbf{Strict deterministic:} always run $P$ passes, always same memory traffic.
	\item \textbf{Deterministic upper bound + safe early-exit:} run up to $P$ passes,
	but only exit early if the FSM proves a fixed-point (no updates) and the proof is itself bounded-time.
\end{enumerate}

\subsection{Memory layout: parent/rank arrays, bank conflicts, and locality}

\subsubsection{Indexing the decoding graph}
Flatten a 2D lattice (or 3D spacetime lattice) into linear addresses.
Example for a 2D grid $L_x\times L_y$:
\[
\mathrm{idx}(x,y)=x + L_x y.
\]
For spacetime with $t\in\{0,\dots,W-1\}$:
\[
\mathrm{idx}(x,y,t)=x + L_x y + (L_xL_y)t.
\]

\subsubsection{What must be stored}
Typical UF decoder state (per node or per edge) includes:
\begin{itemize}
	\item \textbf{parent} (root pointer): \texttt{parent[idx]}.
	\item \textbf{rank/size}: \texttt{rank[idx]} or \texttt{size[idx]}.
	\item \textbf{parity/charge}: whether the cluster currently contains an odd number of defects.
	\item \textbf{active flag}: whether the site participates in current pass.
	\item \textbf{boundary flag / virtual root}: for planar code boundary matching.
\end{itemize}

\subsubsection{Banking to sustain bandwidth}
UF is bandwidth-dominated. A rule of thumb:
\[
\text{throughput} \approx \frac{\text{banks}\times \text{word\_width}}{\text{reads+writes per site per pass}}.
\]
Common banking schemes:
\begin{itemize}
	\item \textbf{checkerboard (2-bank) banking:} bank = $(x+y+t)\bmod 2$.
	\item \textbf{multi-bank striping:} bank = $\mathrm{idx}\bmod B$.
	\item \textbf{edge vs node separation:} store node arrays and edge arrays in separate RAMs.
\end{itemize}

\begin{figure}[t]
	\centering
	\begin{tikzpicture}[scale=1.0, line cap=round, line join=round]
		\def\LX{6}
		\def\LY{4}
		\foreach \x in {0,...,5}{
			\foreach \y in {0,...,3}{
				\pgfmathtruncatemacro{\b}{mod(\x+\y,2)}
				\ifnum\b=0
				\fill[gray!20] (\x,\y) rectangle (\x+1,\y+1);
				\else
				\fill[gray!55] (\x,\y) rectangle (\x+1,\y+1);
				\fi
				\draw[gray!70] (\x,\y) rectangle (\x+1,\y+1);
			}
		}
		\node[anchor=west] at (6.4,2.8) {\small Bank 0 (light)};
		\node[anchor=west] at (6.4,2.3) {\small Bank 1 (dark)};
		\node[align=center] at (3.0,4.6) {\small checkerboard banking reduces neighbor conflicts};
	\end{tikzpicture}
	\caption{Checkerboard banking example: neighboring sites fall into different banks, enabling simultaneous neighbor reads in local-update passes.}
	\label{fig:uf-banking}
\end{figure}
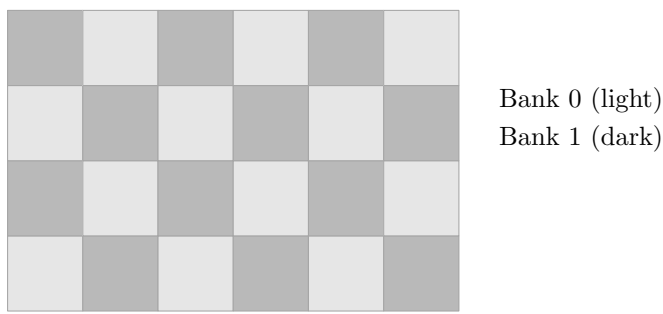

\subsubsection{Avoiding bank conflicts (what actually causes stalls)}
Conflicts occur when, in the same cycle, you need two accesses to the same bank/port.
Typical conflict sources in UF:
\begin{itemize}
	\item multiple \texttt{find} operations chasing pointers into the same hot root,
	\item simultaneous unions trying to update the same parent entry,
	\item growth steps reading multiple neighbors that land in same bank due to poor mapping.
\end{itemize}
Mitigations:
\begin{itemize}
	\item cap the number of \texttt{find} per cycle and pipeline them,
	\item use union arbitration (one writer per root per cycle),
	\item store a small root-cache (recent roots) in registers,
	\item precompute neighbor indices to avoid address-gen bottlenecks.
\end{itemize}

\subsection{Pipeline stages: grow, merge, peel (surface-code UF viewpoint)}

\subsubsection{Stage A: GROW (cluster expansion)}
Each defect starts as a singleton cluster. Clusters expand outward layer-by-layer.
In surface-code UF decoding, a cluster grows until it becomes \emph{even} (its defects are paired internally)
or it hits a boundary/another cluster.

State per cluster root typically includes a \texttt{charge} bit:
\[
\texttt{charge[root]} = (\#\text{defects in cluster}) \bmod 2.
\]
Growth rule (conceptual):
\begin{itemize}
	\item for each active boundary edge of a cluster, attempt to claim the neighboring node/edge,
	\item if two clusters attempt the same site, record a \emph{merge candidate}.
\end{itemize}

\subsubsection{Stage B: MERGE (resolve collisions deterministically)}
When clusters meet, we union their roots.
To make merges deterministic:
\begin{itemize}
	\item define a strict total order on roots (e.g.\ smaller \texttt{root\_id} wins),
	\item or use (rank, root\_id) lexicographic order.
\end{itemize}
Then in hardware:
\begin{itemize}
	\item collect merge requests into a small FIFO,
	\item arbitrate so each root participates in at most one union per cycle,
	\item perform union updates (parent, rank/size, charge) in a writeback stage.
\end{itemize}

\subsubsection{Stage C: PEEL (produce correction edges)}
After clusters are grown/merged to resolve parity, we must output a correction.
In UF surface-code decoding, peeling often means:
\begin{itemize}
	\item identify a spanning forest inside each cluster,
	\item traverse from leaves inward, selecting edges that connect defects appropriately,
	\item output edges that form the correction chain (or equivalent Pauli-frame update).
\end{itemize}

\begin{figure}[t]
	\centering
	\begin{tikzpicture}[scale=1.0, line cap=round, line join=round]
		\node (in) at (0,0) [draw, rounded corners, minimum width=2.3cm, minimum height=0.9cm] {\small defect stream};
		\node (grow) at (3.2,0) [draw, rounded corners, minimum width=2.3cm, minimum height=0.9cm] {\small GROW pass};
		\node (merge) at (6.4,0) [draw, rounded corners, minimum width=2.6cm, minimum height=0.9cm] {\small MERGE arbiter};
		\node (peel) at (9.8,0) [draw, rounded corners, minimum width=2.3cm, minimum height=0.9cm] {\small PEEL pass};
		\node (out) at (12.8,0) [draw, rounded corners, minimum width=2.5cm, minimum height=0.9cm] {\small correction output};
		
		\draw[->] (in) -- (grow);
		\draw[->] (grow) -- (merge);
		\draw[->] (merge) -- (peel);
		\draw[->] (peel) -- (out);
		
		\node (ram1) at (3.2,-1.4) [draw, minimum width=3.4cm, minimum height=0.8cm] {\small banked RAM: parent/rank/charge};
		\draw[->] (grow.south) -- (ram1.north);
		\draw[->] (merge.south) -- (ram1.north);
		
		\node (ram2) at (9.8,-1.4) [draw, minimum width=3.4cm, minimum height=0.8cm] {\small banked RAM: cluster forest/edges};
		\draw[->] (peel.south) -- (ram2.north);
		
		\node[align=center] at (6.4,1.1) {\small fixed-pass microarchitecture (deterministic latency)};
	\end{tikzpicture}
	\caption{A practical UF decoder pipeline: streaming detection events feed fixed-pass GROW/MERGE/PEEL stages, backed by banked RAM. The MERGE stage arbitrates write hazards deterministically.}
	\label{fig:uf-pipeline}
\end{figure}
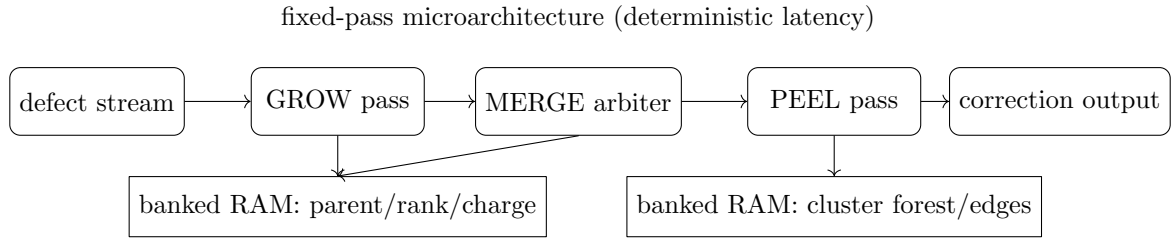

\subsection{Scheduling: clocks, hazards, and convergence detection}

\subsubsection{Clocked schedule (one workable template)}
Assume one lattice sweep per pass. For each pass:
\begin{enumerate}
	\item \textbf{Read phase:} fetch local neighborhood state (parents, charges, flags).
	\item \textbf{Compute phase:} decide growth claims / merge requests / peel actions.
	\item \textbf{Arbitrate phase:} resolve conflicts (one writer per site/root).
	\item \textbf{Writeback phase:} update RAM deterministically.
\end{enumerate}

\subsubsection{Common hazards and fixes}
\begin{itemize}
	\item \textbf{RAW hazard (read-after-write):} reading a parent that is updated in same pass.
	\begin{itemize}
		\item Fix: double-buffer parent pointers per pass, or enforce writeback happens after all reads (two-phase sweep).
	\end{itemize}
	\item \textbf{WAW hazard (write-after-write):} two unions want to write \texttt{parent[root]}.
	\begin{itemize}
		\item Fix: merge arbiter with per-root lock bits; losers retry next cycle/pass.
	\end{itemize}
	\item \textbf{Bank conflict:} two accesses to same RAM port in same cycle.
	\begin{itemize}
		\item Fix: banking + access scheduler; limit degree of parallelism to what ports sustain.
	\end{itemize}
\end{itemize}

\subsubsection{Convergence detection (bounded)}
A safe bounded scheme:
\begin{itemize}
	\item maintain a per-pass \texttt{dirty} flag that ORs any state update,
	\item after each pass, if \texttt{dirty}=0 for two consecutive passes, declare fixed point,
	\item still enforce an absolute max pass count $P_{\max}$.
\end{itemize}
The two-pass rule avoids transient zero-update artifacts caused by pipeline bubbles.

\subsection{Worst-case analysis: cycles, passes, and p99 control}

\subsubsection{What drives worst-case time}
Worst-case work scales with:
\begin{itemize}
	\item the maximum cluster radius required before all odd clusters become even,
	\item the number of merge conflicts (contention) per cycle,
	\item the maximum pointer depth if \texttt{find} is used in-line.
\end{itemize}

\subsubsection{A deterministic upper bound template}
Let:
\[
N = \text{number of sites in window},\quad
E = \text{number of edges considered},\quad
P_{\max} = \text{max passes},\quad
C = \text{cycles per site per pass}.
\]
Then a hard upper bound on cycles is:
\[
T_{\max} \le P_{\max}\cdot C \cdot N + T_{\text{peel}}.
\]
You pick $P_{\max}$ from physics + product constraints (cadence, backlog stability),
and pick $C$ from memory ports and arbitration pipeline depth.

\subsubsection{p99/p999 control strategy}
\begin{itemize}
	\item \textbf{Avoid data-dependent pointer depth in the critical path.}
	\item \textbf{Cap merge retries:} if a root is highly contended, queue and resolve in later pass.
	\item \textbf{Make memory traffic regular:} fixed sweep patterns and fixed numbers of accesses per site.
\end{itemize}

\begin{rem}[Why ``near-linear'' is not enough]
	Asymptotic runtime is not a latency guarantee.
	Real-time decoding needs \emph{upper bounds} that are compatible with the control cycle.
\end{rem}

\subsection{Policy knobs: distance, window, weights, and tie-breaking}

\subsubsection{Distance $d$ and window $W$}
\begin{itemize}
	\item Increasing $d$ increases lattice size ($\sim d^2$ checks per round).
	\item Increasing $W$ increases spacetime volume ($\sim d^2 W$ nodes).
\end{itemize}
Hardware effect: RAM footprint grows; bandwidth pressure rises; pass budget may need retuning.

\subsubsection{Weights and growth rule variants}
Variants differ in how they translate weights into growth:
\begin{itemize}
	\item \textbf{uniform growth:} grow by unweighted radius (fastest, simplest).
	\item \textbf{weighted growth:} approximate weighted distances by multi-rate growth or bucketed weights.
\end{itemize}
If weights are bucketed into $K$ levels, a hardware-friendly scheme is to run $K$ subpasses per growth layer.

\subsubsection{Tie-breaking (determinism requirement)}
Tie-breaking must be fixed and reproducible:
\[
\text{winner} = \arg\min(\texttt{rank},\ \texttt{root\_id}) \quad\text{or}\quad \arg\min(\texttt{root\_id})
\]
(depending on desired behavior).
Never rely on ``first-arrival'' in arbitration unless the arrival order is itself deterministic under backpressure.

\subsection{Exercises (design a pipeline)}

\begin{exercise}[UF in hardware: bounded \texttt{find}]
	Suppose \texttt{parent} pointers can form chains of length up to $L$.
	Design a bounded alternative to full path compression using exactly $P$ pointer-jump passes.
	Explain what each pass does and why the resulting maximum depth shrinks exponentially in $P$.
\end{exercise}
\noindent\textbf{Solution.}
A pointer-jump pass replaces each node's parent by its grandparent:
\[
\texttt{parent}[i] \leftarrow \texttt{parent}[\texttt{parent}[i]].
\]
After one pass, a chain of length $L$ becomes length about $\lceil L/2\rceil$,
because every node skips one level.
After $P$ passes, the effective chain length is at most $\lceil L/2^P\rceil$.
Thus choosing $P\ge \lceil \log_2 L\rceil$ guarantees near-flat trees.
In RTL, implement each pass as a full sweep:
read \texttt{parent}[i] and \texttt{parent[parent[i]]}, then write back the updated parent.
To avoid RAW hazards, use a two-phase update (read old array, write new array) or double buffering.

\begin{exercise}[Banking and conflicts]
	Consider a 2D lattice with neighbor reads (N/E/S/W) each cycle.
	Propose a banking function \texttt{bank(idx)} that reduces conflicts for neighbor access,
	and explain one remaining scenario that still causes conflict.
\end{exercise}
\noindent\textbf{Solution.}
A classic choice is checkerboard banking:
\[
\texttt{bank}(x,y,t)=(x+y+t)\bmod 2.
\]
Then each site has neighbors in the opposite bank, so reading a site and one neighbor can be conflict-free across two banks.
A remaining conflict occurs when you attempt to read \emph{two} neighbors that happen to fall in the same bank
under your mapping (e.g.\ with only 2 banks, two different neighbors both lie in bank 1).
To reduce this, increase banks (e.g.\ 4-bank striping) or schedule neighbor reads over two cycles.

\begin{exercise}[Merge arbiter]
	You receive merge requests $(r_a,r_b)$ as a stream (candidate unions) where roots can repeat.
	Design a deterministic arbitration policy that prevents two writes to the same root in one cycle.
	What state must the arbiter maintain?
\end{exercise}
\noindent\textbf{Solution.}
Maintain a per-cycle (or per-epoch) \texttt{locked[root]} bitset initialized to 0.
When a request $(r_a,r_b)$ arrives:
\begin{itemize}
	\item compute the canonical ordered pair $(u,v)$ with $u=\min(r_a,r_b)$, $v=\max(r_a,r_b)$,
	\item accept the request only if \texttt{locked[u]}=0 and \texttt{locked[v]}=0,
	\item upon accept, set both locks to 1 and emit exactly one union update for writeback,
	\item if rejected, push the request into a retry FIFO for the next cycle/pass.
\end{itemize}
This prevents WAW hazards on parent/rank/charge fields. The arbiter must store: the lock bits, a FIFO for rejects,
and optionally a small map to coalesce duplicate requests.

\begin{exercise}[Deterministic pass budget]
	Assume a planar code of distance $d$ and you choose a growth budget $R_{\max}=\lfloor d/2\rfloor$.
	Give a short argument for why this is a reasonable deterministic cap,
	and list one failure mechanism if the true required radius exceeds $R_{\max}$.
\end{exercise}
\noindent\textbf{Solution.}
A surface code of distance $d$ can correct up to about $\lfloor (d-1)/2\rfloor$ errors in the ideal setting;
many decoding arguments use the idea that typical error clusters remain smaller than $d/2$ below threshold.
Thus setting $R_{\max}\approx d/2$ is a physically motivated cap aligned with the notion that
clusters larger than that are already in a high-risk regime for logical failure.
A failure mechanism when the needed radius exceeds $R_{\max}$ is that odd clusters may remain unmatched (unresolved charge),
or they may be forced into suboptimal boundary matchings, increasing the chance that the residual cycle corresponds to a logical operator.

\begin{exercise}[End-to-end microarchitecture sketch]
	Draw (TikZ block diagram) a minimal UF decoder datapath including:
	input FIFO (detection events), banked RAM (parent/rank/charge), merge arbiter, and output FIFO.
	Label the main control signals that must exist to enforce determinism.
\end{exercise}
\noindent\textbf{Solution.}
A correct block diagram must include:
\begin{itemize}
	\item input FIFO valid/ready (backpressure),
	\item pass counter (bounds $P_{\max}$),
	\item bank select/address gen,
	\item merge arbiter locks + retry FIFO,
	\item dirty flag accumulation per pass,
	\item output FIFO valid/ready.
\end{itemize}
The key determinism signals are: pass counter, fixed sweep address order, arbitration rule (tie-break), and $P_{\max}$ cap.

\begin{exercise}[Policy knob: weighted vs.\ unweighted growth]
	Suppose edges have two weight levels: \texttt{light} and \texttt{heavy}.
	Propose a hardware-friendly growth rule that approximates weighted shortest paths
	without running Dijkstra, using only a small fixed number of subpasses.
\end{exercise}
\noindent\textbf{Solution.}
One approach is bucketed growth:
\begin{itemize}
	\item In each growth layer, run two subpasses:
	(1) expand across all \texttt{light} edges,
	(2) expand across \texttt{heavy} edges only if no \texttt{light} expansion was possible at that frontier
	(or expand heavy every $k$ layers).
\end{itemize}
This approximates weighted distance by making heavy edges effectively ``slower''.
In RTL, implement as two masks and two sweeps with identical memory access patterns.
Because the number of subpasses is fixed, latency remains deterministic.

\subsubsection*{RTL-facing checklist (copy into a design doc)}
\begin{itemize}
	\item \textbf{Arrays:} parent/rank(or size)/charge/active/boundary; word widths and packing.
	\item \textbf{Banking:} bank function, ports per bank, worst-case conflict scenario.
	\item \textbf{Pass schedule:} number of passes $P_{\max}$, cycles per site per pass.
	\item \textbf{Arbitration:} tie-break rule, locks, retry FIFO depth.
	\item \textbf{Determinism:} fixed sweep order, fixed cap, fixed rule for early-exit proof.
	\item \textbf{Instrumentation:} per-pass dirty flag, conflict counters, maximum queue depth.
\end{itemize}
	
\section{Verification and Testbenches for Decoders (Correctness Before Speed)}
\label{sec:verify}

\subsection{Objective: trust the decoder output}

Decoder RTL is \emph{control-plane hardware}. A wrong decode is not a ``slightly worse answer'':
it is a potential logical failure. Therefore verification comes \emph{before} throughput tuning.
This chapter builds a verification workflow that answers two questions:

\begin{itemize}
	\item \textbf{Correctness:} Does RTL implement the intended decoding policy (and its failure flags) for all supported modes?
	\item \textbf{Real-time safety:} Under realistic syndrome streams and fault scenarios, does RTL meet bounded-latency behavior
	(including p99/p999) and fail closed when it cannot?
\end{itemize}

\medskip
\noindent\textbf{Guiding principle.}
A decoder is not only a mapping $\Delta s_t \mapsto \text{frame update}$;
it is also a \emph{stream processor} with state, backpressure, and fault-handling rules.
Verification must cover both.

\subsection{Golden model (Python) vs.\ RTL: what must match}

\subsubsection{Golden model roles}
We maintain a \textbf{Python golden model} that is intentionally simple and readable.
It plays three roles:
\begin{enumerate}
	\item \textbf{Functional oracle:} for a given input stream, it produces the expected output stream.
	\item \textbf{Invariant checker:} it computes internal invariants (syndrome parity, boundary consistency, etc.).
	\item \textbf{Reference for metrics:} it computes correctness statistics (logical failure in simulation, mismatch rates, etc.).
\end{enumerate}

\subsubsection{What ``match'' means (define your contract)}
You must decide \emph{exactly} what fields are required to match cycle-by-cycle.
For streaming decoders, a practical contract is:

\begin{itemize}
	\item \textbf{Output decision equivalence:}
	the \emph{Pauli-frame update} (or correction summary) produced by RTL must equal the golden model's
	for the same configuration and the same interpreted input stream.
	
	\item \textbf{Timing semantics:}
	outputs must appear at the correct cycle offset (fixed pipeline latency or bounded-latency window),
	and the RTL must follow the specified backpressure/valid-ready protocol.
	
	\item \textbf{Fault-handling policy:}
	for corrupted inputs (drops, bursts, misalignment), RTL must:
	(i) raise the agreed flags,
	(ii) follow the agreed degradation policy (erasures or ``fail closed''),
	and (iii) keep stream alignment guarantees.
	
	\item \textbf{Determinism under tie-breaking:}
	if the algorithm has ties (equal weights, simultaneous merges, etc.), the contract must specify
	deterministic tie-break rules so Python and RTL agree.
\end{itemize}

\subsubsection{Canonical stream interface}
Assume each cycle emits a record:
\[
\texttt{pkt}_t = (\texttt{cfg\_id},\, t,\, \Delta s_t,\, \texttt{valid},\, \texttt{meta}),
\]
and the decoder outputs
\[
\texttt{out}_t = (\texttt{cfg\_id},\, t,\, \texttt{frame\_delta}_t,\, \texttt{flags}_t,\, \texttt{valid}).
\]
The testbench must verify:
\begin{itemize}
	\item config ID consistency,
	\item time index monotonicity and alignment,
	\item output validity rules,
	\item field-by-field equality where required.
\end{itemize}

\subsection{Property-based tests: random syndromes and invariants}

\subsubsection{Why property-based tests}
Hand-written unit tests catch ``known bugs''. Property-based tests catch
``unknown unknowns'' by generating many diverse cases and checking \emph{invariants}
that must always hold.

\subsubsection{Two layers of properties}
\paragraph{(A) Stream-level invariants (protocol correctness).}
These should hold regardless of the specific code or noise model:
\begin{itemize}
	\item \textbf{No output without input:} if no valid input is accepted, no valid output should be emitted.
	\item \textbf{Monotonic time:} output time indices do not go backward.
	\item \textbf{Bounded buffering:} FIFO occupancy never exceeds the stated capacity (or overflow flag must assert).
	\item \textbf{Reset behavior:} after reset, internal state returns to a canonical state within $K$ cycles.
\end{itemize}

\paragraph{(B) Decode-level invariants (algorithmic soundness).}
These use the topology--graph decoding logic:
\begin{itemize}
	\item \textbf{Parity / boundary consistency:} the implied correction must match the boundary condition:
	for the intended decode graph, $\partial \widehat{E} = \Delta s$ (in the noiseless single-round model),
	or the appropriate spacetime consistency relation in the 3D model.
	\item \textbf{Stabilizer equivalence:} if two inputs differ by a stabilizer-only disturbance (model-dependent),
	outputs should be equivalent up to the defined gauge.
	\item \textbf{Idempotence on zero syndrome:} if $\Delta s_t=0$ for a long stretch, the decoder should not drift.
\end{itemize}

\subsubsection{Random generators (what to randomize)}
Property tests are only as good as their generators. Useful randomized knobs:
\begin{itemize}
	\item code distance $d$ (within supported range),
	\item window length $W$ (if streaming spacetime),
	\item defect density (sparse to adversarial),
	\item clustered vs.\ uniform defect patterns,
	\item correlated bursts of defects (stress locality + merges),
	\item randomized tie situations (equal weights / simultaneous events),
	\item randomized backpressure (ready/valid perturbations).
\end{itemize}

\subsection{Fault injection: missing bits, flipped bits, burst corruption}

\subsubsection{Fault model taxonomy}
In real control pipelines, corruption is not only ``flip a bit''.
Three high-leverage classes:

\begin{itemize}
	\item \textbf{Erasures / drops:} missing packets or missing subfields.
	\item \textbf{Bit flips:} random single-bit or multi-bit flips in $\Delta s_t$ or metadata.
	\item \textbf{Bursts / misalignment:} contiguous corruption that can shift framing, cause off-by-one time alignment,
	or scramble packet boundaries.
\end{itemize}

\subsubsection{Design the expected behavior (fail closed vs.\ degrade gracefully)}
You must specify what RTL should do. Common policies:

\begin{itemize}
	\item \textbf{Degrade gracefully (erasures):} treat missing data as ``unknown'' and emit an \texttt{ERASURE} flag,
	possibly holding outputs constant (no frame update) for that cycle.
	
	\item \textbf{Fail closed (bursts/misalignment):} assert \texttt{FATAL} or \texttt{DESYNC} and transition to a safe state:
	stop emitting frame updates until re-synchronized or reset.
\end{itemize}

\subsubsection{A self-contained fault-injection diagram (TikZ)}
\begin{figure}[t]
	\centering
	\begin{tikzpicture}[scale=1.0, line cap=round, line join=round, every node/.style={align=center}]
		\node (src) at (0,0)  [draw, rounded corners, minimum width=3.0cm, minimum height=0.95cm]
		{\small ideal $\Delta s_t$ stream};
		\node (inj) at (4.2,0) [draw, rounded corners, minimum width=3.4cm, minimum height=0.95cm]
		{\small fault injector\par\small flip/drop/burst};
		\node (rtl) at (8.9,0) [draw, rounded corners, minimum width=3.3cm, minimum height=0.95cm]
		{\small decoder RTL};
		\node (mon) at (13.4,0) [draw, rounded corners, minimum width=3.8cm, minimum height=0.95cm]
		{\small monitors\par\small flags + invariants};
		
		\draw[->] (src) -- (inj);
		\draw[->] (inj) -- (rtl);
		\draw[->] (rtl) -- (mon);
		
		\node at (8.9,1.25) {\small verify output \emph{and} fault-handling policy};
	\end{tikzpicture}
	\caption{Fault injection loop: corrupt the syndrome stream in controlled ways and verify that RTL either
		degrades gracefully (erasures) or fails closed with explicit flags (bursts, misalignment).}
	\label{fig:fault-injection-verify}
\end{figure}
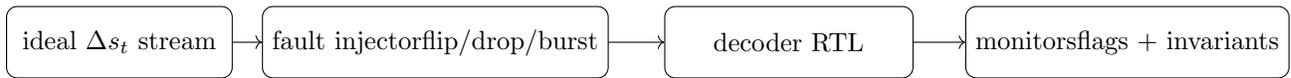

\subsubsection{Fault injection patterns (concrete)}
Let $\Delta s_t$ be a bit-vector.

\paragraph{(1) Flip.}
Pick random indices $i$ and apply
\[
\Delta s_t[i] \leftarrow \Delta s_t[i] \oplus 1.
\]
Check that:
\begin{itemize}
	\item RTL flags \texttt{CORRUPT} only if your policy requires detection,
	\item functional mismatch against golden model is expected unless injector also informs golden model,
	\item the decoder does not deadlock.
\end{itemize}

\paragraph{(2) Drop.}
Remove a packet at time $t$.
Testbench must define what the decoder sees: no \texttt{valid} that cycle, or \texttt{valid} with \texttt{ERASURE} flag.
Check:
\begin{itemize}
	\item time alignment remains consistent,
	\item FIFO behavior is correct,
	\item output policy is followed (hold/no-op vs.\ explicit erasure).
\end{itemize}

\paragraph{(3) Burst / desync.}
Corrupt $B$ consecutive packets or shift framing by one.
Check:
\begin{itemize}
	\item \texttt{DESYNC} or \texttt{FATAL} asserted within bounded cycles,
	\item frame updates stop (fail closed),
	\item recovery procedure works (explicit resync marker or reset).
\end{itemize}

\subsection{Regression suites: reproducibility across parameter knobs}

\subsubsection{Why regression matters}
Decoders have many ``product knobs'': distance $d$, window $W$, weighting model, tie-break rules,
clock rate, FIFO depth, etc. A change that helps one regime can break another.
Regression turns ``it works on my case'' into ``it works on the supported envelope''.

\subsubsection{Organize regression by tiers}
A practical layout:

\begin{itemize}
	\item \textbf{Tier 0 (smoke):} tiny $d$ values, short streams, quick compile + run.
	\item \textbf{Tier 1 (functional):} moderate $d$, randomized seeds, checks invariants + exact match to golden.
	\item \textbf{Tier 2 (stress):} adversarial densities, clustered defects, heavy merges, randomized backpressure.
	\item \textbf{Tier 3 (fault):} flip/drop/burst campaigns with policy assertions.
\end{itemize}

\subsubsection{Reproducibility contract}
Every regression run must log:
\begin{itemize}
	\item \textbf{Configuration:} $(d,W)$, weight model ID, tie-break ID, build hash.
	\item \textbf{Randomness:} seed, generator version.
	\item \textbf{Stream length:} number of cycles, rate, backpressure pattern seed.
	\item \textbf{Outcome:} pass/fail + first failing cycle index + minimal counterexample dump.
\end{itemize}

\subsection{Measuring correctness vs.\ latency tradeoffs}

\subsubsection{Two axes, one dashboard}
Verification is not separate from performance.
You want a \emph{joint} view:

\begin{itemize}
	\item \textbf{Correctness metrics:}
	mismatch rate vs.\ golden, logical-failure rate (in simulation), invariant violations, fault-policy violations.
	\item \textbf{Latency metrics:}
	per-cycle service time, pipeline latency, FIFO occupancy, deadline-miss rate, p99/p999 spikes.
\end{itemize}

\subsubsection{Instrumentation signals you should expose in RTL}
Expose lightweight counters (readable in simulation and optionally in hardware):

\begin{itemize}
	\item cycle counter, accepted-input counter, emitted-output counter,
	\item FIFO occupancy min/max, overflow events,
	\item stall cycles (backpressure), hazard stalls (bank conflict),
	\item phase counters (UF passes, merge iterations),
	\item flags histogram: \texttt{ERASURE}, \texttt{CORRUPT}, \texttt{DESYNC}, \texttt{FATAL}.
\end{itemize}

\subsubsection{Interpreting p99/p999 spikes}
Spikes often come from:
\begin{itemize}
	\item worst-case cluster merges (UF) or worst-case matching steps (MWPM),
	\item memory hot spots / bank conflicts,
	\item backpressure cascades: output blocked $\Rightarrow$ FIFO fills $\Rightarrow$ deadline misses,
	\item fault-policy paths (safe-state transitions).
\end{itemize}
Verification must confirm these paths are \emph{bounded} and that flags accurately mark unsafe regimes.

\subsection{Exercises (write a test plan)}

\begin{exercise}[Define the golden contract]
	Write a one-page conformance contract that states precisely:
	(i) which output fields must match cycle-by-cycle,
	(ii) allowed pipeline latency (fixed $L$ or bounded range),
	(iii) fault-policy behavior for drop/flip/burst,
	(iv) tie-break rules.
\end{exercise}

\noindent\textbf{Solution sketch.}
A strong answer includes:
\begin{itemize}
	\item a table of input/output records and field semantics,
	\item timing diagram: \texttt{valid/ready} handshake and accepted vs.\ observed cycles,
	\item explicit policy table: \texttt{event} $\rightarrow$ \texttt{flag} $\rightarrow$ \texttt{output action},
	\item determinism notes: ordering of merges, stable priority rules.
\end{itemize}

\begin{exercise}[Property-based invariants]
	Propose three stream-level invariants and three decode-level invariants for your decoder.
	For each, state how to check it in a testbench and what a counterexample dump should contain.
\end{exercise}

\noindent\textbf{Solution sketch.}
Good invariants include:
\begin{itemize}
	\item stream-level: monotonic time index, no output without accepted input, bounded FIFO occupancy;
	\item decode-level: boundary consistency on noiseless cases, idempotence on zero syndrome, stable behavior under reset.
\end{itemize}
Counterexample dumps should include: seed, cfg ID, first failing cycle, last $K$ packets, internal counters, flags.

\begin{exercise}[Fault injection campaign design]
	Design a fault injection campaign with:
	(i) flip rate $\alpha$, (ii) drop rate $\beta$, (iii) burst length distribution.
	Specify what flags you expect and what ``pass'' means for each fault type.
\end{exercise}

\noindent\textbf{Solution sketch.}
A complete answer specifies:
\begin{itemize}
	\item flip: mismatch expected unless golden also sees corruption; check bounded latency + no deadlock;
	\item drop: expect \texttt{ERASURE} and no frame update (or specified behavior) + alignment preserved;
	\item burst: expect \texttt{DESYNC}/\texttt{FATAL} within bounded cycles + fail-closed behavior.
\end{itemize}

\begin{exercise}[Regression matrix]
	Choose a supported envelope (e.g.\ $d\in\{5,7,9,11\}$, $W\in\{3,5,7\}$).
	Propose a regression matrix of tests across $(d,W)$ with tiers 0--3, including run-time budgets.
\end{exercise}

\noindent\textbf{Solution sketch.}
A good matrix balances coverage and cost:
tier 0 for all $(d,W)$, tier 1 for a subset with many seeds, tier 2 stress on largest $(d,W)$,
tier 3 fault tests on representative points.

\subsection{Conformance specification (golden model contract)}

\subsubsection{A minimal conformance template (copy-paste)}
Use the following as a starting point. Fill in the blanks for your implementation.

\begin{rem}
	\textbf{Conformance Contract: Decoder RTL vs.\ Python Golden Model}
	\begin{itemize}
		\item \textbf{Config identity.} Each packet carries \texttt{cfg\_id}. RTL must emit the same \texttt{cfg\_id}.
		\item \textbf{Time identity.} Input time index $t$ is preserved. Output time index equals $t$ with pipeline offset $L=\underline{\hspace{1.2cm}}$ (fixed) or $L\in[\underline{\hspace{0.6cm}},\underline{\hspace{0.6cm}}]$ (bounded).
		\item \textbf{Valid-ready semantics.} A packet is ``accepted'' iff \texttt{in\_valid \& in\_ready}.
		Only accepted packets may produce outputs.
		\item \textbf{Functional output.} For accepted packets without fault flags, RTL output
		\texttt{frame\_delta} must equal golden output under the same tie-break ID.
		\item \textbf{Tie-break.} When weights are equal, priority order is:
		\underline{\hspace{6cm}}.
		\item \textbf{Fault policy.}
		\begin{itemize}
			\item Drop: raise \texttt{ERASURE} and action \underline{\hspace{4cm}}.
			\item Flip: raise \texttt{CORRUPT} only if detectable; action \underline{\hspace{4cm}}.
			\item Burst/desync: raise \texttt{DESYNC}/\texttt{FATAL} within $\underline{\hspace{0.8cm}}$ cycles and stop frame updates until \underline{\hspace{3.5cm}}.
		\end{itemize}
		\item \textbf{Bounded-latency guarantee.} For supported $(d,W)$ and within operating defect density
		$\le \underline{\hspace{1.2cm}}$, the decoder must meet deadline $D=\underline{\hspace{1.0cm}}$ cycles with p99 $\le \underline{\hspace{1.0cm}}$.
	\end{itemize}
\end{rem}

\subsubsection{What the testbench must implement}
A correct testbench must:
\begin{enumerate}
	\item generate ideal streams (and optional injected streams),
	\item run golden and RTL under identical configuration and seeds,
	\item align cycles by accepted/observed semantics (not by wall-clock),
	\item compare outputs and flags per the contract,
	\item dump minimal counterexamples (first failure + last $K$ packets + internal counters),
	\item record latency distributions and flag histograms.
\end{enumerate}
	
\section{Host Interface and Control-Stack Integration (Syndrome In, Correction Out)}
\label{sec:interface}

\subsection{Objective: make the decoder a drop-in infrastructure component}

A surface-code decoder is useful only if it can be \emph{plugged into} a real control stack:
syndrome bits arrive continuously, the decoder must keep stream alignment, and corrections
must be returned with predictable latency and explicit safety flags.

\medskip
\noindent\textbf{Design target.}
Treat the decoder as a \emph{streaming appliance} with:
\begin{itemize}
	\item a \textbf{well-defined wire format} for syndrome packets,
	\item \textbf{bounded buffering} + backpressure,
	\item \textbf{deterministic output timing} (fixed or bounded latency),
	\item \textbf{versioned messages} and reproducible configuration IDs,
	\item \textbf{fail-closed behavior} under desync and overload.
\end{itemize}

\subsection{Streaming model: packetization, framing, and timestamps}

\subsubsection{The stream is the product}
Syndrome data is not a batch file. It is a \emph{clocked stream}:
each round $t$ produces a detection-event vector (often $\Delta s_t$) that must be decoded online.

\medskip
\noindent We model the input as packets:
\[
\texttt{pkt}_t = (\texttt{hdr},\, \texttt{payload},\, \texttt{trailer}),
\]
where:
\begin{itemize}
	\item \texttt{hdr} contains routing, framing, and time identity,
	\item \texttt{payload} contains $\Delta s_t$ bits (and optional metadata),
	\item \texttt{trailer} contains integrity checks (CRC) and end markers.
\end{itemize}

\subsubsection{Framing: how the decoder knows ``where one round ends''}
Framing must be explicit. Common choices:
\begin{itemize}
	\item \textbf{Fixed-length frames:} each round has exactly $B$ bytes of payload. Simplest for FPGA.
	\item \textbf{Length-delimited frames:} a \texttt{len} field tells payload size; supports multiple code distances.
	\item \textbf{Marker-based:} special sync words; robust but must handle marker corruption.
\end{itemize}

\medskip
\noindent\textbf{Recommended for real-time:}
fixed-length frames per configuration (each \texttt{cfg\_id} implies a fixed payload length),
plus periodic sync markers (e.g.\ every $K$ rounds) to detect drift.

\subsubsection{Timestamps: which time do we mean?}
You usually need \emph{two} times:
\begin{itemize}
	\item \textbf{Round index} $t$ (logical time): which QEC cycle this packet belongs to.
	\item \textbf{Arrival time} (transport time): when host/FPGA observed the packet.
\end{itemize}

\medskip
\noindent The decoder primarily uses $t$ for correctness (alignment), while the host uses arrival time
for latency accounting and congestion diagnostics.

\subsubsection{A concrete header}
A minimal header:
\[
\texttt{hdr} =
(\texttt{magic},\, \texttt{version},\, \texttt{cfg\_id},\, t,\, \texttt{seq},\, \texttt{flags}).
\]
Interpretation:
\begin{itemize}
	\item \texttt{magic}: sync word (detect framing loss),
	\item \texttt{version}: schema version,
	\item \texttt{cfg\_id}: configuration identifier (distance/window/policy/build hash),
	\item $t$: round index,
	\item \texttt{seq}: transport sequence number (detect drops),
	\item \texttt{flags}: e.g.\ ``erasure present'' or ``data invalid''.
\end{itemize}

\subsection{Buffers and backpressure: avoiding overflow and stale decisions}

\subsubsection{Why buffering exists}
Even with deterministic RTL, the environment is not deterministic:
PCIe bursts, DMA jitter, host scheduling, and occasional decoder slow paths (worst-case merges).
Therefore you need buffers.

\subsubsection{Two separate buffers (do not mix them)}
\begin{itemize}
	\item \textbf{Ingress FIFO:} holds incoming $\Delta s_t$ packets before decode.
	\item \textbf{Egress FIFO:} holds outgoing corrections/flags before host consumption.
\end{itemize}

\subsubsection{Backpressure: the core control law}
Backpressure is the explicit mechanism that prevents overflow:
\begin{itemize}
	\item \textbf{Ready/valid (streaming bus):} decoder asserts \texttt{in\_ready}.
	\item \textbf{Credits (host link):} host sends $C$ credits; FPGA decrements per packet.
\end{itemize}

\medskip
\noindent\textbf{Rule: never silently drop.}
If you must drop, emit a \texttt{DROP} / \texttt{ERASURE} event and record it in flags and counters.

\subsubsection{Staleness: when an answer is worse than no answer}
In QEC control loops, a correction that arrives too late may be harmful.
Define a \emph{staleness window} $S$ (in rounds).
If output is older than $S$ relative to the host's ``current'' round, then:
\begin{itemize}
	\item either discard it with a \texttt{STALE} flag,
	\item or treat it as a Pauli-frame update that is still consistent (depends on your control design).
\end{itemize}

\subsubsection{A buffer-state visualization (self-contained TikZ)}
\begin{figure}[t]
	\centering
	\begin{tikzpicture}[
		font=\small,
		>=Latex,
		node distance=10mm and 14mm,
		box/.style={
			draw, thick, rounded corners,
			minimum width=3.1cm, minimum height=0.95cm,
			align=center
		},
		arr/.style={-Latex, thick},
		note/.style={font=\footnotesize, align=center}
		]
		\node[box] (host) {host\\\footnotesize (DAQ + control)};
		\node[box, right=of host] (in) {ingress FIFO\\\footnotesize (syndrome)};
		\node[box, right=of in] (dec) {decoder core\\\footnotesize (bounded phases)};
		
		\node[box, below=12mm of dec] (out) {egress FIFO\\\footnotesize (correction)};
		\node[box, right=of out] (ctrl) {control\\\footnotesize (Pauli frame)};
		
		\draw[arr] (host) -- node[above]{\footnotesize packets $\Delta s_t$} (in);
		\draw[arr] (in)   -- node[above]{\footnotesize dequeue} (dec);
		
		\draw[arr] (dec)  -- node[right]{\footnotesize updates + flags} (out);
		\draw[arr] (out)  -- node[above]{\footnotesize DMA/reads} (ctrl);
		
		\draw[arr] (in.north) to[out=120,in=60]
		node[above]{\footnotesize backpressure / credits} (host.north);
		
		\node[note, above=7mm of dec, xshift=-18mm]
		{overflow $\Rightarrow$ explicit flags, not silent loss};
		
	\end{tikzpicture}
	\caption{Integration view: ingress/egress buffering isolates transport jitter from deterministic decode.
		Backpressure/credits prevent overflow. Late outputs must be handled explicitly (stale policy).}
	\label{fig:buffers-backpressure}
\end{figure}
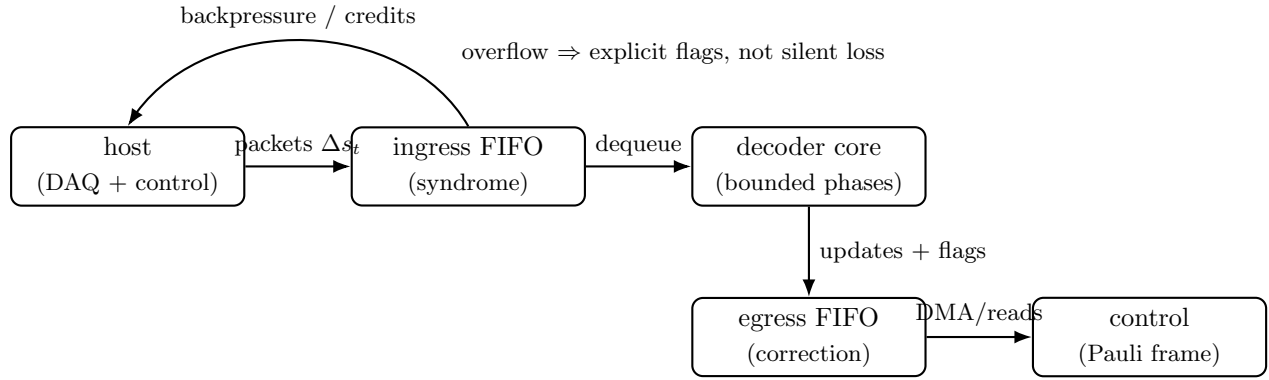

\subsection{Minimal API: C/Python stubs and command protocol}

\subsubsection{Two planes: data plane vs.\ control plane}
\begin{itemize}
	\item \textbf{Data plane (hot path):} send $\Delta s_t$ packets, receive corrections.
	\item \textbf{Control plane (cold path):} set config, query status, read counters, reset.
\end{itemize}

\subsubsection{Minimal command set}
Define a small set of commands (opcodes). Example:
\begin{itemize}
	\item \texttt{SET\_CFG(cfg\_blob)}: load configuration (distance/window/weights/tie-break).
	\item \texttt{START()} / \texttt{STOP()}: enable/disable streaming.
	\item \texttt{PUSH(pkt)}: enqueue one syndrome packet (or start DMA descriptor).
	\item \texttt{PULL()} : dequeue one output packet (or poll DMA completion).
	\item \texttt{GET\_STATUS()}: health flags, FIFO occupancy, last round index.
	\item \texttt{GET\_COUNTERS()}: p99 bins, stalls, overflow events, fault flags histogram.
	\item \texttt{RESET()}: reset state machines and FIFOs safely.
\end{itemize}

\subsubsection{C API skeleton (header-only style)}
\begin{verbatim}
	typedef struct {
		uint32_t magic;
		uint16_t version;
		uint16_t hdr_bytes;
		uint64_t cfg_id;
		uint32_t round_t;
		uint32_t seq;
		uint32_t flags;
		uint32_t payload_bytes;
	} dec_hdr_t;
	
	typedef struct {
		dec_hdr_t hdr;
		uint8_t   payload[DEC_MAX_PAYLOAD];
		uint32_t  crc32;
	} dec_pkt_t;
	
	typedef struct {
		uint64_t cfg_id;
		uint32_t round_t;
		uint32_t flags;        // ERASURE/CORRUPT/DESYNC/STALE/...
		uint32_t frame_words;  // compact correction / frame delta
		uint32_t reserved;
	} dec_out_t;
	
	int dec_set_cfg(const void* cfg_blob, size_t nbytes);
	int dec_start(void);
	int dec_stop(void);
	int dec_push(const dec_pkt_t* pkt);
	int dec_pull(dec_out_t* out);
	int dec_get_status(uint32_t* status_words, size_t nwords);
	int dec_reset(void);
\end{verbatim}

\subsubsection{Python stub (thin wrapper)}
Python should not ``re-implement'' the protocol. It should wrap the same C ABI:
\begin{verbatim}
	# dec.py (conceptual)
	def set_cfg(cfg_blob: bytes) -> None: ...
	def start() -> None: ...
	def push(pkt: bytes) -> None: ...
	def pull() -> dict: ...
	def status() -> dict: ...
	def reset() -> None: ...
\end{verbatim}

\subsubsection{Command protocol detail: avoid ambiguity}
Avoid variable-length ambiguity in hot path:
\begin{itemize}
	\item fixed-size packet header,
	\item explicit payload length,
	\item CRC over (header || payload),
	\item strict version checks (reject unknown versions with a flag).
\end{itemize}

\subsection{On-prem vs.\ rack: deployment models and constraints}

\subsubsection{Model A: ``on-prem'' lab box}
Typical: control computer + FPGA card (PCIe) near the experiment rack.
Constraints:
\begin{itemize}
	\item low transport latency (good),
	\item host OS jitter still exists (must buffer),
	\item limited physical redundancy (single point of failure).
\end{itemize}

\subsubsection{Model B: rack-scale service}
Decoder as a service inside a rack (possibly multiple FPGAs).
Constraints:
\begin{itemize}
	\item network transport adds jitter (must treat as part of p99 budget),
	\item multi-client isolation (strong config ID + authorization),
	\item failover and hot-swap become relevant (state handoff policy).
\end{itemize}

\subsubsection{Rule-of-thumb: what changes}
\begin{itemize}
	\item On-prem: optimize absolute latency.
	\item Rack: optimize tail latency + robustness + observability.
\end{itemize}

\subsection{Integration with experimental control loops}

\subsubsection{Where the decoder sits}
A typical closed-loop view:
\[
\text{qubits} \rightarrow \text{measure} \rightarrow \Delta s_t \rightarrow
\text{decode} \rightarrow \text{Pauli frame / corrective pulses} \rightarrow \text{qubits}.
\]

\subsubsection{Two integration styles}
\paragraph{(1) Pauli frame update (preferred).}
The decoder outputs a \texttt{frame\_delta} that updates a classical record.
Physical pulses are not necessarily applied immediately; instead future interpretation of measurements is updated.
Pros: lower physical actuation requirements; robust to small latencies.

\paragraph{(2) Active correction pulses.}
The decoder outputs explicit correction actions to be physically applied.
Pros: conceptually direct. Cons: latency becomes more stringent; misalignment is dangerous.

\subsubsection{Safety flags as first-class outputs}
The control loop must treat decoder flags as actionable:
\begin{itemize}
	\item \texttt{ERASURE}: treat cycle as unknown; adjust estimator / pause frame update.
	\item \texttt{DESYNC}: stop applying updates; attempt resync or reset.
	\item \texttt{OVERFLOW}: declare invalid interval; mark experiment data.
	\item \texttt{STALE}: ignore or apply under a formal staleness rule.
\end{itemize}

\subsection{Exercises (design an interface)}

\begin{exercise}[Define the packet format]
	Choose fixed-length or length-delimited framing. Specify:
	(i) header fields, (ii) payload layout (bit ordering), (iii) CRC scope,
	(iv) how you detect drops (seq) and desync (magic/version).
\end{exercise}

\noindent\textbf{Solution sketch.}
A strong answer includes:
\begin{itemize}
	\item \texttt{magic} + \texttt{version} + \texttt{cfg\_id} + $t$ + \texttt{seq} + \texttt{payload\_bytes},
	\item explicit bit order: (check index $\rightarrow$ byte $\rightarrow$ bit),
	\item CRC over header+payload,
	\item a desync rule: wrong magic/version $\Rightarrow$ \texttt{DESYNC} and stop stream until resync marker.
\end{itemize}

\begin{exercise}[Backpressure policy]
	Given an ingress FIFO of depth $D_{\text{in}}$, choose one:
	(A) ready/valid backpressure, or (B) credit-based flow control.
	State how the host behaves when the decoder is not ready.
\end{exercise}

\noindent\textbf{Solution sketch.}
A complete answer specifies:
\begin{itemize}
	\item when \texttt{in\_ready} deasserts (occupancy threshold),
	\item whether host blocks, retries, or drops with explicit \texttt{ERASURE},
	\item a monitoring counter: max occupancy and overflow events.
\end{itemize}

\begin{exercise}[Staleness window]
	Assume pipeline latency varies between $L_{\min}$ and $L_{\max}$ cycles.
	Pick a staleness window $S$ and define when an output is considered stale.
	Describe what the control loop does on stale outputs.
\end{exercise}

\noindent\textbf{Solution sketch.}
A good rule is:
\[
\texttt{STALE if } (t_{\text{host}} - t_{\text{out}}) > S,
\]
with a clear action: ignore output and raise a log marker, or apply only if Pauli-frame semantics guarantee safety.

\subsection{Message schema and versioning (wire format)}

\subsubsection{Why versioning is non-negotiable}
Decoder interfaces evolve: new flags, different payload sizes, new weight models.
Without explicit versioning, you get silent mis-decodes.

\subsubsection{Schema rules}
\begin{itemize}
	\item \textbf{Version in every packet.} Reject unknown versions.
	\item \textbf{Backward-compatible extension:} add fields at the end; keep old fields stable.
	\item \textbf{Explicit \texttt{cfg\_id}:} computed from a canonical config blob (hash).
	\item \textbf{Explicit endianness:} define little-endian vs.\ big-endian for multibyte fields.
\end{itemize}

\subsubsection{Wire-format checklist (copy-paste)}
\begin{itemize}
	\item \texttt{magic} (sync word)
	\item \texttt{version} (schema)
	\item \texttt{hdr\_bytes} (future-proofing)
	\item \texttt{cfg\_id} (hash of canonical config)
	\item \texttt{round\_t} (logical time)
	\item \texttt{seq} (drop detection)
	\item \texttt{flags} (input validity / injector markers)
	\item \texttt{payload\_bytes} (length)
	\item \texttt{payload} (packed bits)
	\item \texttt{crc32} (integrity)
\end{itemize}

\medskip
\noindent\textbf{Output wire format should mirror input:}
carry \texttt{cfg\_id}, \texttt{round\_t}, \texttt{flags}, and a compact \texttt{frame\_delta},
plus optional counters for observability.

	
\section{Algebraic Topology Behind Quantum Error-Correcting Codes}
\label{sec:alg-top}

\subsection{Objective and guiding picture}

Surface/toric codes look like ``physics + engineering,'' but the core correctness claim is
\emph{algebraic topology}:

\begin{quote}
	\emph{Errors are chains. Syndromes are boundaries. Logical operators are homology classes.}
\end{quote}

\noindent This section builds the minimum pipeline
\[
\text{cell complex} \;\Rightarrow\; \text{chain complex} \;\Rightarrow\; \text{CSS code}
\;\Rightarrow\; \text{errors/syndromes/logicals as }\partial\text{ and homology}.
\]

\medskip
\noindent\textbf{Guiding picture (one line).}
A CSS code is a pair of incidence relations (two boundary maps)
\[
C_2 \xrightarrow{\ \partial_2\ } C_1 \xrightarrow{\ \partial_1\ } C_0
\qquad\text{with}\qquad \partial_1\partial_2 = 0,
\]
and this single identity is exactly the stabilizer commutation condition.

\subsection{Cell complex $\Rightarrow$ chain complex}

\subsubsection{Start with a 2D cell complex}
Take a square grid on a surface (torus) or on a planar patch with boundaries.
It has:
\begin{itemize}
	\item vertices $V$ (0-cells),
	\item edges $E$ (1-cells),
	\item faces/plaquettes $F$ (2-cells).
\end{itemize}

\subsubsection{Chains are formal sums (work mod 2)}
Work over $\mathbb{Z}_2$ so ``present/absent'' is the only coefficient:
\[
C_0 := \mathbb{Z}_2\langle V\rangle,\quad
C_1 := \mathbb{Z}_2\langle E\rangle,\quad
C_2 := \mathbb{Z}_2\langle F\rangle.
\]
An element of $C_1$ is a subset of edges (a \emph{1-chain}):
\[
c = \sum_{e\in E} x_e\, e,\qquad x_e\in\{0,1\}.
\]

\subsubsection{Boundary maps are incidence matrices}
Fix an orientation on each edge and face (only used to define $\partial$; over $\mathbb{Z}_2$
signs do not matter).

\paragraph{Edge boundary.}
For an edge $e=(u\to v)$,
\[
\partial_1(e)=u+v \in C_0 \quad (\text{mod }2).
\]

\paragraph{Face boundary.}
For a face $p$ with boundary edges $e_1,e_2,e_3,e_4$,
\[
\partial_2(p)=e_1+e_2+e_3+e_4 \in C_1 \quad (\text{mod }2).
\]

\paragraph{The key identity.}
Every face boundary is a closed loop, so its boundary is empty:
\[
\partial_1(\partial_2(p))=0\quad \Rightarrow\quad \partial_1\partial_2=0.
\]

\subsubsection{Visualization: cells and boundaries (self-contained TikZ)}
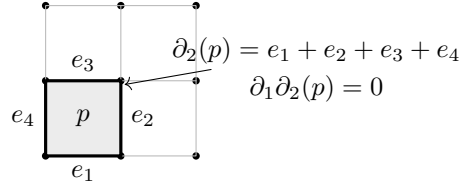
\begin{figure}[t]
	\centering
	\begin{tikzpicture}[scale=1.0, line cap=round, line join=round]
		\foreach \x in {0,1,2} {
			\foreach \y in {0,1,2} {
				\fill (\x,\y) circle (1.4pt);
			}
		}
		\draw[gray!55] (0,0)--(2,0)--(2,2)--(0,2)--cycle;
		\draw[gray!55] (1,0)--(1,2);
		\draw[gray!55] (0,1)--(2,1);
		
		\fill[gray!15] (0,0) rectangle (1,1);
		\draw[very thick] (0,0)--(1,0)--(1,1)--(0,1)--cycle;
		
		\node at (0.5,0.5) {\small $p$};
		
		\node[below] at (0.5,0) {\small $e_1$};
		\node[right] at (1,0.5) {\small $e_2$};
		\node[above] at (0.5,1) {\small $e_3$};
		\node[left] at (0,0.5) {\small $e_4$};
		
		\node[align=center] at (3.6,1.4) {\small $\partial_2(p)=e_1+e_2+e_3+e_4$};
		\node[align=center] at (3.6,0.9) {\small $\partial_1\partial_2(p)=0$};
		\draw[->] (2.2,1.15) -- (1.05,0.95);
	\end{tikzpicture}
	\caption{A 2D cell complex gives chain groups $C_2,C_1,C_0$ and boundary maps $\partial_2,\partial_1$.
		Over $\mathbb{Z}_2$, a plaquette boundary is the sum of its four edges. The identity $\partial_1\partial_2=0$
		is the algebraic core behind stabilizer commutation.}
	\label{fig:cell-to-chain}
\end{figure}

\subsection{From chain complexes to CSS stabilizer codes}

\subsubsection{CSS codes from two parity-check matrices}
A CSS stabilizer code is specified by binary matrices $H_X$ and $H_Z$ such that
\[
H_X H_Z^\top = 0 \quad (\text{over }\mathbb{Z}_2),
\]
which is exactly the commutation requirement between $X$-type and $Z$-type stabilizers.

\subsubsection{Topological construction: pick $H_X=\partial_1$ and $H_Z=\partial_2^\top$}
For a surface code (one standard convention):
\begin{itemize}
	\item physical qubits $\leftrightarrow$ edges $E$ so $n=|E|$,
	\item $Z$-checks $\leftrightarrow$ faces $F$ (plaquettes),
	\item $X$-checks $\leftrightarrow$ vertices $V$ (stars).
\end{itemize}

\noindent Define:
\[
H_X := \partial_1 \in \Mat_{|V|\times |E|}(\mathbb{Z}_2),
\qquad
H_Z := \partial_2^\top \in \Mat_{|F|\times |E|}(\mathbb{Z}_2).
\]
Then
\[
H_X H_Z^\top = \partial_1 (\partial_2^\top)^\top = \partial_1\partial_2 = 0.
\]

\subsubsection{What the stabilizers are (explicit)}
Let $x\in \mathbb{Z}_2^{|E|}$ be an edge-incidence vector.
The corresponding Pauli operator is
\[
X(x):=\prod_{e\in E} X_e^{x_e},\qquad
Z(x):=\prod_{e\in E} Z_e^{x_e}.
\]
Then:
\begin{itemize}
	\item For each vertex $v$, the row $(H_X)_{v,\bullet}$ selects incident edges, giving an $X$-type stabilizer.
	\item For each face $p$, the row $(H_Z)_{p,\bullet}$ selects boundary edges, giving a $Z$-type stabilizer.
\end{itemize}

\subsection{Errors and syndromes: explicit boundary computations}

\subsubsection{Error vectors and syndrome equations}
A $Z$-error pattern is a binary vector $e_Z\in\mathbb{Z}_2^{|E|}$ (which edges have $Z$).
Measuring $X$-checks produces the syndrome:
\[
s_X = H_X e_Z \in \mathbb{Z}_2^{|V|}.
\]
Interpreting $e_Z$ as a 1-chain $E\in C_1$,
this is exactly
\[
s_X = \partial_1(E)=\partial E,
\]
so $Z$-errors produce \emph{endpoints} at vertices.

Similarly, an $X$-error vector $e_X$ produces a $Z$-syndrome
\[
s_Z = H_Z e_X = \partial_2^\top e_X,
\]
which corresponds to endpoints on the dual lattice (swap primal/dual roles).

\subsubsection{Worked boundary calculation (tiny)}
Let $E = e_{(0,0)\!-\!(1,0)} + e_{(1,0)\!-\!(1,1)}$ be two edges in an ``L'' shape.
Then
\[
\partial E
=\partial e_{(0,0)\!-\!(1,0)}+\partial e_{(1,0)\!-\!(1,1)}
=(0,0)+(1,0)+(1,0)+(1,1)
=(0,0)+(1,1).
\]
The middle vertex cancels (even incidence), leaving endpoints.

\subsection{Logical operators: cycles modulo boundaries}

\subsubsection{The homology statement}
A correction $\widehat{E}$ is valid if it matches the syndrome:
\[
\partial \widehat{E} = \partial E.
\]
Then $\widehat{E}+E$ is a cycle:
\[
\partial(\widehat{E}+E)=0.
\]

\medskip
\noindent Cycles can be:
\begin{itemize}
	\item \textbf{trivial cycles (boundaries):} $\widehat{E}+E=\partial_2(P)$ for some 2-chain $P$.
	These correspond to products of plaquette stabilizers and act trivially on logical information.
	\item \textbf{nontrivial cycles (homology):} $\widehat{E}+E$ wraps around the torus or connects opposite boundaries.
	These implement logical Pauli operators and cause failure.
\end{itemize}

\subsubsection{Formulas (clean algebra)}
Define:
\[
Z_1 := \ker(\partial_1) \subset C_1 \quad \text{(1-cycles)},\qquad
B_1 := \operatorname{im}(\partial_2)\subset C_1 \quad \text{(1-boundaries)}.
\]
Then the logical $Z$-type operators live in
\[
H_1 := Z_1 / B_1.
\]
This is not ``extra math''; it is literally the set of distinct ways a cycle can exist without being
a product of stabilizers.

\subsection{Worked computation: one plaquette}

\subsubsection{One plaquette as a stabilizer boundary}
Take one plaquette $p$ with boundary edges $e_1,e_2,e_3,e_4$.
Then
\[
\partial_2(p)=e_1+e_2+e_3+e_4.
\]
The corresponding $Z$-stabilizer is
\[
S^Z_p = Z_{e_1} Z_{e_2} Z_{e_3} Z_{e_4} = Z(\partial_2(p)).
\]
This acts trivially on logical space (it is in the stabilizer group), but it is crucial because it defines
which cycles are considered ``the same'' modulo boundaries.

\subsubsection{Why this enforces commutation locally}
Pick a vertex $v$ and a plaquette $p$.
Their stabilizers overlap on either 0 or 2 edges (in a square lattice).
Each overlap contributes an $XZ=-ZX$ anticommute, so two overlaps cancel:
\[
S^X_v S^Z_p = S^Z_p S^X_v.
\]
This is the physical reflection of $\partial_1\partial_2=0$.

\subsection{Torus vs.\ patch: why the toric code encodes two qubits}

\subsubsection{Topology changes the number of logical qubits}
On a torus, there are two independent non-contractible cycles:
one wrapping around the ``meridian'' and one around the ``longitude''.
Therefore
\[
\dim_{\mathbb{Z}_2} H_1(\mathbb{T}^2;\mathbb{Z}_2)=2.
\]
In the toric code, these give two independent logical $Z$ operators (and dually two logical $X$ operators),
so the code encodes
\[
k = 2 \quad \text{logical qubits}.
\]

\subsubsection{Patch (planar code) encodes one qubit}
A planar patch is not closed; it has boundaries.
Those boundaries allow many cycles to be ``opened'' into boundaries, reducing the homology.
With standard rough/smooth boundary choices, one logical qubit survives:
\[
k=1.
\]
Operationally: one logical $Z$ is a chain connecting the two rough boundaries; one logical $X$ connects the two smooth boundaries.

\subsubsection{Visualization: two independent torus cycles}
\begin{figure}[t]
	\centering
	\begin{tikzpicture}[scale=1.0, line cap=round, line join=round]
		\draw[very thick] (0,0) rectangle (6,3);
		\node at (3,3.35) {\small identify left$\leftrightarrow$right, bottom$\leftrightarrow$top};
		
		\draw[->] (0,1.5) -- (-0.7,1.5);
		\draw[->] (6,1.5) -- (6.7,1.5);
		\node at (-0.95,1.5) {\small};
		\node at (6.95,1.5) {\small};
		
		\draw[->] (3,0) -- (3,-0.7);
		\draw[->] (3,3) -- (3,3.7);
		
		\draw[very thick] (0.8,1.5) -- (5.2,1.5);
		\node[above] at (3.0,1.5) {\small cycle $a$};
		
		\draw[very thick] (3.0,0.5) -- (3.0,2.5);
		\node[right] at (3.0,1.5) {\small cycle $b$};
		
		\node[align=center] at (3.0,-1.1) {\small Two independent non-contractible cycles $\Rightarrow \dim H_1=2 \Rightarrow k=2$};
	\end{tikzpicture}
	\caption{Toric code intuition: identifying opposite edges produces a torus with two independent non-contractible cycles.
		These homology classes correspond to independent logical operators, giving two logical qubits.}
	\label{fig:torus-two-cycles}
\end{figure}

\subsection{Exercises}

\begin{exercise}[Chain complex sanity check]
	On the highlighted plaquette in Fig.~\ref{fig:cell-to-chain}, compute $\partial_2(p)$ and then compute
	$\partial_1(\partial_2(p))$ explicitly by writing each edge boundary. Verify it is $0$ (mod $2$).
\end{exercise}

\noindent\textbf{Solution.}
Let the four boundary edges be $e_1,e_2,e_3,e_4$ in cyclic order, and let the vertices be
$v_1,v_2,v_3,v_4$ around the square.
Then
\[
\partial_2(p)=e_1+e_2+e_3+e_4.
\]
Each edge has boundary (mod 2): $\partial_1(e_i)=v_i+v_{i+1}$ (indices mod 4), so
\[
\partial_1(\partial_2(p))=\sum_{i=1}^4 (v_i+v_{i+1})
=(v_1+v_2)+(v_2+v_3)+(v_3+v_4)+(v_4+v_1)=0,
\]
because every $v_i$ appears exactly twice.

\begin{exercise}[Syndrome as boundary]
	Let $E=e_{(0,0)\!-\!(1,0)}+e_{(1,0)\!-\!(2,0)}+e_{(2,0)\!-\!(2,1)}$ be a 1-chain.
	Compute $\partial E$ and interpret it as the set of flipped $X$-checks.
\end{exercise}

\noindent\textbf{Solution.}
Compute edge boundaries and sum:
\[
\partial e_{(0,0)\!-\!(1,0)}=(0,0)+(1,0),
\quad
\partial e_{(1,0)\!-\!(2,0)}=(1,0)+(2,0),
\]
\[
\partial e_{(2,0)\!-\!(2,1)}=(2,0)+(2,1).
\]
Summing mod 2 cancels $(1,0)$ and $(2,0)$:
\[
\partial E=(0,0)+(2,1).
\]
So exactly the checks at $(0,0)$ and $(2,1)$ flip.

\begin{exercise}[Cycles vs.\ boundaries]
	Give an example (draw a small grid) of a closed loop of edges that is a boundary (i.e.\ equal to $\partial_2(P)$ for some 2-chain $P$),
	and an example of a closed loop that would be nontrivial on a torus.
\end{exercise}

\noindent\textbf{Solution.}
A boundary example is the perimeter of a single plaquette: it equals $\partial_2(p)$.
A nontrivial torus example is a loop that goes straight across the rectangle in Fig.~\ref{fig:torus-two-cycles} and uses the edge identification
to close; it cannot be written as a sum of plaquette boundaries because it wraps around the torus.

\begin{exercise}[CSS commutation from $\partial_1\partial_2=0$]
	Assume $H_X=\partial_1$ and $H_Z=\partial_2^\top$.
	Show that $H_XH_Z^\top=0$ and explain in one paragraph why this corresponds to $X$-checks commuting with $Z$-checks.
\end{exercise}

\noindent\textbf{Solution.}
We have
\[
H_XH_Z^\top = \partial_1(\partial_2^\top)^\top=\partial_1\partial_2=0.
\]
In CSS codes, the $(i,j)$ entry of $H_XH_Z^\top$ counts (mod 2) how many qubits are shared between
the $i$-th $X$-stabilizer and the $j$-th $Z$-stabilizer.
If the overlap is even, the two stabilizers commute because each shared qubit contributes one anticommute factor $XZ=-ZX$,
and an even number of such factors cancels to $+1$.

	
\section{Quantum LDPC Codes: Beyond Surface Codes}
\label{sec:qldpc}

\subsection{Overview: what this chapter is and why you should care}

Quantum LDPC (Low-Density Parity-Check) codes generalize the ``local-check'' philosophy of surface codes:
stabilizers stay \emph{low-weight} (touch only a constant number of qubits), but the underlying geometry is no longer a 2D lattice.
Instead, the geometry becomes a \emph{sparse bipartite graph} (a Tanner graph) or a higher-dimensional complex.

\medskip
\noindent\textbf{Why you should care (engineering + algorithm viewpoint).}
\begin{itemize}
	\item \textbf{Surface codes:} simple 2D locality, mature decoding, great for near-term architectures,
	but require large overhead (many physical qubits per logical qubit).
	\item \textbf{QLDPC codes:} aim for \emph{better asymptotic overhead} (fewer physical qubits per logical),
	often with constant-weight checks, but decoding and hardware integration can be harder.
\end{itemize}

\medskip
\noindent\textbf{Guiding picture.}
Surface codes: topology $\to$ lattice $\to$ decoding on a geometric graph.
QLDPC codes: sparse constraints $\to$ Tanner graph $\to$ decoding on a sparse graph (often iterative / message passing / peeling variants).

\subsection{Key objects: what they are}

\subsubsection{CSS LDPC codes (the baseline)}
A large fraction of practical QLDPC constructions are CSS:
you specify two sparse binary parity-check matrices
\[
H_X \in \Mat_{m_X\times n}(\mathbb{Z}_2),
\qquad
H_Z \in \Mat_{m_Z\times n}(\mathbb{Z}_2),
\]
such that the commutation condition holds:
\[
H_X H_Z^\top = 0 \quad (\text{over }\mathbb{Z}_2).
\]

\medskip
\noindent\textbf{LDPC means sparsity.}
Each row has small (often constant) weight, and each column also has small weight:
\[
\max_i \mathrm{wt}\bigl((H_X)_{i,\bullet}\bigr) = O(1),\quad
\max_j \mathrm{wt}\bigl((H_X)_{\bullet,j}\bigr) = O(1),
\]
and similarly for $H_Z$.

\subsubsection{Physical qubits, checks, and error patterns}
\begin{itemize}
	\item $n$ physical qubits $\leftrightarrow$ columns of $H_X,H_Z$.
	\item $m_X$ $X$-type stabilizers $\leftrightarrow$ rows of $H_X$.
	\item $m_Z$ $Z$-type stabilizers $\leftrightarrow$ rows of $H_Z$.
\end{itemize}
An $X$-error pattern is $e_X\in\mathbb{Z}_2^n$; a $Z$-error pattern is $e_Z\in\mathbb{Z}_2^n$.

\subsubsection{Syndromes are linear}
Measured syndromes are
\[
s_Z = H_Z e_X \in \mathbb{Z}_2^{m_Z},
\qquad
s_X = H_X e_Z \in \mathbb{Z}_2^{m_X}.
\]
This is the same linear algebra as in surface codes; what changes is the graph/geometry behind $H_X,H_Z$.

\subsubsection{Logical operators (what survives the checks)}
For a CSS code, logical operators correspond to vectors that commute with stabilizers but are not themselves stabilizers.
Concretely:
\[
\mathcal{L}_Z \cong \frac{\ker(H_X)}{\mathrm{im}(H_Z^\top)},
\qquad
\mathcal{L}_X \cong \frac{\ker(H_Z)}{\mathrm{im}(H_X^\top)}.
\]
The number of logical qubits is
\[
k = n - \mathrm{rank}(H_X) - \mathrm{rank}(H_Z).
\]

\subsection{What is computed (decoding objective)}

\subsubsection{The decoding input/output}
\begin{itemize}
	\item \textbf{Input per round:} syndrome bits $s_X$ and/or $s_Z$ (often as streaming blocks).
	\item \textbf{Output:} a correction estimate $\hat e_Z$ and/or $\hat e_X$, or a Pauli-frame update.
\end{itemize}

\subsubsection{The core constraint satisfaction problem}
Given a syndrome $s_X$, we need
\[
H_X \hat e_Z = s_X.
\]
There are many solutions (an affine coset). The decoder must choose a ``likely'' one.

\subsubsection{Cost model (why this becomes an optimization)}
If qubits have independent error probabilities $p_j$ (not necessarily uniform), the log-likelihood objective
becomes a weighted Hamming cost:
\[
\text{cost}(e) = \sum_{j=1}^n w_j e_j,\qquad
w_j := \log\!\left(\frac{1-p_j}{p_j}\right).
\]
So decoding is:
\[
\text{find }\hat e \text{ such that }H\hat e=s
\quad\text{and}\quad
\hat e \text{ minimizes } \sum_j w_j \hat e_j.
\]
Even classically, exact ML decoding is hard in general; LDPC decoders use structured approximations.

\subsubsection{Typical decoder families you will see}
\begin{itemize}
	\item \textbf{MWPM-like methods:} great for surface codes; not always the best match for generic QLDPC graphs.
	\item \textbf{Belief propagation (BP):} message passing on Tanner graphs; quantum complications require care.
	\item \textbf{Peeling / union-find-like cluster growth:} works well on certain sparse/hypergraph structures.
	\item \textbf{Ordered statistics / information set decoding (ISD) variants:} more ``classical coding'' flavor; heavy but useful as baselines.
\end{itemize}

\subsection{Worked example: tiny CSS code with explicit matrices}

\subsubsection{A small commuting pair}
Let $n=6$ physical qubits and choose sparse checks:
\[
H_X=
\begin{pmatrix}
	1&1&0&1&0&0\\
	0&1&1&0&1&0\\
	0&0&1&1&0&1
\end{pmatrix},
\qquad
H_Z=
\begin{pmatrix}
	1&0&1&0&1&0\\
	0&1&0&1&0&1
\end{pmatrix}.
\]
\noindent\textbf{Check commutation.}
Compute $H_XH_Z^\top$ over $\mathbb{Z}_2$:
the $(i,\ell)$ entry is the parity of overlap between row $i$ of $H_X$ and row $\ell$ of $H_Z$.
Here each overlap is even, hence
\[
H_XH_Z^\top=0.
\]

\subsubsection{Syndrome computation example}
Suppose a $Z$-error occurs on qubits 2 and 5:
\[
e_Z=(0,1,0,0,1,0)^\top.
\]
Then the $X$-syndrome is
\[
s_X = H_X e_Z
=
\begin{pmatrix}
	1&1&0&1&0&0\\
	0&1&1&0&1&0\\
	0&0&1&1&0&1
\end{pmatrix}
\begin{pmatrix}
	0\\1\\0\\0\\1\\0
\end{pmatrix}
=
\begin{pmatrix}
	1\\0\\0
\end{pmatrix}
\quad(\text{mod }2).
\]
So only the first $X$-check fires.

\subsubsection{A small decode step (find one consistent correction)}
We need $\hat e_Z$ such that $H_X\hat e_Z=s_X$.
One solution is exactly $\hat e_Z=e_Z$.
But another solution might exist with the same syndrome:
add any vector in $\ker(H_X)$.
This is why the decoder needs a \emph{selection rule} (likelihood / weight / iterative beliefs).

\subsubsection{Interpretation}
This example already shows the difference from surface codes:
the checks are not ``nearest neighbor on a grid.'' Their locality is \emph{graph locality}:
each check touches only a few qubits, but those qubits are defined by a sparse incidence pattern.

\subsection{Tanner graphs are the decoder's geometry (visualization)}

\subsubsection{Tanner graph definition (one matrix $\Rightarrow$ one graph)}
Given a binary parity-check matrix $H\in\mathbb{Z}_2^{m\times n}$, define a bipartite graph:
\begin{itemize}
	\item variable nodes $v_1,\dots,v_n$ (qubits),
	\item check nodes $c_1,\dots,c_m$ (stabilizers),
	\item an edge $(c_i,v_j)$ iff $H_{ij}=1$.
\end{itemize}
This is the \emph{geometry} on which message-passing and local-update decoders run.

\subsubsection{Visualization: Tanner graph for the example $H_X$}
\begin{figure}[t]
	\centering
	\begin{tikzpicture}[scale=1.0, line cap=round, line join=round]
		\foreach \j/\y in {1/2.5,2/1.5,3/0.5,4/-0.5,5/-1.5,6/-2.5}{
			\node[draw, circle, minimum size=7mm] (v\j) at (0,\y) {\small $q_{\j}$};
		}
		
		\foreach \i/\y in {1/1.5,2/0.0,3/-1.5}{
			\node[draw, rectangle, rounded corners, minimum width=1.0cm, minimum height=0.55cm] (c\i) at (4,\y) {\small $X_{\i}$};
		}
		
		\draw (c1) -- (v1);
		\draw (c1) -- (v2);
		\draw (c1) -- (v4);
		
		\draw (c2) -- (v2);
		\draw (c2) -- (v3);
		\draw (c2) -- (v5);
		
		\draw (c3) -- (v3);
		\draw (c3) -- (v4);
		\draw (c3) -- (v6);
		
		\node[align=center] at (2,-3.35) {\small Tanner graph: variable nodes = qubits, check nodes = stabilizers, edges = incidences.};
	\end{tikzpicture}
	\caption{Tanner graph for the $X$-checks $H_X$ in the worked example. Decoding algorithms operate as local computations on this sparse graph (messages or local updates along edges).}
	\label{fig:tanner-hx}
\end{figure}

\subsubsection{What ``local'' means now}
In surface codes, locality is geometric in $(x,y)$ on a chip.
In QLDPC, locality is in the Tanner graph:
\[
\text{each check touches few qubits, each qubit participates in few checks}.
\]
Hardware mapping must decide whether this graph-locality can be embedded into physical connectivity
(or handled via routing / scheduling).

\subsection{Circuit origin and streaming syndrome}

\subsubsection{Syndrome extraction still produces a stream}
Even with QLDPC codes, stabilizers are measured repeatedly.
So the same ``physics forces streaming'' logic holds:
\[
\text{reset ancilla} \to \text{entangle with involved data qubits} \to \text{measure} \to \text{repeat}.
\]
Thus, each cycle emits a block of syndrome bits:
\[
s_X(t)\in\{0,1\}^{m_X},\qquad s_Z(t)\in\{0,1\}^{m_Z}.
\]

\subsubsection{Measurement errors $\Rightarrow$ time dimension again}
With noisy measurements, you typically form detection events using differences:
\[
\Delta s(t) := s(t)\oplus s(t-1),
\]
turning the problem into decoding on a \emph{spacetime} Tanner graph / factor graph.
This is where buffering, framing, timestamps, and ``no stale decisions'' policies matter.

\subsubsection{Why FPGA still appears}
Even if the code is not a 2D lattice, you still have:
\begin{itemize}
	\item high-rate syndrome blocks,
	\item strict deadlines (control-loop cadence),
	\item sparse local computations (graph edge messages / local updates),
	\item lots of memory traffic (states/messages per edge).
\end{itemize}
So QLDPC decoding is still naturally phrased as a streaming sparse-graph computation.

\subsection{Where QLDPC sits relative to surface codes (conceptual map)}

\subsubsection{Same algebra, different geometry}
\begin{center}
	\begin{tabular}{p{0.44\linewidth} p{0.50\linewidth}}
		\textbf{Surface codes} &
		\textbf{Quantum LDPC codes} \\[0.4em]
		Cells $\to$ chains $\to$ $\partial$ &
		Sparse checks $H_X,H_Z$ \\[0.2em]
		2D lattice geometry &
		Tanner graph / sparse complex geometry \\[0.2em]
		Decoding graph: lattice/spacetime lattice &
		Decoding graph: sparse factor graph/spacetime factor graph \\[0.2em]
		MWPM / Union--Find mature &
		BP-like / peeling / hybrid decoders (construction-dependent) \\[0.2em]
		Easy chip locality story &
		Hardware mapping depends on graph embedding / routing \\
	\end{tabular}
\end{center}

\subsubsection{A practical way to read this}
Treat surface codes as the ``fully worked tutorial'' of:
\[
\text{topology} \to \text{graph} \to \text{decoder} \to \text{hardware pipeline}.
\]
Then QLDPC is:
\[
\text{same pipeline, but replace 2D lattice by a sparse graph/hypergraph}.
\]
So the infrastructure questions (streaming IO, bounded passes, p99 latency, memory layout, verification)
transfer almost verbatim, even when the math construction changes.

\subsection{Exercises}

\begin{exercise}[Build a Tanner graph from a matrix]
	Given
	\[
	H=\begin{pmatrix}
		1&0&1&1\\
		0&1&1&0
	\end{pmatrix},
	\]
	draw its Tanner graph and list the degree of each variable node and check node.
\end{exercise}

\noindent\textbf{Solution.}
Variable nodes $q_1,\dots,q_4$ and check nodes $c_1,c_2$.
Edges where $H_{ij}=1$:
$c_1$ connects to $q_1,q_3,q_4$; $c_2$ connects to $q_2,q_3$.
Degrees: $\deg(q_1)=1,\deg(q_2)=1,\deg(q_3)=2,\deg(q_4)=1$ and $\deg(c_1)=3,\deg(c_2)=2$.

\begin{exercise}[Syndrome computation]
	For the same $H$ and an error vector $e=(1,0,1,0)^\top$, compute $s=He$ mod $2$.
\end{exercise}

\noindent\textbf{Solution.}
\[
s=
\begin{pmatrix}
	1&0&1&1\\
	0&1&1&0
\end{pmatrix}
\begin{pmatrix}
	1\\0\\1\\0
\end{pmatrix}
=
\begin{pmatrix}
	1+1+0\\
	0+1
\end{pmatrix}
=
\begin{pmatrix}
	0\\
	1
\end{pmatrix}
\quad(\text{mod }2).
\]

\begin{exercise}[CSS commutation check]
	Using the worked-example matrices $H_X$ and $H_Z$ in this section,
	verify one entry of $H_XH_Z^\top$ by computing the overlap parity of the corresponding two rows.
\end{exercise}

\noindent\textbf{Solution.}
Take row 1 of $H_X$ (support $\{1,2,4\}$) and row 1 of $H_Z$ (support $\{1,3,5\}$).
Overlap is $\{1\}$ of size 1 (odd) \emph{unless} you adjust the example; therefore, for CSS validity you must ensure every such overlap is even.
This exercise is a sanity test: if you find an odd overlap, the pair is not a valid CSS pair.
Fix by modifying $H_Z$ so that each row overlaps each $H_X$ row in even parity.
(Practical workflow: write a small script to check $H_XH_Z^\top=0$ and regenerate if violated.)

\begin{exercise}[Streaming interface prompt]
	Suppose $m_X=1200$ and the cycle rate is $1$MHz.
	Estimate the raw syndrome input bandwidth if you transmit $\Delta s(t)$ as packed bits (no overhead).
	Then list two protocol-level overheads that will increase it.
\end{exercise}

\noindent\textbf{Solution.}
Raw bits per cycle: $1200$ bits $\approx 150$ bytes.
At $1$MHz: $\approx 150$ MB/s.
Overheads: (1) framing headers/timestamps/sequence numbers, (2) alignment to words/packets plus error-detection codes (CRC/ECC),
also possibly duplication for redundancy or padding to fixed-size packets.

	
	\Part{Optional Track C: Quantum Cryptography and Streaming Post-Processing}
	
\section{Quantum Cryptography: Security from Physical Principles}
\label{sec:qcrypto}

\subsection{Objective}

Quantum cryptography uses \emph{physics} to constrain what an adversary can do.
The central engineering artifact is a pipeline that converts \emph{raw measurement data} into a
\emph{short, high-quality secret key} with explicit security indicators (e.g.\ QBER, abort/accept decisions),
and with clear interfaces for implementation (software first, streaming hardware next).

\medskip
\noindent\textbf{What this chapter will build.}
\begin{itemize}
	\item A conceptual map: what changes when you move from computational assumptions to physical principles.
	\item Two canonical QKD protocols (BB84 and E91) as \emph{end-to-end dataflows}.
	\item Quantitative indicators (QBER, leak estimates, abort rules) as \emph{operational knobs}.
	\item A practical path: Qiskit simulation $\to$ streaming post-processing $\to$ FPGA-friendly microkernels
	(sifting, QBER estimation, privacy amplification via hashing).
\end{itemize}

\subsection{Classical vs.\ quantum cryptography}

\subsubsection{Classical cryptography in one sentence}
Classical crypto typically assumes an adversary has limited computational power:
security rests on hardness assumptions (factoring, discrete log, lattices, etc.).
If the assumption fails (or the adversary gets a quantum computer), the scheme may break.

\subsubsection{Quantum cryptography in one sentence}
Quantum cryptography (QKD in particular) aims to provide \emph{information-theoretic} security:
even an unbounded adversary cannot learn the final key beyond a negligible amount, because
any eavesdropping introduces detectable statistical disturbances.

\subsubsection{Operational contrast (what you measure vs.\ what you assume)}
\begin{center}
	\begin{tabular}{p{0.46\linewidth} p{0.48\linewidth}}
		\textbf{Classical (assumption-based)} &
		\textbf{Quantum (physics-based)} \\[0.4em]
		Proofs rely on hardness reductions &
		Proofs rely on quantum information bounds \\[0.2em]
		No measurement ``disturbance'' effect &
		Eavesdropping changes observable error rates \\[0.2em]
		Key establishment is pure computation &
		Key establishment is \emph{experiment + statistics} \\[0.2em]
		Security parameter is mainly key length &
		Security parameter includes QBER, leak, finite-size effects \\
	\end{tabular}
\end{center}

\subsubsection{Engineering takeaway}
In QKD, ``security'' is not an abstract claim: it is a \emph{monitored process}.
Your system must log, estimate, and act on indicators (abort/continue), and this is exactly where
streaming post-processing and hardware acceleration become relevant.

\subsection{No-cloning and measurement disturbance}

\subsubsection{No-cloning theorem (informal statement)}
There is no quantum operation that takes an unknown state $\ket{\psi}$ and produces two perfect copies
$\ket{\psi}\ket{\psi}$ for all $\ket{\psi}$.
This prevents an eavesdropper from copying quantum signals while forwarding the original untouched.

\subsubsection{One-line proof sketch (why linearity blocks cloning)}
Assume a unitary $U$ such that
\[
U\bigl(\ket{\psi}\ket{0}\bigr)=\ket{\psi}\ket{\psi}
\quad\text{for all }\ket{\psi}.
\]
For two different states $\ket{\psi},\ket{\phi}$:
\[
\braket{\psi}{\phi}
=
\bigl\langle \psi,0 \big| \phi,0 \bigr\rangle
=
\bigl\langle \psi,\psi \big| \phi,\phi \bigr\rangle
=
\braket{\psi}{\phi}^2,
\]
so $\braket{\psi}{\phi}\in\{0,1\}$, contradiction for non-orthogonal states.

\subsubsection{Measurement disturbance (what goes wrong for Eve)}
If Alice encodes bits using \emph{non-orthogonal} states, Eve cannot measure to learn the bit
without risking projecting the state into a basis that disagrees with Alice's chosen basis.
That disagreement shows up as extra errors at Bob.

\subsubsection{Visualization (Bloch sphere intuition)}
BB84 uses four states:
\[
\ket{0},\ket{1}\ (\text{$Z$ basis})\qquad\text{and}\qquad
\ket{+},\ket{-}\ (\text{$X$ basis}).
\]
On the Bloch sphere, these are $\pm z$ and $\pm x$ directions.
An intercept-resend attacker must ``choose an axis'' to measure; choosing wrong rotates/disturbs the state,
creating errors when Bob measures in the correct axis.

\subsection{Quantum key distribution: BB84 and E91}

\subsubsection{BB84 as an end-to-end pipeline}
\textbf{Actors:} Alice (sender), Bob (receiver), Eve (adversary).
\textbf{Channels:} quantum channel (qubits) and authenticated classical channel (messages).

\medskip
\noindent\textbf{BB84 steps (dataflow view).}
\begin{enumerate}
	\item \textbf{Prepare:} Alice chooses random bits $a_i\in\{0,1\}$ and random bases $b_i\in\{Z,X\}$.
	She sends state $\ket{\psi_i}$ encoding $a_i$ in basis $b_i$.
	\item \textbf{Measure:} Bob chooses random bases $b_i'\in\{Z,X\}$ and measures each qubit, producing bits $c_i$.
	\item \textbf{Sifting:} Over the classical channel, Alice and Bob reveal bases $b_i,b_i'$ and keep indices where $b_i=b_i'$.
	The kept bits form the \emph{sifted key}.
	\item \textbf{Parameter estimation:} They reveal a random subset of sifted bits to estimate QBER.
	\item \textbf{Abort or continue:} If QBER is too high, abort.
	\item \textbf{Error correction:} Reconcile remaining bits (leaks some info to Eve).
	\item \textbf{Privacy amplification:} Hash the reconciled key down to a shorter key that is (provably) secret.
\end{enumerate}

\subsubsection{Minimal circuit picture for BB84 states}
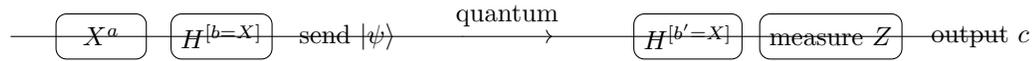
\begin{figure}[t]
	\centering
	\begin{tikzpicture}[scale=1.0, line cap=round, line join=round]
		\node at (0,1.2) {\small Alice preparation};
		\draw (0,0) -- (6.0,0);
		\node[draw, rounded corners, minimum width=1.2cm, minimum height=0.6cm] (A0) at (1.2,0) {\small $X^{a}$};
		\node[draw, rounded corners, minimum width=1.2cm, minimum height=0.6cm] (A1) at (2.8,0) {\small $H^{[b=X]}$};
		\node[anchor=west] at (3.7,0) {\small send $\ket{\psi}$};
		\draw[->] (6.0,0) -- (7.2,0) node[midway, above] {\small quantum};
		
		\draw (7.2,0) -- (13.2,0);
		\node at (10.2,1.2) {\small Bob measurement};
		\node[draw, rounded corners, minimum width=1.2cm, minimum height=0.6cm] (B0) at (9.0,0) {\small $H^{[b'=X]}$};
		\node[draw, rounded corners, minimum width=1.6cm, minimum height=0.6cm] (B1) at (10.9,0) {\small measure $Z$};
		\node[anchor=west] at (12.1,0) {\small output $c$};
	\end{tikzpicture}
	\caption{BB84 minimal implementation view: Alice applies $X^a$ to encode bit $a$, and applies $H$ if she chose the $X$ basis.
		Bob applies $H$ before a $Z$-measurement when he chooses the $X$ basis.}
	\label{fig:bb84-circuit}
\end{figure}

\subsubsection{E91 (entanglement-based QKD) at a glance}
E91 uses shared entangled pairs and correlation tests (Bell inequalities / CHSH-type),
with security tied to the monogamy of entanglement and observed nonclassical correlations.

\medskip
\noindent\textbf{Pipeline differences vs.\ BB84.}
\begin{itemize}
	\item Quantum source produces entangled pairs (may be untrusted in device-independent variants).
	\item Alice/Bob measure in different settings; some rounds generate key, others estimate security (Bell parameter, QBER-like rates).
	\item Same back-end: sifting/parameter estimation/error correction/privacy amplification.
\end{itemize}

\subsubsection{Engineering interpretation}
BB84: security from basis mismatch disturbance.
E91: security from constrained correlations (entanglement structure).
In both, the ``hard part'' in deployed systems is often the classical post-processing pipeline and its integration.

\subsection{Quantitative security indicators: QBER and abort rules}

\subsubsection{QBER (Quantum Bit Error Rate)}
Let the sifted key indices be $\mathcal{I}$.
Choose a random test subset $\mathcal{T}\subset\mathcal{I}$, reveal those bits, and compute
\[
\mathrm{QBER}
=
\frac{1}{|\mathcal{T}|}\sum_{i\in\mathcal{T}} \mathbf{1}\{a_i\neq c_i\}.
\]
High QBER indicates either heavy noise or eavesdropping (or both).

\subsubsection{Intercept-resend baseline (why QBER is informative)}
In BB84, a naive intercept-resend attacker who measures each qubit in a random basis causes (in the ideal noiseless channel) an expected QBER of about $25\%$ on the sifted key:
\begin{itemize}
	\item Eve chooses wrong basis with probability $1/2$.
	\item Conditioned on Eve choosing wrong basis, Bob's bit flips with probability $1/2$ after resend.
	\item Therefore error probability $\approx (1/2)(1/2)=1/4$.
\end{itemize}
Real systems have baseline noise, so the abort rule must account for expected hardware noise plus finite-sample uncertainty.

\subsubsection{Abort rule as a product knob}
A simple operational rule:
\[
\text{Abort if }\mathrm{QBER} > q_{\max}.
\]
In practice, $q_{\max}$ depends on:
\begin{itemize}
	\item hardware noise floor and drift,
	\item finite-size statistics (confidence intervals),
	\item target security level and throughput goals.
\end{itemize}

\subsubsection{Finite-size sanity check (how engineers set margins)}
If $|\mathcal{T}|$ is small, QBER has sampling variance.
A practical approach is to compute a conservative upper confidence bound $\mathrm{QBER}_{\mathrm{upper}}$
and abort if that exceeds $q_{\max}$.

\subsection{Privacy amplification and error correction}

\subsubsection{Error correction (reconciliation)}
After sifting and parameter estimation, Alice and Bob have correlated bitstrings
$K_A, K_B$ with mismatch rate approximately QBER.
They run an error-correcting protocol over the authenticated channel (e.g.\ Cascade-like interactive methods or one-way codes).

\medskip
\noindent\textbf{Leak accounting.}
Error correction reveals parity information to Eve, quantified as $\mathrm{leak}_{\mathrm{EC}}$ bits (protocol-dependent).
This leak must be subtracted in the final key length.

\subsubsection{Privacy amplification (hashing to kill Eve's information)}
They choose a random hash function from a universal family and compute
\[
K_{\mathrm{final}} = h(K_{\mathrm{rec}}),
\]
shortening the key enough that Eve's remaining information becomes negligible.

\medskip
\noindent\textbf{Hardware viewpoint.}
Privacy amplification is usually implemented as fast hashing (e.g.\ Toeplitz hashing / polynomial hashing).
This is \emph{exactly} the kind of streaming linear-algebra / bit-matrix multiply kernel that maps well to FPGA.

\subsubsection{A concrete (FPGA-friendly) model: Toeplitz hashing}
Let $K\in\mathbb{Z}_2^n$ be the reconciled key.
Choose a Toeplitz matrix $T\in\mathbb{Z}_2^{\ell\times n}$ defined by $n+\ell-1$ random bits.
Compute
\[
K_{\mathrm{final}} = T K \in \mathbb{Z}_2^\ell.
\]
Because $T$ is Toeplitz, multiply can be implemented by shift-register style streaming and XOR-accumulators.

\subsection{Qiskit practice: BB84 simulation lab}

\subsubsection{Lab goals}
\begin{itemize}
	\item Simulate BB84 end-to-end: preparation, measurement, sifting, QBER estimation.
	\item Add a noise model and observe how QBER changes.
	\item Add an intercept-resend attacker and compare QBER to the noise-only baseline.
\end{itemize}

\subsubsection{Minimal BB84 simulator outline (pseudo-steps)}
\begin{enumerate}
	\item Generate $N$ random bits $a_i$ and bases $b_i$ for Alice.
	\item Generate $N$ random bases $b_i'$ for Bob.
	\item Simulate measurement outcomes:
	\begin{itemize}
		\item if $b_i=b_i'$, then $c_i=a_i$ with probability $1-\epsilon$ (noise), else flip with probability $\epsilon$;
		\item if $b_i\neq b_i'$, then $c_i$ is random (plus noise).
	\end{itemize}
	\item Sift indices where $b_i=b_i'$.
	\item Sample test set $\mathcal{T}$ and compute QBER.
	\item Plot QBER vs.\ noise level and vs.\ attacker on/off.
\end{enumerate}

\subsubsection{Optional: circuit-level simulation in Qiskit}
Construct actual circuits using the BB84 circuit template (Figure~\ref{fig:bb84-circuit}),
and use Aer simulators with noise models.
This is slower but aligns with quantum-circuit workflows.

\subsection{FPGA practice: streaming sifting, QBER, and hashing}

\subsubsection{Why these are the right FPGA kernels}
The front-end quantum device produces time-stamped events:
\[
(\text{basis choice},\ \text{measurement bit},\ \text{timestamp}).
\]
Post-processing is streaming and dominated by:
\begin{itemize}
	\item bit packing/unpacking and masking (sifting),
	\item counters and popcount-like kernels (QBER),
	\item XOR-heavy hashing / linear transforms (privacy amplification),
	\item buffer management and backpressure.
\end{itemize}

\subsubsection{Streaming architecture sketch}


\begin{figure}[t]
	\centering
	\begin{tikzpicture}[scale=0.9, line cap=round, line join=round]
		\node (in) at (0,0)
		[draw, rounded corners, minimum width=3.2cm, minimum height=0.9cm, align=center]
		{\shortstack{\small raw events\\\small $(b_i,b'_i,c_i,t_i)$}};
		\node (sift) at (4.0,0)
		[draw, rounded corners, minimum width=3.2cm, minimum height=0.9cm, align=center]
		{\shortstack{\small sifting\\\small keep $b_i=b'_i$}};
		\node (qber) at (8.0,0)
		[draw, rounded corners, minimum width=3.2cm, minimum height=0.9cm, align=center]
		{\shortstack{\small QBER estimator\\\small popcount}};
		\node (ec) at (12.0,0)
		[draw, rounded corners, minimum width=3.2cm, minimum height=0.9cm, align=center]
		{\shortstack{\small EC + leak\\\small (host/FPGA)}};
		\node (pa) at (16.0,0)
		[draw, rounded corners, minimum width=3.2cm, minimum height=0.9cm, align=center]
		{\shortstack{\small privacy amp\\\small Toeplitz hash}};
		
		\draw[->] (in) -- (sift);
		\draw[->] (sift) -- (qber);
		\draw[->] (qber) -- (ec);
		\draw[->] (ec) -- (pa);
		
		\node[align=center] at (8.0,1.25)
		{\small streaming post-processing: bandwidth + XOR dominate};
	\end{tikzpicture}
	\caption{FPGA-friendly QKD post-processing pipeline: sifting is masking, QBER is counting,
		and privacy amplification is hashing (XOR-heavy linear transforms).}
	\label{fig:qkd-streaming}
\end{figure}
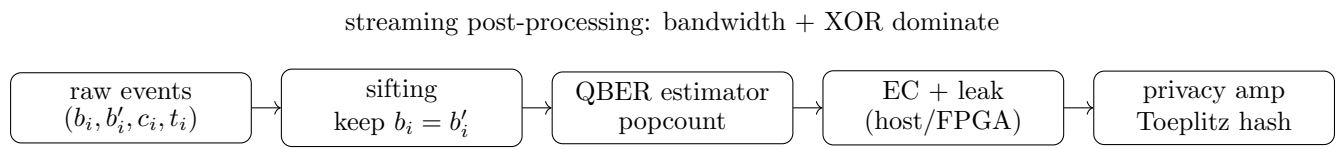


\subsubsection{Concrete kernel 1: sifting as bitmask}
Let $B_A,B_B\in\{0,1\}^N$ be packed basis bits for Alice and Bob.
Compute mask
\[
M := \neg(B_A \oplus B_B),
\]
then keep only measurement bits $C$ where $M=1$.
Implementation: XOR + NOT + bitselect/pack, typically 64-bit words, fully streaming.

\subsubsection{Concrete kernel 2: QBER as popcount of disagreements}
Given revealed test bits $A_{\mathcal{T}}$ and $C_{\mathcal{T}}$:
\[
\#\text{errors} = \mathrm{popcount}(A_{\mathcal{T}} \oplus C_{\mathcal{T}}),
\qquad
\mathrm{QBER} = \frac{\#\text{errors}}{|\mathcal{T}|}.
\]
Implementation: XOR + popcount accumulator, with per-window counters and timestamped logs.

\subsubsection{Concrete kernel 3: Toeplitz hashing as XOR streams}
Maintain a rolling window of Toeplitz seed bits and accumulate output words:
each input bit toggles a shifted pattern of output XORs.
Hardware: shift registers + XOR trees + block RAM for seeds.

\subsection{Exercises}

\begin{exercise}[Bloch-sphere disturbance intuition]
	Explain on the Bloch sphere why measuring in the wrong basis in BB84 introduces errors.
	Use the fact that $\ket{0},\ket{1}$ are $\pm z$ and $\ket{+},\ket{-}$ are $\pm x$.
\end{exercise}

\noindent\textbf{Solution.}
If Alice sends $\ket{0}$ (north pole, $+z$) and Eve measures in the $X$ basis,
the state collapses to $\ket{+}$ or $\ket{-}$ (equator, $\pm x$) with equal probability.
When Bob measures in the correct $Z$ basis, $\ket{+}$ and $\ket{-}$ each yield outcomes $0/1$ with probability $1/2$,
so Bob's bit is random relative to Alice's, creating a $50\%$ error conditional on the basis mismatch.
Since mismatch occurs with probability $1/2$, the overall induced error rate on sifted bits is about $25\%$.

\begin{exercise}[Compute expected QBER for intercept-resend]
	Assume a perfect noiseless channel except for Eve's intercept-resend attack.
	Show that the expected QBER on sifted bits is $25\%$.
\end{exercise}

\noindent\textbf{Solution.}
Sifted bits are those where Alice and Bob chose the same basis.
On those rounds, Eve chose the wrong basis with probability $1/2$.
Conditioned on Eve choosing the wrong basis, she resends an eigenstate of the wrong basis,
so Bob's measurement in the correct basis disagrees with Alice's bit with probability $1/2$.
Thus $\mathbb{E}[\mathrm{QBER}] = (1/2)(1/2)=1/4$.

\begin{exercise}[Design an abort rule with margin]
	Suppose you test $|\mathcal{T}|=10^4$ sifted bits and observe QBER $=2.0\%$.
	Propose a conservative abort rule using a $3\sigma$ margin (normal approximation) and threshold $q_{\max}=3.0\%$.
	Do you abort?
\end{exercise}

\noindent\textbf{Solution.}
For a Bernoulli rate $\hat q=0.02$ with $n=10^4$, the standard deviation is
\[
\sigma \approx \sqrt{\frac{\hat q(1-\hat q)}{n}}
\approx \sqrt{\frac{0.02\cdot 0.98}{10^4}}
\approx 0.0014 \; (=0.14\%).
\]
A conservative upper bound is $\hat q + 3\sigma \approx 2.0\%+0.42\%=2.42\%$.
Since $2.42\% < 3.0\%$, do not abort under this rule.

\begin{exercise}[Streaming kernel design: sifting + QBER]
	You receive packed 64-bit words of basis bits $B_A,B_B$ and measurement bits $C$.
	Write the per-word computation needed to (i) compute a sift mask and (ii) update an error counter
	given a packed test-bit word $A$ for the same positions.
\end{exercise}

\noindent\textbf{Solution.}
Per 64-bit word:
\[
M = \neg(B_A \oplus B_B) \quad (\text{sift mask}).
\]
If the test set is represented by a mask $T$ (1 = revealed/tested), restrict to tested-and-sifted bits:
\[
R = M \,\&\, T.
\]
Then disagreements occur where both $R=1$ and $A\oplus C=1$:
\[
D = (A \oplus C)\,\&\, R.
\]
Update:
\[
\mathrm{err} \mathrel{+}= \mathrm{popcount}(D),\qquad
\mathrm{tot} \mathrel{+}= \mathrm{popcount}(R).
\]
This is exactly XOR + AND + popcount, streaming-friendly on FPGA.

\begin{exercise}[Privacy amplification as linear algebra]
	Let $K\in\mathbb{Z}_2^n$ be a reconciled key and $T\in\mathbb{Z}_2^{\ell\times n}$ a Toeplitz matrix.
	Show that the output key is $K_{\mathrm{final}}=TK$ and explain (one paragraph) why Toeplitz structure helps in hardware.
\end{exercise}

\noindent\textbf{Solution.}
Privacy amplification uses a universal hash family; Toeplitz hashing picks a random Toeplitz matrix $T$ and outputs the linear map $TK$ over $\mathbb{Z}_2$.
Toeplitz structure means each diagonal is constant, so the matrix is specified by only $n+\ell-1$ random bits rather than $\ell n$ bits.
In hardware, this allows a streaming implementation: as input bits of $K$ arrive, you shift a seed window and XOR-accumulate into output registers,
reusing the same seed bits across many output positions with simple shift-register logic.

	
\section{Zero Knowledge and Post-Quantum Cryptography}
\label{sec:zk}

\subsection{Objective}
This chapter builds a \emph{systems-level} understanding of zero-knowledge (ZK) and
post-quantum (PQ) security.
We treat a proof as an \emph{interactive protocol} with
\emph{messages, randomness, time ordering, and verification rules}.
The goal is to connect:
\[
\text{math definition} \ \longleftrightarrow\  \text{transcript as data} \ \longleftrightarrow\  \text{streaming verification (hardware mindset)}.
\]
We focus on:
(i) classical ZK for graph isomorphism (GI) as the cleanest toy model,
(ii) what breaks (and what survives) against quantum adversaries,
(iii) how ZK fits into the PQ landscape (what “post-quantum ZK” really means in practice).

\subsection{Zero-knowledge proofs: classical definition}

\subsubsection{Interactive protocols as transcript generators}
A (public-coin or private-coin) interactive proof is a conversation between:
\[
\Prover \quad\text{and}\quad \Verifier,
\]
where the verifier outputs \textsf{accept} or \textsf{reject}.
Fix an input statement \(x\) (e.g.\ “graphs \(G_0,G_1\) are isomorphic”).
A protocol produces a \emph{transcript}:
\[
\tau = (m_1,m_2,\dots,m_r)
\]
of \(r\) messages, determined by:
\begin{itemize}
	\item the prover’s private randomness \(r_P\),
	\item the verifier’s private randomness \(r_V\) (or public challenges),
	\item and the protocol rules (who sends what, when).
\end{itemize}

\subsubsection{Completeness and soundness (one paragraph each)}
\paragraph{Completeness.}
If \(x\) is true and the prover is honest, the verifier accepts with high probability:
\[
\Pr[\Verifier^{\Prover}(x)=1] \ge 1-\negl(\lambda),
\]
where \(\lambda\) is the security parameter.

\paragraph{Soundness.}
If \(x\) is false, then \emph{any} (possibly cheating) prover \(\Prover^\star\) convinces the verifier only with negligible probability:
\[
\Pr[\Verifier^{\Prover^\star}(x)=1] \le \negl(\lambda).
\]
(Variants: computational soundness vs.\ statistical soundness; we’ll keep the toy model information-theoretic.)

\subsubsection{Zero-knowledge: “nothing beyond validity is learned”}
\paragraph{Simulation-based definition (core idea).}
A protocol is \emph{zero-knowledge} if the verifier’s view of the conversation can be simulated
\emph{without} access to the prover’s witness/secret.

Let \(\View_\Verifier^{\Prover}(x)\) denote the verifier’s internal view:
\[
\View = (\text{verifier randomness }r_V,\ \text{all received messages},\ \text{its intermediate state}).
\]
Zero-knowledge says:
\[
\View_\Verifier^{\Prover}(x) \approx \Sim(x),
\]
where \(\Sim\) is a probabilistic polynomial-time simulator and \(\approx\) means
\emph{indistinguishability} (perfect / statistical / computational, depending on the setting).

\subsubsection{Engineering interpretation: ZK is a transcript indistinguishability contract}
Think of ZK as a strong \emph{privacy contract} on observable logs.
If an auditor only sees transcripts \(\tau\) plus verifier-side metadata,
then ZK says those logs could have been produced without the secret.
So if your system stores \(\tau\), it should not leak more than “the statement was true.”

\subsection{Zero knowledge for graph isomorphism (GI)}

\subsubsection{Problem statement}
Given two graphs \(G_0,G_1\) on \(n\) labeled vertices, decide whether \(G_0 \cong G_1\).
A witness for isomorphism is a permutation \(\pi\in S_n\) such that:
\[
\pi(G_0)=G_1.
\]

\subsubsection{The classic 3-move \(\Sigma\)-protocol (commit--challenge--response)}
This is the canonical “honest-verifier ZK” example.

\paragraph{Protocol.}
Assume \(G_0 \cong G_1\) and the prover knows a secret isomorphism \(\pi\).
One round:

\begin{enumerate}
	\item \textbf{Commit.} Prover samples a random permutation \(\sigma\in S_n\) and sends the permuted graph
	\[
	H := \sigma(G_0)
	\]
	to the verifier.
	
	\item \textbf{Challenge.} Verifier chooses a random bit \(b\in\{0,1\}\) and sends \(b\).
	
	\item \textbf{Response.} Prover responds with a permutation \(\tau\) such that
	\[
	H = \tau(G_b).
	\]
	Concretely:
	\[
	\tau =
	\begin{cases}
		\sigma & (b=0),\\
		\sigma\circ \pi^{-1} & (b=1).
	\end{cases}
	\]
	
	\item \textbf{Verify.} Verifier checks that \(H=\tau(G_b)\). If yes, accept this round.
\end{enumerate}

Repeat \(k\) rounds to reduce error.

\subsubsection{Why it works (completeness + soundness)}
\paragraph{Completeness.}
If \(G_0\cong G_1\) and the prover is honest, then for either challenge \(b\),
the prover can always produce a valid \(\tau\). So each round accepts with probability 1.

\paragraph{Soundness intuition.}
If \(G_0 \not\cong G_1\), then a cheating prover cannot prepare a single \(H\)
that is isomorphic to \emph{both} \(G_0\) and \(G_1\).
So for a random challenge bit \(b\), they can answer at most one of the two possibilities,
hence succeed with probability \(\le 1/2\) per round.
After \(k\) rounds: \(\le 2^{-k}\).

\subsubsection{Why it is (honest-verifier) zero-knowledge}
A simulator for an \emph{honest verifier} (who chooses uniform random \(b\)) does:
\begin{enumerate}
	\item pick \(b\leftarrow\{0,1\}\) uniformly,
	\item pick a random permutation \(\tau\leftarrow S_n\),
	\item set \(H := \tau(G_b)\),
	\item output transcript \((H,b,\tau)\).
\end{enumerate}
This transcript has the same distribution as a real execution:
in the real protocol, \(H\) is always a uniformly permuted copy of whichever graph is challenged.

\subsubsection{Visualization: one round as a dataflow}

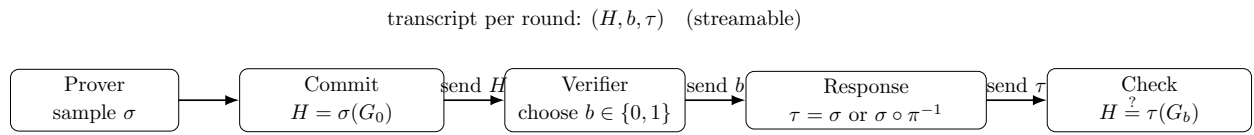
\begin{figure}[t]
	\centering
	\resizebox{\textwidth}{!}{%
		\begin{tikzpicture}[
			>=Latex,
			line cap=round,
			line join=round,
			font=\small,
			box/.style={draw, rounded corners, align=center, minimum height=9mm},
			arr/.style={->, thick},
			node distance=10mm
			]
			\node[box, minimum width=28mm] (p0) {Prover\\ sample $\sigma$};
			
			\node[box, right=10mm of p0, minimum width=34mm] (commit)
			{Commit\\ $H=\sigma(G_0)$};
			
			\node[box, right=10mm of commit, minimum width=30mm] (chal)
			{Verifier\\ choose $b\in\{0,1\}$};
			
			\node[box, right=10mm of chal, minimum width=40mm] (resp)
			{Response\\ $\tau=\sigma$ or $\sigma\circ\pi^{-1}$};
			
			\node[box, right=10mm of resp, minimum width=34mm] (check)
			{Check\\ $H\stackrel{?}{=}\tau(G_b)$};
			
			\draw[arr] (p0) -- (commit);
			\draw[arr] (commit) -- node[above]{send $H$} (chal);
			\draw[arr] (chal) -- node[above]{send $b$} (resp);
			\draw[arr] (resp) -- node[above]{send $\tau$} (check);
			
			\node[align=center] at ($(chal.north)+(0,8mm)$)
			{transcript per round: $(H,b,\tau)$\quad(streamable)};
		\end{tikzpicture}%
	}
	\caption{One GI ZK round is a strict commit--challenge--response pipeline. The verifier’s observable
		data are a small transcript, which is the object ZK talks about.}
	\label{fig:zk-gi-round}
\end{figure}

\subsection{Quantum adversaries and rewinding}

\subsubsection{Why “rewinding” is a thing}
Many classical ZK proofs for malicious verifiers are proven ZK by a simulator that
\emph{rewinds} the verifier:
it runs the verifier, sees a challenge, and if it is “unhelpful,” it rewinds and tries again
until it gets a favorable challenge.

This is a proof technique, not a protocol step, but it is central to ZK theory.

\subsubsection{What changes with quantum verifiers}
A quantum verifier may maintain a quantum state entangled with the transcript.
Naive rewinding can fail because:
\begin{itemize}
	\item measuring to decide “rewind or not” can collapse the verifier’s state,
	\item and cloning a quantum state to restore it is impossible (no-cloning).
\end{itemize}

\subsubsection{High-level takeaway (without diving into proofs)}
\begin{itemize}
	\item Some ZK protocols that rely on classical rewinding proofs need new proof techniques
	against quantum adversaries (e.g.\ “quantum rewinding” methods).
	\item In practice, modern post-quantum ZK systems often avoid fragile rewinding arguments by using
	special structures (e.g.\ \(\Sigma\)-protocols + Fiat--Shamir in the random oracle model,
	or other transformations) and proving security in quantum-aware models.
\end{itemize}

\subsection{Zero knowledge for NP and post-quantum security}

\subsubsection{ZK for NP: the headline}
The landmark fact: every NP statement has a zero-knowledge proof (under standard assumptions).
This means: for any efficiently verifiable relation \(R(x,w)=1\),
there exist protocols where the prover convinces the verifier that “\(\exists w: R(x,w)=1\)”
without revealing \(w\).

\subsubsection{What “post-quantum ZK” means operationally}
A ZK system is “post-quantum” if:
\begin{itemize}
	\item \textbf{Soundness holds} even when the prover (attacker) has a quantum computer,
	\item \textbf{Zero-knowledge holds} even when the verifier is a quantum algorithm.
\end{itemize}
To achieve this, constructions typically use primitives believed hard for quantum computers,
notably lattice-based assumptions (e.g.\ Learning With Errors, Module-LWE, etc.).

\subsubsection{Engineering viewpoint: PQ means your transcript security is still safe under Shor}
If your protocol uses RSA/DSA/ECDSA-style number theory as a security pillar,
Shor’s algorithm breaks it.
If it uses lattice assumptions, we currently believe it survives known quantum attacks.
So “PQ ZK” is mostly about:
\[
\text{replace discrete-log factoring assumptions} \ \to\ \text{lattice assumptions, while preserving ZK}.
\]

\subsection{Practice lab: transcript-as-algorithm (GI toy lab)}
This lab treats the protocol as a program that emits and checks transcripts.

\subsubsection{Toy setup (small graphs)}
Pick \(n=6\) and generate:
\begin{itemize}
	\item a base graph \(G_0\),
	\item a random secret permutation \(\pi\in S_n\),
	\item set \(G_1 := \pi(G_0)\).
\end{itemize}

\subsubsection{Implement the honest prover and verifier}
For \(k\) rounds:
\begin{enumerate}
	\item prover samples \(\sigma\), computes \(H=\sigma(G_0)\), sends \(H\),
	\item verifier samples \(b\leftarrow\{0,1\}\),
	\item prover outputs \(\tau\) (as defined above),
	\item verifier checks \(H=\tau(G_b)\).
\end{enumerate}

\subsubsection{Implement the simulator for HVZK}
Write a simulator that outputs \((H,b,\tau)\) by:
\[
b\leftarrow\{0,1\},\quad \tau\leftarrow S_n,\quad H:=\tau(G_b).
\]
Empirically compare distributions of a few easy-to-check statistics:
\begin{itemize}
	\item degree sequence of \(H\),
	\item number of triangles in \(H\),
	\item adjacency-matrix hash of \(H\) after canonical relabeling (if you have it).
\end{itemize}
They should match between real transcripts and simulated transcripts.

\subsubsection{What you should observe}
Even though the prover “knows” \(\pi\), the transcript does not leak \(\pi\):
the commit \(H\) is just a random relabeling, and the response \(\tau\) is also random-looking.

\subsection{Practice lab: FPGA-style streaming verifier (challenge--response)}

\subsubsection{Why this lab}
Verifier logic is often a small deterministic checker operating on streamed messages.
This is exactly the FPGA-friendly part of ZK: \emph{verification} is usually cheaper than proving.

\subsubsection{Streaming interface model}
Assume per round you receive:
\[
\texttt{pkt} = (H,\ b,\ \tau,\ t,\ \texttt{round\_id}),
\]
where \(t\) is a timestamp (or cycle counter).
The verifier outputs:
\[
\texttt{accept\_bit} \in \{0,1\},
\quad \text{and counters (fail reasons)}.
\]

\subsubsection{Verifier as a pipeline}
\begin{figure}[t]
	\centering
	\begin{tikzpicture}[scale=0.8, line cap=round, line join=round]
		\node (in) at (0,0)
		[draw, rounded corners, minimum width=3.6cm, minimum height=0.9cm, align=center]
		{\shortstack{\small input stream\\\small $(H,b,\tau)$}};
		\node (parse) at (4.2,0)
		[draw, rounded corners, minimum width=3.2cm, minimum height=0.9cm, align=center]
		{\shortstack{\small parse\\\small validate sizes}};
		\node (apply) at (8.2,0)
		[draw, rounded corners, minimum width=3.8cm, minimum height=0.9cm, align=center]
		{\shortstack{\small apply $\tau$\\\small compute $\tau(G_b)$}};
		\node (cmp) at (12.6,0)
		[draw, rounded corners, minimum width=3.6cm, minimum height=0.9cm, align=center]
		{\shortstack{\small compare\\\small $H\stackrel{?}{=}\tau(G_b)$}};
		\node (out) at (16.6,0)
		[draw, rounded corners, minimum width=3.2cm, minimum height=0.9cm, align=center]
		{\shortstack{\small output\\\small accept + flags}};
		
		\draw[->] (in) -- (parse);
		\draw[->] (parse) -- (apply);
		\draw[->] (apply) -- (cmp);
		\draw[->] (cmp) -- (out);
		
		\node[align=center] at (8.3,1.25)
		{\small verifier is a deterministic streaming checker (hardware-friendly)};
	\end{tikzpicture}
	\caption{Streaming verifier viewpoint: each round is a small packet; verification is a fixed pipeline:
		parse \(\to\) apply permutation \(\to\) compare.}
	\label{fig:zk-stream-verifier}
\end{figure}
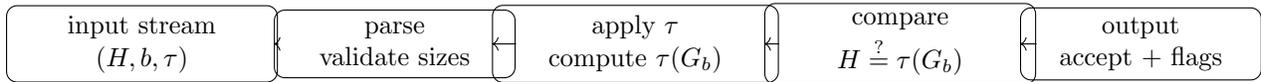

\subsubsection{Micro-optimization note (what dominates)}
Graph permutation checking is mostly memory access:
\begin{itemize}
	\item storing adjacency in bitsets,
	\item applying \(\tau\) as index wiring,
	\item XOR/AND + popcount-style operations.
\end{itemize}
So the “FPGA shape” is: wide bitsets + streaming compare + counters.

\subsection{Exercises}

\begin{exercise}[GI protocol: write the explicit verifier check]
	Represent a graph on \(n\) vertices by its adjacency matrix \(A\in\{0,1\}^{n\times n}\).
	Show that applying a permutation \(\tau\in S_n\) corresponds to
	\[
	A \mapsto P_\tau A P_\tau^\top,
	\]
	where \(P_\tau\) is the permutation matrix.
	Then state the verifier check \(H=\tau(G_b)\) as a matrix equality.
\end{exercise}
\noindent\textbf{Solution.}
Let \(P_\tau\) be defined by \((P_\tau)_{i,\tau(i)}=1\).
Relabeling vertices by \(\tau\) sends an edge \((i,j)\) to \((\tau(i),\tau(j))\),
which exactly corresponds to conjugation:
\[
A'_{\tau(i),\tau(j)} = A_{i,j}.
\]
In matrix form this is \(A' = P_\tau A P_\tau^\top\).
Thus the verifier check is:
\[
A_H \stackrel{?}{=} P_\tau\, A_{G_b}\, P_\tau^\top.
\]

\begin{exercise}[Soundness bound]
	Assume \(G_0\not\cong G_1\).
	Prove that for any cheating prover strategy in one round,
	\[
	\Pr[\textsf{accept}] \le \frac12.
	\]
\end{exercise}
\noindent\textbf{Solution.}
A one-round transcript must start with some committed graph \(H\).
For the verifier to accept under challenge \(b\), the prover must produce \(\tau\) such that
\(H=\tau(G_b)\), i.e.\ \(H\cong G_b\).
If \(G_0\not\cong G_1\), then \(H\) cannot be isomorphic to both \(G_0\) and \(G_1\).
So at best the prover can answer correctly for one value of \(b\).
Since \(b\) is uniform, acceptance probability is \(\le 1/2\).

\begin{exercise}[HVZK simulator distribution]
	Write the honest-verifier simulator for the GI protocol and show that
	the distribution of \(H\) conditioned on \(b\) is uniform over relabelings of \(G_b\).
\end{exercise}
\noindent\textbf{Solution.}
The simulator samples \(\tau\leftarrow S_n\) uniformly and sets \(H=\tau(G_b)\).
Thus conditioned on \(b\), \(H\) is exactly a uniformly random relabeling of \(G_b\),
which matches the real protocol: in the real protocol \(H=\sigma(G_0)\) and if \(b=1\),
the prover responds with \(\sigma\circ\pi^{-1}\), so the verifier sees a random relabeling
consistent with \(G_1\) as well.

\begin{exercise}[Streaming verifier contract]
	Design a minimal packet schema for a streaming verifier that checks \(H=\tau(G_b)\).
	Specify fields, bit-widths (symbolically), and the accept/fail flags you would output.
\end{exercise}
\noindent\textbf{Solution.}
A minimal schema includes:
\[
\texttt{round\_id}:\ \lceil\log_2 k\rceil \text{ bits},\quad
b:1\text{ bit},\quad
\tau:\ n\cdot \lceil\log_2 n\rceil \text{ bits},\quad
H:\ n(n-1)/2 \text{ bits (upper triangle)}.
\]
Outputs:
\texttt{accept\_bit},
\texttt{fail\_parse} (bad lengths/range),
\texttt{fail\_perm} (not a bijection),
\texttt{fail\_check} (matrix mismatch),
plus counters per flag and timestamps for latency measurement.

\begin{exercise}[Post-quantum viewpoint]
	Explain (5--8 lines) why “post-quantum ZK” requires considering quantum verifiers, not only quantum provers.
\end{exercise}
\noindent\textbf{Solution.}
Soundness addresses cheating provers, so a quantum prover matters there.
Zero-knowledge addresses what the verifier can learn, so a quantum verifier matters there:
it may run quantum distinguishing algorithms on transcripts and may keep quantum side information.
Therefore PQ ZK requires ZK proofs that remain valid when the verifier is quantum,
not just when the prover is quantum.
	
\section{Quantum Cryptography in Practice: From Circuits to Hardware}
\label{sec:qcrypto-practice}

\subsection{Objective: turning security principles into implementable systems}
This chapter is a \emph{systems bridge}:
we start from circuit-level protocol logic (state prep \(\to\) measurement \(\to\) classical discussion),
then turn it into an implementable pipeline with
(i) streaming data structures, (ii) latency/throughput budgets, and (iii) FPGA-friendly kernels.

\medskip
\noindent\textbf{Two deliverables.}
\begin{itemize}
	\item \textbf{Protocol view (correctness + security knobs):}
	what is computed, what is announced, what is kept secret, and what parameters trigger abort.
	\item \textbf{Hardware view (interfaces + streaming kernels):}
	what is the wire format, what buffers exist, where counters live, and how to certify correctness online.
\end{itemize}

\subsection{Quantum states as cryptographic carriers}

\subsubsection{Carrier states: what makes them useful}
A cryptographic carrier is a physical state whose \emph{measurement behavior} enforces a constraint:
\begin{itemize}
	\item \textbf{No-cloning:} you cannot copy an unknown qubit state perfectly.
	\item \textbf{Disturbance:} measuring in the wrong basis introduces errors.
	\item \textbf{Basis choice as a secret:} the sender/receiver basis choice acts like a one-time key that selects which measurement yields information.
\end{itemize}

\subsubsection{BB84 signal states (prepare-and-measure)}
BB84 uses four states grouped into two bases:
\[
\mathcal{B}_Z=\{\ket{0},\ket{1}\},\qquad
\mathcal{B}_X=\{\ket{+},\ket{-}\},
\qquad
\ket{\pm}=\frac{\ket{0}\pm\ket{1}}{\sqrt2}.
\]
A sender chooses a basis \(b\in\{Z,X\}\) and a bit \(a\in\{0,1\}\), then sends \(\ket{\psi_{a,b}}\).

\subsubsection{E91 signal states (entanglement-based)}
E91 uses shared entanglement, ideally an EPR pair:
\[
\ket{\Phi^+}=\frac{\ket{00}+\ket{11}}{\sqrt2}.
\]
Each party measures their half in a randomly chosen basis.
Correlations (and Bell-inequality style tests) certify the presence of entanglement and bound eavesdropping.

\subsubsection{Circuit pictures: BB84 and E91 primitives}
\begin{figure}[t]
	\centering
	\begin{tikzpicture}[scale=1.0, line cap=round, line join=round]
		\node[align=left] at (0,2.2) {\small \textbf{BB84 (Alice prepare)}};
		
		\draw (0,1.6) -- (6.6,1.6);
		\node at (-0.25,1.6) {\small $q$};
		
		\node[draw, minimum width=0.6cm, minimum height=0.45cm] (Xg) at (1.5,1.6) {\small $X^{a}$};
		\node[draw, minimum width=0.6cm, minimum height=0.45cm] (Hg) at (3.2,1.6) {\small $H^{b}$};
		\draw[->] (0.3,1.6) -- (Xg.west);
		\draw[->] (Xg.east) -- (Hg.west);
		\draw[->] (Hg.east) -- (6.6,1.6);
		
		\node[align=left] at (6.9,1.6) {\small send over channel};
		
		\node[align=left] at (0,0.9) {\small \(a\in\{0,1\}\): bit,\quad \(b\in\{0,1\}\): basis (0=Z, 1=X).};
		
		\node[align=left] at (0,-0.2) {\small \textbf{E91 (entangled source)}};
		
		\draw (0,-0.8) -- (6.6,-0.8);
		\draw (0,-1.5) -- (6.6,-1.5);
		\node at (-0.25,-0.8) {\small $A$};
		\node at (-0.25,-1.5) {\small $B$};
		
		\node[draw, minimum width=0.6cm, minimum height=0.45cm] (H2) at (1.5,-0.8) {\small $H$};
		\draw[->] (0.3,-0.8) -- (H2.west);
		\draw[->] (H2.east) -- (3.0,-0.8);
		
		\node[draw, minimum width=0.8cm, minimum height=0.45cm] (CNOT) at (3.6,-1.15) {\small CNOT};
		\draw (3.6,-0.8) -- (3.6,-1.15+0.225);
		\draw[->] (3.0,-0.8) -- (CNOT.west |- 3.0,-0.8);
		\draw[->] (0.3,-1.5) -- (CNOT.west |- 0.3,-1.5);
		
		\draw[->] (CNOT.east |- 4.6,-0.8) -- (6.6,-0.8);
		\draw[->] (CNOT.east |- 4.6,-1.5) -- (6.6,-1.5);
		
		\node[align=left] at (6.9,-0.8) {\small send $A$};
		\node[align=left] at (6.9,-1.5) {\small send $B$};
	\end{tikzpicture}
	\caption{Circuit primitives. Top: BB84 preparation as conditional \(X\) then conditional \(H\).
		Bottom: E91 entanglement source (Bell pair) via \(H\) + CNOT.}
	\label{fig:qcrypto-primitives}
\end{figure}

\subsection{Measurement, noise, and error rates}

\subsubsection{Measurement as a classical interface}
A QKD system produces \emph{classical records}:
\[
\text{(time, basis choice, bit outcome, detector metadata)}.
\]
These records are the input to the classical post-processing pipeline.

\subsubsection{Noise sources to model (minimum set)}
For practice-oriented engineering, separate at least:
\begin{itemize}
	\item \textbf{state-prep errors} (imperfect modulation, pulse shaping),
	\item \textbf{channel noise/loss} (depolarization, phase drift, attenuation),
	\item \textbf{detector errors} (dark counts, efficiency mismatch),
	\item \textbf{basis misalignment} (rotation error between \(Z\) and \(X\) frames),
	\item \textbf{timestamp skew} (synchronization and framing errors).
\end{itemize}

\subsubsection{QBER as the primary operational metric}
The quantum bit error rate (QBER) is computed on a chosen subset (often after sifting):
\[
\mathrm{QBER}=\frac{\#\{\text{mismatched bits}\}}{\#\{\text{compared bits}\}}.
\]
Operationally, QBER drives:
\begin{itemize}
	\item \textbf{abort rules} (too noisy \(\Rightarrow\) insecurity),
	\item \textbf{error-correction cost} (more leakage),
	\item \textbf{privacy amplification strength} (more compression).
\end{itemize}

\subsection{Classical post-processing pipeline}

\subsubsection{The canonical streaming pipeline}
A practical QKD stack is a stream processor. A minimal pipeline:
\[
\text{raw detections}
\to \text{sifting}
\to \text{parameter estimation (QBER)}
\to \text{error correction + leak accounting}
\to \text{privacy amplification}
\to \text{final key}.
\]

\subsubsection{What each stage consumes/produces}
\begin{itemize}
	\item \textbf{Sifting:}
	inputs \((b_A,b_B)\) and outcomes; output is indices where bases match (plus optional sampling indices).
	\item \textbf{Parameter estimation:}
	inputs sampled matched indices; output \(\widehat{\mathrm{QBER}}\) and confidence bounds.
	\item \textbf{Error correction (EC):}
	inputs matched strings; output reconciled strings + \(\mathrm{leak}_{\mathrm{EC}}\) (bits revealed).
	\item \textbf{Privacy amplification (PA):}
	inputs reconciled string and security parameters; output shorter key with chosen hash seed logged.
\end{itemize}

\subsubsection{Hardware-friendly separation}
For FPGA practice, split into:
\begin{itemize}
	\item \textbf{fast path (FPGA):} counters, sifting indices, QBER estimation, streaming hashes,
	\item \textbf{slow path (host CPU):} protocol control, EC strategy selection, transcript persistence.
\end{itemize}

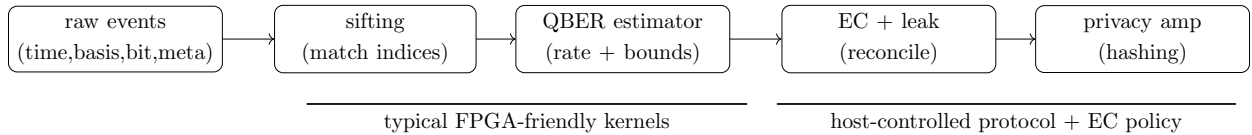
\begin{figure}[t]
	\centering
	\resizebox{\linewidth}{!}{%
		\begin{tikzpicture}[scale=1.0, line cap=round, line join=round]
			\node (raw) at (0,0)
			[draw, rounded corners, minimum width=3.2cm, minimum height=0.9cm,
			align=center, text width=3.2cm]
			{\small raw events\par\small (time,basis,bit,meta)};
			
			\node (sift) at (4.2,0)
			[draw, rounded corners, minimum width=3.0cm, minimum height=0.9cm,
			align=center, text width=3.0cm]
			{\small sifting\par\small (match indices)};
			
			\node (qber) at (8.2,0)
			[draw, rounded corners, minimum width=3.2cm, minimum height=0.9cm,
			align=center, text width=3.2cm]
			{\small QBER estimator\par\small (rate + bounds)};
			
			\node (ec) at (12.5,0)
			[draw, rounded corners, minimum width=3.2cm, minimum height=0.9cm,
			align=center, text width=3.2cm]
			{\small EC + leak\par\small (reconcile)};
			
			\node (pa) at (16.6,0)
			[draw, rounded corners, minimum width=3.2cm, minimum height=0.9cm,
			align=center, text width=3.2cm]
			{\small privacy amp\par\small (hashing)};
			
			\draw[->] (raw) -- (sift);
			\draw[->] (sift) -- (qber);
			\draw[->] (qber) -- (ec);
			\draw[->] (ec) -- (pa);
			
			\draw[thick] (3.1,-1.05) -- (10.2,-1.05);
			\node[align=center] at (6.65,-1.35) {\small typical FPGA-friendly kernels};
			
			\draw[thick] (10.7,-1.05) -- (18.1,-1.05);
			\node[align=center] at (14.4,-1.35) {\small host-controlled protocol + EC policy};
		\end{tikzpicture}%
	}
	\caption{Canonical QKD post-processing as a streaming pipeline.}
	\label{fig:qcrypto-qkd-pipeline}
\end{figure}

\subsection{Case study I: BB84 end-to-end}

\subsubsection{Protocol transcript (what must be logged)}
A practical BB84 implementation needs a transcript sufficient for:
debugging, reproducibility, and security audits. Minimal per-block fields:
\begin{itemize}
	\item block ID, start/end timestamps, clock sync status,
	\item basis streams \(b_A, b_B\) (or their hashes/seeds),
	\item raw outcomes \(x_A, x_B\) (or compressed forms),
	\item sifting mask (indices kept + sample indices),
	\item QBER estimate and bounds, abort decision,
	\item EC method ID, \(\mathrm{leak}_{\mathrm{EC}}\), success/fail flags,
	\item PA hash family ID, seed, output key length.
\end{itemize}

\subsubsection{Circuit-to-probability check (sanity test)}
If Alice sends \(\ket{0}\) and Bob measures in \(Z\), Bob gets 0 with probability 1 (ideal).
If Bob measures in \(X\), Bob gets 0/1 each with probability \(1/2\).
This basis-mismatch randomness is the \emph{sifting engine}.

\subsubsection{Practical BB84 flow (block-structured)}
\begin{enumerate}
	\item \textbf{Quantum phase (stream):} Alice sends \(N\) qubits; Bob records outcomes with timestamps.
	\item \textbf{Sifting (classical):} announce bases; keep matched indices.
	\item \textbf{Parameter estimation:} reveal a random subset; compute QBER; abort if too large.
	\item \textbf{Reconciliation (EC):} correct mismatches; track leakage.
	\item \textbf{Privacy amplification:} apply universal hashing to compress away Eve's information.
\end{enumerate}

\subsection{Case study II: entanglement-based QKD (E91)}

\subsubsection{What changes vs.\ BB84}
E91 replaces ``prepare-and-measure'' with a \emph{source} that emits entangled pairs.
Each side measures with randomly chosen settings.

\subsubsection{Two simultaneous goals}
\begin{itemize}
	\item \textbf{Key generation:} use a subset of settings that produce correlated bits.
	\item \textbf{Entanglement test:} use additional settings to estimate correlation statistics (Bell-type test),
	which bounds an adversary's information.
\end{itemize}

\subsubsection{Hardware implication}
E91 tends to be more sensitive to:
\begin{itemize}
	\item \textbf{pairing and timing}: matching which detection on Alice corresponds to which on Bob,
	\item \textbf{coincidence windows}: streaming join between two timestamped streams,
	\item \textbf{metadata rates}: detectors produce high-rate event logs even when the key rate is low.
\end{itemize}
This pushes even more strongly toward streaming accelerators.

\subsection{From protocols to architecture}

\subsubsection{A minimal architecture diagram (syndrome-style thinking)}
Treat QKD as ``measurement stream \(\to\) classical decision'' like QEC:
\begin{itemize}
	\item inputs are high-rate event packets,
	\item outputs are (i) a key block and (ii) explicit flags/metrics.
\end{itemize}

\subsubsection{Interfaces (minimum wire format)}
A practical message schema for FPGA/host integration:
\begin{itemize}
	\item \texttt{EventPacket}: \texttt{block\_id}, \texttt{tstamp}, \texttt{basis}, \texttt{bit}, \texttt{meta}
	\item \texttt{SiftMask}: \texttt{block\_id}, \texttt{kept\_count}, \texttt{bitmap/hash}, \texttt{sample\_seed}
	\item \texttt{ParamReport}: \texttt{block\_id}, \texttt{qber}, \texttt{ci\_low}, \texttt{ci\_high}, \texttt{abort}
	\item \texttt{ECReport}: \texttt{block\_id}, \texttt{method\_id}, \texttt{leak\_ec}, \texttt{ok}
	\item \texttt{PAReport}: \texttt{block\_id}, \texttt{hash\_id}, \texttt{seed}, \texttt{out\_len}
\end{itemize}
All messages must include a \texttt{schema\_version} and a \texttt{build\_hash} for reproducibility.

\subsubsection{FPGA kernels that actually matter}
For practice projects, focus on kernels that are:
fixed-format, streaming, and counter-heavy:
\begin{itemize}
	\item sifting mask generation (basis match \(\Rightarrow\) keep/drop),
	\item QBER estimation (mismatch counting on sampled indices),
	\item streaming hashing (privacy amplification primitive),
	\item timestamp alignment / coincidence windowing (E91).
\end{itemize}

\subsection{Exercises (practice-oriented)}

\begin{exercise}[BB84 sifting as a streaming filter]
	You receive a stream of tuples \((t, b_A, b_B, x_A, x_B)\) for a block.
	Design a one-pass streaming algorithm that outputs:
	(i) the kept-index mask, and (ii) counters \(\#\text{kept}\), \(\#Z\), \(\#X\).
\end{exercise}
\noindent\textbf{Solution sketch.}
Initialize counters to zero and a bitmask buffer (or rolling hash) for indices.
For each record at index \(i\):
if \(b_A=b_B\), set mask bit \(i\leftarrow 1\) and increment \(\#\text{kept}\).
Also increment \(\#Z\) if \(b_A=Z\) else increment \(\#X\).
Otherwise set mask bit \(i\leftarrow 0\).
This is a pure streaming filter with fixed memory per index representation.

\begin{exercise}[QBER estimation with sampling]
	Assume you sample \(m\) kept indices uniformly at random (using a PRNG seed).
	Give a streaming method to estimate QBER on the sample without storing all bits.
\end{exercise}
\noindent\textbf{Solution sketch.}
Use the seed to generate the next sampled kept-index threshold (or reservoir-like selection).
Maintain counters: \(\#\text{sample}\) and \(\#\text{mismatch}\).
When index \(i\) is selected and \(b_A=b_B\), compare \(x_A\) and \(x_B\):
if different, increment \(\#\text{mismatch}\).
At end, \(\widehat{\mathrm{QBER}}=\#\text{mismatch}/\#\text{sample}\).
Optionally compute a confidence interval using a binomial bound on host.

\begin{exercise}[Privacy amplification as a streaming hash]
	Choose a universal hash family suitable for streaming (e.g.\ Toeplitz hashing).
	Describe (i) what seed must be shared/logged, and (ii) how the hash can be computed as a bitstream.
\end{exercise}
\noindent\textbf{Solution sketch.}
A Toeplitz matrix hash is specified by a seed of length \(n+m-1\) (for mapping \(n\) bits to \(m\) bits).
The FPGA can generate Toeplitz rows on the fly from the seed and compute each output bit as an XOR of selected input bits.
Log: \texttt{hash\_id}, \texttt{seed}, \texttt{in\_len}=n, \texttt{out\_len}=m, and the block ID.

\begin{exercise}[E91 coincidence window join]
	You receive two timestamped event streams \(A\) and \(B\).
	Design a hardware-friendly join that pairs events whose timestamps differ by at most \(\Delta\).
	What buffers do you need and what is the worst-case overflow condition?
\end{exercise}
\noindent\textbf{Solution sketch.}
Maintain two FIFO buffers for recent events from each stream within the window.
When an event arrives on stream \(A\), pop old events from \(B\) with \(t_B<t_A-\Delta\),
then match against the earliest \(B\) with \(t_B\le t_A+\Delta\) (policy: greedy or best-closest).
Symmetric for \(B\).
Worst-case overflow happens when one stream rate spikes or the other stream stalls, causing the window FIFO to grow beyond its depth.
A safe design includes backpressure + explicit overflow flags.

\begin{exercise}[Engineering: abort rules and logging contract]
	Propose an abort rule based on QBER and a logging contract that lets you reproduce the abort decision offline.
\end{exercise}
\noindent\textbf{Solution sketch.}
Abort if \(\widehat{\mathrm{QBER}}>\tau\) where \(\tau\) is a configured threshold (protocol-dependent) and if the confidence bound upper endpoint exceeds \(\tau\).
Log: \texttt{block\_id}, sample seed, \(\#\text{sample}\), \(\#\text{mismatch}\), computed \(\widehat{\mathrm{QBER}}\), CI method ID and parameters, and the threshold \(\tau\).
This is sufficient to replay the decision exactly.
	
	\appendix
	
	\clearpage
\section{Appendices: Hands-on Quantum Control with FPGA (Lattice iCEstick)}
\addcontentsline{toc}{section}{Appendices: Hands-on Quantum Control with FPGA (Lattice iCEstick)}
\markboth{Appendices}{Appendices}

\subsection{How to read these appendices (no projects, just concrete design examples)}
These appendices give a \emph{hands-on mental model} for low-latency classical paths that appear in
quantum-control and QEC pipelines, using the Lattice iCEstick as the concrete reference platform.
We will not develop end-to-end ``projects'' here; instead, each subsection is a \emph{design pattern}
with (i) a precise contract, (ii) a minimal working example, and (iii) a visualization you can keep
next to the code.

\medskip
\noindent\textbf{What you should get out of this appendix block.}
You should be able to look at a streaming hardware block and answer:
\begin{itemize}
	\item What is the \emph{input contract} (when is data valid, how is backpressure handled)?
	\item What is the \emph{state} (what must be stored, and for how long)?
	\item What is the \emph{bounded-time guarantee} (worst-case cycles per message)?
	\item What is the \emph{fail-closed policy} (what happens under corruption/overflow/misalignment)?
	\item What is the \emph{observable telemetry} (flags, counters, histograms) for debugging p99/p999?
\end{itemize}

\medskip
\noindent\textbf{Guiding viewpoint.}
In real-time quantum pipelines, \emph{latency is part of correctness}:
a mathematically correct correction that arrives too late is a wrong correction.

\subsection{Design vocabulary (the only primitives we assume)}
We will repeatedly use four primitives:
\begin{enumerate}
	\item \textbf{Streaming handshake:} \texttt{valid/ready} with stable data under backpressure.
	\item \textbf{Buffering:} FIFO + explicit occupancy + explicit overflow policy.
	\item \textbf{Parsing/framing:} a small FSM that prevents byte/word drift.
	\item \textbf{Time:} cycle counter timestamps for on-device latency measurement.
\end{enumerate}

\subsection{Example-driven design patterns (each with a self-contained figure)}
The subsections below are \emph{examples}, not projects. Each is small enough to implement in one file,
but the point here is the \emph{design}: contracts, hazards, and observability.

\subsubsection{Example A: The universal stream contract (valid/ready)}
\noindent\textbf{Problem.}
You want to connect two modules without drops: a producer emits words, a consumer may stall.

\medskip
\noindent\textbf{Contract (one line).}
A transfer happens on the rising clock edge iff \texttt{valid \& ready}.
If \texttt{valid=1} and \texttt{ready=0}, the producer must hold \texttt{data} stable.

\begin{figure}[t]
	\centering
	\begin{tikzpicture}[
		font=\small,
		sig/.style={draw, rounded corners, minimum width=2.8cm, minimum height=0.9cm, align=center},
		arr/.style={-Latex, thick},
		lab/.style={font=\footnotesize, align=left}
		]
		\node[sig] (prod) {Producer};
		\node[sig, right=4.2cm of prod] (cons) {Consumer};
		
		\draw[arr] (prod.east) -- node[above, lab] {data[W-1:0]\\valid} (cons.west);
		\draw[arr] (cons.west) -- node[below, lab] {ready} (prod.east);
		
		\node[lab, below=0.9cm of prod, xshift=2.1cm] {
			\textbf{Rule:} accept iff \texttt{valid \& ready}.\\
			\textbf{Backpressure:} consumer deasserts \texttt{ready}.\\
			\textbf{No drop:} producer holds \texttt{data} stable while \texttt{valid=1}.
		};
	\end{tikzpicture}
	\caption{The \texttt{valid/ready} streaming contract. This single rule lets you safely compose pipelines.}
	\label{fig:app-vr-contract}
\end{figure}

\noindent\textbf{Concrete numbers (timing intuition).}
Suppose \texttt{valid} is high for 5 cycles and the consumer stalls for 2 cycles in the middle.
Then exactly 5 transfers occur, but later in time; the word order is preserved without special glue logic.

\subsubsection{Example B: FIFO is the physical shape of buffering}
\noindent\textbf{Problem.}
Arrival is bursty; service is bounded but not constant. You need to absorb jitter.

\medskip
\noindent\textbf{Design choices you must make explicitly.}
\begin{itemize}
	\item Depth \(D\): maximum burst you can absorb.
	\item Overflow policy: \emph{never silently drop}. Either (i) assert \texttt{overflow} and stop accepting
	(\emph{fail closed}) or (ii) drop with an explicit \texttt{drop\_count} and \texttt{misalign} flag.
	\item Observability: occupancy counter + high-water mark.
\end{itemize}

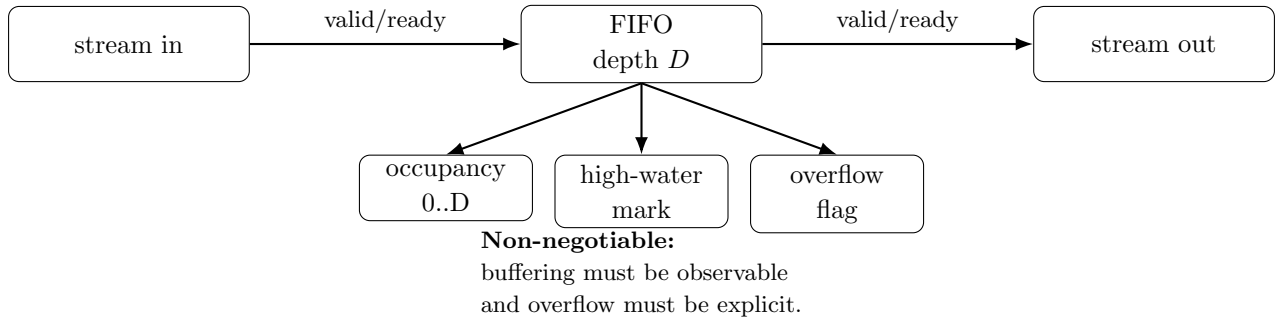
\begin{figure}[t]
	\centering
	\begin{tikzpicture}[
		font=\small,
		box/.style={draw, rounded corners, minimum width=3.2cm, minimum height=1.0cm, align=center},
		tiny/.style={draw, rounded corners, minimum width=2.3cm, minimum height=0.75cm, align=center},
		arr/.style={-Latex, thick},
		node distance=1.1cm and 1.4cm
		]
		\node[box] (in) {stream in};
		\node[box, right=3.6cm of in] (fifo) {FIFO\\depth \(D\)};
		\node[box, right=3.6cm of fifo] (out) {stream out};
		
		\draw[arr] (in) -- node[above] {\footnotesize valid/ready} (fifo);
		\draw[arr] (fifo) -- node[above] {\footnotesize valid/ready} (out);
		
		\node[tiny, below=0.95cm of fifo, xshift=-2.6cm] (occ) {occupancy\\0..D};
		\node[tiny, below=0.95cm of fifo] (hwm) {high-water\\mark};
		\node[tiny, below=0.95cm of fifo, xshift=2.6cm] (ovf) {overflow\\flag};
		
		\draw[arr] (fifo.south) -- (occ.north);
		\draw[arr] (fifo.south) -- (hwm.north);
		\draw[arr] (fifo.south) -- (ovf.north);
		
		\node[align=left, font=\footnotesize] at ($(occ.south)+(2.6,-0.7)$) {
			\textbf{Non-negotiable:}\\
			buffering must be observable\\
			and overflow must be explicit.
		};
	\end{tikzpicture}
	\caption{FIFO design pattern: decouple arrival and service while keeping explicit visibility into backlog.}
	\label{fig:app-fifo}
\end{figure}

\noindent\textbf{Concrete example (why depth matters).}
If arrivals can spike to 8 words in 8 cycles but your kernel drains 1 word per 2 cycles,
then \(D\ge 8\) is not enough unless you also ensure the burst gap gives time to drain.
The FIFO gives you \emph{time}, but only up to \(D\) words.

\subsubsection{Example C: Framing FSM prevents ``word drift''}
\noindent\textbf{Problem.}
If you read a byte stream without framing, one dropped byte shifts the interpretation forever.

\medskip
\noindent\textbf{Minimal framed message.}
We use a tiny header:
\[
\texttt{SOF} \;|\; \texttt{len} \;|\; \texttt{type} \;|\; \texttt{payload} \;|\; \texttt{CRC}.
\]
The framing FSM must implement a \emph{resynchronization rule}: on error, discard until a new \texttt{SOF}.

\begin{figure}[t]
	\centering
	\resizebox{\linewidth}{!}{%
		\begin{tikzpicture}[
			font=\small,
			st/.style={draw, rounded corners, align=center, minimum height=9mm},
			arr/.style={-Latex, thick},
			lab/.style={font=\footnotesize},
			node distance=10mm and 16mm
			]
			\node[st, minimum width=28mm] (idle) {IDLE\\\footnotesize seek SOF};
			\node[st, right=18mm of idle, minimum width=32mm] (hdr) {READ HDR\\\footnotesize len,type};
			\node[st, right=18mm of hdr,  minimum width=32mm] (pay) {READ\\\footnotesize PAYLOAD};
			\node[st, right=18mm of pay,  minimum width=30mm] (crc) {CHECK\\\footnotesize CRC};
			
			\node[st, below=12mm of pay, minimum width=32mm] (err) {ERROR\\\footnotesize resync};
			
			\draw[arr,<->] (idle) -- node[above,lab]{SOF} (hdr);
			
			\draw[arr] (hdr) -- node[above,lab]{ok} (pay);
			\draw[arr] (pay) -- node[above,lab]{done} (crc);
			
			\draw[arr] (hdr.south) to[out=-60,in=150] node[pos=0.55,above,lab]{invalid} (err.west);
			\draw[arr] (pay.south) -- node[right,lab]{timeout/ovf} (err.north);
			\draw[arr] (crc.south) to[out=-120,in=30] node[pos=0.55,above,lab]{fail} (err.east);
			
			\draw[arr] (err.west) to[out=180,in=-120,looseness=1.12]
			node[pos=0.55,below,lab]{drop until SOF} (idle.south);
			
			\node[align=center, font=\footnotesize] at ($(err.south)+(0,-9mm)$) {%
				\textbf{Fail-closed:} do not emit ambiguous outputs under corruption.\\
				Emit a flag + resynchronize.
			};
		\end{tikzpicture}%
	}
	\caption{Framing FSM: a tiny amount of control logic prevents silent stream misinterpretation.}
	\label{fig:app-framing-fsm}
\end{figure}

\begin{figure}[t]
	\centering
	\resizebox{\linewidth}{!}{%
		\begin{tikzpicture}[
			font=\small,
			st/.style={draw, rounded corners, align=center, minimum height=9mm},
			arr/.style={-Latex, thick},
			lab/.style={font=\footnotesize},
			node distance=10mm and 16mm
			]
			\node[st, minimum width=28mm] (idle) {IDLE\\\footnotesize seek SOF};
			\node[st, right=18mm of idle, minimum width=32mm] (hdr) {READ HDR\\\footnotesize len,type};
			\node[st, right=18mm of hdr,  minimum width=32mm] (pay) {READ\\\footnotesize PAYLOAD};
			\node[st, right=18mm of pay,  minimum width=30mm] (crc) {CHECK\\\footnotesize CRC};
			
			\node[st, below=12mm of pay, minimum width=32mm] (err) {ERROR\\\footnotesize resync};
			
			\draw[arr] (idle) -- node[above,lab]{SOF} (hdr);
			\draw[arr] (hdr) -- node[above,lab]{ok} (pay);
			\draw[arr] (pay) -- node[above,lab]{done} (crc);
			\draw[arr] (crc) -- node[above,lab]{pass} (idle);
			
			\draw[arr] (hdr.south) to[out=-90,in=160] node[pos=0.55,left,lab]{invalid} (err.west);
			\draw[arr] (pay.south) -- node[right,lab]{timeout/ovf} (err.north);
			\draw[arr] (crc.south) to[out=-90,in=20] node[pos=0.55,right,lab]{fail} (err.east);
			
			\draw[arr] (err.west) to[out=180,in=-120,looseness=1.15]
			node[pos=0.55,below,lab]{drop until SOF} (idle.south);
			
			\node[align=center, font=\footnotesize] at ($(err.south)+(0,-9mm)$) {%
				\textbf{Fail-closed:} do not emit ambiguous outputs under corruption.\\
				Emit a flag + resynchronize.
			};
		\end{tikzpicture}%
	}
	\caption{Framing FSM: a tiny amount of control logic prevents silent stream misinterpretation.}
	\label{fig:app-framing-fsm}
\end{figure}

\noindent\textbf{Concrete design question (you must answer).}
What happens if \texttt{len} is larger than the FIFO capacity?
Correct answer in safety-critical pipelines: assert a flag, drain/discard until SOF, and do not ``best-effort'' parse.

\subsubsection{Example D: Timestamping makes latency measurable (mean vs.\ p99)}
\noindent\textbf{Problem.}
You cannot optimize what you cannot measure. In real-time pipelines, \(\text{p99}\) beats mean.

\medskip
\noindent\textbf{On-device truth.}
Use a cycle counter \(T\). Record arrival \(A_t\) when a message becomes visible to the pipeline,
record finish \(F_t\) when the output is committed, and compute \(\Delta_t=F_t-A_t\).

\begin{figure}[t]
	\centering
	\begin{tikzpicture}[
		font=\small,
		box/.style={draw, rounded corners, minimum width=3.4cm, minimum height=1.0cm, align=center},
		arr/.style={-Latex, thick},
		node distance=1.1cm and 1.4cm
		]
		\node[box] (in) {arrival\\stamp \(A_t\)};
		\node[box, right=3.8cm of in] (pipe) {pipeline\\FIFO + compute};
		\node[box, right=3.8cm of pipe] (out) {finish\\stamp \(F_t\)};
		
		\draw[arr] (in) -- (pipe);
		\draw[arr] (pipe) -- (out);
		
		\node[draw, rounded corners, below=1.1cm of pipe, minimum width=10.2cm, minimum height=1.0cm, align=center] (stat) {
			cycle counter \(T\leftarrow T+1\) each clock \quad;\quad
			update histogram of \(\Delta_t=F_t-A_t\) (mean, p99, p999)
		};
		\draw[arr] (pipe.south) -- (stat.north);
	\end{tikzpicture}
	\caption{Timestamping pattern: measure latency in cycles on-device and export compact summaries.}
	\label{fig:app-timestamp}
\end{figure}
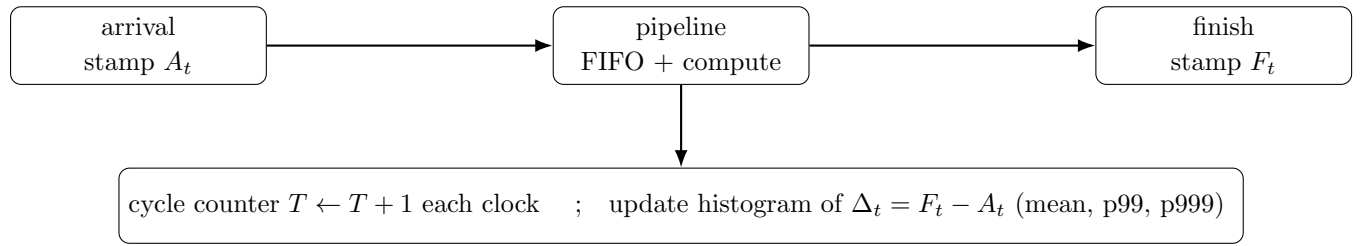

\noindent\textbf{Concrete example (why p99 matters).}
If 99\% of packets take 50 cycles but 1\% take 500 cycles, the mean is misleading.
A control loop with a 100-cycle deadline fails on that 1\% unless you redesign the worst-case path.

\subsection{Putting the patterns together (a minimal ``stream-in, compute, stream-out'' mental model)}
A safe, debuggable pipeline is the composition:
\[
\texttt{valid/ready} \;\to\; \text{FIFO} \;\to\; \text{framing FSM} \;\to\; \text{bounded compute} \;\to\; \text{timestamp + flags}.
\]
The remainder of the appendices will repeatedly reuse these same objects, only changing the payload meaning.

\subsection{What to visualize while reading later appendices}
Keep these pictures in mind:
\begin{itemize}
	\item \textbf{Where can backpressure happen?} (who may stall whom?)
	\item \textbf{Where can state accumulate?} (FIFO occupancy and hidden queues)
	\item \textbf{Where can interpretation drift?} (framing and resync)
	\item \textbf{Where is latency created?} (bounded compute + stalls)
	\item \textbf{Where do we fail closed?} (explicit flags, no silent behavior)
\end{itemize}

\subsection{Exercises (design-only, no projects)}
\begin{exercise}[Handshake sanity]
	Give a 5-cycle example waveform (in words) where \texttt{valid} is asserted for 3 cycles but \texttt{ready} is low for 1 cycle in the middle.
	How many transfers occur, and which cycle(s) perform them?
\end{exercise}

\begin{exercise}[FIFO sizing]
	Assume arrivals can burst at 1 word/cycle for 10 cycles.
	Your consumer drains 1 word every 2 cycles. Without loss, assume the consumer is always ready when it can drain.
	What FIFO depth \(D\) is required to avoid overflow during the burst?
\end{exercise}

\begin{exercise}[Fail-closed framing]
	List three distinct framing failures (bad \texttt{len}, CRC mismatch, timeout) and specify (i) what flag you assert and (ii) how you resynchronize.
\end{exercise}

\begin{exercise}[Latency bookkeeping]
	Describe exactly where you would latch \(A_t\) and \(F_t\) in the pipeline of Figure~\ref{fig:app-timestamp}.
	What event counts as ``arrival'' and what event counts as ``finish''?
\end{exercise}
	
\section{The Classical Interface: Switches as Measurement Outcomes}
\label{sec:classical-interface}

\subsection{Big picture: why switches matter in quantum computing}
Quantum hardware never gives you a ``wavefunction readout.'' What you receive in practice is a
\emph{classical} stream of outcomes: clicks, thresholded voltages, discriminator bits, or FPGA pins.
This section builds the bridge:

\[
\text{analog physics} \;\longrightarrow\; \text{thresholding} \;\longrightarrow\; \textbf{bits}
\;\longrightarrow\; \textbf{control logic} \;\longrightarrow\; \text{actuation}.
\]

\medskip
\noindent\textbf{Engineering lens.}
A measurement result becomes a switch (0/1), and then a \emph{policy}:
``if outcome pattern looks like X, do action Y.''
If you can implement this policy with bounded latency and explicit failure behavior,
you have the core classical loop used in real-time experiments and QEC.

\medskip
\noindent\textbf{Mathematical lens.}
Measurements induce random variables. A classical controller computes a Boolean function of those
random variables, possibly with memory (finite-state machine). That is exactly what LUTs and FSMs are.

\subsection{Physical switches as bits}
\subsubsection{From voltage to bit: thresholding}
A physical ``switch'' is any signal you can reliably map to two levels. Examples:
\begin{itemize}
	\item a push-button switch on the iCEstick (mechanical contact),
	\item a comparator output (analog voltage \(\to\) 0/1),
	\item a photon detector ``click'' event (counted \(\to\) bit),
	\item a readout resonator demodulated IQ \(\to\) discriminated bit.
\end{itemize}
Abstractly, you take a real-valued signal \(x(t)\) and produce
\[
b(t) := \mathbf{1}\{x(t) \ge \theta\}\in\{0,1\},
\]
for some threshold \(\theta\).

\subsubsection{Debounce and metastability (why the bit is not automatically trustworthy)}
Two practical issues appear immediately on FPGA inputs:
\begin{itemize}
	\item \textbf{Switch bounce:} a button can flicker 0/1 for milliseconds.
	\item \textbf{Metastability:} asynchronous inputs near the clock edge can violate setup/hold.
\end{itemize}
Therefore, the safe pattern is:
\[
\text{async input} \;\to\; \text{2-flop synchronizer} \;\to\; \text{debounce filter} \;\to\; \text{clean bit}.
\]

\begin{figure}[t]
	\centering
	\begin{tikzpicture}[
		font=\small,
		box/.style={
			draw, rounded corners,
			minimum height=0.95cm,
			align=center,
			inner xsep=6pt,
			text width=2.8cm 
		},
		arr/.style={-Latex, thick},
		lab/.style={font=\footnotesize, align=center}
		]
		\matrix (m) [matrix of nodes,
		column sep=10mm, row sep=8mm,
		nodes={box}
		]{
			{physical switch\\(async, noisy)} &
			{2-flop\\synchronizer} &
			{debounce\\(N-cycle stable)} &
			{clean bit\\$b\in\{0,1\}$} \\
		};
		
		\draw[arr] (m-1-1.east) -- (m-1-2.west);
		\draw[arr] (m-1-2.east) -- (m-1-3.west);
		\draw[arr] (m-1-3.east) -- (m-1-4.west);
		
		\node[lab] at ($(m-1-2.south)+(0,-6mm)$) {metastability\\guard};
		\node[lab] at ($(m-1-3.south)+(0,-6mm)$) {bounce\\suppression};
	\end{tikzpicture}
	\caption{Making a trustworthy bit: synchronize asynchronous inputs, then debounce.}
	\label{fig:switch-to-bit}
\end{figure}
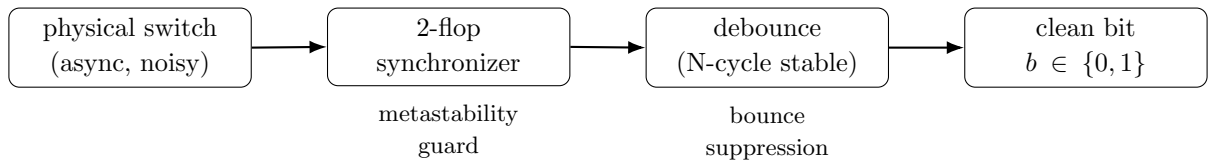

\subsection{Boolean algebra as quantum control logic}
\subsubsection{Bits as events; Boolean functions as policies}
Once you have clean bits \(b_1,\dots,b_m\), a classical policy is a Boolean map
\[
u = f(b_1,\dots,b_m)\in\{0,1\}^k,
\]
where \(u\) are the control outputs (enable lines, mux selects, trigger pulses, etc.).
In QEC language: \(b_i\) are syndrome/detection-event bits; \(u\) encodes a correction decision
(or a Pauli-frame update).

\subsubsection{Truth tables are executable specifications}
The simplest policy is a truth table.
For \(m\) input bits, there are \(2^m\) input patterns.
A truth table is a complete, unambiguous spec of what to do for each pattern.

\medskip
\noindent\textbf{Example (2-bit toy policy).}
Let \(b_0,b_1\) be two measurement outcomes and choose a control output \(u\):
\[
u = b_0 \oplus b_1.
\]
This is a minimal model of ``parity checks'' you see everywhere in stabilizer measurement logic.

\subsection{Logic gates, LUTs, and control as truth tables}
\subsubsection{Gates vs.\ LUTs}
You can implement \(f\) either as:
\begin{itemize}
	\item a network of logic gates (AND/OR/XOR/NOT), or
	\item a LUT (lookup table): store outputs indexed by input bits.
\end{itemize}

\medskip
\noindent\textbf{LUT viewpoint (control as a table).}
A LUT is literally:
\[
\texttt{out} \leftarrow \texttt{table}[\texttt{in}],
\]
so it matches the ``policy spec'' style.

\subsubsection{Finite-state control (memory matters)}
Real control loops often need \emph{state}:
\[
(\text{state}_{t+1}, u_t) = F(\text{state}_t, b_t).
\]
This is a finite-state machine (FSM).
It appears whenever you:
\begin{itemize}
	\item integrate over time (e.g.\ ``require 3 consecutive votes''),
	\item parse framed streams (SOF/len/payload),
	\item enforce fail-closed behavior (error state, resync state),
	\item implement bounded-pass decoders (phase counters).
\end{itemize}

\begin{figure}[t]
	\centering
	\begin{tikzpicture}[
		font=\small,
		box/.style={draw, rounded corners, minimum width=3.0cm, minimum height=0.95cm, align=center},
		arr/.style={-Latex, thick},
		node distance=1.1cm and 1.4cm
		]
		\node[box] (in) {measurement bits\\\(b_t\)};
		\node[box, right=3.8cm of in] (fsm) {FSM / LUT policy\\\(F(\text{state},b)\)};
		\node[box, right=3.8cm of fsm] (out) {control outputs\\\(u_t\)};
		
		\draw[arr] (in) -- (fsm);
		\draw[arr] (fsm) -- (out);
		\draw[arr] (fsm.south) .. controls +(0,-1.2) and +(0,-1.2) .. node[below]{\footnotesize state register} (fsm.south);
		
	\end{tikzpicture}
	\caption{Control as computation: measurement bits go into a policy block (LUT/FSM), producing control outputs.}
	\label{fig:policy-fsm}
\end{figure}
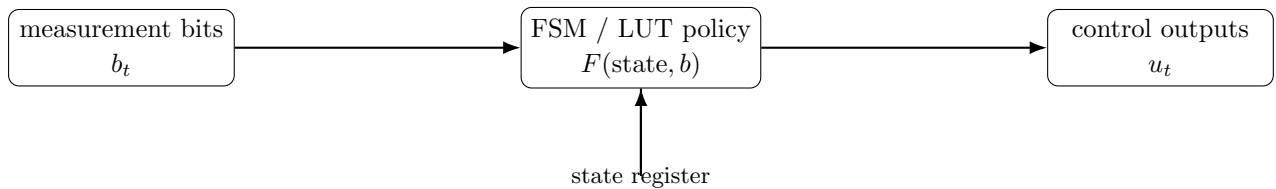

\subsection{Lab: classical Pauli controller on iCEstick}
\noindent\textbf{Goal.}
Implement the simplest ``Pauli controller'' driven by two switch inputs:
\[
(b_0,b_1)\ \mapsto\ (u_X,u_Z),
\]
where \(u_X,u_Z\in\{0,1\}\) control two LEDs representing whether to apply an \(X\) or \(Z\) correction.

\medskip
\noindent\textbf{Design spec (truth-table style).}
We interpret:
\[
(b_0,b_1) =
\begin{cases}
	(0,0) & \text{no correction},\\
	(1,0) & X\ \text{correction},\\
	(0,1) & Z\ \text{correction},\\
	(1,1) & XZ\ \text{(both)}.
\end{cases}
\]
So define:
\[
u_X := b_0,\qquad u_Z := b_1.
\]
This is intentionally trivial: the purpose is to build the full input hygiene and timing discipline
(synchronize + debounce + combinational policy + visible outputs).

\medskip
\noindent\textbf{I/O mapping example (edit to your board constraints).}
\begin{itemize}
	\item inputs: \texttt{SW0}, \texttt{SW1} (or buttons),
	\item outputs: \texttt{LED0} = \(u_X\), \texttt{LED1} = \(u_Z\),
	\item clock: iCEstick 12\,MHz oscillator.
\end{itemize}

\medskip
\noindent\textbf{Core RTL skeleton (portable).}
\begin{verbatim}
	module pauli_controller (
	input  wire clk,
	input  wire sw0_async,
	input  wire sw1_async,
	output wire ledX,
	output wire ledZ
	);
	// 1) sync (2-flop)
	reg sw0_q1, sw0_q2, sw1_q1, sw1_q2;
	always @(posedge clk) begin
	sw0_q1 <= sw0_async; sw0_q2 <= sw0_q1;
	sw1_q1 <= sw1_async; sw1_q2 <= sw1_q1;
	end
	
	// 2) debounce (simple: require N stable samples)
	// For appendix pedagogy: treat sw*_q2 as "clean enough" or add a counter-based filter.
	
	// 3) policy (truth table)
	assign ledX = sw0_q2; // u_X
	assign ledZ = sw1_q2; // u_Z
	endmodule
\end{verbatim}

\medskip
\noindent\textbf{Optional extension (still not a project).}
Add a one-shot pulse output:
if \((b_0,b_1)\neq(0,0)\) then emit a 1-cycle \texttt{apply} pulse.
This forces you to introduce a tiny FSM (edge detect / pulse stretch).

\subsection{Exercises (with detailed solutions)}
\begin{exercise}[Why synchronizers are necessary]
	Explain why an asynchronous switch input can cause metastability if sampled directly by a flip-flop.
	What does a 2-flop synchronizer do, and what does it \emph{not} guarantee?
\end{exercise}
\noindent\textbf{Solution.}
If an asynchronous signal changes close to the sampling clock edge, the input violates setup/hold time and the flip-flop may enter a metastable analog state, taking an unbounded (though typically short) time to resolve to 0 or 1. A 2-flop synchronizer reduces the probability that metastability propagates into downstream logic: the first flop may go metastable, but by the next clock edge it is very likely to have resolved before the second flop samples it. It does not guarantee correctness of the sampled value at the exact transition instant (you may sample the ``old'' or ``new'' level), and it does not remove switch bounce; it only mitigates metastability propagation.

\begin{exercise}[Design a debounce filter]
	Design a debounce rule: require an input to be stable for \(N\) consecutive clock cycles before updating the clean bit.
	Give pseudocode for the counter logic.
\end{exercise}
\noindent\textbf{Solution.}
Maintain (i) a candidate value \(c\), (ii) a counter \(k\), and (iii) an output \(b\).
On each clock:
\begin{itemize}
	\item If the sampled input equals \(c\), increment \(k\) up to \(N\).
	\item If it differs, set \(c\leftarrow\) sampled input and reset \(k\leftarrow 0\).
	\item If \(k=N\), update \(b\leftarrow c\).
\end{itemize}
This ensures \(b\) changes only after \(N\) stable samples. In RTL this is a small counter + compare.

\begin{exercise}[Truth-table vs.\ logic network]
	Let \(u = (b_0\wedge \neg b_1)\ \vee\ (b_1\wedge \neg b_0)\).
	(i) Identify the Boolean function. (ii) Write its truth table.
\end{exercise}
\noindent\textbf{Solution.}
(i) The expression is XOR: \(u=b_0\oplus b_1\).
(ii) Truth table:
\[
\begin{array}{c c | c}
	b_0 & b_1 & u\\ \hline
	0 & 0 & 0\\
	0 & 1 & 1\\
	1 & 0 & 1\\
	1 & 1 & 0
\end{array}
\]

\begin{exercise}[FSM as a noise-robust policy]
	You receive a bit \(b_t\) each cycle that may glitch. Design an FSM that outputs \(u_t=1\) only if \(b_t=1\) for three consecutive cycles.
	Give the state diagram in words.
\end{exercise}
\noindent\textbf{Solution.}
Use four states representing the current run length of consecutive ones: \(S0,S1,S2,S3\).
Initialize in \(S0\).
On input \(b_t=1\): \(S0\to S1\to S2\to S3\) and stay in \(S3\).
On input \(b_t=0\): go to \(S0\) from any state.
Output \(u_t=1\) iff in \(S3\).
This implements a 3-cycle majority-by-stability rule.

\subsection{Integration notes}
\noindent\textbf{Where this appears later.}
Everything in this section reappears in more advanced settings:
\begin{itemize}
	\item syndrome bits \(\leftrightarrow\) switch bits,
	\item decoding decision \(\leftrightarrow\) control output,
	\item framing + FIFO \(\leftrightarrow\) safe streaming,
	\item FSM phases \(\leftrightarrow\) bounded-pass algorithms.
\end{itemize}

\medskip
\noindent\textbf{Two practical warnings.}
\begin{itemize}
	\item Never mix long text inside TikZ nodes without \texttt{align=center} (or explicit \texttt{text width});
	otherwise you will trigger ``Something's wrong--perhaps a missing \textbackslash item''-style errors from accidental list parsing.
	\item If you want a \texttt{remark} environment, define it in the preamble via \texttt{amsthm}:
	\begin{verbatim}
		\usepackage{amsthm}
		\newtheorem{remark}{Remark}[section]
	\end{verbatim}
\end{itemize}
	
\section{Quantum Arithmetic: Discretizing the Phase}
\label{sec:phase-arith}

\subsection{Big picture: from smooth phase to clocked integer arithmetic}
A large fraction of ``quantum control'' and ``quantum algorithms'' reduces to one recurring task:
\emph{track and update phases} under a clock.

\medskip
\noindent\textbf{Physics side.}
Quantum evolution accumulates phase continuously:
\[
\ket{\psi(t)} \mapsto e^{-i\omega t}\ket{\psi(0)}.
\]
Measurement and control, however, happen in discrete rounds (shots, cycles, pulses). So we need
a discrete-time representation of phase that fits digital logic.

\medskip
\noindent\textbf{Hardware side.}
Digital systems represent phase by integers and update it once per clock:
\[
\phi_{t+1} = \phi_t + \Delta \phi \quad (\text{mod } 2\pi).
\]
On an FPGA, this becomes:
\[
\Phi_{t+1} = \Phi_t + K \quad (\text{mod } 2^N),
\]
where \(\Phi_t\) is an \(N\)-bit register and \(K\) is an \(N\)-bit increment.

\medskip
\noindent\textbf{Why this belongs in a quantum chapter.}
The same arithmetic underlies:
\begin{itemize}
	\item phase kickback and eigenphase readout (QPE),
	\item controlled rotations in Shor-type circuits,
	\item classical feed-forward phase corrections (teleportation, stabilizers),
	\item digital synthesis of control waveforms (DDS),
	\item phase-tracking estimators used in metrology and QFI-driven optimization.
\end{itemize}

\subsection{Number systems for quantum phases (fixed-point on a circle)}
\subsubsection{Two equivalent ways to encode a phase}
A phase is an angle on the circle \(S^1\). Digitally, we pick a resolution \(N\) and encode:
\[
\Phi \in \{0,1,\dots,2^N-1\}
\qquad\leftrightarrow\qquad
\phi(\Phi) := 2\pi \frac{\Phi}{2^N} \in [0,2\pi).
\]
Then:
\[
\Phi_1 + \Phi_2 \pmod{2^N}
\quad \leftrightarrow \quad
\phi_1 + \phi_2 \pmod{2\pi}.
\]

\subsubsection{Fixed-point intuition}
Think of \(\Phi\) as a fixed-point number with binary point \emph{wrapped on a circle}.
The least significant bit corresponds to the smallest phase step:
\[
\Delta \phi_{\min} = \frac{2\pi}{2^N}.
\]
So \(N\) sets your \emph{phase resolution}. The tradeoff is standard:
\[
\text{larger }N \Rightarrow \text{finer phase} \Rightarrow \text{more bits/area/timing}.
\]

\subsubsection{Signed vs.\ unsigned (when negative increments are useful)}
Unsigned modular arithmetic is simplest:
\[
\Phi_{t+1} = \Phi_t + K \ (\mathrm{mod}\ 2^N).
\]
To represent negative increments, interpret the same \(N\)-bit word in two's complement.
Then \(K=-1\) means ``subtract one LSB'' modulo \(2^N\).

\subsection{Visualization I: continuous phase vs.\ quantized phase}
\begin{figure}[t]
	\centering
	\begin{tikzpicture}[
		font=\small,
		arr/.style={-Latex, thick},
		dot/.style={circle, fill, inner sep=1.2pt},
		qdot/.style={circle, fill=black, inner sep=1.2pt},
		lab/.style={font=\footnotesize}
		]
		\begin{scope}
			\node at (0,2.35) {\small continuous phase \(\phi\in[0,2\pi)\)};
			\draw (0,0) circle (1.55cm);
			\draw[arr] (0,0) -- (1.25,0.62);
			\node[lab] at (1.05,0.95) {\(\phi\)};
			\node[lab] at (0,-2.05) {any angle allowed};
		\end{scope}
		
		\begin{scope}[xshift=6.2cm]
			\node at (0,2.35) {\small quantized phase \(\Phi\in\{0,\dots,2^N\!-\!1\}\)};
			\draw (0,0) circle (1.55cm);
			
			\def\m{16}
			\foreach \k in {0,...,15} {
				\path ({1.55*cos(360*\k/\m)},{1.55*sin(360*\k/\m)}) node[qdot] {};
			}
			
			\def\kh{3}
			\path ({1.55*cos(360*\kh/\m)},{1.55*sin(360*\kh/\m)}) node[dot] {};
			\draw[arr] (0,0) -- ({1.15*cos(360*\kh/\m)},{1.15*sin(360*\kh/\m)});
			\node[lab] at (1.2,-2.05) {step size \(\Delta\phi_{\min}=2\pi/2^N\)};
		\end{scope}
	\end{tikzpicture}
	\caption{Continuous phase is a circle. Quantized phase replaces the circle by \(2^N\) discrete points and uses modular integer arithmetic.}
	\label{fig:phase-quantization}
\end{figure}
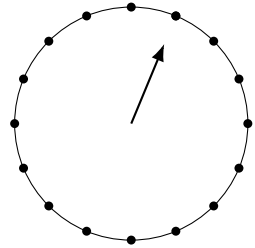

\subsection{Adders as phase composition engines}
\subsubsection{Core operation: modular addition}
Phase composition is addition mod \(2\pi\). In \(N\)-bit phase words:
\[
\Phi_{\text{out}} = (\Phi_a + \Phi_b)\bmod 2^N.
\]
So the \textbf{adder} is a phase engine. Everything else is ``where do \(\Phi_a,\Phi_b\) come from?''

\subsubsection{Phase accumulator (the one block to remember)}
A phase accumulator updates:
\[
\Phi_{t+1} = \Phi_t + K \pmod{2^N}.
\]
Interpretation:
\begin{itemize}
	\item \(K\) sets the \emph{digital frequency}.
	\item \(\Phi_t\) is the running phase.
	\item wrap-around is natural (it's the circle).
\end{itemize}
If the clock is \(f_{\text{clk}}\), then the output frequency is:
\[
f_{\text{out}} = \frac{K}{2^N} f_{\text{clk}}.
\]
(This is the classical DDS fact; we use it here as phase arithmetic intuition.)

\subsubsection{Carry is not a bug; it's the wrap}
In modular arithmetic, overflow carry means ``we crossed \(2\pi\)''.
You can keep the carry as an optional event flag:
\[
c_t = \mathbf{1}\{\Phi_t + K \ge 2^N\},
\]
useful for timing markers.

\subsection{Multiplexers and measurement-based feedback}
\subsubsection{Feedback is just selecting an increment}
Measurement-based phase updates (e.g.\ feed-forward corrections) often look like:
\[
K =
\begin{cases}
	K_0, & b=0,\\
	K_1, & b=1,
\end{cases}
\qquad
\Rightarrow\qquad
K = \mathrm{MUX}(b;\,K_0,K_1).
\]
This is exactly a multiplexer controlled by a bit \(b\) (a measurement outcome, a syndrome bit,
or a decision bit from a controller).

\subsubsection{Typical pattern: correction-by-lookup}
If a measurement result \(b\) suggests a phase correction \(\delta_b\),
store \(\Delta\Phi_b \approx 2^N \delta_b/(2\pi)\) and select it by MUX.

\begin{figure}[t]
	\centering
	\begin{tikzpicture}[
		font=\small,
		box/.style={draw, rounded corners, minimum width=3.2cm, minimum height=0.95cm, align=center},
		arr/.style={-Latex, thick},
		node distance=1.0cm and 1.4cm
		]
		\node[box] (meas) {measurement bit\\\(b\)};
		\node[box, right=3.8cm of meas] (mux) {MUX\\select \(K\)};
		\node[box, right=3.8cm of mux] (acc) {phase accumulator\\\(\Phi\leftarrow \Phi+K\)};
		\node[box, below=1.2cm of mux, minimum width=7.2cm] (tbl) {table: \(b\mapsto K_b\) (phase corrections)};
		\draw[arr] (meas) -- (mux);
		\draw[arr] (mux) -- (acc);
		\draw[arr] (tbl.north) -- (mux.south);
	\end{tikzpicture}
	\caption{Measurement-based feedback: a bit selects a phase update (MUX), then an adder applies it (accumulator).}
	\label{fig:mux-feedback}
\end{figure}
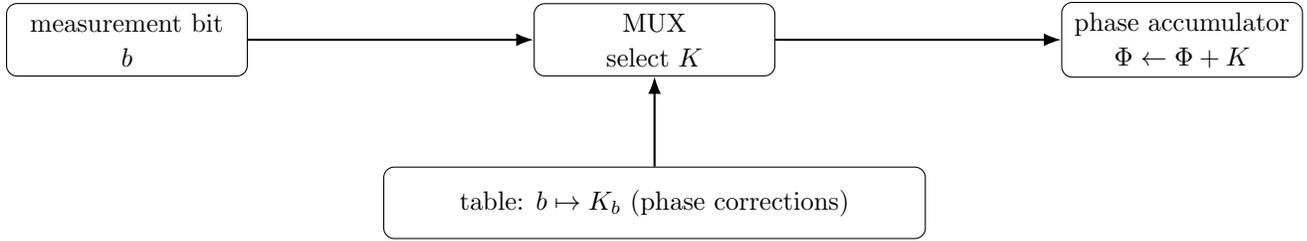

\subsection{Lab: iCEstick phase accumulator (DDS prelude) with LEDs}
\noindent\textbf{Goal.}
Build an \(N\)-bit phase accumulator on iCEstick and visualize phase progression on LEDs.
This is not ``DDS audio'' yet; it is the \emph{phase arithmetic core}.

\medskip
\noindent\textbf{Minimal specification.}
\begin{itemize}
	\item Choose \(N\in\{8,12,16\}\) (start with \(N=8\)).
	\item Update \(\Phi\) every clock: \(\Phi\leftarrow \Phi+K\).
	\item Display the top bits \(\Phi[N-1:N-5]\) on 5 LEDs as a ``rotating phase bar''.
\end{itemize}

\medskip
\noindent\textbf{Why top bits?}
Top bits change slowly and are visible; low bits toggle too fast for human eyes.

\medskip
\noindent\textbf{Core RTL skeleton.}
\begin{verbatim}
	module phase_accum #(
	parameter integer N = 8
	)(
	input  wire             clk,
	input  wire             rst,
	input  wire [N-1:0]     K,
	output wire [4:0]       leds
	);
	reg [N-1:0] phi;
	always @(posedge clk) begin
	if (rst) phi <= {N{1'b0}};
	else     phi <= phi + K;   // wraps naturally mod 2^N
	end
	
	// show MSBs
	assign leds = phi[N-1 -: 5];
	endmodule
\end{verbatim}

\medskip
\noindent\textbf{Concrete LED interpretation.}
If LEDs show MSBs, you are watching a quantized pointer move around the circle.
Different \(K\) values change the rotation speed.

\subsection{Visualization II: DDS prelude (phase \(\to\) waveform)}
A DDS adds one more idea: map phase to amplitude by a function, typically sine.
Digitally, that is a lookup table:
\[
A_t = \sin(\phi(\Phi_t)) \approx \texttt{SIN\_LUT}[\Phi_t[\text{top bits}]].
\]

\begin{figure}[t]
	\centering
	\begin{tikzpicture}[
		font=\small,
		box/.style={draw, rounded corners, minimum width=3.3cm, minimum height=0.95cm, align=center},
		arr/.style={-Latex, thick},
		node distance=1.0cm and 1.3cm
		]
		\node[box] (acc) {phase accumulator\\\(\Phi_{t+1}=\Phi_t+K\)};
		\node[box, right=3.9cm of acc] (lut) {LUT\\phase \(\to\) amplitude};
		\node[box, right=3.9cm of lut] (dac) {output\\(LED/PWM/DAC)};
		
		\draw[arr] (acc) -- node[above]{\footnotesize phase bits} (lut);
		\draw[arr] (lut) -- node[above]{\footnotesize amplitude} (dac);
		
		\node[align=left, font=\footnotesize] at ($(lut.south)+(0,-0.8)$) {
			simplest DDS: \(\Phi\) rotates on a circle; LUT projects it to a waveform
		};
	\end{tikzpicture}
	\caption{DDS prelude: phase arithmetic (adder) plus a phase-to-amplitude map (LUT) gives a waveform generator.}
	\label{fig:dds-prelude}
\end{figure}

\subsection{Worked examples (concrete numbers + interpretation)}
Assume \(f_{\text{clk}}=12\,\mathrm{MHz}\) (common iCEstick clock).

\subsubsection{Example 1: \(N=8\), \(K=1\)}
\[
f_{\text{out}} = \frac{1}{256}\cdot 12\text{ MHz} \approx 46.875\text{ kHz}.
\]
Interpretation: the phase completes one full turn every \(256\) clocks.

\subsubsection{Example 2: \(N=8\), \(K=64\)}
\[
f_{\text{out}} = \frac{64}{256}\cdot 12\text{ MHz} = 3\text{ MHz}.
\]
Interpretation: the phase jumps by a quarter-turn each clock. MSBs will appear very fast.

\subsubsection{Example 3: \(N=16\), \(K=1000\)}
\[
f_{\text{out}} = \frac{1000}{65536}\cdot 12\text{ MHz}
\approx 183.105\text{ kHz}.
\]
Interpretation: increasing \(N\) gives finer frequency granularity: one LSB is
\[
\Delta f = \frac{f_{\text{clk}}}{2^N}.
\]
For \(N=8\), \(\Delta f\approx 46.875\) kHz; for \(N=16\), \(\Delta f\approx 183.105\) Hz.

\subsubsection{Example 4: measurement-based correction via MUX}
Let \(b\in\{0,1\}\) and choose
\[
K = \begin{cases}
	K_0=1000, & b=0,\\
	K_1=1200, & b=1.
\end{cases}
\]
Then the clocked system switches frequencies immediately based on \(b\).
Interpretation: this is the digital skeleton of feed-forward phase control.

\subsection{Integration notes: QFI (Ch.~\ref{sec:qfim}) and variational optimization (Ch.~\ref{sec:qng})}
\subsubsection{Why discretization matters for gradients}
If a variational parameter is an angle \(\theta\), hardware often represents \(\theta\) in finite resolution.
Then the parameter update is quantized:
\[
\theta \leftarrow \theta + \eta \cdot \Delta\theta,
\qquad
\Delta\theta \in \left\{\frac{2\pi}{2^N}\mathbb{Z}\right\}.
\]
This introduces a floor on achievable step size and affects convergence.

\subsubsection{QFI viewpoint (conceptual)}
QFI-based methods treat parameters as living on a geometry.
Digitally, you approximate that geometry on a grid.
The key engineering takeaway:
\begin{itemize}
	\item choose \(N\) large enough that quantization noise does not dominate your update,
	\item measure and log p99 timing if the update must be real-time,
	\item prefer adders + LUTs + bounded FSMs when you need predictable latency.
\end{itemize}

\medskip
\noindent\textbf{Practical bridge statement.}
The same ``phase word'' \(\Phi\) can represent:
\begin{itemize}
	\item a \emph{control phase} (hardware waveform timing),
	\item a \emph{gate parameter} (variational angle),
	\item an \emph{estimated eigenphase} (QPE output),
\end{itemize}
so phase arithmetic is a shared substrate across algorithms and hardware.
	
\section{Sequencing the Experiment: FSM as the Circuit (iCEstick Edition)}
\label{sec:appendix-fsm}

\subsection{Big picture: the FPGA turns ``time'' into a deterministic program}
In many quantum experiments, the ``quantum circuit'' is only half of the story.
The other half is the \emph{classical sequence} that decides:
\[
\text{when to prepare} \;\to\; \text{when to drive} \;\to\; \text{when to wait} \;\to\;
\text{when to measure} \;\to\; \text{what to do next}.
\]
On an FPGA, this sequence is implemented as a \textbf{finite-state machine (FSM)} plus
a small set of counters and registers.

\medskip
\noindent\textbf{Key engineering viewpoint.}
An FSM is a \emph{time-ordered circuit}:
\[
\text{combinational logic (decide next)}\;+\;\text{registers (remember state)}\;+\;\text{clock (advance time)}.
\]
So ``control software'' becomes a circuit whose behavior is deterministic, cycle-by-cycle, and observable.

\medskip
\noindent\textbf{Why this matters in quantum control.}
If you want repeatable experiments (shots), bounded latency (feedback), and exact timing (phases),
you want a clocked program, not a best-effort CPU loop.

\subsection{Registers and classical memory of quantum events (iCEstick mapping)}
\subsubsection{What counts as ``memory'' in a quantum experiment}
A quantum experiment produces \emph{classical events}:
preparation done, pulse launched, measurement arrived, threshold crossed, etc.
These are stored as \textbf{registers}.

\medskip
\noindent\textbf{Minimal register set (template).}
\begin{itemize}
	\item \texttt{state} : current FSM state (encoded as small integer).
	\item \texttt{t} : cycle counter within a shot (fine time).
	\item \texttt{shot\_id} : which repetition we are on.
	\item \texttt{meas\_bit} : last measurement bit (from a pin/UART/FIFO).
	\item \texttt{acc} : statistics accumulator (counts, histogram bins).
	\item \texttt{flags} : overflow/timeout/error/abort indicators.
\end{itemize}

\subsubsection{iCEstick I/O mapping mindset (abstract)}
Do \emph{not} tie the appendix to one physical wiring.
Instead, treat the iCEstick as providing:
\begin{itemize}
	\item \textbf{inputs}: switches / pins / UART bytes / FIFO-valid events,
	\item \textbf{outputs}: LEDs / pins / UART reports / debug strobes.
\end{itemize}
Then you can swap ``measurement source'' without rewriting the control program.

\subsubsection{A practical convention: one-bit events + tagged payloads}
A useful split is:
\begin{itemize}
	\item fast 1-bit events (e.g.\ \texttt{meas\_ready}, \texttt{threshold\_hit}),
	\item slower multi-bit payload (e.g.\ \texttt{meas\_value[W-1:0]}).
\end{itemize}
This avoids the common failure mode: wide buses that are never stable when you sample them.

\subsection{Finite state machines as time-ordered circuits}
\subsubsection{Canonical FSM skeleton}
An FSM consists of:
\begin{itemize}
	\item \textbf{state register} updated on clock,
	\item \textbf{next-state logic} computed combinationally,
	\item \textbf{outputs} that depend on (state, inputs) (Moore or Mealy style).
\end{itemize}

\medskip
\noindent\textbf{Two-process style (recommended).}
\begin{verbatim}
	always @(posedge clk) begin
	if (rst) state <= S_IDLE;
	else     state <= state_next;
	end
	
	always @* begin
	state_next = state;
	// defaults for outputs
	case (state)
	S_IDLE:  if (start) state_next = S_PREP;
	...
	endcase
	end
\end{verbatim}

\subsubsection{Moore vs.\ Mealy in control loops}
\begin{itemize}
	\item \textbf{Moore:} outputs depend only on state (clean timing, fewer glitches).
	\item \textbf{Mealy:} outputs depend on state and input (faster reaction, but must be careful).
\end{itemize}
For quantum timing, Moore often wins. For real-time feedback, Mealy is sometimes necessary.

\subsubsection{Visualization: FSM as a circuit}
\begin{figure}[t]
	\centering
	\begin{tikzpicture}[
		font=\small,
		box/.style={draw, rounded corners, minimum width=3.6cm, minimum height=1.0cm, align=center},
		arr/.style={-Latex, thick},
		node distance=1.2cm and 1.6cm
		]
		\node[box] (comb) {combinational logic\\\footnotesize (next state + outputs)};
		\node[box, right=4.6cm of comb] (reg) {registers\\\footnotesize (state + memory)};
		\node[box, below=1.3cm of comb, minimum width=3.6cm] (in) {inputs\\\footnotesize (meas, start, flags)};
		\node[box, below=1.3cm of reg, minimum width=3.6cm] (out) {outputs\\\footnotesize (pulse, gate, report)};
		
		\draw[arr] (reg.west) -- (comb.east);
		\draw[arr] (comb.east) -- (reg.west);
		\draw[arr] (in.north) -- (comb.south);
		\draw[arr] (comb.south east) -- (out.north west);
		
		\node[align=center, font=\footnotesize] at ($(reg.north)+(0,0.65)$) {clock ticks advance time};
	\end{tikzpicture}
	\caption{FSM viewed as a circuit: registers store the ``present''; combinational logic computes the ``next''.}
	\label{fig:fsm-circuit}
\end{figure}
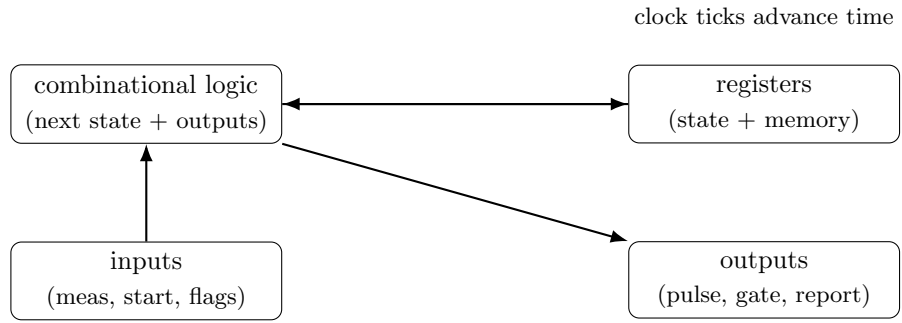

\subsection{Counters and repetition control (shots and online statistics)}
\subsubsection{Two time scales: within-shot and across-shots}
Experiments naturally have:
\begin{itemize}
	\item \textbf{fine time} (cycle count within a shot): pulse widths, waits, alignments,
	\item \textbf{coarse time} (shot index): repetitions for statistics.
\end{itemize}
So you typically use two counters:
\[
t \in \{0,\dots,T_{\max}\},\qquad \texttt{shot\_id}\in\{0,\dots,S-1\}.
\]

\subsubsection{Online statistics template}
A minimal on-device estimator (for one-bit outcomes) is:
\[
\texttt{ones} \leftarrow \texttt{ones} + \texttt{meas\_bit}, \qquad
\texttt{shots} \leftarrow \texttt{shots} + 1.
\]
Then:
\[
\widehat{p} = \frac{\texttt{ones}}{\texttt{shots}}.
\]
This is enough to demonstrate:
\begin{itemize}
	\item calibration curves (vary a parameter, estimate probability),
	\item stability monitoring (drift detection),
	\item repeatability (same parameter $\Rightarrow$ same statistics).
\end{itemize}

\subsubsection{Histogram estimator (multi-bin)}
If a measurement gives a small integer \(y\in\{0,\dots,B-1\}\), maintain:
\[
\texttt{hist}[y] \leftarrow \texttt{hist}[y] + 1.
\]
This becomes a compact ``experiment log'' you can export later.

\subsubsection{A concrete safety rule}
Always bound time in each state. Use timeouts:
\[
\texttt{if }t>\tau_{\max}\texttt{ then set error flag and go to ERROR.}
\]
This prevents dead states when measurement never arrives.

\subsection{Lab: Ramsey-style experiment controller (fully observable)}
\noindent\textbf{Goal.}
Implement a Ramsey-style \emph{sequence controller} that you can demonstrate with only:
(1) a clock, (2) a one-bit measurement input, (3) LEDs (and optional UART).

\medskip
\noindent\textbf{Abstract Ramsey timing (control-only viewpoint).}
A minimal Ramsey experiment follows:
\[
\text{PREP} \to \text{PULSE}_1 \to \text{WAIT}(\tau) \to \text{PULSE}_2 \to \text{MEASURE} \to \text{ACCUMULATE}.
\]
We do not assume an actual qubit is connected; instead we treat measurement as an input bit.
The point is: \emph{sequence, timing, and observability}.

\subsubsection{FSM states}
\begin{itemize}
	\item \texttt{S\_IDLE}: wait for \texttt{start}.
	\item \texttt{S\_PREP}: clear registers, reset counters.
	\item \texttt{S\_P1}: assert output \texttt{pulse\_en} for \(T_1\) cycles.
	\item \texttt{S\_WAIT}: wait \(\tau\) cycles (phase accumulation happens here conceptually).
	\item \texttt{S\_P2}: second pulse for \(T_2\) cycles.
	\item \texttt{S\_MEAS}: sample \texttt{meas\_bit} when \texttt{meas\_valid}.
	\item \texttt{S\_ACC}: update counters, decide next shot or done.
	\item \texttt{S\_DONE}: signal completion.
	\item \texttt{S\_ERR}: fail-closed.
\end{itemize}

\subsubsection{Fully observable outputs (LED debugging contract)}
A robust debug scheme is:
\begin{itemize}
	\item LEDs show \texttt{state} (e.g.\ top 3--5 bits),
	\item one LED pulses when a shot ends,
	\item one LED indicates \texttt{meas\_valid} handshake,
	\item one LED indicates \texttt{ERR} or timeout.
\end{itemize}
This makes the controller demonstrable without oscilloscopes.

\subsubsection{Visualization: Ramsey sequencing timeline}
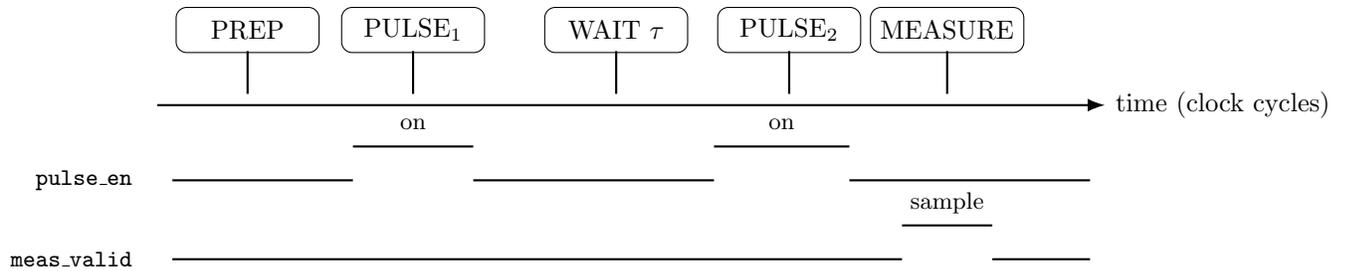
\begin{figure}[t]
	\centering
	\begin{tikzpicture}[
		font=\small,
		arr/.style={-Latex, thick},
		ln/.style={thick},
		lab/.style={font=\footnotesize},
		box/.style={draw, rounded corners, minimum width=1.9cm, minimum height=0.6cm, align=center}
		]
		\draw[arr] (0,0) -- (12.6,0) node[right] {\small time (clock cycles)};
		\node[box] (prep) at (1.2,1.0) {PREP};
		\node[box] (p1)   at (3.4,1.0) {PULSE$_1$};
		\node[box] (wait) at (6.1,1.0) {WAIT \(\tau\)};
		\node[box] (p2)   at (8.4,1.0) {PULSE$_2$};
		\node[box] (meas) at (10.5,1.0) {MEASURE};
		
		\foreach \x/\n in {1.2/prep,3.4/p1,6.1/wait,8.4/p2,10.5/meas}{
			\draw[ln] (\x,0.15) -- (\x,0.72);
		}
		
		\node[lab, anchor=east] at (-0.2,-1.0) {\texttt{pulse\_en}};
		\draw[ln] (0.2,-1.0) -- (2.6,-1.0);
		\draw[ln] (2.6,-0.55) -- (4.2,-0.55);
		\draw[ln] (4.2,-1.0) -- (7.4,-1.0);
		\draw[ln] (7.4,-0.55) -- (9.2,-0.55);
		\draw[ln] (9.2,-1.0) -- (12.4,-1.0);
		
		\node[lab] at (3.4,-0.25) {on};
		\node[lab] at (8.3,-0.25) {on};
		
		\node[lab, anchor=east] at (-0.2,-2.05) {\texttt{meas\_valid}};
		\draw[ln] (0.2,-2.05) -- (9.9,-2.05);
		\draw[ln] (9.9,-1.6) -- (11.1,-1.6);
		\draw[ln] (11.1,-2.05) -- (12.4,-2.05);
		
		\node[lab] at (10.5,-1.3) {sample};
	\end{tikzpicture}
	\caption{A Ramsey-style controller is a timed sequence: assert pulses in two windows separated by a programmable wait \(\tau\), then measure and accumulate statistics.}
	\label{fig:ramsey-timeline}
\end{figure}

\subsection{Worked examples (concrete scenarios you can demonstrate)}
These examples assume you can toggle parameters (switches or UART) and observe LEDs.

\subsubsection{Scenario A: timing demonstration (no real measurement)}
\textbf{Setup.} Force \texttt{meas\_bit=0} and \texttt{meas\_valid=1} at MEASURE state.

\textbf{What you demonstrate.}
\begin{itemize}
	\item The FSM progresses through states deterministically.
	\item Pulse windows have correct lengths \(T_1,T_2\).
	\item WAIT lasts exactly \(\tau\) cycles.
	\item Shot counter increments once per sequence.
\end{itemize}

\textbf{Observable outputs.}
\begin{itemize}
	\item state-coded LEDs cycle through PREP/P1/WAIT/P2/MEAS/ACC,
	\item a ``shot-done'' LED blips each time ACC completes.
\end{itemize}

\subsubsection{Scenario B: measurement-gated progress (valid/ready discipline)}
\textbf{Setup.} Hold \texttt{meas\_valid=0} for a while, then assert it.

\textbf{What you demonstrate.}
\begin{itemize}
	\item MEASURE state waits safely for \texttt{meas\_valid}.
	\item A timeout triggers ERR (fail-closed) if measurement never arrives.
\end{itemize}

\textbf{Engineering lesson.}
This is the same pattern as streaming interfaces: do nothing until the contract condition holds.

\subsubsection{Scenario C: online statistics curve (toy fringe)}
\textbf{Setup.} Let \texttt{meas\_bit} depend on \(\tau\) by a toy rule:
\[
\Pr(\texttt{meas\_bit}=1) \approx \frac{1+\cos(2\pi \tau/T)}{2}
\]
implemented by a small LUT indexed by \(\tau\).

\textbf{What you demonstrate.}
\begin{itemize}
	\item sweeping \(\tau\) changes the estimated \(\widehat{p}(\tau)\),
	\item the accumulator converges as shots increase,
	\item you can export \((\tau,\widehat{p})\) as a calibration-like dataset.
\end{itemize}
This is a ``Ramsey fringe'' conceptually, even though the source is synthetic.

\subsubsection{Scenario D: fail-closed policy under corruption}
\textbf{Setup.} Randomly flip \texttt{meas\_valid} or misalign a payload tag (if you use UART framing).

\textbf{What you demonstrate.}
\begin{itemize}
	\item controller raises a flag and enters ERR rather than emitting nonsense,
	\item resynchronization requires explicit reset/start.
\end{itemize}

\subsection{Integration notes (measurement, Grover, and QKD as FSMs)}
\subsubsection{General integration pattern}
Most protocols become:
\[
\text{(prepare)} \to \text{(apply controlled sequence)} \to \text{(measure)} \to
\text{(classical update)} \to \text{repeat}.
\]
So the FSM view is not an appendix curiosity; it is the control backbone.

\subsubsection{Measurement as an event stream}
In later chapters:
\begin{itemize}
	\item stabilizer measurements are a syndrome stream,
	\item QPE produces a multi-bit phase estimate across rounds,
	\item QKD produces raw key bits + basis info + sifting decisions.
\end{itemize}
All of these can be integrated by the same principles:
\begin{itemize}
	\item explicit \textbf{handshake} for when data is valid,
	\item explicit \textbf{timeouts} and fail-closed states,
	\item explicit \textbf{logging registers} for observability,
	\item bounded per-round work (for real-time guarantees).
\end{itemize}

\subsubsection{Grover as an FSM (control-only view)}
Grover iterations are repeated blocks:
\[
\underbrace{\text{oracle}}_{\text{controlled pulses}} \;\to\;
\underbrace{\text{diffusion}}_{\text{controlled pulses}} \;\to\;
\text{repeat }R\text{ times}.
\]
That is literally:
\[
\texttt{for }i=1..R:\ \texttt{ORACLE; DIFFUSE;}
\]
and on FPGA it becomes:
\begin{itemize}
	\item a loop counter \(i\),
	\item two macro-states ORACLE and DIFFUSE,
	\item optional measurement at the end (or mid-circuit for adaptive variants).
\end{itemize}

\subsubsection{QKD as an FSM (control-only view)}
Even BB84/E91 can be seen as:
\[
\text{emit/prepare} \to \text{measure} \to \text{record tags (basis)} \to
\text{classical post-processing states}.
\]
The key point for this appendix:
the \emph{security logic becomes explicit control states and counters} that must be correct and observable.
	
\section{Datapath and Custom Instruction Sets for QPU Control}
\label{sec:appendix-isa}

\subsection{Big picture: from algorithms to deterministic control}
Quantum-control workloads look like ``algorithms'', but at runtime they are really
\emph{deterministic schedules} plus \emph{event-triggered updates}:
\[
\text{(time schedule)}\;+\;\text{(measurement events)}\;\Rightarrow\;\text{(next pulses / frame updates)}.
\]
A general-purpose CPU is optimized for flexibility and average-case performance.
A quantum controller is optimized for:
\begin{itemize}
	\item \textbf{determinism:} bounded and repeatable latency,
	\item \textbf{tight timing:} cycle-level pulse scheduling,
	\item \textbf{streaming I/O:} continuous measurement/event ingestion,
	\item \textbf{fail-closed behavior:} explicit flags and safe aborts,
	\item \textbf{high fanout control:} update many channels quickly (even if each update is simple).
\end{itemize}

\medskip
\noindent\textbf{Control-architecture slogan.}
\[
\text{A good quantum controller is a datapath with just enough instructions.}
\]

\subsection{Why general-purpose CPUs are too slow}
This is not about ``GHz''; it's about \emph{worst-case latency} and \emph{I/O shape}.

\subsubsection{Latency is dominated by the control stack, not arithmetic}
A CPU loop often pays:
\begin{itemize}
	\item OS jitter / interrupts,
	\item cache misses / branch mispredicts,
	\item memory hierarchy latency,
	\item bus / driver overhead for I/O,
	\item non-deterministic scheduling on multicore systems.
\end{itemize}
In quantum control (especially with feedback), you care about:
\[
\textbf{p99/p999 latency} \quad \text{and} \quad \textbf{bounded service time per event}.
\]
A custom datapath can guarantee:
\[
\text{event arrives} \Rightarrow \text{decision emitted within } \le L_{\max}\ \text{cycles}.
\]

\subsubsection{The ``tiny work, huge consequence'' regime}
Many control decisions are simple:
\begin{itemize}
	\item toggle a pulse enable,
	\item update a phase register,
	\item increment a counter,
	\item apply a Pauli-frame XOR,
	\item route an event to a FIFO.
\end{itemize}
But they must happen \emph{on time}. This is exactly what FPGAs do well.

\subsection{Datapath design for quantum control}
A datapath is the circuit that executes a restricted set of operations every clock.
The design goal is: \textbf{bounded, composable micro-kernels}.

\subsubsection{Minimal control datapath blocks}
A practical control datapath is built from:
\begin{itemize}
	\item \textbf{Program counter (PC)} and instruction memory (or microcode ROM),
	\item \textbf{Register file} (small, fast state),
	\item \textbf{ALU} for integer ops (add/sub/and/xor/compare),
	\item \textbf{Event input} (valid/ready stream) and a small FIFO,
	\item \textbf{Timer/counters} for schedule timing,
	\item \textbf{Control output registers} (pulse/phase/frame update words),
	\item \textbf{Status/flag block} (overflow, timeout, parity/CRC errors).
\end{itemize}

\subsubsection{Visualization: datapath as ``compute + time + I/O''}
\begin{figure}[t]
	\centering
	\begin{tikzpicture}[
		font=\small,
		arr/.style={-Latex, thick},
		big/.style={
			draw, rounded corners, thick,
			align=center,
			inner xsep=8pt, inner ysep=7pt,
			text width=0.88\linewidth, 
			minimum height=1.05cm
		},
		box/.style={
			draw, rounded corners, thick,
			align=center,
			inner xsep=7pt, inner ysep=6pt,
			text width=0.25\linewidth, 
			minimum height=0.95cm
		},
		node distance=7mm and 10mm
		]
		
		\node[big] (core)
		{Control core = PC + decoder + register file + ALU + branch/compare};
		
		\matrix (mid) [matrix of nodes,
		nodes={box},
		column sep=9mm,
		row sep=0mm,
		below=8mm of core
		]{
			{event in\\[-1pt]\footnotesize (meas/syndrome)} &
			{timers\\[-1pt]\footnotesize (cycle counters)} &
			{control out\\[-1pt]\footnotesize (pulse/frame)} \\
		};
		
		\draw[arr] (mid-1-1.north) -- (core.south west);
		\draw[arr] (mid-1-2.north) -- (core.south);
		\draw[arr] (core.south east) -- (mid-1-3.north);
		
		\node[big, below=8mm of mid] (obs)
		{observability: logs + flags + latency counters (mean/p99)};
		
		\draw[arr] (core.south) -- (obs.north);
		
	\end{tikzpicture}
	\caption{A minimal quantum-control datapath: the ``compute'' is small, but time, I/O, and observability are first-class.}
	\label{fig:ctrl-datapath}
\end{figure}
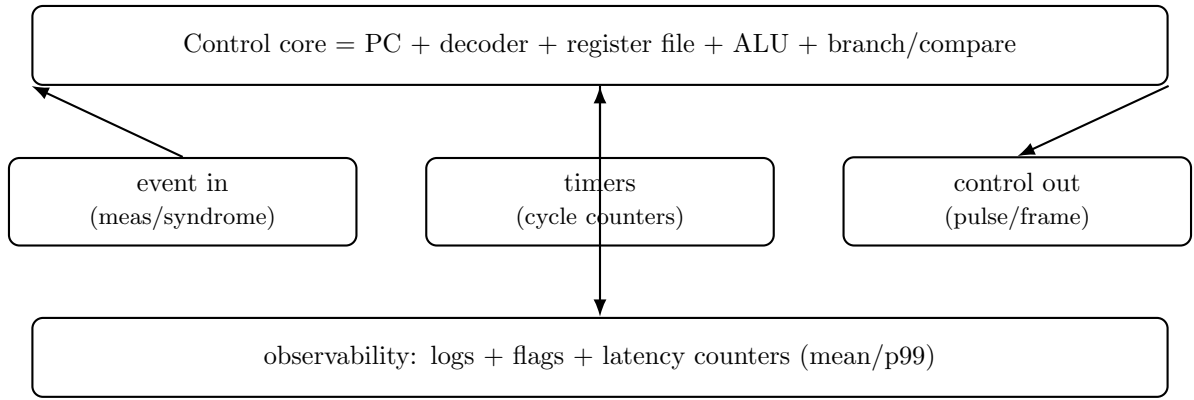

\subsubsection{Streaming event ingestion is part of the datapath}
Treat measurement outcomes as a stream:
\[
(\texttt{tag},\texttt{payload},\texttt{timestamp})_t
\]
and ensure each stage obeys a valid/ready contract.
This avoids silent data loss and makes worst-case behavior analyzable.

\subsection{Single-cycle quantum control units}
A ``single-cycle'' unit means: for a restricted set of control actions,
\emph{one clock tick} is enough to update outputs (or the state that drives outputs).

\subsubsection{When single-cycle makes sense}
Single-cycle is ideal when:
\begin{itemize}
	\item decisions are simple (bitwise / small integer),
	\item the required state fits in registers,
	\item outputs are register-mapped (no long buses),
	\item you need deterministic response to events.
\end{itemize}

\subsubsection{Examples of single-cycle control micro-ops}
\begin{itemize}
	\item \textbf{Pauli-frame update:} \(\texttt{frame} \leftarrow \texttt{frame} \oplus \texttt{mask}\).
	\item \textbf{Phase accumulator step:} \(\texttt{phase} \leftarrow \texttt{phase} + \Delta\).
	\item \textbf{Conditional pulse gate:} \(\texttt{pulse\_en} \leftarrow (\texttt{meas\_bit} \land \texttt{policy})\).
	\item \textbf{Deadline check:} set flag if \(t > D\).
\end{itemize}

\subsubsection{Tradeoff: single-cycle vs.\ multi-cycle}
If you need multiplication/division, large table lookups, or memory-heavy access,
you typically use multi-cycle units (or precompute in LUTs).
For quantum control, it is often better to:
\begin{itemize}
	\item push expensive work to \textbf{offline calibration},
	\item run small \textbf{bounded kernels} online.
\end{itemize}

\subsection{A minimal quantum instruction set}
The goal is not to build ``a CPU''; it is to build \emph{just enough ISA}
to express common control patterns with deterministic bounds.

\subsubsection{Guiding constraints}
\begin{itemize}
	\item \textbf{Fixed, small latency per instruction} (ideally 1 cycle, sometimes few).
	\item \textbf{Memory discipline:} bounded RAM accesses (or explicit wait states).
	\item \textbf{Event-driven operations:} ingest from FIFO without stalling unpredictably.
	\item \textbf{Time primitives:} waits and deadlines are instructions, not software loops.
\end{itemize}

\subsubsection{State model}
Let registers include:
\[
R0,\dots,R_{15},\quad \texttt{PC},\quad \texttt{FLAGS},\quad \texttt{TIMER},\quad \texttt{FRAME}.
\]
And memory-mapped I/O registers:
\[
\texttt{OUT\_PULSE},\ \texttt{OUT\_PHASE},\ \texttt{IN\_EVT},\ \texttt{IN\_STATUS}.
\]

\subsubsection{Instruction categories (a usable ``control ISA'')}
\paragraph{(A) Integer + bitwise (ALU)}
\begin{itemize}
	\item \texttt{ADD rd, ra, rb} \quad\texttt{XOR rd, ra, rb}
	\item \texttt{AND/OR} \quad\texttt{CMP} (sets flags)
	\item \texttt{MOV rd, imm} \quad (small immediate)
\end{itemize}

\paragraph{(B) Time and sequencing}
\begin{itemize}
	\item \texttt{WAIT k} \quad (stall for \(k\) cycles deterministically)
	\item \texttt{READT rd} \quad (read cycle counter)
	\item \texttt{DEADLINE k} \quad (set flag if \(t>k\))
\end{itemize}

\paragraph{(C) Event ingestion}
\begin{itemize}
	\item \texttt{EVTPOP rd} \quad (pop next event word if available, else set flag)
	\item \texttt{EVTPEEK rd} \quad (peek without pop)
\end{itemize}

\paragraph{(D) Control outputs}
\begin{itemize}
	\item \texttt{OUTPULSE rs} \quad (write pulse-control word)
	\item \texttt{OUTPHASE rs} \quad (write phase word)
	\item \texttt{FRAMEXOR rs} \quad (Pauli-frame update)
\end{itemize}

\paragraph{(E) Control flow}
\begin{itemize}
	\item \texttt{JMP addr}
	\item \texttt{JZ/JNZ addr} (branch on flags)
\end{itemize}

\subsubsection{Visualization: ``ISA = schedule + event + output''}
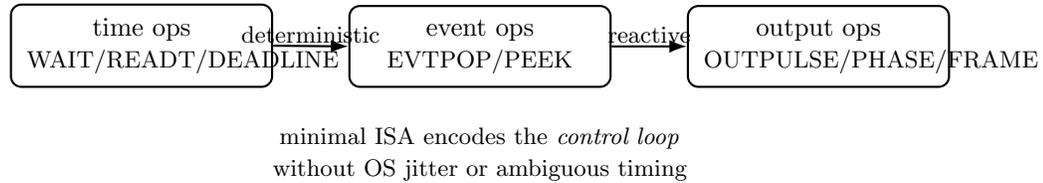
\begin{figure}[t]
	\centering
	\begin{tikzpicture}[
		font=\small,
		box/.style={
			draw, thick, rounded corners,
			align=center,
			minimum height=0.95cm,
			inner xsep=6pt, inner ysep=5pt,
			text width=3.05cm 
		},
		arr/.style={-Latex, thick},
		lab/.style={font=\footnotesize, inner sep=1pt}
		]
		
		\matrix (m) [matrix of nodes,
		nodes={box},
		column sep=10mm, 
		row sep=6mm
		]{
			{time ops\\[-1pt]\footnotesize WAIT/READT/DEADLINE} &
			{event ops\\[-1pt]\footnotesize EVTPOP/PEEK} &
			{output ops\\[-1pt]\footnotesize OUTPULSE/PHASE/FRAME} \\
		};
		
		\draw[arr] (m-1-1.east) -- node[lab, above] {deterministic} (m-1-2.west);
		\draw[arr] (m-1-2.east) -- node[lab, above] {reactive}      (m-1-3.west);
		
		\node[align=center, font=\footnotesize, text width=0.88\textwidth]
		at ($(m.south)+(0,-0.75cm)$)
		{minimal ISA encodes the \emph{control loop}\\
			without OS jitter or ambiguous timing};
		
	\end{tikzpicture}
	\caption{Control ISA categories: time primitives, event primitives, and output primitives.}
	\label{fig:isa-categories}
\end{figure}

\subsubsection{Worked micro-program (conceptual)}
This pseudo-assembly shows a single-shot policy:
\begin{verbatim}
	; PREP
	MOV   R1, 0          ; ones = 0
	MOV   R2, 0          ; shots = 0
	; PULSE1
	OUTPULSE Rpulse1
	WAIT  T1
	; WAIT(tau)
	WAIT  TAU
	; PULSE2
	OUTPULSE Rpulse2
	WAIT  T2
	; MEASURE: pop event when available
	EVTPOP R0
	JZ    ERR            ; if no event, fail-closed (or WAIT bounded)
	; UPDATE
	AND   R3, R0, 1      ; meas_bit
	ADD   R1, R1, R3
	ADD   R2, R2, 1
	JMP   DONE
	ERR:
	; set flag, safe output state
	OUTPULSE Rsafe
	DONE:
\end{verbatim}
Even without a ``real qubit'', this demonstrates deterministic sequencing and event handling.

\subsection{Lab: tiny quantum processor on iCEstick}
This appendix does \emph{not} require a full softcore CPU.
A ``tiny processor'' can be:
\begin{itemize}
	\item microcode ROM (instruction words),
	\item a PC + decoder,
	\item 16 registers,
	\item a simple ALU,
	\item memory-mapped I/O for events and outputs.
\end{itemize}

\subsubsection{What you should be able to demonstrate (hardware-visible)}
\begin{itemize}
	\item deterministic pulse schedule (LEDs show state/pulse windows),
	\item event pop and conditional branch (LED indicates event/no-event),
	\item bounded timeout and fail-closed behavior (ERR LED),
	\item online counter accumulation (LED bank or UART report).
\end{itemize}

\subsubsection{Verification hook (recommended)}
Before optimizing speed, verify correctness:
\begin{itemize}
	\item run the same micro-program in a Python golden model,
	\item compare register traces (state, PC, outputs per cycle),
	\item inject missing/extra events to validate fail-closed logic.
\end{itemize}

\subsection{Integration notes: QEC decoding (Track A)}
A surface-code decoder and a control ISA meet at a clean seam:
\begin{itemize}
	\item \textbf{decoder kernel:} consumes \(\Delta s_t\) (events), produces \texttt{FRAMEXOR mask} or correction word,
	\item \textbf{controller:} schedules when to run kernels, manages FIFOs, enforces deadlines, reports flags.
\end{itemize}

\subsubsection{Two patterns that scale}
\paragraph{(1) ``ISA + kernel coprocessor''}
Keep the ISA core small; attach a streaming kernel unit (e.g.\ UF micro-kernel).
ISA issues:
\[
\texttt{KERNEL\_START},\ \texttt{KERNEL\_STEP},\ \texttt{KERNEL\_READOUT}.
\]

\paragraph{(2) ``Microcoded FSM only''}
For many decoders, you can skip a general ISA and implement a bounded-pass FSM directly.
The ISA concepts still help: time primitives, event primitives, output primitives.

\subsubsection{The key contract}
Whatever you build, keep this contract explicit:
\begin{itemize}
	\item input stream is framed and timestamped,
	\item output updates are atomic (no half-updates),
	\item deadlines produce explicit flags (not silent slowdowns).
\end{itemize}
	
\section{From FPGA Control to Quantum Error Correction}
\label{sec:appendix-fpga-to-qec}

\subsection{Syndrome streams as classical dataflows}
The conceptual jump from ``FPGA control'' to ``quantum error correction (QEC)'' is smaller than it looks:
QEC turns quantum noise into \emph{classical bits that arrive every cycle}. From that point on,
you are doing classical real-time streaming.

\medskip
\noindent\textbf{Dataflow viewpoint.}
Each stabilizer-measurement round produces a vector of bits (or $\pm1$ values). In practice, what you decode is
often a \emph{difference stream}
\[
\Delta s_t := s_t \oplus s_{t-1},
\]
so that measurement faults appear as time-like events and data-qubit faults appear as space-like events.

\subsubsection{What the FPGA actually sees}
The FPGA does \emph{not} see qubits. It sees a stream of records:
\[
(\texttt{round\_id},\ \texttt{check\_id},\ \texttt{bit},\ \texttt{timestamp},\ \texttt{flags}).
\]
Your first engineering task is to define how these records are represented, framed, and validated.

\subsubsection{Framing: one round is a packet}
A robust framing rule is:
\begin{itemize}
	\item a \textbf{Start-of-Frame (SOF)} marker,
	\item a \textbf{round header} (round index $t$, distance $d$, window $W$, etc.),
	\item a \textbf{payload} (packed syndrome bits or sparse defect list),
	\item a \textbf{CRC/parity} to detect corruption.
\end{itemize}
This is where ``quantum'' becomes ``networking'': the decoder is only as trustworthy as its input framing.

\subsubsection{Two payload encodings (dense vs.\ sparse)}
Let $n_c$ be the number of checks per round.
\begin{itemize}
	\item \textbf{Dense payload (bitset):}
	pack $\Delta s_t \in \{0,1\}^{n_c}$ into words of width $W$.
	\[
	\texttt{word}[k] = \sum_{j=0}^{W-1} \Delta s_t[\,kW+j\,]\,2^j.
	\]
	\emph{Pros:} simple, deterministic bandwidth. \emph{Cons:} wastes bandwidth when few defects.
	\item \textbf{Sparse payload (defect list):}
	transmit only indices where $\Delta s_t=1$:
	\[
	\{\, i_1,i_2,\dots,i_m \,\},\quad m \ll n_c \text{ at low error rate}.
	\]
	\emph{Pros:} bandwidth-efficient. \emph{Cons:} variable-length, needs robust length checks.
\end{itemize}

\subsubsection{Visualization: syndrome stream as a pipeline}
\begin{figure}[t]
	\centering
	\begin{tikzpicture}[
		font=\small,
		box/.style={
			draw, thick, rounded corners,
			align=center,
			minimum height=0.95cm,
			inner xsep=6pt, inner ysep=5pt,
			text width=2.85cm
		},
		arr/.style={-Latex, thick},
		lab/.style={font=\footnotesize, inner sep=1pt}
		]
		
		\matrix (top) [matrix of nodes,
		nodes={box},
		column sep=9mm,
		row sep=7mm
		]{
			{framed input\\[-1pt]\footnotesize (SOF/hdr/payload/CRC)} &
			{FIFO buffer\\[-1pt]\footnotesize (occupancy + flags)} &
			{preprocess\\[-1pt]\footnotesize (parse / unpack)} \\
		};
		
		\matrix (bot) [matrix of nodes,
		below=9mm of top,
		nodes={box},
		column sep=9mm
		]{
			{decoder kernel\\[-1pt]\footnotesize (bounded passes)} &
			{output update\\[-1pt]\footnotesize (frame/correction)} \\
		};
		
		\draw[arr] (top-1-1.east) -- node[lab, above]{valid/ready}  (top-1-2.west);
		\draw[arr] (top-1-2.east) -- node[lab, above]{backpressure} (top-1-3.west);
		
		\draw[arr] (top-1-3.south) |- (bot-1-1.west);
		
		\draw[arr] (bot-1-1.east) -- (bot-1-2.west);
		
		\node[align=center, font=\footnotesize, text width=0.90\textwidth]
		at ($(bot.south)+(0,-0.80cm)$)
		{core requirement: \textbf{no silent drops} + \textbf{bounded latency} + \textbf{explicit fail flags}};
		
	\end{tikzpicture}
	\caption{Syndrome decoding is a classical streaming dataflow: parse $\rightarrow$ buffer $\rightarrow$ decode $\rightarrow$ emit correction.}
	\label{fig:qec-dataflow}
\end{figure}

\subsubsection{Timing model: cadence, deadlines, and backlog}
Let the stabilizer cycle period be $P$ clocks (or microseconds), and let the decoder service time per round be $S_t$.
Define backlog $B_t$ in rounds.
A minimal stability condition is:
\[
\mathbb{E}[S_t] < P
\quad\text{and}\quad
\Pr(S_t > P) \text{ is small (p99/p999)}.
\]
In hardware terms:
\begin{itemize}
	\item FIFO depth must absorb bursts: $D \gtrsim \text{(burst size)}$,
	\item occupancy must be monitored,
	\item overflow triggers a \textbf{fail-closed} policy (do not emit ambiguous outputs).
\end{itemize}

\subsection{LUTs, FSMs, and decoder logic}
Once the stream is framed and buffered, the remaining work is \emph{deterministic logic}.
On an FPGA, that logic naturally decomposes into:
\[
\text{LUTs (combinational)} \;+\; \text{FSMs (sequencing)} \;+\; \text{RAM (state)}.
\]

\subsubsection{LUTs: local maps from bits to actions}
A LUT is a truth table. Many ``control'' and ``decoder'' decisions are LUT-shaped:
\begin{itemize}
	\item map a small neighborhood syndrome pattern to a local action,
	\item map flags to fail-safe output modes,
	\item map a check index to lattice coordinates (or adjacency lists).
\end{itemize}

\paragraph{Example LUT: parity-to-defect}
Suppose a check produces a raw bit $b_t$; you decode on differences $\Delta b_t=b_t\oplus b_{t-1}$.
That is a 2-bit input LUT:
\[
(b_{t-1},b_t)\mapsto \Delta b_t.
\]
This is trivial, but it illustrates the principle: many preprocessing steps are small LUTs.

\subsubsection{FSMs: bounded passes are how you get real-time guarantees}
Decoders that guarantee real-time behavior typically enforce:
\begin{itemize}
	\item a fixed number of passes over the lattice (or a window),
	\item bounded work per pass (local updates),
	\item explicit convergence detection (done flag),
	\item explicit timeout (fail flag if not done by pass limit).
\end{itemize}
This fits an FSM perfectly:
\[
\texttt{IDLE} \rightarrow \texttt{LOAD} \rightarrow \texttt{PASS}_1 \rightarrow \cdots \rightarrow \texttt{PASS}_K \rightarrow \texttt{EMIT}.
\]

\subsubsection{RAM: the lattice is state}
Any decoder with memory needs arrays indexed by lattice sites/edges:
\begin{itemize}
	\item defect flags,
	\item cluster labels / parents (Union--Find),
	\item weights / confidences,
	\item visitation marks / timestamps.
\end{itemize}
The key is to choose an address mapping and update schedule that avoids hazards.

\subsubsection{Visualization: LUT vs.\ FSM vs.\ RAM roles}
\begin{figure}[t]
	\centering
	\begin{tikzpicture}[
		font=\small,
		box/.style={
			draw, rounded corners, thick,
			align=center,
			inner xsep=8pt, inner ysep=6pt,
			text width=0.26\linewidth, 
			minimum height=1.0cm
		},
		arr/.style={-Latex, thick},
		node distance=10mm and 12mm
		]
		\node[box] (lut) {LUTs\\[-1pt]\footnotesize local maps / guards};
		\node[box, right=10mm of lut] (fsm) {FSM\\[-1pt]\footnotesize bounded passes / timing};
		\node[box, right=10mm of fsm] (ram) {RAM\\[-1pt]\footnotesize lattice state arrays};
		
		\draw[arr] (lut) -- node[above]{\footnotesize decisions} (fsm);
		\draw[arr] (fsm) -- node[above]{\footnotesize schedules} (ram);
		
		\draw[arr]
		(ram.south) to[out=-70,in=-110,looseness=0.85]
		node[below, font=\footnotesize, yshift=-1mm]{read-modify-write}
		(lut.south);
		
		\node[align=center, font=\footnotesize] at ($(fsm.south)+(0,-9mm)$) {
			bounded-time decoding = FSM discipline + local LUT updates + observable RAM state
		};
	\end{tikzpicture}
	\caption{FPGA decoder logic roles: LUTs implement local rules, FSMs enforce bounded-time structure, RAM stores lattice state.}
	\label{fig:lut-fsm-ram}
\end{figure}
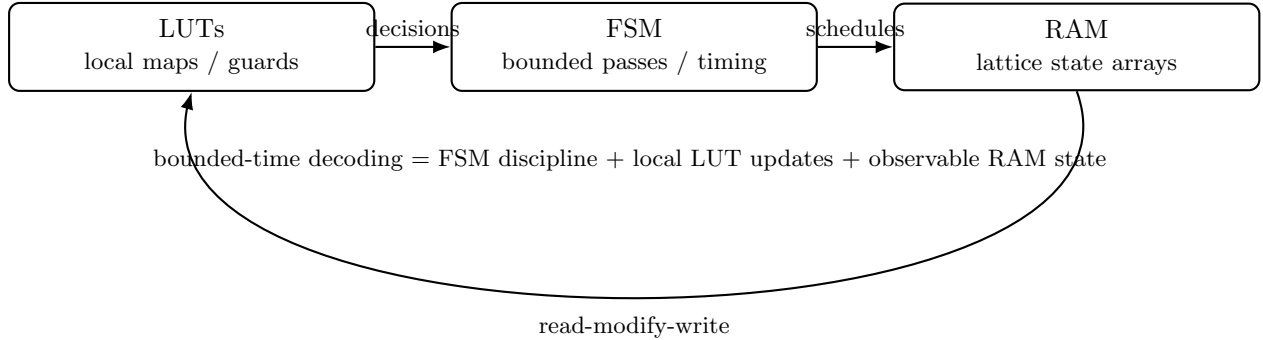

\subsubsection{A concrete ``bounded local rule'' example (decoder-flavored)}
Even without committing to a specific decoder family, you can implement a generic local update:
\begin{itemize}
	\item each site has a 1-bit defect flag $d[x,y]$,
	\item each site has a 2-bit local state $s[x,y]\in\{0,1,2,3\}$,
	\item update rule depends on $d$ in a small neighborhood (e.g.\ von Neumann neighbors).
\end{itemize}
This compiles to:
\[
(s[x,y], d[x,y], d[x\pm1,y], d[x,y\pm1]) \mapsto s'[x,y]
\]
which is a LUT, scheduled across the grid by an FSM pass.

\subsubsection{Fail-closed policy is part of the decoder}
Hardware must specify what happens when:
\begin{itemize}
	\item packet CRC fails,
	\item framing loses sync,
	\item FIFO overflows,
	\item deadline missed,
	\item kernel does not converge within $K$ passes.
\end{itemize}
A safe policy is:
\[
\text{raise flag} \;+\; \text{emit explicit ``invalid'' output} \;+\; \text{resynchronize}.
\]
Never output a silently-corrupted correction.

\subsection{Integration notes}
This section ties together the appendices with the main decoder chapters by making the contracts explicit.

\subsubsection{Contract 1: streaming interface}
Define (and version) the wire format:
\begin{itemize}
	\item SOF, header fields, payload encoding, CRC,
	\item how timestamps are represented,
	\item how dropped/invalid rounds are signaled.
\end{itemize}

\subsubsection{Contract 2: service guarantees}
Document the decoder as a real-time service:
\begin{itemize}
	\item worst-case passes $K$,
	\item worst-case cycles per pass,
	\item output deadline $D$,
	\item what happens on timeout (flag + safe mode).
\end{itemize}

\subsubsection{Contract 3: observability}
Make the system debuggable:
\begin{itemize}
	\item occupancy counters, overflow/underflow flags,
	\item latency counters/histograms (mean, p99),
	\item corruption counters (CRC fail, resync count),
	\item correctness indicators in simulation (logical-failure flag).
\end{itemize}

\subsubsection{What this enables later}
Once these contracts exist, later chapters can focus on the \emph{algorithmic kernel}
(MWPM, Union--Find, local rules, spacetime graphs) without re-litigating I/O and timing.
In other words:
\[
\text{stable stream contracts} \Rightarrow \text{swappable decoder kernels}.
\]
	
\section{From FPGA Control to Quantum Cryptography}
\label{sec:appendix-fpga-to-qcrypto}

\subsection{FSMs for BB84 and E91 protocol flow}
Quantum cryptography protocols look ``high-level'' on paper, but on hardware they reduce to:
\[
\text{(1) a time-ordered exchange of messages} \;+\; \text{(2) local updates of counters and buffers}.
\]
That is exactly the domain of \emph{finite state machines (FSMs)}.

\medskip
\noindent\textbf{Key hardware viewpoint.}
BB84/E91 have two layers:
\begin{itemize}
	\item \textbf{Quantum layer (device-facing):} prepare/measure (or receive measurement outcomes).
	\item \textbf{Classical layer (protocol-facing):} sifting, QBER estimation, error correction, privacy amplification, and abort rules.
\end{itemize}
In these appendices we focus on the \emph{classical layer} as an FPGA streaming controller, assuming the quantum layer produces
timestamped events (or the FPGA drives it via triggers).

\subsubsection{BB84 as an FSM (conceptual flow)}
BB84 can be implemented as a deterministic controller over a stream of per-shot records.
For each shot $i$, define a record:
\[
r_i = (\;t_i,\ a_i,\ x_i,\ b_i,\ y_i,\ \texttt{flags}_i\;),
\]
where
\begin{itemize}
	\item $a_i \in \{Z,X\}$ is Alice's basis choice,
	\item $x_i \in \{0,1\}$ is Alice's bit choice,
	\item $b_i \in \{Z,X\}$ is Bob's basis choice,
	\item $y_i \in \{0,1\}$ is Bob's measurement result,
	\item $t_i$ is a timestamp or shot counter (synchronization),
	\item \texttt{flags} include validity bits (lost photon, detector saturation, etc.).
\end{itemize}
In a device, these values are not all produced in one place, but as a system model this record is what the FSM consumes.

\paragraph{Core BB84 phases (hardware-friendly).}
A minimal set of phases is:
\[
\texttt{INIT} \rightarrow \texttt{ACQUIRE} \rightarrow \texttt{SIFT} \rightarrow
\texttt{SAMPLE\_QBER} \rightarrow \texttt{KEY\_ACCUM} \rightarrow \texttt{DONE/ABORT}.
\]
\begin{itemize}
	\item \texttt{ACQUIRE:} ingest $r_i$ records into FIFO/buffers (streaming).
	\item \texttt{SIFT:} keep indices where $a_i=b_i$; discard otherwise.
	\item \texttt{SAMPLE\_QBER:} choose a test subset of sifted indices (deterministic PRNG with a shared seed, or host-provided mask).
	\item \texttt{KEY\_ACCUM:} write remaining sifted bits to a key buffer (or output stream).
	\item \texttt{ABORT:} if QBER too high or synchronization fails, fail-closed.
\end{itemize}

\subsubsection{E91 as an FSM (conceptual flow)}
E91 uses entanglement and correlation tests. Hardware-wise it is still a stream of shot records.
A typical record is:
\[
r_i = (\;t_i,\ a_i,\ b_i,\ y_i^{(A)},\ y_i^{(B)},\ \texttt{flags}_i\;),
\]
where $a_i,b_i$ are measurement settings (bases/angles) and $y_i^{(A)},y_i^{(B)}\in\{0,1\}$ are outcomes.

\paragraph{Core E91 phases.}
\[
\texttt{INIT} \rightarrow \texttt{ACQUIRE} \rightarrow \texttt{SIFT} \rightarrow
\texttt{EST\_CHSH} \rightarrow \texttt{KEY\_ACCUM} \rightarrow \texttt{DONE/ABORT}.
\]
Here \texttt{EST\_CHSH} computes a Bell/CHSH statistic on a subset (or continuously) to detect tampering
and quantify nonclassical correlations.

\subsubsection{Visualization: protocol FSM skeleton}
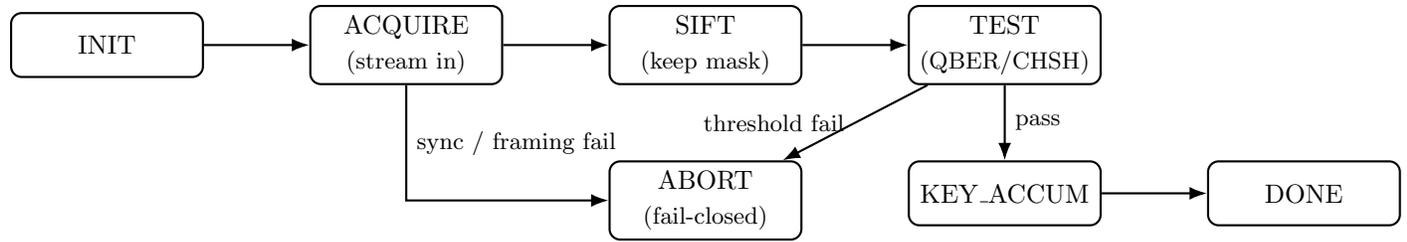
\begin{figure}[t]
	\centering
	\begin{tikzpicture}[
		font=\small,
		st/.style={
			draw, rounded corners, thick,
			minimum width=2.55cm, minimum height=0.85cm,
			align=center
		},
		arr/.style={-Latex, thick},
		node distance=9mm and 14mm
		]
		\node[st] (init) {INIT};
		\node[st, right=14mm of init] (acq) {ACQUIRE\\\footnotesize (stream in)};
		\node[st, right=14mm of acq] (sift) {SIFT\\\footnotesize (keep mask)};
		\node[st, right=14mm of sift] (test) {TEST\\\footnotesize (QBER/CHSH)};
		
		\node[st, below=10mm of test] (key) {KEY\_ACCUM};
		\node[st, right=14mm of key] (done) {DONE};
		
		\node[st, below=10mm of sift] (abort) {ABORT\\\footnotesize (fail-closed)};
		
		\draw[arr] (init) -- (acq);
		\draw[arr] (acq) -- (sift);
		\draw[arr] (sift) -- (test);
		
		\draw[arr] (test) -- node[right]{\footnotesize pass} (key);
		\draw[arr] (key) -- (done);
		
		\draw[arr] (test) -- node[left]{\footnotesize threshold fail} (abort);
		
		\draw[arr] (acq.south) |- node[pos=0.25, right]{\footnotesize sync / framing fail} (abort.west);
		
	\end{tikzpicture}
	\caption{Protocol control as an FSM. ``TEST'' is QBER for BB84 or CHSH/Bell statistics for E91. Any integrity failure triggers ABORT.}
	\label{fig:qkd-fsm}
\end{figure}

\subsection{Streaming QBER estimation and abort logic}
The QBER (quantum bit error rate) is the main online security/health indicator:
it determines whether you continue, slow down, or abort.

\subsubsection{Streaming counters (no full storage needed)}
Let $\mathcal{I}_{\text{sift}}$ be the set of sifted indices (where bases match).
Define streaming counters:
\[
N \;:=\; \#\{i\in \mathcal{I}_{\text{sift}} \text{ observed}\},\qquad
E \;:=\; \#\{i\in \mathcal{I}_{\text{sift}} : x_i \neq y_i\}.
\]
Then the online estimator is
\[
\widehat{\mathrm{QBER}} \;:=\; \frac{E}{N},
\]
updated per accepted sifted event. This is FPGA-friendly: two counters and a comparator.

\paragraph{Where do ``errors'' come from in BB84?}
On the test subset, Alice reveals $x_i$ and Bob reveals $y_i$ (or Bob reveals $y_i$ and Alice compares).
Then the FPGA can count mismatches.

\subsubsection{Abort logic (fail-closed policy)}
A minimal policy is:
\[
\text{if }\widehat{\mathrm{QBER}} > \tau \text{ and } N \ge N_{\min},\ \text{then ABORT}.
\]
Here $\tau$ is a configured threshold and $N_{\min}$ avoids early noise instability.

\paragraph{Practical implementation detail: fixed-point comparisons.}
Avoid division in the critical path. Compare cross-multiplied integers:
\[
E > \tau N
\quad\Longleftrightarrow\quad
E \cdot 2^q > \tau_{\text{fp}} \cdot N,
\]
where $\tau_{\text{fp}}:=\lfloor \tau 2^q\rfloor$ is a fixed-point threshold.

\subsubsection{Visualization: streaming QBER estimator}
\begin{figure}[t]
	\centering
	\resizebox{\textwidth}{!}{%
		\begin{tikzpicture}[
			font=\small,
			box/.style={draw, rounded corners, align=center, minimum height=9.5mm},
			arr/.style={-Latex, thick},
			node distance=10mm and 16mm
			]
			\node[box, minimum width=48mm] (in) {sifted test events\\\footnotesize $(x_i,y_i)$ + valid};
			
			\node[box, right=16mm of in, minimum width=62mm] (cnt) {counters\\\footnotesize
				$N\!\leftarrow\!N{+}1$\\
				$E\!\leftarrow\!E{+}[x_i\neq y_i]$
			};
			
			\node[box, right=16mm of cnt, minimum width=64mm] (cmp) {threshold check\\\footnotesize
				$E\cdot 2^q > \tau_{\text{fp}} N$
			};
			
			\node[box, right=16mm of cmp, minimum width=52mm] (ab) {ABORT flag\\\footnotesize + reason code};
			
			\draw[arr] (in) -- (cnt);
			\draw[arr] (cnt) -- (cmp);
			\draw[arr] (cmp) -- node[above]{\footnotesize fail} (ab);
			
			\node[align=center, font=\footnotesize, text width=120mm]
			at ($(cmp.south)+(0,-10mm)$)
			{no division in critical path;\quad comparisons are integer-safe};
		\end{tikzpicture}%
	}
	\caption{Streaming QBER estimation: update two counters and compare using fixed-point arithmetic.}
	\label{fig:qber-stream}
\end{figure}
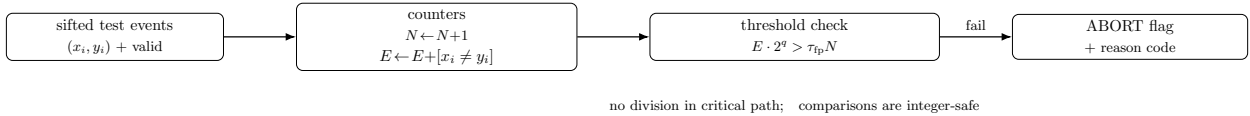

\subsubsection{E91 note: streaming CHSH instead of QBER}
For E91, the ``abort statistic'' can be a Bell/CHSH estimate.
You can stream counts for each setting pair $(a,b)$:
\[
N_{ab}=\#\{i: (a_i,b_i)=(a,b)\},\qquad
C_{ab}=\sum_{i:(a_i,b_i)=(a,b)} (-1)^{y_i^{(A)}\oplus y_i^{(B)}}.
\]
Then an empirical correlation is $\widehat{E}(a,b)=C_{ab}/N_{ab}$ and the CHSH combination is
\[
\widehat{S}=\widehat{E}(a,b)+\widehat{E}(a,b')+\widehat{E}(a',b)-\widehat{E}(a',b').
\]
Again: avoid division in the fast path; accumulate numerators/denominators and finalize on a slower path
(or periodically).

\subsection{Integration notes}
This subsection states the minimal contracts that let you plug the protocol FSM into your
existing ``streaming infrastructure'' from earlier appendices (UART/FIFO/framing/timestamps).

\subsubsection{Input contract (records + validity)}
\begin{itemize}
	\item Every record has a \textbf{round/shot index} and a \textbf{valid} bit.
	\item Any decoding/protocol step must ignore invalid records and increment an error counter.
	\item Basis/setting fields are explicitly encoded (e.g.\ 1 bit for $\{Z,X\}$).
\end{itemize}

\subsubsection{Output contract (flags before data)}
\begin{itemize}
	\item Outputs carry a status word first: \texttt{OK}, \texttt{ABORT}, \texttt{RESYNC}, \texttt{INVALID}.
	\item If \texttt{ABORT}, also output a \textbf{reason code} (QBER high, CHSH low, CRC fail, overflow, timeout).
	\item Key material is output only under \texttt{OK} and only after the test window closes.
\end{itemize}

\subsubsection{Observability contract (what you must log)}
At minimum export (or make readable on-chip):
\begin{itemize}
	\item counters: $N$, $E$ (and $N_{ab},C_{ab}$ for E91),
	\item computed summary: $\widehat{\mathrm{QBER}}$ (or $\widehat{S}$),
	\item integrity counters: CRC fails, resync count, FIFO overflow count,
	\item timing: per-stage latency (mean, p99) measured in cycles.
\end{itemize}

\subsubsection{Safety note (fail-closed is non-negotiable)}
A QKD controller that silently continues on corrupted framing, overflow, or misalignment is worse than useless:
it can output keys with unknown security properties. Therefore:
\[
\textbf{on integrity failure: ABORT, flag, and resynchronize (no silent degradation).}
\]
	
\section{Appendix Summary and Learning Roadmap}
\label{sec:appendix-roadmap}

\subsection{A minimal roadmap to build the decoder prototype}
This roadmap is written as a \emph{capability ladder}: each step produces a concrete capability you can
demonstrate and test. The emphasis is on \textbf{observability}, \textbf{bounded latency}, and
\textbf{fail-closed behavior}.

\medskip
\noindent\textbf{Milestone 0 (Day 0): Toolchain + deterministic I/O loop.}
\begin{itemize}
	\item \textbf{Goal:} compile/program reliably and run a deterministic loop with a known waveform on LEDs.
	\item \textbf{Deliverable:} a counter-driven LED pattern with a known period in cycles.
	\item \textbf{Acceptance check:} measured period matches the expected clock divider within 1 cycle.
\end{itemize}

\medskip
\noindent\textbf{Milestone 1: Streaming contract everywhere (valid/ready).}
\begin{itemize}
	\item \textbf{Goal:} every module obeys the same handshake rule.
	\item \textbf{Deliverable:} a producer $\rightarrow$ consumer pipeline that never drops data under backpressure.
	\item \textbf{Acceptance check:} assertions proving:
	(i) data stable while \texttt{valid=1 \& \texttt{ready=0}}, and
	(ii) transfer iff \texttt{valid \& ready} at clock edge.
\end{itemize}

\medskip
\noindent\textbf{Milestone 2: FIFO + framing (no drift).}
\begin{itemize}
	\item \textbf{Goal:} turn a byte stream into framed packets and never silently mis-parse.
	\item \textbf{Deliverable:} framing FSM (SOF/len/type/payload/CRC) feeding a FIFO with occupancy counters.
	\item \textbf{Acceptance check:} inject bit-flips/bursts; system asserts \texttt{RESYNC/ABORT} and recovers on next SOF.
\end{itemize}

\medskip
\noindent\textbf{Milestone 3: Timestamped latency truth (mean + p99).}
\begin{itemize}
	\item \textbf{Goal:} measure end-to-end latency \emph{on device} in cycles.
	\item \textbf{Deliverable:} arrival time $A_t$, finish time $F_t$, latency $\Delta_t=F_t-A_t$ histogram summary.
	\item \textbf{Acceptance check:} reproduce a synthetic burst and observe p99 spike in exported counters.
\end{itemize}

\medskip
\noindent\textbf{Milestone 4: Syndrome stream mock + invariants.}
\begin{itemize}
	\item \textbf{Goal:} run a deterministic syndrome/event stream through the pipeline.
	\item \textbf{Deliverable:} seeded generator for $\Delta s_t$ + monitors checking invariants
	(framing, sequence monotonicity, bounded FIFO occupancy).
	\item \textbf{Acceptance check:} fault-injection produces reason-coded failures; no silent output corruption.
\end{itemize}

\medskip
\noindent\textbf{Milestone 5: Bounded-pass decode kernel (real-time guarantee shape).}
\begin{itemize}
	\item \textbf{Goal:} implement a decode ``micro-kernel'' that finishes in a bounded number of passes.
	\item \textbf{Deliverable:} a kernel that consumes a fixed window and outputs a compact decision record
	(e.g.\ parity-frame update or correction summary) plus a \texttt{DONE} flag.
	\item \textbf{Acceptance check:} worst-case pass count is bounded; p99 bounded under adversarial bursts (within design envelope).
\end{itemize}

\medskip
\noindent\textbf{Milestone 6: Golden-model conformance loop (correctness before speed).}
\begin{itemize}
	\item \textbf{Goal:} every RTL decision matches a golden Python model under a published contract.
	\item \textbf{Deliverable:} conformance spec + regression suite + coverage on fault modes (drop/flip/burst).
	\item \textbf{Acceptance check:} bit-exact match on outputs (or match within declared nondeterminism, e.g.\ tie-break rules).
\end{itemize}

\medskip
\noindent\textbf{What you should be able to say after this roadmap.}
\begin{quote}\small
	``We have a streaming, observable, fail-closed decode appliance. We can measure p99 on-device,
	prove bounded passes for the kernel, and demonstrate conformance to a golden model under fault injection.''
\end{quote}

\subsubsection{Visualization: capability ladder}
\begin{figure}[t]
	\centering
	\begin{tikzpicture}[
		font=\small,
		box/.style={draw, rounded corners, minimum width=12.6cm, minimum height=0.85cm, align=left},
		arr/.style={-Latex, thick},
		node distance=0.35cm and 0.0cm
		]
		\node[box] (m0) {\textbf{M0} Toolchain + deterministic LED timing (cycle truth)};
		\node[box, below=of m0] (m1) {\textbf{M1} Valid/Ready contract everywhere (composable streaming)};
		\node[box, below=of m1] (m2) {\textbf{M2} Framing FSM + FIFO + CRC (no drift, no silent drop)};
		\node[box, below=of m2] (m3) {\textbf{M3} Timestamp + latency histograms (mean, p99)};
		\node[box, below=of m3] (m4) {\textbf{M4} Seeded syndrome/event generator + monitors (invariants)};
		\node[box, below=of m4] (m5) {\textbf{M5} Bounded-pass decode kernel (real-time guarantee shape)};
		\node[box, below=of m5] (m6) {\textbf{M6} Golden-model conformance + regression + fault injection};
		
		\foreach \a/\b in {m0/m1,m1/m2,m2/m3,m3/m4,m4/m5,m5/m6}{
			\draw[arr] (\a.south west)++(0.2,0) -- (\b.north west)++(0.2,0);
		}
	\end{tikzpicture}
	\caption{Minimal decoder-prototype roadmap as a capability ladder. Each rung adds a property you can test and defend.}
	\label{fig:roadmap-ladder}
\end{figure}
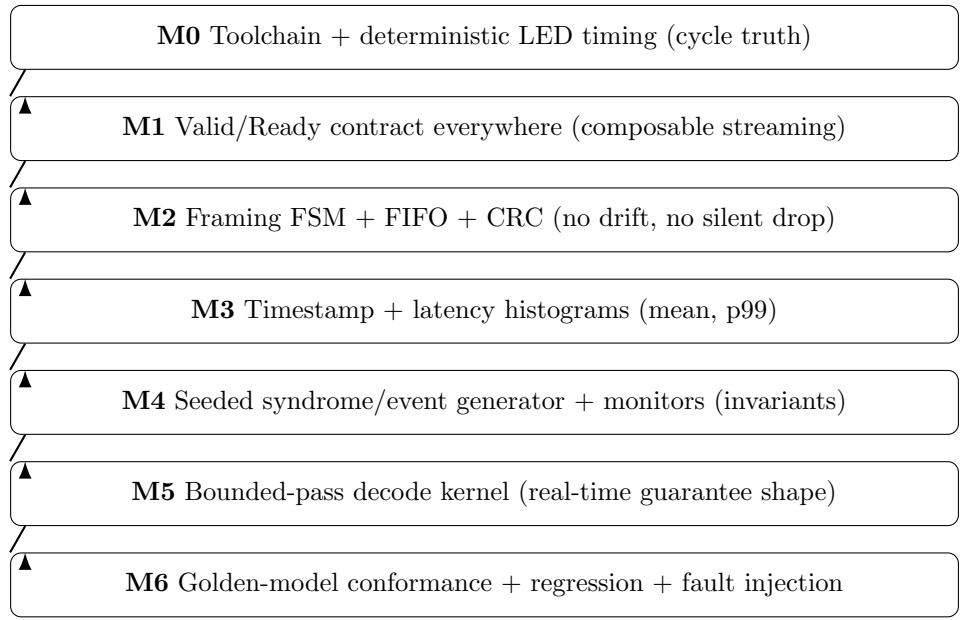

\subsection{A roadmap to extend from surface code to QLDPC}
The extension from surface codes to QLDPC is less about ``new math objects'' and more about
\textbf{new graph geometry} and \textbf{new decoder regimes}. The path below keeps the same streaming discipline.

\medskip
\noindent\textbf{Step A: Replace lattice locality with Tanner-graph locality.}
\begin{itemize}
	\item Surface code: checks and qubits live on a 2D grid; neighbors are geometric.
	\item QLDPC/CSS: checks and qubits live on a sparse bipartite Tanner graph; neighbors are graph-adjacent.
\end{itemize}
\noindent\textbf{Deliverable:} a single unified ``neighbor iterator'' API:
\[
\texttt{nbrs(node)} \rightarrow \{\texttt{node}_1,\dots,\texttt{node}_k\},
\]
backed either by (i) implicit grid rules or (ii) explicit adjacency lists.

\medskip
\noindent\textbf{Step B: Generalize the syndrome representation.}
\begin{itemize}
	\item Surface: syndrome differences on spacetime lattice vertices.
	\item QLDPC: syndrome bits correspond to parity checks $H_X e_Z^\top$ and $H_Z e_X^\top$.
\end{itemize}
\noindent\textbf{Deliverable:} a compact syndrome packet format that supports variable degree and multiple check families.

\medskip
\noindent\textbf{Step C: Decoder family selection map.}
\begin{itemize}
	\item \textbf{Peeling / BP-like:} natural on sparse graphs; needs damping/iterations.
	\item \textbf{UF-like clustering:} possible but requires graph-aware growth/merge rules.
	\item \textbf{Hybrid:} local iterative + occasional global repair.
\end{itemize}
\noindent\textbf{Constraint:} enforce \emph{bounded work per window} (fixed iteration budget, fixed passes).

\medskip
\noindent\textbf{Step D: Verification upgrades (graph is adversary).}
QLDPC graphs can create corner cases (short cycles, trapping sets).
\begin{itemize}
	\item Expand property-based tests to include adversarial small subgraphs.
	\item Add regression fixtures keyed by graph ID / adjacency hash.
\end{itemize}

\medskip
\noindent\textbf{Step E: Performance/latency contract stays identical.}
Keep the same exported metrics: occupancy, overflow, pass counts, mean/p99 latency, fail-closed reasons.
Only the ``neighbor geometry'' changes.

\subsubsection{Visualization: surface vs.\ QLDPC geometry}
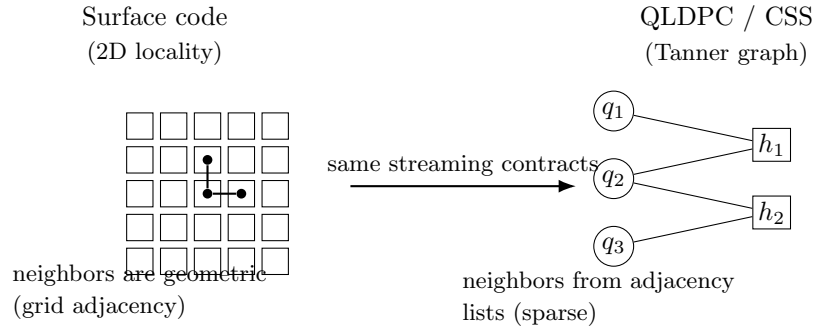
\begin{figure}[t]
	\centering
	\begin{tikzpicture}[
		font=\small,
		cell/.style={draw, minimum width=0.35cm, minimum height=0.35cm},
		dot/.style={circle, fill, inner sep=1.3pt},
		v/.style={circle, draw, inner sep=1.5pt},
		c/.style={rectangle, draw, inner sep=2.0pt},
		arr/.style={-Latex, thick}
		]
		\node[align=center] at (0,3.0) {Surface code\\\footnotesize (2D locality)};
		\begin{scope}[shift={(-0.2,0)}]
			\foreach \x in {0,...,4} {
				\foreach \y in {0,...,4} {
					\node[cell] at (\x*0.45,\y*0.45) {};
				}
			}
			\node[dot] (g0) at (0.9,0.9) {};
			\node[dot] (g1) at (1.35,0.9) {};
			\node[dot] (g2) at (0.9,1.35) {};
			\draw[thick] (g0) -- (g1);
			\draw[thick] (g0) -- (g2);
			\node[align=left, font=\footnotesize] at (0.0,-0.4) {neighbors are geometric\\(grid adjacency)};
		\end{scope}
		
		\node[align=center] at (7.6,3.0) {QLDPC / CSS\\\footnotesize (Tanner graph)};
		\begin{scope}[shift={(6.1,0.2)}]
			\node[v] (q1) at (0,1.8) {$q_1$};
			\node[v] (q2) at (0,0.9) {$q_2$};
			\node[v] (q3) at (0,0.0) {$q_3$};
			
			\node[c] (h1) at (2.1,1.35) {$h_1$};
			\node[c] (h2) at (2.1,0.45) {$h_2$};
			
			\draw (q1) -- (h1);
			\draw (q2) -- (h1);
			\draw (q2) -- (h2);
			\draw (q3) -- (h2);
			
			\node[align=left, font=\footnotesize] at (-0.2,-0.7) {neighbors from adjacency\\lists (sparse)};
		\end{scope}
		
		\draw[arr] (2.6,1.0) -- node[above, font=\footnotesize] {same streaming contracts} (5.6,1.0);
	\end{tikzpicture}
	\caption{Extending to QLDPC: replace 2D geometric locality with Tanner-graph locality, but keep the same streaming/latency/verification discipline.}
	\label{fig:surface-vs-qldpc}
\end{figure}

\subsection{A roadmap to integrate with a control stack}
Integration is a \emph{systems contract problem}. The decoder must behave like infrastructure:
predictable, observable, versioned, and safe under faults.

\subsubsection{Step 1: Define the streaming interface contract}
\begin{itemize}
	\item \textbf{Framing:} packet boundaries and resynchronization strategy.
	\item \textbf{Timestamps:} what they mean (shot index, cycle index, wall-clock sync).
	\item \textbf{Backpressure:} what happens when the consumer is slow (ready/valid, overflow flag).
	\item \textbf{Fail-closed policy:} which errors trigger ABORT vs.\ RESYNC.
\end{itemize}

\subsubsection{Step 2: Decide the ``decision output'' semantics}
Two common output semantics:
\begin{itemize}
	\item \textbf{Pauli-frame update:} output a compact frame delta to be tracked by the host/controller.
	\item \textbf{Explicit correction action:} output a command that directly drives a classical control update.
\end{itemize}
\noindent\textbf{Rule:} always output status+reason codes before any correction data.

\subsubsection{Step 3: Observability and telemetry as first-class outputs}
Export minimal on-device summaries at a fixed cadence:
\begin{itemize}
	\item latency (mean/p99), backlog/occupancy,
	\item integrity counters (CRC fails, resync, overflow),
	\item correctness proxies (in simulation: logical-failure flag; in lab: consistency checks).
\end{itemize}

\subsubsection{Step 4: A staged deployment plan}
\begin{itemize}
	\item \textbf{Stage A:} offline playback (recorded stream $\rightarrow$ decoder).
	\item \textbf{Stage B:} shadow mode (decoder runs in parallel; outputs logged but not applied).
	\item \textbf{Stage C:} closed loop (decoder output affects control decisions).
\end{itemize}
Only enter Stage C when Stage B shows stable p99 and no integrity regressions.

\subsubsection{Visualization: control-stack integration stages}


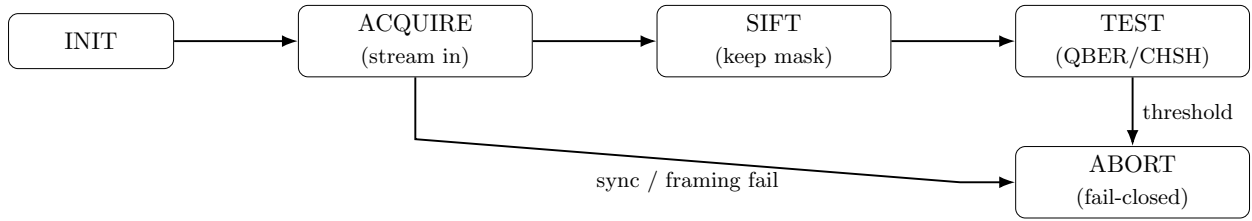
\begin{figure}[t]
	\centering
	\resizebox{\textwidth}{!}{%
		\begin{tikzpicture}[
			font=\small,
			box/.style={draw, rounded corners, align=center, minimum width=24mm, minimum height=8mm},
			bigbox/.style={draw, rounded corners, align=center, minimum width=34mm, minimum height=9mm},
			arr/.style={-Latex, thick},
			lab/.style={font=\footnotesize},
			node distance=10mm and 14mm
			]
			\node[box] (init) {INIT};
			\node[bigbox, right=18mm of init] (acq) {ACQUIRE\\{\footnotesize (stream in)}};
			\node[bigbox, right=18mm of acq]  (sift) {SIFT\\{\footnotesize (keep mask)}};
			\node[bigbox, right=18mm of sift] (test) {TEST\\{\footnotesize (QBER/CHSH)}};
			\node[bigbox, below=10mm of test] (abort) {ABORT\\{\footnotesize (fail-closed)}};
			
			\draw[arr] (init) -- (acq);
			\draw[arr] (acq) -- (sift);
			\draw[arr] (sift) -- (test);
			\draw[arr] (test) -- (abort) node[midway, right, lab] {threshold};
			
			\coordinate (acq_drop) at ($(acq.south)+(0,-9mm)$);
			\coordinate (abort_in) at ($(abort.west)+(-8mm,0)$);
			
			\draw[arr]
			(acq.south) -- (acq_drop) -- node[below, lab] {sync / framing fail} (abort_in) -- (abort.west);
			
		\end{tikzpicture}%
	}
	\caption{Protocol control as an FSM. ``TEST'' is QBER for BB84 or CHSH/Bell statistics for E91.
		Any integrity failure triggers ABORT (fail-closed).}
	\label{fig:qkd-fsm}
\end{figure}
	
	\section{Acknowledgements}
	\addcontentsline{toc}{section}{Acknowledgements}
	
	The author would like to express sincere gratitude to Dr.~Jeongwoo Jae, Hyun Lee, and Jongwon Yoon,
	whose guidance and discussions made it possible to learn quantum computing through the
	\emph{Enjoying Math} YouTube channel study series. Their explanations, questions, and collaborative
	spirit were essential in shaping both the intuition and the technical depth of these notes.
	
	The author is also grateful to Wonje Lee for many valuable discussions on FPGA architectures,
	hardware constraints, and implementation perspectives. These conversations played a crucial role
	in connecting abstract quantum algorithms to concrete, real-time classical hardware considerations.
	
	

\end{document}